\documentclass{article}
\usepackage[margin=.8in,left=.8in]{geometry}
\usepackage{amsthm}
\usepackage{amssymb}
\usepackage{amsmath}

\usepackage{adjustbox}

\usepackage{mathtools}
\usepackage{xcolor}
\usepackage{colortbl}    
\usepackage{stmaryrd}    
\usepackage[safe]{tipa} 

\usepackage{hyperref}
\hypersetup{
  colorlinks = true,
 linkcolor= blue, citecolor=blue           
}

\hypersetup{
        colorlinks=true,
        linkcolor=blue,
        filecolor=blue,      
        urlcolor=blue,
        citecolor=blue,
        linktoc=page
    }

\let\PLAINthebibliography\thebibliography
\renewcommand\thebibliography[1]{
  \PLAINthebibliography{#1}
  \setlength{\parskip}{0.5pt}
  \setlength{\itemsep}{0.5pt plus .3ex}
}

\newcommand{\CartesianSpaces}{\mathrm{CartSp}}

\newcommand{\FrechetManifolds}{\mathrm{FrMfd}}

\newcommand{\DiffeologicalSpaces}{\mathrm{DflSp}}

\newcommand{\SmoothManifolds}{\mathrm{SmthMfd}}

\newcommand{\SmoothSets}{\mathrm{SmthSet}}

\newcommand{\SuperSmoothSets}{\mathrm{Sup}\SmoothSets}

\newcommand{\FormalSmoothSets}{\mathrm{Th}\SmoothSets}

\newcommand{\ThickenedCartesianSpaces}{\mathrm{Th}\CartesianSpaces}

\newcommand{\SuperCartesianSpaces}{\mathrm{SupCartSp}}

\newcommand{\ev}
{\mathrm{ev}}

\newcommand{\loc}
{\mathrm{loc}}

\newcommand{\var}
{\mathrm{var}}
\usepackage{enumitem} \setlist{nosep} 

\usepackage{tikz}
\usetikzlibrary{cd}

\DeclareMathAlphabet{\mathpzc}{OT1}{pzc}{m}{it} 

\newcommand\mathscr[1]{\scalebox{1.1}{$\mathpzc{#1}$}}

\definecolor{darkblue}{rgb}{0.05,0.25,0.65}
\definecolor{darkgreen}{RGB}{20,140,10}
\definecolor{lightgray}{rgb}{0.9,0.9,0.9}
\definecolor{darkorange}{RGB}{200,100,5}
\definecolor{darkyellow}{rgb}{.91,.91,0}
\definecolor{lightolive}{RGB}{189,183,107}

\definecolor{greenii}{RGB}{20,140,10}
\definecolor{orangeii}{RGB}{200,100,5}

\usepackage{multirow}

\usepackage{enumerate} 

\makeatletter
\newcommand\makebig[2]{%
  \@xp\newcommand\@xp*\csname#1\endcsname{\bBigg@{#2}}%
  \@xp\newcommand\@xp*\csname#1l\endcsname{\@xp\mathopen\csname#1\endcsname}%
  \@xp\newcommand\@xp*\csname#1r\endcsname{\@xp\mathclose\csname#1\endcsname}%
}
\makeatother

\makebig{biggg} {3.0}
\makebig{Biggg} {3.5}
\makebig{bigggg}{4.0}
\makebig{Bigggg}{4.5}

\usepackage{mathptmx}
\usepackage{graphicx}
\DeclareRobustCommand{\coprod}{\mathop{\text{\fakecoprod}}}
\newcommand{\fakecoprod}{%
  \sbox0{$\prod$}%
  \smash{\raisebox{\dimexpr.9625\depth-\dp0}{\scalebox{1}[-1]{$\prod$}}}%
  \vphantom{$\prod$}%
}

\DeclareFontFamily{OMX}{MnSymbolE}{}
\DeclareSymbolFont{mnomx}{OMX}{MnSymbolE}{m}{n}
\SetSymbolFont{mnomx}{bold}{OMX}{MnSymbolE}{b}{n}
\DeclareFontShape{OMX}{MnSymbolE}{m}{n}{
    <-6>  MnSymbolE5
   <6-7>  MnSymbolE6
   <7-8>  MnSymbolE7
   <8-9>  MnSymbolE8
   <9-10> MnSymbolE9
  <10-12> MnSymbolE10
  <12->   MnSymbolE12}{}

\usepackage[new]{old-arrows}   

\usepackage{cleveref}

\crefformat{section}{\S#2#1#3} 
\crefformat{subsection}{\S#2#1#3}
\crefformat{subsubsection}{\S#2#1#3}

\theoremstyle{italics}
\newtheorem{theorem}{Theorem}[section]
\newtheorem{lemma}[theorem]{Lemma}
\newtheorem{proposition}[theorem]{Proposition}
\newtheorem{corollary}[theorem]{Corollary}
\theoremstyle{definition}
\newtheorem{definition}[theorem]{Definition}

\newtheorem{example}[theorem]{Example}

\newtheorem{remark}[theorem]{Remark}

\newtheorem{literature}[theorem]{Literature}

\renewcommand{\emph}{\textit}

\usepackage{latexsym}
\usepackage[all]{xy}
\usepackage{color}

\usepackage{bbm} 

\newcommand{\id}{\mathrm{id}}   			

\newcommand{\CB}{\mathcal{B}}

\newcommand{\CC}{\mathcal{C}}

\newcommand{\CD}{\mathcal{D}}

\newcommand{\CF}{\mathcal{F}}

\newcommand{\CG}{\mathcal{G}}

\newcommand{\CH}{\mathcal{H}}

\newcommand{\CI}{\mathcal{I}}
\newcommand{\CJ}{\mathcal{J}}
\newcommand{\CK}{\mathcal{K}}

\newcommand{\CL}{\mathcal{L}}

\newcommand{\CN}{\mathcal{N}}
\newcommand{\CO}{\mathcal{O}}

\newcommand{\CP}{\mathcal{P}}

\newcommand{\CR}{\mathcal{R}}

\newcommand{\CT}{\mathcal{T}}

\newcommand{\CX}{\mathcal{X}}

\newcommand{\CZ}{\mathcal{Z}}

\newcommand{\CE}{\mathcal{E}}
\newcommand{\frg}{\mathfrak{g}}				

\newcommand{\FR}{\mathbbm{R}}     			
\newcommand{\NN}{\mathbbm{N}}     			
\newcommand{\DD}{\mathbbm{D}}     			

\newcommand{\dd}{\mathrm{d}}     			

\newcommand{\pr}{\mathsf{pr}}     			

\newcommand{\comment}[1]{}     				
\def\tyng(#1){\hbox{\tiny$\yng(#1)$}}			
\def\tyoung(#1){\hbox{\tiny$\young(#1)$}}			

\newcommand{\beq}{\begin{eqnarray}}
	\newcommand{\eeq}{\end{eqnarray}}

\definecolor{outrageousorange}{rgb}{1.0, 0.43, 0.29}

\newcommand{\om}{\omega}
\newcommand{\epsi}{\epsilon}
\newcommand{\nn}{\nonumber}

\def\R{{R}}

\usepackage{eulervm}

\usepackage{tocloft}

\begin{document}

\title{\vspace{-1cm} \Large\bf 
Field Theory via Higher 
Geometry I:  \,
Smooth Sets of Fields
}
\author{
Grigorios Giotopoulos${}^\ast$  
\qquad 
Hisham Sati${}^{\ast \, \dagger}$  
}
\date{}

\maketitle

\begin{abstract}
Most modern theoretical considerations of the physical world suggest that nature is at a minimum:
(1) field-theoretic,
(2) smooth,
(3) local, 
(4) gauged,
(5) containing fermions,
and last but not least:
(6) non-perturbative.
Tautologous as this may sound to experts of the field, it is remarkable that the mathematical notion of geometry which reflects {\it all} of these aspects -- namely, 
as we will explain: ``{\it supergeometric homotopy theory}'' -- has received little attention even by mathematicians and remains unknown to most 
physicists. Elaborate algebraic machinery is known for {\it perturbative} field theories both at the classical and quantum level, but in order to tackle the deep open questions of the subject, 
these will need to be lifted to a global geometry of physics. Prior to considering any notion of non-perturbative quantization procedure, by necessity, this must first be accomplished at the classical and pre-quantum level.

\smallskip 
Our aim in this series is, first, to introduce inclined physicists to this theory, second to fill mathematical gaps in the existing literature, 
and finally to rigorously develop the full power of supergeometric homotopy theory and apply it to the analysis of fermionic (not {\it necessarily} 
super-symmetric) field theories. Secondarily, this will also lead to a streamlined and rigorous perspective of the type that we hope would also be desirable to mathematicians.

\smallskip 
In this first part, we explain how classical bosonic Lagrangian field theory (variational Euler-Lagrange theory) finds a natural home
in the ``topos of smooth sets'', thereby neatly setting the scene for the higher supergeometry discussed in later parts of the series.
This introductory material will be largely known to a few experts but has never been comprehensively laid out before. 
A key technical point we make is to regard jet bundle geometry systematically in smooth sets instead of just its subcategories 
of diffeological spaces or even Fr{\'e}chet manifolds -- or worse simply as a formal object. Besides being more transparent and powerful, it is only on this backdrop that 
a reasonable supergeometric jet geometry exists, needed for satisfactory discussion of any field theory with fermions.

\end{abstract}

\vspace{-.7cm}

 \begin{center}
 \begin{minipage}{12cm}
 {\small 
   \tableofcontents
 }
 \end{minipage}
\end{center}

\vfill

\hrule
\vspace{2pt}

{
\footnotesize
\noindent
\def\arraystretch{1}
\tabcolsep=0pt
\begin{tabular}{ll}
${}^*$\,
&
Mathematics, Division of Science; and
\\
&
Center for Quantum and Topological Systems,
\\
&
NYUAD Research Institute,
\\
&
New York University Abu Dhabi, UAE.  
\end{tabular}
\hfill
\raisebox{-10pt}{
\includegraphics[width=3cm]{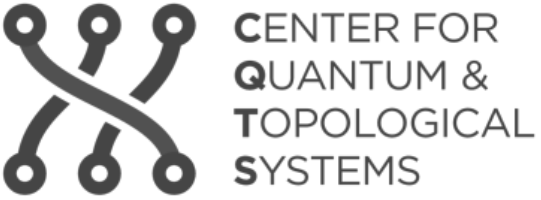}
}

\vspace{1mm} 
\noindent ${}^\dagger$The Courant Institute for Mathematical Sciences, NYU, NY

\vspace{.1cm}

\noindent
The authors acknowledge the support by {\it Tamkeen} under the 
{\it NYU Abu Dhabi Research Institute grant} {\tt CG008}.
}

\newpage

\section{Introduction} 

\noindent
{\bf The open problem.} The big open problem of contemporary fundamental physics is the formulation of {\it non-}perturbative 
strongly coupled systems (e.g. \cite{BakulevShirkov10}\cite[\S 4.1]{HollandsWald15}), which at the classical and pre-quantum level largely means the 
formulation of {\it global} or {\it integrated} and in any case {\it topologically nontrivial} structures, 
but still equipped with analytic/differential hence
{\it smooth structure}.

\smallskip 
However, much or even most of the existing literature tends to focus on perturbative, namely infinitesimal (aka ``formal'') approximations 
to the full global geometries, where everything can still be described by {\it algebra} alone. At this algebraic level, it is 
fairly straightforward to add in all the relevant bells and whistles, which leads to the consideration of ``BV-BRST complexes'' 
given by differential super/graded algebras (see \cite{FR12}\cite{Li17}\cite{Mnev17}\cite{CSCS23}). Even in their sophisticated 
formulations, invoking homotopy Lie (``$L_\infty$'') and algebr{\it oid}-structure (see e.g. \cite{CostelloWilliam1}\cite{CostelloWilliam2}\cite{JRSW19}), 
these constructions see just an infinitesimal slice in the full configuration space of physical 
fields, enough to do perturbation theory but generally failing to see the core of the physical system. 

\smallskip 
The core example here 
is quantum chromodynamics (QCD, Yang-Mills theory coupled to fermionic quark matter), where perturbation theory is 
of no value for deriving its ordinary hadronic bound states that constitute the tangible world around us. Due to the lack of a comprehensive 
mathematical formulation for non-perturbative QCD -- pronounced as a mathematical `'Millennium Problem'' by the 
Clay Math Institute (cf. \cite[\S 9.1]{Shnir05}\cite{Chatterjee19}\cite{RobertsSchmidt20}), particle physicists currently mostly abandon any attempt to 
understand the physics analytically and 
instead resort to computer experiments (``lattice QCD''). However, a non-perturbative formulation 
of quantum QCD ought to exist, involving the {\it global geometric} and topologically non-trivial structures of hadrodynamics, such as
its WZW-terms, instantons, and Skyrmions.  Assuming that this ought to be achieved by a rigorous quantization process, it is necessary to first construct a rigorous and non-perturbative formulation of the classical theory.

\smallskip 
Here we are concerned with laying out missing modern mathematical foundations towards this goal.
Throughout this text and the rest of the series, we will intend to develop and formalize the theory rigorously at the abstract level 
and yet, in parallel, show in detail how this implies the usual formulas and manipulations of the physics literature.
While our main goal is not to dwell on a long list of examples, we have chosen to illustrate the various 
concepts via particularly simple examples.

\medskip

\noindent
{\bf The geometry of physics.}
The modern mathematical machinery that serves to reflect the above notions -- even while it may 
superficially seem unfamiliar or even esoteric -- is actually quite close to physical intuition. We promise that the inclined physicist
who bears with us just a little will discover with us the most gratifying calculus for modern physics.

\medskip 
To start with, while it is {\it topos theory} (e.g. \cite{MacLaneMoerdijk94}\cite{Johnstone02}\cite{Nesin}\footnote{ We only assume that the reader has basic knowledge of the basic concepts of category theory -- categories, functors, (co)limits --, such as provided in the exposition \cite[\S 2]{Geroch85} which is aimed at mathematical physicists. An exposition of further details may be found in lecture notes such as \cite{Awodey06}\cite{Nesin}. }) that underlies the conceptualization of all of the above aspects of physical geometry,
we rush to share a well-kept secret: The core idea of topos-theoretic geometry is profoundly {\it physical} and secretly known in many 
aspects to physicists, even if under different names.  
Namely, the core idea of topos-theoretic geometry is that in order {\it to know a geometric space} -- such as to know the spacetime of the 
observable universe  -- is ({\it not} to define a set of points equipped with a long list of extra structures, etc. but) 
{\it to probe it} --  quite in the sense of experimental physics, and even quite in the sense of string theory where the 
universe appears as whatever it is that {\it probe-branes} see that traverse it.
In particular, the {\it smooth sets} in our title are just the spaces probe-able by smooth coordinate charts $\FR^k$; we start 
in \cref{IntuitiveSmoothSetsSection} with an exposition of this powerful perspective. 

\smallskip 
In subsequent parts of the series, it will be pleasing to see how a simple variation of the collection of probes (such as from Cartesian
spaces $\FR^k$ to Cartesian {\it super}spaces $\FR^{k\vert q}$) bootstraps with ease the richest notions of physical geometry for us.

\medskip

\noindent
{\bf The power of smooth sets.}
The point here is that the geometry of {\it field theory} is all about {\it smooth mapping spaces} -- the field spaces -- which are 
outside the scope of traditional differential geometry. Under the (awkward) assumption that spacetime is compact, these have been 
traditionally been discussed with their structure of infinite-dimensional {\it Fr{\'e}chet manifolds}; As will be made clear in this text, this approach carries heavy analytical buggage which is practically unnecessary for most field theoretic purposes. More recently their structure 
as {\it diffeological spaces} has gained more attention, but many constructions remain 
cumbersome with these restricted tools, and most importantly do not naturally generalize to the fermionic setting.
Nevertheless, all of these approaches are {\it faithfully} subsumed in the topos of smooth sets (see Lit. \ref{LiteratureSmoothSpaces}), where field-theoretic constructions 
find a more natural home, as we shall see:
\vspace{-2mm}
$$
  \begin{tikzcd}[
    row sep=1pt,
    column sep=20pt
  ]
    \mathclap{
    \scalebox{.7}{
      \bf
      \color{darkblue}
      \bf
      \def\arraystretch{.9}
      \begin{tabular}{c}
        Cartesian 
        \\
        spaces
      \end{tabular}
    }
    }
    &&
    \mathclap{
    \scalebox{.7}{
      \bf
      \color{darkblue}
      \bf
      \def\arraystretch{.9}
      \begin{tabular}{c}
        smooth
        \\
        manifolds
      \end{tabular}
    }
    }
    &&
    \mathclap{
    \scalebox{.7}{
      \bf
      \color{darkblue}
      \bf
      \def\arraystretch{.9}
      \begin{tabular}{c}
        Fr{\'e}chet
        \\
        manifolds
      \end{tabular}
    }
    }
    &&
    \mathclap{
    \scalebox{.7}{
      \bf
      \color{darkblue}
      \bf
      \def\arraystretch{.9}
      \begin{tabular}{c}
        diffeological
        \\
        spaces
      \end{tabular}
    }
    }
    &&
    \mathclap{
    \scalebox{.7}{
      \bf
      \color{darkorange}
      \bf
      \def\arraystretch{.9}
      \begin{tabular}{c}
        smooth
        \\
        sets
      \end{tabular}
    }
    }
    \\
    \CartesianSpaces
    \ar[
      rr,
      hook
    ]
    &&
    \SmoothManifolds
    \ar[
      rr,
      hook
    ]
    &&
    \FrechetManifolds
    \ar[
      rr,
      hook
    ]
    &&
    \DiffeologicalSpaces
    \ar[
      rr,
      hook
    ]
    &&
    \SmoothSets
    \\[+4pt]
    \ar[
      rrrrrrrr,
      hook,
      "{
      \scalebox{.7}{\color{gray}
        fully faithful inclusions
        of ever more general 
        categories of smooth spaces 
      }
      }"{description}
    ]
    &&&&&&&&
    {}
  \end{tikzcd}
$$

\newpage

\noindent
{\bf Smooth sets for Lagrangian field theory.}
A good geometry of field theory should make the differential geometry of the field space come out
formally just the way one would expect if it {\it were} a finite-dimensional smooth manifold. For example,

\begin{itemize}[leftmargin=.5cm]
\item  the (off-shell) {\it field space} (Def. \ref{SectionsSmoothSet}) should be a {\it smooth space} of smooth sections of a smooth {\it field bundle} $F$ over spacetime $M$:
\vspace{-1mm} 
\begin{equation}
  \label{FieldSpaceInIntroduction}
  \scalebox{.7}{
    \color{darkblue}
    \bf
    Field space
  }
  \;\;
  \mathcal{F}
  \;=\;
  \Gamma_M( F )
  \;\;
  \scalebox{.7}{
    \color{darkblue}
    \bf
    \def\arraystretch{.85}
    \begin{tabular}{c}
      Smooth space of
      \\
      smooth sections of
      \\
      smooth field bundle
    \end{tabular}
  }
\end{equation}
\vspace{-1.5mm}

\item a {\it Lagrangian density} $\CL$ should be a smooth map \eqref{NonLocalSmoothLagrangian} from fields to smooth differential forms on spacetime:
\vspace{-7mm} 
\begin{equation}
  \label{LagrangianDensityInIntroduction}
  \begin{tikzcd}[
    row sep=-4pt, column sep=large
  ]
    \ar[
      rr,
      phantom,
      "\raisebox{-20pt}{{
        \scalebox{.7}{
          \color{darkgreen}
          \bf
          Lagrangian density \phantom{aaaa}
        }
      }}"
    ]
    &&
    {}
    \\
    \mathllap{
      \scalebox{.7}{
        \color{darkblue}
        \bf
        Field space
      }    
      \;\;
    }
    \mathcal{F}
    \ar[
      rr,
      "{ \mathcal{L} }"
      {swap}
    ]
    &&
   \Omega^{d}(M)
    \mathrlap{
      \scalebox{.7}{
        \color{darkblue}
        \bf
        \def\arraystretch{.9}
        \begin{tabular}{c}
          Differential forms
          \\
          on spacetime
        \end{tabular}
      }    
      \;\;
      }
    \\
    \phi &\longmapsto& \mathcal{L}(\phi)
  \end{tikzcd}
\end{equation}

\item
and a {\it local Lagrangian density} should smoothly factor (Lem. \ref{LocalLagrangianisSmooth}) through the smooth jet bundle via 
a smooth {\it jet prolongation} $j^\infty$ \eqref{smoothjetprolongation}
\vspace{-4mm} 
\begin{equation}
  \begin{tikzcd}[column sep=huge]
    \mathllap{
      \scalebox{.7}{
        \color{darkblue}
        \bf
        Field space
      }
      \;\;\;
    }
    \mathcal{F}
    \ar[
      rr,
      "{
        j^\infty
      }",
      "{
        \scalebox{.7}{
          \color{darkgreen}
          \bf
          extract field derivatives
        }
      }"{yshift=9pt}
    ]
    \ar[  
      rrrr,
      bend right=12pt,
      "{ \mathcal{L} }"
      {description}
    ]
    &&
    \Gamma_M
    \big(
      J^\infty_M
      F
    \big)
    \ar[
      rr,
      "{
        \scalebox{.7}{
          \color{darkgreen}
          \bf
        Lagrangian density bundle map
        }
      }"{yshift=8pt}
    ]
    &&
    \Omega^{d}(M)
    \mathrlap{
      \scalebox{.7}{
        \color{darkblue}
        \bf
        \def\arraystretch{.9}
        \begin{tabular}{c}
          Differential forms
          \\
          on spacetime
        \end{tabular}
      }
    }
  \end{tikzcd}
\end{equation}

\vspace{-1mm} 
\item
The integral\footnote{This assumes that $M$ is compact, or that the fields have suitable support. This need not be the case, 
and the general criticality condition will be described in \cref{OnShellFieldsCriticalSetSubsection}.} of the Lagrangian
density should give a smooth function \eqref{LocalActionisSmooth} on the smooth field space:
\vspace{-4mm}
\begin{equation}
  \label{ActionFunctionalInIntroduction}
  \begin{tikzcd}[
    column sep=60pt,
    row sep=4pt
  ]
    &&
    {}
    \\[-3pt]
    \mathllap{
      \scalebox{.7}{
        \color{darkblue}
        \bf
        Field space
      }
      \;\;\;\;
    }
    \mathcal{F}
    \ar[
      rr,
      "{
        \scalebox{.7}{
          \color{darkgreen}
          \bf
          action functional
        }
      }",
      "{S:\;\;
        \phi 
        \;\;\mapsto\;\;
        \int_M 
          \mathcal{L}(\phi)
      }"{swap}
    ]
    &&
    \FR
    \mathrlap{
      \scalebox{.7}{
        \color{darkblue}
        \bf
        \def\arraystretch{.9}
        \begin{tabular}{c}
          Smooth space of
          \\
          real numbers
        \end{tabular}
      }
      \;\;\;\;
    }
  \end{tikzcd}
\end{equation}

\vspace{-2mm} 
\noindent so that this {\it action functional} should just be a smooth {\it action function} now!

\item Now for every smooth 1-parameter family of fields
\vspace{-2mm} 
$$
  \phi_t \;\colon\, 
  \begin{tikzcd}[
    column sep=50pt
  ]
    \FR^1 
    \ar[
      rr,
      "{
        \scalebox{.7}{
          \color{darkgreen}
          \bf
          smooth curve in field space
        }
      }"{yshift=0pt}
    ]
    &&
    \mathcal{F} \;,
  \end{tikzcd}
$$

\vspace{-2mm} 
\noindent there should be the corresponding 
naive differentiation of the action function(al) which represents its classical {\it variation}, which in the case of a \textit{local} action is proportional 
to the smooth {\it Euler-Lagrange differential operator} applied to the fields and evaluated on the corresponding \textit{tangent} vector of the fields
(see \cref{OnShellFieldsCriticalSetSubsection}, especially Prop. \ref{Crit(S)functoriality}):
\vspace{.6cm}
\begin{equation}
  \label{VariationalPrincipleInIntroduction}
  \overset{
    \mathclap{
      \hspace{-9pt}
      \rotatebox{+45}{
        \rlap{
          \hspace{-8pt}
          \scalebox{.7}{
            \color{darkblue}
            \bf
            derivative of
          }  
        }
      }
    }
  }{
    \,\partial_t\,
  }
  \overset{
    \mathclap{
      \hspace{-1.5pt}
      \rotatebox{+45}{
        \rlap{
          \hspace{-8pt}
          \scalebox{.7}{
            \color{darkblue}
            \bf
            action functional
          }
        }
      }
    }
  }{
  \,S\,
  }
  (
  \overset{
    \mathclap{
      \hspace{-1.5pt}
      \rotatebox{+45}{
        \rlap{
          \hspace{-9pt}
          \scalebox{.7}{
            \color{darkblue}
            \bf
            along field family
          }
        }
      }
    }
  }{
  \phi_t
  }
  )|_{ t = 0 }
  \overset{
    \raisebox{-10pt}{
      \scalebox{.7}{
        \color{darkgreen}
        \bf
        \def\arraystretch{.9}
        \begin{tabular}{c}
          smooth 
          \\
          variational
          \\
          principle
        \end{tabular}
      }
    }
  }{
    \;\;=\;\;
  }
  \int_M
  \big\langle
  \overset{
    \mathclap{
      \hspace{+8pt}
      \rotatebox{+45}{
        \rlap{
          \hspace{-18pt}
          \scalebox{.7}{
            \color{darkblue}
            \bf
            \def\arraystretch{.9}
            \begin{tabular}{c}
              Euler-Lagrange 
              \\
              equations
            \end{tabular}
          }
        }
      }
    }
  }{
    \mathcal{EL}
  }
  (\phi_0)
    ,\,
  \overset{
    \mathclap{
      \hspace{+5pt}
      \rotatebox{+45}{
        \rlap{
          \hspace{-18pt}
          \scalebox{.7}{
            \color{darkblue}
            \bf
            \def\arraystretch{.9}
            \begin{tabular}{c}
              evaluated on the
              \\
              field variation
            \end{tabular}
          }
        }
      }
    }
  }{
    \partial_t \phi_t 
  }
   |_{t = 0}
  \big\rangle \, ,
\end{equation}

\vspace{-1mm} 
\noindent where the Euler-Lagrange operator $\mathcal{EL}$ is a smooth section of the {\it variational cotangent bundle}
(see Def. \ref{LocalCotangentBundleToFieldSpace}):
\vspace{-2mm} 
\begin{equation}
  \label{VariationalCotangentBundleInIntroduction}
  \begin{tikzcd}[
    column sep=huge,
    row sep=20pt
  ]
    & 
   T_\mathrm{var}^*\CF
    \ar[
      d,
      "{
        \scalebox{.7}{
          \color{darkblue}
          \bf
          \def\arraystretch{.9}
          \begin{tabular}{c}
            variational 
            \\
            cotangent bundle
            \\
            of field space
          \end{tabular}
        }
      }"
    ]
    \\
    \mathcal{F}
    \ar[r, equals]
    \ar[
      ur,
      dashed,
      "{
        \mathcal{EL}
      }"
    ]
    &
    \mathcal{F}
    \mathrlap{\, ,}
  \end{tikzcd}
\end{equation}

\vspace{-2mm} 
\noindent where  $T_\mathrm{var}^*\CF :=\Gamma_M\big(
      \wedge^d T^* M \otimes V^* F
    \big)$. This is the smooth geometric incarnation of the {\it equations of motion} of the field theory.

\item
Finally, the {\it solutions to the equations of motion} should be a smooth space of {\it on-shell fields}, obtained simply as the (``critical'') locus 
in field space where the Euler-Lagrange operator vanishes, hence it should be a fiber product / pullback in smooth spaces of this form
(Cor. \ref{CriticalSmoothSetPullback}):
\vspace{-4mm} 
\begin{equation}
  \label{CriticalLocusInIntroduction}
  \begin{tikzcd}[row sep=small,
    column sep=80pt
  ]
    &
    \overset{
      \mathclap{
        \raisebox{1pt}{
          \scalebox{.7}{
            \color{darkblue}
            \bf
            on-shell field space
          }
        }
      }
    }{
      \mathrm{Crit}(S)
    }
    \ar[dl]
    \ar[dr]
    \ar[
      d,
      phantom,
      "{
        \scalebox{.7}{
          \color{darkgreen}
          \bf
          inside field space
        }
      }"{pos=.65}
    ]
    \\    
    \CF
    \ar[
      dr,
      "{
        \mathcal{EL}
      }"{pos=.4},
      "{
        {
          \scalebox{.7}{
            \color{darkgreen}
            \bf
            \def\arraystretch{.9}
            \begin{tabular}{c}
              where the 
              \\
              Euler-Lagrange
              \\
              operator
            \end{tabular}
          }
        }      
      }"{swap, pos=.8, yshift=3pt}
    ]
    &
    {}
    &
    \CF   \mathrlap{\,,}
    \ar[
      dl,
      "{
      0_\CF
      }"{swap, pos=0.4},
      "{
      {
              \;\;
          \scalebox{.7}{
            \color{darkgreen}
            \bf
            vanishes
          }
        }
      }"{pos=.7}
    ]
    \\
    &
    T^*_\mathrm{var} \CF
  \end{tikzcd}
\end{equation}

\end{itemize}

\vspace{-1mm} 
Our {\bf main result} in this mainly expository first part is the general demonstration that -- when carried 
 out in $\SmoothSets$ --  these desiderata are naturally, conveniently, and rigorously achieved. In the same vein, we explain in detail how a much 
 larger class of field theoretic concepts find their natural home in $\SmoothSets$. In subsequent articles, 
 we enhance this same story to richer notions of generalized smooth spaces, such as including (anticommuting) fermionic fields.

\medskip

\newpage 

\noindent
{\bf Summary and results.} Alongside presenting a comprehensive account of bosonic classical field theory, 
we fill several technical gaps that have remained in the literature (precise citations will be given throughout the text).
Throughout, we will illustrate many of the following concepts using the simple running examples of vector-valued field theory, 
pure electromagnetism, and particle mechanics, although
of course our formulation applies in much greater generality to accommodate a vast array of examples, as appearing in Literature \ref{lit-cft} below. 
\begin{itemize}
\setlength{\itemsep}{.5mm}
\item In \cref{IntuitiveSmoothSetsSection}, we incrementally build up intuition towards the definition of smooth sets and indicate how the Yoneda 
Lemma comes into play as a consistency requirement. Then we indicate how this naturally allows for a precise  formulation of the geometry of 
field theory as described above, while leaving the full details for the following sections. 

\item In \cref{SmoothSetsSection},
we describe the canonical smooth set structure on any mapping space between two smooth sets, that is via the internal hom functor of smooth sets. 
We then show how the (extended) Yoneda embedding preserves the Fr\'{e}chet mapping space between two manifolds, by combining 
it with the Exponential Law property of the latter. We employ the smooth mapping space construction to define the smooth space of sections of 
a fiber bundle -- and hence the full non-perturbative space of fields -- as a smooth set, and use this to motivate the 
definition of tangent vectors, smooth tangent bundle and smooth vector fields on the field space. We then describe how vector 
fields on field space are interpreted as ``smooth infinitesimal symmetries" via smooth 1-parameter groups of diffeomorphisms of the field space. 

\item 
We close off \cref{Sec-bos} in \cref{ClassifyingSpaceSection} by defining a non-concrete smooth set, the moduli space of smooth $n$-forms, 
and show how it can be used to define an alternative notion 
of forms on an arbitrary smooth set. Although we do not use it explicitly here (apart from Rem. \ref{CriticalityViaModuliSpaceOf1-forms}), 
we point out some of its uses in the mathematics and physics literature. Its relation to the differential forms as maps out of the (synthetic) tangent 
bundle, and in particular that of field spaces as introduced in \cref{SmoothSetsSection}, will be fully elaborated in \cite{GSS-2}).

\item We set up local Lagrangians in \cref{Sec-locLag}. 
We start in \cref{JetBundleSection} by recalling the description of the infinite jet bundle as an infinite-dimensional Fr\'{e}chet manifold -- a 
projective limit of finite-dimensional manifolds in the category $\FrechetManifolds$ (a ``locally pro-manifold"). 
We extend previous results on the description of smooth functions 
out of the infinite jet bundle to values in any finite-dimensional manifold (and any other locally pro-manifold). We then use the embedding 
of Fr\'{e}chet manifolds into smooth sets to explain how the infinite jet prolongation map is naturally viewed as a \textit{smooth} map within 
this category. This allows for the reformulation of the standard description of Lagrangian densities, currents, and charges on the field space, 
as pullbacks of smooth (bundle) maps defined on the infinite jet bundle, via the infinite jet prolongation, and hence as smooth maps out of
the field space. 

\item We close the section by defining in \cref{Sec-Currents} the appropriate notion of a (smooth) symmetry of a local Lagrangian field theory.
This, and the following sections, will make clear that the only infinite-dimensional manifold/analytic properties necessary 
for the description of (off-shell) smooth, local field theory are:

\begin{itemize} 
\item[{\bf (i)}] the \textit{paracompact} manifold structure on $J^\infty_M F$ and 

\item[{\bf (ii)}] the characterization of smooth maps out of $J^\infty_M F$ as maps that \textit{locally factor} through 
finite-order jet bundles. 
\end{itemize} 
\noindent The remaining properties employed are purely universal categorical properties of the smooth set incarnation of $J^\infty_M F$.  

\item  \cref{DifferentialGeometryOnTheInfiniteJetBundle} is essentially a recasting of results on the differential geometry of the infinite 
jet bundle, but now comprehensively spelled out in smooth sets. Having embedded the infinite jet bundle in smooth sets, in \cref{InfinityJetTangentBundleSection} 
we define its tangent bundle as a limit of the finite order tangent bundles -- computed directly in smooth sets. We describe how this recovers 
the (several equivalent) definitions and coordinate descriptions of tangent vectors and vector fields in the 
existing literature, by intuitively unraveling their definitions internal to smooth sets. 

\item In  \cref{HorizontalSplittingSection}, we define the smooth vertical subbundle of $J^\infty_M F$ --  and hence smooth vertical 
vector fields --  directly in smooth sets, recovering the standard notions used in the literature. We then recall the description 
of the point-set horizontal splitting of the tangent bundle of $J^\infty_M F$ over $M$, and prove how this is 
actually a \textit{smooth} splitting when correctly interpreted in the category of smooth sets.

\item In \cref{DifferentialFormsOnJetBundleSection}, we define forms on the infinite jet bundle as smooth $\FR$-valued, 
fiberwise-linear maps out of its tangent bundle and detail explicitly how this recovers the traditional definitions and coordinate descriptions. 
We close the section by showing how differential forms of \textit{globally} finite order may be equivalently viewed as de Rham forms, i.e., 
as maps from $J^\infty_M F$ into the classifying space of de Rham forms introduced in \cref{ClassifyingSpaceSection}.

\item In \cref{VariationalBicomplexSection}, we use the smooth splitting of \cref{HorizontalSplittingSection} to define a bi-complex 
structure on differential forms on $J^\infty_M F$, and hence recover the variational bicomplex. We then recall the standard definitions 
of source forms, interior Euler operator, functional forms, and their properties along with well-known results on the cohomology of the 
(augmented) variational bicomplex and the Euler--Lagrange complex. We describe how the vertical differential and the Euler operator encode
the `integration by parts' algorithm (at the level of the jet bundle) via the Euler--Lagrange source form of a Lagrangian density.

\item  In \cref{{sec-onshell}}, we define the ``shell" of a Lagrangian density on $J^\infty_M F$ as the (smooth) subset where 
the Euler--Lagrange source form vanishes. Similarly, having explicitly shown how every $(d,1)$-source form defines a (smooth) 
map out of the field space, we define the ``on-shell"  space of fields as the (smooth) subset of fields whose plots vanish under 
the induced Euler--Lagrange operator. We then define the prolongated version of the shell, as the appropriate incarnation of 
the on-shell fields inside the jet bundle.
We describe how one may pull back forms on $J^\infty_M F$ to spacetime forms via the infinite jet prolongation of any field, and show 
how this is compatible with the horizontal differential on $J^\infty_M F$ and the usual de Rham differential on the spacetime $M$. 
This allows for a version of Stokes' Theorem on field space, by integrating pullbacks of horizontal $p$-forms along compact oriented 
submanifolds, leading to the notion of (off-shell and on-shell) conserved $p$-form currents and charges on field space. 
We close the section by showing how a (finite) local symmetry of a Lagrangian field theory preserves the smooth 
subspace of on-shell fields.

\item In \cref{OnShellFieldsCriticalSetSubsection}, we rigorously define the set of critical points of
an arbitrary real-valued smooth map on a smooth set, and furthermore show how this naturally generalizes to a definition 
of arbitrary critical 
$\FR^k$-plots. We then prove that this assignment indeed defines a smooth set in the case of local Lagrangian field theories 
over compact spacetimes, which furthermore coincides with the smooth space of on-shell fields. Subsequently, we extend this 
description to field theories over non-compact spacetimes, whereby the action functional is ill-defined over the full field space, 
with the criticality condition being replaced by an appropriate criticality condition for the Lagrangian itself. 
After remarking on the special treatment required for field theories over spacetimes with \textit{boundary}, 
we explain in detail the relation of the criticality condition with the moduli space of 1-forms of \cref{SmoothSetsSection}.
We conclude by employing the criticality description to show that an arbitrary spacetime covariant symmetry of 
a local Lagrangian field theory preserves the smooth space of on-shell fields.

\item In \cref{EvolutionaryVectorFieldsAndNoetherTheoremsSection}, we define (smooth) \textit{local} vector fields on field space, 
along with their corresponding evolutionary vector fields and prolongated counterparts on the infinite jet bundle. 
After recalling known properties of (prolongated) evolutionary vector fields, we define \textit{infinitesimal} local symmetries of 
a local Lagrangian field theory, and show how every such symmetry induces a smooth on-shell conserved current (Noether's First Theorem). 
Along the way, we describe in detail how these notions recover the standard formulas appearing in the physics literature.

\item
We proceed in \cref{Sec-gauge}  by motivating the notion of infinitesimal local gauge symmetries by considering a particular simple 
but illustrative case, before defining the notion in full 
generality. 
We then show  how every such symmetry induces interrelations between the components of the Euler--Lagrange
operator (Noether's Second Theorem), 
and hence smooth conserved currents of a special kind -- which in particular (cohomological) situations are trivial on-shell. 

\item Lastly, in \cref{Sec-Cauchy}  we describe how one may rigorously and concisely define the notion of initial data for a local 
Lagrangian field theory on a codimension-1 submanifold of spacetime using smooth sets, together with the corresponding notion of 
a Cauchy surface. We close the section by recasting and proving the well-known fact that (finite) gauge symmetries prevent 
the existence of Cauchy surfaces.

\item In \cref{TheBicomplexOfLocalFormsSection} we study the presymplectic structure of local field theories. 
 In \cref{sec-cartan}, we define the appropriate notion of a tangent bundle and differential forms on the product of (off-shell)
 fields and the underlying spacetime, $\CF\times M$. After recalling the smooth (prolongated) evaluation map valued in the jet bundle,
 $\CF\times M\rightarrow J^\infty_M F$ we intuitively motivate and rigorously define its pushforward bundle map between the
 corresponding tangent bundle. We use the pushforward to pull back the variational bicomplex on $J^\infty_M F$ and define 
 the local bicomplex of differential forms on $\CF\times M$. We carefully show how these abstract and globally defined 
 objects recover the standard notation of local forms on $\CF\times M$, along with the actions of the corresponding differentials. 
 We then describe the local Cartan calculus satisfied by insertions of local vector fields, which lifts the evolutionary 
 Cartan calculus of the jet bundle. Finally, we show how local differential forms of \textit{globally} finite order may
 be equivalently viewed as de Rham forms, i.e., as maps from $\CF\times M$ into the classifying space of de Rham forms from \cref{ClassifyingSpaceSection}.

\item In \cref{SecPresymplecticCurrentPropertiesAndBrackets},  we describe the presymplectic current on $\CF\times M$, 
and rigorously prove several of its properties, and similarly for its restriction to the on-shell product $\CF_{\CE \CL}\times M$. 
To that end, an appropriate definition of a smooth tangent bundle to $\CF_{\CE \CL}$ is required, which we motivate intuitively 
and rigorously define. Among other results, we show that the presymplectic current is on-shell conserved and that
(infinitesimal) gauge symmetries imply its degeneracy. We use the presymplectic current to define a notion of (off-shell) 
Hamiltonian currents and define an induced `Poisson-like' bracket structure. We relate this to the bracket of Noether 
currents and hint at its higher Lie algebraic nature. Finally, we explain that the corresponding on-shell notions of Hamiltonian 
currents and their brackets follow similarly (Rem. \ref{OnshellBracketsCaveats}), but are only well-defined if the local
Cartan calculus descends to $\CF_{\CE \CL}\times M$. Indeed, we stress that the question of whether the local Cartan calculus descends
to the on-shell fields is \textit{precisely} where infinite dimensional/analytic details are again relevant, via the prolongated shell
$S_L^\infty \hookrightarrow J^\infty_M $ of Euler-Lagrange source form (Rem. \ref{OnShellCartanCalculusCaveats}).

\item 
In the final part \cref{Sec-cov}, we describe how one may transgress (`integrate') local differential 
forms on $\CF\times M$ over submanifolds to produce smooth differential forms on the actual field space $\CF$. In particular, 
we show how the presymplectic $(d-1,2)$-current transgresses to a presymplectic $2$-form on field space $\CF$, whose on-shell restriction 
depends only on the cobordism class of the transgressing submanifold, hence 
defining the covariant phase space $(\CF_{\CE \CL}, \Omega_\CL)$ as a presymplectic smooth set. By the degeneracy result of the
presymplectic current from \cref{SecPresymplecticCurrentPropertiesAndBrackets} it follows that, in the presence of gauge symmetries, 
the covariant phase space cannot be symplectic. If a Cauchy surface exists, we describe how to use the induced 
isomorphism to its initial data smooth set to produce the associated (non-covariant) phase space. Next, 
we define the algebra of (off-shell) Hamiltonian functionals on $\CF$ with respect to the symplectic $2$-form, 
describe its Poisson algebra structure,
and show how the transgression intertwines the brackets currents and functionals. Finally, we explain how the corresponding 
on-shell notion of Hamiltonian functionals and their Poisson structure follows similarly, at least in the case where the local 
Cartan calculus descends to $\CF_{\CE\CL}\times M$, and hence further transgresses to $\CF_{\CE \CL}$.  

\item We provide an outlook in \cref{outlook}  on how to carry our 
approach further to include the rigorous description of infinitesimal, fermionic, and (higher) gauge theoretic aspects of field theory. 
These will be fleshed out in full detail in the upcoming installments of this series. 

\end{itemize}

\vspace{-1mm} 
\begin{literature}[{\bf Smooth spaces}]
\label{LiteratureSmoothSpaces}
Different approaches to generalized smooth spaces have been explored in the past (see \cite{Stacey} for a survey).
More traditional definitions based on underlying topological spaces include Smith \cite{Smith}, Sikorski \cite{Sik}, and 
Mostow \cite{Mostow}. 
The definition of {\it diffeological spaces} has grown out of the work of  Chen \cite{Chen}\cite{Chen82}, 
Fr\"olicher \cite{Fro81}\cite{Fro82}, and then notably Souriau \cite{Sour} (see \cite{BH11} for exposition). 
A comprehensive discussion of differential geometry with diffeological spaces 
was developed by Iglesias-Zemmour \cite{IZ85}\cite{IZ13}. 
The more general definition of {\it smooth sets} which we use here, as general sheaves on the site $\CartesianSpaces$,
is from Schreiber \cite[Def. 1.2.197]{dcct} (their smooth mapping spaces were discussed in \cite[\S 1.2.2.4]{dcct}), 
elementary exposition is in the lecture notes \cite[\S 2]{Nesin} and further discussion in the broader context 
of higher topos theory is in \cite[Ex. 3.18]{Orbifold}\cite[(3.143)]{Bundles}; see also \cite{Schreiber24}.
\end{literature}

\vspace{-2.5mm} 
\begin{literature}[{\bf Variational bicomplex}]
Early geometric perspectives on the calculus of variations for field theory purposes are given in \cite{Trautman}\cite{Herman}\cite{Kom68}\cite{Kom69}\cite{Kuper}\cite{AA78}. 
The variational bicomplex plays an important role in this context. Detailed expositions are given in \cite{An91}\cite{Vitolo}. 
The cohomological properties of the variational bicomplex, including local exactness, are established within various 
approaches, for instance in \cite{Tul77}, then \cite{Vin84-spec} (Spencer spectral sequence), \cite{Takens79}\cite{AD80}\cite{Tul80} 
(explicitly and constructively), \cite{Tsu82} (Koszul complex), \cite{BDK}\cite{Wald}\cite{DVHTV}\cite{Dickey92} (algebraic techniques), 
\cite{BG} (characteristic cohomology), and \cite{AF97} (Lie algebra cohomology). 
\end{literature}

\vspace{-2.5mm} 
\begin{literature}[{\bf Classical Lagrangian field theory}]\label{lit-cft}
Clearly, the literature here is vast and we will not attempt to do it justice. Constructions and surveys taking various approaches 
within the differential geometric and/or mathematical physics context 
include 
\cite{Trautman}\cite{Sn}\cite{Thirring}\cite{AA80}\cite{Bleecker}\cite{Vin84}\cite{Zuckerman}\cite{BSF}\cite{Morandi}\cite{HenneauxTeitelboim92}\cite{GMM}\cite{ACDSR}\cite{GMS97}
\newline \cite{DF99}\cite{Freed}\cite{FF03}\cite{Sardanashvily09}\cite{GMS09}\cite{DR11}\cite{Frankel}\cite{CMR12}\cite{Urs}\cite{DSV15}\cite{Kr15}\cite{MH16}\cite{BFR19}\cite{Mar}.
The approach to field theory in terms toposes of sheaves over various
kinds of probes takes inspiration from the work of Schreiber (see \cite{dcct} \cite{Urs}\cite{Schreiber24}
and the early live lecture notes \cite{Schreiber}, written with a
viewpoint towards perturbative quantum field theory).
 Another approach close to ours as a rigorous formalization of field theories seems to be the unpublished (live) lecture notes by Blohmann  \cite{Blohmann23b},
which we discovered at a late stage of writing this manuscript. 
Naturally, in terms of content, there is some overlap on how certain field-theoretic concepts are being formalized.
  However, our development takes a different route, which does not emphasize an underlying set of points 
(a viewpoint which, in particular, will be necessary for the fermionic case), and hence our approach is more amenable 
to (further) physically desirable generalizations. 
Indeed, in  \cite{Blohmann23b} (and also in Ref. \cite{Del}) the perspective is to view 
\textit{bosonic} Lagrangian field theory as taking place in the 
full subcategory  (Pro-) $\DiffeologicalSpaces$ of (Pro-) smooth sets. While this is sufficient in the bosonic setting, our more general
point of view of sheaves over Cartesian spaces naturally generalizes to include infinitesimal structure and fermionic fields  
(see \cref{outlook} and \cite{GSS-2}). Another crucial technical difference 
is their treatment of the infinite jet bundle as a pro-manifold, which forces them to work in Pro$\DiffeologicalSpaces$ compared to 
ours as a Fr\'{e}chet manifold -- which thus naturally embeds 
into actual smooth sets (see \cref{JetBundleSection} and Rem. \ref{ProVsLocProJetBundle}).

A further difference is in terms of approach: Our focus and perspective throughout are on the actual smooth space of field $\CF$ 
-- most results are developed and stated directly on the field space $\CF$   without explicitly using the local bicomplex structure 
on $\CF \times M$, until strictly necessary, but rather using the variational bicomplex on $J^\infty_M F$. The perspective of 
\cite{Blohmann23b} is to first develop the local bicomplex (albeit in pro-diffeological spaces, rather than smooth sets as we do), 
and then express the corresponding field-theoretic results in terms of the corresponding cohomology. Our approach allows, for instance,
a conceptually clearer description of (finite) symmetries (Def. \ref{FiniteSymmetryofLagrangianFieldTheory},
Lem. \ref{InfinitesimalSpacetimeCovariantSymmetries}), their action on the on-shell field space  
(Prop. \ref{LocalSymmetryPreservesOnshellSpace}, Prop. \ref{SymmetryPreservesOnshellSpace}), its description as a critical locus \cref{OnShellFieldsCriticalSetSubsection}, 
and the discussion of initial data and Cauchy surfaces \cref{Sec-Cauchy}. Moreover, part of the (non-expert) mathematical 
physics community is not familiar with the hands-on geometrical 
manipulations on  $J^\infty_M F$, and further their all-important interplay with structures on $\CF$. We hope our treatment to 
be useful in bridging this gap.

We only delve into the local bicomplex and rigorously recover the description (and notation) and several results 
along the lines of \cite{Zuckerman}\cite{DF99} in the final section, whereby the local bicomplex is indeed conceptually 
(and notationally) useful for defining and discussing the presymplectic structure of local field theories. 
Indeed, compared to \cite{Blohmann23b}, our presentation moreover includes a detailed discussion of the 
off-shell presymplectic current of field theories and its induced bracket structures of Noether and Hamiltonian local currents, 
which (in good situations) restrict to the on-shell product $\CF_{\CE \CL}\times M$ -- and similarly for the corresponding 
transgressed structures on the actual field space $\CF$ and $\CF_{\CE \CL}$ upon integration along an appropriate submanifold.
Further specific technical discrepancies in terms of definitions and proofs will be noted throughout.

\end{literature}



\addtocontents{toc}{\protect\vspace{-10pt}}
\section{Bosonic field spaces as smooth sets} \label{Sec-bos}
\subsection{Intuitive approach to smooth sets}\label{IntuitiveSmoothSetsSection}

In this section, we gently motivate and then discuss the definition of {\it smooth sets} -- generalized smooth spaces that subsume smooth manifolds, 
but also for instance their smooth mapping spaces and hence the field spaces of field theories. This will pave the way to discuss these topics
more rigorously in \cref{SmoothSetsSection}.
These examples of smooth sets are also known as {\it diffeological spaces}, but beyond that also the {\it moduli spaces of differential forms} 
exist as smooth sets.

\smallskip 
Despite this powerful generality of smooth sets, their definition is actually much {\it simpler} than the typical tools of infinite-dimensional 
analysis (nevertheless they faithfully subsume for instance infinite-dimensional Fr{\'e}chet manifolds, see Prop. \ref{FrManRestrEmbedding}) 
and we want to offer the physicist reader a neat perspective on smooth sets which, while fully precise (see \cref{SmoothSetsSection}), 
is quite close to the operational idea of ``spaces'' actually used in physics.

\medskip

\noindent
{\bf Intuition for smooth sets.}
The basic idea of generalized smooth spaces
$\CG$ (like smooth sets) is {\it not} to declare them as sets of points with extra structure (as with so many traditional definitions), 
but instead to provide an operational meaning of how to explore or to {\it probe} $\CG$. This is not unlike the intuition in string theory 
that spacetime is whatever brane {\it probes} detect as they traverse it. Indeed, if we write 
\begin{itemize}
\item 
$\Sigma$ for a probe brane's worldvolume (a finite-dimensional smooth manifold),
\item $\CG$ for the would-be target space to be explored (a generalized smooth space),
\end{itemize}
then the probe brane's trajectory should be a smooth map (for some definition of ``smooth map'' to be determined)
\vspace{-2mm} 
\begin{equation}
  \label{SchematicsOfProbes}
  \begin{tikzcd}[column sep=large]
    \Sigma
    \ar[
      rr,
      "{
        \scalebox{.7}{
          \color{darkgreen}
          \bf
          trajectory
        }
      }",
      "{
        \scalebox{.7}{
          (``plot'')
        }
      }"{swap, yshift=-2pt}
    ]
    &&
    \CG
    \mathrlap{\,.}
  \end{tikzcd}
\end{equation}

\vspace{-2mm} 
\noindent The idea is that knowing about $\CG$ is equivalent to knowing about the system of these plots of $\CG$.
There is an evident minimum of structure that such a system of plots carries:
\begin{itemize}
  \item[{\bf (i)}]
    For each probe manifold $\Sigma$ there should be a set of plots
    \vspace{-2mm} 
    $$
      \Sigma \,\in\, \SmoothManifolds
      \;\;\;\;\;\;
        \Rightarrow
      \;\;\;\;\;\;
      \mathrm{Plots}(\Sigma, \CG)
      \;\in\;
      \mathrm{Set}\;.
    $$

    \vspace{-2mm}
    For instance, for $\Sigma =\{*\}$ a point the set $\mathrm{Plots}(*,\CG)$ is to be interpreted as the `set of points' in $\CG$, 
    while for $\Sigma =\FR^1$ the real line the set $\mathrm{Plots}(\FR^1,\CG)$ is to be interpreted as `the set of smooth lines' in $\CG$. 
    More generally, $\mathrm{Plots}(\FR^k,\CG)$ is `set of $\FR^k$-shaped smooth plots' in $\CG$ and analogously for any 
    probe manifold $\Sigma\in \SmoothManifolds$.
  \item[{\bf (ii)}]
    The precomposition of a plot with a smooth map between probe manifolds should again be a plot
    \vspace{-2mm} 
    $$
      \begin{tikzcd}[column sep=large]
        \Sigma'
        \ar[
          rrr,
          "{ f }",
          "{
            \scalebox{.7}{
              ordinary smooth map
            }
          }"{swap, yshift=-2pt}
        ]
        \ar[
          rrrrr,
          rounded corners,
          to path={
               ([yshift=-00pt]\tikztostart.south)
            -- ([yshift=-17pt]\tikztostart.south)
            -- node[yshift=-7pt] {
                  \scalebox{.7}{
                    composite plot
                  }  
               }
               node[yshift=6pt]{
                 \scalebox{.7}{
                   $ f^\ast p$
                 }
               }
               ([yshift=-17pt]\tikztotarget.south)
            -- ([yshift=-00pt]\tikztotarget.south)
          }
        ]
        &&&
        \Sigma
        \ar[
          rr,
          "{
            p
          }",
          "{
            \scalebox{.7}{
              plot
            }
          }"{swap, yshift=-2pt}
        ]
        &&
        \CG
      \end{tikzcd}
    $$
    
    \vspace{-2mm} 
\noindent    such that 
    \begin{itemize}
      \item[({\bf a)}]
        precomposition with an identity map is the identity operation on the set of plots,
      \item[{\bf (b)}] 
      precomposition with two successive maps is the same as precomposing with their composites
    \end{itemize}
    \begin{equation}
      \label{PreasheafConditionForSmoothSets}
      \begin{tikzcd}[row sep=small]
        \SmoothManifolds^{\mathrm{op}}
        \ar[
          rr,
          "{
            \mathrm{Plots}(-, \CG)
          }"
        ]
        &&
        \mathrm{Set}
        \\[-12pt]
        \Sigma
        \ar[
          from=dd,
          "{
            f
          }"
        ]
        &\mapsto&
        \mathrm{Plots}(\Sigma, \CG) \phantom{.}
        \ar[
          dd,
          "{
            f^\ast
          }"
        ]
        \ar[
          dddd,
          rounded corners,
          to path={
               ([xshift=+00pt]\tikztostart.east)
            -- ([xshift=+10pt]\tikztostart.east)
            -- node[sloped, rotate=180] {
                 \scalebox{.7}{\colorbox{white}{$
                   (f \circ g)^\ast
                 $}}
               }
               ([xshift=+10pt]\tikztotarget.east)
            -- ([xshift=+00pt]\tikztotarget.east)
          }
        ]
        \\
        \\
        \Sigma'\!\!
        \ar[
          from=dd,
          "{
            g
          }"
        ]
        &\mapsto&
        \!\!\!\mathrm{Plots}(\Sigma', \CG)
        \ar[
          dd,
          "{
            g^\ast
          }"
        ]
        \\
        \\
        \Sigma''\!\!\!\!\!
        \ar[
          uuuu,
          rounded corners,
          to path={
               ([xshift=-00pt]\tikztostart.west)
            -- ([xshift=-10pt]\tikztostart.west)
            -- node[sloped, rotate=180]{
                 \scalebox{.7}{\colorbox{white}{$
                   f \circ g
                 $}}
               }
               ([xshift=-10pt]\tikztotarget.west)
            -- ([xshift=-00pt]\tikztotarget.west)
          }
        ]
        &\mapsto&
        \!\!\!\!\!\!\mathrm{Plots}(\Sigma'', \CG)
      \end{tikzcd}
    \end{equation}
    Of course, this means to say that the system 
    $\mathrm{Plots}(-,\CG)$
    of plots of $\CG$ constitutes a {\it contravariant functor} from smooth manifolds to sets, also called 
    a {\it presheaf} on $\SmoothManifolds$, as indicated above.

    Note that -- at this point of bootstrapping the generalized space $\CG$ into existence --  the arrow 
    notation \eqref{SchematicsOfProbes} for a $\Sigma$-plot of $\CG$ is only schematic, since $\CG$ is only in the 
    process of being defined by these very plots. Even once it is defined, it nominally lives in a different category than 
    the probes $\Sigma$. However, the magic of this bootstrap process is that in the end the would-be-plots
    \eqref{SchematicsOfProbes} of a generalized space end up being the actual plots! This is the content of nothing but 
    the {\it Yoneda Lemma}, Prop. \ref{YonedaLemmaForSmoothSets} below.

    \item[{\bf (iii)}] If a probe manifold $\Sigma$ is {\it covered} by open subsets 
    $$
      \big\{
        \Sigma_i 
          \xhookrightarrow{\;\; \iota_i \;\;} 
        \Sigma 
      \big\}_{i \in I}
    $$ 
    then plots by $\Sigma$ should be the same as $I$-tuples of plots by the $\Sigma_i$ which match on their overlaps,
    expresses an evident locality condition of plots:   

    \vspace{-2mm} 
    \begin{equation}
    \label{SheafConditionFormSmoothSets}
    \hspace{-4cm}
      \begin{tikzcd}[
        row sep=-7pt, column sep=45pt
      ]
        \scalebox{.7}{
          \color{darkblue}
          \bf
          any ``large'' plot
        }
        \ar[
          rr,
          phantom,
          "{\hspace{-10mm}
            \scalebox{.7}{
              \color{darkgreen}
              \bf
              equivalently decomposes into
            }
          }"
        ]
        &&
        \scalebox{.7}{
          \color{darkblue}
          \bf
           ``small'' plots
        }
        \mathrlap{
          \scalebox{.7}{
            \color{darkblue}
            \bf
             that match where they overlap
          }          
        }
        \\
        \mathrm{Plots}(\Sigma,\, \CG)
        \ar[
          rr,
          "{
            \sim
          }"
        ]
        &&
        \Big\{
          \big(
            p_i 
              \in          
            \mathrm{Plots}(\Sigma_i,\, \CG)
          \big)_{i \in I}
        \mathrlap{
          \;\Big\vert\;
          \underset{
            i,j \in I
          }{\forall}
          \mbox{
            $p_i = p_j$ on $\Sigma_i \cap \Sigma_j$
          }
        \Big\}.
        }
        \\
        \big(
          \Sigma 
          \xrightarrow{p}
          \CG
        \big)
        &\mapsto&
        \big(
          \Sigma_i 
            \xrightarrow{ \iota_i }
          \Sigma 
          \xrightarrow{p}
          \CG
        \big)_{i \in I}        
      \end{tikzcd}
    \end{equation}

    \vspace{-2mm} 
\noindent    Of course, that these decomposition maps are bijections is known as the {\it sheaf condition}
which characterizes the presheaf $\mathrm{Plots}(-,\CG)$  on $\SmoothManifolds$ \eqref{PreasheafConditionForSmoothSets}
    as sheaves with respect to the ``Grothendieck topology'' of open covers:
    \vspace{-3mm} 
    $$
      \begin{tikzcd}[row sep=-2pt]
        \scalebox{.7}{
            \color{darkblue}
            \bf
            sheaves
        }
        \ar[
          rr,
          phantom,
          "{
            \scalebox{.7}{
              \color{darkgreen}
              \bf
              among
            }
          }"
        ]
        &&
        \scalebox{.7}{
            \color{darkblue}
            \bf
            pre-sheaves
        }
        \\
        \mathrm{Sh}(\SmoothManifolds)
        \ar[
          rr,
          hook
        ]
        &&
        \mathrm{PSh}(\SmoothManifolds)
      \end{tikzcd}
    $$
\end{itemize}

  \vspace{-2mm} 
\noindent  And that is already our definition of smooth sets!
A {\it smooth set} is conceived as a sheaf -- namely: the sheaf ``of its probes'' -- on the category $\SmoothManifolds$ 
of smooth manifolds with respect to the open covers:
  \vspace{-2mm} 
  $$
    \SmoothSets
    \;\;
    :=
    \;\;
    \mathrm{Sh}(\SmoothManifolds)
    \,.
  $$

   \vspace{-2mm} 
\noindent   In fact, the definition can be made simpler still: Since smooth manifolds, by definition, are always covered 
by ``Cartesian spaces'' $\FR^n$ (for some $n \in \mathbb{N}$, being their dimension), the locality/sheaf 
condition \eqref{SheafConditionFormSmoothSets} says that to know a smooth set it is sufficient to probe it by Cartesian spaces
  \vspace{-2mm} 
$$
  \big\{
    \FR^k
    \xrightarrow{\;}
    \FR^{k'}
  \big\}
  \;\;
    =:
  \;\;
  \mathrm{CartSp}
  \xhookrightarrow{\;\; \iota \;\;}
  \SmoothManifolds  
  \,.
$$

  \vspace{-2mm} 
\noindent  in that the restriction map is an equivalence of categories
  \vspace{-2mm} 
$$
  \begin{tikzcd}
  \mathrm{Sh}(\SmoothManifolds)
  \ar[
    rr,
    "{
      \sim
    }"{swap},
    "{
      \iota^\ast
    }"
  ]
  &&
  \mathrm{Sh}(\mathrm{CartSp})
  \end{tikzcd}
$$

  \vspace{-2mm} 
\noindent  This means that we may equivalently define smooth sets operationally just as above, 
but with the probes $\Sigma$ ranging just over the Cartesian spaces:

\begin{definition}[\bf  Smooth sets]
\label{SmoothSets}
The category of \textit{smooth sets}
is the category of sheaves over the site of Cartesian spaces
  \vspace{-2mm} 
\begin{equation}
  \label{SmoothSetsDefinedAsSheaves}
  \begin{tikzcd}[sep=0pt]
	\SmoothSets
    &:=&
    \mathrm{Sh}(\mathrm{CartSp})
    \\
    \mathrm{Plots}(-,\CG)
    &\leftrightarrow&
    \CG(-)
  \end{tikzcd}
\end{equation}

  \vspace{-2mm} 
\noindent  with respect to the ({\it differentiably-good}) open covers.
\end{definition}

\begin{remark}[\bf Jargon and perspective: ``Gros'' vs. ``petit'' sheaves]\label{GrosVsPetitSheaves}
$\,$

\noindent {\bf (i)}  The reader with previous exposition to the notion of {\it sheaves} has probably seen them in 
a seemingly different context:
  Typically the first examples of sheaves considered are defined not on a category like $\SmoothManifolds$ which contains 
  {\it all} manifolds (or $\mathrm{Top}$ containing {\it all} topological spaces) but instead on the category of open 
  subsets of a {\it single} manifold (or a {\it single} topological space). 

\noindent {\bf (ii)}  While these are two special cases of the same mathematical definition of ``sheaf'', and while 
close relations between both cases certainly exist,\footnote{For instance the site sheaf condition described in \eqref{SheafConditionFormSmoothSets}
may be equivalently read as: A presheaf $\CG:\text{CartSp}\rightarrow \text{Set}$ 
is a smooth set if and only if $\CG(\FR^k)$ defines a sheaf on $\FR^k$, in the ``petit'' sense of topological spaces, 
for each $\FR^k\in \text{CartSp}$. Here by $\CG(\FR^k)$ we mean the presheaf defined via 
$\CG|_{\text{Open}(\FR^k)}$ by restricting on the full subcategory open subsets of $\FR^k$ - identified 
via (smooth) open embeddings $\FR^k \hookrightarrow \FR^k$.  
}
 there is quite a qualitative difference between these two cases. To bring out this difference, some authors 
 refer to the category of sheaves on a fixed space as 
 its {\it petit topos} and to the category of sheaves on all spaces as a {\it gros topos}.

\noindent {\bf (iii)}  The perspective of {\it gros} toposes is less widely appreciated, but this is decidedly 
the perspective that we need and use here. The above motivation 
 is meant to show that, despite all the jargon, gros toposes of generalized spaces have a quite transparent definition closely reflecting the operational notion of 
 spaces explored via their plots by smaller probe spaces. 
\end{remark}

\begin{remark}[\bf Sheafification]\label{Sheafification}
It is a fact, as with any sheaf category, that
there exists an (essentially) unique functor  (see e.g. \cite[p. 128]{MacLaneMoerdijk94})
$$
L \, : \, \mathrm{PSh}(\CartesianSpaces) \longrightarrow \SmoothSets
$$
which constructs a smooth set (a sheaf) out of any presheaf on Cartesian spaces.  The ``\textit{Sheafification}'' functor L is fully characterized as being left adjoint to the canonical inclusion
functor 
$$
j\, : \, \SmoothSets \longhookrightarrow \mathrm{PSh}(\CartesianSpaces)
$$
which views smooth sets as generic presheaves on Cartesian spaces. We will not need the explicit form of the sheafification construction in this manuscript, but it is worthwhile noting a couple of its properties. Namely, by its explicit construction it preserves finite limits, and being a left adjoint it also further preserves all colimits. Lastly, it acts trivially on the (essential) image of the inclusion functor, namely $L(\CG) \cong \CG$ if $\CG$ already satisfies the sheaf condition \eqref{SheafConditionFormSmoothSets}.
\end{remark}

To get our foot on the ground, it is instructive to look at the following basic example of smooth sets.

\begin{example}[{\bf Manifolds as smooth sets}]
\label{ManRestrYonedaEmbedding}
Every ordinary smooth manifold $M \in \SmoothManifolds$ (such as a Cartesian space $\FR^n$) immediately 
has the structure of a smooth set simply by taking its 
system of plots to be given by the ordinary smooth functions into it:
\begin{equation}
  \label{PlotsOfSmoothManifolds}
  \mathrm{Plots}(\Sigma, M)
  \;\;
  :=
  \;\;
  C^\infty(\Sigma, M)
  \,.
\end{equation}
For any smooth map of probes $f:\Sigma^{'}  \rightarrow \Sigma$, the corresponding pullback of plots  $f^*:\mathrm{Plots}(\Sigma,M)\rightarrow \mathrm{Plots}(\Sigma^{'} ,M)$ is given by the \textit{actual} precomposition $(-)\circ f$ of smooth maps in $\SmoothManifolds$.
\end{example}
For emphasis, for any smooth manifold $M\in \SmoothManifolds$ let us write 
\vspace{-2mm} 
\begin{equation}
  \label{SmoothManifoldAsASmoothSet}
  y(M) \,\in\, \SmoothSets
\end{equation}

\vspace{-2mm} 
\noindent to indicate contexts in which we think of a smooth manifold as a smooth set, via \eqref{PlotsOfSmoothManifolds}.

\medskip

\noindent
{\bf Consistency of smooth sets: The Yoneda lemma.}
This identification \eqref{SmoothManifoldAsASmoothSet} now leads to a potential consistency issue: Above we bootstrapped 
smooth sets $\CG$ out of existence by declaring
that for each manifold $M$ there is a set $\mathrm{Plots}(M,\,\CG)$ of would-be plots of $\CG$ by $M$. 
But notice that before this bootstrap, there is no actual notion yet of what it would mean to have an actual map 
$M \xrightarrow{\;} \CG$, as highlighted above.
But now that we see that manifolds themselves may be regarded as smooth sets $y(M)$, there {\it is} such a notion, 
namely given by the smooth maps of smooth sets 
$y(M) \xrightarrow{\;} \CG$, being elements in the {\it hom-set} $\mathrm{Hom}_{\SmoothSets} \big(y(M), \CG\big)$.
Continuing in the spirit of understanding generalized spaces via their probes, a smooth map 
$\varphi \,:\, \CG \to \CH$ between smooth sets should 
take the $\Sigma$-plots $p$ of $\CG$ to $\Sigma$-plots  $\varphi^\Sigma_\ast p$ of $\CH$, and it 
should be {\it defined} by doing so consistently, namely such that all the following diagrams commute:
\vspace{-2mm} 
\begin{equation}
  \label{HomsOfSmoothsets}
  \begin{tikzcd}[column sep=large, row sep=small]
    &&
    \CG
    \ar[
      rr,
      "{
        \varphi
      }"{description},
      "{
        \scalebox{.7}{
          \color{darkgreen}
          \bf
            a smooth map between smooth sets
        }
      }"{yshift=7pt}
    ]
    &&
    \CH
    \ar[
      r,
      phantom,
      shift left=13pt,
      "{
        \hspace{-1cm}
        \scalebox{.7}{
          \rlap{
            hence an
            element of their Hom-set
          } 
        }
      }"
    ]
    &
    \varphi
      \,\in\,
    \mathrm{Hom}_{\SmoothSets}(\CG, \CH)
    \\
    \Sigma
    \ar[
      from=dd,
      "{
        f
      }"
    ]
    &\mapsto&
    \mathrm{Plots}(\Sigma, \CG)
    \ar[
      dd,
      "{
        f^\ast_{\CG}
      }"
    ]
    \ar[
      rr,
      "{
        \varphi^{\Sigma}_\ast
      }"{description},
      "{
        \scalebox{.7}{
          \color{darkgreen}
          \bf
          takes plots to plots
        }
      }"{yshift=5pt}
    ]
    &&
    \mathrm{Plots}(\Sigma, \CH)
    \ar[
      dd,
      "{
        f^\ast_{\CH}
      }"
    ]
    \\
    &&
    \ar[
      rr,
      phantom,
      "{
        \scalebox{.7}{
          \color{darkorange}
          \bf
          \def\arraystretch{.9}
          \begin{tabular}{c}
            compatibly with 
            \\
            precomposition
          \end{tabular}
        }
      }"
    ]
    &&
    {}
    \\
    \Sigma'
    &\mapsto&
    \mathrm{Plots}(\Sigma', \CG)
    \ar[
      rr,
      "{
        \varphi^{\Sigma'}_\ast
      }"{description}
    ]
    &&
    \mathrm{Plots}(\Sigma', \CH)
  \end{tikzcd}
\end{equation}

\vspace{-2mm} 
\noindent In categorical language, such an assignment $\varphi : \CG \rightarrow \CH$ defines a natural transformation of contravariant functors, and so
\vspace{-2mm}
$$
\text{Hom}_{\SmoothSets}(\CG,\CH)\equiv \text{Nat}(\CG,\CH)\, .
$$ 

\vspace{-2mm} 
\noindent
Hence in order for our theory of general smooth spaces to be self-consistent, it must be the case that the {\it a priori} notion of
plots agrees with the {\it a posteriori} notion. That this is indeed the case is exactly the statement of the famous Yoneda Lemma:

\begin{proposition}[\bf Yoneda Lemma for smooth sets]\label{YonedaLemmaForSmoothSets}
For $\CG \,\in\, \SmoothSets$ and $M \,\in\, \SmoothManifolds$, there is a natural bijection between the $M$-plots
  of $\CG$ \eqref{SmoothSetsDefinedAsSheaves}
  and the smooth maps of smooth sets \eqref{HomsOfSmoothsets} 
  from $y(M)$ \eqref{SmoothManifoldAsASmoothSet} 
  to $\CG$:
  \vspace{-2mm} 
  $$
    \begin{tikzcd}[
      row sep=-3pt, 
      column sep=30pt]
     \mathclap{
       \scalebox{.7}{
         \color{darkblue}
         \bf
         \def\arraystretch{.9}
         \begin{tabular}{c}
           the defining plots
           \\
           of a smooth set by
           \\
           a smooth manifold
         \end{tabular}
       }
     }
     \ar[
       rr,
       phantom,
       "{
         \scalebox{.7}{
           \color{darkgreen}
           \bf
           \def\arraystretch{.9}
           \begin{tabular}{c}
             \color{darkgreen}
             \bf
             end up being
             \\
             equivalent
           \end{tabular}
         }
       }"
     ]
     &&
     \mathclap{
       \scalebox{.7}{
         \color{darkblue}
         \bf
         \def\arraystretch{.9}
         \begin{tabular}{c}
           the actual plots of
           \\
           the smooth manifold
           \\
           regarded as a smooth set
         \end{tabular}
       }
     }
      \\
      \CG(M) \equiv \mathrm{Plots}(M,\CG)
      \ar[
        rr,
        "{
          \sim
        }"
      ]
      &&
      \mathrm{Hom}_{\SmoothSets}\big(y(M), \CG\big)
      \\
      \hspace{1.5cm} \varphi^{M}_\ast(
        \mathrm{id}_M
      )
      &\longmapsfrom&
      \varphi
    \end{tikzcd}
  $$
\end{proposition}

\vspace{-2mm} 
\noindent Choosing $\CG=y(N) \in \SmoothSets$ for an arbitrary manifold $N$, the above immediately implies that the ``\textit{ embedding}'' functor 
\vspace{-2mm}
\begin{align*}
   y:  \SmoothManifolds &\longhookrightarrow \SmoothSets \\
        M&\longmapsto y(M):= \mathrm{Hom}_{\SmoothManifolds}(-, M)
\end{align*}

\vspace{-2mm}
\noindent is \textit{fully faithful}. Hence any results and constructions on finite-dimensional smooth manifolds may
equivalently be phrased in terms of their smooth set incarnation (and vice-versa).





\medskip 
Our running motivating example of smooth sets
is the following:
\begin{example}[\bf Field spaces as smooth sets]
Let 
\begin{itemize}[leftmargin=.5cm]
\item $M$ be a smooth manifold modeling {\it spacetime} (possibly with boundary, such as for a closed temporal interval), 
\item $F\rightarrow M$ be a smooth fiber bundle modeling the nature of the {\it dynamical fields} (the {\it field bundle}).
\end{itemize}
Then the (``off-shell'') {\it field configurations} or {\it field histories} should be the smooth sections $\phi \,\in\, \Gamma_{M}(F)$ of the field bundle.
Now for a  variational field theory, the dynamics of the fields should be determined a `smooth' map -- the {\it action functional} -- of the form
\vspace{-2mm} 
\begin{equation}
  \label{ActionFunctionalInMotivation}
  \mathllap{
    \scalebox{.7}{
      \color{darkgreen}
      \bf
      \def\arraystretch{.9}
      \begin{tabular}{c}
        Action 
        \\
        functional
      \end{tabular}
    }
  }
  S
   \;:\;
  \overset{
    \mathclap{
      \raisebox{7pt}{
        \scalebox{.7}{
          \color{darkblue}
          \bf
          \def\arraystretch{.9}
          \begin{tabular}{c}
            off-shell
            \\
            field space
          \end{tabular}
        }
      }
    }
  }{
    \Gamma_M(F) 
  }
  \;\;
    \xrightarrow{
      \quad \mathrm{smooth}\quad 
    }
  \FR
\end{equation}

\vspace{-2mm} 
\noindent in that its critical locus, where
its `\textit{variation vanishes} $\delta S = 0$', consists of the physical (``on-shell'') field histories which satisfy their equations of motion:
\vspace{-2mm} 
\begin{equation}
  \label{CriticalocusInMotivation}
  \overset{
    \mathclap{
      \raisebox{7pt}{
        \scalebox{.7}{
          \color{darkblue}
          \bf
          \def\arraystretch{.9}
          \begin{tabular}{c}
            critical
            \\
            locus
          \end{tabular}
        }
      }
    }
  }{
    \mathrm{Crit}(S)
  }
    \;:=\;
  \overset{
    \mathclap{
      \raisebox{7pt}{
        \scalebox{.7}{
          \color{darkblue}
          \bf
          \def\arraystretch{.9}
          \begin{tabular}{c}
            on-shell
            \\
            field space
          \end{tabular}
        }
      }
    }
  }{
    \Gamma_M(F)_{{}_{\delta S = 0}}
  }
    \;\subset\; 
  \Gamma_M(F)
  \,.
\end{equation}
The problem now is to make precise sense of \eqref{ActionFunctionalInMotivation} and \eqref{CriticalocusInMotivation}. If $\Gamma_M(F)$ were a {\it finite-dimensional manifold} 
then $S$ would be an ordinary smooth function and $\mathrm{Crit}(S)$ its ordinary critical locus of vanishing derivatives. Indeed, many physics textbooks behave as 
if this were the case, but the reality is far from it: If $M$ is not compact -- and for realistic field theories it is not -- then $\Gamma_M(F)$ is not even 
a Fr{\'e}chet manifold! (See around \eqref{FrechetMappingSpaces} below.) Hence what one needs is a good notion of such generalized smooth spaces for which 
\eqref{ActionFunctionalInMotivation} makes rigorous sense while
the naive geometric intuition 
for obtaining \eqref{CriticalocusInMotivation} 
still essentially makes sense.

This is very naturally accomplished with smooth sets:
First, consider the case where the field bundle is a trivial bundle $M \times N$ with fiber $N$. (For instance, for real scalar field theory we would have $N = \FR$.) 
Then the space \eqref{ActionFunctionalInMotivation} of off-shell fields must be just the space of smooth maps $M \to N$. 
Our task is hence to define the $\Sigma$-shaped plots into such a space of maps. But, for each point $\sigma \in \Sigma$,
such a plot should be a smooth map $\phi_\sigma : M \to V$ and the dependence of these maps on $\Sigma$ should better be smooth, too, in that the combined map
\vspace{-2mm} 
$$
  \begin{tikzcd}[row sep=-4pt]
    \mathllap{
        \phi_{(-)}(-)
        \;\;:\;\;
    }
    \Sigma
      \times 
    M 
    \ar[
      rr
    ]
    &&
    N
    \\
    (\sigma,m)
    &\longmapsto&
    \phi_\sigma(m)
  \end{tikzcd}
$$

\vspace{-2mm} 
\noindent is smooth. This indeed defines a smooth set, which we denote by angular brackets:
  \vspace{-2mm} 
\begin{equation}
  \label{SmoothMappingSet}
  \mathllap{
    \scalebox{.7}{
      \color{darkblue}
      \bf
      \def\arraystretch{.9}
      \begin{tabular}{c}
        smooth 
        \\
        mapping space
      \end{tabular}
    }
  }
  \def\arraystretch{1.3}
  \begin{array}{l}
    \mathbold{\Gamma}_M(M \times N)
    \,=\,
    [M,\,N]
    \;\;\;\in\;
    \SmoothSets
    \\
    \mathrm{Plots}\big(
      \Sigma
      ,\,
      [M,\,N]
    \big)
    \;:=\;
    C^\infty\big(
      \Sigma \times M
      ,\,
      N
    \big)
    \mathrlap{\,.}
  \end{array}
\end{equation}

  \vspace{-2mm} 
\noindent  Notice how for $M = \ast$ the point, this recovers \eqref{PlotsOfSmoothManifolds}.

In \eqref{SmoothMappingSet} we may naturally think of a plot $\phi_\Sigma \,\in\, \mathrm{Plots}\big(\Sigma,\, [M,N] \big)$ 
as a {\it (smoothly) $\Sigma$-parametrized family of fields}.
For instance, if $\Sigma = \FR^1_t$ then $\phi_t \equiv \phi_{\FR^1_t}:\FR^1_t\times M\rightarrow N$ is a smooth {\it 1-parameter family of fields}.
Hence for an action functional \eqref{ActionFunctionalInMotivation} to be a smooth map \eqref{HomsOfSmoothsets}, 
it in particular must take such 1-parameter families of fields to smooth 1-parameter families of real numbers:
\vspace{-2mm} 
$$
  \begin{tikzcd}[row sep=-4pt]
    \Gamma_M(M \times N)
    \ar[
      rr,
      "{
        S
      }"
    ]
    &&
    \FR
    \\
    \mathllap{
      S^{\FR^1}_\ast
      \;:\;\;\;
    }
    C^\infty\big(
      \FR_t \times M
      ,\,
      N
    \big)
    \ar[
      rr
    ]
    &&
    C^\infty(\FR_t, \FR)
    \\
    \phi_{t}
    &
    \longmapsto
    &
    \big(
      t_0
      \,\mapsto\,
      S(\phi_{t=t_0})
    \big)
    \mathrlap{\,.}
  \end{tikzcd}
$$

  \vspace{-2mm} 
\noindent  But this now allows to easily define the variation $\delta S$ (see  \eqref{CriticalocusInMotivation})
at $\phi_0$ along the variation encoded in the family $\phi_{t}$ by reducing to the usual derivative of smooth 
functions on $\FR$:
  \vspace{-2mm} 
$$
  \delta_{\phi_{t}} S
  \,=\,
  \frac{\dd}{ \dd t }
  S\big(
    \phi_{t}
  \big)\big\vert_{t = 0}
  \,.
$$

  \vspace{-1mm} 
\noindent The critical locus is then the set

\vspace{-2mm}
$$
\mathrm{Crit}(S)=\Big\{\phi\in C^\infty(M,N) \; \big{|}\; \delta_{\phi_t} S = \frac{\dd}{\dd t} S (\phi_t)|_{t=0}=0 , \hspace{0.3cm}
  \forall \, \phi_t\in  C^\infty(\FR^1_t\times M,\, N) \, \, \, \mathrm{s.t.} \, \, \, \phi_{t=0}=\phi \Big\} \, .
$$ 

  \vspace{0mm} 
\noindent Notice here how the notion of smooth sets allows to perform standard operations of finite-dimensional 
differential geometry {\it plot-wise} on ordinary manifolds, and thereby extend these notions to smooth sets like 
$\Gamma_M(F)$ which by no means have a (finite-dimensional) manifold structure. Moreover, even in the finite-dimensional 
setting, it is often the case that the critical locus cannot be supplied with a smooth submanifold structure. In contrast, 
we shall see that both in the finite and infinite-dimensional (local) field theory setting, the critical locus has 
a natural smooth subset structure (for more details on this see \cref{OnShellFieldsCriticalSetSubsection}).

%
\end{example}

\begin{remark}[\bf Traditional Fr{\'e}chet and jet bundle technology]
The traditional approach to formalizing field spaces \eqref{ActionFunctionalInMotivation} models them as infinite-dimensional Fr{\'e}chet manifolds, which works (only) 
in the (physically often unjustified)
\footnote{For instance, relevant here is compact support in spatial directions vs. compact support in time direction.}
special case that the spacetime manifold $M$ is {\it compact}. In that case (only), the set of smooth functions $C^{\infty}(M,N)$
carries the structure of a Fr\'echet manifold with local charts taking values in certain well-behaved infinite-dimensional vector spaces:
\vspace{-1.5mm} 
\begin{equation}
  \label{FrechetMappingSpaces}
  \mbox{$M$ compact}
  \quad
    \Rightarrow
  \quad
  C^{\infty}(M,N)
  \,\in\,
  \FrechetManifolds
\end{equation}

\vspace{-1.5mm} 
\noindent
 The same is in fact true for 
any set of sections $\Gamma_M(F)$ of a fiber bundle over a compact $M\in \SmoothManifolds$ \cite{KrieglMichor}.

\medskip 
 If $M$ is noncompact, the manifold description via infinite-dimensional charts in this approach becomes very subtle and somewhat unnatural due to the many choices appearing and extra assumptions needed;
 see \cite[\S 10.10]{Michor}\cite[Ch. IX]{KrieglMichor}.

\medskip 
Given such an infinite-dimensional manifold structure, one can give meaning to the smoothness of the functional $S$, and its variation  $\delta S$. Naturally, this approach carries heavy functional analytical baggage, which from the point of 
view of practicing physics is more than often redundant. Furthermore, it is not canonically generalizable to non-compact spacetimes - which are 
the norm in physics.

\medskip 
A partial alternative, technically simpler resolution is achieved by noticing that functionals in fundamental field theories are \textit{local}, 
that is they only depend on the value of a field $\phi \in \Gamma_M(F)$ and its jets at a point $x\in M$. The mathematical 
incarnation of this is that the action may be written as
\vspace{-2mm} 
\begin{align*}
	S: \Gamma_M(F) &\longrightarrow \FR \\
	\phi &\longmapsto \int_{M} L(j^{\infty}\phi), 
\end{align*}

\vspace{-2mm} 
\noindent where $L:J^\infty_M F\rightarrow \Omega^d(M)$ is a `smooth' function on the infinite jet bundle $J^{\infty}_M F$ with values in densities 
$\Omega^{d}(M)$, and $j^{\infty}:\Gamma_M(F) \rightarrow \Gamma_M (J^\infty F)$ being the infinite jet prolongation. 
The idea is that, even though the infinite jet bundle is necessarily an infinite-dimensional space, it is much easier 
to describe the actual field space $\Gamma_M(F)$.
\footnote{There are essentially two distinct ways to consider $J^\infty_M F$ as a smooth space, either as a genuine Fr\'{e}chet manifold or as a purely formal object. 
We will recall the former description before embedding in our picture of smooth sets, and will comment on the relation with the latter.}
The usual manipulations of variational calculus are then delegated to analogous operations in the ``variational bicomplex" (see \cref{VariationalBicomplexSection})  
of differential forms on $J^\infty_M F$, with the end products `pulled-back' 
to the field space $\Gamma_M(F)$. This bypasses most of usual the difficulties of the smooth structure on $\Gamma_M(F)$, 
essentially by avoiding defining it altogether.\end{remark}

In the present text, we wish to advocate an even more convenient setting to make sense of the smooth structure of infinite-dimensional 
field spaces entirely \textit{operationally} by consistently answering the question:
\begin{quote}
 \textit{``What are the ways we can smoothly probe the would-be space with the simple probe Cartesian spaces?"}
 \end{quote} 
 This is naturally achieved by considering field spaces as objects in the sheaf topos of smooth sets over the site of Cartesian spaces.
This avoids most of the functional analysis alluded to, faithfully subsumes and combines the above two approaches (along with the 
corresponding variational calculus), has many positive categorical
properties and, furthermore, naturally generalizes to include fermionic fields and fields with non-trivial internal symmetries
such as higher gauge fields (see \cref{outlook}).

\subsection{Field spaces, tangent bundles and diffeomorphisms}
\label{SmoothSetsSection}

Having intuitively motivated the relevance and definition of smooth sets, here we transition to a more precise and rigorous description. 
To start off, we recall how Fr\'{e}chet manifolds (see \cite{Hamilton}\cite{DGV}) embed into smooth sets. Recalling further 
the Exponential Law property of Fr\'{e}chet mapping spaces between manifolds (with compact domain) and its embedding to 
smooth sets, we show how one can arrive at the same object by working directly in smooth sets. This will allow for a general
definition of a smooth mapping space between any two smooth sets, and in particular for a smooth set structure on sections of 
bundles (over potentially non-compact spacetimes). Using the smooth structure on the space of fields, we proceed to rigorously 
define tangent vectors, smooth vector fields, diffeomorphisms, and furthermore  differential forms on the field space as smooth
fiber-wise linear maps out of its tangent bundle. The discussion will fit into a more general framework \cite{GSS-2}.

\begin{proposition}
[\bf Fr\'{e}chet manifolds as smooth sets {\cite{Losik92}\cite{Fro81}}]
\label{FrManRestrEmbedding}
Consider the category $\FrechetManifolds$ of (infinite-dimensional) Fr\'echet manifolds.
The restricted  embedding along 
\vspace{-1mm} 
$$
{\rm CartSp}\longhookrightarrow {\rm Man}\longhookrightarrow {\rm FrMan}
$$

\vspace{-2mm} 
\noindent
defines a fully faithful embedding
\vspace{-2mm} 
\begin{align}
y: {\rm FrMan}&\longhookrightarrow {\rm SmoothSet}
\\[-2pt]
G &\longmapsto {\rm Hom}_{{\rm FrMan}}(-,G)|_{{\rm CartSp}}\, , \nn
\end{align}

\vspace{-2mm} 
\noindent where on the right-hand side we consider \textit{smooth Fr\'echet} maps.
\end{proposition}
This result says that the infinite-dimensional field spaces appearing in physics, at least with the domain being
compact and hence equipped with a Fr\'{e}chet manifold structure, may naturally be viewed as smooth sets. Moreover, 
given two finite-dimensional manifolds $M,N\in \SmoothManifolds$, with $M$ 
compact, the Fr\'echet manifold mapping space $C^{\infty}(M,N)_{\FrechetManifolds}$ satisfies the crucial property 
for the purposes of variational calculus, known as the ``exponential law" or ``internal hom" property 
\cite[Cor. 11.9]{Michor}\cite[Thm 42.14]{KrieglMichor}:
\footnote{However, it does not define an internal hom functor on finite-dimensional manifolds, since it takes 
values inside the larger Fr\'echet category instead. Furthermore, it is not defined on arbitrary Fr\'echet 
manifolds either.}


\begin{proposition}[\bf Exponential Law for Fr\'echet mapping space]
Let $M,N$ be finite-dimensional manifolds, with $M$ compact. Then there is a canonical bijection 
\vspace{-2mm} 
$$
{\rm Hom}_{{\rm FrMan}}\big(S,C^{\infty}(M,N)_{{\rm FrMan}}\big) 
\cong_{\text{Set}} {\rm Hom}_{{\rm Man}}(S\times M, N)\, ,
$$

\vspace{-2mm} 
\noindent
for any smooth manifold $S\in \SmoothManifolds$, and moreover naturally in $S$.
\end{proposition}

 Explicitly, $f\in \text{Hom}_{\FrechetManifolds}\big(S,C^{\infty}(M,N)_{\FrechetManifolds}\big)$ is mapped to 
 $\hat{f} \in \text{Hom}_{\SmoothManifolds}(S\times M, N)$ where 
 $\hat{f}(s,m):= f(s) (m)$, and the theorem guarantees this assignment is in fact a bijection. Intuitively, this says that smooth maps 
$S\rightarrow C^{\infty}(M,N)_{\FrechetManifolds}$ into the mapping space are the same as smoothly $S$-parametrized maps 
from $M$ to $N$. In particular, it is true for all $S\in \text{CartSp}\subset \SmoothManifolds$, 
and so combined with Prop. \ref{FrManRestrEmbedding} translates to the statement that 
\vspace{-2mm} 
\begin{align}\label{FrechetMappingEmbedded}
	y\big( C^{\infty}(M,N)_{\FrechetManifolds}\big)  \cong_{\SmoothSets} \text{Hom}_{\SmoothManifolds}(-\times M,N)
\end{align}

\vspace{-2mm} 
\noindent as smooth sets.	
On the other hand, as with any (pre)sheaf category, $\SmoothSets$ has an honest internal hom functor 
\vspace{-2mm} 
$$
[-,-]:\SmoothSets^{op}\times \SmoothSets \longrightarrow \SmoothSets\, .
$$ 

\begin{definition}[\bf  \bf Smooth set internal hom]
\label{SmoothSetinternalHom}
Let $\CG,\CH \in \SmoothSets$, the \textit{smooth mapping set} $[\CG,\CH] \in \SmoothSets$ is defined by 
\vspace{-2mm} 
\begin{align}[\CG,\CH](\FR^k):= \mathrm{Hom}_{\SmoothSets}\big(y(\FR^k)\times \CG, \CH\big) \, .
\end{align}
\end{definition}
The internal hom property of $[-,-]$ on representable sheaves is simply the Yoneda Lemma \ref{YonedaLemmaForSmoothSets}
\vspace{-2mm} 
$$
\mathrm{Hom}_{\SmoothSets}\big(y(\FR^k),[\CG,\CH]\big)\cong_{\mathrm{Set}}[\CG,\CH](\FR^k):=  \mathrm{Hom}_{\SmoothSets}\big(y(\FR^k)\times \CG, \CH\big)\, . 
$$

\vspace{-2mm} 
\noindent On general non-representable smooth sets, it follows since any such is the colimit of representables, 
and that $\mathrm{Hom}$ functors preserve (co)limits. Arguments of this form are extremely useful and standard in categories of 
(pre)sheaves \cite{MacLaneMoerdijk94}, and will be used throughout these series. For the convenience of the reader, we include the 
detailed calculation in this instance. Since any smooth set $\CX$ may be written as a colimit
$\CX\cong \mathrm{colim}_{i}^{\SmoothSets}\,y(\FR^{k_i})$ (e.g. \cite[p. 42]{MacLaneMoerdijk94}), we have 
\vspace{-2mm} 
\begin{align*}
\mathrm{Hom}_{\SmoothSets}\big(\CX,[\CG,\CH]\big)&\cong\mathrm{Hom}_{\SmoothSets}
\big(\mathrm{colim}_{i}^{\SmoothSets}\,y(\FR^{k_i}),[\CG,\CH]\big) \cong \mathrm{lim}_i^{\mathrm{Set}} \, \mathrm{Hom}_{\SmoothSets}\big(y(\FR^{k^i}), [\CG, \CH] \big) 
\\ 
&\cong \mathrm{lim}_i^{\mathrm{Set}} \, \mathrm{Hom}_{\SmoothSets}\big(y(\FR^{k^i})\times \CG, \CH\big) \cong 
\mathrm{Hom}_{\SmoothSets}\big(\mathrm{colim}_i^{\SmoothSets} \,y(\FR^{k^i})\times \CG, \CH\big) \\  &\cong \mathrm{Hom}_{\SmoothSets}\big(\CX\times\CG,\CH\big) 
\end{align*}

\vspace{-2mm} 
\noindent where in the first line we used the fact that $\mathrm{Hom}$ functors preserve (co)limits, 
while in the second line we used the internal hom property on representables and that colimits commute with products in (pre)sheaf categories.
\footnote{Strictly speaking, the canonical identification
$\mathrm{colim}_i^{\SmoothSets} \big(\,y(\FR^{k^i})\times \CG \big) \cong \mathrm{colim}_i^{\SmoothSets} \,y(\FR^{k^i})\times \CG$
requires justification. Namely, it is easy to see it holds if the colimit is taken in the presheaf category, since therein it is computed probe-wise as a colimit in Set, where commuting colimits with products is obvious. However, colimits in sheaf categories are given the sheafification of the corresponding presheaf colimit (Rem. \ref{LimitsAndColimitsExist}). The result then follows by recalling that the sheafication functor $L$ commutes with products and that $L \CG \cong \CG $ for $\CG$ a sheaf  (Rem. \ref{Sheafification}).}

\medskip 
Restricting our attention to the case where $\CG=y(M)$ and $\CH=y(N)$ are representable by the 
finite-dimensional manifolds above, this gives another notion of a smooth structure $\big[y(M),y(N)\big]$
on the smooth mapping set between two manifolds, which exactly coincides with that of the intuitive
description \eqref{SmoothMappingSet}. 
In the case where $M$ is compact, it also recovers that of 
$C^\infty(M,N)_\FrechetManifolds$ of \eqref{FrechetMappingSpaces} under the Yoneda embedding.

\begin{lemma}[{\bf Yoneda preserves Fr\'echet mapping space}]
Let $y(M),y(N)\in\SmoothSets$ for $M,N\in \SmoothManifolds$, with $M$ compact. The embedding of the
Fr\'echet manifold $C^{\infty}(M,N)_\FrechetManifolds$ in smooth sets is isomorphic to the mapping 
smooth set $\big[y(M),y(N)\big]$, 
\vspace{-5mm} 
\begin{align}
y\big(C^{\infty}(M,N)_{{\rm FrMan}}\big)\cong_{{\SmoothSets}} \big[y(M),y(N)\big]\, .
\end{align}
\begin{proof} We have the sequence of isomorphisms (of sets)
\vspace{-2mm} 
\begin{align*}\mathrm{Hom}_{\SmoothSets}\Big(y(\FR^k),\,\big[y(M),y(N)\big]\Big)&
\cong \mathrm{Hom}_{\SmoothSets}\Big(y(\FR^k)\times y(M), y(N)\Big) \\ 
	&\cong \mathrm{Hom}_{\SmoothSets}\Big(y(\FR^k\times M), y(N)\Big)\\
	&\cong \mathrm{Hom}_{\SmoothManifolds}\big(\FR^k\times M, N\big)\, ,
\end{align*}

\vspace{-2mm} 
\noindent naturally in $\FR^k$, where in the first line we use the Yoneda Lemma \ref{YonedaLemmaForSmoothSets} / internal hom property, 
in the second that the Yoneda embedding preserves finite products, and in the third 
by the faithfulness of the Yoneda embedding. The statement follows by property \eqref{FrechetMappingEmbedded} of 
 the Fr\'echet mapping space 
 \vspace{-2mm} 
 $$	
 y\big( C^{\infty}(M,N)_{\FrechetManifolds}\big)  \cong_{\SmoothSets} \mathrm{Hom}_{\SmoothManifolds}(-\times M,N)\, .
 $$ 

\vspace{-6mm} 
\end{proof}


\end{lemma}
This result (also proved in \cite[Lem. A.1.7]{Waldorf} by more explicit means) exhibits one of the powers of the topos of smooth sets. 
It means that we can arrive at the appropriate smooth structure on a mapping set of manifolds by the simple 
abstract formula \eqref{SmoothSetinternalHom}, bypassing all the functional analytical technology 
involving Fr\'echet manifold theory. Furthermore, the smooth set $[y(M),y(N)]$ exists and satisfies the
internal hom / Exponential Law property even if $M$ is not compact, in which case $C^{\infty}(M,N)$ can no longer be equipped with 
a Fr\'echet structure. 

\begin{example}[\bf Vector-valued field space]
\label{VectorFieldTheoryFieldSpace}
In this case, the field bundle is $F=M\times W$ where $W$ a finite-dimensional real vector space (e.g., $W=\FR$ for scalar field theory). 
The set of field configurations is given by smooth $W$-valued functions, $\Gamma_M(M\times W)= C^\infty(M,W)$. 
It follows that the $\FR^k$-plots of the smooth set of fields $\CF=[y(M),y(W)]$
are given by smooth maps of manifolds
\vspace{-1mm} 
$$
\phi^k: \FR^k\times M \longrightarrow W\, .
$$
\end{example}

To stress the point further, a smooth structure is naturally defined on more general 
field spaces consisting of sections of an arbitrary fiber bundle over a potentially noncompact base. 

\begin{definition}[\bf  Smooth sets of smooth sections]
\label{SectionsSmoothSet}
Let $\pi: F\rightarrow M$ be a fiber bundle of smooth manifolds, with set of smooth sections $\Gamma_{M}(F)$. 
The smooth set of sections $\CF=\mathbold{\Gamma}_M(F)\in \SmoothSets$ is defined by
\vspace{-1mm} 
\begin{align}
	\CF(\FR^k)\equiv \mathbold{\Gamma}_M(F) (\FR^k):=\{\phi^k:\FR^k\times M \rightarrow F \; | \; \pi\circ \phi^k = \pr_2 \}\, ,
\end{align}

\vspace{-1mm} 
\noindent where $\FR^k\in \mathrm{CartSp}$ and $\pr_2 :\FR^k\times M\rightarrow M$ is the projection onto M. That is,
$\phi^k: \FR^k\times M \rightarrow F$ such that
\vspace{-2mm} 
	\[ 
\xymatrix@=1.2em{ &&  F \ar[d]^{\pi}
	\\ 
	\FR^k\times M \ar[rru]^-{\phi^k} \ar[rr]^-{\pr_2} && M
}   
\]

\vspace{-2mm} 
\noindent commutes, and so equivalently $\CF(\FR^k)\cong \Gamma_{M\times \FR^k}(\pr_2^* F)$.
\end{definition}

In simple words, $\CF(\FR^k)$ is the set of smoothly $\FR^k$-parametrized sections of $F\rightarrow M$. Note that 
set theoretically, there is a natural injection \footnote{This is essentially the internal hom of the category Set of sets.}
\vspace{-2mm} 
\begin{align*}
\CF(\FR^k)&\longhookrightarrow \mathrm{Hom}_{\mathrm{Set}}\big(\FR^k, \Gamma_M(F)\big)
\\[-2pt]
\phi^k &\longmapsto \hat{\phi}^k
\end{align*}

\vspace{-2mm} 
\noindent  where $\hat{\phi}^k(x)=\phi^k(x,-)\in \Gamma_M(F)$. Hence $\CF(\FR^k)$ \textit{defines} the subset of `smooth maps' 
from the set $\FR^k$ into the set $\Gamma_M(F)$. Keep in mind that this latter interpretation only holds for smooth
sets that have an underlying notion of set of points, i.e., ``diffeological spaces"\footnote{Equivalently, these may be characterized as the ``concrete'' sheaves among all smooth sets \cite{BH11}.}, while the definition can in fact be applied 
to sections of arbitrary `bundle' $\CG\rightarrow \CH$ of smooth sets.

\medskip 
A more conceptual way to arrive at the above 
definition is the following. Note that the \textit{set} of smooth sections $\Gamma_M(F)\in \mathrm{Set}$ may be identified 
with the fiber product 
$\Gamma_M(F) \cong C^\infty(M,F)\times_{C^{\infty}(M,M)} \{\mathrm{id}_M\}$ in Set, i.e., the pullback set 
\vspace{-2mm} 
\[
\xymatrix@=1.6em  {\Gamma_M(F) \ar[d] \ar[rr] &&   C^{\infty}(M,F) \ar[d]^{\pi_*} 
	\\ 
\{\mathrm{id}_M\} \, \ar@{^{(}->}[rr]  && C^{\infty}(M,M)
\, . } 
\]

\vspace{-2mm} 
\noindent Instead, we may remember the smooth structure of all the objects appearing in the diagram and compute the pullback in 
smooth sets, as $[M,F]\times_{[M,M]} y(\{\mathrm{id}_M\})$.  Pullbacks in any (pre)sheaf category are computed objectwise 
in Set, i.e.,
\vspace{-2mm} 
\[
\xymatrix@=1.6em  {[M,F]\times_{[M,M]} y(\{\mathrm{id}_X\}) \big(\FR^k\big) \ar[d] \ar[rr] &&   [M,F](\FR^k) \ar[d]^{\pi_*} 
	\\ 
	\{\mathrm{id}_M\}(\FR^k) \ar@{^{(}->}[rr]  && [M,M](\FR^k)
	 } 
\]

\vspace{-2mm} 
\noindent for any $\FR^k\in \mathrm{CartSp}$. By calculating the fiber product above in Set, for each probe $\FR^k$, it is 
immediate that the result coincides with 
$\mathbold{\Gamma}_M(F)$ of Def. \ref{SectionsSmoothSet}. The top arrow exhibits the smooth set of sections $\CF$ as a smooth subspace of
the full mapping space 
\vspace{-2mm}
$$
\CF\longhookrightarrow [M,F]\, ,
$$ 

\vspace{0mm}
\noindent as intuitively expected from the point-set inclusion $\Gamma_M(F)\hookrightarrow C^\infty(M,F)$.

\begin{remark}[\bf Limit and colimit constructions exist]\label{LimitsAndColimitsExist} In fact all (small) limits and colimits exist in any sheaf category, and in 
particular smooth sets. Limits are computed plot-wise in Set, as in the above example, while for colimits one has to further
sheafify the resulting presheaf (Rem. \ref{Sheafification}). This means that constructions that do not exist in 
finite- or infinite-dimensional manifolds that appear in physics, such as intersections and quotients, exist after 
embedding fully faithfully to this larger, better-behaved category.
\end{remark}

\begin{remark}[\bf Field space as a sheaf of sheaves]\label{FieldSpaceAsASheafOfSheaves}
As is well known, an arbitrary fiber bundle $F\rightarrow M$ might have no \textit{global} sections (and hence also no global $\FR^k$-parametrized sections). 
Thus, strictly speaking,  Def. \ref{SectionsSmoothSet} as stated is potentially null. In reality, what one should keep in mind is the following: 
The assignment of (local) sections of $F$ over $M$ forms a sheaf on $M$ in the petit sense of Rem. \ref{GrosVsPetitSheaves},
\vspace{-2mm} 
\begin{align*}
\Gamma_{(-)}(F): \mathrm{Open}(M)&\longrightarrow \mathrm{Set} \\
U\subset M &\longmapsto \Gamma_U(F)\, .
\end{align*}

\vspace{-2mm} 
\noindent Def. \ref{SectionsSmoothSet} may then be applied locally on $M$, i.e., on each set of local sections $\Gamma_U(F)$, defining 
a sheaf \footnote{That is, a sheaf with respect to the product coverage on the site $\mathrm{Open}(M)\times \CartesianSpaces$.}
\vspace{-2mm}
\begin{align*}
\mathbold{\Gamma}_{(-)}(F): \mathrm{Open}(M)\times \CartesianSpaces &\longrightarrow \mathrm{Set} \\
U\times \FR^k &\longmapsto \mathbold{\Gamma}_U (F) (\FR^k)\cong \Gamma_{U\times\FR^k}(\pr_2^*F) \, .
\end{align*}

\vspace{-1mm}
\noindent This functor may be equivalently thought of as a \textit{petit} sheaf on $M$ valued in the \textit{gros} sheaves of  $\SmoothSets$, 
or as a (gros) sheaf on $\CartesianSpaces$ valued in (petit) sheaves on $M$. The takeaway for the rest of the manuscript is that when we talk 
about the `smooth set of sections' of a field bundle, we will implicitly mean local sections over an arbitrary open set $U\subset M$. 
Indeed, all statements and definitions regarding the smooth set of fields $\CF=\mathbold{\Gamma}_M(F)$ we make apply (functorially) for each 
open set $U\subset M$. For this reason, and to avoid excess confusion for readers new to these concepts, we will be tacitly suppressing the
mentioning of local sections and the corresponding petit sheaf aspect -- recalling it only when strictly necessary.
\end{remark}

The smooth structure on the field space can be used to define many smooth geometrical constructions of interest. 
For instance, the natural evaluation map $\mathrm{ev}: \Gamma_M(F) \times M \rightarrow F $, 
$(\phi,x) \mapsto \phi(x)$ extends (uniquely) to a smooth map\footnote{In fact, such an evaluation map $\mathbold{\Gamma}_{\CH}(\CG) \times \CH \rightarrow \CG$ exists for any `bundle' $\CG\rightarrow \CH$ of smooth sets and its corresponding smooth set of sections. }
\vspace{-1mm} 
\begin{align}\label{SmoothEvaluationMap}
\mathrm{ev}\,:\, \CF\times y(M) &\longrightarrow y(F) 
\\[-2pt]
(\phi^k, x^k) &\longmapsto \phi^k \circ (\id_{\FR^k}, x^k) \in y(F)(\FR^k) \, \nn 
\end{align}

\vspace{-1mm} 
\noindent
for any pair of plots $\phi^k \in \CF(\FR^k)$ and $x^k\in y(M)(\FR^k)=C^\infty(\FR^k, M)$. Moreover, smooth 
real-valued functions on the field space $\CF$ are defined as maps of smooth sets
\vspace{-2mm} 
\begin{align}\label{SmoothFunctionsOnFieldSpace}
C^\infty(\CF):= \mathrm{Hom}_{\SmoothSets}\big(\CF,y(\FR)\big) \, .
\end{align}

\vspace{-2mm} 
\noindent In other words, a smooth function $f\in C^\infty(\CF)$ defines a map of sets 
\vspace{-2mm} 
\begin{align*}
f \,:\, \Gamma_M(F) &\longrightarrow \FR
\\[-2pt]
\phi &\longmapsto f(\phi)
\end{align*}

\vspace{-2mm} 
\noindent which furthermore maps \textit{smooth} $\FR^k$-plots of fields to \textit{smooth} functions in 
$C^\infty(\FR^k, \FR)$, naturally with respect to pullbacks by maps of probes. 
The algebra structure of $C^\infty(\CF)$ follows plot-wise from that of the target $y(\FR)$. Furthermore, since 
for any $\FR^1$-plot 
of fields\footnote{We will denote generic $\FR^k$-plots by $\phi^k$. In the case of $k=1$, we will denote $\FR^1$-plots 
by $\phi_t=\phi^1$ since these will play a special role throughout. This is a slight abuse of notation, common in the 
physics literature, whereby the symbol $t$ does \textit{not} denote evaluation of the function $\phi^1$ at $t\in \FR$, 
but instead is simply a placeholder. Evaluation at $t_0\in \FR$ in this notation will be denoted by 
$\phi_{t=t_0}:=\phi^1(t_0,-)$.} $\phi_t\in \mathbold{\Gamma}_M(F)(\FR^1)\cong \mathrm{Hom}_{\SmoothSets}\big(y(\FR^1), \mathbold{\Gamma}_M(F)\big)$ 
the composition $f\circ \phi_t $ defines a smooth map $\FR \rightarrow \FR$, there is an induced derivation
\vspace{-2mm} 
\begin{align}\label{SmoothPathInducedDerivation}
    C^\infty(\CF) &\longrightarrow \FR \\
    f &\longmapsto \partial_t(f\circ \phi_t) \vert_{t=0} \, . \nn 
\end{align} 

\vspace{-2mm} 
\noindent This is suggestive of a notion of tangent vectors on field spaces via the use of $\FR^1$-plots of fields, e.g.,
as motivated in \cite{DF99}.

\begin{example}[\bf Tangent vectors on field space]\label{TangentVectorsOnFieldSpace} 
Let $\phi_t\in \mathbold{\Gamma}_M(F)(\FR^1)$ be a $\FR^1$-plot of fields, i.e.,
an $\FR^1$-parametrized section $\phi_t:\FR^1\times M \rightarrow F$. For each $x\in M$, by the section condition, we have a smooth curve $\phi_t(x):\FR^1\rightarrow F$ 
whose image is contained in the fiber over $x$. Thus,
$$
\partial_t\phi_t(x) |_{t=0} \;\; \in V_{\phi_0(x)} F
$$
defines a \textit{vertical} tangent vector at $\phi_0(x)\in F$, where $V_{\phi_0(x)}(F)\subset T_{\phi_0(x)} F$ 
is the subspace of vertical tangent vectors. Varying over $x\in M$, we get a \textit{smooth} section 
$\partial_t\phi_t |_{t=0}:M \rightarrow VF$ covering $\phi_0$
\vspace{-2mm} 
	\[ 
\xymatrix@=1.7em{ &&  VF \ar[d]
	\\ 
	M\ar[rru]^{\partial_t\phi_t |_{t=0}~~} \ar[rr]^{~~~{\phi}_0} && F \, ,
}   
\]

\vspace{-2mm} 
\noindent where $VF\hookrightarrow TF \rightarrow F \rightarrow M$ is the vertical sub-bundle of the tangent bundle of $F$. 
Equivalently, the above diagram defines a section of the pullback bundle $\phi_0^*VF$, i.e.,
$\partial_t\phi_t |_{t=0} \in \Gamma_M(\phi_0^*VF)$. Obviously, any two $\FR^1$-plots 
$\phi_t, \phi'_t$ over $\phi_0$ define the same section of $\phi^*_0 VF$ if and only if 
they agree up to first derivatives in $t\in \FR$, at $t=0$ for each $x\in M$.

\medskip 
Following the intuition of tangent vectors as `first order infinitesimal smooth curves' in the space of fields, 
the set $\Gamma_M(\phi^*_0 VF)$ is interpreted as $T_{\phi_0}\big(\Gamma_M(F)\big)$, the set of tangent vectors 
at the field configuration $\phi_0$. The full set of tangent vectors  is 
\vspace{-2mm} 
\begin{align}
T\big(\Gamma_M(F)\big) := \bigcup_{\phi_0 \in \Gamma_M(F)} \Gamma_M(\phi_0^* VF) \cong_{\mathrm{Set}} \Gamma_M(VF) \, .
\end{align}

\vspace{-2mm} 
\noindent In  \cite{GSS-2}, we will enrich our smooth spaces with infinitesimal structure, and the 
intuition of tangent vectors as first-order infinitesimal smooth curves will become a rigorous definition, 
recovering the above set of tangent vectors - along with its natural (infinitesimally thickened) smooth structure (Def. \ref{KinematicalTangentBundleToFieldSpace}) - in a more 
direct manner.
In fact, this will be a special case of a 
general construction that applies to any `infinitesimally thickened smooth set'.  
\end{example}

\begin{example}[\bf Vector-valued field space tangent vectors]\label{ScalarFieldTheoryTangetVectors}
Since the field bundle is a \textit{trivial vector} bundle $F=M\times W$, it follows that 
\vspace{-2mm} 
\begin{align*}
T\big(\Gamma_M(F)\big) &:=\Gamma_M\big( V(M\times W) \big) \cong \Gamma_M(M\times W\times W) \\& =C^\infty(M,W\times W) \cong C^\infty(M,W)\times C^\infty(M,W) \, .
\end{align*}

\vspace{-2mm} 
\noindent Hence a tangent vector at a point in field space $\phi \in C^\infty(M,W)$ is determined by a pair of smooth W-valued maps
$(\phi,\psi)\in T_\phi \CF \subset C^\infty(M,W)\times C^\infty(M,W)$. In terms of an $\FR^1$-plot of fields, since the target $W$ 
is a vector space, this may be represented by the linear plot
$$
\phi + t\cdot \psi \;:\; \FR^1 \times M \longrightarrow W \, . 
$$
\end{example}

We motivated the definition of tangent vectors by differentiating $\FR^1$-plots in field space. However, unlike tangent vectors on 
a finite-dimensional manifold, it is not immediately clear that every tangent vector $\CZ_\phi\in T_\phi \CF = \Gamma_M(VF)$ covering 
$\phi\in \CF$ may be identified as the `derivative' of some $\FR^1$-plot $\phi_t\in \CF(\FR^1)$, in the case of an arbitrary field bundle $F\rightarrow M$.
This is indeed the case, but as the proof of the following lemma will show, this rests on delicate geometrical and topological results 
regarding smooth fiber bundles, their sections and their vector fields. We include the full details for completeness 
(a sketch of a proof also appears in \cite{Blohmann23b}).

\vspace{1mm} 
\begin{lemma}[{\bf Line-plots represent tangent vectors}]\label{LinePlotsRepresentTangentVectors}
For any $\CZ_\phi \in \Gamma_M(VF)$ covering $\phi=\pi_F \circ \CZ_\phi$
\vspace{-2mm} 
	\[ 
\xymatrix@R=1.3em{ && VF\ar[d]^{\pi_F} \\&&  F \ar[d]
	\\ 
	M\ar[rruu]^{\CZ_{\phi} } \ar[rru]^>>>>>>>{{\phi}} \ar[rr]&& M \, ,  
}   
\]

\vspace{-2mm} 
\noindent there exists a $\phi_t:M\times \FR\rightarrow F$ such that $\phi_0=\phi$ and $\partial_t \phi_t |_{t=0} = \CZ_\phi$. That is, the map
\vspace{-2mm} 
\begin{align*} \mathbold{\Gamma}_M(F)(\FR^1) &\longrightarrow T\big( \Gamma_M(F)\big)
\\[-2pt]
\phi_t &\longmapsto \partial_t \phi_t |_{t=0} 
\end{align*}

\vspace{-2mm} 
\noindent of Ex. \ref{TangentVectorsOnFieldSpace} is surjective. 
\end{lemma}
\begin{proof}
Let $\CZ_\phi \in \Gamma_M(VF)$ and $\phi=\pi_F \circ \CZ_\phi \in \Gamma_M(F)$ as above. By construction, $\phi$ is a smooth 
section of the bundle $F\rightarrow M$, and hence its image $\phi(M)\subset F$ is an embedded submanifold of $F$. Moreover, since manifolds 
are assumed to be Hausdorff, it is also \textit{closed} in $F$. Now notice that $\CZ_\phi$ is equivalently a section of $\phi^*F$ over $M$, 
and yet equivalently a section $\hat{\CZ}_\phi:\phi(M)\rightarrow VF|_{\phi(M)}$ of the vertical subbundle restricted to the (embedded) 
submanifold $\phi(M)$, since $\phi:M\xrightarrow {~} \phi(M)\subset F$ is a diffeomorphism.

Since $\phi(M)$ is a closed (and hence properly) embedded submanifold, the section $\hat{\CZ}_\phi$ can be extended to a smooth section of 
the vector bundle $VF$ over $F$ (see e.g. \cite[Ex. 10.9]{Lee}). That is, there exists a section $\psi:F\rightarrow VF$ such that
\vspace{-2mm} 
	\[ 
\xymatrix@=1.6em{& &&  VF \ar[d]^{\pi_F}
	\\ 
	\phi(M)\ar[rrru]^{\hat{\CZ}_\phi} \ar[r]& F\ar[rru]_>>>>>>{\psi}  \ar[rr]&& F \, ,  
}   
\]

\vspace{-2mm} 
\noindent commutes. Crucially, since $\phi(M)$ is closed in F, this extension can be chosen to have compact support - by potentially 
multiplying with a bump function whose value is $1$ on $\phi(M)$. Thus $\psi$ is a vector field on $F$ with \textit{compact support}, 
and hence integrates to a global flow \cite[Thm 9.16]{Lee}. 
That is, there exists a smooth 1-parameter group of diffeomorphisms 
\vspace{-2mm} 
$$
\Psi_t : \FR^1 \times F\longrightarrow F \, ,
$$

\vspace{-2mm} 
\noindent  that differentiates to $\psi$  
$$
\partial_t \Psi_t |_{t=0}= \psi \; \in \Gamma_F(VF) \, .
$$
Restricting the map on $M\times \FR^1$ via the embedding $M\times \FR^1 \xrightarrow{(\phi,\id)} \phi(M)\times \FR^1\hookrightarrow F\times \FR^1$ 
yields a section
\vspace{-2mm} 
	\[ 
\xymatrix@=1.2em{ &&&  F \ar[d]^{\pi}
	\\ 
	\FR^1\times M \ar[rrru]^-{\phi_t} \ar[rrr]^-{p_2} &&& M
}   
\]

\vspace{-2mm} 
\noindent over M, and so a line plot $\phi_t\in \CF(\FR^1)$ that differentiates to $\CZ_\phi$. 
\end{proof}
\begin{remark}[\bf Tangent vectors, paths, and derivations]
\label{TangentVectorsPathsOfFieldsAndDerivations}
Let $\phi_t\in \mathbold{\Gamma}_M(F)(\FR^1)\cong \mathrm{Hom}_{\SmoothSets}\big(y(\FR^1), \mathbold{\Gamma}_M(F)\big)$ 
be a $\FR^1$-plot 
of fields with $\phi_{t=0}=\phi \in \Gamma_M(F)$ with induced derivation
\vspace{-2mm} 
\begin{align}
    C^\infty(\CF) &\longrightarrow \FR 
    \\[-2pt]
    f &\longmapsto {\partial_t}(f\circ \phi_t) |_{t=0} \, . \nn 
\end{align}

\vspace{-2mm} 
\noindent It is obvious that the induced derivation actually depends only on the germ $\FR^1$-plots around $0\in \FR^1$. 
In other words, any two plots $\phi_t, \phi'_t: \FR^1\times M\rightarrow F$ which agree on some open $(-\epsi,\, +\epsi)\times M\subset \FR^1\times M$ 
define the same derivation. However, unlike the finite-dimensional manifold case, we see no obvious reason why this derivation 
depends only on the corresponding tangent vector at $\partial_t \phi_t|_{t=0} \in T_{\phi}(\CF)=\Gamma_M(\phi^* VF)$. 
Hence it might not necessarily descend to an action of tangent vectors at $\phi$.\footnote{Ref. \cite[\S 6.54]{IZ13} claims something similar, 
but the dependence proved is in fact between germs and tangent vectors of $\FR^1$-plots.} However, 
this \textit{does hold} for the subset of \textit{local} function(al)s $C^\infty_{\mathrm{loc}}(\CF)\subset C^\infty(\CF)$, 
which correspond to functions that factor through the infinite jet bundle, as we will make precise in the next section 
with Definition \ref{LocalFunctionsOnFieldSpace} and later on with Remark \ref{TangentVectorsAndDerivationsofLocalFunctions}. 
This potential `pathology' of tangent vectors and smooth functions
is automatically cured if we work in the better category of infinitesimally thickened smooth sets, where smooth maps preserve 
the infinitesimal structure -- by definition (see \cite{GSS-2}).  
\end{remark}

As it is, the would-be tangent bundle above is simply a set. Since it is naturally identified with sections of 
a bundle, we may apply Def. \ref{SectionsSmoothSet} to naturally view it as a smooth set.

\begin{definition}[\bf  Smooth tangent bundle to field space]\label{KinematicalTangentBundleToFieldSpace}
 The smooth tangent bundle to a field space $\CF=\mathbold{\Gamma}_M(F)$ is defined by
 \vspace{-2mm} 
\begin{align}
T\CF:= \mathbold{\Gamma}_M(VF)\, \, ,
\end{align}

\vspace{-2mm} 
\noindent as the smooth set of sections of $\Gamma_M(VF)$ via Def. \ref{SectionsSmoothSet}.
\end{definition}
In particular, an $\FR^k$-plot of the tangent bundle $T\big(\mathbold{\Gamma}_M(F)\big)$ corresponds to a pair $(\CZ_{\phi^k_0},\phi^k_0)$ 
of $\FR^k$-parametrized sections over $M$ such that 
\vspace{-2mm} 
	\[ 
\xymatrix@=1.6em{ &&  VF \ar[d]
	\\ 
	\FR^k \times M\ar[rru]^{\CZ_{\phi^k_0} } \ar[rr]^>>>>>>>{{\phi}^k_0} && F \, 
}   
\]

\vspace{-2mm} 
\noindent commutes. Note that this immediately shows that the fiber-wise  $\FR$-linear structure of the plain set bundle
$\Gamma_M(VF) \rightarrow \Gamma_M(F)$ extends to a smooth $y(\FR)$-linear map 
\vspace{-2mm}
\begin{align}\label{FiberWiseLinearStructure}
+ \;:\; T\CF \times T\CF &\longrightarrow T\CF \\
\Big(\CZ^1_{\phi_0^k}, \, \CZ^2_{\phi^k_0}\Big)&\longmapsto \CZ^1_{\phi_0^k}+ \CZ^2_{\phi^k_0} \nn
\end{align}
\vspace{-2mm}

By arguing as in Lem. \ref{LinePlotsRepresentTangentVectors}, any such map $\CZ_{\phi^k_0}:\FR^k\times M \rightarrow VF$ 
may be (non-uniquely) represented\footnote{An $\FR^k$-plot of the tangent bundle $T\CF$ is equivalently a section $\hat{\CZ}_{\phi^k_0}$ 
of the pullback bundle $V(\pr_2^* F) \cong \pr_2^*(VF) \rightarrow \FR^k\times M$. The proof of Lem. \ref{LinePlotsRepresentTangentVectors}
then applies verbatim to produce a map $\Psi^k_t: \FR^1\times \pr_2^* F\rightarrow \pr_2^* F$ which differentiates to $\hat{\CZ}_\phi$ 
along the image of $\phi^k_0$.} by an $\FR^1\times \FR^k$-plot $\phi^k_t$ of $\mathbold{\Gamma}_M(F)$ such that 
$\phi^k_{t=0}= \phi^k_0$, and so we might often write 
\vspace{-2mm} 
\begin{align}\label{PlotsOfTangentBundleViaCurvesofPlots}
\CZ_{\phi^k_0} = \partial_t \phi^k_t |_{t=0} \, .
\end{align}

\vspace{-2mm} 
\noindent As expected, the evident projection $\Gamma_M(VF)\rightarrow \Gamma_M(F)$ extends to  smooth projection map 
\vspace{-2mm} 
\begin{align*}\pi_\CF \;:\; T\CF&\longrightarrow
\CF \\
\big(\CZ_{\phi^k_0},\phi^k_0\big)&\longmapsto \phi^k_0
	\, . 
\end{align*}

\vspace{-2mm} 
\noindent Hence, we may define smooth vector fields in this infinite-dimensional setting as actual geometrical smooth sections of 
the tangent bundle.
\begin{definition}[\bf  Vector fields on field space]\label{VectorFieldsOnFieldSpace}
The set of smooth vector fields on the field space $\CF=\mathbold{\Gamma}_M(F)$ is defined as smooth sections of 
its tangent bundle
\vspace{-2mm} 
\begin{align}
\CX(\CF):= \big\{\CZ:\CF\rightarrow T\CF \, |\, \pi_\CF \circ \CZ  = \id_\CF \big\} \, .
\end{align}
That is, smooth maps $\CZ: \CF\rightarrow T\CF $ such that the diagram of smooth sets
\vspace{-2mm} 
	\[ 
\xymatrix@=1.6em{ &&  T\CF \ar[d]^{\pi_\CF}
	\\ 
	\CF \ar[rru]^{\CZ } \ar[rr]^>>>>>>{\id_\CF} && \CF \, 
}   
\]

\vspace{-2mm} 
\noindent commutes.
\end{definition}
Let us unwind the definition slightly. On $*$-plots, such a section defines a map of sets
\vspace{-2mm} 
\begin{align*}
\CZ \,:\, \Gamma_M(
F)&\longrightarrow \Gamma_M(VF) 
\\[-2pt]
\phi&\longmapsto \CZ_\phi
\end{align*}

\vspace{-2mm} 
\noindent which assigns a tangent vector $\CZ_\phi\in T_\phi\big(\Gamma_M(F)\big)=\Gamma_M(\phi^*VF)$ to every field 
configuration $\phi\in \Gamma_M(F)$, and so a vector field in the intuitive sense. The smoothness condition is 
then the further requirement that under this point-wise assignment\footnote{This interpretation of smoothness only 
holds for concrete smooth sets, i.e., diffeological spaces with an underlying notion of points \cite{BH11}. 
In particular, it does not generalize to non-concrete smooth sets as stated, and hence will also not generalize 
when working with smooth fermionic spaces \cite{GSS-2}.}, \textit{smooth} $\FR^k$-plots of field configurations, 
i.e., smoothly $\FR^k$-parametrized sections $\phi^k:\FR^k\times M\rightarrow F$, are mapped to \textit{smooth} 
$\FR^k$-plots of tangent vectors, i.e., smoothly $\FR^k$-parametrized sections $\CZ_{\phi^k}: \FR^k\times M \rightarrow VF$. 
The most common examples of smooth vector fields on field spaces arising in physics are \textit{local} vector fields, 
that is vector fields that factor through the infinite jet bundle of $J^\infty_M F$ - in a sense that we will make fully 
precise in  \cref{EvolutionaryVectorFieldsAndNoetherTheoremsSection}.

 \medskip 
Vector fields on field space $\CF$ should be interpreted as `infinitesimal smooth diffeomorphisms', in direct analogy with the finite-dimensional case, 
as will become clear in the next examples. Firstly, the definition of diffeomorphisms in smooth sets is immediate.

\begin{definition}[\bf  Diffeomorphisms of the field space]
\label{SmoothAutomorphisms}
The set of \textit{diffeomorphisms} of $\CF$, i.e., smooth automorphisms of $\CF$,  is defined as the subset invertible self-morphisms 
of smooth sets,
\vspace{-1mm} 
$$
\mathrm{Diff}(\CF):= \big\{\CP:\CF\rightarrow \CF \; |\; \exists\,  \CP^{-1}:\CF \rightarrow \CF  \big\}\,\subset \mathrm{Hom}_{\SmoothSets}(\CF,\CF)  \, .
$$
\end{definition}
In finite-dimensional smooth manifolds, vector fields can be obtained by differentiating smooth 1-parameter families of diffeomorphisms, i.e.,
smooth paths of diffeomorphisms. This analogously carries over to the case of the field space $\CF$. A useful intermediate step to deduce this differentiation process is via the smooth set of paths in the field space, constructed as an application of Def. \ref{SmoothSetinternalHom}.

\begin{definition}[\bf  Path space of fields]\label{PathSpaceOfFields}
The \textit{smooth path space} of fields is defined as 
\vspace{0mm}
$$P(\CF):=\big[y(\FR^1), \CF \big] \, .
$$

\vspace{-2mm}
\noindent
That is, 
the smooth set with $\FR^k$-plots\footnote{In this description, maps of probes $f:\FR^{k'}\rightarrow \FR^k$ act on plots via the pullback $(\id_{\FR^1}\times f)^*:\mathbold{\Gamma}_M(F)(\FR^1\times \FR^k)\rightarrow \mathbold{\Gamma}_M(F)(\FR^1\times \FR^{k'})$.}
\vspace{-1mm} 
$$
\big[y(\FR^1), \mathbold{\Gamma}_M(F)\big](\FR^k):= \mathrm{Hom}_{\SmoothSets}\big(y(\FR^1\times \FR^k)\, ,
\, \mathbold{\Gamma}_M(F)\big)\cong \mathbold{\Gamma}_M(F)(\FR^1\times \FR^k) \, .
$$
\end{definition}

\vspace{-1mm}
For each $t_0\in \FR^1_t$ there is a corresponding smooth projection 
\vspace{-2mm} 
\begin{align*}\mathrm{ev}_{t_0} \,:\, P(\CF)&\longrightarrow \CF
\\[-2pt]
\phi^k_t &\longmapsto \phi^k_{t=t_0}
\end{align*}

\vspace{-2mm} 
\noindent which evaluates each path (of $\FR^k$-plots) of fields at $t_0\in \FR^1$. Furthermore, the projection at $t_0=0$ obviously 
factors through the smooth tangent bundle of $\mathbold{\Gamma}_M(F)$ via
\vspace{-3mm} 
\begin{align*}\partial_t|_{t=0} \;:\; P(\CF)&\longrightarrow T(\CF)
\\[-2pt]
\phi^k_t&\longmapsto \partial_t \phi^k_{t}|_{t=0} \, ,
\end{align*}
which by Lem. \ref{LinePlotsRepresentTangentVectors} (and its application to $\FR^k$-plots) is an epimorphism (see Eq. \eqref{PlotsOfTangentBundleViaCurvesofPlots} and its footnote).

\begin{example}[\bf Vector fields as infinitesimal diffeomorphisms]
\label{VectorFieldsAsInfinitesimalDiffeomorphisms}
There is a natural notion of a smooth 1-parameter family of diffeomorphisms on $\CF$, making use of the canonical smooth mapping 
space $[\CF,\CF]$ of Def. \ref{SmoothSetinternalHom}. This is a $\FR^1$-plot of self-morphisms
$$
\CP_t \in [\CF,\CF](\FR^1):=\mathrm{Hom}_{\SmoothSets}\big(y(\FR^1)\times \CF,\CF\big) 
\cong\mathrm{Hom}_{\SmoothSets}\big(\CF \,, \,  P(\CF)  \big) 
$$ 
that `is a diffeomorphism for each $t_0\in \FR^1$', where the latter isomorphism is the internal hom property. \footnote{Alternatively (and equivalently), 
there is a canonical way to supply $\mathrm{Diff}(\CF)$ with a smooth set structure and then consider its line plots instead. 
In fact, there exists a sub-object of automorphisms  $\mathbold{Aut}(\CF)\hookrightarrow [\CF,\CF] $ of the self-mapping space
in any sheaf category, whose plots are given by
\vspace{-1mm}
$$\mathbold{Aut}(\CF)(U):=\big\{\tilde{f}=(\pr_U\times f) :y(U)\times \CF \rightarrow y(U)\times \CF \, \, |\,\, \exists \,
\tilde{f}^{-1}=(\pr_U\times g): y(U)\times \CF \rightarrow y(U)\times \CF \big\}.$$}
 Explicitly, it is a smooth map 
 \vspace{-3mm} 
\begin{align*}
\CP_t \,:\, \CF  &\longrightarrow \mathbold{P}(\CF) \\[-2pt]
\phi^k&\longmapsto \phi^k_t
\end{align*}

\vspace{-2mm} 
\noindent such that the composition along the projections
$\mathrm{ev}_{t_0}: P(\mathbold{\Gamma}_M(F))\rightarrow \mathbold{\Gamma}_M(F)$

\vspace{-3mm} 
\begin{align*}
\CP_{t=t_0}\,:\, \CF  &\longrightarrow \CF\\ 
\phi^k&\longmapsto \phi^k_{t=t_0}
\end{align*}

\vspace{-2mm} 
\noindent is invertible for any $t_0\in \FR^1$. Consider now any such 1-parameter family $\CP_t\in [\CF,\CF](\FR^1)$ that 
starts at the identity map, i.e., $\CP_{t=0}=\id_\CF$. Equivalently, this is a section of the bundle 
$\mathrm{ev}_{t=0}:\mathbf{P}(\CF)\rightarrow \CF$ and hence composing along the projection
$\partial_t|_{t=0}:P(\CF)\rightarrow T(\CF) $ yields a vector field
\vspace{-3mm} 
\begin{align*}
\partial_t|_{t=0}\circ \CP_t \;:\; \CF&\longrightarrow T(\CF) \\
\phi^k &\longmapsto \partial_t(\phi^k_t)|_{t=0}
\end{align*}

\vspace{-2mm} 
\noindent on $\CF$. Hence the `infinitesimal part' of a 1-parameter diffeomorphism connected to the identity is indeed a vector 
field on the field space $\CF$. On the other hand, unlike the finite-dimensional manifold case, one should not in general expect 
to be able to integrate vector fields to 1-parameter families of diffeomorphisms. \footnote{In finite dimensions, integral 
curves along vector fields always exists - at least for a small `time' interval $(-\epsi,\epsi)\subset \FR^1$. This is not necessarily true in
the infinite-dimensional smooth set setting.}
\end{example}

Note that any smooth 1-parameter family of diffeomorphisms $\CP_t:\CF \rightarrow \mathbold{P}(\CF)$ connected to the identity defines a derivation
\vspace{-4mm} 
\begin{align}\label{PathSectionDerivation}
    C^\infty(\CF) &\longrightarrow C^\infty(\CF)      
\end{align}

\vspace{-2mm} 
\noindent where a function $\phi^k\mapsto g\circ \phi^k$ is mapped to the smooth function $\phi^k\mapsto {\partial_t}(g\circ \phi^k_t)|_{t=0}$, extrapolating 
Eq. \eqref{SmoothPathInducedDerivation}. As with Remark \eqref{TangentVectorsPathsOfFieldsAndDerivations}, this derivation is not necessarily determined 
by the induced vector field $\partial_t|_{t=0}\circ \CP_t \in \CX(\CF)$. However, this \textit{does hold} on the subset of local functions 
on $\CF$, where every \textit{local} vector field defines a derivation of $C^\infty_{\mathrm{loc}}(\CF)$, as will be detailed in Rem. 
\ref{TangentVectorsAndDerivationsofLocalFunctions}. More generally, this potential pathology of vector fields will be automatically cured 
when we consider our spaces as thickened smooth sets \cite{GSS-2}, where every vector field will necessarily define a derivation. 

\begin{example}[\bf Diffeomorphisms for vector-valued field space via the target]\label{DiffeomorphismForVector-valuedFieldTheoryViaTarget}
In the case of vector-valued field space with $F=M\times W$ (Ex. \ref{VectorFieldTheoryFieldSpace}), bundle automorphisms correspond
to diffeomorphisms of the target $g:W\rightarrow W$.
\footnote{This holds for any $\sigma$-model field space, i.e., with target any manifold $N\in \SmoothManifolds$.}
For instance, general linear transformations constitute a particular example, whereby choosing a basis $\{e_a\}_{a=1,\cdots, N}$ for $W\cong \FR^N$, 
\vspace{-2mm} 
\begin{align*}
g:W&\longrightarrow W \, 
\\[-2pt]
 e_a &\longmapsto g^{b}_{\, \, a} \cdot e_b
\end{align*}

\vspace{-2mm} 
\noindent for $\big[g^b_{\,\,a}\big]\in \mathrm{GL}(n,\FR)$. There is an induced diffeomorphism on field space given by postcomposition 
of (plots of) fields
\vspace{-2mm} 
\begin{align*}
g_*: \CF &\longrightarrow \CF 
\\[-2pt]
 \phi^k &\longmapsto g\circ \phi^k \, .
\end{align*}

\vspace{-2mm} 
\noindent  It follows that a 1-parameter family of such transformations, i.e., 
\vspace{-2mm} 
\begin{align*}
g_t\; :\; \FR^1\times W&\longrightarrow W \, 
\\[-1pt]
 (t, w^a \cdot e_a ) &\longmapsto w^a \cdot g^{b}_{\, \, a}(t) \cdot e_b
\end{align*}

\vspace{-2mm} 
\noindent 
for some $\big[g^b_{\,\, a}(t)\big]:\FR^1 \rightarrow \mathrm{GL}(n,\FR)$, defines a $1$-parameter family of field space diffeomorphisms
\vspace{-2mm} 
\begin{align*}
\CP_{t}=(g_t)_* \;:\; \CF&\longrightarrow \mathbold{P}(\CF) \, \\
\phi^k =\phi^{k,a} \cdot e_a &\longmapsto \phi^k_t = \phi^{k,a} \cdot g^{b}_{\, \, a}(t) \cdot e_b\, . 
\end{align*}

\vspace{-1mm} 
\noindent 
Assuming $g_{t=0}=\id_W$ is the identity map on $W$, with $\big[A^{b}_{\, \, a}\big]=\big[\dot{g}^b_{\, \, a}(0)\big] \in \mathfrak{gl}(n,\FR)$ 
the induced Lie algebra element, the differentiation of the previous example (Ex. \ref{VectorFieldsAsInfinitesimalDiffeomorphisms}) 
explicitly yields the smooth vector field $\CZ^A= \partial_t|_{t=0}\circ(g_t)_*$ on field space,
\vspace{-2.5mm} 
\begin{align*}
\CZ^A \;:\; \CF\cong [M,W] &\longrightarrow T(\CF)\cong [M,W\times W] \, \\
\phi^k = \phi^{k,a} \cdot e_a &\longmapsto \big(\phi^k \, ,\, \phi^{k,a} \cdot A^b_{\, \, a} \cdot e_b\big)\, . 
\end{align*}

\vspace{-2mm} 
\noindent This is a rigorous way to make sense of the smooth vector field which in the physics literature is often denoted 
abusively\footnote{Via an (a priori) unjustified analogy to the coordinate formulas of smooth vector fields in the 
finite-dimensional setting.} by 
\vspace{-2mm} 
$$ 
\CZ^A(\phi)= \phi^a \cdot A^{b}_{\, \, a} \cdot  \frac{\delta }{\delta \phi^b}\, ,
$$

\vspace{-1mm} 
\noindent  and more often as an `infinitesimal transformation of the field'
\vspace{-1mm} 
$$
\delta_A \phi^b = A^b_{\,\, a} \cdot \phi^a \, .
$$ 
  
\noindent 
Completely analogously, any translation on $W$ acting as $w\mapsto w+c$ for some $c=c^a \cdot e_a \in W$, defines a diffeomorphism 
on field space by postcomposition. The induced vector field of such (1-parameter families of) translations is often denoted as
\vspace{-2mm}
$$
\CZ^c(\phi)=c^a \cdot \frac{\delta}{\delta \phi^a}\,  \hspace{1.5cm} \mathrm{or} \hspace{1.5cm} \delta_{c} \phi^a = c^a \, ,
$$

\vspace{-2mm} 
\noindent and is in particular `constant' in field space.
\end{example}

\begin{example}[\bf Diffeomorphisms for vector-valued field space via the base]\label{DiffeomorphismForVector-valuedFieldTheoryViaBase}
Consider the case of a trivial field bundle $F=M\times W$ and let $f:M\rightarrow M$ diffeomorphism of the base spacetime. 
There is an induced diffeomorphism on field space acting by precomposition of (plots of) fields
\vspace{-2mm} 
\begin{align*}
f^*:[M,W] &\longrightarrow [M,W]\\
     \phi^k &\longmapsto \phi^k \circ \big(\id_{\FR^k}, f\big)\, .
\end{align*}

\vspace{-2mm} 
\noindent It follows that a smooth $1$-parameter family of spacetime diffeomorphisms 
$
f_t \,:\, \FR^1 \times M\longrightarrow M
$
induces a $1$-parameter family of field space diffeomorphisms
\vspace{-3mm} 
\begin{align*}
(f_t)^* \;:\; \CF&\longrightarrow \mathbold{P}(\CF) \, \\
\phi^k &\longmapsto \phi^k_t = \phi^k \circ \big(\id_{\FR^k}, f_t\big)\, . 
\end{align*}

\vspace{-2mm} 
\noindent Assuming $f_{t=0}=\id_M$ is the identity map on $M$, with induced vector field 
$v= \partial_t f_t |_{t=0} =v^\mu \cdot \frac{\partial}{\partial x^\mu}\in \Gamma(TM)$ on spacetime, the differentiation 
of Ex. \ref{VectorFieldsAsInfinitesimalDiffeomorphisms} immediately yields the vector field $\CZ^v= \partial_t |_{t=0} \circ (f_t)^*$,
\vspace{-2mm} 
\begin{align*}
\CZ^v \;:\; [M,W] &\longrightarrow [M,W\times W] \, \\
\phi^k  &\longmapsto \big(\phi^k \, ,\, \dd_M \phi^k\circ (\id_{\FR^k},v) \big)\, , 
\end{align*}
where $\dd_M \phi^k : \FR^k\times TM\rightarrow TW\cong W\times W$ is the de Rham differential on $M$. In local coordinates for $M$ 
and a basis for $W$,  $\dd_M \phi^k = \partial_\mu \phi^{k,a} \cdot \dd x^\mu \cdot e^a$, and so 
$\dd_M \phi \circ v = \mathbb{L}_\nu(\phi^a) \cdot e_a =  \partial_\mu \phi^a \cdot v^\mu \cdot e_a $, where $\mathbb{L}_\nu$ is the Lie derivative 
along $\nu\in \Gamma(TM)$. In the physics literature, such vector fields are denoted by
\vspace{-2mm} 
$$
\CZ^v(\phi) = \mathbb{L}_v(\phi^a) \cdot \frac{\delta}{\delta \phi^a}= \partial_\mu \phi^a \cdot v^\mu \cdot \frac{\delta}{\delta \phi^a}\, , 
$$

\vspace{-2mm} 
\noindent  or even by
$$
\delta_\nu \phi^a = \mathbb{L}_\nu (\phi)^a = \nu^\mu \cdot \partial_\mu \phi^a \, ,
$$
and one says ``the field $\phi$ transforms as a scalar". More than often in physics, the spacetime $M$ is supplied with a metric $g$, 
and the vector field $v\in \Gamma(TM)$ a Killing vector field. A particular instance is Minkowski space $(M,g)=(\FR^4,\eta)$, 
with Killing vector fields $v\in \mathfrak{iso}(1,3)$ being elements of the Poincar\'e Lie algebra. 
\end{example}

\begin{remark}[\bf Diffeomorphisms via field bundle automorphisms]\label{DiffeomorphismsViaFieldBundleAutomorphisms}
Generally, any field bundle automorphism $\tilde{f}:F\rightarrow F$
\vspace{-2mm} 
\[
\xymatrix@R=1em@C=1.6em  {F \ar[d] \ar[rr]^{\tilde{f}}  &&   F \ar[d]
	\\ 
M\, \ar[rr]^{f}  && M
} 
\]

\vspace{-1mm} 
\noindent
covering a diffeomorphism $f:M\rightarrow M$ (not necessarily the identity), yields a field space diffeomorphism
\vspace{-2mm} 
\begin{align*}
\CF &\longrightarrow \CF
\\[-2pt]
 \phi^k &\longmapsto \tilde{f} \circ \phi^k \circ (\id_{\FR^k}, f^{-1}) \, .
\end{align*}

\vspace{-1mm} 
\noindent 
One may consider 1-parameter families of such, and differentiate to obtain smooth vector fields on $\CF$. In this picture, 
Ex. \ref{DiffeomorphismForVector-valuedFieldTheoryViaTarget} corresponds to the case of the trivial bundle $F=M\times W$ with $\tilde{f}=(\id_M, g)$ 
covering the identity on $M$, while Ex. \ref{DiffeomorphismForVector-valuedFieldTheoryViaBase} corresponds to 
\vspace{-2mm} 
\[
\xymatrix@R=1em@C=1.6em  {M\times W \ar[d] \ar[rr]^{(f,\, \id_W)}  &&   M\times W \ar[d]
	\\ 
M\, \ar[rr]^{f}  && M
\, . } 
\]

\noindent We note, however, that for a general field fiber bundle $F\rightarrow M$ and a given diffeomorphism $f:M\rightarrow M$, the pullback $f^*\phi$ of 
a field is a section of $f^*F\rightarrow M$ and \textit{not} of $F\rightarrow M$. That is, to act on the actual space of fields via $\mathrm{Diff}(M)$, 
it is indeed necessary to lift the given spacetime diffeomorphism to a bundle morphism $\tilde{f}:F\rightarrow F$. However, this is not always possible, 
contrary to the case of the trivial field bundle. Bundles with such a lifting property are called ``gauge natural" bundles\footnote{An archetypical example 
is the tangent/cotangent bundles and tensor powers thereof. The prime example of a field space employing natural bundles is pure gravity in its 
$(0,2)$-tensor metric formulation, whereby the field space consists of (symmetric, non-degenerate) sections of $\wedge^2 T^*M$, i.e., 
$\mathrm(Met)_M:= \Gamma^{\mathrm{sym. nondeg.}}_M(\wedge^2 T^*M)$. Any diffeomorphism $f: M\rightarrow M$ lifts to a bundle morphism which allows 
to pull back metrics, defining a smooth diffeomorphism $\CP:= f^*: \mathrm(Met)_M \rightarrow \mathrm(Met)_M$. Strictly speaking, this is only valid 
for pure gravity, or at most coupled to fields that are sections of trivial or natural bundles. In particular, it breaks down when (fermionic) 
spinors  enter the picture.} (see \cite{KMS}). 
Nevertheless, for many purposes of field theory 
(e.g., for computing conserved currents, see \cref{EvolutionaryVectorFieldsAndNoetherTheoremsSection}) it suffices to lift only the infinitesimal version 
of spacetime diffeomorphisms, i.e., vector fields on $M$. That is, there exists a short exact sequence
\vspace{-2mm} 
\begin{align*}
0\longrightarrow \Gamma_F(VF) \longrightarrow \Gamma_F(TF) \longrightarrow \Gamma_F(\pi^*TM)\longrightarrow 0 \, ,
\end{align*}

\vspace{-2mm} 
\noindent 
a splitting of which can be used to lift infinitesimal diffeomorphisms on spacetime, i.e., elements of $\Gamma_M(TM)$.\footnote{Splitting 
as vector spaces is sufficient. 
Indeed, such a lift given by a connection on $F\rightarrow M$ is a Lie algebra homomorphism, if and only if the connection is flat.} Such a splitting is a \textit{connection}
on $F\rightarrow M$, which always exists. In the relevant field bundles of physics, the corresponding connection is often fixed as a background structure in the choice of 
the Lagrangian density (Def. \ref{LocalLagrangianDensity}), or is otherwise a dynamical field itself. The data in this situation is enough to define an induced vector 
field on $\CF$, from that on the spacetime $M$ (see \cite[\S 5.3]{Gregory}
for a particular instance in the case of the first-order formulation of general relativity, and sources therein). 
We will not expand on this intricacy further in this manuscript.
\end{remark}
Having defined a notion of a smooth tangent bundle $T\CF$ on a field space $\CF$, which does have a natural fiber-wise linear structure (Eq. \eqref{FiberWiseLinearStructure}), there is a natural corresponding notion of differential forms, as fiber-wise linear and antisymmetric maps out of $T\CF$. 
\begin{definition}[\bf  Differential forms on field space]\label{DifferentialFormsOnFieldSpace} 
The set of differential m-forms on $\CF=\mathbold{\Gamma}_M(F)$ is defined as
\begin{align}\Omega^m (\CF) := \mathrm{Hom}^{\mathrm{fib.lin. an.}}_{\SmoothSets}
\big(T^{\times m}\CF\,,\, \FR \big) \, , 
\end{align} 

\vspace{-1mm} 
\noindent i.e., smooth real-valued, fiber-wise linear antisymmetric maps with respect to the fiber-wise linear structure (Eq. \eqref{FiberWiseLinearStructure}) on the $m$-fold fiber product
\vspace{0mm} 
$$
T^{\times m}\CF:= T\CF\times_{\CF}\cdots \times_{\CF} T\CF
$$ 

\vspace{0mm} 
\noindent of the tangent bundle over the $\CF$.  
\end{definition}
Exactly as with finite-dimensional manifolds, the collection of differential forms of all degrees forms a graded $\FR$-vector space
\vspace{-2mm}
$$
\Omega^\bullet(\CF):= \bigoplus_{m\in \NN} \Omega^{m}(\CF )\, ,
$$

\vspace{-1mm}
\noindent with a well-defined notion of a wedge product $\wedge$ and contraction $\iota_X$ for any vector field $\CZ \in \CX(\CF)$. What is not at all obvious is the existence of a differential along $\dd_\CF: \Omega^\bullet(\CF) \rightarrow \Omega^{\bullet+1}(\CF)$ (see Rem. \ref{CaveatsWithTheBicomplexOfProduct}). Nevertheless, we will see in \cref{TheBicomplexOfLocalFormsSection} that this is precisely the correct notion of differential forms to host the subspace of  smooth and \textit{local} forms on the field 
space, whereby such a differential exists and satisfies a Cartan calculus with respect to \textit{local} vector fields. This will encode in a rigorous manner 
the kind of differential forms and their manipulations exactly as they are performed in the physics literature.

\begin{remark}[\bf Notions of tangent bundles]The construction of a kinematical tangent bundle from Ex. \ref{TangentVectorsOnFieldSpace} 
and Def. \ref{KinematicalTangentBundleToFieldSpace} can be generalized to any diffeological space \cite{Hector}, i.e., 
any smooth set with underlying set of `enough' points. 
However, in general, there are different choices for the smooth set structure on the resulting tangent bundle set 
\cite{CSW}\cite{ChristensenWu}.
 Even worse, one may also define tangent vectors extrinsically as derivations of the corresponding smooth function algebra \cite{CSW}\cite{ChristensenWu}, since as we 
 alluded to the two notions are not necessarily equivalent. As this text will make clear, the correct and useful notion of 
 the tangent bundle to a field space is that of Def. \ref{KinematicalTangentBundleToFieldSpace}. In fact, will show how we may define the tangent bundle 
 to the field space, and in fact any (thickened) smooth set, 
 as a particular mapping space  smooth set in a canonical manner via synthetic differential geometry in \cite{GSS-2}, 
 bypassing these ambiguities and recovering the above discussion. 
 \end{remark} 

\subsection{Classifying space of de Rham forms}\label{ClassifyingSpaceSection}
All the previous examples of smooth sets have a notion of underlying \textit{sets of points}, in that
$\CG(\FR^k)\subset \mathrm{Hom}_{\mathrm{Set}}(\FR^k,G_{s})$ for some plain set $G_{s}\in \mathrm{Set}$, for each probe space 
$\FR^k\in \mathrm{CartSp}$ (cf. Def. \ref{SmoothSetinternalHom}). We say such smooth sets are 
\textit{concrete}. The full subcategory of concrete smooth sets is naturally identified with \textit{diffeological spaces} 
 \cite{IZ13} \cite{BH11}. However, not all smooth sets are concrete and in fact 
there are useful non-concrete generalized smooth spaces.\footnote{This is, in fact, generically the case for fermionic field spaces 
which constitute \textit{super} smooth sets \cite{GSS-2}.} One such smooth set is the ``moduli space of de Rham n-forms'', which in particular
allows for an alternative definition of forms on an arbitrary smooth set.
\begin{definition}[\bf  $n$-forms as smooth set]
\label{nFormsSmoothSet}
For each $n\in \NN$, we define the \textit{moduli space of de Rham $n$-forms} $\mathbold{\Omega}_{\mathrm{dR}}^{n}\in \SmoothSets$ by
\vspace{-4mm} 
\begin{align}
 \mathbold{\Omega}_{\mathrm{dR}}^{n}(\FR^k):= \Omega^{n}(\FR^k) \, .
\end{align}

\vspace{-2mm} 
\noindent 
That is, the assignment of $n$-forms on each $\FR^k\in\mathrm{CartSp}$.
\end{definition} 
This is a presheaf, since forms pull back along maps of manifolds, and further a sheaf since locally defined forms glue.  
Naturally, the equivalent sheaf $\mathbold{\Omega}_{\mathrm{dR}}^n$ over the category of all manifolds is given by the same formula. 
Note that for $n=0$ this sheaf is the same as the embedding of $\FR$ into smooth sets, since
\vspace{-1mm} 
$$
\mathbold{\Omega}_{\mathrm{dR}}^{0}(\FR^k)=\Omega^0(\FR^k)=C^{\infty}(\FR^k)=y(\FR)(\FR^k)\, .
$$

\vspace{-2mm} 
\noindent In particular, $\mathbold{\Omega}_{\mathrm{dR}}^0$ is a concrete smooth set. For $n>0$, $\mathbold{\Omega}_{\mathrm{dR}}^{n}$ is no longer 
concrete, and in fact quite unintuitive from the point-set perspective. Indeed, for a fixed $n \in \NN$, 
$$
\mathbold{\Omega}_{\mathrm{dR}}^{n}(\FR^k)=\Omega^{n}(\FR^k)=\{*\}
$$

\vspace{0mm} 
\noindent for all $k<n$ and so there is a single point in $\mathbold{\Omega}_{\mathrm{dR}}^{n}$, with all $p$-dimensional `plots' 
degenerating to the same point for $k<n$. However, for $k\geq n$
$$
\mathbold{\Omega}_{\mathrm{dR}}^{n}(\FR^k)= \Omega^{n}(\FR^k)\, ,
$$ 

\vspace{-0mm} 
\noindent is an infinite set, and so there exists an infinite number of $k$-dimensional plots in 
$\mathbold{\Omega}_{\mathrm{dR}}^{n}$ for $k\geq n$. 

\smallskip 
We now consider operations on the moduli spaces. For any $n\geq 0$, there exists a smooth map incarnation of the de Rham differential
\vspace{-3mm} 
\begin{align}\label{smoothdeRham}
\dd_{\mathrm{d R}}: \mathbold{\Omega}_{\mathrm{dR}}^{n}&\longrightarrow \mathbold{\Omega}_{\mathrm{dR}}^{n+1}
\end{align}

\vspace{-2mm} 
\noindent defined plot-wise by
\vspace{-2mm} 
\begin{align*}
	\Omega^n(\FR^k)&\longrightarrow \Omega^{n+1}(\FR^k) \\
	\om_{\, \FR^k}&\longmapsto \dd_{\FR^k} \om_{\, \FR^k}
\end{align*}

\vspace{-2mm} 
\noindent for each $\FR^k\in \mathrm{CartSp}$ and $\om_{\FR^k}\in \Omega^n(\FR^k)$, with $\dd_{\FR^k}$ 
being the usual
de Rham differential on $\FR^k$. This constitutes a smooth map, since the usual de Rham differential 
commutes with pullbacks of manifolds. Thus, the collection of the moduli spaces of $n$-forms for $n\geq 0$ 
inherits a cochain complex structure, now in smooth sets $
\big(\mathbold{\Omega}_{\mathrm{dR}}^{\bullet}, \dd\big) \in \mathrm{Ch}(\SmoothSets)$,
with the module structure being that of the smooth set real numbers $y(\FR)\in \SmoothSets$.
Similarly, for any $n,m\geq 0$ there exists a smooth map incarnation of the wedge product
\vspace{-2mm} 
\begin{align}\label{SmoothClassifyingWedgeProduct}
 \wedge \;:\; \mathbold{\Omega}_{\mathrm{dR}}^{n}\times \mathbold{\Omega}_{\mathrm{dR}}^{m} \longrightarrow \mathbold{\Omega}_{\mathrm{dR}}^{n+m}   
\end{align}

\vspace{-2mm} 
\noindent defined plot-wise by
\vspace{-2mm} 
\begin{align*}
\Omega^{n}(\FR^k)\times \Omega^{m}(\FR^k)&\longrightarrow \Omega^{n+m}(\FR^k) 
\\
\big(\omega_{\FR^k}, \omega'_{\FR^k}\big)&\longmapsto \omega_{\FR^k} \wedge_{\FR^k} \omega'_{\FR^k} 
\end{align*}

\vspace{-1mm} 
\noindent  for each $\FR^k\in \mathrm{CartSp}$, with $\omega_{\FR^k} \in \Omega^{n}(\FR^k)$ and 
$\omega'_{\FR^k}\in \Omega^{m}(\FR^k)$, respectively. It follows that $\dd$ is a (graded) differential with respect to $\wedge$, 
due to the corresponding plot-wise property. The above natural 
maps of smooth sets can be summarized in the following lemma.

\begin{lemma}[{\bf Induced structure on the moduli space}]\label{InducedStructureOnTheModuliSpace}
The differential graded commutative $\FR$-algebra (DGCA) structure of forms $\Omega^\bullet(\FR^k)$ on each 
Cartesian space $\FR^k\in \mathrm{Cart}$  induces a DGCA structure on the moduli space of forms 
$\mathbold{\Omega}_{\mathrm{dR}}^\bullet\in \SmoothSets$. That is, the triple 
$\big(\mathbold{\Omega}_{\mathrm{dR}}^\bullet, \dd_{\mathrm{d R}}, \wedge\big)$ forms a DGCA  \footnote{Explicitly, this means that the maps of smooth sets $\dd_{\mathrm{dR}},\, \wedge, \,+, \, $ the unit $1_{\FR}: y(\FR) \longrightarrow \Omega^0_{\mathrm{dR}}\hookrightarrow \Omega^{\bullet}_{\mathrm{dR}}$, and 
the corresponding counit satisfy precisely the same structural equations as those of a DGCA of real vector spaces.} to
smooth sets,
\vspace{-2mm} 
\begin{align}\label{ClassifyingFormsDGCAstructure}
\big(\mathbold{\Omega}_{\mathrm{dR}}^\bullet, \dd_{\mathrm{d R}} , \wedge \big)  \in \mathrm{DGCA}(\SmoothSets) \, ,
\end{align}

\vspace{-2mm} 
\noindent as a module over $y(\FR)\in \SmoothSets$.
\end{lemma}
The name ``moduli space" is justified, for the moment only partially, for the following reason. 
For any smooth manifold $M\in \SmoothManifolds$ viewed as a smooth set, we have
\vspace{-2mm} 
\begin{align}\label{ClassifyingFormsClassifyManifoldForms}
\mathrm{Hom}_{\SmoothSets}(y(M),\,   
 \mathbold{\Omega}_{\mathrm{dR}}^{n})\cong 
\Omega^{n}(M)\cong \mathrm{Hom}_{\SmoothManifolds}^{\mathrm{fib.lin.}}( T^{\times n}M,\,  \FR)
\end{align}

\vspace{-1mm} 
\noindent where the first isomorphism follows by Yoneda Lemma \ref{YonedaLemmaForSmoothSets} and the equivalence $\mathrm{Sh(Cart)}\cong \mathrm{Sh(Man)}$. 
That is, smooth maps from a manifold $M$ to the smooth set $\mathbold{\Omega}_{\mathrm{dR}}^{n}$ coincide with the set of 
$n$-forms on $M$, hence $\mathbold{\Omega}_{\mathrm{dR}}^{n}$ indeed serves as a space that `modulates $n$-forms on manifolds'. 
We may extend the above identification to a definition of $n$-forms which applies to any smooth set 
$\CF\in \SmoothSets$. 

\medskip 
The following is essentially the definition of forms in the restricted case of Diffeological spaces \cite{IZ13},
now extended to non-concrete smooth sets. This definition will generalize to the synthetic differential context of 
(infinitesimally) thickened smooth sets (see \cite{GSS-2}), where as far as we know its original conception along with its full 
classifying nature arose.

\begin{definition}[\bf  De Rham forms on smooth sets]
\label{nformsonSmoothSet}
	The set of \textit{smooth de Rham $n$-forms} on a generalized smooth space 
 $\CF\in\SmoothSets$ is defined by
 \vspace{-2mm} 
	\begin{align}
		\Omega^{n}_{\mathrm{dR}}(\CF):= \mathrm{Hom}_{\SmoothSets}(\CF,  \mathbold{\Omega}_{\mathrm{dR}}^{n})\,.
	\end{align}
\end{definition}
That is, since morphisms of smooth sets are given 
by natural transformations, 
 an $n$-form $\om \in \Omega^{n}(\CF)$ is an (functorial) assignment 
 \vspace{-3mm} 
 \begin{align*}
 	\CF(\FR^k) &\longrightarrow \Omega^{n}(\FR^k) 
  \\[-2pt]
 	    \phi_{\FR^k}&\longmapsto \omega_{f_{\FR^k}} \, ,
 \end{align*}

 \vspace{-2mm} 
\noindent
mapping each `$\FR^k$-plot' in $\CF$ to an n-form on $\FR^k$. The de Rham differential and wedge product of forms on a smooth space are defined 
naturally by composing with their universal incarnations \eqref{smoothdeRham} and \eqref{SmoothClassifyingWedgeProduct}. 

\begin{definition}[\bf  Differentials and wedge products of forms]
\label{DifferentialofSmoothMap}
$\,$

\noindent {\bf (i)} The de Rham 1-form $\dd S\in \Omega^1_{\mathrm{dR}}(\CF)$  of a smooth map $S:\CF\rightarrow y(\FR)$ 
is defined as the composition
\vspace{-2mm} 
\begin{align}
	\dd_{\mathrm{d R}} S \;:\; \CF \xlongrightarrow{S} y(\FR)\cong \mathbold{\Omega}_{\mathrm{dR}}^0 \xlongrightarrow{\dd_{\mathrm{d R}}} \mathbold{\Omega}_{\mathrm{dR}}^1 \, .
\end{align}

\vspace{-1mm} 
\noindent {\bf (ii)} Similarly, the differential $\dd_{\mathrm{d R}}\om \in \Omega^{n+1}_{\mathrm{dR}}(\CF)$ of an $n$-form $\om \in \Omega^n(\CF)$ 
is defined as the composition 
\vspace{-2mm} 
$$
\dd_{\mathrm{d R}} \om \;:\; \CF \xlongrightarrow{\;\;\om \;\;} \mathbold{\Omega}_{\mathrm{dR}}^n \xlongrightarrow{\; \dd_{\mathrm{d R}}\;} \mathbold{\Omega}_{\mathrm{dR}}^{n+1}\, .
$$

\vspace{-1mm} 
\noindent {\bf (iii)} Finally, the wedge product of two forms $\om \in \Omega^{n}_{\mathrm{dR}}(\CF)$, $\om'\in \Omega^{m}_{\mathrm{dR}}(\CF)$ 
is defined as the composition
\vspace{-1mm} 
\begin{align*}
\om\wedge \om' \;:\; \CF \xlongrightarrow{\;(\om,\om')\;}\mathbold{\Omega}_{\mathrm{dR}}^{n}\times \mathbold{\Omega}_{\mathrm{dR}}^{m} 
\xlongrightarrow{\; \wedge \;}\mathbold{\Omega}_{\mathrm{dR}}^{n+m} \, .
\end{align*}
\end{definition}

By Lem. \ref{InducedStructureOnTheModuliSpace}, it follows that the collection of differential forms on 
any smooth space $\CF$ inherits the structure of a DGCA over the real numbers
\vspace{-1mm}
$$
(\Omega^{\bullet}_\mathrm{dR}(\CF), \dd_{\mathrm{d R}}, \wedge)\, \in \mathrm{DGCA}_\FR\, .
$$

\vspace{-1mm}
\noindent However, this notion of forms on a smooth set, say a field space $\CF= \mathbold{\Gamma}_M(F)$, does not have an obvious `contraction' 
operation with the corresponding vector fields. In particular, this is in contrast to the differential forms of Def. \ref{DifferentialFormsOnFieldSpace}. 
\begin{remark}[\bf Classifying nature of $\mathbold{\Omega}_{\mathrm{dR}}^\bullet$]
\label{ClassifyingFormsDoNotClassifyInSmoothSet}
Although we have defined smooth differential forms as maps into a `classifying/moduli space', having all the standard algebraic 
properties one would expect, it is not clear what exactly they `classify/modulate'.

\vspace{1mm} 
\noindent {\bf (i)} For the case of representable objects, i.e.,
finite-dimensional manifolds, they do classify the familiar differential forms as fiber-wise linear maps out of the 
tangent bundle \eqref{ClassifyingFormsClassifyManifoldForms}. 
However, as we have noted before, 
there are several notions of `a tangent bundle' for a general smooth set, each of which having its own drawbacks. 

\vspace{1mm} 
\noindent {\bf (ii)} We will come back to the fully-fledged classifying nature of this sheaf once we enrich our site with infinitesimal spaces. 
In that setting, there will be a natural notion of a tangent bundle for any generalized smooth space; 
see \cite{GSS-2}. 
For those generalized spaces that have a fiber-wise linear structure on their tangent bundle, this sheaf will 
classify forms in the traditional sense of fiber-wise linear maps out of the tangent bundle, as in the particular case of 
Def. \ref{DifferentialFormsOnFieldSpace} for field spaces; 
see \cite{GSS-2}. In the current restricted setting of smooth sets, one can still show that a certain subset of \textit{local} differential 
forms (Def. \ref{BicomplexOfLocalForms}) be viewed equivalently as maps into the classifying space $\Omega^{\bullet}_{\mathrm{d R}}$
(Lem. \ref{LocalDiffFormsAsDeRhamForms}).

\vspace{1mm} 
\noindent {\bf (iii)} Nevertheless, we note that the moduli space of $n$-forms does find important applications even in the current restricted setting 
of sheaves over Cartesian spaces. For example: 

\begin{itemize}[leftmargin=.8cm]
\item[{\bf (a)}] In \cite{IZ13} it is shown that forms on quotients of finite-dimensional manifolds 
\footnote{Such quotients do not necessarily result in smooth manifolds (e.g., the irrational torus), but they may be naturally viewed 
as smooth sets since arbitrary colimits exist in $\SmoothSets$.} in the above sense define a sensible de Rham cohomology.

\item[{\bf (b)}] For $M^d$ a smooth spacetime manifold of dimension $d$, the mapping space 
$\big[M^d, \mathbold{\Omega}_{\mathrm{dR}}^{d+2}\big]$ may be seen as a valid home for ``anomaly polynomials'' of Green-Schwarz type. 
Namely, these are $(d+2)$-forms $I_{d+2}(\phi^k)$ associated naturally with $\FR^k$-plots of field configurations $\phi^k$, 
and hence in total form a natural system of differential forms on $\FR^k \times M$,
which is just a single smooth map from $\CF$ into the mapping space $[M^d, \mathbold{\Omega}_{\mathrm{dR}}^{d+2}]$. 
Crucially, such a map is necessarily trivial on fields and $\FR^1$-plots of fields, since 
$[M^d, \mathbold{\Omega}^{d+2}_{\mathrm{dR}}](\FR^k)\cong \Omega^{d+2}(\FR^k\times M^d)$ -- which vanishes for 
$k\leq 1$ (i.e., the mapping space is also not concrete). Nevertheless, there is non-trivial data in the anomaly polynomial,
encoded on higher $\FR^k$-plots of fields. If we assume that the integral of such forms over $M^d$ exists 
(for instance, assuming that $M^d$ is compact) then this construction serves to produce a smooth de Rham 2-form on the field space, as follows:
\vspace{-2mm} 
\begin{equation}
  \label{AnomalyPolynomials}
  \begin{tikzcd}[
    column sep=45pt,
    row sep=0pt
  ]
    {\CF}
    \ar[
      rr,
      "{
        I_{{}_{d+2}}
      }"
    ]
    \ar[
      rrrr,
      rounded corners,
      to path={
           ([xshift=-2pt,yshift=+00pt]\tikztostart.north)          
        -- ([xshift=-2pt,yshift=+9pt]\tikztostart.north)          
        -- node {
          \scalebox{.7}{
            \colorbox{white}{\bf 
              \color{darkgreen}
              curvature of anomaly line bundle on field space
            }
          }
        }
           ([yshift=+08pt]\tikztotarget.north)          
        -- ([yshift=+00pt]\tikztotarget.north)          
      }
    ]
    &&
    {
      \big[
        M^d,
        \boldsymbol{\Omega}_{\mathrm{dR}}^{d+2}
      \big]
    }
    \ar[
      rr,
      "{ \int_{M^d} }"
    ]
    &&
    \boldsymbol{\Omega}^2_{\mathrm{dR}} \;.
    \\[-7pt]
    (
      \overset{
        \mathclap{
          \raisebox{1pt}{
            \scalebox{.7}{\bf 
              \color{darkblue}
              probe
            }
          }
        }
      }{
        \FR^k
      }
      \times 
      M^d 
        \xrightarrow{\phi^k} F
    )
    \ar[
      rr, 
      phantom, 
      "{\longmapsto}"
    ]
    &&
    I_{d+2}(\phi^k)
    \ar[
      rr, 
      phantom, 
      "{\longmapsto}"
    ]
    &&
    \int_{M^d}
    I_{{}_{d+2}}(\phi^k)
    \\[-6pt]
    \mathclap{
    \scalebox{.7}{\bf 
      \color{darkblue}
      \def\arraystretch{.9}
      \begin{tabular}{c}
        $\FR^k$-parametrized
        \\
        family of fields
      \end{tabular}
    }
    }
    &&
    \mathclap{
    \scalebox{.7}{\bf
      \color{darkblue}
      \def\arraystretch{.9}
      \begin{tabular}{c}
        anomaly polynomial
        \\
        in
        $\Omega^{d+2}(\FR^k\times M^d)$
      \end{tabular}
    }
    }
    &&
    \mathclap{
    \scalebox{.7}{\bf 
      \color{darkblue} 
      \def\arraystretch{.9}
      \begin{tabular}{c}
        local anomaly
        \\
        in
        $\Omega^{2}(\FR^k)$
      \end{tabular}
    }
    }
  \end{tikzcd}
\end{equation}

\vspace{-2mm} 
\noindent As indicated in the diagram, if the $I_{d+2}$ are indeed the anomaly polynomials of a Green-Schwarz type anomaly, 
then the 2-form on field space exhibited by the composite arrow is a precise incarnation of what is commonly referred to 
as the curvature of the ``anomaly line bundle'' (cf. \cite[p. 21]{Freed02}, where our ``$\FR^k$'' appears as ``$T$''). 

Of course, in the entire anomaly line bundle
corresponding to \eqref{AnomalyPolynomials} appears similarly as a map $\CF \to \mathbf{B}\mathrm{U}(1)_{\mathrm{conn}}$ 
to the smooth stack of circle bundles with connection \eqref{BGConn}. Alternatively, for the usual case of \textit{local}
anomaly polynomials, these may be equivalently interpreted as $(d,2)$-local forms on $\CF\times M$ (Def. \ref{BicomplexOfLocalForms}),
with the relation of the two interpretations being via Lem. \ref{LocalDiffFormsAsDeRhamForms}.

\item[{\bf (c)}] 
Similarly, {\it invariant} polynomials of degree $n$ (such as underlying higher Chern-Simons forms) are naturally interpreted as maps out 
of a moduli stack $\mathbf{B}G_{\mathrm{conn}}$ \eqref{BGConn} into the classifying sheaf of differential forms
(see \cite[p. 65]{FSS12}\cite[p. 2]{FSS13}\cite{FH13}\cite[\S 5.4.3]{dcct}\cite[\S 3.4]{FSS14}):
\begin{equation}
  \label{CharPolynomials}
  \mathbf{B}G_{\mathrm{conn}}
  \longrightarrow
  \mathbold{\Omega}_{\mathrm{dR}}^{2n}
  \,.
\end{equation}

\end{itemize}
\end{remark}

\paragraph{\bf Smooth structure on de Rham forms of a smooth set.}
By construction, the set of smooth de Rham forms on a smooth space $\Omega^{n}_{\mathrm{dR}}(\CF)$ may be promoted to 
a smooth set using the internal hom \eqref{SmoothSetinternalHom}
\vspace{-1mm} 
$$
{\mathbold{\Omega}}_{\mathrm{dR}}^n(\CF)(\FR^k) 
= \mathrm{Hom}_{\SmoothSets}\big(y(\FR^k)\times \CF, \, \mathbold{\Omega}_{\mathrm{dR}}^{n}\big) \, .
$$

\vspace{-1mm} 
\noindent This does contain what one would interpret as `smoothly $\FR^k$-parametrized n-forms on $\CF$', but 
is much larger.  To see this, consider the case when $\CF=y(M)$ is a manifold where, by the Yoneda Lemma,
$$
\mathbold{\Omega}_{\mathrm{dR}}^n \big(y(M)\big)(\FR^k) = \mathrm{Hom}_{\SmoothSets}\big(y(\FR^k)\times y(M),
\mathbold{\Omega}_{\mathrm{dR}}^{n}\big)\cong \Omega^{n}(\FR^k\times M)\, .
$$ 
Instead, we would like to consider $\FR^k$-parametrized forms on a manifold $M$ as the subalgebra of vertical 
$n$-forms along $\FR^k\times M\rightarrow M$, i.e., 
$$
\Omega^{n}_{\mathrm{Vert}}\big(y(M)\big)(\FR^k)= \Omega^n(\FR^k\times M) \big/ \Omega^{n\geq 1} (\FR^k)
\cong C^\infty(\FR^k)\hat{\otimes} \Omega^{n}(M)\cong {\mathbold{\Gamma}}_M(\wedge^n T^*M)(\FR^k)\, ,
$$
where $ \hat{\otimes}$ denotes the completed projective tensor product of $\FR$-vector spaces.
Analogously working internally to smooth sets, we instead promote $\Omega^{n}_{\mathrm{dR}}(\CF)$ to a smooth set vertically by
\begin{align}
\label{SmoothFormsVerticalStructure}
\Omega^{n}_\mathrm{dR, Vert}(\CF)(\FR^k):=C^{\infty}(\FR^k)\hat{\otimes} \Omega^{n}(\CF)= 
\mathrm{Hom}_{\SmoothSets}\big(\CF,\mathbold{\Omega}_{\mathrm{dR}}^{n}\hat{\otimes} C^\infty(\FR^k)\big) \, ,
\end{align}
where the tensor product on the right-hand side is meant to be computed object-wise. 
\footnote{It can be checked that this smooth set structure corresponds exactly to the so-called ``concretification'' 
of the mapping space $[\CF,\mathbold{\Omega}^n_{\mathrm{dR}}]$ (see \cite[\S 1.2.3.3]{dcct}). } 
In this case, the de Rham differential also extends vertically to a smooth map
\vspace{-2mm} 
\begin{align}\label{SmoothVerticaldeRham}
\dd_\mathrm{Vert} \;:\; \Omega^{n}_{\mathrm{dR, Vert}}(\CF)
\longrightarrow \Omega^{n+1}_\mathrm{dR, Vert}(\CF) 
\end{align}

\vspace{-2mm} 
\noindent as $C^\infty(\FR^k)\hat{\otimes} \Omega^n_\mathrm{dR} (\CF) \xlongrightarrow{\id \otimes \dd}
C^\infty(\FR^k)\hat{\otimes} \Omega^{n+1}_\mathrm{dR}(\CF)$,
for each $\FR^k\in \mathrm{Cart}$. Similarly, the wedge product may be extended vertically to a smooth map 
$$
\wedge_\mathrm{Vert} \;:\; \Omega^n_{\mathrm{dR, Vert}}(\CF)\times \Omega^m_\mathrm{dR, Vert}(\CF) \longrightarrow \Omega^{n+m}_\mathrm{dR, Vert}(\CF).
$$

\newpage 
 
\begin{remark}[\bf Cotangent bundles of smooth sets]
For the case of a manifold $M\in \SmoothManifolds$, $1$-forms are naturally identified with sections of the cotangent bundle 
$\Omega^{1}(M)\cong \Gamma_M(T^*M)$, and analogously for $n$-forms. The vertical smooth structure on $n$-forms defined 
above coincides with that on sections of a vector bundle, as in Def. \ref{SectionsSmoothSet}. However, for 
a general smooth space $\CF$, there is no natural notion of a cotangent bundle over it -- but we still want to
think of de Rham forms on $\CF$ as sections of a would-be bundle, 
hence the vertical smooth structure. \footnote{For the case when $\CF$ is a diffeological space, there is a (somewhat indirect) 
way to construct a bundle  whose sections coincide with $\Omega^n_\mathrm{d R}(\CF)$; see \cite[Ch.6]{IZ13}.}
\end{remark}

To recap, there is a consistent and very general definition of smooth forms on any generalized 
smooth space, finite or infinite-dimensional, concrete or not, which evades the use of any 
functional analysis. In \cref{OnShellFieldsCriticalSetSubsection}, we will describe how this notion of forms on the field space 
may be used to encode the traditional variation of an action functional (albeit not in the most straightforward manner, 
see Rem. \ref{CriticalityViaModuliSpaceOf1-forms}). To that end, we now move into the description of the infinite jet bundle,
as it forms a crucial component of (local) variational field theory.

\newpage 

\addtocontents{toc}{\protect\vspace{-10pt}}
\section{Local Lagrangians} \label{Sec-locLag}

We first recall the traditional infinite-dimensional Fr\'{e}chet, locally pro-manifold structure 
 on $J^{\infty}_M F$, following the description of \cite{Takens79}\cite{Saunders89}\cite{KS17}. Then, viewing the infinite
 jet bundle as a smooth set, we describe how to employ smooth maps out of it to define (smooth) \textit{local} Lagrangians, \textit{local} 
 currents and charges, and \textit{local} functionals on the corresponding smooth field space $\CF$. 
 We conclude by defining the appropriate notion of a (finite and smooth) symmetry of a local Lagrangian field theory. In the subsequent sections, 
 we will use the smooth set description of the infinite jet bundle to naturally construct its tangent bundle, vector fields, and differential 
 forms internally within SmoothSet while, in parallel, showing these constructions recover all the corresponding standard 
 notions employed in the classical literature.

 \medskip 
In \cite{GSS-2}, we will detail how one may equivalently detect (or define) the jet 
 bundle directly within (formal) smooth sets, therefore bypassing
 the infinite-dimensional technicalities mentioned below, and further setting the scene to define the appropriate notion of jet 
 bundles in the fermionic setting. In particular, an (infinite) jet \eqref{inftyJetTraditional} of a fiber bundle $F\rightarrow M$
 at a point $p$ in the base will be equivalently a section of $F$ over the \textit{`infinitesimal neighborhood'} of $p\in M$.

\subsection{Infinite jet bundles as locally pro-manifolds} 
\label{JetBundleSection}

 Let $\pi_M: F\rightarrow M$ be a fiber bundle and $p\in M$
a point in the base. For any $k\in\NN$, the set of $k^{th}$-order jets of sections at $p\in M$ 
 is traditionally defined as the equivalence classes of (local) sections such that their partial derivatives agree, on some 
 (and hence any) local chart, up to order $k$ at $p$:
 \vspace{-2mm} 
\begin{align}\label{kjetTraditional} 
 J^{k}_p(F):= \Big\{j^k_p\phi=[\phi] \, \big{|}\, \phi \sim \phi' \in \Gamma_M(F) \iff \partial_{I} \phi^a (p) 
 = \partial_{I} {\phi'}^a (p) \hspace{0.3cm} \forall\, \,  0\leq |I| \leq k \Big\} \, ,
 \end{align}

 \vspace{-2mm} 
\noindent where $I$ denotes (symmetric) multi-indices. The collection of $k$-jets forms  a finite-dimensional manifold 
\vspace{-2mm} 
 $$
 J^k_M(F)=\underset{p\in M}{\bigcup} J^k_p(F) \quad  \in \quad  \SmoothManifolds \;, 
 $$ 

  \vspace{-2mm} 
\noindent  and in particular a fiber bundle  over $M$ for each $k\in \NN$. For instance, a compatible coordinate chart $\{x^\mu, u^a\}$ 
 on a trivialization of $F\rightarrow M$ induces a natural coordinate chart on $J^1_M(F)$ denoted by 
 $\{x^\mu,u^a, u^a_\mu \}$. Similarly, there are induced coordinate charts 
 $\big\{x^\mu,\{u_I^a\}_{|I|\leq k} \big\}:= \big\{x^\mu,u^a, u^a_\mu, u^{a}_{\mu_1 \mu_2}, \cdots , u^a_{\mu_1\cdots \mu_k}\big\}$
 on each $J^k_M F$, with extra coordinates corresponding 
 to the higher partial derivatives appearing in the definition of a $k$-jet (and hence symmetric in the base indices). 
 There is a natural sequence of smooth maps $\pi^k_{k-1}: J^{k}_M(F)\rightarrow J^{k-1}_M(F)$, in particular surjective submersions, 
 which `forget' the highest derivatives, and so a diagram of smooth manifolds  
 \vspace{-2mm} 
 \begin{align*}
 \longrightarrow J^k_M(F) \longrightarrow J^{k-1}_M(F) \longrightarrow
 \cdots \longrightarrow  J^1_M(F)\longrightarrow J^0_M(F)\cong F \, .    
 \end{align*}

 \vspace{-2mm} 
\noindent Furthermore, there is a canonical map of smooth sections, the $k^{th}$-order \textit{jet prolongation}
\vspace{-2mm} 
\begin{align}\label{kjetprolongation}
j^k \;:\; \Gamma_M(F) \longrightarrow \Gamma_M\big (J^k_M (F) \big) 
\end{align}

\vspace{-2mm} 
\noindent
where $j^k\phi (p) := j^k_p \phi $, for each $k\in \NN$. 

\medskip 
The infinite jet bundle $J^\infty_M F$ is supposed to be the projective limit of the above diagram. It exists when 
computed in the category Set of sets (or topological spaces), and so each fiber $J^\infty_p(F)$ will consist of 
equivalence classes of sections whose partial derivatives agree to arbitrary order, i.e.,
\vspace{-2mm} 
\begin{align}\label{inftyJetTraditional} 
 J^{\infty}_p(F)= \Big\{j^\infty_p\phi=[\phi] \; \big{|}\; \phi \sim \phi' \in \Gamma_M(F) 
 \iff \partial_{I} \phi^a (p) = \partial_{I} {\phi'}^a (p) \hspace{0.3cm} \forall \, \, 0 \leq |I| \Big\} \, .
 \end{align}

 \vspace{-2mm} 
\noindent Set-theoretically, this guarantees the existence of a projection map 
\vspace{-2mm} 
$$
\pi_k\equiv \pi^\infty_k \;:\;  J^\infty_M(F)\longrightarrow J^k_{M}(F)
$$

\vspace{-2mm}
\noindent for each $k\in \NN$ and a map of sections, the \textit{infinite jet prolongation} 
\vspace{-2mm} 
\begin{align}\label{inftyjetprolongation}
   j^\infty \;:\; \Gamma_M(F)\longrightarrow \Gamma_M\big(J^\infty_M(F)\big)\, , 
\end{align}

\vspace{-2mm} 
\noindent where $j^\infty \phi (x):= j^\infty_{x} \phi$. However, the limit does not exist in $\SmoothManifolds$ since it is an 
infinite limit of manifolds of increasing dimension, hence necessarily infinite-dimensional.

\medskip 
A natural way to evaluate
the limit is via the fully faithful embedding $\SmoothManifolds\hookrightarrow \FrechetManifolds$ of finite-dimensional manifolds 
into Fr\'echet manifolds, whereby the limit $J^\infty_M F := \mathrm{lim}_k^{\FrechetManifolds}J^k_M F$ exists
as an infinite-dimensional, paracompact manifold by virtue of the maps in the diagram being surjective submersions 
\cite{Takens79}\cite{Saunders89}\cite{KS17}.

\begin{definition}[\bf  Fr{\'e}chet jet bundles]
The \textit{infinite jet bundle} $J^\infty_M F$ is the (para-compact and Hausdorff) Fr\'{e}chet manifold defined by the limit 
\vspace{-1mm} 
$$
  J^\infty_M(F)
  :=
  \mathrm{lim}^{\FrechetManifolds}_{k}
  \, \, J^k_M(F) \quad \in \quad  \FrechetManifolds \, 
$$ 

\vspace{0mm} 
\noindent whose local model is 
$\FR^\infty
  =
\mathrm{lim}^{\FrechetManifolds}_k
\, \, \FR^k \, \, \in \, \, \FrechetManifolds $, with 
$\big\{x^\mu, \{u_I^a\}_{0\leq |I|}\big\} := \big\{x^\mu, u^a,  u^a_{\mu_1}, u^a_{\mu_1 \mu_2}, \cdots\big\}$
being the local coordinate charts. 
\end{definition}

From our point of view, the 
manifold and explicit chart description of the infinite jet bundle will mostly play an auxiliary role, so as to make 
contact with the existing literature. The universal properties will play a more central role. For instance, the projections
$\pi^\infty_{k}:J^\infty_M(F)\rightarrow J^k_M(F)$ are identified with the universal cone projections, and hence become 
smooth Fr\'echet maps. More generally, such infinite-dimensional limits of finite-dimensional manifolds form a subcategory 
of Fr\'{e}chet manifolds \cite{KS17}.

\begin{definition}[\bf  Locally pro-manifolds]\label{LocProMan}
We define the category of \textit{locally pro-manifolds}
\vspace{-2mm} 
$$
\mathrm{LocProMan}\longhookrightarrow \FrechetManifolds\, ,
$$ 

\vspace{-2mm} 
\noindent to be the full subcategory of Fr\'{e}chet manifolds consisting of projective limits of finite-dimensional manifolds. 
\end{definition}
Note that, by Prop. \ref{FrManRestrEmbedding}, locally pro-manifolds also embed fully faithfully into smooth sets via
\vspace{-2mm} 
\begin{align*}
    \mathrm{LocProMan}\longhookrightarrow \FrechetManifolds\xhookrightarrow{\quad y \quad} \SmoothSets\, .
\end{align*} 

\vspace{-1mm} 
\noindent The defining properties of such infinite-dimensional manifolds, and in particular of the resulting infinite-dimensional manifold 
\vspace{-4.5mm} 
\begin{align}\label{JetBundleLocPro}
J^\infty_M(F):=\mathrm{lim}_k^{\FrechetManifolds}\, \, J^k_M(F) \quad  \in \quad \mathrm{LocProMan}
\end{align}

 \vspace{-2mm} 
\noindent with which we are concerned is the characterization of smooth Fr\'echet maps, 
$\mathrm{Hom}_{\FrechetManifolds}\big(J^{\infty}_M(F), \Sigma  \big)$ 
and $\mathrm{Hom}_{\FrechetManifolds}\big(\Sigma, J^\infty_M(F)\big)$, with codomain and domain being a finite-dimensional manifold $\Sigma \in \SmoothManifolds$,
respectively.
The characterization of smooth maps into $J_M^\infty(F)$ is straightforward.

\begin{lemma}[{\bf Smooth maps into the jet bundle}]
\label{FcnsIntoInftyJet}
Let $\Sigma \in \SmoothManifolds$ be a finite-dimensional manifold. A map of sets $f:\Sigma \rightarrow J^\infty_M(F)$ is smooth if 
and only if  
\vspace{-1mm}
$$
\pi^{\infty}_k\circ f \;:\; \Sigma \longrightarrow J^k_M(F)
$$

\vspace{-1mm} 
\noindent is smooth for each $k\in \NN$. Furthermore,  smooth maps $f:\Sigma \rightarrow J^{\infty}_M(F)$ are in
1-1 correspondence with families of smooth maps $\{f_k: \Sigma \rightarrow J^{k}_M(F) \, | \, k \in \NN \}$ such that
\vspace{-1mm} 
	\[ 
\xymatrix@R=1.6em@C=4em{ \Sigma \ar[rrd]^{f_{k_1}} \ar[d]_{f_{k_2}}&&  
	\\ J^{k_2}_M(F) 
	\ar[rr]^<<<<<<<{\pi^{k_2}_{k_1}} && J^{k_1}_M(F) 
}   
\]

\vspace{-1mm} 
\noindent commutes for each pair $k_2\geq k_1$.
\begin{proof}
This is an immediate consequence of the universal cone property of the limit 
\vspace{-2mm} 
\begin{align*}J^\infty_M(F)\longrightarrow \cdots
 \longrightarrow J^k_M(F) \longrightarrow J^{k-1}_M(F) \longrightarrow \cdots \longrightarrow 
 J^1_M(F)\longrightarrow J^0_M(F)\cong F \, .    
 \end{align*}

 \vspace{-7mm} 
\end{proof}
\end{lemma}
In particular, this result holds for all $\Sigma \in \mathrm{CartSp}\hookrightarrow \SmoothManifolds$, and so gives 
a more explicit description of the smooth set
\vspace{-1mm} 
$$
y(J^\infty_M F ):=\mathrm{Hom}_{\FrechetManifolds}(-,J^\infty_M F )|_{\mathrm{CartSp}} \quad \in \quad  \SmoothSets
$$ 

\vspace{-1mm} 
\noindent incarnation of the infinite jet bundle along the embedding of Prop. \ref{FrManRestrEmbedding}. 
Note that limits in SmoothSet are computed objectwise, which means that the smooth set above is equivalently the limit formed in smooth sets 
\vspace{-2mm} 
\begin{align}\label{JetBundleSmoothSetlimit}y(J^\infty_M F) \cong \mathrm{lim}_k^{\SmoothSets}\,  y( J^k_M F) \, .
\end{align}

\vspace{-2mm} 
\noindent by first embedding each finite-dimensional jet bundle. Concretely, we may consider $\FR^n$-plots of $y(J^\infty_M F)$ as families
of compatible plots of each $J^k_ M F$,
\vspace{-3mm} 
$$
y(J^\infty_M F)(\FR^n)\cong\Big\{ \big\{s^n_k:\FR^n\rightarrow J^k_M F\;\;  \big|  \;\;   \pi^{k}_{k-1}\circ  s^n_{k}
= s^n_{k-1}\big\}_{k\in \NN} \Big\}\, .
$$

\vspace{-2mm} 
\noindent
In particular, an infinity jet $j^\infty_p \phi\in J^\infty_M F$, i.e., a $*$-plot in $y(J^\infty_M F)$ is equivalently represented by the 
compatible family of $k$-jets $\{j^k_p \phi=\pi_k(j^\infty_p)\}_{k\in \NN}$, as expected by the underlying set-theoretic limit \eqref{inftyJetTraditional}.

\begin{remark}[\bf Locally pro-manifolds as smooth set limits]
\label{LocProMantoSmoothSetReflectsLimits}
The same statements hold for any locally pro-manifold 
$
  G^\infty
  =
  \mathrm{lim}_k^{\FrechetManifolds}G^k\in \mathrm{LocProMan}
$. 
That is, 
\vspace{-2mm} 
$$
  y(G^\infty)
  :=
  y\big(\mathrm{lim}_k^{\FrechetManifolds} G^k\big)
  \cong 
  \mathrm{lim}_k^{\SmoothSets}\big( y(G^k)\big)
  \, ,
$$ 

\vspace{-1mm} 
\noindent i.e., the embedding $y:\FrechetManifolds\hookrightarrow \SmoothSets$ reflects projective limits of finite-dimensional manifolds.
In fact, $y:\FrechetManifolds\hookrightarrow \SmoothSets$ reflects all limits of finite-dimensional manifolds, and hence preserves 
all limits in $\mathrm{LocProMan}$, by definition of the latter.
\end{remark}

Moreover, the limit property implies that the infinite jet prolongation \eqref{inftyjetprolongation} maps smooth sections of $F$ to smooth sections of 
$J^\infty_M(F)$, since composition with each of $\pi^\infty_k$ gives a smooth section of $J^k_M(F)$. Smooth sections of the infinite jet bundle 
naturally form a smooth set $\mathbold{\Gamma}_M(J^\infty F)$ as in Def. \ref{SectionsSmoothSet}, i.e., by considering smoothly 
$\FR^k$-parametrized sections of $J^\infty_M(F)$. Moreover, the infinite jet prolongation canonically extends (uniquely) to a \textit{smooth} map
\vspace{-1mm} 
\begin{align}\label{smoothjetprolongation}
    j^{\infty} \;:\; \mathbold{\Gamma}_M(F) &\longrightarrow \mathbold{\Gamma}_M(J^\infty F) 
    \\[-1pt]
    \phi^k &\longmapsto j^\infty \phi^k \nn 
\end{align}

\vspace{-1mm} 
\noindent  vertically. Explicitly, for $\phi^k\in \mathbold{\Gamma}_M(F)(\FR^k)$ a smoothly $\FR^k$-parametrized section (Def. \ref{SectionsSmoothSet}), 
$j^\infty\phi^k \in \mathbold{\Gamma}_M(J^\infty F)(\FR^k)$ is defined by $j^\infty \phi^k (x,u):= j^\infty\big(\iota^*_{u\rightarrow \FR^k}\phi^k \big) (x)$, 
i.e., the prolongation is applied point-wise with respect to the probe space. It is immediate to see that the assignment is functorial 
with respect to pullbacks of maps $f:\FR^m\rightarrow \FR^k$ of Cartesian probe spaces, hence defining a map of smooth sets.

\medskip 
The characterization of smooth maps out of $J^\infty_M(F)$ is more delicate and, for that purpose, we use the following \cite{Takens79}\cite{KS17}.

\begin{proposition}[\bf Smooth maps out of the jet bundle]
\label{FunctionsOnInftyJetBundle}
Let $\pi_k : J^\infty_M(F) \rightarrow J^k_M(F)$ be the canonical projection for each $k\in \NN$. 
A function $f:J^{\infty}_M(F)\rightarrow \FR $ (a map of sets) is smooth if and only if locally around every point $x\in J^{\infty}_M(F)$ 
it factors through $\pi_k : J^{\infty}_M(F) \rightarrow J^{k}_M(F)$ for some $k\in \NN$. That is, for each $x\in J^\infty_M(F)$ 
there exists a $k\in \NN$, a neighborhood $U\subset J^k_M(F)$ around $\pi_k(x)\in J^k_M(F)$ and a smooth function of (finite-dimensional)
manifolds $\tilde{f}^k_U: J^k_M(F)|_{U}\rightarrow \FR$ such that the following diagram commutes
\vspace{-1mm} 
	\[ 
\xymatrix@R=1.3em@C=4em{ J^\infty_M(F)|_{\pi_k^{-1}(U)} \ar[rrd]^-{f} \ar[d]_{\pi_k}&&  
	\\ 
	J^{k}_M(F)|_U \ar[rr]^<<<<<<<{\tilde{f}^{\, k}_U} && \FR\;.
}   
\]
\end{proposition}
We note that the proof of this proposition, and later the discussion of Rem. \ref{1formInfinityJetAsFrechetMap}, 
are the only place where analytical details regarding Fr\'{e}chet theory contribute
to our discussion. 
The rest of the presentation is mainly based on universal properties of smooth sets along with the above result, 
where necessary. The result is readily extended if the target manifold $\FR$ is replaced by $\FR^n$, 
and consequently by any finite-dimensional manifold $N\in \SmoothManifolds$.

\begin{corollary}[\bf  Smooth maps valued in manifolds]
\label{FunctionsInfJetToManifold}
Let $\Sigma \in \SmoothManifolds$ be a finite-dimensional manifold. A function $f:J^{\infty}_M(F)\rightarrow \Sigma $ (of sets) 
is smooth if and only if, locally around every point $x\in J^{\infty}_M(F)$, it factors through 
$\pi_k : J^{\infty}_M(F) \rightarrow J^{k}_M(F)$ for 
some $k\in \NN$.
\end{corollary}
\begin{proof}
Consider first the case where $\Sigma=\FR^2$, with $f=J^\infty_M (F)\rightarrow \FR^2$ a map of sets. It is a smooth map 
of Fr\'echet manifolds if and only if $f_i:=\pr_i\circ f: J^\infty_M (F)\rightarrow \FR $ is smooth for both $i=1,2$. 
By Prop. \ref{FunctionsOnInftyJetBundle}, this is the case if and only if around each $x\in J^\infty_M(F)$ there 
exist $k_1,k_2\in \NN$, neighborhoods $U_1 \subset J^{k_1}_M(F),\, U_2\subset J^{k_2}_M(F)$ and smooth functions 
$\tilde{f}^{\, k_i}_{U_i}:J^{k_i}_M(F)|{U_{i}}\rightarrow \FR$ such that the diagrams
\vspace{-2mm} 
	\[ 
\xymatrix@R=1.5em@C=4em{ J^\infty_M(F)|_{\pi_{k_i}^{-1}(U_i)} \ar[rrd]^-{f_i} \ar[d]_{\pi_{k_i}}&&  
	\\ 
	J^{k_{i}}_M(F)|_{U_i} \ar[rr]^<<<<<<<<{\tilde{f}^{\, k_i}_{U_i}} && \FR
}   
\]

\vspace{-1mm} 
\noindent 
commute for $i=1,2$. Assume (without loss of generality) that $k_2\geq k_1$ and let $\pi^{k_2}_{k_1}:J^{k_2}_M(F) \rightarrow J^{k_1}_M(F)$
denote the canonical projection.
Define $\bar{f}^{\, k_1}_{U}:= (\pi^{k_2}_{k_1})^*\tilde{f}^{\, k_1}_{U_1}: J^{k_2}_{M}(F)|_{U}\longrightarrow \FR$, 
where $U=U_2 \cap \pi^{-1}(U_1) $, then the diagram  
\vspace{-2mm} 
	\[ 
\xymatrix@R=1.6em@C=5em{ J^\infty_M(F)|_{\pi_{k_2}^{-1}(U)} \ar[rrd]^-{(f_{1},f_{2})} \ar[d]_{\pi_{k_2}}&&  
	\\ 
	J^{k_{2}}_M(F)|_{U} \ar[rr]^<<<<<<<<<{\big(\bar{f}^{\; k_1}_{U},\; \tilde{f}^{\; k_2}_{U_2}\big)} && \FR^2
}   
\]

\vspace{-2mm} 
\noindent commutes since each of its compositions with the two projections $\pr_i:\FR^2\rightarrow \FR$ commute. 
Hence, a map $f:J^\infty_M(F) \rightarrow \FR^2$ is smooth if and only if it locally factors through finite order jet bundles. 

The case for the codomain being $\FR^n$ follows similarly by induction of the above argument. 
Lastly, let $\Sigma$ be an n-dimensional manifold, then a map $f:J^\infty_M(F) \rightarrow \Sigma$ is smooth if and only if 
$\psi \circ f:J^\infty_M(F)|_{f^{-1}V} \rightarrow \FR^n$ is smooth for every local chart 
$\psi:V\subset \Sigma \xrightarrow{\sim} \FR^n$. By the previous argument, this is true if and only if for each
$x\in J^\infty_M (F) $ there exists some $k\in \NN$, a neighborhood $U\subset J^k_M(F)$ of $\pi^k(x)$ and a 
smooth map $\tilde{f}_U^k:J^k_M(F)\rightarrow V\subset \Sigma$ such that the following diagram commutes
\vspace{-2mm} 
	\[ 
\xymatrix@R=1.4em@C=1.8em{ J^\infty_M(F)|_{\pi_k^{-1}(U)} \ar[rrd]^-{f} \ar[d]_{\pi_k}&& & &
	\\ 
	J^{k}_M(F)|_U \ar[rr]^<<<<<<{\tilde{f}^{\, k}_U} && V\subset \Sigma \ar[rr]^{\psi} && \FR^n 
}   
\]

\vspace{-1mm} 
\noindent For each local chart $\psi$, the corresponding diagram commutes if and only if its left triangle commutes, 
 giving the result.
\end{proof}

\newpage 

We may combine the above two results to describe the set of smooth maps 
$\mathrm{Hom}_{\mathrm{LocProMan}}(J^\infty_M F, G^\infty)$ from the infinite jet bundle to any other 
projective limit of finite-dimensional manifolds $G^\infty=\mathrm{lim}_j^{\FrechetManifolds}G^j$.

\begin{proposition}[\bf Smooth maps of locally pro-manifolds]\label{SmoothMapsofLocProMan}
Let $G^\infty=\underset{{\rm FrMan}}{{\rm lim}}G^j$ be any locally pro-manifold.

\noindent {\bf (i)} A map of sets 
$f:J^\infty_M F \rightarrow G^\infty $ is smooth if and only if 
\vspace{-3mm}
$$
p^{\infty}_j \circ f \;:\; J^\infty_M F \longrightarrow G^\infty \xlongrightarrow{\;\; p^\infty_j\;\;} G^j 
$$

\vspace{-1mm} 
\noindent is smooth for each $k\in \NN$, and hence if each $p^\infty_j\circ f$ locally around every $x\in J^\infty_M F$ 
locally factors through $\pi_k:J^\infty_M (F) \rightarrow J^k_M (F)$ for some $k\in \NN$, where 
$p^\infty_j: G^\infty \rightarrow G^j$ denotes the universal cone projections of $G^\infty$.

\noindent {\bf (ii)} Furthermore,  
smooth maps $f:J^{\infty}_M(F) \rightarrow G^\infty$ are in 1-1 correspondence with compatible families of 
smooth maps \newline 
$\{f_j: J^\infty_M F \rightarrow G^j  \, 
| \, j \in \NN \}$ such that
\vspace{-2mm} 
	\[ 
\xymatrix@R=1.7em@C=4em{ J^\infty_M (F)\ar[rrd]^{f_{j_1}} \ar[d]_{f_{j_2}}&&  
	\\ G^{j_2} 
	\ar[rr]^<<<<<<<<{p^{j_2}_{j_1}} && G^{j_1} 
}   
\]

\vspace{-2mm} 
\noindent commutes for each pair $k_2\geq k_1$, and hence, such that furthermore each $f_j$ locally factors through 
$\pi_k:J^\infty_M (F) \rightarrow J^k_M (F)$ for some $k\in \NN$.
\end{proposition}
\begin{proof}
By the Fr\'{e}chet limit property of $G^\infty$, a map of sets $f:J^\infty_M F \rightarrow G^\infty$ is smooth if and only if each 
$p^\infty_j \circ f :J^\infty_M F\rightarrow G^j$, and furthermore each such compatible family of smooth maps 
$\{f_j: J^\infty_M \rightarrow G^j  \, | \, j \in \NN \}$ defines a smooth map $f:J^\infty_M F\rightarrow G^\infty$. 
Each of the terms $G^j$ in the limit are  by assumption finite-dimensional manifolds $G^j\in \SmoothManifolds$, and hence by Cor. \ref{FunctionsInfJetToManifold} each $f_j$ locally factors through a finite order jet bundle .
\end{proof}

\begin{remark}[\bf Pro-manifold vs Fr\'{e}chet jet bundle]\label{ProVsLocProJetBundle} 
$\,$

\noindent {\bf (i)} The limit over jet bundles may alternatively be computed formally as a pro-object in the categorical sense. 
In this case, the resulting object is 
not a Fr\'echet manifold, but simply a \textit{formal}
limit of finite-dimensional manifolds instead. A smooth function on such a formal object is \textit{by definition} 
one that \textit{globally} factors through some $J^k_M(F)$. Said otherwise, the smooth functions $C^\infty_\mathrm{glb}(J^\infty_M F)$ 
on the formal limit are the union over $k\in \NN$ of those on $J^k_M(F)$. \footnote{Namely, the colimit of the algebras 
of functions on each finite order jet bundle. This is, in particular, precisely the definition that \cite{Anderson89} uses.} 
Note that these naturally sit inside the algebra of locally finite order functions
$C^\infty_\mathrm{glb}(J^\infty_M F)\hookrightarrow C^\infty(J^\infty_M F)$, as a sub-algebra.

\vspace{1mm} 
\noindent {\bf (ii)} Crucially, however, 
globally finite order functions on $J^\infty_M F$ do not form a (petit) sheaf on the underlying topological space, in contrast to those of 
locally finite order, and hence the study of the two is of different flavor. 
The formal approach is implicit in 
\cite{Anderson84}\cite{Anderson89}\cite{DF99}, while  \cite{Takens79}\cite{Saunders89}\cite{KS17} employ the Fr\'echet limit described above. 

\vspace{1mm} 
\noindent {\bf (iii)} Our choice of the latter is based upon the simple structure of maps into the Fr\'echet manifold $J^\infty_M F$ of Lem. \ref{FcnsIntoInftyJet}, 
and the natural embedding into smooth sets of Prop. \ref{FrManRestrEmbedding} -- thus viewing it as a proper geometrical space on the same footing with field 
spaces. Further support for this choice is given by the observation that the resulting smooth set is equivalently the limit over jet bundles computed directly
in smooth sets,\footnote{This does not render the Fr\'{e}chet picture redundant, for it is this incarnation that provides the intuitive point set
description of smooth maps out of $J^\infty_M F$.} as in  \eqref{JetBundleSmoothSetlimit}. 

\vspace{1mm} 
\noindent {\bf (iv)} We note that the texts \cite{Blohmann23b}\cite{Del}  attempt to view (bosonic) Lagrangian field theory as taking place in the 
full subcategory  $\DiffeologicalSpaces $ of smooth sets, while also requiring that the infinite jet bundle is only a formal pro-object.
The drawbacks in that approach are twofold:

\begin{itemize} 
\item[{\bf (a)}] restricting to concrete sheaves, while sufficient for most aspects of bosonic field theory, does not naturally generalize to 
include fermionic fields, or infinitesimal structure;

\item[{\bf (b)}] insisting on viewing the infinite jet bundle as a pro-object requires the introduction of considerably extra categorical machinery, 
while still not embedding into $\SmoothSets$ (or $\DiffeologicalSpaces$), and hence not treating all objects appearing on the same categorical and 
geometrical footing. The resolution therein is to develop the rest of the field-theoretic concepts in Pro$\SmoothSets$ (or Pro$ \DiffeologicalSpaces$). 
This is a somewhat heavy conceptual step, and seems unnecessary from our perspective. 
\end{itemize} 
\end{remark} 

\newpage 
\subsection{Local Lagrangians, currents and symmetries}
\label{Sec-Currents}

In the case of interest to field theory, the above formulation of the infinite jet bundle correctly encodes the explicit 
form of local Lagrangian densities. 

\begin{definition}[\bf  Local Lagrangian density]
\label{LocalLagrangianDensity}
A local Lagrangian density is a map of smooth sections $\CL:\Gamma_M(F) \rightarrow \Omega^d(M)$ of the form 
\vspace{-3mm} 
\begin{align}
\CL \;:\; \Gamma_M(F) &\longrightarrow \Omega^d(M)\\
\phi &\longmapsto  L \circ j^\infty \phi \nn
\end{align}

\vspace{-2mm} 
\noindent  where $j^\infty$ is the jet prolongation $j^{\infty}:\Gamma_M(F) \rightarrow \Gamma_M(J^\infty F)$ and $L$ is
a \textit{smooth bundle} map
\vspace{-2mm} 
\[ 
\xymatrix@C=1.8em@R=.2em  {J^\infty_M F \ar[rd] \ar[rr]^L &   & \wedge^d T^*M \ar[ld]
	\\ 
& M & 
}.  
\]
\end{definition}
This definition does indeed reflect the formulas written in the physics literature. Locally, the Lagrangian density $L$
may be written as  $L=\bar{L}\big(x^\mu,\{u_I^a\}_{|I|\leq k}\big) \cdot \dd x^1\cdots \dd x^d $ for some smooth function $\bar{L}\in C^\infty(J^\infty_M F)$, 
which by Prop. \ref{FunctionsOnInftyJetBundle} depends (locally) on jets up to degree $|I|=k$. Hence, abusing notation 
slightly, the value of the local Lagrangian density $\CL(\phi)$ on a field $\phi\in \Gamma_M(F)$ 
may be represented by
\vspace{-2mm} 
$$
\CL(\phi) = L\circ j^\infty\phi = L\big(x^\mu, \phi^a,\{\partial_I\phi^a\}_{|I|\leq k}\big) 
= \bar{L}\big(x^\mu, \phi^a, \{\partial_I\phi^a\}_{|I|\leq k} \big) \cdot \dd x^1\cdots \dd x^d = \bar{\CL}(\phi) \cdot \dd x^1\cdots \dd x^d\, ,
$$

\vspace{-2mm} 
\noindent
as is common in the physics literature (see e.g. \cite{HenneauxTeitelboim92}\cite{DF99}).

\begin{example}[\bf O($n$)-model Lagrangian]\label{VectorValuedFieldTheoryLagrangian}
Let $M\times W\rightarrow M$ be the trivial vector bundle over spacetime M equipped with a Lorentzian (or Riemannian) metric $g$ and $W\cong \FR^N$ 
equipped with a Euclidean inner product $\langle-,-\rangle$. The ``${\rm O}(n)$-model'' Lagrangian\footnote{Higher polynomial terms, 
i.e., ``interactions'', are also considered. Generally, this is accomplished by adding a term $V(\langle\phi,\phi \rangle)$ for $V:\FR \rightarrow \FR$
some polynomial function.} is given by 
\vspace{-1mm} 
\begin{align*}
\CL(\phi)&=\tfrac{1}{2} \langle \dd_M\phi\, ,\,  \star \dd_M\phi \rangle + \tfrac{1}{2}\big(  c_2 \cdot \langle \phi , \phi \rangle 
+ \tfrac{1}{2} c_4 \cdot (\langle \phi,\phi \rangle)^2 \big) \cdot \dd \mathrm{vol}_g \\ 
&=\tfrac{1}{2}\big(\langle \dd_M\phi\, ,\,  \dd_M\phi \rangle_g +  c_2 \cdot \langle \phi , \phi \rangle + \tfrac{1}{2} c_4 \cdot 
(\langle \phi,\phi \rangle)^2 \big) \cdot \dd \mathrm{vol}_g \\ 
&= \tfrac{1}{2}\big(g^{\mu \nu} \cdot \partial_\mu \phi^a \cdot \partial_\nu \phi_a  +  c_2 \cdot \phi^a \phi_a + \tfrac{1}{2} c_4 \cdot 
\phi^a \phi_a \cdot \phi^b \phi_b \big) \cdot \sqrt{|g|}\cdot \dd x^1 \cdots \dd x^d
\end{align*}

\vspace{-1mm} 
\noindent for some ``coupling constants'' $c_2,c_4\in \FR$, where $*_g: \Omega^{k}\rightarrow \Omega^{d-k}(M)$ is the Hodge operator
of $g$, $\phi^a \cdot \phi_a := \langle \phi , \phi \rangle =
\delta^a_{\, \, b} \cdot \phi^a \cdot \phi^b$
and similarly for $\langle -, -\rangle_g:= g \otimes \langle -, - \rangle $ as pairing on $\Omega^{1}(M)\otimes W$. By the local
coordinate expression, the Lagrangian density is local $\CL(\phi) = L\circ j^\infty(\phi)$, for the Lagrangian bundle map
$L:J^\infty_{M}(M\times W)\rightarrow \wedge^d(T^*M)$ over M given locally by
\vspace{-1mm} 
$$
L=  \tfrac{1}{2}\big(g^{\mu \nu} \cdot u^a_\mu\cdot u^b_\nu \cdot \delta^a_{\, \, b} +  
c_2 \cdot u^a  u^b\cdot \delta^a_{\, \, b}+ \tfrac{1}{2}c_4 \cdot  u^a  u^c\cdot \delta^a_{\, \, c}\cdot u^b  u^d\cdot \delta^b_{\, \, d} \big)
\cdot\sqrt{|g|} \cdot \dd x^1 \cdots \dd x^d\, .
$$

\vspace{0mm} 
\noindent In the case of Minkowski (or Euclidean) spacetime $(M,g)=(\FR^d,\eta)$, the expression simplifies further with the 
determinant being $\sqrt{|\eta|}=1$. 
Taking further the fiber to be one dimensional $W\cong \FR$, the Lagrangian reduces to that of ``$\phi^4$ scalar field theory''. 
Choosing instead the spacetime $(M,g)=(\FR^1_t, \delta)$ to be the $1$-dimensional `time' line with the canonical metric, 
the Lagrangian reduces to that of particle mechanics coupled to a background potential.
\end{example}

Collecting the above results and observations, we see that any local Lagrangian density constitutes a smooth map.

\begin{lemma}[{\bf Local Lagrangian is smooth}]
\label{LocalLagrangianisSmooth}
A local Lagrangian density $\CL:\Gamma_M(F) \rightarrow \Omega^d(M)$ canonically extends to a smooth map of smooth sets
\vspace{-3.5mm} 
\begin{align*}
\CL \;:\; \mathbold{\Gamma}_{M}(F) \; \longrightarrow \;
\Omega^{d}_{\mathrm{Vert}}(M)\cong \Omega^{d}(M)\hat{\otimes} y(\FR)\, ,
\end{align*}
where $\hat{\otimes}$ denotes the (plot-wise)  completed projective tensor product.
\begin{proof}
For any $\tilde{\phi}^k \in \mathbold{\Gamma}_M(J^\infty F)(\FR^k)$ smoothly $\FR^k$-parametrized section of $J^\infty_M(F)$,  
the Lagrangian density smooth bundle map $L:J^\infty F \rightarrow \wedge^d T^*M$ underlying $\CL$ induces the following  diagram
\vspace{-2mm} 
\[ 
\xymatrix@C=1.8em@R=.4em  {& J^\infty_M F \ar[rd] \ar[rr]^L &   & \wedge^d T^*M \ar[ld]
	\\ 
 \FR^k\times M \ar[ru]^{\tilde{\phi}^k} \ar[rr] & & M & 
} \, ,  
\]

\vspace{-2mm} 
\noindent where the outer part commutes since the two inner triangles commute. Hence the composition $L\circ \tilde{\phi}^k$
defines a smoothly $\FR^k$-parametrized $n$-form on $M$, i.e., $L\circ \tilde{\phi}^k \in \Omega^{d}_{\mathrm{Vert}}(M)(\FR^k)$.
Moreover, the assignment $\tilde{\phi}^k \mapsto L\circ \tilde{\phi}^k$ is manifestly functorial with respect to pullbacks 
of maps $\FR^m\rightarrow \FR^k$, and so $L$ naturally defines a map of smooth sets 
\vspace{-2mm} 
$$
L \;:\; \mathbold{\Gamma}_M(J^\infty F) \; \longrightarrow \; \Omega^{d}_{\mathrm{Vert}}(M)\, .
$$
\newpage 
\vspace{-2mm} 
\noindent The smooth Lagrangian $\CL:= L\circ j^\infty$ is defined by pre-composing with the smooth infinite prolongation map 
$j^\infty:\mathbold{\Gamma}_M(F)\rightarrow \mathbold{\Gamma}_M(J^\infty F)$ from \eqref{smoothjetprolongation}, that is,
\vspace{-3mm} 
\begin{align}
 \CL \;:\; \mathbold{\Gamma}_M(F) & \; \longrightarrow \; \Omega^{d}_{\mathrm{Vert}}(M)\\
 \phi^k & \; \longmapsto \; L(j^\infty \phi^k) \, , \nn
\end{align}

\vspace{-2mm} 
\noindent where $\phi^k\in \mathbold{\Gamma}_M(F)(\FR^k)$. Finally, it is a smooth map as a composition of smooth maps.
\end{proof}
\end{lemma}

At this point, we have described all the ingredients needed to define a bosonic, smooth and local field theory via objects and maps in the 
category of smooth sets.
\begin{definition}[\bf  Local Lagrangian Field Theory]\label{LocalLagrangianFieldTheoryDefinition}
A (bosonic) \textit{smooth, local Lagrangian field theory} is defined to be a pair
\vspace{-2mm}
$$
(\CF,\CL)
$$

\vspace{-2mm} 
\noindent where $\CF=\mathbold{\Gamma}_M(F)\in \SmoothSets$ is a smooth field space of sections 
of a (finite-dimensional) fiber bundle $F$ over the space-time $M$, and 
$\CL= L\circ j^\infty : \CF\rightarrow \Omega^{d}_\mathrm{Vert}(M) $ is the smooth Lagrangian density defined by a smooth bundle map
$L:J^\infty_M F\rightarrow \wedge^d T^*M$. 
\end{definition}

\begin{remark} [\bf Finite- vs. infinite-degree jet bundles in physics] 
$\,$

\noindent {\bf (i)} Each of the fundamental Lagrangians appearing in theoretical physics, e.g. General Relativity, Yang--Mills \footnote{Strictly speaking, 
our discussion applies verbatim for Yang--Mills and Chern--Simons \textit{only if} the field space is taken to be $\CF=\Omega^1(M,\frg)=\Gamma_M(T^*M\times_M \frg)$, 
i.e., connections on \textit{the trivial $G$-bundle} $M\times G\rightarrow M$. Generally, however, the field space consists of connections $A\in \mathrm{Conn}(P)$ 
on an arbitrary principal $G$-bundle over $M$, which are not (canonically) the set of sections of a fiber bundle over $M$. For a fixed topological 
sector $P\rightarrow M$, one can choose a reference connection $A_0$ to yield a bijection $\mathrm{Conn}(P)\cong_{A_0}\Gamma_M\big(\mathrm{ad}(P)\big)$ 
where $\mathrm{ad}(P)=P\times_G \frg $ is the adjoint bundle. The jet bundle formalism then applies on the corresponding smooth set of the right-hand side, 
therefore all diffeomorphisms and symmetries defined necessarily preserve $A_0$. Namely, such an approach treats the \textit{perturbation} theory around 
a fixed connection $A_0$. 
Consistently accounting for \textit{all} topological sectors nonperturbatively \textit{and} furthermore (local) gauge symmetries simultaneously 
forces a further natural generalization of smooth sets (Rem. \ref{OnGaugeEquivalencyandRedundancy}).} 
and Chern--Simons theories, factor \textit{globally} through a finite degree jet bundle $J^k_M F$.  \footnote{In particular, these fundamental Lagrangians all depend
(at most) on second-order jets. However, higher-order jet dependence corrections are often suggested to support deviations on experimental data.}

\noindent {\bf (ii)} Generic algebraic operations on fixed order Lagrangians, such as integration by parts (and further operations in the 
variational bicomplex; see \cref{VariationalBicomplexSection}), result in objects that necessarily factor through higher jet bundles, 
and so it is only natural to not fix an order in the jet bundle and consider the infinite limit instead. 

\noindent {\bf (iii)} To our knowledge, there are no explicit examples of theories that only \textit{locally} factor through finite order jet bundles, as in the
locally pro-manifold picture. Working in this slightly more general setting is nevertheless natural in order to place ourselves within the
convenient setting of smooth sets. On the other hand, to our knowledge, there is no fundamental physical principle that excludes the possibility of 
locally finite order Lagrangians, and it would be interesting to find explicit examples and properties of such theories.
 \end{remark} 

\begin{remark}[\bf Non-local Lagrangian Field Theory]\label{NonLocalSmoothLagrangian}
Although we will not pursue this here, the setting of smooth sets naturally accommodates a more general notion of a smooth Lagrangian field theory, 
given by  a field space $\CF=\mathbold{\Gamma}_M (F)$ and a (possibly non-local) smooth Lagrangian density
\vspace{-2mm} 
\begin{align*}
\CL \;:\; \CF \longrightarrow \Omega^{d}_{\mathrm{Vert}}(M)\, .
\end{align*}

\vspace{-2mm} 
\noindent Common examples of such Lagrangians include (effective) field theories which factor through direct products of $J^\infty_M(F)$ 
(also termed ``multi-local"). Nevertheless, as will be made clear throughout this manuscript, the \textit{locality} property is a crucial 
ingredient for the majority of well-known constructions and results of classical field theory. 
\end{remark}
 
The statement (and proof) of Lem. \ref{LocalLagrangianisSmooth} holds in more generality, for differential operators
viewed as maps of sections of bundles that factor through the infinite jet bundle.

\begin{lemma}[{\bf Differential operators as smooth maps}]\label{DiffOpsAsSmoothMaps}
Let $F$ and $G$ be two (finite-dimensional) fiber bundles over $M$, and let $P$ be a smooth bundle map covering a diffeomorphism $f:M\rightarrow M$
\vspace{-2mm} 
\[ 
\xymatrix@C=4em@R=.5em  {J^\infty_M F \ar[dd] \ar[rr]^P &   &  G \ar[dd]
	\\  \\ 
M \ar[rr]^f &  &  M \, .
}    
\]

\vspace{-1mm} 
\noindent {\bf (i)} $P$ naturally defines a smooth map 
\vspace{-3mm} 
\begin{align*}
P \;:\; \mathbold{\Gamma}_M(J^\infty F)& \; \longrightarrow \; \mathbold{\Gamma}_M(G)
\\
\tilde{\phi}^k& \;\longmapsto \; P\circ \tilde{\phi}^k\circ (\id_{\FR^k}, f^{-1})\, , 
\end{align*}

\vspace{-2mm} 
\noindent where $\tilde{\phi}^k \in \mathbold{\Gamma}_M(J^\infty F)(\FR^k)$ is any $\FR^k$-parametrized section
of $J^\infty_M(F)$. 
\newpage 
\noindent {\bf (ii)} Furthermore, it defines a differential operator
\vspace{-2mm} 
\begin{align*}
\CP \;:\; \Gamma_M(F)& \;\longrightarrow \; \Gamma_M(G)\\[-1pt]
\phi&\;\longmapsto \;  P\circ j^\infty (\phi) \circ f^{-1}
\end{align*}

\vspace{-1mm} 
\noindent which extends to a smooth map $\CP:= P\circ j^\infty$, where 
$j^\infty:\mathbold{\Gamma}_M(F)\rightarrow \mathbold{\Gamma}_M(J^\infty F)$ is the 
smooth infinite jet prolongation \eqref{smoothjetprolongation}
\vspace{-1mm}
\begin{align}
\CP \; : \; \mathbold{\Gamma}_M(F)&\longrightarrow \mathbold{\Gamma}_M(G)
\\[-2pt]
\phi^k&\; \longmapsto\; P\circ j^\infty (\phi^k)  \circ (\id_{\FR^k}, f^{-1})\;. \nn 
\end{align}

\begin{proof}
This follows as in Lem. \ref{LocalLagrangianisSmooth}.
\end{proof}
\end{lemma}
For future reference, let us note that any map smooth $P:J^\infty_M F \rightarrow G$ covering a diffeomorphism as in Lem. \ref{DiffOpsAsSmoothMaps}
canonically induces a \textit{smooth} `prolongated' map of infinite jet bundles $\pr P :J^\infty_M F \rightarrow J^\infty_M G$ 
covering $P$. \footnote{A special case of this definition, where it is restricted to bundle maps $J^\infty_M F\rightarrow G$ arising
as pullbacks of bundle maps $F\rightarrow G$, appears in \cite[Def. 7.2.10]{Saunders89}  in the Fr\'{e}chet setting
and \cite[Def 1.2]{Anderson89}  in the pro-manifold setting.}

\begin{definition}[\bf  Prolongation of jet bundle map]\label{ProlongationOfJetBundleMap}
Let $P:J^\infty_M F \rightarrow G$ be a smooth bundle map covering a diffeomorphism $f:M\rightarrow M$. The \textit{prolongation} 
of $P$ is the bundle map defined by
\vspace{-1mm} 
\begin{align*}
\pr P \;:\; J^\infty_M F &\longrightarrow J^\infty_M G \\
j^\infty_p \phi &\longmapsto j^\infty_{f(p)} \big(P \circ j^\infty \phi \circ f^{-1}\big) \, ,
\end{align*}

\vspace{-1mm} 
\noindent where $\phi:U\rightarrow F$ is any representative local section. By construction, it covers the original bundle map  
\vspace{-2mm} 
	\[ 
\xymatrix@C=1.6em@R=1em{ &&  J^\infty_M G \ar[d]
	\\ 
	J^\infty_M F \ar[rru]^{\pr P } \ar[rr]^>>>>>>>{P} && G \, ,
}   
\]

\vspace{-2mm} 
\noindent and hence also the diffeomorphism $f:M\rightarrow M$.
\end{definition}
The prolongation $\pr P:J^\infty_M F \rightarrow J^\infty_M G$ is: 

\noindent {\bf (a)} well-defined, i.e., independent of the choice representative section $\phi: U\rightarrow F$ since the right-hand 
side manifestly depends only on the derivatives of $\phi$ at $p\in M$, all of which are encoded in $j^\infty_p \phi$,

\noindent {\bf (b)} smooth as an application of Prop. \ref{SmoothMapsofLocProMan} by expanding in local coordinates.

\noindent By construction, it follows that the induced smooth map  on sections 
$\pr P \circ j^\infty : \mathbold{\Gamma}_M(F) \rightarrow \mathbold{\Gamma}_{M}(J^\infty G) $ (Lem. \ref{DiffOpsAsSmoothMaps}) 
is related to that of $\CP:= P\circ j^\infty (-) \circ(\id,f^{-1}):\mathbold{\Gamma}_M(F)\rightarrow \mathbold{\Gamma}_M(G) $ by 
\vspace{-2mm}
\begin{align}\label{ProlongatedActionOnSections}
 \pr P \circ j^\infty = j^\infty \circ \CP  \, . 
\end{align}

Let us return to the field-theoretic setting where the \textit{smooth} differential operator maps of Lem. \ref{DiffOpsAsSmoothMaps} 
naturally appear. As explained, of course, one such case is a smooth Lagrangian $\CL: \CF \rightarrow \Omega^{d}_{\mathrm{Vert}}(M)$ 
induced by a bundle map $J^\infty_M F \rightarrow \wedge^d T^*M$ over $M$. Completely analogously, we may define \textit{local} 
$p$-form valued maps on the field space, for any $0\leq p\leq d$. Indeed, consider Lem. \ref{DiffOpsAsSmoothMaps} in the case 
where $G=\wedge^p T^*M$ as a bundle over $M$. Then $\mathbold{\Gamma}_M(\wedge^p T^*M)\cong \Omega^{p}_{\mathrm{Vert}}(M)$, 
with its $*$-plots being usual smooth $p$-forms on the base $M$. 
In this case, precomposing a bundle map $J^\infty_M F \rightarrow \wedge^p T^*M$ 
by the jet prolongation map gives a smooth map
$\CF \rightarrow \Omega^{p}_{\mathrm{Vert}}(M) \, .$

\begin{definition}[\bf  Local currents on field space]\label{CurrentOnFieldSpace}
A local (smooth) \textit{$p$-form current} on field space is a smooth map
\vspace{-1mm} 
$$
\CP:= P \circ j^\infty  \;:\; \CF \; \longrightarrow \; \Omega^{p}_{\mathrm{Vert}}(M) 
$$

\vspace{-1mm} 
\noindent  induced by a smooth bundle map $P:J^\infty_M F \rightarrow \wedge^p T^*M$.  
\end{definition}

\begin{example}[\bf Vector-valued field space currents]\label{VectorValuedFieldSpaceCurrents} 
Constructing local currents for trivial vector field bundles as in Ex. \ref{VectorFieldTheoryFieldSpace} is straightforward. For instance:

\noindent {\bf (i)} Let $\phi = \phi^a\cdot e_a$ be some basis decomposition of $\CF=[M,W]$. Then 
\vspace{-1mm}
$$
\CP_1(\phi) = \dd_M \phi^a = \partial_\mu \phi^a \cdot \dd x^\mu \quad \in \quad \Omega^1(M)
$$

\vspace{-1mm}
\noindent
defines a local $1$-form current, for $P_1: J^\infty_M F \rightarrow T^*M$ given locally by $u^a_\mu \cdot \dd x^\mu$.

\noindent {\bf (ii)} Let $B= (B_{[ab]} ):W\times W\rightarrow \FR$ be any antisymmetric bilinear map. Then 
\vspace{-2mm}
$$
\CP_2(\phi)= B\circ(\dd_M\phi, \dd_M\phi) =  B_{[ab]} \cdot \dd_M \phi^a \wedge \dd_M \phi^b \quad \in \quad \Omega^2(M)
$$ 

\vspace{-1mm}
\noindent
defines a local $2$-form current, for $P_2: J^\infty_M F \rightarrow \wedge^2 T^*M $ given locally by $B_{[ab]}\cdot u^a_\mu u^b_\nu \cdot \dd x^\mu\wedge \dd x^\nu  $. 

\vspace{1mm} 
\noindent {\bf (iii)} Recall the field space vector fields $\CZ^A(\phi) = A^a_{\, b} \cdot \phi^b \cdot \frac{\delta}{\delta \phi^b}$ and 
$\CZ^\nu (\phi) = \partial_\mu \phi^a \cdot \nu^\mu \cdot \frac{\delta}{\delta \phi^a}$ from Ex. \ref{DiffeomorphismForVector-valuedFieldTheoryViaTarget} 
and Ex. \ref{DiffeomorphismForVector-valuedFieldTheoryViaBase}, respectively. Then 
\vspace{-2mm}
\begin{align*}
\CP_{\CZ^A}(\phi) 
&= - \langle \CZ^A(\phi),\, \star \dd_M \phi \rangle := - \CZ^A(\phi)^a \wedge \star \dd_M\phi 
= -A^a_{\, b} \cdot \phi^b \wedge \star \dd_M \phi_a   \quad \in \quad \Omega^{d-1}(M) 
\\  \CP_{\CZ^\nu}(\phi)& = \iota_\nu \big(  \CL(\phi)\big) - \langle \CZ^\nu(\phi), \, \star \dd_M \phi \rangle = \iota_\nu \big(  \CL(\phi)\big)  - \nu^\mu \cdot \partial_\mu \phi^a \wedge \star \dd_M \phi_a  \quad \in \quad \Omega^{d-1}(M) 
\end{align*}

\vspace{-1mm}
\noindent 
define local $(d-1)$-form currents. The corresponding jet bundle maps can be read off by expanding in coordinates.
\end{example}

\newpage 
Integration defines another natural map of smooth sets. More precisely, if the base spacetime $M$ is compact 
and oriented,\footnote{Throughout this text we shall assume the spacetime $M$ is orientable, with a fixed choice of orientation.
Treating non-orientable spacetimes requires (minor) technical modifications (e.g., Lagrangians and currents as valued 
in forms twisted by the orientation bundle) which we will make explicit here. Such modifications are implicit in \cite{DF99}, 
and technical details may be found in \cite{Del}.}
then integration along $M$ defines a smooth map 
\vspace{-2mm} 
\begin{align}\label{SmoothIntegrationMap}
\int_M : \;\; \Omega^{d}_{\mathrm{Vert}}(M) \longrightarrow y(\FR)\, .
\end{align}

\vspace{-2mm} 
\noindent
Explicitly, for any $\FR^k$-parametrized top-form $\om_{\FR^k}\in \Omega^{d}_{\mathrm{Vert}}(M)(\FR^k)$, the value 
of the function $\int_M \om_{\FR^k} \in y(\FR)(\FR^k)\cong C^\infty(\FR^k, \FR) $ is given by  
$$
\bigg(\int_M \om_{\FR^k}\! \bigg)(x) := \int_M \iota_{x}^* \om_{\FR^k}\, ,
$$
where $\iota_x$  is the inclusion $M\xrightarrow{\sim}\{x\}\times M\subset \FR^k\times M $. 
In plain words, one integrates along $M$ while keeping the $\FR^k$-dependence fixed. If $M$ is not compact, 
the smooth integration map is only defined on the smooth subset of compactly supported 
forms $\Omega^{d}_{\mathrm{Vert},c}(M)\hookrightarrow \Omega^d_\mathrm{Vert}(M)$.

\medskip 
Composition of a smooth Lagrangian $\CL$ with the smooth integration map defines the action functional as 
a map of smooth sets 
\vspace{-2mm} 
\begin{align}\label{LocalActionisSmooth}
S=\int_M \circ \;\; \CL \; :\;  \mathbold{\Gamma}_M(F) \longrightarrow y(\FR)\, .
\end{align}
If $\CL$ is interpreted as a $d$-form current, then the action may be thought of as its charge over $M$. 
If the spacetime $M$ is not compact, the integration map is only defined on the smooth subset of compactly supported fields
$\mathbold{\Gamma}_{M,c}(F)\hookrightarrow\mathbold{\Gamma}_M(F)$, 
on which the action functional is still defined as a smooth map. \footnote{Or more generally smooth fields with 
appropriate decay behavior. Alternatively, one integrates against test bump functions instead.} 

\medskip  

Analogously to the case of the action functional, we may integrate a $p$-form 
current along any $p$-dimensional (oriented, compact) sub-manifold to obtain a smooth, `local'
real-valued function on field space.

\begin{definition}[\bf  Charges on field space]\label{ChargesOnFieldSpace}
The \textit{charge} of a local $p$-form current $\CP$ along a $p$-dimensional oriented, compact submanifold 
$\Sigma^p \hookrightarrow M$ is the smooth map
\begin{align*}
\CP_{\Sigma^p}:= \int_{\Sigma^p} \CP = \int_{\Sigma^p} P\circ j^\infty 
\quad 
: \quad
\CF \; \longrightarrow \; 
\Omega^{d}_{\mathrm{Vert}}(M) \; \longrightarrow \; y(\FR)\,.
\end{align*}

\end{definition}
The physical interpretation is the traditional one. For any field $\phi \in \CF(*)$, a point in field space, the charge map 
assigns the real number  $\int_{\Sigma^p}P \circ j^\infty \phi \in \FR$, i.e., the total flux through the submanifold 
$\Sigma^p$ of the current $\CP(\phi)$ corresponding to the chosen field configuration. 
The upshot is that such charges are in fact smooth maps within the category of smooth sets, which allows one to
study them rigorously. 
In particular, we may employ currents and their charges as the 
basic ingredients for general \textit{local} functionals on field space, which are a particular subalgebra of all smooth real-valued maps.

\begin{definition}[\bf  Local functionals on field space]\label{LocalFunctionsOnFieldSpace}
The algebra of (smooth) \textit{local real-valued functionals} 
\vspace{-1mm} 
$$
C^\infty_\mathrm{loc}(\CF)\longhookrightarrow C^\infty(\CF)
$$ 

\vspace{-1mm} 
\noindent on the field space $\CF$ is the (minimal) sub-algebra generated by charges of currents, and more generally by integrating currents
$\CP:=P \circ j^\infty : \CF\rightarrow \Omega^{p}_{\mathrm{Vert}}(M)$ against bump functions along non-compact oriented 
submanifolds. \footnote{This can further be enlarged to include limits of charges against bump functions, and hence observables 
of distributional nature. This extension will not be needed for our discussion.} 

\end{definition}

\begin{remark}[\bf Algebras of currents and local functionals]\label{AlgebrasOfCurrentsAndLocalFunctions}
 The set of currents inherits a (graded) algebra structure from the target $\Omega^{\bullet}_{\mathrm{Vert}}(M)$, i.e.,
 the wedge product of the currents induced by 
 $P:J^\infty_M F \rightarrow \wedge^p T^*M$ and $P':J^\infty_M F \rightarrow \wedge^{p'} T^*M$ is given by
\vspace{-2mm} 
 $$
 \CP \wedge \CP' := P\wedge P' \circ j^\infty \;\; : \;\; \CF \longrightarrow \Omega^{p+p'}_\mathrm{Vert}(M)\, .
 $$

 \vspace{0mm} 
\noindent The resulting current induces charges $(\CP\wedge \CP')_{\Sigma^{p+p'}}\in C^\infty_{\mathrm{loc}}(\CF)$ by 
integrating over submanifolds 
$\Sigma^{p+p'}$ of dimension $p+p'$. Note that this is not directly related to the product of the corresponding charges,
 \vspace{-2mm} 
 $$
 \CP_{\Sigma^p}\cdot \CP'_{\Sigma^{p'}}= \bigg(\int_{\Sigma^p} P\circ j^\infty\bigg)\cdot \bigg(\int_{\Sigma^{p'}} P'\circ j^\infty\bigg) \in C^\infty_{\mathrm{loc}}(\CF)   
 $$

  \vspace{-2mm} 
\noindent induced by first integrating against submanifolds $\Sigma^p$, $\Sigma^{p'}$ of dimension $p$ and $p'$, respectively.
\end{remark}
In the case of $p=0$, where $\Omega^{0}_{\mathrm{Vert}}(M)= [M,\FR] $ is the smooth space of smooth $\FR$-valued functions on $M$. 
By integration over $0$-dimensional manifolds,  i.e.,
points in $M$, we mean evaluation at the given point.

\begin{example}[\bf Field point amplitude]\label{FieldPointAmplitudeExample}
For any chart $\big\{x^\mu, \{u_I^a\}_{0\leq |I|}\big\}$ of $J^\infty_M F$ above a point $x\in M$, we have the smooth, local point evaluation maps 
$$
\hat{\phi}^a_{I}(x):= \mathrm{ev}_x \circ u^a_I \circ j^\infty 
\quad : \quad 
\CF \longrightarrow \Omega^{0}_{\mathrm{Vert}}(M)\cong [y(M),y(\FR)] \longrightarrow y(\FR) \, .
$$
On $*$-plots $\CF$, i.e., on field configurations $\Gamma_M(F)$, this extracts the amplitude (of the derivatives) of the 
 field at the point $x$, in the given chart,
 $$
 \hat{\phi}^a_I(x) (\phi) = \partial_I \phi^a (x) \, ,
 $$
 and similarly for $\FR^k$-plots of fields. In the physics literature, these `point observables' are a common source of confusion, 
 since they are usually denoted by exactly the same symbol as the field itself.
\end{example}

We close off this section by defining the appropriate notion of a symmetry of a Lagrangian field theory. With the notion of diffeomorphisms $\CF\rightarrow \CF$ as in Def. \ref{SmoothAutomorphisms} at hand, one can expect a (smooth) symmetry of
a local Lagrangian field theory $(\CF,\CL)$ to be any diffeomorphism $\CD:\CF\rightarrow \CF$ that preserves the Lagrangian, $\CL\circ \CD = \CL$. 
However, since `exact local Lagrangians' induce trivial dynamics (see Eq. \eqref{ExactLagrangianTrivialELequation} onwards), it is 
natural to relax the condition to preserve the Lagrangian up to an exact Lagrangian. Indeed, it is this notion of symmetry (in its local version) that preserves
the corresponding (smooth) space of on-shell fields (Prop. \ref{LocalSymmetryPreservesOnshellSpace}), and (in its infinitesimal version) induces (smooth) conserved
currents via Noether's 1st theorem (Prop. \ref{Noether1st}). In the current setting of \textit{finite} symmetries, as we will see, physical examples dictate that
a Lagrangian may be further only preserved up to a pullback by a spacetime diffeomorphism, in which case the (smooth) space of on-shell fields is again preserved (Prop. \ref{SymmetryPreservesOnshellSpace}).

\begin{definition}[\bf  Symmetry of Lagrangian field theory]\label{FiniteSymmetryofLagrangianFieldTheory}
$\,$

\noindent {\bf (i)} A \textit{local symmetry} of a local Lagrangian field theory $(\CF,\CL)$ is a \textit{local} diffeomorphism $\CD :\CF\rightarrow \CF$ 
such that there exists a local current (d-1)-form $\CK$ that makes  diagram 
\vspace{-2mm} 
\[ 
\xymatrix@C=6em@R=.2em  {\CF \ar[rd]_-{\CL + \dd_M  \CK ~~} \ar[rr]^{\CD} &   & \CF \ar[ld]^-{\CL}
	\\ 
& \Omega^{d}_{\mathrm{Vert}}(M) & 
} 
\]

\vspace{-2mm} 
\noindent commute. That is, $\CL\circ \CD = \CL + \dd_M \CK $ where 
$\CK:= K\circ j^\infty : \CF\rightarrow \Omega^{d-1}_{\mathrm{Vert}}(M)$ for some bundle map 
$K:J^\infty_M F \rightarrow \wedge^{d-1} T^*M$ and $\CD:= D\circ j^\infty: \CF\rightarrow \CF$ for some bundle map 
\vspace{-2mm} 
\[ 
\xymatrix@C=2.8em@R=.2em  {J^\infty_M F \ar[rd] \ar[rr]^D &   &  F \ar[ld]
	\\ 
& M & 
}   
\]

\vspace{-2mm} 
\noindent as per Lem. \ref{DiffOpsAsSmoothMaps}.

\noindent {\bf (ii)} A \textit{spacetime covariant symmetry} of a local Lagrangian field theory $(\CF,\CL)$ is a diffeomorphism $\CD:\CF\rightarrow \CF$, 
the differential operator induced by a bundle map $D:J^\infty_M F \rightarrow F$ covering a diffeomorphism $f:M\rightarrow M$, such that there exist 
local current (d-1)-form $\CK$ that makes the diagram
\vspace{-2mm} 
\[
\xymatrix@C=6em@R=2em  {\CF \ar[d]^{\CL+\dd_M \CK} \ar[rr]^{\CD} &&   \CF\ar[d]^{\CL} 
	\\ 
 \Omega^{d}_{\mathrm{Vert}}(M)\ar[rr]^-{f^*}  && 
\Omega^{d}_{\mathrm{Vert}}(M)\, . } 
\]

\vspace{-2mm}
\noindent commute. In other words, $\CL\circ \CD = f^* \circ \CL + f^* \circ \dd_M \CK $ where 
$\CK:= K\circ j^\infty : \CF\rightarrow \Omega^{d-1}_{\mathrm{Vert}}(M)$ for some bundle map 
$K:J^\infty_M F \rightarrow \wedge^{d-1} T^*M$ and $\CD:= 
D\circ j^\infty (-) \circ (\id, f^{-1}): \CF\rightarrow \CF$ for some bundle map
\vspace{-2mm} 
\[ 
\xymatrix@C=4em@R=.8em  {J^\infty_M F \ar[dd] \ar[rr]^P &   &  F \ar[dd]
	\\  \\ 
M \ar[rr]^f &  &  M \, .
}    
\]

\vspace{-2mm}
\noindent as per Lem. \ref{DiffOpsAsSmoothMaps}. Colloquially, one may say $\CD$ is a ``local symmetry of $(\CF,\CL)$ up to a spacetime diffeomorphism''.
\end{definition}
Note that the invertibility of a local diffeomorphism $\CD:\CF\rightarrow \CF$ implies that the induced\footnote{By dimensionality, it is obvious that 
$D=\pi^\infty_0 \circ \tilde{D}:J^\infty_M F\rightarrow F$ itself cannot be invertible.} prolongated bundle map $\pr P:J^\infty_MF \rightarrow J^\infty_M F$
(Def. \ref{ProlongationOfJetBundleMap}) is necessarily an automorphism, with the inverse of $\CD$ given by $\CD^{-1}=(\pi^\infty_0 \circ (\pr D)^{-1}) \circ j^\infty$.
A similar statement holds for the diffeomorphism of a spacetime covariant symmetry. The diffeomorphisms of Ex. \ref{DiffeomorphismsViaFieldBundleAutomorphisms} 
arising from bundle automorphisms maps $\tilde{f}: F\rightarrow F$ covering diffeomorphisms are a special case of the above, by precomposing via 
$\pi^\infty_0 : J^\infty_M F\rightarrow F$.

\newpage 
We have already introduced the ingredients of an explicit example of such symmetries.
\begin{example}[\bf Symmetries of ${\rm O}(N)$-model]\label{LocalSymmetriesOfO(n)Model}
Consider the ${\rm O}(N)$-model Lagrangian $\CL:[M,W]\rightarrow \Omega^{d}_{\mathrm{Vert}}(M)$ of Ex. \ref{VectorValuedFieldTheoryLagrangian}, 
\vspace{-1mm} 
\begin{align*}
\CL(\phi)&=\tfrac{1}{2}\big(\langle \dd_M\phi\, ,\,  \dd_M\phi \rangle_g +  c_2 \cdot \langle \phi , \phi \rangle + \tfrac{1}{2} c_4 \cdot 
(\langle \phi,\phi \rangle)^2 \big) \cdot \dd \mathrm{vol}_g \, . 
\end{align*}

\vspace{-1mm}
\noindent {\bf (i)} Recall the postcomposition induced diffeomorphisms $\CD=g_*:[M,W]\rightarrow [M,W]$ of Ex. \ref{DiffeomorphismForVector-valuedFieldTheoryViaTarget}.
Choosing $g:W\rightarrow W$ to be an orthogonal transformation, $e_a\mapsto g^{b}_{\, \, a}\cdot e_b$ for $[g^{b}_{\, \, a}]\in \mathrm{O}(N,\FR)$, 
which (by definition) preserves the inner product on $W$, we obtain a local symmetry of the Lagrangian field theory:
\vspace{-1mm}
\begin{align*}
 \CL \circ \CD(\phi):&= \CL(g\circ \phi)= \tfrac{1}{2}\big(\langle \dd_M(g\circ \phi)\, ,\, 
 \dd_M(g\circ \phi) \rangle_g +  c_2 \cdot \langle g\circ \phi , g\circ \phi \rangle + \tfrac{1}{2} c_4 \cdot 
(\langle g\circ \phi, g\circ \phi \rangle)^2 \big) \cdot \dd \mathrm{vol}_g \\
&=
\tfrac{1}{2}\big(\langle g \circ \dd_M\phi\, ,\,  g\circ \dd_M\phi \rangle_g +  c_2 \cdot \langle \phi , \phi \rangle + \tfrac{1}{2} c_4 \cdot 
(\langle \phi,\phi \rangle)^2 \big) \cdot \dd \mathrm{vol}_g \\
&=
\tfrac{1}{2}\big(\langle \dd_M\phi\, ,\,  \dd_M\phi \rangle_g +  c_2 \cdot \langle \phi , \phi \rangle + \tfrac{1}{2} c_4 \cdot 
(\langle \phi,\phi \rangle)^2 \big) \cdot \dd \mathrm{vol}_g
\\
&= \CL(\phi) \, ,
\end{align*}

\vspace{-1mm}
\noindent where $\dd_M(g\circ \phi) = \dd_M (g^a_{\, \, b} \cdot \phi^b \cdot e_a )= g^{a}_{\, \, b} \cdot \dd_M \phi^b \cdot e_a = g\circ \dd_M \phi$, since 
$g^{a}_{\,\, b} \in \FR$ for each $a,b$. The calculation carries through identically for $\FR^k$-plots of fields. This is a symmetry of the theory with
$\dd_M \CK=0$ and $f=\id_M$.\footnote{Any choice of $(d-1)$-form current $\CK+\dd_M \CT$ for some `trivial' current $\dd_M \CT$ will serve to fill the diagram.}  
The symmetry is local, with $\CD=g_*$ 
factoring through $J^\infty_M F$ via $J^0 F= F$ as in Rem. \ref{DiffeomorphismsViaFieldBundleAutomorphisms}.

\noindent {\bf (ii)} Recall the precomposition induced diffeomorphisms $\CD:=f^*:[M,W]\rightarrow [M,W] $ of Ex. \ref{DiffeomorphismForVector-valuedFieldTheoryViaBase}.
Choosing $f:M\rightarrow M$ to be an isometry of the background metric $g$ \footnote{For instance, a Poincar\'e transformation in the case of Minkowski 
space $(M,g)=(\FR^d,\eta)$ .}, i.e., $f^*g=g$ and so further $f^* \dd\mathrm{vol}_g= \dd \mathrm{vol}_g$, we obtain a spacetime covariant symmetry of 
the Lagrangian field theory:
\vspace{-1mm}
\begin{align*}
 \CL \circ \CD(\phi):&= \CL(f^*\phi)= \tfrac{1}{2}\big(\langle \dd_M(f^* \phi)\, ,\, 
 \dd_M(f^* \phi) \rangle_g +  c_2 \cdot \langle  f^*\phi , f^*\phi \rangle + \tfrac{1}{2} c_4 \cdot 
(\langle f^*\phi,  f^*\phi \rangle)^2 \big) \cdot \dd \mathrm{vol}_g \\
&=
\tfrac{1}{2}\big(\langle f^*\dd_M \phi\, ,\,  f^*\dd_M \phi \rangle_{f^*g} +  c_2 \cdot f^* \langle  \phi , \phi \rangle + \tfrac{1}{2} c_4 \cdot 
f^*(\langle \phi,  \phi \rangle)^2 \big) \cdot f^*\dd \mathrm{vol}_g \\
&=
\tfrac{1}{2}\big(f^*\langle \dd_M \phi\, ,\,  \dd_M \phi \rangle_{g} +  c_2 \cdot f^* \langle  \phi , \phi \rangle + \tfrac{1}{2} c_4 \cdot 
f^*(\langle \phi,  \phi \rangle)^2 \big) \cdot f^*\dd \mathrm{vol}_g
\\
&= f^*\circ \CL(\phi) \, ,
\end{align*}

\vspace{-1mm}
\noindent where we used standard properties of pullback map on forms, metrics and functions on the spacetime $M$. Evidently, this is a spacetime 
covariant symmetry of the theory with $\dd_M \CK=0$ which is furthermore \textit{not} local. Nevertheless, by Ex.
\ref{DiffeomorphismForVector-valuedFieldTheoryViaBase} the corresponding vector field on $[M,W]$ is ``local'' (Def. \ref{LocalVectorFields})
factoring through $J^\infty_M F$ via $J^1_M F$, and the above finite spacetime covariant symmetry induces an \textit{infinitesimal local} 
symmetry of the theory (see Lem. \ref{InfinitesimalSpacetimeCovariantSymmetries} and Ex. \ref{InfinitesimalSpacetimeSymmetryOfO(N)-model}).
\end{example}

\begin{remark}[\bf On spacetime covariant symmetries]\label{OnSpacetimeCovariantSymmetries}
The physics literature usually focuses on \textit{local} symmetries of Def. \ref{FiniteSymmetryofLagrangianFieldTheory} as a symmetry 
of a local Lagrangian field theory, and more often only in its infinitesimal version (see Def. \ref{InfinitesimalLocalSymmetryOfLagrangian}) 
with the all-important application in Noether's Theorems (see Prop. \ref{Noether1st},  Prop. \ref{Noether2nd}). 

\noindent {\bf (a)} We generalize slightly with Def. \ref{FiniteSymmetryofLagrangianFieldTheory}\,{(ii)} to include, in particular,
cases where one may lift a spacetime diffeomorphism to a diffeomorphism on field space (Rem. \ref{DiffeomorphismsViaFieldBundleAutomorphisms}). 
This is necessary to accommodate important examples appearing in physics, with Ex. \ref{LocalSymmetriesOfO(n)Model} being a particular instance. 
Other important field theories with such symmetries include the (2nd order) metric formulation of General Relativity, Yang-Mills theories, 
and Chern-Simons theories.\footnote{As emphasized in Rem. \ref{DiffeomorphismsViaFieldBundleAutomorphisms}, finite spacetime diffeomorphisms 
apply for pure general relativity, and at most when it is coupled to fields with trivial field bundles. Similarly for Yang-Mills 
and Chern-Simons theories, it strictly applies only to the trivial topological sector.}

\noindent {\bf (b)} Crucially, the infinitesimal version of any spacetime covariant symmetries is, in fact, an infinitesimal \textit{local} symmetry (Lem. \ref{InfinitesimalSpacetimeCovariantSymmetries}), which we believe is the reason their finite `non-local' aspect is often bypassed in the literature. 

\noindent {\bf (c)} One can check that \textit{local} symmetries $\mathrm{Diff}_{\mathrm{loc}}^{\CL}(\CF) \hookrightarrow \mathrm{Diff}_{\mathrm{loc}}(\CF)$
a local Lagrangian field theory form a (smooth) subgroup of all local diffeomorphisms.\footnote{The action of a composition of two \textit{local} 
symmetries gives 
\vspace{-2mm} 
$$
\CL\circ(\CD_1\circ \CD_2) =(\CL \circ \CD_1)\circ \CD_2 = (\CL+ \dd_M \circ \CK_1 ) \circ \CD_2 = \CL + \dd_M \circ \CK_2 + \dd_M \circ  (\CK_1 \circ \CD_2)\,.
$$

\vspace{-2mm} 
\noindent The latter term is also local since $\CK_1 \circ \CD_2= K_1 \circ j^\infty \circ D_2 \circ j^\infty = (\CK_1 \circ \pr D_2) \circ j^\infty $ where
$\pr D_2:J^\infty_M F\xrightarrow{\sim} J^\infty_M F$ is the prolongated bundle automorphism (Def. \ref{ProlongationOfJetBundleMap}). 
Thus the composed local diffeomorphism preserves the Lagrangian up to the exact local Lagrangian $\dd_M \circ (K_2 + K_1 \circ \pr D_2) \circ j^\infty$.}
 From a strict mathematical point of view, however, for a given local Lagrangian
$\CL=L\circ j^\infty$, the pullback smooth Lagrangian $f^*\circ\CL:\CF \rightarrow \Omega^d_{\mathrm{Vert}}(M)$ 
is \textit{not} local in the sense of Def. \ref{LocalLagrangianDensity}. In principle, one could expand the definition to include the orbits 
of local Lagrangians under $\mathrm{Diff}(M)$, such that symmetries of Lagrangian field theories may genuinely form a \textit{groupoid}
(and not a group). We do not follow this route in the current manuscript, as to avoid unnecessary confusion with the standard nomenclature.
\end{remark}

To rigorously phrase the variational calculus and infinitesimal local symmetries of local Langrangians, the notions of vector fields 
and differential forms on the infinite jet bundle become indispensable. We now move to the description of these concepts within the 
category of smooth sets.

\newpage 

\addtocontents{toc}{\protect\vspace{-10pt}}
\section{Differential geometry on the infinite jet bundle}\label{DifferentialGeometryOnTheInfiniteJetBundle}

\subsection{Tangent bundle and vector fields}
\label{InfinityJetTangentBundleSection}
The original manuscripts on the infinite jet bundle \cite{Takens79}\cite{Saunders89}\cite{Anderson89} define vector fields algebraically as
derivations on the algebra of functions $C^\infty(J^\infty_M F)$. We define the tangent bundle to the infinite jet bundle directly as a smooth set, 
and show how vector fields in the above sense are recovered as its geometrical smooth sections inside the category of smooth sets. When we enrich 
our spaces with infinitesimal structure, we will show in \cite{GSS-2} how this is naturally recovered as the `synthetic tangent bundle' of the infinite jet bundle. 

\medskip 
Recall by \eqref{JetBundleSmoothSetlimit} that the infinite jet bundle smooth set may be equivalently identified as 
$y(J^\infty_M F) \cong \mathrm{lim}_k^{\SmoothSets} \, y(J^k_M F)\, .$ Similarly, there is an induced diagram of finite-dimensional tangent bundles
\vspace{-2mm}
 \begin{align*}
 \longrightarrow T(J^k_M F) \xlongrightarrow{\dd{\pi^k_{k-1}}} T(J^{k-1}_M F) 
 \longrightarrow \cdots \longrightarrow  T(J^1_M F)\xlongrightarrow{\dd \pi^1_0} T(J^0_M F)\cong TF \, . 
 \end{align*}

 \vspace{-2mm}
\noindent with the maps being the pushforwards of the projections $\{\pi^k_{k-1}:J^k_M F \rightarrow J^{k-1}_M F\}_{k\in \NN}$. 
Embedding the diagram along $\SmoothManifolds\xhookrightarrow{y} \SmoothSets$, we define the tangent bundle of $y(J^\infty_M F)$ 
as the corresponding limit.

\begin{definition}[\bf  Infinite jet tangent bundle]
\label{InfinityJetTangentBundle}
The tangent bundle smooth set $T\big(y(J^\infty_M F)\big)\in \SmoothSets$ of the infinite jet bundle $y(J^\infty_M F) $ is defined as
\vspace{-3mm}
\begin{align}
  T\big(y(J^\infty_M F) \big)
  := 
  \underset{k \in \mathbb{N}}{\mathrm{lim}}^{\SmoothSets} 
  \, 
  y\big(T (J^k_M F)
  \big)\, . 
  \end{align}
\end{definition}
The points of this space, i.e., the set of tangent vectors of $J^\infty_M F$ is given by
 \vspace{-1mm}
$$
  T\big(y(J^\infty_M F)\big)(*)
  = 
  \underset{k \in \mathbb{N}}{\mathrm{lim}}^{\SmoothSets}\, 
  y\big(T (J^k_M F)\big)(*)
  :=
  \underset{k \in \mathbb{N}}{\mathrm{lim}}^{\mathrm{Set}}\big(T(J^k_M F)\big)\, ,  
$$

 \vspace{-2mm}
\noindent which is represented by  
 \vspace{-2mm}
$$
\bigcup_{s\in J^\infty_M F} T_s(J^\infty_M F) := \bigcup_{s\in J^\infty_M F}\big\{ \{X_{s}^k\in T_{\pi_k(s)} (J^k_M F)\,  |  \, 
\dd \pi^{k}_{k-1} X^{k}_s= X^{k-1}_s\}_{k\in \NN} \, \big\}\ . 
$$

 \vspace{-2mm}
\noindent That is, a tangent vector $X_s\in T_s(J^\infty_M F)$ at $s=j^\infty_p \phi \in J^\infty_M F$ is represented by an (infinite) family 
of tangent vectors $\{X^k_s \in T_{\pi_k(s)} (J^k_M  F)\}_{k\in \NN} $ on each finite order tangent bundle at 
$\pi_k(s)=\pi_k(j^\infty_p \phi) = j^k_p \phi \in J^k_M F$, compatible along the pushforward projections. 
In a local coordinate chart $\big\{x^\mu, \{u_I^a\}_{0\leq |I|}\big\}$ of $J^\infty_M F\in \mathrm{LocProMan}$ around $s$, such a family
may be represented by an infinite (formal) sum
\vspace{-3mm}
\begin{align}\label{InfinityJetTangentFormalSum}
X_s = X^\mu \, \frac{\partial}{\partial x^\mu}\Big\vert_s + \sum_{|I|=0}^\infty Y^a_I \frac{\partial}{\partial u_I^a} \Big\vert_s \, ,
\end{align}

 \vspace{-2mm}
\noindent for an infinite list of real numbers $\{X^\mu,Y_I^a\}\subset \FR$, with each $X^k$ corresponding to the case where the 
sum is terminated at order $|I|=k$.

\begin{lemma}[{\bf Infinite jet tangent vectors as infinitesimal curves}]
\label{InfinityJetTangentVectorsAsCurves} 
The set of tangent vectors $T\big(y(J^\infty_M F)\big)(*)$ is in bijection with equivalence classes 
of curves in $ J^\infty_M F\in \mathrm{LocProMan}$
\vspace{-1mm} 
\begin{align}
T\big(y(J^\infty_M F)\big)(*)\cong {\rm Hom}_{{\rm FrMan}} (\FR^1, J^\infty_M F) /\sim_{\CO(t^1)} 
\cong  y(J^\infty F) (\FR^1)/\sim_{\CO(t^1)} \, ,
\end{align}

 \vspace{-2mm}
\noindent
where the equivalence relation is agreement up to first order derivatives at $0\in \FR^1$.
\begin{proof}

By the limit property of $T\big(y(J^\infty F)\big)$, a tangent vector $X_s$ at a point $s\in J^\infty_M F$ corresponds uniquely 
to a compatible family $\{X_s^k \in T_{\pi_k(s)}(J^k_M F)\}_{k\in \NN}$ of tangent vectors on each finite order jet bundle. 
Since each $J^k_M F$ is a finite-dimensional manifold, each tangent vector $X^k_s$ in the family is represented by some curve
$\gamma^{k}_s : \FR^1 \rightarrow J^k_M F$ through $\pi_k(s)\in J^k_M F$, i.e.,
$X^k_s = \dot{\gamma}^{k}_s(0) = [\gamma]$, with the equivalence relation being that the derivatives agree up to first-order derivatives. 
The compatibility relation $\dd \pi^{k}_{k-1} X_s^{k+1}= X_s^k$ implies that the family of curves is compatible in the sense that
$[\pi^k_{k-1} \circ \gamma^{k}_s] = [\gamma^{k-1}_s]$.

By restricting on a compatible family of charts around each $\pi_k(s)\in J^k_M F$, a representative family of curves may be chosen such 
$$
\pi^k_{k-1} \circ \gamma^{\, k}_s = \gamma^{k-1}_s : \FR^1 \longrightarrow J^k_M F \, ,
$$
for each $k\in \NN$. By the limit property of $J^\infty_M F$, such a family uniquely corresponds to a curve $\gamma_s: \FR^1 \rightarrow J^\infty_M F$. 
A different choice of a representative family $\{\gamma_s^{k'} :\FR^1 \rightarrow J^k_M F\}_{k\in \NN}$ determines another curve 
$\gamma_s':\FR^1\rightarrow J^\infty_M F$, which however is equivalent to $\gamma_s$ under the induced relation. The induced equivalence on 
$\mathrm{Hom}_{\FrechetManifolds}(\FR^1, J^\infty_M F)$ is as expected: Any two curves through $s\in J^\infty_M F$ are equivalent
$\gamma_s\sim \gamma_s'$ if and only if 
\vspace{-2mm} 
$$
\frac{\dd}{\dd t}\Big\vert_{t=0} (\phi \circ \gamma_s) =\frac{\dd}{\dd t}\Big\vert_{t=0} (\phi \circ \gamma_s')
$$
for any local chart $\phi$ valued in $\FR^\infty$ around $s\in J^\infty_M$.
\end{proof}
\end{lemma}
The viewpoint of equivalence classes of curves is taken as a definition of tangent vectors in \cite{Saunders89}. 
On the other hand, the definition as a smooth set allows us to consider more than bare points in the tangent 
bundle of $J^\infty_M F$. The $\FR^n$-plots are given by
\vspace{-2mm} 
\begin{align}\label{TangentInfinityJetPlots}
T\big(y(J^\infty_M F)\big)(\FR^n)= {\rm lim}_{k}^{{\SmoothSets}}\, y\big(T (J^k_M F)\big)(\FR^n):=
{\rm lim}_{k}^{\mathrm{Set}}\, {\rm Hom}_{{\rm Man}}\big(\FR^n,T(J^k_M F) \big)\, ,
\end{align}
which is represented by
\vspace{-3mm} 
$$
\bigcup_{s^n\in y(J^\infty_M F)(\FR^n)}
\Big\{ \{X_{s^n}^k:\FR^n\rightarrow T(J^k_M F)\;  \big\vert  \;
\dd \pi^{k}_{k-1}\circ  X^{k}_{s^n}= X^{k-1}_{s^n}\}_{k\in \NN} \, \Big\}\ , 
$$

\vspace{-2mm} 
\noindent where $X^k_{s^n}(x)\in T_{\pi_k\circ s^n(x)}J^k_M F$ for each $x \in \FR^n$ being implicit in the notation. 
In local charts, one may further represent such families with infinite formal sums as done with $*$-plots
\vspace{-2mm} 
\begin{align}\label{InfinityJetTangentPlotsFormalSum}
X_{s^n} = X^\mu_{s^n} \, \frac{\partial}{\partial x^\mu}\Big\vert_{s^n} + 
\sum_{|I|=0}^\infty Y^a_{I,s^n} \frac{\partial}{\partial u_I^a} \Big\vert_{s^n} \, ,
\end{align}

\vspace{-2mm} 
\noindent where now $\big\{X^\mu_{s^n}, \{Y^a_{I,s^n}\}_{0\leq |I|}\big\}$ denote an infinite list of smooth functions, with $X^k_{s^n}$ 
corresponding to the case where the sum is terminated at order $|I|=k$.

On the fiber $\FR^n$-plots over each $\FR^n$-plot of $J^\infty_M F$, there exists a $C^\infty(\FR^n)$-linear structure 
induced by the linear structures on each $T(J^k_M F)$, natural in $\FR^n$. Thus, there is a map of smooth sets
\vspace{-2mm} 
\begin{align}\label{InfinityJetTangentBundleLinearStructure}
 +\;\;:\;\; T\big(y(J^\infty_M F)\big) \times_{y(J^\infty_M F)} T\big(y(J^\infty_M F)\big)&
 \longrightarrow T\big(y(J^\infty_M F)\big)   
 \\
 \big(\{X^k_{s^n}\}_{k\in \NN}\, ,\,  \{\tilde{X}^k_{s^n}\}_{k\in \NN} \big) &
 \longmapsto \{X^k_{s^n}+\tilde{X}^k_{s^n}\}_{k\in \NN} \, , \nn 
\end{align}

\vspace{-2mm} 
\noindent which respects the $y(\FR)$ fiber multiplication smooth map
\vspace{-2mm} 
\begin{align}\label{InfinityJetTangentBundleScalarMultiplcation}
\cdot \;\; : \;\; y(\FR)\times T\big(y(J^\infty_M F)\big) &\longrightarrow T\big(y(J^\infty_M F)\big) \\
\big(f_n,\{X^k_{s^n}\}_{k\in \NN}\big) &\longmapsto \{f_n\cdot X_{s^n}^k\}_{k\in \NN}
\nn \, ,
\end{align}

\vspace{-2mm} 
\noindent where $f_n:\FR^n\rightarrow \FR$ is a plot of $\FR$. In terms of their coordinate formal sum representations, 
the smooth addition map corresponds to the intuitive formal addition 
\vspace{-2mm} 
$$
\!\!\! \bigg(\!\! X^\mu_{s^n} \, \frac{\partial}{\partial x^\mu}\Big\vert_{s^n} + \sum_{|I|=0}^\infty 
\!\!\! Y^a_{I,s^n} \frac{\partial}{\partial u_I^a} \Big\vert_{s^n} \!\!\bigg) 
+ \bigg(\!\! \tilde{X}^\mu_{s^n} \, \frac{\partial}{\partial x^\mu}\Big\vert_{s^n} + 
\sum_{|I|=0}^\infty \!\!\! \tilde{Y}^a_{I,s^n} \frac{\partial}{\partial u_I^a} \Big\vert_{s^n}\!\! \bigg)
:= 
\big(X^\mu_{s^n}+\tilde{X}^\mu_{s^n}\big) \, 
\frac{\partial}{\partial x^\mu}\Big\vert_{s^n} + 
\sum_{|I|=0}^\infty \!\!\!\! \big(Y^a_{I,s^n}+\tilde{Y}^a_{I,s^n}\big) \frac{\partial}{\partial u_I^a} \Big\vert_{s^n}  .
$$

\vspace{-2mm} 
\noindent and similarly for the fiber scalar multiplication.

\begin{remark}[\bf Infinite jet tangent bundle as a Fr\'{e}chet manifold]
\label{TangentInfinityJetBundleFrechetManifold}
Another option would be to show the limit exists as an object of LocProMan, $T(J^\infty_M F):= \mathrm{lim}_k^{\FrechetManifolds} T(J^k_M F)$ 
which is suggested in \cite{Saunders89}. Indeed, this is the case and the proof is identical to that for 
$J^\infty_M F:=\mathrm{lim}_k^{\FrechetManifolds} J^k_M F$ of \cite{Takens79}\cite{Saunders89}. The local model is $\FR^\infty\times \FR^\infty\in \FrechetManifolds$, 
with local charts on $T(J^\infty_M F)$ taking the form $\big\{x^\mu, \{u^a_{I}\}_{0\leq |I|},\dot{x}^\mu , \{\dot{u}^a_{I}\}_{0\leq |I|} \big\}$. 
However, upon further embedding the resulting Fr\'{e}chet manifold into $\SmoothSets$, this will necessarily coincide with our definition 
above as per Rem. \ref{LocProMantoSmoothSetReflectsLimits}. 
The following discussion could be carried within $\FrechetManifolds$ and then embedded in smooth sets. We choose to do this directly in $\SmoothSets$, so as to stress 
that most of the constructions and definitions appearing depend mostly on universal properties, and rarely on analytical details.
\end{remark}

From the plot descriptions, there is an evident smooth projection map
\vspace{-2mm} 
$$
p: T\big(y(J^\infty_M F)\big)\longrightarrow y(J^\infty_M F)\, , 
$$

\vspace{-2mm} 
\noindent as expected intuitively. Naturally, we define vector fields on $J^\infty_M F $ as sections of the projection map, 
internally to SmoothSet.
\begin{definition}[\bf  Vector fields on the jet bundle]\label{VectorFieldsOnInftyJetBundle}
The set of smooth vector fields on the infinite jet bundle is defined as smooth sections of its tangent bundle
\vspace{-2mm} 
\begin{align}
\CX(J^\infty_M F):= \Gamma_{J^\infty_M F}\big(T(J^\infty_M F)\big) =
\Big\{X:y(J^\infty_M F)\rightarrow T\big(y(J^\infty_M F)\big) 
\; |\; p\circ X = \id_{y(J^\infty_M 
 F)}\Big\}\, .
\end{align}
\end{definition}

Let us unwind the definition in more concrete terms. By the limit property for $T\big(y(J^\infty_M F)\big)$ 
a map $X:y(J^\infty_M F)\rightarrow T\big(y(J^\infty_M F)\big)$ corresponds to a family
$\{X^k:y( J^\infty_M F) \rightarrow y\big(T(J^k_M F)\big)\}_{k\in\NN}$ such that the diagram
\vspace{-2mm} 
	\[ 
\xymatrix@R=1.3em@C=3em{ &&  y\big(T(J^k_M F)\big) \ar[d]^{\dd \pi^k_{k-1}}
	\\ 
	y(J^\infty_M F) \ar[rru]^{X^k} \ar[rr]^{\;\;\; X^{k-1}} && y\big(T(J^{k-1}_M F)\big)
}    
\]

\vspace{-2mm} 
\noindent commutes for all $k\in \NN$. Since $\mathrm{LocProMan}\xhookrightarrow{y}\SmoothSets$ is fully faithfully, 
this is equivalently a family of Fr\'{e}chet  manifold maps $\{X^k: J^\infty_M F \rightarrow T(J^k_M F)\}_{k\in\NN}$ 
such that the corresponding diagram commutes, i.e., $\dd \pi^k_{k-1} \, X^k = X^{k-1} : J^\infty_M F \rightarrow T (J^{k-1}_M F)$. 
Moreover, the section condition of X over $y(J^\infty_M F)$ interpreted in terms of Fr\'{e}chet manifolds as above,
corresponds to the condition $$X^k(s)\in T_{\pi^\infty_k (s)}J^{k}_M F$$ for each $s\in J^\infty_M F$ and every
$k\in \NN$. By Lem. \ref{FunctionsInfJetToManifold}, each $X^k:J^\infty_M F \rightarrow T(J^k_M F)$ is locally
of finite order, i.e., locally around each $s\in J^\infty_M F$, 
$$
X^k|_{\pi^{-1}_{j_s}(U_{\pi_{j_s}(s)})}= \pi^*_{j_s} (X^k_{j_s})\, ,
$$ 
for some $X^k_{j_s}:J^{j_s}_M F |_{U_{\pi_{j_s}(s)}} \rightarrow T(J^k_M F) $. In particular, such a family 
$\{X^k: J^\infty_M F\rightarrow T(J^k_M F)\}_{k\in \NN} $ 
may be represented in a local chart of $J^\infty_M F$ by an infinite (formal) sum
\vspace{-3mm} 
\begin{align}\label{InftyJetBundleVectorFieldCoordinates}
X= X^\mu \frac{\partial}{\partial x^\mu} + \sum_{|I|=0}^\infty Y^a_I \frac{\partial}{\partial u_I^a} \, , 
\end{align} 

\vspace{-2mm} 
\noindent for an infinite list of (locally defined) smooth functions $\big\{X^\mu,\{Y_I^a\}_{0\leq |I|}\big\}\subset C^\infty (J^\infty_M F), 
$\footnote{Even though these are defined in some chart $V\subset J^\infty_M F$, it may be that they only locally 
factor through finite order jet bundles around each $x\in V\subset J^\infty_M F$.}
with each $X^k$ corresponding to the case where the sum is terminated at order $|I|=k$.

\medskip 
Of course, it might be that every smooth map $X^k:J^\infty_M F \rightarrow T(J^k_M F)$ in a compatible family is \textit{globally} 
of finite order on $J^\infty_M F$, so that  $X^k= \pi^*_{j_k} X^k_{j_k}$ for some $X^k_{j_k}:J^{j_k}_M F\rightarrow T(J^k_M F)$. 
In such a case, the infinite sum representation above has coefficients of fixed orders, i.e., of order up to $j_k$ when the sum 
is terminated at $|I|=k$. We denote the vector subspace of globally finite order vector fields by
$$
\CX^{\mathrm{glb}}(J^\infty_M F)\subset \CX(J^\infty_M F)\, .
$$
These are the vector fields that usually appear in the field-theoretic examples, and so \cite{Anderson89} focuses 
on this vector subspace. However, this subset cannot be identified with the full set of smooth sections of some bundle. 
For this reason, we allow for locally finite order vector fields as in \cite{Takens79}\cite{Saunders89}. 

\medskip 
The fiber-wise $y(\FR)$-linear structure on $T\big(y(J^\infty_M F)\big)$ induces a $C^\infty(J^\infty_M F)$-module 
structure on the set of vector fields on $\CX(J^\infty_M F) := \Gamma_{J^\infty_M F}\big(T(J^\infty_M F)\big)$, 
simply by composition along \eqref{InfinityJetTangentBundleLinearStructure} and 
\eqref{InfinityJetTangentBundleScalarMultiplcation}. It is easy to see that in the formal coordinate representation above, 
this corresponds to the usual formal addition and multiplication, e.g. 
\vspace{-2mm} 
\begin{align}
\label{InfinityJetVectorFieldsModuleStructure}
f\cdot X = f\cdot X^\mu \frac{\partial}{\partial x^\mu} +
\sum_{|I|=0}^\infty f\cdot Y^a_I \frac{\partial}{\partial u_I^a} \, ,
\end{align}

\vspace{-2mm} 
\noindent 
for any $f\in C^\infty(J^\infty_M F)$. Note, however, the subset of globally finite order vector fields is not 
a module of $C^\infty(J^\infty_M F)$ but only of $C^\infty_{\mathrm{glb}}(J^\infty_M F)$.

\medskip 
The following is stated in \cite[\S 2.5]{Takens79} with ``proof left as trivial". We provide a proof for completeness. 

\begin{lemma}[{\bf Vector fields and derivations on the jet bundle}]
\label{InfinityJetVectorFieldsDerivations} 
Vector fields on $J^\infty_M F$ are in 1-1 correspondence with derivations 
$C^\infty(J^\infty_M F)\rightarrow C^\infty(J^\infty_M F)$,
$$
\CX(J^\infty_M F)\cong {\rm Der}\big(C^\infty(J^\infty_M F)\big)\, .
$$
\begin{proof}
Let $X$ be a vector field  on $J^\infty _M F$ represented by a family $\{X^k: J^\infty_M F\rightarrow T(J^k_M F)\}_{k\in \NN} $, 
as above. Recall any function $f\in C^\infty(J^\infty_M F)$ is, locally around $s\in J^\infty_M F$, the pullback 
$\pi^*_{k_s} f_{k_s} = f|_{\pi^{-1}_{k_s}(U_{\pi_{k_s}(s)})}$ for some $f_{k_s} \in C^\infty\big(U_{\pi_{k_s}(s)} \subset J^{k_s}_M F\big)$. 
Define $X(f)\in C^\infty(J^\infty_M F)$ by
\vspace{-3mm} 
\begin{align*}
X(f)(s):&= X^{k_s}(s)\big(f_{k_s}\big), 
\end{align*}

\vspace{-2mm} 
\noindent which is well defined, since for any other local representative $f_{k'_s}$, without loss of generality $k'_s\geq k_s$, we have
\vspace{-2mm} 
$$
X^{k_s}(s) (f_{k_s})= \dd\pi^{k'_s}_{k_s} X^{k
_s}(s) (f_{k_s})= X^{k'_s}(s) \big({\pi^{k'_s}_{k_s}}^* f_{k_s}\big)=X^{k'_s}(s) (f_{k'_s})\, .
$$

\vspace{-2mm} 
\noindent The derivation property on a product $f\cdot g\in C^\infty(J^\infty_M F)$ follows since it holds locally for any 
finite order representative. It remains to show that $X(f):J^\infty M\rightarrow \FR$ is indeed smooth. 
Since $X^k$ is also locally of finite order 
\vspace{-2mm} 
$$
X^{k_s}|_{\pi^{-1}_{j_s}(U_{\pi_{j_s}(s)})}= \pi^*_{j_s} X^{k_s}_{j_s}\, ,
$$ 

\vspace{-2mm} 
\noindent for some $X^{k_s}_{j_s}:J^{j_s}_M F |_{U_{\pi_{j_s}(s)}} \rightarrow T(J^{k_s}_M F) $, without loss of generality
$j_s\geq k_s$, we get that
\vspace{-2mm} 
$$
X^{k_s}_{j_s}(f_{k_s}) \in C^\infty \big(U_{\pi_{j_s}(s)}\subset J^{j_s}_M F\big) 
$$

\vspace{-2mm} 
\noindent given by $X^{k_s}_{j_s}\big(f_{k_s}\big)(q)= X^{k_s}_{j_s}(q) \big(f_{k_s}\big).$ It follows that
\vspace{-2mm} 
\begin{align*}
X\big(f\big)(s)&=X^{k_s}(s)\big(f_{k_s}\big)= \big(\pi^*_{j_s} X^{k_s}_{j_s} \big)(s) \big(f_{k_s}\big)\\  
&= X^{k_s}_{j_s}\big(\pi_{j_s}(s)\big) \big(f_{k_s}\big) = X^{k_s}_{j_s}\big(f_{k_s}\big) \big(\pi_{j_s}(s)\big) \\
&= \pi^*_{j_s}\big(X^{k_{s}}_{j_s} (f_{k_s})\big) (s)
\end{align*}

\vspace{-2mm} 
\noindent that is, $X(f)|_{\pi^{-1}_{j_s}(U_{\pi_{j_s}(s)})}=\pi^*_{j_s}\big(X^{k_s}_{j_s}(f_{k_s})\big)$ around each $s\in J^\infty_M F$. 
Being locally of finite order, $X(f):J^\infty_M F\rightarrow \FR$ is smooth.

It follows similarly that any derivation $\tilde{X}:C^\infty(J^\infty_M F)\rightarrow C^\infty(J^\infty_M F)$ acts on local finite order representatives as above, 
and hence defines a vector field on $J^\infty_M F$.
\end{proof}
\end{lemma}
The statement and proof restrict to the case of globally finite order vector fields and derivations of globally finite order functions,
$$
\CX^{\mathrm{glb}}(J^\infty_M F)\cong \mathrm{Der}\big(C^\infty_{\mathrm{glb}}(J^\infty_M F)\big)\, .
$$ 
Similarly, it further restricts to tangent vectors 
and germs of smooth functions at any point $s\in J^\infty _M F$,
$$
T_s(J^\infty_M F)\cong \mathrm{Der}\big(C^\infty_s(J^\infty_M F), \FR \big)\, ,
$$
which coincides with the definition of tangent vectors of \cite{Takens79}.

\medskip 
In a local chart for $J^\infty_M F$, using the formal sum representative 
$X= X^\mu \frac{\partial}{\partial x^\mu} + \sum_{|I|=0}^\infty Y^a_I \frac{\partial}{\partial u_I^a} \, , $ 
the action on a function $f\in C^\infty (J^\infty_M F)$ corresponds to 
\vspace{-3mm} 
$$
X(f)|_{\pi^{-1}_{k_s}(U_{\pi_{k_s}(s)})}:=X^\mu \frac{\partial f}{\partial x^\mu} + \sum_{|I|=0}^{k_s} Y^a_I \frac{\partial f}{\partial u^a_I}\, ,
$$

\vspace{-1mm} 
\noindent  where the formal sum necessarily terminates to the (local) finite order of $f=\pi_{k_s}^* f_{k_s}$ around each point $s\in J^\infty_M F$ 
`since the higher derivatives act trivially'. The derivation point of view allows to define a Lie algebra structure on $\CX(J^\infty_M F)$, as usual by 
\vspace{-1mm} 
\begin{align}\label{LieBracketInfinityJetBundle}
[X,\tilde{X}](f):= X(\tilde{X}(f))- \tilde{X}\big(X(f)\big) \, ,
\end{align}

\vspace{-1mm} 
\noindent  which locally may be represented via the usual coordinate formula, now involving an infinite formal sum. 
Naturally, it restricts to a Lie algebra structure on the subspace of globally finite order vector fields.

 \subsection{Horizontal splitting}
 \label{HorizontalSplittingSection}
In finite-dimensional manifolds, the tangent bundle $TE\rightarrow E$ of every smooth fiber bundle $E\rightarrow M$ splits \textit{noncanonically} into 
a vertical and horizontal subbundle $VE\oplus HE\rightarrow E$, with each splitting corresponding to a \textit{choice} of a connection on $E$. 
In the context of infinite-dimensional manifolds and smooth sets, such a splitting is not guaranteed to exist.\footnote{For \textit{any} 
chosen notion of a tangent bundle, when it exists.}
Nevertheless, in the case of the infinite jet bundle, the tangent bundle
$T(J^\infty_M F)\rightarrow J^\infty_M F$ does have a \textit{smooth} horizontal splitting, which is in fact \textit{canonical}.

 \medskip 
The total projection $\pi^{\infty}_M:y(J^\infty_M F)\xrightarrow{\pi^\infty_0} y(F) \xrightarrow{\pi_M} y(M)$ to the base $M$ induces 
a `pushforward projection' between the tangent bundle smooth sets over $M$
\vspace{-2mm} 
\begin{align}\dd \pi^\infty_M \;:\; T \big( y(J^\infty_M F)\big) &\longrightarrow y(TM) 
\end{align}

\vspace{-2mm} 
\noindent 
which acts on $\FR^n$-plots of \eqref{TangentInfinityJetPlots} as
\begin{align*}
\big\{X_{s^n}^k:\FR^n\rightarrow T(J^k_M F)\;  \big\vert  \;  \dd \pi^{k}_{k-1}\circ  X^{k}_{s^n}= X^{k-1}_{s^n}\big\}_{k\in \NN} 
\;\; \longmapsto \;\; 
\dd \pi^k_M \circ X^k_{s^n} \in y(TM)(\FR^n) \, , 
\end{align*}
for any $X^k_{s^n}$ in the representing set of the plot, and $\dd \pi^k_M:T(J^k_M F)\rightarrow TM$ the pushforward of 
$\pi^k_M:J^k_M F\rightarrow M$. This is well-defined since 
$\dd \pi^{k+1}_M \circ X^{k+1}_{s^n} = \dd\pi^k_{M} \circ \dd \pi^{k+1}_k \circ X^{k+1}_{s^n}= \dd \pi^k_M \circ X^k_{s^n}$. 
At the level of points and in local coordinates, the map acts on a tangent vector $X_s\in T\big(y(J^\infty_M)\big) (*)$ at $s\in J^\infty_M F$ by
\vspace{-2mm}
$$
X_s = X^\mu \, \frac{\partial}{\partial x^\mu}\Big\vert_s + \sum_{|I|=0}^\infty Y^a_I \frac{\partial}{\partial u_I^a} \Big\vert_s \longmapsto X^\mu \, \frac{\partial}{\partial x^\mu}\Big\vert_{\pi_M^\infty(s)}\, . 
$$

\vspace{-2mm}

\begin{definition}[\bf  Vertical subbundle]
\label{VerticalJetBundle}
The (smooth) vertical subbundle of $T\big(y(J^\infty_M F)\big)$ is defined as the equalizer of $\dd \pi^\infty_M $ and the canonical map 
$0_M:T\big(y(J^\infty_M F)\big) \rightarrow y(TM)$,
\vspace{-4mm}
\begin{align}
\xymatrix@=1.6em  {
VJ^\infty_M F:= \mathrm{eq}\Big( \! T\big(y(J^\infty_M F)\big) \ar@<-.6ex>[rr]_-{0_M} \ar@<.6ex>[rr]^-{\dd \pi^\infty_M } && y(TM) \! \Big).}
\end{align}
\end{definition}

Concretely, the vertical sub-bundle is the sub-object $VJ^\infty_M F \longhookrightarrow TJ^\infty_M F$
whose $\FR^n$-plots are represented by families 
\vspace{-1mm} 
\begin{align}\label{VerticalSubbundlePlots}
\Big\{X_{s^n}^k:\FR^n\rightarrow T(J^k_M F)\;  \big|  \;  \dd \pi^{k}_{k-1}\circ  X^{k}_{s^n}= X^{k-1}_{s^n}\, , \, 
\dd \pi^k_M \circ X^{k}_{s^n}=(\pi^\infty_M\circ s^n, 0)\Big\}_{k\in \NN}\, .
\end{align}

\vspace{-1mm} 
\noindent For instance, a vertical tangent vector at $s\in J^\infty_M F$, $X_s \in VJ^\infty_M(F)(*)$ is represented by 
a family of compatible vertical vectors $\{X^k_s\in V_{\pi^\infty_k(s)} J^k_M F\}$ in the usual finite-dimensional sense. 
In a local coordinate chart for $J^\infty_M F$, such a family is represented by an infinite formal sum 
\vspace{-3mm} 
\begin{align}\label{InftyJetBundleVerticalTangentCoordinate}
X_s = 0 + \sum_{|I|=0}^\infty Y_I^a \frac{\partial}{\partial u^a_I} \Big\vert_s\, .
 \end{align}
 
\begin{example}[\bf Prolongated tangent vectors to field space]
\label{InfinityJetVerticalVectorFromProlongationofPlot}
Let $\phi_t:\FR^1\times M \rightarrow F$ be a $\FR^1$-parametrized section of $F\rightarrow M$, i.e.,
an $\FR^1$-plot of $\mathbold{\Gamma}_M(F)$ as in Def. \ref{SectionsSmoothSet}. 
The jet prolongation \eqref{smoothjetprolongation} of the plot $j^\infty \phi_t : \FR^1\times M \rightarrow J^\infty_M F$ 
defines an $\FR^1$-parametrized section of $J^\infty_M F\rightarrow M$. For each $x\in M$ we get a smooth curve 
$j^\infty \phi_t(x): \FR^1 \rightarrow J^\infty_M F$, through $j^\infty \phi_0 (x) \in J^\infty _M F$, which by 
Lem. \ref{InfinityJetTangentVectorsAsCurves} defines a \textit{vertical} tangent vector denoted suggestively by 
\vspace{-1mm} 
$$
\partial_t j^\infty \phi_t(x) |_{t=0} \;\; \in  V_{j^\infty \phi_0 (x)} J^\infty_M F \, .
$$

\vspace{-1mm} 
\noindent 
The verticality follows immediately in terms of the compatible family representation
\vspace{-1mm} 
$$
\big\{ \partial_tj^k \phi_t(x) |_{t=0} \;\; \in V_{j^k \phi_0 (x)} J^k_M F\big\}_{k\in \NN} \, ,
$$

\vspace{-1mm} 
\noindent   where each tangent vector is vertical in the usual sense, since each curve $j^k \phi_t (x): \FR^1 \rightarrow J^k_M F$ is
contained in the fiber above $x\in M$. Yet equivalently, the action on a smooth function $f\in C^\infty (J^\infty_M F)$
via Lem. \ref{InfinityJetVectorFieldsDerivations} is given by 
\vspace{-1mm} 
$$
\partial_t j^\infty \phi_t(x) \big\vert_{t=0} (f) := 
\frac{\partial }{\partial t}\Big\vert_{t=0}\big( (j^{\infty}\phi_t(x))^* f\big) 
=\frac{\partial }{\partial t}\Big\vert_{t=0} \big(f\circ j^\infty\phi_t (x)\big)\, ,
$$
whereby in a local chart for $J^\infty_M F$ and using the chain rule \footnote{The chain rule applies as in 
the finite-dimensional manifold case, since around the point $j^\infty \phi_0(x)\in J^\infty_M F$, 
$f$ is necessarily of finite order. In particular, the sum in the 
formula below terminates.} 
\vspace{-2mm} 
$$
\partial_t j^\infty \phi_t(x) \big\vert_{t=0}  (f)=
\sum_{|I|=0}^\infty \frac{\partial}{\partial t}\Big\vert_{t=0}\big(\partial_I \phi^a_t(x)\big) 
\cdot \frac{\partial}{\partial u^a_I} f\big(j^\infty \phi_0(x)\big)\, .
$$

\vspace{-2mm} 
\noindent  Thus in local coordinates, 
\vspace{-2mm} 
$$
\partial_t j^\infty \phi_t(x) \big\vert_{t=0}  = \sum_{|I|=0}^\infty \frac{\partial}{\partial t}\Big\vert_{t=0}\big(\partial_I \phi^a_t(x)\big) 
\cdot \frac{\partial}{\partial u^a_I} \Big\vert_{j^\infty \phi_0(x)}= 
\sum_{|I|=0}^\infty \frac{\partial}{\partial x^I}\big(\partial_t \phi^a_t|_{t=0}\big)(x) 
\cdot \frac{\partial}{\partial u^a_I} \Big\vert_{j^\infty \phi_0(x)}\, .
$$  

\vspace{-1mm} 
\noindent  Varying over $x\in M$, we get a section $\partial_t j^\infty\phi_t |_{t=0}:M \rightarrow VJ^\infty F$ 
covering $j^\infty \phi_0$
\vspace{-1mm} 
	\[ 
\xymatrix@R=1.5em@C=4em{ &&  VJ^\infty F \ar[d]
	\\ 
	M\ar[rru]^{\partial_t j^\infty\phi_t |_{t=0}\phantom{aa}} \ar[rr]^>>>>>>>>{{j^\infty \phi}_0} && J^\infty F \, .
}   
\]

\vspace{-2mm} 
\noindent By the local coordinate formula, this section depends only on the induced tangent vector $\partial_{t}\phi_{t}|_{t=0}\in T \CF(*)$ on field space. 
Thus, since any tangent vector $\CZ_\phi \in  T \CF=\mathbold{\Gamma}_M(VF)(*)$ is represented by such line plots (Lem. \ref{LinePlotsRepresentTangentVectors}), 
the above defines a `jet prolongation' map which we suggestively\footnote{This is not only useful notation. Indeed, it is not hard to see that 
$VJ^\infty_M F\cong  J^\infty_M (VF)$ as bundles (of smooth sets) over $M$, making the notation consistent with the jet prolongation of a section.} denote by the same symbol
\vspace{-1mm}
\begin{align*}
j^\infty \;:\; T\CF = \mathbold{\Gamma}(VF)& \; \longrightarrow \; \mathbold{\Gamma}( V J^\infty_M F) 
\\[-1pt]
\CZ_\phi= \CZ^a_\phi \cdot \frac{\partial}{\partial u^a} &\; \longmapsto \; j^\infty \CZ_\phi = \sum_{|I|=0}^\infty \frac{\partial \CZ_\phi^a}{\partial x^I}
\cdot \frac{\partial}{\partial u^a_I}\, ,
\end{align*}

\vspace{-1mm}
\noindent with a similar form on $\FR^k$-plots of $\CF$. For future reference, we note that for each choice of $x\in M$,
composition with evaluation at $x$ may be viewed as a (smooth) `pushforward' map
\vspace{-2mm}
\begin{align}\label{pushforwardsoftangentfieldvectors}
T\CF &\longrightarrow VJ^\infty_M F \\
\CZ_\phi &\longmapsto (j^\infty  \CZ_\phi)(x) \nn 
\end{align}

\vspace{-1mm}
\noindent of tangent vectors on the field space $\CF$, to (vertical) tangent vectors on the infinite jet bundle $J^\infty_M F$.
\end{example}
Naturally, smooth sections of the vertical sub-bundle $VJ^\infty F \rightarrow J^\infty F$ as in
Def. \ref{VectorFieldsOnInftyJetBundle}, define the 
distribution of \textit{vertical} vector fields on $J^\infty_M F$
$$
\CX_V(J^\infty_M F):= \Gamma_{J^\infty_M F} (VJ^\infty_M F) \, .
$$
From the coordinate representation of general vector fields \eqref{InftyJetBundleVectorFieldCoordinates} and 
vertical tangent vectors \eqref{InftyJetBundleVerticalTangentCoordinate}, it follows that any vertical vector 
field $X\in \CX_V(J^\infty_M F)$ is locally represented by an infinite sum of the form
\vspace{-2mm} 
$$
X=0 + \sum_{|I|=0}^\infty Y^a_I \frac{\partial}{\partial u_I^a}\, ,
$$

\vspace{-2mm} 
\noindent for arbitrary (local) functions $\{Y^a_I\}\subset C^\infty(J^\infty_M F)$. 
It follows immediately that $\CX_V(J^\infty_M F)$ is closed under the Lie bracket of $\CX(J^\infty_M F)$.

\medskip 
As with finite-dimensional fibrations, the vertical sub-bundle fits into a natural short exact sequence of smooth sets 
\footnote{By which 
we mean a short exact sequence of vector spaces of fiber $\FR^n$-plots over each $\FR^n$-plot of $J^\infty_M F$.}
over $y(J^\infty_M F)$ 
\vspace{-2mm} 
\begin{align}\label{TangentInfinityJetExactSequence} 0_{J^\infty_M F}\longrightarrow V J^\infty_M F\longhookrightarrow T\big(y(J^\infty_M F)\big) 
\longrightarrow y(J^\infty_M F) \times_{y(M)} y(TM) \longrightarrow 0_{J^\infty_M F}\, ,
\end{align}
where the fibered product on the right is defined now in SmoothSet, with $\FR^n$-plots being pairs of plots that project to the same plot in $M$.
Equivalently, this may be computed as the limit of the finite-dimensional fibered products $J^k_M F \times_M TM$ in LocProMan, 
$J^\infty_M F \times_M TM := \mathrm{lim}_k^{\FrechetManifolds}(J^k_M F \times_M TM)$, \footnote{The coordinate charts on which takes the form 
$\{x^\mu, \dot{x}^\mu, u^a_{\mu} , u^{a}_{\mu_1 \mu_2}, \cdots \}$  with $\{\dot{x}^\mu\}$ denoting the fiber coordinates of $TM$.}
and then embedded via $y:\mathrm{LocProMan}\hookrightarrow \SmoothSets$, so that
$$
y(J^\infty_M F \times_M TM) \cong y(J^\infty_M F)\times_{y(M)}y(TM)\, .
$$
The third map is naturally given on $\FR^n$-plots by
\begin{align*}
\{X_{s^n}^k:\FR^n\rightarrow T(J^k_M F)\,  |  \,  \dd \pi^{k}_{k-1}\circ  X^{k}_{s^n}= X^{k-1}_{s^n}\}_{k\in \NN} 
\longmapsto (s^n, \dd \pi^k_M \circ X^k_{s^n})  \, , 
\end{align*} extending the usual point-set surjection $T(J^\infty_M F) \rightarrow J^\infty_M F \times_M TM$.
The crucial property of the infinite jet bundle is that the above sequence has a \textit{canonical splitting}
\vspace{-1mm} 
$$
H \;:\; y(J^\infty_M F)\times_{y(M)} y(TM) \longrightarrow T\big(y(J^\infty_M F)\big)\, ,
$$

\vspace{-2mm} 
\noindent  in contrast to the corresponding sequence of any finite order jet bundle $J^k_M F$.
That is, there is a canonical connection on $J^\infty_M F$, usually referred to as the \textit{Cartan connection}. 
We recall the description of the splitting at the point set level as described in \cite{Takens79}\cite{Saunders89}, 
and we show how this extends to a smooth map of bundles over $J^\infty_M F$ in SmoothSet.

\medskip 
Let $j^k_p\phi \in J^k_M F$ and choose a representative (local) section $\tilde{\phi}:U\subset M\rightarrow F$ such that
$j^k \tilde{\phi} (p)= j^k_p \phi$. The induced pushforward map
$$ 
\dd (j^k \tilde{\phi})_p :T_p M \longrightarrow T_{j^k_p \phi} (J^k_M F) \, ,
$$
 is given in local coordinates by
 \vspace{-3mm} 
\begin{align*}
X^\mu\frac{\partial}{\partial x^\mu} \Big\vert_p \; \longmapsto \; & \;\;
X^\mu \bigg(\frac{\partial}{\partial x^\mu} \Big\vert_{j^k_p \phi} + \sum_{|I|=0}^{k} \frac{\partial}{\partial x^\mu} 
\frac{\partial}{\partial x^I } \tilde{\phi}^a (p) \cdot    \frac{\partial}{\partial u^a_I} \Big\vert_{j^k_p \phi}\bigg)
\\[-2pt]
 &= X^\mu \bigg(\frac{\partial}{\partial x^\mu} \Big\vert_{j^k_p \phi} + \sum_{|I|=0}^{k} u^a_{I+\mu}(j^{k+1}_p \tilde{\phi}) 
 \cdot    \frac{\partial}{\partial u^a_I} \Big\vert_{j^k_p \phi}\bigg)\, .
\end{align*}

\vspace{-2mm} 
\noindent The map depends only on the $(k+1)$-jet at $p\in M$ of the chosen representative section $\tilde{\phi}$, 
and evidently \textit{smoothly} so. Hence, for each $k\in \NN$ the assignment defines a smooth map of bundles over $J^{k}_M F$
\vspace{-1mm} 
\begin{align}\label{TangentkJetbundleSplitting}
H^k \;:\; J^{k+1}_M F \times_M TM &\longrightarrow T(J^k_M F) \\
(j^{k+1}_p \phi, X_p) &\longmapsto \dd(j^k \phi)_p (X_p) \nn
\end{align}

\vspace{-1mm} 
\noindent where, by slight abuse of notation, $\phi:U\subset  M\rightarrow F $ on the right-hand side is any representative local section of
$j^{k+1}_p \phi$. It is easy to see that these fit into the commutative diagram
\vspace{-2mm} 
\begin{align}\label{CartanDistribCompatibility}
\begin{gathered}
\xymatrix@R=.8em@C=1.6em  {J^{k+1}_M F\times_M TM \ar[dd]_{\pi^{k+1}_k\times \id} \ar[rr]^-{H^k} &&  T(J^k_M F) \ar[dd]^{\dd \pi^{k}_{k-1} } 
	\\ \\
J^k_M F\times_M TM \ar[rr]^-{H^{k-1}}  && T(J^{k-1}_M F) 
\, . }
\end{gathered}
\end{align}

\vspace{-1mm} 
\noindent  
The intuition behind the smooth horizontal splitting of $T\big(y(J^\infty_M F)\big)$ is that the above smooth 
bundle maps `stabilize at the $k\rightarrow \infty$ \, limit'. Indeed, there is an injective \textit{point-set} map
\vspace{-1mm} 
\begin{align}\label{PointPlotInfinityJetSplitting}
J^\infty_M F \times_M TM &\longrightarrow T\big(y(J^\infty_M F)\big) (*) \\
(j^{\infty}_p \phi \, ,\,  X_p ) &\longmapsto \dd(j^\infty \phi)_p(X_p):= \big\{\dd(j^k \phi)_p (X_p) \big\}_{k\in \NN} \, \nn ,
\end{align}
which splits the $*$-plot sequence of \eqref{TangentInfinityJetExactSequence}, where $\phi:U\subset M\rightarrow F$ is a  representative 
local section of $j^\infty_p \phi$. In the coordinate representation \eqref{InfinityJetTangentFormalSum} of tangent vectors 
in $T_{j^\infty_p\phi} J^\infty_M F$, the map takes the form
\vspace{-2mm} 
\begin{align}\label{InftyJetBundleHorizontalTangentCoordinate}
\bigg(j^{\infty}_p \phi \, ,\,  X^\mu \frac{\partial}{\partial x^\mu}\Big\vert_p \bigg) \; \longmapsto \;
X^\mu \bigg(\frac{\partial}{\partial x^\mu} \Big\vert_{j^\infty_p \phi} + \sum_{|I|=0}^{\infty} u^a_{I+\mu}(j^{\infty}_p \phi) 
\cdot    \frac{\partial}{\partial u^a_I} \Big\vert_{j^\infty_p \phi}\bigg)\, ,
\end{align}

\vspace{-2mm} 
\noindent as usually written in the existing literature, e.g. \cite{Saunders89}. The tangent vectors on $J^\infty_M F$ in the image are
interpreted as horizontal lifts of tangent vectors on $M$.

For future reference, note that by construction, for any choice of field configuration $\phi \in \CF(*)$, this defines a \textit{smooth} pushforward map
\vspace{-2mm}
\begin{align}\label{PushforwardAlongProlongationOfField}
\dd (j^\infty \phi) \;:\; TM &\longrightarrow T\big(y(J^\infty_M F)\big) \\
X_p &\longmapsto \dd(j^\infty \phi)_p (X_p) := X^\mu \bigg(\frac{\partial}{\partial x^\mu} \Big\vert_{j^\infty_p \phi} + \sum_{|I|=0}^{\infty}
\frac{\partial \phi^a} {\partial x^{I+\mu}} (p)
\cdot    \frac{\partial}{\partial u^a_I} \Big\vert_{j^\infty_p \phi}\bigg)\, , \nn
\end{align}

\vspace{-2mm}
\noindent with a similar action on higher plots of tangent vectors on $M$. Smoothness here follows by the limit property of $T\big(y(J^\infty_M F)\big)$,
and since each of the finite jet order maps are by construction smooth.

In fact, more generally, the map of Eq. \eqref{PointPlotInfinityJetSplitting} is naturally identified as part of a smooth 
splitting of the corresponding 
smooth sets.

\begin{proposition}[\bf Smooth splitting]
\label{SmoothSplittingProp}
The family of smooth bundle maps $\{H^k:J^{k+1}_M F \times_M TM\longrightarrow  T(J^k_M F)\}_{k\in \NN}$ determines a map of smooth sets  
\vspace{-2mm} 
\begin{align}
H \;:\; y(J^\infty_M F)\times_{y(M)} y(TM) \longrightarrow T\big(y(J^\infty_M F)\big)\, ,
\end{align}
which splits the corresponding exact sequence \eqref{TangentInfinityJetExactSequence}.
\begin{proof}
This is essentially an application of (a slight variation of) Prop. \ref{SmoothMapsofLocProMan}. More explicitly, by the limit property 
of $T\big(y(J^\infty_M F)\big)=\mathrm{lim}_k^{\SmoothSets} y\big(T(J^k_M F)\big)$ and the fully faithful embedding 
$y:\mathrm{LocProMan}\hookrightarrow \FrechetManifolds\hookrightarrow \SmoothSets$,
\vspace{-1mm} 
\begin{align*}
\mathrm{Hom}_{\SmoothSets}\Big(y(J^\infty_M F)\times_{y(M)} y(TM) \, , \, T\big(y(J^\infty_M F)\big)  \Big) 
&\cong \mathrm{Hom}_{\SmoothSets}\Big(y(J^\infty_M F\times_M TM)\, , \, T\big(y(J^\infty_M F)\big)  \Big)
\\[-2pt] 
&\cong \mathrm{lim}_k^{\mathrm{Set}}\, 
\mathrm{Hom}_{\SmoothSets}\Big(y(J^\infty_M F\times_M TM)\, , \, y\big(T(J^k_M F)\big)  \Big)\\
&\cong \mathrm{lim}_k^{\mathrm{Set}}\,  \mathrm{Hom}_{\FrechetManifolds}\big(J^\infty_M F\times_M TM\, , \,T(J^k_M F)  \big)\, . 
\end{align*}
Hence a smooth map $f:y(J^\infty_M F\times_M TM) \rightarrow T\big(y(J^\infty_M F\big) $ corresponds to a family of smooth Fr\'{e}chet maps
$\{f^k: J^\infty_M F\times_M TM \rightarrow T(J^k_M F)\}_k\in \NN$ such that $\dd \pi^k_{k-1} \circ f^k = f^{k-1}$, and vice-versa. 

\noindent By Lem. \ref{FunctionsInfJetToManifold}, the (set theoretic) maps
\vspace{-1mm} 
$$
(\pi_{k+1}\times \id)^* H^k \;:\; J^\infty_M F\times_M TM \; \longrightarrow \; J^{k+1}_M F\times_M TM \longrightarrow T(J^k_M F)
$$ 
are smooth Fr\'{e}chet maps for each $k\in \NN$, \footnote{Lemma \ref{FunctionsInfJetToManifold} applies verbatim for any limit in LocProMan, for instance in the case of
$J^\infty_M F\times_M TM\in \mathrm{LocProMan}$.} being the pullback of finite order maps (globally in this case). Furthermore, by the 
commutativity of diagram \eqref{CartanDistribCompatibility} they satisfy
\vspace{-1mm} 
\begin{align*}
\dd\pi^k_{k-1}\circ \big(\pi_{k+1}\times \id)^*H^k&= \dd \pi^{k}_{k-1}\circ H^k \circ (\pi_{k+1}\circ \id) 
\\[-1pt] 
&= H^{k-1}\circ (\pi^{k+1}_k \times \id ) \circ (\pi_{k+1} \times \id) 
\\[-1pt]
&= H^{k-1}\circ (\pi_{k}\times \id)
\\[-1pt]
&= (\pi_k \times \id)^* H^{k-1} \, .
\end{align*}

\vspace{-1mm} 
\noindent Thus the family $\big\{(\pi_{k+1}\times \id)^*H^k:J^\infty_M F\times_M TM \rightarrow T(J^k_M F)\big\}_{k\in \NN}$ uniquely corresponds to a 
map $H:y(J^\infty_M F \times_M TM) \rightarrow T\big(y(J^\infty_M F)\big)\, .$ The underlying point set map is that of \eqref{PointPlotInfinityJetSplitting},
and since this splits the $*$-plot sequence, it follows that the $\FR^n$-plot sequences split too.
\end{proof}
\end{proposition}

In a local coordinate chart for $J^\infty_M F$, the explicit action on $\FR^n$-plots may be seen as
\begin{align*}
    y(J^\infty_M F \times_M TM)(\FR^n) & \; \longrightarrow \; T\big(y(J^\infty_M  F)\big) (\FR^n)
    \\
\bigg(s^n \, ,\,  X^\mu \frac{\partial}{\partial x^\mu}\Big\vert_{\pi_M\circ s^n} \bigg) & \; \longmapsto \;
X^\mu \bigg(\frac{\partial}{\partial x^\mu} \Big\vert_{ s^n} + \sum_{|I|=0}^{\infty} u^a_{I+\mu}\circ s^n 
\cdot    \frac{\partial}{\partial u^a_I} \Big\vert_{s^n}\bigg)\, ,
\end{align*}

\vspace{-2mm} 
\noindent where $\{X^\mu:\FR^n\rightarrow M\}$ denote the components of the corresponding plot in $TM$.

\begin{corollary}[\bf  Canonical horizontal splitting]
\label{CanonicalSmoothSplitting}
There is a canonical isomorphism of smooth sets
\vspace{-2mm} 
\begin{align}VJ^\infty_M F \times_{y(J^\infty_M F)}\big( y(J^\infty_M F) \times_{y(M)} y(TM) \big) \xlongrightarrow{\sim} T\big(y(J^\infty_M F)\big)
\end{align}

\vspace{-1mm} 
\noindent
over $y(J^\infty_M F)$. We denote the induced splitting by
\vspace{-1mm} 
$$
T\big(y(J^\infty_M F)\big) \cong VJ^\infty_M F \oplus H J^\infty_M F\, ,
$$

\vspace{-1mm} 
\noindent
where the plots of $H J^\infty_M F$ are given by the image of the map of Prop. \ref{SmoothSplittingProp}.
\end{corollary}
The smooth sub-bundle $HJ^\infty_M F$ is called the \textit{Cartan Distribution} on the infinite jet bundle.
The direct sum here means the fibered product as above, computed in SmoothSet, with the linear structure being fiberwise 
for each $\FR^n$-plot over the induced plot of the base $y(J^\infty_M F)$. This furthermore implies the splitting of 
\textit{vector spaces}\footnote{It may naturally be extended to an isomorphism of the corresponding smooth sets of sections.},
in fact of $C^\infty(J^\infty_M F)$-modules, of smooth vector fields  
\vspace{-2mm} 
\begin{align} \label{InftyJetBundleVectorFieldSplitting}
\CX(J^\infty_M F) \cong \CX_V(J^\infty_M F)\oplus \CX_H(J^\infty_M F )
\end{align}

\vspace{-1mm} 
\noindent  where the two components denote smooth sections of the corresponding bundles, as in Def. \ref{VectorFieldsOnInftyJetBundle}.
Hence, any vector field $X$ on $J^\infty_M F$ may be written as $X=X_V+ X_H$, with $X_V$ and $X_H$ denoting sections of the 
vertical and horizontal smooth sub-bundles, respectively. Taking into account the local coordinate representations for 
a general vector field \eqref{InftyJetBundleVectorFieldCoordinates}
\vspace{-2mm} 
$$
X= X^\mu \frac{\partial}{\partial x^\mu} + \sum_{|I|=0}^\infty Y^a_I \frac{\partial}{\partial u_I^a}\, ,
$$

\vspace{-2mm} 
\noindent and those of vertical \eqref{InftyJetBundleVerticalTangentCoordinate} and horizontal \eqref{InftyJetBundleHorizontalTangentCoordinate}
tangent vectors,  it follows that the vertical and horizontal component  vector fields may locally be represented by
\vspace{-4mm} 
\begin{align}\label{HorizontalVerticalVectorFieldCoordinates}
X_V&=\sum_{|I|=0}^{\infty}\big(Y_I^a - X^\mu \cdot u^a_{I+\mu}\big) \cdot \frac{\partial}{\partial u^a_I} 
\\[-3pt] 
\nn
X_H&= X^\mu \bigg( \frac{\partial}{\partial x^\mu} + \sum_{|I|=0}^{\infty} u^{a}_{I+\mu} \frac{\partial}{\partial u^a_I} \bigg) \, , 
\end{align}

\vspace{-1mm} 
\noindent 
reproducing the usual formulas \cite{Anderson89}\cite{Saunders89}. In particular, if $\hat{X}^\mu \frac{\partial}{\partial x^\mu} \in \Gamma_M(TM)$ 
is a vector field on the base, then its horizontal lift is given by $(\pi^{\infty}_M)^*\hat{X}^\mu \big( \frac{\partial}{\partial x^\mu} 
+ \sum_{|I|=0}^{\infty} u^{a}_{I+\mu} \frac{\partial}{\partial u^a_I} \big) \in \CX_H(J^\infty_M F)$. 

\begin{example}[\bf Horizontal lift of coordinate basis]
It is customary to denote the local basis for 
horizontal vector fields, i.e., the horizontal lift of the local basis of coordinate vector fields $\{\frac{\partial}{\partial x^\mu}\}$,  by
\vspace{-2mm} 
\begin{align}\label{HorizontalVectorFieldBasisNotation}
D_\mu :=\frac{\partial}{\partial x^\mu} + \sum_{|I|=0}^{\infty} u^{a}_{I+\mu} \frac{\partial}{\partial u^a_I}\, . 
\end{align}

\vspace{-2mm} 
\noindent  If $f:J^\infty_M F \rightarrow \FR$ is a smooth function then for any smooth section $\phi\in \Gamma_M(F)$, we have 
$f\circ j^\infty \phi \in C^\infty(M)$. The lifts $\{D_\mu\}$ encode the action of $\{\frac{\partial}{\partial x^\mu}\}$ on $f\circ j^\infty \phi$ 
via the chain rule, i.e.,
\vspace{-1mm} 
\begin{align}\label{HorizontalVectorFieldBasisAction}
     D_\mu (f) \circ j^\infty \phi & =  \bigg(\frac{\partial f}{\partial x^\mu} + \sum_{|I|=0}^{\infty}   
     u^a_{I+\mu} \frac{\partial f}{\partial u_{I}^a} \bigg)\circ j^\infty \phi=\frac{\partial f}{\partial x^\mu}\circ j^{\infty}\phi + \sum_{|I|=0}^{\infty}  
     \frac{\partial}{\partial x^\mu} \frac{\partial \phi^a}{\partial x^I}\cdot  \Big(\frac{\partial f}{\partial u_{I}^a} \circ j^{\infty}\phi\Big)  
     \nn \\[-2pt]
     &= \frac{\partial}{\partial x^\mu} \big(f\circ j^\infty \phi\big) \, . 
\end{align}
\end{example}

Finally, we note that the vertical vector fields $\CX_V(J^\infty_M F)$ are closed under the Lie bracket \eqref{LieBracketInfinityJetBundle}, by construction. 
A crucial property of the Cartan connection is that the horizontal vector fields are also closed under the Lie bracket,
\vspace{-1mm} 
\begin{align}\label{CartanDistributionInvolutive}
[X^1_H, X^2_H] \, \in \CX_H(J^\infty_M F) \, ,
\end{align}

\vspace{-1mm} 
\noindent  for all $X^1_H, X^2_H \in \CX_H(J^\infty_M F)$, as can be easily checked in local coordinates. Namely, the Cartan connection is \textit{flat}. 
Naturally, the splitting descends on the subspace of global finite order vector fields 
\vspace{-1mm} 
$$
\CX^{\mathrm{glb}}(J^\infty_M F)\cong \CX^{\mathrm{glb}}_V(J^\infty_M F)\oplus \CX_H^{\mathrm{glb}}(J^\infty_M F)\, ,
$$

\vspace{-1mm} 
\noindent with the same local representation formulas and involutive properties.

\subsection{Differential forms}\label{DifferentialFormsOnJetBundleSection}

Given that the tangent bundle of $J^\infty_M F$ has a natural fiber-wise linear structure, we will 
define differential forms as smooth $y(\FR)$-linear maps $T\big(y(J^\infty_M F)\big) \rightarrow y(\FR)$. 
This will allow us to recover the differential forms as defined in \cite{Takens79}\cite{Anderson89}\cite{Saunders89}. 
Furthermore, as we will explain in \cite{GSS-2}, this definition will be in line with the classifying nature
$\mathbold{\Omega}_{\mathrm{dR}}^1$ in the extended topos of thickened smooth sets, by abstract arguments; see Rem. \ref{ClassifyingFormsDoNotClassifyInSmoothSet}.
This identification is not completely clear in the current setting of smooth sets, and so we will use a different symbol for definite forms 
as fiber-wise linear maps out of the tangent bundle. Nevertheless, we will prove that forms of `globally finite order' 
on $J^\infty_M F$, in the traditional sense, may be naturally identified as a subalgebra of de Rham forms defined via the classifying space (Lem. \ref{JetBundleDiffFormsAsDeRhamForms}).

\begin{definition}[\bf  1-forms infinity jet bundle]
\label{1formsInftyJetBundle}
The set of differential 1-forms on the infinite jet bundle is defined as
\vspace{-2mm} 
\begin{align}\Omega^1(J^\infty_M F) := \mathrm{Hom}^{\mathrm{fib.lin.}}_{\SmoothSets}\big(T(J^\infty_M F)\,,\, y(\FR) \big)
\end{align} 

\vspace{-1mm} 
\noindent with respect to the fiber-wise linear structure of \eqref{InfinityJetTangentBundleLinearStructure} and \eqref{InfinityJetTangentBundleScalarMultiplcation}. 
\end{definition}
As with finite-dimensional manifolds, any fiberwise linear map $\om:T(J^\infty_M F)\rightarrow y(\FR)$ defines a map of $C^\infty(J^\infty_M F)$-modules  
\vspace{-2mm} 
$$
\om \;:\; \Gamma\big(T(J^\infty_M F)\big) \longrightarrow C^\infty(J^\infty_M F)
$$

\vspace{-2mm} 
\noindent  by pre-composing as
$$
\big(X:y(J^\infty_M F)\rightarrow T(J^\infty_M F) \big) \longmapsto \big(\om \circ X : y(J^\infty_M F) \rightarrow y(\FR) \big) \, ,
$$
where on the right-hand side we identify $\om\circ X \in C^\infty(J^\infty_M F)= \mathrm{Hom}_{\FrechetManifolds}(J^\infty_M F, \FR)$ since 
$y:\FrechetManifolds\hookrightarrow \SmoothSets$ is fully faithful. Tracing through the identifications, in a local coordinate chart around 
$s\in J^\infty_M F$, such a map $\om$ preserving the $C^\infty(J^\infty_M F)$-module structure \eqref{InfinityJetVectorFieldsModuleStructure}
must take the form
\vspace{-2mm} 
$$
X^\mu \frac{\partial}{\partial x^\mu} + \sum_{|I|=0}^\infty Y^a_I \frac{\partial}{\partial u_I^a} \;\;\longmapsto  \;\;
X^\mu\cdot \om\Big(\frac{\partial}{\partial x^\mu}\Big) + \sum_{|I|=0}^\infty Y^a_I 
\cdot \om \Big(\frac{\partial}{\partial u^a_I}\Big) \, \in C^\infty(J^\infty_M F)\,  .
$$

\vspace{-2mm} 
\noindent  Since the sum on the right-hand side is assumed to be a well-defined smooth function on $J^\infty_M F$ for all vector fields $X$, 
it must necessarily terminate at some finite $|I|=k_{s}\in \NN$ and so $\omega(\frac{\partial}{\partial u_I^a})=0$ for $|I|>k_s$. 
Hence, on the chart around $s\in J^\infty_M F$, we may represent $\om$ by a \textit{finite sum}
\vspace{-4mm} 
\begin{align}
\label{1formInfinityJetLocalCoordinates}
\om=\om_\mu \dd x^\mu + \sum_{|I|=0}^{k_s} \om_a^I \dd u^a_I\, ,  
\end{align}

\vspace{-2mm} 
\noindent where $\big\{\om_\mu, \{\om_a^I\}_{|I|\leq k_s}\big\}\subset C^\infty(J^\infty_M F)$ are some locally defined smooth functions. 
By further restricting to a smaller neighborhood $U_s \subset J^\infty_M F$ around $s$, we may assume these functions are also of 
finite order and, without loss of generality, of the same order $k_s$. Given this local description, we see that the map 
$\om: \CX(J^\infty_M F)\rightarrow C^\infty(J^\infty_M F)$ is necessarily of the form
\vspace{-1mm} 
\begin{align}\label{InfinityJet1formLocallyPullback}
\om(X)(s') = \om_{k_{s}}\big(\pi_{k_s}(s')\big) \left(\dd \pi_{k_s}\left( X(s')\right)\right) 
\end{align}

\vspace{-1mm} 
\noindent  locally for any $s'\in U_s\subset J^\infty_M$ for some local form $\om_{k_s}\in \Omega^1(J^{k_s}_M F)$ and $k_s\in \NN$. 
Hence, we may think of 1-forms on $J^\infty_M F$ as being `locally the pullback' of finite order forms.

\medskip 
More precisely, we see that every 1-form on $\om$ on $J^\infty_M F$ is represented by a family 
\vspace{-2mm} 
\begin{align}\label{1formInfinityJetViaLocalFiniteOrderFamilies}
\Big\{\om^{k_s}\in \Omega^1(U_{\pi_{k_s}(s)})\; \big\vert \; U_{\pi_{k_s}(s)}\subset J^{k_s}_M F \Big\}_{s\in J^\infty_M F}
\end{align}

\vspace{-2mm} 
\noindent of compatible locally defined 1-forms on finite order jet bundles. Compatible here means that for any two 
$s,s' \in J^\infty_M F$ with $k_{s'}\geq k_s$ such that $(\pi_{k_s}^{k_{s'}})^{-1}(U_{k_s}) \cap U_{k_s'} \neq \emptyset\subset J^{k_s}_M F$, 
we have $(\pi_{k_s}^{k_{s'}})^*\om_{k_s}=\om_{k_{s'}}$.\footnote{Two such families $\{\om^{k_s}\}_{s\in J^\infty_M F}, \{\tilde{\om}^{k_s}\}_{s\in J^\infty_M F}$ 
determine the same 1-form on $J^\infty_M F$ 
if and only if each pair $(\om^{k_s},\tilde{\om}^{k_s})$ of finite degree local forms agree on their common overlap in $J^{k_s}_M F$. 
Formally, such families represent 
elements in the limit of the algebras of \textit{germs of $1$-forms} on $J^k_M F$; see e.g. \cite{GMS00}.}
Formula \eqref{InfinityJet1formLocallyPullback} also determines the form of the corresponding map out of the tangent bundle 
(Def. \ref{1formsInftyJetBundle}) on $*$-plots 
\vspace{-1mm} 
\begin{align*}
\om \;:\; T(J^\infty_M F)(*)&\longrightarrow \FR 
\\[-2pt]
X_s &\longmapsto \om_{k_s}\big(\pi_{k_s}(s)\big)\big(\dd \pi_{k_s} (X_s)\big)
\end{align*}

\vspace{-1mm} 
\noindent  for some locally defined form $\om_{k_s}\in \Omega^1(J^{k_s}_M F)$ around $\pi_{k_s}(s)\in J^{k_s}_M F$,
and consequently on any $\FR^n$-plot. 

\medskip 
Note that any finite order globally defined form $\om^k \in \Omega^1(J^k_M F)$ fits the above description and defines a 
$1$-form on $J^\infty_M F$ by the same formulas. We denote the vector subspace of globally finite order 1-forms by
\vspace{-1mm} 
\begin{align}\label{GloballyFinOrder1formsInfinityJet}\Omega^1_{\mathrm{glb}}(J^\infty_M F)\subset \Omega^1(J^\infty_M F)\, ,
\end{align}

\vspace{-1mm} 
\noindent in analogy to the case of smooth functions (or 0-forms). The local finite order viewpoint agrees with \cite{Takens79}, 
while both \cite{Anderson89}\cite{Saunders89} define forms as being globally the pullback of finite order forms. 
By the discussion above, the former has a natural $C^\infty ( J^\infty_M F)$-module structure and corresponds to the 
full set of smooth fiber-wise linear maps out of the actual tangent bundle, while it also defines a (petit) sheaf on the topological 
space $|J^\infty_M F|$. The latter is only is a module over $C^\infty_{\mathrm{glb}}(J^\infty_M F)\subset C^\infty(J^\infty_M F)$ 
and corresponds to a subset of linear maps out of the tangent bundle, and it does \textit{not} define a sheaf on the underlying topological space.

\begin{remark}[\bf Infinity jet 1-form as Fr\'echet map]
\label{1formInfinityJetAsFrechetMap}
Alternatively, if we consider $T(J^\infty_M F)$ as a locally pro-manifold as per Rem. \ref{TangentInfinityJetBundleFrechetManifold}, then by the 
fully faithful embedding $y:\mathrm{LocProMan}\hookrightarrow \SmoothSets$ a 1-form is equivalently a fiber-wise linear Fr\'{e}chet map
$\om :T(J^\infty_M F)\rightarrow \FR$. By Lem. \ref{FunctionsOnInftyJetBundle} applied to the projective limit $T(J^\infty_M F)$, such a 
smooth map is locally around any\footnote{Strictly speaking, this holds around an open neighborhood of any 
$X_s \in T(J^\infty_M F)$. The topology of $T(J^\infty_M F)$ and the linearity assumption of $\om$ are crucial to descend to an open 
neighborhood $s\in J^\infty_M F $.}
$s\in J^\infty_M F$
\vspace{-1mm} 
$$
\om=(\dd \pi_{k_s})^* \om_{k_s}\, ,
$$

\vspace{-1mm} 
\noindent where the pullback is by $\dd \pi_{k_s}:T(J^\infty_M F)\rightarrow T(J^{k_s}_M F)$. When interpreted as acting on vector fields,
this reproduces the formula above. Furthermore, the locally pro-manifold $J^\infty_M F$ is paracompact and has partitions of unity \cite{Takens79}. 
By extending tangent vectors to vector fields, via a partition of unity, it follows that any $C^\infty(J^\infty_M F)$-linear map 
$\CX(J^\infty_M F)\rightarrow C^\infty(J^\infty_M F)$ defines a fiberwise linear smooth map $T(J^\infty_M F)\rightarrow \FR$. 
As with finite-dimensional manifolds, this is another way to witness the bijection 
\vspace{-2mm} 
$$
\Omega^1(J^\infty_M F)= \mathrm{Hom}_{\SmoothSets}^{\mathrm{fib.lin.}}\big(T(J^\infty_M F )\, , \, \FR) 
\cong \mathrm{Hom}_{C^\infty(J^\infty_M F)\mathrm{-Mod}}\Big(\Gamma\big(\CX(J^\infty_M F) \, , \, C^\infty(J^\infty_M F) \Big) \, .
$$ 
\end{remark}

Completely analogously, we may define $m$-forms on $J^\infty_M F$, with the analogous identifications following verbatim.

\begin{definition}[\bf  Forms on infinite jet bundle]
\label{mformsInftyJetBundle}
The set of differential m-forms on the infinite jet bundle is defined as
\vspace{-1mm} 
\begin{align}\Omega^m(J^\infty_M F) := \mathrm{Hom}^{\mathrm{fib.lin. an.}}_{\SmoothSets}
\big(T^{\times m}(J^\infty_M F)\,,\, y(\FR) \big) \, , 
\end{align} 

\vspace{-1mm} 
\noindent i.e., smooth real-valued, fiber-wise linear antisymmetric maps with respect to the fiber-wise linear structure induced by
\eqref{InfinityJetTangentBundleLinearStructure} and \eqref{InfinityJetTangentBundleScalarMultiplcation}, on the $m$-fold fiber product
\vspace{-2mm} 
$$
T^{\times m}(J^\infty_M F):= T(J^\infty_M F)\times_{y(J^\infty_M F)}\cdots \times_{y(J^\infty_M F)} T(J^\infty_M F)
$$ 

\vspace{-2mm} 
\noindent of the tangent bundle over the infinite jet bundle.
\end{definition}

Concretely, the $m$-fold fiber product of the tangent bundle over $J^\infty_M F$ is the smooth set with $\FR^n$-plots given by $m$-tuples of 
plots covering the same plot in $y(J^\infty_M F)$,
\vspace{-2mm} 
\begin{align}\label{mfoldTangentInfinityJetPlots}
 T^{\times m}(J^\infty_M F)(\FR^n)
 =\bigcup_{s^n\in y(J^\infty_M F)(\FR^n)}\big\{(X^1_{s^n},\cdots, X^m_{s^n})\in T(J^\infty_M F)(\FR^n)\times \cdots \times T(J^\infty_M F)(\FR^n)\big\} \, ,
\end{align}

\vspace{-2mm} 
\noindent
where $ p \circ X^1_{s^n}=\cdots =p\circ X^m_{s^n} = s^n \in y(J^\infty_M F)(\FR^n)$ being implicit, with each $X^i_{s^n}\in T(J^\infty_M F)(\FR^n)$ as in
\eqref{TangentInfinityJetPlots}. Equivalently, this may be seen as the limit computed directly in smooth sets
$$
 T^{\times m}(J^\infty_M F) \cong 
\mathrm{lim}_k^{\SmoothSets}\, y\big(T(J^k_M F) \times_{J^k_M F}\cdots \times_{J^k_M F} T(J^k_M F)\big) \, . 
$$
Repeating the discussion for the case of a single vector field (Def. \ref{VectorFieldsOnInftyJetBundle}), it follows that sections of the $m$-fold fibered product 
above correspond to $m$-tuples of vector fields. Thus, any $m$-form $\om \in \Omega^m(J^\infty_M F)$ defines an antisymmetric map of $C^\infty(J^\infty_M F)$-modules  
$$
\om \;:\; \CX(J^\infty_M F)\times\cdots\times \CX(J^\infty_M F)  \longrightarrow C^\infty(J^\infty_M F)
$$
by pre-composing as
\vspace{-1mm} 
$$
\big((X^1,\cdots, X^m):y(J^\infty_M F)\rightarrow  T^{\times m}(J^\infty_M F) \big) \longmapsto
\big(\om \circ (X^1,\cdots, X^m) : y(J^\infty_M F) \rightarrow y(\FR) \big) \, .
$$

\vspace{-1mm} 
\noindent
As with 1-forms, in local coordinates in a neighborhood $U_s$ around a point $s\in J^\infty_M F$ we may represent an $m$-form by
\vspace{-2mm} 
\begin{align}\label{mformInfinityJetLocalCoordinates}
\om = \sum_{p+q=m} \, \sum_{I_1,\cdots, I_p=0}^{k_s} \om_{\mu_1\cdots \mu_p a_1 \cdots a_q}^{I_1\dots I_q} 
\dd x^{\mu_1}\wedge \cdots \wedge \dd x^{\mu_p}\wedge \dd u^{a_1}_{I_1}\wedge\cdots \wedge \dd u^{a_q}_{I_q}\, ,
\end{align}

\vspace{-2mm} 
\noindent where the sum terminates at some finite order $k_s$, and the coefficients are locally defined functions on $J^{k_s}_M F$. 
That is, the map $\om: \CX(J^\infty_M F)\times \cdots \times \CX(J^\infty_M F)  \rightarrow C^\infty(J^\infty_M F)$ 
is necessarily of the form
\vspace{-2mm} 
\begin{align}\label{InfinityJetmformLocallyPullback}
\om(X^1,\cdots, X^m)(s') = \om_{k_{s}}\big(\pi_{k_s}(s')\big) \Big(\dd \pi_{k_s}\big( X^1(s')\big),\cdots, \dd \pi_{k_s}\big( X^m(s')\big)\Big)  ,
\end{align}

\vspace{-1mm} 
\noindent locally for any $s'\in U_s\subset J^\infty_M$ for some local form $\om_{k_s}\in \Omega^m(J^{k_s}_M F)$ and $k_s\in \NN$. 
Hence, we may think of m-forms on $J^\infty_M F$ as being `locally the pullback' of finite order forms, and represent them by 
compatible families of locally defined finite order $m$-forms 
as in \eqref{1formInfinityJetViaLocalFiniteOrderFamilies}. The previous formula determines the form of the corresponding
map out of the tangent bundle 
(Def. \ref{mformsInftyJetBundle}) on $*$-plots 
\vspace{-3mm} 
\begin{align*}
\om \;:\; T(J^\infty_M F)\times_{J^\infty_M F} \cdots \times_{J^\infty_M F} T(J^\infty_M F)(*)&\longrightarrow \FR \\
(X^1_s,\cdots, X^m_s) &\longmapsto \om_{k_s}\big(\pi_{k_s}(s)\big)\big(\dd \pi_{k_s} (X^1_s),\cdots, \dd \pi_{k_s} (X^m_s) \big)
\end{align*}

\vspace{-2mm} 
\noindent for some locally defined form $\om_{k_s}\in \Omega^1(J^{k_s}_M F)$ around $\pi_{k_s}(s)\in J^{k_s}_M F$, and consequently on any $\FR^n$-plot. 
Thus, as in Rem. \ref{1formInfinityJetAsFrechetMap}, there is in fact a bijection
\begin{align}
\label{FormsInfinityJetBundleAsVectorFieldMaps}
\Omega^m(J^\infty_M F)&= \mathrm{Hom}_{\SmoothSets}^{\mathrm{fib.lin. an.}}\big(T^{\times m}(J^\infty_M F )\, , \, \FR) \nn 
\\ &\cong \mathrm{Hom}_{C^\infty(J^\infty_M F)-\mathrm{Mod}}^{\mathrm{antis.}}\big(\CX(J^\infty_M F)\times \cdots\times \CX(J^\infty_M F) \, 
, \, C^\infty(J^\infty_M F) \big) \, . \  
\\
&\cong \bigwedge^{m}_{C^\infty(J^\infty_M F)} \Omega^1(J^\infty_M F) \, . \nn
\end{align}

\vspace{-2mm} 
\noindent We denote the vector subspace of globally finite order m-forms, i.e., those determined `as a pullback' of
a single globally defined m-form 
$\om^{k}\in \Omega^m(J^k_M F)$, by
\vspace{-2mm} 
\begin{align}\label{GloballyFinOrdermformsInfinityJet}\Omega^m_{\mathrm{glb}}(J^\infty_M F)\subset \Omega^m(J^\infty_M F)\, ,
\end{align}
with the same comments applying as in $1$-forms $\eqref{GloballyFinOrder1formsInfinityJet}$. 

It is straightforward to algebraically define the usual Cartan calculus on $\Omega^\bullet(J^\infty_M F):= \bigoplus_{m\in \NN}\Omega^m(J^\infty_M F)$, 
with all the maps naturally descending to the subspace $\Omega^\bullet_{\mathrm{glb}}(J^\infty_M F)$. For instance,
the de Rham differential is given by
\vspace{-2mm} 
\begin{align}\label{deRhamInfinityJet}
\dd \;:\; \Omega^m(J^\infty_M F)&\longrightarrow \Omega^{m+1}(J^\infty_M F)
\end{align}

\vspace{-2mm} 
\noindent where 
\vspace{-2mm} 
\begin{align*}
\dd \om (X^0,\cdots, X^{m}):= \sum_{i=0}^{m} (-1)^i\,  X^i\big(\om(X^0,\cdots,\hat{X}^i,\cdots, X^m\big)\big) + \sum_{i<j}(-1)^{i+j}\, 
\om \big([X_i,X_j],X_0,\cdots, \hat{X}^i, \cdots, \hat{X}^j,\cdots, X^m \big) \, .
\end{align*}

\vspace{-2mm} 
\noindent Equivalently, in terms of the local finite order representatives \eqref{InfinityJetmformLocallyPullback}
around some $s\in J^\infty_M F$
\vspace{-2mm} 
\begin{align*}
\dd\om(X^0,\cdots, X^m)(s') = \dd\om_{k_{s}}\big(\pi_{k_s}(s')\big) \Big(\dd \pi_{k_s}\big( X^0(s')\big),\cdots,
\dd \pi_{k_s}\big( X^m(s')\big)\Big) \, ,
\end{align*}

\vspace{-2mm} 
\noindent or in terms of the local coordinate representation \eqref{mformInfinityJetLocalCoordinates}, via the usual formula
\vspace{-2mm} 
$$
\dd \om = \bigg(\dd x^\mu \wedge  \frac{\partial}{\partial x^\mu}+
\sum_{|I|=0}^{\infty} \dd u^a_I \wedge  \frac{\partial}{\partial u^a_I}\bigg) \om  \, , 
$$

\vspace{-2mm} 
\noindent with the sum necessarily terminating at some finite order $k_s$ around each $s\in J^\infty_M F$, since the coefficients of $\om$ are locally of finite order. 
From either representation above, it follows that $\dd^2=0$ and so $\big(\Omega^\bullet(J^\infty_M F), \dd\big)$ defines a cochain complex. 
It turns out both the locally \cite{Takens79}\cite{GMS00} and globally \cite{Anderson89} finite order de Rham cohomologies agree with the
cohomology the fiber bundle $F\rightarrow M$.

\begin{proposition}[\bf Total de Rham cohomology jet bundle]
\label{TotaldeRhamCohomologyInfinityJet}
The projection map $\pi^\infty_0 :J^\infty_M F\rightarrow F$ induces a quasi-isomorphism 
\vspace{-1mm} 
$$
(\pi^\infty_0)^* \;:\; \Omega^{\bullet}(F) \longrightarrow 
\Omega^{\bullet}(J^\infty_M F)_{\mathrm{glb}}\hookrightarrow \Omega^{\bullet}(J^\infty_M F)\, ,
$$

\vspace{-2mm} 
\noindent
and so isomorphisms on cohomology
\vspace{-2mm} 
\begin{align}\label{}
H^\bullet_{\mathrm{dR}}(J^\infty_M F) \cong H^\bullet_{\mathrm{dR,glb}}(J^\infty_M F)\cong H^\bullet_{\mathrm{dR}}(F)\, .
\end{align}

\vspace{-2mm} 
\noindent In particular, any closed $m$-form $\om$ on $J^\infty_M F$ decomposes as $\om= (\pi_0^\infty)^*\tilde{\om} + \dd k$ for some 
$\tilde{\om}\in \Omega^m(F)$ and $k\in \Omega^{m-1}(J^\infty_M F)$.
\end{proposition}

The Lie derivative of a function $f\in C^\infty(J^\infty_M F)$ along any vector field $X$ is defined by
$\mathbb{L}_X(f):= X(f)\in C^\infty(J^\infty_M F)$. It extends to $m$-forms by $\mathbb{L}_X:=[\dd, \iota_X]$, 
\vspace{-3mm} 
\begin{align}
\label{mformInfinityJetLieDerivative}
\mathbb{L}_X \;:\; \Omega^m(J^\infty_M F)&\longrightarrow \Omega^{m}(J^\infty_M F) 
\\[-2pt]
\om &\longmapsto \dd(\iota_X \om) + \iota_X (\dd \om) \, ,\nn
\end{align}

\vspace{-2mm} 
\noindent where $\iota_X\om:= \om(X,-,\cdots, -) \in \Omega^{m-1}(J^\infty_M F)$ denotes the contraction map. 
As with the de Rham differential, this takes the usual coordinate form around any point $s\in J^\infty_M F$, 
whereby $\om$ and $X$ are of finite order. The usual Cartan calculus identities between the maps $\dd,\, \mathbb{L}_X$ 
and $\iota_X$ also follow, since they hold locally around every point 
$s\in J^\infty_M F$ for the corresponding finite order representatives.

 \paragraph{\bf Relation to de Rham forms on $J^\infty_M F$.} Since the infinite jet bundle is a smooth set $y(J^\infty_M F)\in \SmoothSets$, 
 the notion of de Rham $m$-forms (Def. \ref{nformsonSmoothSet}) as maps into the classifying space $\mathbold{\Omega}^m_\mathrm{d R}$ applies.
 In this case\footnote{Notice, such a simple argument cannot be applied to the case of arbitrary differential forms (Def. \ref{DifferentialFormsOnFieldSpace}) 
 and de Rham forms (Def. \ref{nformsonSmoothSet}) on the actual field space $\CF=\mathbold{\Gamma}_M(F)$. Nevertheless, the above result can be 
 used to view \textit{local} differential forms on $\CF\times M$ (Def. \ref{BicomplexOfLocalForms}) as de Rham forms (Lem. \ref{LocalDiffFormsAsDeRhamForms}).},
 the relation to the traditional notion of differential $m$-forms as maps out of the tangent bundle $T(J^\infty_M F)$ (Def. \ref{mformsInftyJetBundle}) 
 is immediate, at least for those forms that are of globally finite order (cf. Rem. \ref{ClassifyingFormsDoNotClassifyInSmoothSet}).
\begin{lemma}[Differential forms on $J^\infty_M F$ as de Rham forms]\label{JetBundleDiffFormsAsDeRhamForms}
The subalgebra $\Omega^\bullet_\mathrm{glb} (J^\infty_M F)\hookrightarrow \Omega^\bullet (J^\infty_M F)$ of globally finite order differential forms 
is canonically identified with a subalgebra of de Rham forms on the infinite jet bundle. That is, there is a canonical DGCA injection 
\vspace{-3mm} 
\begin{align*}
 \Omega^\bullet_\mathrm{glb} (J^\infty_M F)\longhookrightarrow \Omega^\bullet_\mathrm{d R} (J^\infty_M F)\, .
\end{align*}
\end{lemma}
\begin{proof}
This follows by the finite dimensional manifold identification of Eq. \eqref{ClassifyingFormsClassifyManifoldForms}. More explicitly, 
let $\om \in \Omega^m_\mathrm{glb} (J^\infty_M F)\hookrightarrow \Omega^m (J^\infty_M F)$ be a differential form of globally finite order, i.e.,
\vspace{-2mm}
$$
\om = \pi^*_k \om_k = \om_k \circ \dd \pi_k \;\; : \;\; T^{\times m}(J^\infty_M F) \longrightarrow T^{\times m}(J^k_M F) \xlongrightarrow{\om_k} \FR 
$$ 

\vspace{-2mm}
\noindent
for a unique $\om_k \in \Omega^m(J^k_M F)$, where $k$ is the minimal such order. In particular, $\om_k$ is a differential form on a 
finite-dimensional manifold, and so by the Yoneda Lemma \ref{YonedaLemmaForSmoothSets} it corresponds uniquely to a map
(see Eq. \eqref{ClassifyingFormsClassifyManifoldForms}) 
\vspace{-2mm}
$$
\tilde{\om}_k \;:\; y(J^k_M F) \longrightarrow \mathbold{\Omega}^m_{\mathrm{dR}} \, 
$$ 

\vspace{-2mm}
\noindent
into the classifying space, i.e., a de Rham differential form on $y(J^k_M F)$. Precomposing with the projection 
$y(\pi_k) : y(J^\infty_M F)\rightarrow y(J^k_M F)$ we get 
\vspace{0mm}
$$
\tilde{\om}:= \tilde{\om}_k\circ y(\pi_k) \;\;  :  \;\;  J^\infty_M F \longrightarrow \mathbold{\Omega}^m_{\mathrm{dR}} \, ,
$$

\vspace{-1mm}
\noindent
which is the (unique) de Rham $m$-form corresponding to the traditional pullback form $\om= \om_k \circ \dd \pi_k$. 

It follows similarly that under this identification, the classifying the de Rham differential $\dd_{\mathrm{d R}}$ corresponds to the
traditional differential $\dd$ of (globally finite order) forms on $J^\infty_M F$, and similarly for the corresponding wedge products. 
Thus the DGCA of globally finite order forms on $J^\infty_M F$ embeds into the de Rham forms (Def. \ref{ClassifyingFormsDGCAstructure}) 
on $J^\infty_M F$ defined via the classifying space $\mathbold{\Omega}_{\mathrm{d R}}^\bullet$.
\end{proof}
We expect, but do not prove here, that de Rham forms actually exhaust all differential forms on $J^\infty_M F$. In other words, the above 
inclusion should extend to a canonical bijection of all differential forms on $J^\infty_M F$ and those defined via the classifying space
\footnote{We note that this does not follow from a generic categorical argument, in that $J^\infty_M F$ is a \textit{limit} of finite
order jet bundles and not a \textit{colimit}.} 
\vspace{-3mm}
\begin{align*}
 \Omega^\bullet (J^\infty_M F)\longhookrightarrow \Omega^\bullet_\mathrm{d R} (J^\infty_M F)\, .
\end{align*}

\vspace{-2mm}
\noindent On the other hand, there is a canonical and all-important splitting of differential forms on $J^\infty_M F$, into horizontal and 
vertical components, which is naturally explained via the traditional tangent bundle picture, and is not at all obvious if one uses the 
classifying space picture. This will become clear in the following section.

\newpage 

\addtocontents{toc}{\protect\vspace{-10pt}}
\section{Euler-Lagrange dynamics via the infinite jet bundle}

There is a natural bicomplex structure $\Omega^{\bullet,\bullet} (J^\infty_M F)$ induced on the differential forms on the infinite jet bundle, 
the so-called \textit{variational bicomplex}. The name ``variational" is justified, among other reasons, as the vertical differential may be
used to rigorously encode the integration-by-parts algorithm and the explicit algebraic form of the Euler--Lagrange equations via the
``Euler--Lagrange source form'' of a Lagrangian. Using the latter, we may express the space of on-shell fields as a smooth set. 

\medskip 
Furthermore, the horizontal forms and horizontal differential may be used to prove a version of Stokes' Theorem on field space, and hence
inducing local (conserved) currents and charges on $\CF$ by `pulling back' (horizontally) closed forms on $J^\infty_M F$. 
The latter pullback viewpoint on currents on field space is furthermore useful in showing that local symmetries of a field theory 
preserve the on-shell space of fields.   

\subsection{The variational bicomplex and EL equations}
\label{VariationalBicomplexSection}

The splitting of the tangent bundle (Cor. \ref{CanonicalSmoothSplitting}) induces a splitting on the set of 1-forms 
$\Omega^1(J^\infty_M F)$.

\begin{definition}[\bf  Horizontal and vertical 1-forms]
\label{HorizontalVertical1formDefinition}
Let   $\om:T(J^\infty_M F)\rightarrow y(\FR)$ be a differential 1-form on the infinite jet bundle.

\noindent {\bf (i)} $\omega$ is called \textit{horizontal} if the following composite is the zero map
\vspace{-2mm} 
$$
\om|_{{}_{VJ^\infty_M F}} \;:\; VJ^\infty_M F\longhookrightarrow T(J^\infty_M F) \longrightarrow y(\FR) \;.
$$

\noindent {\bf (ii)} Similarly, it is called \textit{vertical} if the following composite is the zero map
\vspace{-2mm} 
$$
\om|_{{}_{HJ^\infty_M F}} \;:\; HJ^\infty_M F\longhookrightarrow T(J^\infty_M F) \longrightarrow y(\FR) \;. 
$$ 
\end{definition}
Note that these conditions are equivalent to $\om$ vanishing on all horizontal or vertical $*$-plots (and hence all such $\FR^n$-plots), respectively. 
That is, vanishing on all horizontal or vertical tangent vectors at each point in $J^\infty_M F$, respectively. Acting on vector fields, 
a $1$-form $\om$ is horizontal or vertical if and only if
\vspace{-2mm}
$$
\om(X_V)=0 \hspace{1cm} \mathrm{or} \hspace{1cm}  \om(X_H)=0
$$

\vspace{-2mm}
\noindent for all $X_V\in \CX_V(J^\infty_M F)$, 
 or all $X_H\in \CX_H(J^\infty_M F)$, respectively. Since $1$-forms are $C^\infty(J^\infty_M F)$-linear maps, every 1-form $\om$ uniquely
 decomposes as $\om=\om_H+\om_V$, and so
 \vspace{-1mm} 
\begin{align}\label{1formsInfinityJetSplitting}
\Omega^1(J^\infty_M F) \cong \Omega^1_H(J^\infty_M F)\oplus \Omega^1_V(J^\infty_M F)
\end{align}

\vspace{-1mm} 
\noindent with the subspaces denoting vertical and horizontal 1-forms, respectively. Under the local coordinates representations
\eqref{HorizontalVerticalVectorFieldCoordinates}, 
\eqref{1formInfinityJetLocalCoordinates} around a point $s\in J^\infty_M F$, horizontal and vertical 1-forms take the form
\vspace{-2mm} 
\begin{align}
\label{HorizontalVertical1formLocalCoordinates}
\om_H= (\om_H)_\mu \cdot \dd x^\mu \, ,\hspace{1cm}  \hspace{1cm}  \om_V = \sum_{|I|=0}^{k_s} (\om_V)_a^{I} \cdot 
(\dd u^a_I - u^{a}_{I+\mu} \dd x^\mu)=:\sum_{|I|=0}^{k_s} (\om_V)_a^{I} \cdot \theta^a_I \, ,
\end{align}

\vspace{-2mm} 
\noindent for some coefficient functions of local finite order, respectively. The locally spanning set of vertical 1-forms denoted by 
$$
\big\{ \theta^a_I:= \dd u^a_I - u^a_{I+\mu} \dd x^\mu \big\}_{|I|\in \NN}
$$ 
is also referred to as `\textit{contact basis forms}' and the ideal $\Omega_C^\bullet(J^\infty_M F) \hookrightarrow \Omega^\bullet(J^\infty_M F)$ 
they generate the `\textit{contact forms}', as in e.g. \cite{Anderson89}. Projecting onto each of the subspaces, the de Rham differential 
$\dd:C^\infty(J^\infty_M F)\rightarrow \Omega^1(J^\infty_M F)$ decomposes as $\dd= \dd_H + \dd_V$    
\vspace{-2mm} 
\begin{align}\label{TotalDifferentialDecomposesOn0forms}
d_H+d_V \;:\; C^\infty_M(F)\longrightarrow \Omega^1_H(J^\infty_M F)\oplus \Omega^1_V(J^\infty_M F)
\end{align}

\vspace{-2mm} 
\noindent where in local coordinates
\vspace{-2mm} 
$$
\dd_H(f)= \bigg(\frac{\partial f}{\partial x^\mu} + \sum_{I=0}^{k_s}     u^a_{I+\mu} \frac{\partial f}{\partial u_{I}^a} \bigg)
\cdot \dd x^\mu= D_\mu(f) \cdot \dd x^\mu  \, ,\hspace{1cm}  \hspace{1cm}  \dd_{V}(f) = \sum_{|I|=0}^{k_s} \frac{\partial f}{\partial u^a_I} \cdot  \theta^a_I \, .
$$
In particular, the action on local coordinates is given by
\begin{align}\label{HorizontalVerticalDifferentialonCoordinates}
\dd_H (x^\mu)= \dd x^\mu, \hspace{1cm} \dd_V (x^\mu)= 0, \hspace{1cm} \dd_H(u^a_I) = u^a_{I+\mu} \dd x^\mu,  \hspace{1cm} \dd_V (u^a_I) = \theta^a_I \, .
\end{align}

\begin{example}[\bf Action of prolongated tangent vector]
The action of the vertical tangent vector $\partial_t j^\infty \phi_t (x) |_{t=0}$ induced by an $\FR^1$-plot
$\phi_t\in \mathbold{\Gamma}_M (F)(\FR^1)$ of smooth sections, as in Ex. \ref{InfinityJetVerticalVectorFromProlongationofPlot}, 
may be equivalently given as
\vspace{-1mm} 
\begin{align}
\label{InfinityJetVerticalVectorFromProlongationofPlotActionVia1form}
{\partial_t}|_{t=0} \big(f\circ j^\infty \phi_t(x)\big) &= \dd f\big(\partial_t j^\infty \phi_t(x)|_{t=0} \big) = 
(\dd_H + \dd_V) (f)  \big(\partial_t j^\infty \phi_t(x)|_{t=0} \big) \\ 
&= \dd_V f \big(\partial_t j^\infty \phi_t(x)|_{t=0} \big)= \iota_{\partial_t j^\infty \phi_t(x)|_{t=0} } \dd_V f\, , \nn 
\end{align}

\vspace{-1mm} 
\noindent by the verticality of the tangent vector, or explicitly in coordinates.
\end{example}

Proceeding in a similar manner, the splitting of $1$-forms \eqref{1formsInfinityJetSplitting} induces a decomposition on $m$-forms
\vspace{-2mm} 
\begin{align}\label{mformsInfinityJetSplitting}
\Omega^{m}(J^\infty_M F)\cong \bigoplus_{p+q=m} \Omega^{p,q}(J^\infty_M F)
\end{align}

\vspace{-1mm} 
\noindent
 via 
$ \Omega^m(J^\infty_M F) \cong \bigwedge^m \Omega^1(J^\infty_M F)\cong \bigwedge^m \big(\Omega^1_H(J^\infty_M F)\oplus \Omega^1_V(J^\infty_M F)\big)$. 
In other words, $\om\in \Omega^{p,q}(J^\infty_M F)$ if and only if 
$$\om(X_V^1,\cdots, X_V^p,-,\cdots,-)=0\, , \hspace{1cm} \mathrm{and} \hspace{1cm} \om(X_H^1,\cdots, X_H^q,- \cdots , -)=0 $$
for any vertical $\{X_V^1,\cdots, X_V^p\}\subset \CX_V(J^\infty_M F)$ and any horizontal $\{X_H^1,\cdots, X_H^q\} \subset \CX_H(J^\infty_M F)$. 
Thus there is a bi-grading on the algebra of differential forms
\vspace{-2mm} 
\begin{align}\label{BigradingonInfinityJetForms}
\Omega^\bullet(J^\infty_M F)\cong\Omega^{\bullet,\bullet}(J^\infty_M F):= \bigoplus_{m\in \NN} \bigoplus_{p+q=m} \Omega^{p,q}(J^\infty_M F )\, ,
\end{align}
where in a local chart around $s\in J^\infty_M F$ a $(p,q)$-form $\om \in \Omega^{p,q}(J^\infty_M F)$ takes the form
\vspace{-1mm} 
\begin{align}\label{pqformInfinityJetLocalCoordinates}
\om &= \sum_{I_1,\cdots, I_p=0}^{k_s} \om_{\mu_1\cdots \mu_p a_1 \cdots a_q}^{I_1\dots I_q} \dd x^{\mu_1}\wedge \cdots \wedge \dd x^{\mu_p}\wedge 
\theta^{a_1}_{I_1}\wedge\cdots \wedge  \theta^{a_q}_{I_q}\, , \\
&=  \sum_{I_1,\cdots, I_p=0}^{k_s} \om_{\mu_1\cdots \mu_p a_1 \cdots a_q}^{I_1\dots I_q} \dd_H x^{\mu_1}\wedge \cdots
\wedge \dd_H x^{\mu_p}\wedge \dd_V u^{a_1}_{I_1} \wedge\cdots \wedge \dd_V u^{a_q}_{I_q} \, .\nn 
\end{align}

\vspace{-2mm} 
\noindent
By the explicit local coordinate description, it is clear that any Lagrangian density map $L:J^\infty_M F \rightarrow \wedge^d T^*M$ of 
Def. \ref{LocalLagrangianDensity} is equivalently
a horizontal $(d,0)$-form $L\in \Omega^{d,0}(J^\infty_M F).$ More generally, smooth bundle maps
valued in $\wedge^p T^*M$ (over M) are equivalently horizontal $(p,0)$-forms on $J^\infty_M F$,
\vspace{-2mm}
\begin{align}\label{BundleMapsAsHorizontalForms}
P:J^\infty_M F \longrightarrow \wedge^p T^*M \hspace{0.5cm} \iff  \hspace{0.5cm} P \, \in \,  \Omega^{p,0}(J^\infty_M F)\, .
\end{align}

\begin{definition}[\bf  Horizontal and vertical differentials]
\label{TwoDifferentialsInfinityJet} 
The horizontal and vertical differentials of degree $(p,q)$ 
\begin{align}
\dd_H^{p,q} \;:\; \Omega^{p,q}(J^\infty_M F)\longrightarrow  \Omega^{p+1,q}(J^\infty_M F), \hspace{2cm} 
\dd_V^{p,q}\; : \;
\Omega^{p,q}(J^\infty_M F)&\longrightarrow \Omega^{p,q+1}(J^\infty_M F) 
\end{align}
are defined by 
$$
\dd_H^{p,q}:= \pr_{p+1,q}\circ \dd|_{\Omega^{p,q}}, \hspace{2cm} \dd_V^{p,q}:= \pr_{p,q+1}\circ \dd|_{\Omega^{p,q}}\, ,
$$
where $\pr_{p,q}:\Omega^{p+q}(J^\infty_M F) \rightarrow \Omega^{p,q}(J^\infty_M F)$ denotes the subspace projection.
\end{definition}
We stress that the differential property $\dd_H^{2}=0, \, \dd_V^2=0$, is non-trivial, and follows from the fact that the total differential 
decomposes as 
\vspace{-3mm} 
\begin{align}\label{TotalDifferentialDecomposesOnpqforms}
\dd|_{\Omega^{p,q}}= \dd_H^{p,q}+ \dd_V^{p,q} \, ,
\end{align}

\vspace{-1mm} 
\noindent which in turn hinges upon the involutive properties of the vertical and horizontal sub-bundles \eqref{CartanDistributionInvolutive}. 
Indeed, while the decomposition is true by definition on 0-forms \eqref{TotalDifferentialDecomposesOn0forms}, on 1-forms one has 
\vspace{-1mm} 
$$
\dd|_{\Omega^{1,0}}= \dd_H^{1,0}+ \dd_V^{1,0} + \pr_{0,2}\circ \dd|_{\Omega^{1,0}} \, .
$$

\vspace{-1mm} 
\noindent Thus it must be the case that $\dd \om_V (X_H^1,X_H^2)=0$ for all $\om_V\in \Omega^{1}_V(J^\infty_M F)$ and any horizontal 
vector fields $X^1_H, X^2_H$. Indeed,
$$
\dd \om_V (X^1_H, X^2_H) = X^1_H\big(\om_V(X_H^2)\big) - X^2_H\big(\om_V(X_H^1)\big) - \om_V ([X^1_H,X^2_H]) =0  
$$

\vspace{-1mm} 
\noindent since the first two terms vanish by verticality of $\om_V$, and the third similarly since $[X^1_H, X^2_H]$ is also horizontal \eqref{CartanDistributionInvolutive}. 
The same argument shows that $\dd \om|_{\Omega^{0,1}}= \dd_H^{0,1}+ \dd_V^{0,1}$, since the vertical vector fields are also involutive. 
Since the bi-graded algebra $\Omega^{\bullet,\bullet}(J^\infty_M F)$ is generated by $C^{\infty}(J^\infty_M F), \, \Omega^1_{V}(J^\infty_M F)$ and 
$\Omega^1_{H}(J^\infty_M F)$ it follows that the decomposition \eqref{TotalDifferentialDecomposesOnpqforms} holds. 

\medskip 
As is customary, we omit reference to the $(p,q)$ degrees on the differentials, and write simply 
\vspace{-2mm} 
$$
\dd = \dd_H + \dd_V\, ,
$$

\vspace{-2mm} 
\noindent whereby the differential property $0=\dd^{2}=\dd_H^2 + \dd_V^2 + (\dd_H \circ \dd_V +\dd_V\circ \dd_H) $ implies
\vspace{-1mm} 
\begin{align}\label{BicomplexRelations}
\dd_{H}^2=0, \hspace{1.5cm} \dd_V^2=0, \hspace{1.5cm} \dd_H \circ \dd_V = - \dd_V \circ \dd_H \, , 
\end{align}

\vspace{-1mm} 
\noindent since each of the terms maps into distinct subspaces of the bi-graded algebra $\Omega^{\bullet,\bullet}(J^\infty_M F)$. 
The anti-commutation relations and \eqref{HorizontalVerticalDifferentialonCoordinates} determine the action on the local basis of horizontal and vertical 1-forms 
\vspace{-1mm} 
\begin{align*}
\dd_H ( \dd x^\mu)= 0, \hspace{1cm} \dd_V (\dd x^\mu)= 0, \hspace{1cm} \dd_H( \dd_V u^a_I) = 
- \dd_V u^{a}_{I+\mu} \wedge \dd x^\mu ,  \hspace{1cm} \dd_V ( \dd_V  u^a_I) = 0 \, ,
\end{align*}

\vspace{-1mm} 
\noindent and hence on any $(p,q)$-form \eqref{pqformInfinityJetLocalCoordinates}, by derivation property of the differentials and their action of 
functions \eqref{TotalDifferentialDecomposesOn0forms}. 

\begin{definition}[\bf  Bundle variational bi-complex]
The \textit{variational bi-complex} of a fiber bundle $F\rightarrow M$ is the bi-complex
\vspace{-1mm} 
\begin{align}
\big(\Omega^{\bullet,\bullet}(J^\infty_M F), \, \dd_H, \, \dd_V \big) \, ,
\end{align}

\vspace{-2mm} 
\noindent depicted graphically as
\vspace{-3mm} 
\[ 
\xymatrix@R=1.4em@C=4em{
\vdots & \vdots &\vdots  &  & \vdots  \\
\Omega^{0,2}(J^\infty_M F)  \ar[u]^-{\dd_V} \ar[r]^-{\dd_H} & \Omega^{1,2}(J^\infty_M F)  \ar[u]^-{\dd_V} \ar[r]^-{\dd_H} & 
\Omega^{2,2}(J^\infty_M F)   \ar[u]^-{\dd_V} \ar[r]^-{\dd_H}& \cdots\ar[r]^-{\dd_H} & \Omega^{d, 2} (J^\infty_M F)   \ar[u]^-{\dd_V}\\
\Omega^{0,1}(J^\infty_M F)\ar[r]^-{\dd_H} \ar[u]^-{\dd_V}& \Omega^{1,1}(J^\infty_M F) 
\ar[u]^-{\dd_V} \ar[r]^-{\dd_H} & \Omega^{2,1}(J^\infty_M F)  \ar[u]^-{\dd_V} \ar[r]^-{\dd_H}  & 
\cdots \ar[r]^-{\dd_H} & \Omega^{d,1} (J^\infty_M F) \ar[u]^-{\dd_V}  \\ 
\Omega^{0}(J^\infty_M F) \ar[r]^{\dd_H} \ar[u]^-{\dd_V} &\Omega^{1,0}(J^\infty_M F)  \ar[u]^-{\dd_V} \ar[r]^-{\dd_H} &
\Omega^{2,0}(J^\infty_M F)  \ar[u]^-{\dd_V} \ar[r]^-{\dd_H}& \cdots \ar[r]^-{\dd_H}& \Omega^{d,0}(J^\infty_M F) \ar[u]^-{\dd_V} \, . 
}   
\]
\end{definition}

The vertical and horizontal cohomology of the variational bicomplex above is described in detail in \cite{Takens79}, and in \cite{Anderson89}
for the subcomplex of globally finite order forms, with both cohomologies agreeing on abstract grounds \cite{GMS00}. 
In particular, they prove the following important and useful result.

\begin{proposition}[\bf Takens' acyclicity theorem]\label{AcyclicityTheorem}
The horizontal rows $\big(\Omega^{\bullet,q\geq1}(J^\infty_M F)\, , \dd_H\big)$ are exact, apart from degrees $(d, q)$. Similarly the vertical columns $\big(\Omega^{p,\bullet}\, , \dd_V\big)$, apart from degrees $(p,q\leq f)$, where $f$ is the dimension of the fiber, $ \mathrm{dim}(F) = f + d $.
\end{proposition}

This means that any potentially non-vanishing cohomology is concentrated in a finite portion of the diagram. For the purposes of classical field theory, most of the useful information is in the bottom row and right-most column of the diagram. Indeed, it is convenient to extend the bi-complex to the right by the column 
of \textit{functional forms}, which in turn will allow for combining the bottom row and right-most column in a single cochain complex. 

\begin{remark}[\bf Globally vs. locally finite order]
The definitions and results that follow apply for both cases of locally finite order $\Omega^{\bullet}(J^\infty_M F)$ and globally finite order 
$\Omega^\bullet_{\mathrm{glb}}(J^\infty_M F)$ on the infinite jet bundle, as in \cite{Takens79}\cite{Anderson89}\cite{GMS00}. 
Therefore, we will suppress mentioning the distinction between the cases.
\end{remark}

\begin{definition}[\bf  Source forms]
\label{SourceForms}
For $q\geq 1$, the \textit{source forms} of degree $(d,q)$ are defined as the subspace 
\begin{align}
\Omega^{d,q}_{s}(J^\infty F ):=\Omega^{d,q-1} (J^\infty F)\wedge \Omega^1_{V,0}(J^\infty F) \, , 
\end{align}

\vspace{-2mm} 
\noindent
where
\vspace{-1mm}
$$
\Omega^1_{V,0}(J^\infty F):= \pr_{0,1}\big((\pi^\infty_0)^* \Omega^1(F)\big) \subset \Omega^1_V(J^\infty F) \, 
$$
\vspace{-2mm}
\noindent
is the space of vertical 1-forms that arise by pulling back 1-forms on $F$.

\end{definition}

\smallskip 
In particular, a source form $P \in\Omega^{d,q}_{s}(J^\infty F ) $ always has a `leg along $\dd_V u^a$', i.e., 
$
P = P_a\wedge \dd_V u^a
$,
for some $\{P_a\}\subset \Omega^{d,q-1} (J^\infty F)$.

\begin{lemma}[{\bf Differential operator of source form}]
\label{SourceFormsDefineDifferentialOperators}
$\,$

\noindent {\bf (i)} Any $(d,1)$-source form $ P \in \Omega_s^{d,1}(J^\infty F)$ naturally defines a smooth bundle map 
$P:J^\infty_M F \rightarrow \wedge^d T^*M \otimes V^*F $ over $F$, and so in turn over $M$, 
\footnote{We note that $V^*F\rightarrow F \rightarrow M$ is a \textit{vector} bundle over $F$, but in general only a \textit{fiber} over $M$. 
Hence, the tensor product of bundles $ \wedge^d T^*M \otimes V^*F $ is a shorthand for the tensor product of vector bundles over $F$, i.e., 
$\pi_F^*(\wedge^d T^*M)\otimes_F V^*F \rightarrow F \rightarrow M$.}  and hence a differential operator 
\vspace{-1mm} 
$$
\CP \,:\, \Gamma_M(F) \longrightarrow \Gamma_M\big(\wedge^d T^*M\otimes V^*F)\, ,
$$

\vspace{0mm} 
\noindent
where $V^*F\rightarrow F\rightarrow M$ is the dual vector bundle to the vertical tangent bundle $VF\hookrightarrow F\rightarrow M$.

\noindent {\bf (ii)} Furthermore, by Lem. \ref{DiffOpsAsSmoothMaps}, it is (uniquely) extended to a map of smooth sets
\vspace{-1mm} 
$$
\CP \,:\, \mathbold{\Gamma}_M (F) \longrightarrow \mathbold{\Gamma}_M(\wedge^d T^*M\otimes V^*F)\, .
$$

\begin{proof}
Let $P$ be 
a $(d,1)$-source form, locally of the form $P=P_a \wedge \dd_V u^a$. Since each $P_a\in \Omega^{d,0}(J^\infty_M F)$ is horizontal top-form, we have
$$
P=P_a \wedge \dd_V u^a= P_a \wedge (\dd u^a - u^a_\mu \cdot \dd x^\mu) = P_a \wedge \dd u^a + 0\, .
$$
The local coordinate $1$-forms $\{\dd u^a= \dd_{J^\infty_M F} u^a\}$ are the pullback of the differential of fiber coordinates on $F$, 
\vspace{-2mm} 
$$
\dd u^a = \dd (\pi^*_0 u^a)= \pi^*_0 (\dd_F u^a)\, ,
$$

\vspace{-1mm} 
\noindent which locally span the fibers of $V^*F$. Thus $P$ is equivalently a bundle map over $F$ (and so over M),
$$
P:J^\infty(F) \longrightarrow \wedge^d T^*M\otimes V^*F \, . 
$$

\vspace{-2mm} 
\noindent By  Lem. \ref{DifferentialofSmoothMap}, the corresponding differential operator is defined by precomposing 
with the jet prolongation of any given section
\vspace{-3mm} 
\begin{align*}
\CP: \Gamma_{M}(F) &\longrightarrow \Gamma_M (\wedge^d T^*M \otimes V^* F) \\
\phi &\longrightarrow P\circ j^\infty \phi \, .
\end{align*}

\vspace{-2mm} 
\noindent and similarly for its smooth extension on $\FR^k$-plots. In local coordinates,
$\phi^b(x^\mu) \mapsto P_a(x^\mu, \phi^b, \partial_\mu \phi^b,\cdots) \cdot \dd_F u^a$. 
\end{proof}
\end{lemma}

Note that the result \textit{does not extend} to higher order source forms as stated. There is a special subspace of source forms, 
\textit{the functional forms} \cite{Anderson89}, defined as the image of a natural projection map, which parametrizes a class of 
currents on $\CF$ via contraction with (vertical) vector fields on $J^\infty_M F$ (see Rem. \ref{SmoothMapsInducedbyFunctionalForms} below
and \cite{Anderson89} for more details).

\begin{definition}[\bf  Interior Euler operator]
\label{InteriorEulerOperator} 
$\,$

\noindent {\bf (i)} For any $q\geq 1$, the \textit{interior Euler operator} 
\vspace{-2mm} 
\begin{align}\mathcal{I} \;:\; \Omega^{d,q}(J^\infty F) \longrightarrow \Omega^{d,q}(J^\infty F)
\end{align}

\vspace{-2mm} 
\noindent is defined by the formula 
\vspace{-2mm} 
$$
\om \; \longmapsto \; \frac{1}{q}\dd_v u^a \wedge \bigg( \sum_{|I|=0}^\infty (-1)^{|I|} \mathbb{L}_{D_{I}}
\Big(\iota_{\frac{\partial}{\partial u^a_I}} \om \Big) \!\!\bigg) 
$$

\vspace{-2mm} 
\noindent where $\mathbb{L}_{D_I}=(\mathbb{L}_{D_1})^{I_1}\circ (\mathbb{L}_{D_2})^{I_2}\circ \cdots \circ (\mathbb{L}_{D_d})^{I_d}$ denotes the 
composition of the Lie derivatives \eqref{mformInfinityJetLieDerivative} with respect to the lifts of the 
coordinate vector fields \eqref{HorizontalVectorFieldBasisNotation}, for every multi-index $I=(I_1,\cdots, I_d)$. 

\noindent {\bf (ii)} The set of \textit{functional forms} is defined as
\vspace{-2mm} 
\begin{align*}
\Omega^{d,q}_{\mathrm{f}}(J^\infty F) := \mathrm{Im}(\CI) \subset \Omega^{d,q}_s(J^\infty) \, .
\end{align*}
\end{definition}

The interior Euler operator appears in \cite{Anderson89} in a more coordinate-dependent fashion, 
while the above form -- as far as we know -- appears only in 
\cite{Blohmann23b}. 

\begin{proposition}[\bf Properties of the interior operator]\label{InteriorEulerProperties}\cite[Thm 2.12]{Anderson89}

\noindent {\bf (i)} $\CI$ is a projection operator 
\vspace{-1mm} 
$$
\CI \circ \CI = \CI\, .
$$

 \noindent {\bf (ii)}  The composition $\CI\circ \dd_H: \Omega^{d-1,q}(J^\infty F)\longrightarrow \Omega^{d,q}(J^\infty F)$ 
vanishes and furthermore
\vspace{-2mm} 
$$
\mathrm{ker}(\CI)=\mathrm{Im}(\dd_H)\, .
$$

\vspace{-2mm} 
\noindent {\bf (iii)}  The map $\CI \circ \dd_V: \Omega^{d,q}\longrightarrow \Omega^{d,q+1}$ is nilpotent
\vspace{-2mm} 
$$
(\CI\circ \dd_V) \circ (\CI \circ \dd_V) =0 \, .
$$

\vspace{-2mm} 
\noindent {\bf (iv)} 
 The $(d,1)$-functional forms coincide with $(d,1)$-source forms
\vspace{-2mm} 
$$
\Omega^{d,1}_{\mathrm{f}}(J^\infty F):=\mathrm{Im}(\CI)\cong \Omega^{d,1}_S(J^\infty F)\, .
$$

\end{proposition} 
The first two items are the non-trivial part of the proposition, the third follows from the first two, while the fourth 
follows directly from the defining formula in Def. \ref{InteriorEulerOperator}. Note that the first two items imply that 
\vspace{-1.5mm} 
$$
\Omega^{d,q}(J^\infty F) \cong \Omega^{d,q}_f(J^\infty F)\oplus \dd_H \Omega^{d-1,q}(J^\infty F)\, ,
$$

\vspace{-1.5mm} 
\noindent and so in particular functional forms parametrize $(d,q)$-forms modulo horizontally exact forms
\vspace{-1.5mm} 
\begin{align*}
  \Omega^{d,q}_f(J^\infty F) \cong \Omega^{p,q}(J^\infty F)/\dd_H\Omega^{d-1,q}(J^\infty F)\, .
\end{align*}

\vspace{-1mm} 
The map defined on functional forms
\vspace{-2mm} 
\begin{align}
\delta_V:= \CI\circ \dd_V \; :\; \Omega^{d,q}_f(J^\infty F)\longrightarrow \Omega^{d,q+1}_f(J^\infty F) 
\end{align}

\vspace{-2mm} 
\noindent is called the (higher) Euler operator. Since it squares to zero, it follows that 
the variational complex may be \textit{augmented} to the right
\vspace{-4mm} 
$$
\xymatrix@R=2em@C=3em{
\vdots & \vdots &\vdots  &  & \vdots  & \vdots \\
\Omega^{0,2}(J^\infty_M F)  \ar[u]^-{\dd_V} \ar[r]^-{\dd_H} & \Omega^{1,2}(J^\infty_M F)  \ar[u]^-{\dd_V} \ar[r]^-{\dd_H} & 
\Omega^{2,2}(J^\infty_M F)   \ar[u]^-{\dd_V} \ar[r]^-{\dd_H}& \cdots\ar[r]^-{\dd_H} & \Omega^{d, 2} (J^\infty_M F) \ar[r]^-{\CI}  \ar[u]^-{\dd_V} 
& \Omega^{d,2}_f(J^\infty_M F) \ar[u]_-{\delta_V}\\
\Omega^{0,1}(J^\infty_M F)\ar[r]^-{\dd_H} \ar[u]^-{\dd_V}& \Omega^{1,1}(J^\infty_M F) \ar[u]^-{\dd_V} \ar[r]^-{\dd_H} 
& \Omega^{2,1}(J^\infty_M F)  \ar[u]^-{\dd_V} \ar[r]^-{\dd_H}  & \cdots \ar[r]^-{\dd_H} & \Omega^{d,1} (J^\infty_M F) \ar[u]^-{\dd_V}  \ar[r]^-{\CI} 
& \Omega^{d,1}_f(J^\infty_M F)\cong \Omega^{d,1}_s(J^\infty_M F) \ar[u]_-{\delta_V} \\ 
\Omega^{0}(J^\infty_M F) \ar[r]^{\dd_H} \ar[u]^-{\dd_V} &\Omega^{1,0}(J^\infty_M F)  \ar[u]^-{\dd_V} \ar[r]^-{\dd_H} 
& \Omega^{2,0}(J^\infty_M F)  \ar[u]^-{\dd_V} \ar[r]^-{\dd_H}& \cdots \ar[r]^-{\dd_H}& \Omega^{d,0}(J^\infty_M F) \ar[u]^-{\dd_V} \ar[ur]_-{\delta_V}& 
}   
$$
where, by the second property of the above proposition, all the rows except the bottom are still horizontally exact. 
The bottom row and right-most column connect into a single complex. 

\begin{definition}[\bf  Euler-Lagrange complex]
\label{EulerLagrangeComplex}
The \textit{Euler--Lagrange} complex $\Omega^\bullet_{\mathrm{EL}}(J^\infty_M F)$ of $J^\infty_M F$ is defined as
\begin{align}
\Omega^0(J^\infty_M F) \xrightarrow{\; \dd_H \; } \Omega^{1,0}(J^\infty_M F) \xrightarrow{\; \dd_H \;} \cdots 
\xrightarrow{\; \dd_H \;} \Omega^{d,0}(J^\infty_M F)\xrightarrow{\; \delta_V \;} \Omega^{d,1}_s(J^\infty_M F)
\xrightarrow{\; \delta_V \;}\Omega^{d,2}_f(J^\infty_M F)\xrightarrow{\; \delta_V \;} \cdots \, ,
\end{align}
where the notation implies $\Omega^{p}_{\mathrm{EL}}(J^\infty_M F):= \Omega^{p,0}(J^\infty_M F)$ for $0\leq p \leq d$ 
and $\Omega^{d+q}_{\mathrm{EL}}(J^\infty_M F):= \Omega^{d,q}_f(J^\infty_M F)$ for $1\leq q$.
\end{definition}
In particular, if $L\in \Omega^{d,0}(J^\infty_M F)$ is a Lagrangian density given in local coordinates by 
$L=\bar{L}\big(x^\mu,\{u^a_I\}_{0\leq |I|}\big) \cdot \dd x^1\cdots \dd x^d$, then the `integration by parts' 
algorithm (with differentials and derivatives taken on $J^\infty_M F$) gives
\vspace{-2mm} 
\begin{align}\label{LagrangianVerticalDifferentialDecomposition} 
\dd_V L  & = \sum \frac{\partial \bar{L}}{\partial u^a_I}  \cdot \dd_V u^a_I \wedge \dd x^1\cdots \dd x^n \nn \\ 
&= \sum (-1)^{|I|} D_I \bigg(\! \frac{\partial \bar{L}}{\partial u^a_I} \!\bigg) \cdot \dd_V u^a \wedge \dd x^1\cdots \dd x^n +\dd_H \theta_L \\
&=  \delta_V L + \dd_H \theta_L \, , \nn 
\end{align}

\vspace{-2mm} 
\noindent for some $\theta_L \in \Omega^{d-1,1}(J^\infty_M F)$, defined up to the addition of a horizontally closed $(d-1,1)$-form 
(and hence exact by Prop. \ref{AcyclicityTheorem}). This can be checked directly in coordinates using Def. \ref{InteriorEulerOperator} 
and explicit formulas for arbitrary order Lagrangians can be found, for instance, in 
\cite{Del}\cite{Blohmann23b}. For the simple (but generic) case where the Lagrangian is globally of order $k=1$, we have
\vspace{-2mm}
\begin{align}\label{1stOrderLagrangianDecomposition}
\theta_L = - \dd_V u^a \wedge  \mathbb{L}_{\frac{\partial}{\partial u^a_\mu}}\Big(  \iota_{\frac{\partial}{\partial x^\mu}} L\Big)  \, . 
\end{align}

\vspace{-2mm}
\noindent 
Generally, the interpretation is that $\theta$ absorbs the dependence on the vertical forms $\{\dd_V u^a_I\}_{|I|\geq 1}$. 
Abstractly, this is implied by the exactness of the second row in the augmented bicomplex. Equivalently, this reads 
$$
\delta_V L= \CI (\dd_V L) = \dd_V L - \dd_H \theta_L 
$$
and so the projection operator $\CI$ acts locally via partial integration. Indeed, this is exactly the purpose of 
the defining formula (Def. \ref{InteriorEulerOperator}). For reasons that will become apparent below, we denote the 
component of $\delta_V = \CI \circ \dd_V$ acting on Lagrangians by
\vspace{-2mm} 
$$
E:=\delta_V \;\; : \;\; \Omega^{d,0}(J^\infty_M F)\longrightarrow  \Omega^{d,1}_s(J^\infty_M F) \, ,
$$

\vspace{-2mm} 
\noindent and call source forms in its image \textit{Euler--Lagrange source forms}. The map $E$ itself is often called the
\textit{Euler--Lagrange differential} on $J^\infty_M F$. Note that when the Lagrangian density happens to be 
exact $L= \dd_H T$, then the induced Euler--Lagrange source form is trivial 
\vspace{-2mm} 
\begin{align}\label{ExactLagrangianTrivialELequation}EL = \delta_V L = \CI (\dd_V \dd_H T) = \CI(\dd_H \dd_V T) = 0 \, . 
\end{align}

\vspace{-2mm} 
\noindent For this reason, horizontally exact Lagrangian densities are known as \textit{trivial} Lagrangians. 

\newpage 
\subsection{On-shell space of fields and conserved currents}
\label{sec-onshell}

Using the explicit expression of the Euler--Lagrange form above, there is an immediate and natural description of the on-shell spaces of fields.
A useful intermediate step is to  define the smooth subspace of the infinite jet bundle where the Euler--Lagrange sources form vanishes.

\begin{definition}[\bf  Shell of Lagrangian]\label{ShellOfLagrangian}
The \textit{shell} of a Lagrangian $L\in \Omega^{d,0}(J^\infty_M F)$ is defined to be the smooth subspace of the jet bundle 
\vspace{-2mm} 
\begin{align*}
\mathrm{S}_L \longhookrightarrow y(J^\infty_M F) \, ,
\end{align*} 

\vspace{-2mm} 
\noindent on which the Euler--Lagrange source form vanishes. Equivalently, the shell is the pullback/intersection of smooth sets 
  \vspace{-2mm} 
 \[
\xymatrix@=1.6em  {S_L \ar[d] \ar[rr] &&   y(J^\infty_M F)\ar[d]^{E L} 
	\\ 
	y(J^\infty_M F)\ar[rr]^-{0^\infty_{F}}  && \wedge^d T^*M \otimes V^*F
	\, ,} 
\]

  \vspace{-2mm} 
\noindent where the Euler--Lagrange source form $EL$ is considered as a bundle map, via Lem. \ref{SourceFormsDefineDifferentialOperators}, 
and $0^\infty_F: J^\infty_M F \rightarrow \wedge^d T^*M \otimes V^*F$ is the (non-constant) fiberwise 0-bundle map over $F$. 
\footnote{This is induced by pulling back the zero section $0_F$ of the vector-bundle $\wedge^d T^*M \otimes V^*F \rightarrow F$
via $\pi_0 : J^\infty_M F\rightarrow F$.} 
\end{definition}
 Yet equivalently, the shell and the above diagram may be reinterpreted as the intersection of the zero section and the section
 induced by the $EL$-source form, of the induced pullback bundle $(\pi_0)^*\big(\wedge^d T^*M \otimes V^*F\big)$ over $J^\infty_M F$ itself. 
 Explicitly, the shell is the smooth space with points
$S_L(*)=\{s\in J^\infty_M F \, | \, EL(s)=0_{\pi_0(s)} \}$
and more generally $\FR^k$-plots 
\vspace{-1mm} 
\begin{align*}
\mathrm{S}_L(\FR^k)=\big\{s^k:\FR^k\rightarrow J^\infty_M F \; \big{|}\; EL\big(s^k(x)\big) =0_{\pi_0(s^k(x))} \quad \forall x\in \FR^k \big\}\, .
\end{align*}
According to Lem. \ref{SourceFormsDefineDifferentialOperators}, for any Lagrangian $L\in \Omega^{d,0}(J^\infty_M F )$, the resulting 
Euler--Lagrange source form $EL=\delta_V L\in \Omega^{d,1}_s  (J^\infty F)$ defines a differential operator and, in particular, a smooth map 
\vspace{-1mm} 
\begin{align}\label{EulerLagrangerSmoothSetMap}
\mathcal{EL} \;:\; \mathbold{\Gamma}_M (F) & \; \longrightarrow\;
\mathbold{\Gamma}_M(\wedge^d T^*M\otimes V^*F) \\[-2pt]
\phi^k &\;\longmapsto \; EL \circ j^\infty \phi^k\, . \nn  
\end{align}

\vspace{-2mm} 
\noindent By \eqref{HorizontalVectorFieldBasisAction}, in local coordinates this takes the usual form of 
the Euler--Lagrange operator acting on a field $\phi$, 
\vspace{-2mm} 
$$ 
\CE \CL(\phi)= EL\circ j^\infty \phi= \sum_{|I|=0}^{\infty} (-1)^{|I|} \frac{\partial}{\partial x^I}
\bigg(\frac{\partial \bar{L}}{\partial u^a_I} \circ j^\infty\phi \bigg) \cdot \dd_F u^a \wedge \dd x^1\cdots \dd x^n\, ,
$$

\vspace{-2mm} 
\noindent whose coefficients, by abuse of notation, are usually expressed as
\vspace{-2mm} 
\begin{align}\label{EulerLagrangeEquationsLocallyAbusingNotation}
\CE \CL_a (\phi) = EL_a \circ j^\infty \phi = \sum_{|I|=0}^{\infty} (-1)^{|I|} \frac{\partial}{\partial x^I} 
\bigg(\frac{\delta \bar{L}\big(x^\mu ,\{\partial_J \phi^b\}_{|J|\leq k} \big)}{\delta (\partial_I \phi^a )} \bigg)\, .
\end{align}

\begin{remark}[\bf Treating partial derivatives as independent]\label{TreatingPartialDerivativesAsIndependent}
For the above textbook form of the Euler--Lagrange equations (Eq. \eqref{EulerLagrangeEquationsLocallyAbusingNotation}) to make sense computationally, 
one implicitly treats the functions $\{\partial_I \phi^a\}_{|I|\leq k}$ on spacetime $M$ as independent variables -- which is factually not the case -- and 
acts formally via the corresponding partial differentiation. The rigorous justification for this treatment is given by the proper definition of the 
differential Euler--Lagrange operator (map \eqref{EulerLagrangerSmoothSetMap}) and its well-defined local coordinate formula,
whereby the coordinates $\{u^a_I\}_{|I|\leq k }$ are truly independent.
\end{remark}

\begin{definition}[\bf  On-shell smooth space of fields]
\label{OnshellSpaceOfFields}
$\,$

\noindent {\bf (i)} 
A field configuration $\phi$ is said to be \textit{on-shell} for a local Lagrangian $\CL$ if it satisfies the Euler--Lagrange equation
\vspace{-1mm} 
$$
\mathcal{EL}(\phi) = EL\circ j^\infty \phi = 0_\phi \in \Gamma_M(\wedge^d T^*M \otimes V^* F)\, ,
$$

\vspace{-1mm}
\noindent  where $0_\phi:= 0_F \circ \phi : M\rightarrow F \rightarrow \wedge^d T^*M \otimes V^*F $ is the canonical section covering
$\phi\in \Gamma_M(F)$, given by postcomposing with the zero-section of $\wedge^d T^*M \otimes V^*F\rightarrow F$.

\noindent {\bf (ii)}  Similarly, an $\FR^k$-plot of fields $\phi^k$ is on-shell if it satisfies $\mathcal{EL}(\phi^k)=EL\circ j^\infty\phi^k =0_{\phi^k} 
\in \mathbold{\Gamma}_M(\wedge^d T^*M \otimes V^*F)(\FR^k)$. 

\noindent {\bf (iii)}  The on-shell $\FR^k$-plots of fields define the 
\textit{smooth space of on-shell fields}, i.e., the smooth subspace of the smooth field space 
\vspace{-3mm} 
\begin{align*}
\CF_{\CE \CL} \longhookrightarrow \CF=\mathbold{\Gamma}_M(F)
\end{align*}

\vspace{-2mm} 
\noindent with 
\vspace{-2mm} 
\begin{align*}
 \CF_{\CE \CL}(\FR^k):=\big\{\phi^k \in \CF(\FR^k) \, \big{|} \, \mathcal{EL}(\phi^k)=EL\circ j^\infty\phi^k =0_{\phi^k} \big\} \, . 
\end{align*}
\end{definition}

\newpage 
Equivalently, an $\FR^k$-plot of fields $\phi^k\in \CF(\FR^k)$ is on-shell if and only if its prolongation 
$j^\infty \phi^k: y(\FR^k\times M) \rightarrow y(J^\infty_M F)$ 
factors through the shell of L
\vspace{-1mm} 
$$
j^\infty \phi^k: y(\FR^k \times M) \longrightarrow \mathrm{S}_L \longhookrightarrow y(J^\infty_M F)\, .
$$
In \cref{OnShellFieldsCriticalSetSubsection}, we will show how the on-shell space of fields is 
identified with the smooth critical subset of the 
corresponding action functional (or smooth Lagrangian), in an appropriate sense of the criticality condition. As we will make explicit therein
(Cor. \ref{CriticalSmoothSetPullback}), the above definition is in fact a universal construction within smooth sets -- i.e., pullback/intersection 
similar to that of the shell (Def. \ref{ShellOfLagrangian}) within the infinite jet bundle. 

\begin{example}[\bf O($n$)-model on-shell fields]\label{O(n)ModelOnshellFields} 
Recall the O($n$)-model Lagrangian from Ex. \ref{VectorValuedFieldTheoryLagrangian}.
We will consider the particular case of $c_4=0$ for the sake of brevity. 
The corresponding Euler--Lagrange operator is given by 
\vspace{-1mm}
$$
\CE\CL (\phi) = - \dd_M \star \dd_M \phi + \star \phi \quad \in \quad \,  \Gamma_M(\wedge^d T^*M \otimes W)  \, ,
$$

\vspace{-1mm}
\noindent
where the fibers of $V^*(W\times M) \cong W^*\times W \times M$ are implicitly identified with $W\times W$ via the Euclidean product on $W$. Equivalently as usually expressed, by application of a further Hodge dual isomorphism 
\vspace{-2mm}
$$
\star \CE \CL(\phi) = - \Delta \phi + \phi \quad \in \quad C^\infty (M,W)\, ,
$$

\vspace{-2mm}
\noindent
where $\Delta = \dd_M \star \dd_M \star$ is the Laplace--Beltrami operator. Thus, the smooth set of on-shell fields is 
comprised of (plots of) fields satisfying 
\vspace{-1mm}
$$
\Delta(\phi) = \phi \, .
$$

\vspace{-1mm}
\noindent 
The corresponding source form is easily read in coordinates. In terms of the decomposition at the level of the jet bundle we have 
$\delta L = EL + \dd_H \theta_L $ where
\vspace{-1mm} 
\begin{align}\label{O(n)ModelJetBundleVariationalDecomposition}
\theta_L=  - \dd_V u^a \wedge  \mathbb{L}_{\frac{\partial}{\partial u^a_\mu}}\Big(  \iota_{\frac{\partial}{\partial x^\mu}} L\Big) 
=  - \dd_V u^a \wedge g^{\mu \nu} u_{\nu \, a} \cdot \iota_{\frac{\partial}{\partial x^\mu}} \dd \mathrm{vol}_g 
=: - \langle \dd_V u \,, \star \dd_H u \rangle_g
 \quad \in \quad \Omega^{d-1,1}(J^\infty_M F) \, .
 \end{align}
\end{example}

 Although the shell $S_L\hookrightarrow J^\infty_M F$ (Def. \ref{ShellOfLagrangian}) serves as a good stepping stone 
towards the actual on-shell field space $\CF_{\CE \CL}\hookrightarrow \CF$ (Def. \ref{OnshellSpaceOfFields}), it is not quite 
its correct avatar within the infinite jet bundle. Indeed, a solution $\phi$ of a PDE such as the Euler--Lagrange equation 
$\CE \CL(\phi)=0$ of a field theory automatically satisfies an infinite list of implied differential equations. Intuitively, 
locally in coordinates and trivializations of the corresponding bundles, these are `generated' by applying arbitrary derivatives
on the original Euler--Lagrange differential condition \eqref{EulerLagrangeEquationsLocallyAbusingNotation}
\vspace{-2mm}
$$
\CE \CL_a (\phi)=0 \quad \implies \quad \frac{\partial}{\partial x^I} \CE \CL_a (\phi)=0  \, .
$$

\noindent These implied conditions, however, are not reflected in the shell $S_L\hookrightarrow J^\infty_M F$. In particular, 
for an Euler--Lagrange source form $EL= \pi_k^* EL_k : J^\infty_M F\rightarrow J^k_M F \rightarrow \wedge^d T^*M \otimes V^*F $ 
of global order $k$, an arbitrary point $s=j^\infty_p\phi \in S_L$ is such that $EL(s)=EL_k \circ \pi_k (s)= EL_k \big(j^k(\phi)\big) =0$, 
i.e., with conditions on the jet components up to order $k$. Crucially, there are no conditions on the higher jets, such as those induced
by the implied differential equations. Said otherwise, in local coordinates $\{x^\mu, \{u^a_I\}_{0\leq |I|} \}$ for $J^\infty_M F$, 
the vanishing of the Euler--Lagrange form imposes conditions on the coordinates up to order $|I|=k$, but leaves the rest of the coordinates free.

\medskip 
Taking into account the global structure of the field bundle $F\rightarrow M$ and the corresponding Euler--Lagrange differential 
operator $\CE \CL: \Gamma_M(F)\rightarrow \Gamma_M(\wedge^d T^*M \otimes V^*F)$, this translates to fact that
\vspace{-1mm}
$$
\CE \CL(\phi) = 0_\phi \, \in \Gamma_M( \wedge^d T^*M \otimes V^*F) \, \quad \implies \quad j^\infty \big(\CE \CL(\phi)\big) 
=0_{j^\infty(\phi)} \, \in \Gamma_M\big(J^\infty (\wedge^d T^*M \otimes V^*F) \big) \, .
$$

\vspace{-1mm}
\noindent 
Expanding the induced condition on the right and recalling Eq. \eqref{ProlongatedActionOnSections}, the (locally defined) implied 
conditions may be equivalently 
 encoded as
\footnote{For the reverse implication, notice that for any \textit{linear} differential
operator $\CK: \Gamma_M(\wedge^d T^*M \otimes V^*F) \rightarrow \Gamma_M(G)$, with underlying vector bundle map 
$K:J^\infty_M(\wedge^d T^*M \otimes V^*F)\rightarrow G$ valued in an arbitrary vector bundle $G$ over $F$ (and hence over M), $
\CE \CL(\phi) = 0_\phi \, \in \Gamma_M( \wedge^d T^*M \otimes V^*F) $ implies
\vspace{-1mm}
$$\CK \circ \CE \CL(\phi) = K\circ j^\infty \circ  E L \circ j^\infty (\phi)= K \circ \pr EL (j^\infty \phi) = 0 \quad  \in \quad \Gamma_M( G) \, .
$$}

\vspace{-3mm}
\begin{align*}
j^\infty \circ \CE \CL(\phi) = \pr EL \circ j^\infty \phi =  0_{j^\infty(\phi)} \quad \in \quad \Gamma_M\big(J^\infty (\wedge^d T^*M \otimes V^*F)\big)\, .
\end{align*}

\vspace{-1mm}
\noindent 
In other words, the prolongation $j^\infty \phi$ of a solution $\phi \in \CF_{\CE \CL}$ does not only factor through the shell $S_L$, i.e.,
zero locus of $EL:J^\infty_M F \rightarrow \wedge^d T^*M \otimes V^*F $, but furthermore through the `prolongated shell' -- the zero
locus of the prolongated Euler--Lagrange bundle map $\pr EL :J^\infty_M F \rightarrow J^\infty_M( \wedge^d T^*M \otimes V^*F )$.

\newpage 
\begin{definition}[Prolongated Shell of Lagrangian]\label{ProlongatedShellOfLagrangian}
The \textit{prolongated shell} of a Lagrangian $L\in \Omega^{d,0}(J^\infty_M F)$ is defined to be the smooth subspace of the jet bundle 
\vspace{-2mm} 
\begin{align*}
\mathrm{S}^\infty_L \longhookrightarrow y(J^\infty_M F) \, ,
\end{align*} 

\vspace{-2mm} 
\noindent on which the prolongated Euler--Lagrange source form $\pr EL$ vanishes. That is, the prolongated shell is the pullback/intersection of smooth sets 
  \vspace{-2mm} 
 \[
\xymatrix@=1.6em  {S^\infty_L \ar[d] \ar[rr] &&   y(J^\infty_M F)\ar[d]^{\pr E L} 
	\\ 
	y(J^\infty_M F)\ar[rr]^-{\pr 0^\infty_{F}}  && J^\infty_M (\wedge^d T^*M \otimes V^*F ) 
	\, ,} 
\]

  \vspace{-1mm} 
\noindent where the $\pr EL$ is the prolongation (Def. \ref{ProlongationOfJetBundleMap}) of the Euler--Lagrange source form $EL$ viewed as a 
bundle map, via Lem. \ref{SourceFormsDefineDifferentialOperators}, and $\pr 0^\infty_F: J^\infty_M F \rightarrow J^\infty_M (\wedge^d T^*M \otimes V^*F )$
is the (non-constant) fiberwise 0-bundle map over $F$.
\end{definition}
Yet equivalently, the prolongated shell and the above diagram may be reinterpreted as the intersection of the zero section and the section induced by
the prolongated $EL$-source form, of the induced pullback bundle $(\pi_0)^*\big( J^\infty_M (\wedge^d T^*M \otimes V^*F)\big)$ over $J^\infty_M F$ itself. 
Explicitly, the prolongated shell is the smooth space with points
$S_L(*)=\big\{s\in J^\infty_M F \; | \; \pr EL(s)=0_{\pi_0(s)}  \big\}$ and $\FR^k$-plots along the same lines.

\begin{remark}[\bf Manifold structure and nomenclature of the prolongated shell]\label{ManifoldStructureOnTheShell}
$\,$

\noindent {\bf (i)} 
Even though the ambient infinite jet bundle $J^\infty_M F$ is a Fr\'{e}chet manifold, it is not guaranteed that (the set of points of) the prolongated shell $S_L^\infty$
can be supplied with a Fr\'{e}chet submanifold structure. This is apparent even in the globally finite order Lagrangian case, where the corresponding 
Euler--Lagrange form is (the pullback of) a globally finite order map, say $EL^k:J^k_M F \rightarrow \wedge^d T^*M \otimes V^*F$. In this case, even the 
question of the (finite order) shell $S_{L,k}\hookrightarrow J^k_M F$ having a smooth (finite-dimensional) manifold structure demands the transversality 
of the maps $EL^k$ and $0_F$. This is not always guaranteed and depends on the explicit form of the Lagrangian and the induced Euler--Lagrange form.

\vspace{1mm} 
\noindent {\bf (ii)} 
 Even if this is the case, 
it might still be the case that they become non-transversal upon any finite order prolongation\footnote{These are defined analogously, so that
$\pr^q EL^k (j^{k+q}_p \phi) := j^q_p\big(P\circ j^k \phi \big)$ for some (local) representative section $\phi$. Equivalently, they may be 
defined directly via the $q$-order prolongation of $EL:= (\pi_k)^*EL^k$ on $J^\infty_M F$ as
$\pr^q EL(j^\infty_p \phi) := j^q_p \big(EL \circ j^\infty(\phi)\big)= j^q_p(EL^k \circ j^k \phi)= \pr^q EL^k \circ \pi_{q+k}(j^\infty_p\phi)$.} 
$\pr^q EL^k : J^{k+q}_M F\rightarrow J^q(\wedge^d T^*M\otimes V^*F)$, so that their zero locus intersection $S^{q}_{L,k}\hookrightarrow J^{k+q}_M F$ is not a smooth submanifold. Nevertheless, there is a smooth set
structure on each finite order prolonged shell, and the prolonged shell $S^\infty_L$ as their (projective) limit, since all limits exist in the category of smooth sets
(see also Ex. \ref{FiniteDimensionalCriticalSmoothSet}). More intricate details about the potential analytic nature required by constructions in 
local field theory will be noted in Rem. \ref{OnShellCartanCalculusCaveats}.

\vspace{1mm} 
\noindent {\bf (iii)} 
We note that the description of the `avatar' of an arbitrary partial differential equation, as the zero locus of the corresponding prolongated 
bundle map inside the infinite jet bundle, was first introduced by Vinogradov in \cite{Vin81}\cite{Vin84} and goes under the name \textit{diffiety}. 
Its geometrical properties and relation to the actual solutions of PDEs have been much studied since (see \cite{Vin13} for a modern review). 
Our nomenclature as the (prolonged) shell in the specific case of Euler--Lagrange operators is intended to make contact with the 
``on-shell fields'' nomenclature from the physics literature, thus being those fields that factor through the (prolonged) shell.
\end{remark}

\smallskip 
Given the exactness properties of the variational bi-complex, along with the identification of the total cohomologies of Prop. \ref{TotaldeRhamCohomologyInfinityJet} 
and Prop. \ref{AcyclicityTheorem}, 
together with enough diagram chasing, one can arrive at the following result on the cohomology of the Euler--Lagrange complex
\cite{Takens79}\cite{Anderson89}\cite{GMS00}.

\begin{proposition}[\bf Euler--Lagrange cohomology]
\label{EulerLagrangeComplexCohomology}
The projection map $\pi^\infty_0 :J^\infty_M F\rightarrow F$ along with the projections $\pr_{p,0} $ and $\CI\circ \pr_{d,q}$, onto horizontal 
and $(d,q)$-source forms respectively, induce a quasi-isomorphism 
\vspace{-2mm} 
$$
\Omega^{\bullet}(F) \longrightarrow 
\Omega^{\bullet}_{\mathrm{EL}}(J^\infty_M F)\, ,
$$

\vspace{-2mm} 
\noindent and so isomorphisms on the corresponding cohomologies.
\end{proposition}

In particular, this implies that $H^{n}_{\mathrm{EL}}(J^\infty_M F) = 0$ for $n>\mathrm{dim}(F)$. Furthermore, any $\dd_H$-closed horizontal form 
$\om_{p,0}\in \Omega^{p,0}(J^\infty_M F)$ for $0\leq p < d$ decomposes as
\vspace{-2mm} 
$$
\om_{p,0}=\pr_{p,0}\circ  (\pi^{\infty}_0)^*(\tilde{\om}_{p}) + \dd_H k_{p-1,0}\,,
$$

\vspace{-2mm} 
\noindent for some closed $\tilde{\om}_{p}\in \Omega^p(F)$ and $k_{p-1,0}\in \Omega^{p-1,0}(J^\infty_M F)$, while the same holds for any density 
$L\in \Omega^{d,0}(J^\infty_M F)$ which is \textit{$\delta_V$-closed} instead. Since $\delta_V L= EL =0$ means the Lagrangian has trivial 
Euler--Lagrange derivatives as per Eq. \eqref{ExactLagrangianTrivialELequation}, the cohomology $H^d_{\mathrm{dR}}(F)$ parametrizes variationally trivial Lagrangian densities. 
\begin{remark}[\bf The (global) inverse problem of the calculus of variations] 
Given a (smooth) differential 
operator 
\vspace{-2mm} 
\begin{align*}
\mathcal{P} \;:\; \mathbold{\Gamma}_M (F) & \;\longrightarrow\; \mathbold{\Gamma}_M(\wedge^d T^*M\otimes V^*F) \\
\phi &\; \longmapsto \; P \circ j^\infty \phi\, , 
\end{align*}

\vspace{-2mm} 
\noindent for some $(d,1)$-source form $P\in \Omega^{d,1}_s(J^\infty F)$, does there exist a Lagrangian $(d,0)$-form such that $P=\delta_V L$? 
In words, is the ``\textit{source equation}"
$\CP=0$ \textit{``variational"}? A necessary condition is that $\delta_V P=0$. This condition guarantees the 
existence of a Lagrangian \textit{locally} on 
$J^\infty_M F$ \cite{Tonti69}. Globally, the obstruction to be variational lies in the $(d+1)$-cohomology of the Euler--Lagrange 
complex  (Def. \ref{EulerLagrangeComplex}), and hence by Prop. \ref{EulerLagrangeComplexCohomology}, in the $(d+1)$ de Rham cohomology $H^{d+1}_{\mathrm{d R}}$(F) of the underlying field bundle. This problem and its reduction to the cohomological question is historically one of the main reasons for the introduction of the Euler--Lagrange complex. 
\end{remark}

\smallskip 
Forms on the infinite jet bundle may be pulled back to forms on spacetime, along any prolongation of a field configuration $j^\infty \phi : M \rightarrow J^\infty_M F$, in a manner that is compatible with the horizontal differential $\dd_H$ on
$J^\infty_M F$ and the de Rham differential $\dd_M$ of the base $M$.

\begin{definition}[\bf  Pullback via prolongated section]\label{PullbackViaProlongatedSection}
    Let $\phi\in \Gamma_M(F)$ be a smooth section and $j^\infty \phi\in \Gamma_M(J^\infty F)$ its prolongation. 
    Define the pullback morphism of 
differential forms 
\vspace{-1mm} 
$$
(j^\infty \phi)^*  \;:\; \Omega^{\bullet}(J^\infty_M F)\longrightarrow \Omega^\bullet(M)
$$ 
as follows: For any $m$-form 
$\om\in \Omega^m(J^\infty_M F)$, let $(j^\infty \phi)^* \om\in \Omega^m(M)$ denote the $m$-form on the base $M$ given by 
\vspace{-1mm} 
$$
(j^\infty \phi)^*\om (X_p) := \om_{j^\infty \phi(p)} \big(\dd (j^\infty \phi)_p X_p\big)\, ,
$$
for all $X_p\in T_p M$, $p\in M$, where $d(j^\infty \phi)_p$ denotes the pushforward map from  \eqref{PushforwardAlongProlongationOfField}.
\end{definition}

\begin{lemma}[{\bf Horizontal differential and base de Rham compatibility}]
\label{HorizontalDifferentialBaseDeRhamCompatibility}
 The pullback $(j^\infty \phi)^*  : \Omega^{\bullet}(J^\infty_M F)\longrightarrow \Omega^\bullet(M)$ \textit{respects the horizontal differential}
 \vspace{-3mm}
\begin{align}
(j^\infty \phi)^* \dd_H \om = \dd_M (j^\infty \phi)^* \om \, ,
\end{align}

\vspace{-2mm} 
\noindent and \textit{vanishes on forms of non-zero vertical degree}
\vspace{-2mm} 
$$
(j^\infty \phi)^* \om_{p,q\geq1}  =0\, .
$$
\begin{proof}
It suffices to check the relations on $C^\infty(J^\infty_M F)$ and the local basis of $\Omega^{1,1}(J^\infty_M F)$, 
since these generate $\Omega^{\bullet,\bullet}(J^\infty_M F)$. In local coordinates, recall by
\eqref{InftyJetBundleHorizontalTangentCoordinate},
\vspace{-3mm} 
$$
\dd(j^\infty \phi)_p \Big(\frac{\partial}{\partial x^\mu}\Big\vert_p\Big)  = \frac{\partial}{\partial x^\mu} \Big\vert_{j^\infty_p \phi} 
+ \sum_{|I|=0}^{\infty} u^a_{I+\mu}(j^{\infty}_p \phi) \cdot    \frac{\partial}{\partial u^a_I} \Big\vert_{j^\infty_p \phi} \, . 
$$ 

\vspace{-2mm} 
\noindent On the local coordinate functions $\{ x^\mu\}$ for $J^\infty_M F$, we have  
\vspace{-1mm} 
$$
(j^\infty \phi)^* \dd_H x^\mu = (j^\infty \phi)^*\dd x^\mu = \dd_M (j^\infty \phi^* x^\mu) = \dd_M x^\mu \, 
$$

\vspace{-1mm} 
\noindent where the first equality is due to \eqref{HorizontalVerticalDifferentialonCoordinates}, 
the second by acting on $\frac{\partial}{\partial x^\mu}|_p$, and the last equality due to the section condition 
of $j^\infty \phi$ over M. Similarly, for any smooth function $f\in C^\infty(J^\infty_M F)$, we have 
\vspace{-1mm} 
\begin{align*}
     (j^\infty \phi)^* \dd_H f &= (j^\infty \phi)^*\big(D_\mu(f) \cdot \dd x^\mu)= 
     \big( D_\mu(f)\circ j^\infty \phi\big) \cdot (j^\infty \phi)^* \dd x^\mu\\
     &=\frac{\partial }{\partial x^\mu} (f\circ j^\infty \phi) \cdot \dd_M x^\mu 
\end{align*}

\vspace{-1mm} 
\noindent with the last equality being the chain rule as in \eqref{HorizontalVectorFieldBasisAction}. On the local basis of horizontal 1-forms 
$\{\dd_H x^\mu= \dd x^\mu\}$ it is immediate that
\vspace{-1mm} 
$$
(j^\infty \phi)^* \dd_H(\dd_H x^\mu) =0=\dd_M (\dd_M x^\mu) = \dd_M \big((j^\infty \phi)^* \dd_H x^\mu \big) \, , 
$$
while on the vertical $1$-forms $\{\theta^a_I=\dd_V u^a_I\}$
\begin{align*} (j^\infty \phi)^* \dd_H (\dd_V u^a_I) &= (j^\infty \phi)^*(- \dd_V u^a_{I+\mu} \wedge \dd x^\mu) = 0 \, , 
\end{align*}

\vspace{-2mm} 
\noindent since  $\dd_V u^a_{I+\mu} \wedge \dd x^\mu$ vanishes when applied to two horizontal vectors, such as the horizontal lifts 
$\dd(j^\infty \phi)_p(\frac{\partial}{\partial x^\mu}|_p)$. Similarly, $\dd_V u^a_I$ vanishes on horizontal vectors by
definition, and so 
$(j^\infty \phi)^* \dd_V u^a_I =0$ which also implies 
\vspace{-1mm} 
$$
\dd_M\big((j^\infty \phi)^* \dd_V u^a_I \big)=0 =  (j^\infty \phi)^* \dd_H (\dd_V u^a_I)\, .
$$

\vspace{-1mm} 
\noindent The above proves that $(j^\infty \phi)^*  \dd_H = \dd_M (j^\infty \phi)^* $, while the fact that 
$(j^\infty \phi)^*  \dd_V u^a_I = 0$ 
furthermore implies the vanishing result $(j^\infty \phi)^*  \om_{p,q\geq 1}=0$ by expression \eqref{pqformInfinityJetLocalCoordinates}.
\end{proof}
\end{lemma}
By a straightforward
calculation along the lines of the proof above, and recalling the identification of  \eqref{BundleMapsAsHorizontalForms}, 
the currents of Def. \ref{CurrentOnFieldSpace} induced by horizontal forms $P\in \Omega^{p,0}(J^\infty_M F)$ are then equivalently given by 
\vspace{-1mm} 
\begin{align}\label{CurrentsViaPullback}
\CP: \Gamma_M(F) &\longrightarrow \Omega^{p}(M)\\
\phi &\longmapsto P\circ j^\infty (\phi)  = (j^\infty \phi)^*P \nn \, , 
\end{align}

\vspace{-2mm} 
\noindent where the pullback denotes that of forms in the sense of Def. \ref{PullbackViaProlongatedSection}.

\begin{corollary}[\bf  Stokes' Theorem on field space]\label{StokesTheoremLocalFunctions}
Consider a horizontal $(p,0)$-form $P\in \Omega^{p,0}(J^\infty_M F)$ and its horizontal differential $(p+1,0)$-form
$\dd_H P \in \Omega^{p+1,0}(J^\infty_M F)$.

\noindent {\bf (i)} For any compact oriented submanifold $B^{p+1}\hookrightarrow M$ with boundary $\partial B^{p+1}=\Sigma^p$,  we have 
\vspace{-2mm} 
$$
\int_{B^{p+1}}  \dd_H P \circ j^\infty\phi  =  
\int_{\Sigma^p} P \circ j^\infty \phi  \, ,
$$

\vspace{-1mm} 
\noindent 
for any field $\phi \in \CF=\Gamma_M(F)$. 

\noindent {\bf (ii)} Similarly, the same formulas hold for $\FR^k$-plots of fields, valued in $C^\infty(\FR^k,\FR)$. 
That is, the induced charges defined by $\CP_{ \Sigma^{p}}:= \int_{ \Sigma^{p}} P \circ j^\infty_M $ and 
$\int_{B^{p+1}}  \dd_H P\circ j^\infty_M$ agree as smooth real-valued functions 
\vspace{-1mm} 
$$
\CF \longrightarrow y(\FR)
$$

\vspace{-2mm} 
\noindent on field space.

\begin{proof}
By Stokes' Theorem, 
    \vspace{-3mm} 
\begin{align*} 
\int_{B^{p+1}}  \dd_H P \circ j^\infty\phi  &=  
\int_{B^{p+1}}  (j^\infty \phi)^*\dd_H P = \int_{B^{p+1}}\dd_M  (j^\infty \phi)^*P \nn 
\\
&= \int_{\Sigma^{p}} (j^\infty \phi)^*P = \int_{\Sigma^p} P \circ j^\infty \phi  \, ,
\end{align*}

\vspace{-2mm}
\noindent and similarly for plots of fields.
\end{proof}
\end{corollary}
The above corollary naturally leads to the notion of conserved local currents. In particular, if $P$ is horizontally closed 
$\dd_H P=0$, then 
\vspace{-3mm} 
\begin{align}\label{ConservedCurrentFormula}
0=\dd_M \circ \CP \;:\; \CF \longrightarrow \Omega^{p+1}_{\mathrm{Vert}}(M)
\end{align}
and the so current is `conserved'. In particular  $\CP_{\partial B^{p+1}}(\phi)=0$, i.e., the total flux/current passing through
\textit{any bounding} 
submanifold $\partial B^{p+1}\hookrightarrow M$ is zero, for any field configuration $\phi$ (and similarly plots of fields).

\begin{definition}[\bf  Conserved currents and charges]
A current $\CP:\CF \rightarrow \Omega^{p}_{\mathrm{Vert}}(M)$ induced by a horizontally closed $(p,0)$-form $P\in \Omega^{p,0}(J^\infty_M F)$ 
is called a \textit{conserved current}. The induced charges are called \textit{conserved charges}.
\end{definition}
The intuition for a conserved charge is the usual one. If $\Sigma^p =\Sigma^p_{-}\coprod \Sigma^p_{+}$ 
is the disjoint union of two cobordant submanifolds, $\Sigma^p=\partial B^{p+1}$, 
then $\CP_{\Sigma^p} = 0 = \CP_{\Sigma^p_+}- \CP_{\Sigma^p_{-}}$ and so 
\vspace{-2mm} 
\begin{align}\label{ConservedChargeOverCobordism}
\CP_{\Sigma^p_+} = \CP_{\Sigma^p_{-}} \;:\; \CF \longrightarrow \FR \, ,
\end{align}

\vspace{-2mm} 
\noindent 
i.e., the charge is conserved along the cobordism $B^{p+1}$. 
\begin{remark}[\bf Spacetime cohomology and triviality of currents]\label{SpacetimeCohomologyAndTrivialityOfCurrents} 
Let $\CP=P\circ j^\infty $ be a conserved local current on $\CF$, i.e., $\dd_M \CP = (\dd_H P)\circ j^\infty = 0$ as a smooth map on $\CF$. 
If the de Rham cohomology $H^p_{\mathrm{d R}}(M)$ vanishes, then for any fixed field configuration $\phi \in \CF(*)$ there exists 
some $T_{\phi} \in \Omega^{p-1}(M)$ such that 
\vspace{-2mm}
\begin{align*}
 \CP(\phi) = \dd_M T_\phi \hspace{0.5cm}\in \hspace{0.5cm}  \Omega^{p}(M)\, .   
\end{align*}

\vspace{-2mm}
\noindent However, the dependence of $T_\phi$ on $\phi\in \CF$ is by no means guaranteed to be \textit{local} or \textit{smooth}.
\end{remark}

If $\CP(\phi)=\dd_M \CT_\phi$ happens to be exact,
then its charges are identically zero along \textit{any} closed submanifold, and one says $\CP$ is a trivial current at $\phi\in \CF(*)$.
Of course, it may still have non-trivial charges over submanifolds with boundary, but the nomenclature is standard. 

\begin{definition}[\bf  Trivial local currents and charges]
A local current $\dd_M \CT :\CF \rightarrow \Omega^{p}_{\mathrm{Vert}}(M)$ induced by a horizontally exact $(p,0)$-form 
$\dd_H T\in \Omega^{p,0}(J^\infty_M F)$ is called a \textit{trivial current}.
\end{definition}
It follows that the cohomology $H^{\bullet< d}_{\mathrm{EL}}(J^\infty_M F)$ of Euler--Lagrange complex (Def. \ref{EulerLagrangeComplex})
up to degree $d-1$, and so equivalently by Prop. \ref{EulerLagrangeComplexCohomology} the de Rham cohomology
$H^{\bullet < d}_{\mathrm{dR}}(F)$, parametrizes the classes of non-trivial (off-shell) conserved local currents on the 
field space $\CF$. In particular, if the degree $p$-cohomology of $F\rightarrow M$ vanishes, then every conserved
$p$-form local current on $\CF$ is trivial (in a local manner, in contrast with Rem. \ref{SpacetimeCohomologyAndTrivialityOfCurrents}). 
For instance, if the field bundle is a vector bundle over a contractible spacetime, e.g. 
$W\times \FR^d \rightarrow \FR^d$ as in the case of vector-valued field theory (Ex. \ref{VectorValuedFieldTheoryLagrangian}) 
on Minkowski spacetime, then all of its cohomology vanishes and hence all 
conserved (off-shell) charges are necessarily trivial. In particular, the currents from Ex. \ref{VectorValuedFieldSpaceCurrents}(i), (ii)
are easily seen to be conserved off-shell, by a direct application of the de Rham differential.

\medskip 
In classical field theory, the currents of most interest are only preserved on the smooth subspace $\CF_{\CE \CL}\hookrightarrow \CF$ of on-shell fields.
\begin{definition}[\bf  On-shell conserved currents and charges]
\label{OnshellConservedCurrentsandCharges}
A $(p,0)$-form $P$ is `\textit{horizontally closed on the (prolongated) shell of} $L$' if $\dd_H P \circ \iota_{\mathrm{S}^\infty_L} =0$ 
as a bundle map
\vspace{-2mm} 
\begin{align*}\mathrm{S}^\infty_L \longhookrightarrow y(J^\infty_M F) \longrightarrow y(\wedge^{p}T^*M)\, 
\end{align*}

\vspace{-2mm} 
\noindent  over $M$. The induced current $\CP$ and charges are called \textit{on-shell conserved current and charges}.
\end{definition}
It follows by Stokes' Theorem on field space, Cor. \ref{StokesTheoremLocalFunctions}, that 
\vspace{-2mm} 
\begin{align*}
0=\dd_M \circ \CP \;:\; \CF_{\CE \CL} \longhookrightarrow \CF \longrightarrow \Omega^{p+1}_{\mathrm{Vert}}(M)\, ,
\end{align*}

\vspace{-2mm} 
\noindent  and so is indeed `conserved on-shell'. Similarly
\vspace{-2mm} 
\begin{align*}
\CP_{\Sigma^p_+} = \CP_{\Sigma^p_{-}} \;:\; \CF_{\CE \CL}\longhookrightarrow \CF \longrightarrow y(\FR) ,
\end{align*}

\vspace{-2mm} 
\noindent  for the corresponding charges for any $\Sigma^p=\Sigma^p_{-}\coprod \Sigma^p_{+}$ bounding a submanifold 
$B^{p+1}$. Of course, such conserved currents arise when (infinitesimal) local symmetries exist, via Noether's First Theorem (Prop. \ref{Noether1st}). 
In particular, the O($n$)-model local $(d-1)$-currents from Ex. \ref{VectorValuedFieldSpaceCurrents}\,{{(iii)}} are not off-shell conserved, 
but can be checked to be on-shell conserved when the vector fields employed are infinitesimal symmetries.

\begin{remark}[\bf EL-cohomology and on-shell currents]
Crucially, the cohomology $H^{\bullet<d}_{\mathrm{dR}}(F)$ of the field bundle \textit{does not} parametrize on-shell 
conserved currents. Indeed, even if the cohomology vanishes, and hence all off-shell currents are trivial, there 
can still be non-trivial on-shell conserved currents. We will come back to the distinction of off- and on-shell 
conserved currents in the form of Noether's first and second theorems in the following section.
\end{remark}

We close off this section by proving that a local symmetry of a Lagrangian field theory 
(Def. \ref{FiniteSymmetryofLagrangianFieldTheory}) preserves the smooth subspace of on-shell fields, 
justifying its definition. 

\begin{proposition}[{\bf Local symmetries preserve the on-shell space of fields}]\label{LocalSymmetryPreservesOnshellSpace}
Any local symmetry $\CD:\CF\rightarrow \CF$ of a Lagrangian field theory such that $\CL \circ \CD = \CL + \dd_M \CK $,
where $\CK = K\circ j^\infty$ for some $K\in \Omega^{d-1,0}(J^\infty_M F)$, preserves the smooth subspace
$\CF_{\CE \CL}\hookrightarrow \CF$ of on-shell fields. Namely, the local diffeomorphism $\CD:\CF\rightarrow \CF$ 
restricts to a diffeomorphism
$$
\CD|_{\CF_{\CE \CL}} \;:\; \CF_{\CE \CL}\xlongrightarrow{\sim} \CF_{\CE \CL}\, .
$$ 
\end{proposition}
\begin{proof}
By Lem. \ref{HorizontalDifferentialBaseDeRhamCompatibility}, the trivial smooth Lagrangian $\dd_M \CK$ is equivalently $(\dd_H K)\circ j^\infty$.
So, by Eq. \eqref{ExactLagrangianTrivialELequation},
its corresponding smooth Euler--Lagrange operator vanishes. Hence, 
$$
\CE (\CL\circ \CD)  = \CE(\CL + \dd_M \CK) = \CE \CL + \CE (\dd_M\CK) = \CE \CL\, ,
$$
i.e., the corresponding smooth Euler--Lagrange operators of $\CL\circ \CD $ and $\CL$ coincide. It follows they define the same on-shell space of fields. 

Next, we note that the Euler--Lagrange operator of the local Lagrangian $\CL \circ \CD$ is further equivalently given
by $\CE (\CL \circ \CD) = \CE \CL \circ \CD$. To see this, first note that by relation \eqref{ProlongatedActionOnSections}
\vspace{-2mm} 
$$\CL \circ \CD = L \circ j^\infty \circ D \circ j^\infty = L\circ \pr D \circ j^\infty = (\pr D^*L)\circ j^\infty\, ,
$$ 

\vspace{-2mm} 
\noindent where $\pr D:J^\infty_M F\xrightarrow{\sim} J^\infty_M F$ is the prolongated bundle map (Def. \ref{ProlongationOfJetBundleMap}), 
and the final equation is the pullback in terms of horizontal forms $\Omega^{d,0}(J^\infty_M F)$. It is tedious but straightforward 
to check that the pullback via a prolongated bundle map as above commutes with the  Euler operator on $\Omega^{\bullet,\bullet}(J^\infty_M F)$ 
and so $E(\pr D^* L)=\pr D^* EL $ (e.g. in local coordinates, see Prop. 1.6 and Cor 3.22 of \cite{Anderson89}\,\footnote{The statement 
of Prop 1.6 therein is for a prolongation of a bundle map $F\rightarrow F$, but the proof can be directly extended to include prolongations 
of bundle maps $J^\infty_M F\rightarrow F$.}). Precomposing with the jet prolongation $j^\infty \phi$ of any field configuration gives
\begin{align*}
 \CE (\CL \circ \CD)(\phi):&= E(\pr D^*L) \circ j^\infty \phi   = (j^\infty\phi)^* \, \big( \pr D^* EL \big)   \\ 
 &= (\pr D \circ j^\infty \phi)^* EL =EL \circ (\pr D \circ  j^\infty \phi)  \\
 &= 
  EL \circ \big (j^\infty \circ \CD(\phi) \big)= \CE\CL \circ \CD(\phi)\, .
 \end{align*}
It follows that if $\phi$ 
is an on-shell field for $\CL$, then $\CE\CL(\phi)=0$ implies $\CE\CL \big(\CD(\phi)\big)=0$ and similarly for $\FR^k$-plots of on-shell
fields. In other words `the image' 
of $\CD$ sits inside $\CF_{\CE \CL}$, and since $\CD$ is by assumption invertible, the result follows.
\end{proof}

Note that the result of this proposition can be interpreted colloquially as the statement that a local symmetry of an action/Lagrangian is also
a symmetry of the corresponding Euler--Lagrange equations. The proof here relies crucially on the fact that $E(\pr D^* L)=\pr D^* EL$,
for any prolongation of a bundle map $D:J^\infty_M F\rightarrow F$ over $M$. This is also true (using identical arguments) 
for prolongations of bundle maps $D_f: J^\infty F \rightarrow F$ covering  diffeomorphisms $f:M\rightarrow M$. Overall, this shows 
that the Euler--Lagrange differential operators $\CE \CL: \CF\rightarrow \mathbold{\Gamma}_M(\wedge^d T^*M \otimes V^*F)$ are 
`covariant' under the action of the induced diffeomorphisms on field space  $\CD_f:= D_f \circ j^\infty(-) \circ f^{-1} : \CF\rightarrow \CF$,
\vspace{-2mm}
\begin{align}\label{CovarianceofEulerLagrangeOperator}
\CE(\CL \circ \CD_f) := E(\pr D_f^* L) \circ j^\infty \equiv \CE \CL \circ \CD_f \, .    
\end{align}

\vspace{-2mm}
\noindent
In Prop. \ref{SymmetryPreservesOnshellSpace}, we will use the critical locus characterization of the on-shell space of fields to 
show that any (spacetime covariant or local) symmetry descends to a diffeomorphism on the on-shell space of fields. We note 
also that the proof of Prop. \ref{SymmetryPreservesOnshellSpace} does \textit{not} explicitly rely on the covariance of the Euler--Lagrange operators. 

\begin{example}[\bf O($n$)-model symmetries vs. on-shell fields]
Consider the O($n$)-model on-shell fields from Ex. \ref{O(n)ModelOnshellFields}, i.e., $\phi \in C^\infty(M,W)$ such that
\vspace{-1mm}
$$
\CE \CL(\phi) = - \Delta \phi + \phi= - \star \dd \star \dd \phi + \phi =0\, .
$$

\vspace{-1mm}
\noindent For the local finite symmetries of  $\CD: [M,W]\rightarrow [M,W]$ of the Lagrangian from Ex. \ref{LocalSymmetriesOfO(n)Model}, 
given by $\CD(\phi)= g\cdot \phi = g^{a}_{\, b}\cdot  \phi^b \cdot e_b $, it follows that 
\vspace{-2mm}
\begin{align*}
\CE \CL\big(\CD(\phi)\big) &= - \Delta (g\cdot \phi) + g\cdot \phi \\
&= - g \cdot \big( \Delta(\phi) + \phi) 
\\&=0\, ,
\end{align*}

\vspace{-2mm}
\noindent since $g$ is constant when viewed as a function on M $g:M\rightarrow O(n)$, and hence commutes with the de Rham differential and Hodge dual. 
For the spacetime covariant symmetry $\CD : [M,W]\rightarrow [M,W]$, given by pulling back along an isometry $\CD(\phi)= f^* \phi$, it follows that 
\vspace{-2mm}
\begin{align*}
\CE \CL\big(\CD(\phi)\big) &= - \Delta (f^* \phi) + f^* \phi= - f^* \big( \star \dd \star \dd \phi + \phi \big) \\
&= - f^* \big( \Delta(\phi) + \phi) 
\\&=0\, ,
\end{align*}

\vspace{-2mm}
\noindent since pullbacks commute with the de Rham differential, and further since the pullback by an isometry commutes with the Hodge dual.

\end{example}

\subsection{On-shell fields as a smooth critical set}
\label{OnShellFieldsCriticalSetSubsection}
In this section, we show how the usual variation algorithm of local action functionals, or more generally local Lagrangians, appearing 
in the physics literature may be rigorously interpreted in smooth sets. This completely avoids having to deal with any functional 
analytic and infinite-dimensional manifold technicalities in the actual field space. The resulting space of classical solutions, i.e.,
the space of ``extrema of the action functional", naturally forms a smooth set and coincides with the space of on-shell fields of 
Def. \ref{OnshellSpaceOfFields}. 

\medskip 
Let $\CF$ be a smooth set and $S:\CF\rightarrow y(\FR)$ a smooth map, e.g. a field space and an action functional. 
Evaluating on any $\FR^1$-plot $\phi_t$, the quantity $S(\phi_t)$ is an element of $y(\FR)(\FR^1_t)=C^\infty(\FR^1_t,\FR)$ 
and so we may compute its derivative with respect to $t\in \FR^1_t$, i.e., 
\vspace{-2mm}
\begin{align}\label{PartialDerivativeOfActionOnFR^1plots}
\CF(\FR^1_t) &\longrightarrow C^\infty(\FR^1_t)\\
 \phi_t &\longmapsto \partial_t S(\phi_t) \, . \nn 
\end{align}

\vspace{-1mm}
\noindent Denote by $\CF_{\phi}(\FR^0\times \FR^1_t)\subset \CF(\FR^1_t)$ the subset of paths 
$\phi_t\in \CF(\FR^1_t)\cong \CF(\FR^0\times \FR^1_t)$ such that $\phi_{t=0}:= \iota^*_0 \phi_t = \phi$, where 
$\iota_{0}: \{0\}\cong \FR^0 \hookrightarrow \FR^1$. We say $\phi\in \CF(\FR^0)$ is a \textit{critical point} 
of $S$ if the map (of sets)
\vspace{-1mm}
\begin{align*}
	\CF_{\phi}(\FR^0 \times \FR^1_t)&\xlongrightarrow{\quad \partial_t S \quad} 
 C^\infty(\FR^1_t) \xlongrightarrow{\;\;\mathrm{ev}_{t=0}\;\;} \FR  \nn 
 \\
	\phi_t & \quad \longmapsto \quad {\partial_t} S(\phi_t) \quad \longmapsto \;\;  {\partial_t}S(\phi_t)|_{t=0} 
 \, \, 
\end{align*}

\vspace{-1mm}
\noindent  vanishes.\footnote{It is easy to check that the analogous point-criticality condition for variations with higher dimensional paths, 
e.g. $\phi_{(t_1,t_2)} \in \CF(\FR^2)$ with $\phi_{(t_1,t_2)=0}=\phi^0$, is redundant since it is implied by the 1-parameter 
variation condition.}
The interpretation of the formulas above is the intuitive one, we have a notion of paths  through a point in our space of fields 
and we differentiate at the given point. 

\medskip 
Smooth sets allow us to consider more than the bare set 
$\mathrm{Crit}(S)(*)\subset \CF(*)$ of critical points. Denoting similarly by 
$\CF_{\phi^1}(\FR^1\times\FR^1_t)\subset \CF(\FR^2)$ the `surfaces' $\phi^1_t\in \CF(\FR^2)$ such that  
$\phi^1_{t=0}:= \iota_{\FR^1\times \{0\}}^* \phi^1_t =\phi^1$, where 
$\iota_{\FR^1\times \{0\}}: \FR^1\cong \FR^1\times \{0\} \hookrightarrow \FR^2$. 
Indeed, we say $\phi^1 \in \CF(\FR^1)$ is a \textit{critical line} or critical $\FR^1$-plot of S if the map (of sets)
\begin{align*}
	\CF_{\phi^1}(\FR^1 \times \FR^1_t)&\xlongrightarrow{\quad \partial_t S\quad} 
C^\infty(\FR^1 \times \FR^1_t )\xlongrightarrow{\;\;\mathrm{ev}_{t=0} \;\;} C^\infty(\FR) \nn 
 \\
	\phi_t^1 & \quad \longmapsto \qquad {\partial_t} S(\phi_t^1) \qquad \;\; \longmapsto \;\; {\partial_t}S(\phi_t^1)|_{t=0} 
\, 
\end{align*}

\vspace{-1mm}
\noindent 
vanishes. Note that the first map takes only derivative with respect to $t\in \FR^1$. Naturally, we may consider more than the 
set of critical lines of $S$ by $\mathrm{Crit}(S)(\FR^1)\subset \CF(\FR^1)$. Inducting the above formula to higher dimensional plots 
and their 1-parameter variations, we have the following definition.
\begin{definition}[\bf  Critical plots of smooth map]\label{CriticalRkPoints}
Let $S:\CF \rightarrow y(\FR)$ be a smooth map. The \textit{critical $\FR^k$-plots} of $S$ is the subset of $\FR^k$-plots
\vspace{-3mm} 
\begin{align}
		\mathrm{Crit}(S)(\FR^k):=  \Big\{\phi^k\in \CF(\FR^k) \; \big{|}\; \partial_t S (\phi^k_t)|_{t=0}=0, \hspace{0.3cm}
  \forall \, \phi^k_t\in  \CF_{\phi^k}(\FR^k\times \FR^1_t) \Big\} ,		
\end{align}
where $\CF_{\phi^k}(\FR^k\times \FR^1_t) = \big\{\phi^k_t \in \CF(\FR^k\times \FR^1_t) \, | \, \phi^k_{t=0}= \phi^k \in \CF(\FR^k) \big\}$. 
\end{definition}
Hence, there is an assignment of sets of $\FR^k$-critical plots
\vspace{-1mm}
$$
\FR^k \longmapsto \mathrm{Crit}(S)(\FR^k)\, ,
$$

\vspace{-2mm}
\noindent  for each $k\in \NN$. A priori, if $\CF$ and $S$ are an arbitrary smooth set and smooth map respectively, this assignment 
might not be functorial under maps $f:\FR^k\rightarrow \FR^{k'}$. Even if it does, it might still not satisfy the sheaf condition. 
That is, it \textit{may not} define a smooth set. Categorically, this is due to the 
fact that the above definition is not - in general - identified with a universal construction in SmoothSet, e.g.
a limit of some sort.\footnote{The culprit is \textit{not} simply the fact that, generally, the quantity 
$\partial_t S(\phi^1_t)|_{t=0}$ does not depend only on `$\partial_t \phi_t |_{t=0}$' as in Rem. \ref{TangentVectorsPathsOfFieldsAndDerivations}. 
Indeed, this condition alone would not be sufficient for functoriality.}
Nevertheless, the structure of local action functionals on field spaces guarantees the functoriality of the critical set.

\begin{proposition}[\bf Functoriality of the critical set]
\label{Crit(S)functoriality}
Let $M$ be a compact manifold without boundary, $F\rightarrow M$ a fiber bundle, $\CF$ its smooth set of sections and 
$S:\CF \rightarrow  y(\FR)$ a (smooth) action functional given by 
\vspace{-2mm}
\begin{align*}
	S = \int_M \circ \;\; \CL
\end{align*}  

\vspace{-1mm}
\noindent
for some local Lagrangian $\CL: \CF \rightarrow \Omega^{d}_\mathrm{Vert}(M)$.
Then the assignment
\vspace{-1mm}
\begin{align*}
\FR^k \longmapsto \mathrm{Crit}(S)(\FR^k)
\end{align*}

\vspace{-1mm}
\noindent  defines a smooth set, which coincides with the on-shell space of fields $\CF_{\CE \CL }$.
\begin{proof}
The smooth map $S: \CF \rightarrow y(\FR)$ is
defined as (see \eqref{LocalActionisSmooth})
\vspace{-2mm} 
$$
S=\int_{M}\circ \;\;  \CL \;\;:\;\; \CF \longrightarrow \Omega^{d}_{\mathrm{Vert}}(M)\longrightarrow y(\FR) 
$$

\vspace{-2mm} 
\noindent  by composing the smooth Lagrangian density map $\CL$ of Prop. \ref{LocalLagrangianisSmooth} corresponding to 
 $L:J^\infty_M F \rightarrow \wedge^d T^* M$,  with integration along $M$. Intuitively, derivatives only act along $M$, 
 and for each $\FR^k$-plot of fields the $\FR^k$-dependence is carried along by multiplication of functions on $\FR^k$.
 The standard integration by parts calculation may now be expressed rigorously through the infinity jet prolongation as follows.
 Firstly, by viewing the Lagrangian density bundle map locally as $L=\bar{L}(x^\mu,u^a_I) \cdot \dd x^1\cdots \dd x^d$, and working 
 along the lines of Ex. \ref{TangentVectorsAndDerivationsofLocalFunctions}, we have
 \vspace{-2mm} 
\begin{align*}
\partial_t (L\circ j^\infty \phi_t ) |_{t=0}&=(j^\infty \phi_0)^* 
\big( \iota_{\partial_t j^\infty \phi_t |_{t=0}} \dd_V L) \\
&= (j^\infty \phi_0)^* 
\Big(\iota_{\partial_t j^\infty \phi_t |_{t=0}} \delta_V L + \iota_{\partial_t j^\infty \phi_t |_{t=0}} \dd_H \theta_L \Big) 
\end{align*}

\vspace{-1mm} 
\noindent  with the latter equality being Eq. \eqref{LagrangianVerticalDifferentialDecomposition}. Next, following the description
of Lem. \ref{SourceFormsDefineDifferentialOperators}, the first term coincides with the natural pairing of the corresponding Euler--Lagrange 
differential operator $\mathrm{\CE \CL}(\phi_0)\in \Gamma_M(V^*F\otimes \wedge^d T^*M)$ evaluated at $\phi_0$,  with the tangent vector
$\partial_t \phi_t |_{t=0}\in \Gamma_M(VF)=T\CF(*)$ at $\phi_0\in \CF(*)$ as in Def. \ref{KinematicalTangentBundleToFieldSpace}. 
That is,
 $$
(j^\infty \phi)^* \big(\iota_{\partial_t j^\infty \phi_t |_{t=0}} \delta_V L \big) = \langle \mathcal{E}\mathcal{L}(\phi_0)\, , 
\, \partial_t \phi_t|_{t=0} \rangle \in \Omega^{d}(M) \, .
 $$
 Further, the commutation relation $[\iota_{\partial_t j^\infty \phi_t |_{t=0}},\dd_H ]=0 $ may be easily checked, for instance 
 in local coordinates. Along with the compatibility result of Lem. \ref{HorizontalDifferentialBaseDeRhamCompatibility}, 
 the second term then becomes an exact form on $M$ 
\begin{align*}
(j^\infty \phi)^* \big(\iota_{\partial_t j^\infty \phi_t |_{t=0}} \dd_H \theta_L) \big) 
&=- (j^\infty \phi)^*\big(\dd_H \iota_{\partial_t j^\infty \phi_t |_{t=0}} 
 \theta_L \big) \\
& = - \dd_M \big( (j^\infty \phi)^* \iota_{\partial_t j^\infty \phi_t |_{t=0}} 
 \theta_L \big) \in \Omega^d(M) \, . 
 \end{align*}
 Let us now make the useful, for the moment only notational, shorthand
\begin{align}\label{ThetaNotationalShorthand}
 \dd_M \theta_\CL\big(\partial_t \phi_t |_{t=0}\big):= \dd_M \big( (j^\infty \phi)^* \iota_{\partial_t j^\infty \phi_t |_{t=0}} 
 \theta_L \big) \, ,
 \end{align}
 whose underlying mathematical meaning will be made explicit in \cref{TheBicomplexOfLocalFormsSection}.
 
 With this at hand, we arrive at the more familiar form
 \vspace{-1mm} 
\begin{align}\label{VariationOfLagrangianDensityOnPoints}
\partial_t \CL (\phi_t ) |_{t=0}&=
\langle \mathcal{E}\mathcal{L}(\phi_0)\, , \, \partial_t \phi_t|_{t=0} \rangle -  \dd_M \theta_\CL\big(\partial_t \phi_t |_{t=0}\big) \quad \in \quad  \Omega^d(M)\, , 
\end{align}

\vspace{-1mm}
\noindent 
which in coordinates reads as the usual integration by parts algorithm. Note that here it is globally justified for arbitrary 
non-trivial fiber bundles and local Lagrangians. Composing with the integral, 
the variation of the action functional is thus given by
\vspace{-3mm} 
\begin{align*}
\partial_{t} S(\phi_t)|_{t=0}&=\int_M \partial_t \CL( \phi_t ) |_{t=0}\\
&= \int_M \big\langle \mathcal{E}\mathcal{L}(\phi_0)\, , \, \partial_t \phi_t|_{t=0} \big\rangle 
\;\; - \;\; \int_M  \dd_M \theta_\CL\big(\partial_t \phi_t |_{t=0}\big) \\
 &= \int_M \big\langle \mathcal{E}\mathcal{L}(\phi_0)\, , \, \partial_t \phi_t|_{t=0} \big\rangle \,,
\end{align*}
where the second term vanishes due to the compactness and boundary-less assumption on $M$. 

All the steps above follow through similarly on $\FR^k$-plots of sections, and so 
\vspace{-1mm} 
\begin{align}\label{VariationOfLagrangianDensityOnRkPoints}
\partial_t \big(L\circ j^\infty \phi^k_t \big)|_{t=0}
&=\big\langle \mathcal{E}\mathcal{L}(\phi^k_0)\, , \, \partial_t \phi^k_t|_{t=0} \big\rangle 
-  \dd_M \theta_\CL\big(\partial_t \phi^k_t |_{t=0}\big) \quad \in  \quad \Omega^d_\mathrm{Vert}(M)\big(\FR^k\big)\, 
\end{align}

\vspace{-1mm} 
\noindent  taking value in  
$\Omega^{d}_\mathrm{Vert}(M)\big(\FR^k\big)\cong \Omega^{d}(M)\hat{\otimes} C^\infty(\FR^k)$, with 
\vspace{-1mm} 
$$
\mathcal{E}\mathcal{L}(\phi^k_0)\in \mathbold{\Gamma}_M\big(\wedge^d T^*M\otimes V^*F\big)(\FR^k)
\cong \Gamma_M\big(\wedge^d T^*M\otimes V^*F\big) \hat{\otimes} C^\infty(\FR^k)
$$ 

\vspace{-1mm} 
\noindent  as in  \eqref{EulerLagrangerSmoothSetMap} and similarly $\partial_t \phi^k_t |_{t=0} \in T\CF(\FR^k)= 
\mathbold{\Gamma}_M(VF)(\FR^k)\cong \Gamma_M(VF)\hat{\otimes} C^\infty(\FR^k)$ 
as in Def. \ref{KinematicalTangentBundleToFieldSpace}. The pairing $\langle-,-\rangle$ is extended linearly by multiplying
the $C^\infty(\FR^k)$ components.\footnote{In other words, it defines a smooth map $\langle-,-\rangle: {\mathbold{\Gamma}}_M(\wedge^d T^*M \otimes V^* F)\times
{\mathbold{\Gamma}}_M(VF) \rightarrow \Omega^{d}_{\mathrm{Vert}}(M)$.}
Composing with integration, the second term once again vanishes, thus
$$
\partial_{t}S(\phi^k_t)|_{t=0} = \int_{M}\langle \mathcal{E}\mathcal{L} (\phi^k_0)\, , \,  \partial_t \phi^k_t |_{t=0} \rangle \in C^\infty(\FR^k)\, ,
$$
taking values in $C^\infty(\FR^k)$. By Lem. \ref{LinePlotsRepresentTangentVectors} 
 and its application to higher plots (see relation \eqref{PlotsOfTangentBundleViaCurvesofPlots} and its footnote), the equation 
 holds for all tangent vectors at $\phi^k_0\in \CF(\FR^k)$.  By the fundamental lemma 
of the calculus of variations  
 (see, e.g., \cite{Jost}),  the pairing 
	$\int_M \langle-,-\rangle \, \dd t$ is non-degenerate and so
 \vspace{-2mm} 
\begin{align}\label{ELcondition}
	\phi_0^k \in \mathrm{Crit}(S)(\FR^k)\iff \mathcal{E}\mathcal{L}(\phi^k_0)
 =0_{\phi^k_0} \in C^{\infty}(\FR^k)\hat{\otimes} \Gamma_M(\wedge^d T^*M \otimes V^*F)\, .
\end{align}

\vspace{-2mm} 
\noindent By verticality of $\CE \CL$, the latter condition is functorial (invariant) under pullbacks along maps
$f:\FR^{k'} \rightarrow \FR^{k}$ of plots, and so $\mathcal{E}\mathcal{L}\big(\phi^k_0\big)= 0_{\phi^k_0} \implies 
\mathcal{E}\mathcal{L}\big(f^* \phi^k_0\big)= 0_{f^*\phi^k_0}$. 
Hence if 
$\phi^k_0\in \mathrm{Crit}(S)(\FR^k)\subset \CF(\FR^k)$, then $f^*\phi^k_0\in \mathrm{Crit}(S)(\FR^{k'})\subset \CF(\FR^{k'})$ 
for all $f:\FR^{k'}\rightarrow \FR^k$ and the result follows.
\end{proof}
\end{proposition}

The criticality condition being equivalent to the Euler--Lagrange equations \eqref{ELcondition} implies a universal construction of the smooth
critical set of a local Lagrangian field theory. To that end, recall that $\wedge^d T^*M \otimes V^*F$ really stands 
for the tensor product $\pi_F^*\big(\wedge^d T^*M \big) \otimes_F V^*F$ of the pullback top exterior bundle over $F$, and hence 
has a zero-section $0_F : F\rightarrow \wedge^d T^*M \otimes V^*F$. Postcomposition of (plots of) sections of $F$, i.e., (plots of) 
fields in $\CF$, with the zero section $0_F$ yields the (non-constant) smooth map
\vspace{-2mm}
\begin{align}\label{VariationalCotangentZeroSection}
0_\CF \;:\; \CF &\longrightarrow \mathbold{\Gamma}_M(\wedge^d T^*M \otimes V^*F)
\\[-2pt] 
\phi^k &\longrightarrow 0_F \circ \phi^k\, . \nn 
\end{align}

\vspace{-2mm} 
\begin{corollary}[\bf  Critical smooth set as a pullback]\label{CriticalSmoothSetPullback}
$\,$

\noindent {\bf (i)} The bare set of critical points $\mathrm{Crit}(S)(*)$ is given equivalently by the pullback/intersection set
\vspace{-2mm} 
\[
\xymatrix@=1.6em  {\mathrm{Crit}(S)(*) \ar[d] \ar[rr] &&   \Gamma_{M}(F)\ar[d]^{\CE \CL} 
	\\ 
	\Gamma_{M}(F)\ar[rr]^-{0_{\Gamma_M(F)}}  && \Gamma_M(\wedge^d T^*M \otimes V^*F)
	\, . } 
\]
\vspace{-3mm} 
\noindent {\bf (ii)} The smooth critical set  $\mathrm{Crit}(S)$ coincides with the incarnation of this
diagram in the topos of smooth sets, i.e., 
\[
\xymatrix@=1.6em  {\mathrm{Crit}(S) \ar[d] \ar[rr] &&   \CF \ar[d]^{\mathcal{EL}} 
	\\ 
	\CF \ar[rr]^-{0_\CF}  && \mathbold{\Gamma}_M(\wedge^d T^*M \otimes V^*F)
	\, . } 
\]
\end{corollary}
 
\begin{proof}
The first statement can easily be checked via \eqref{ELcondition}, while the second 
follows using the (unique) smooth extension of the Euler--Lagrange differential operator from \eqref{EulerLagrangerSmoothSetMap}.   
\end{proof}

\begin{remark}[\bf Subspace diffeology] We should note that for the concrete, i.e., diffeological, field spaces appearing in \textit{bosonic} field theory, 
the smooth structure on $\mathrm{Crit}(S)$ we have defined coincides with the so-called \textit{subspace diffeology}. This point of view, however, 
does not generalize to non-concrete field spaces - such as those appearing in \textit{fermionic} field theory. On the contrary,
the smooth set viewpoint we have laid out will directly generalize when considering infinitesimally thickened smooth bosonic field spaces 
and even smooth supersets of fermions and bosons, while at the same time paralleling the discussions of fermionic on-shell spaces in the 
physics literature \cite{GSS-2}.
\end{remark}

\begin{example}[\bf Field theory in zero dimension]\label{FiniteDimensionalCriticalSmoothSet}
	Consider the case of zero-dimensional field theory, i.e., $M=*$ and $F=N\times *$. Then $\CF\cong y(N)$, and the above calculation 
 reduces to the condition
   \vspace{-2mm} 
	$$
 \partial_t S(\phi^k_t)\big|_{t=0}=\; \dd^{N}_{\phi^k_0}\, S \big(\dot{\phi}^k(0)\big)=0\;,
 $$

  \vspace{-2mm} 
\noindent where the right-hand side is interpreted, for each $x\in \FR^k$, as the pushforward of a tangent vector 
$\dot{\phi}^k(0,x)\in T_{\phi^k_0(x)}N$  under the differential of $S$ at the point $\phi^k_{0}(x)\in N$. 
Equivalently, this is the condition that the 1-form $\dd_N S\in \Omega^1(N) \cong \Gamma_N(T^*N)$,
 viewed as a section of the cotangent bundle, vanishes at the point $\phi^k_0(x)\in N$, for every $x\in \FR^k$. 
 Since this condition holds for each point in the image of the $\phi^k_0:\FR^k\rightarrow N$, for all tangent vectors 
 at each point, it follows it is invariant under pullbacks along $f:\FR^{k'}\rightarrow \FR^k$. Yet equivalently and 
 more concisely, the critical smooth set of $S:y(M)\rightarrow y(\FR)$ is the intersection/pullback diagram in $\SmoothSets$,
   \vspace{-1mm} 
 \[
\xymatrix@=1.6em  {\mathrm{Crit}(S) \ar[d] \ar[rr] &&   y(N)\ar[d]^{\dd_N S} 
	\\ 
	y(N)\ar[rr]^-{0}  && y(T^*N)
	\, . } 
\]

  \vspace{-2mm} 
\noindent Note that the above critical locus might not have a finite-dimensional smooth manifold structure \footnote{That is, in the case where
the intersection is non-transversal.} -- i.e., the pullback diagram might not exist in $\SmoothManifolds$ -- but it does have a natural 
smooth set structure since \textit{any} such intersection/pullback diagram exists in $\SmoothSets$. \footnote{Depending on one's goals, 
 this smooth structure might not be sufficient, and one might have to choose a different category to compute the intersection. 
 For instance, appropriate choices include thickened smooth sets \cite{GSS-2}, which do detect intersection multiplicities in 
 contrast to smooth sets, or more generally their homotopical versions as in derived geometry (see \cite{Carchedi}).}
\end{example}

Comparing with the finite-dimensional example above, we see that the smooth set 
$\mathbold{\Gamma}_M(\wedge^d T^*M \otimes V^*F)$ serves as a substitute of the cotangent bundle 
for the field space $\CF$ and \textit{local} action functionals.
\begin{definition}[\bf  Variational cotangent bundle]\label{LocalCotangentBundleToFieldSpace}
The \textit{local or variational cotangent bundle} $\pi_\CF: T^*_{\mathrm{var}} \CF\rightarrow \CF$ of a smooth field space 
$\CF=\mathbold{\Gamma}_M(F)\in \SmoothSets$ is defined as the smooth vector bundle 
\vspace{-3mm} 
\[
\xymatrix@=1em  { \mathbold{\Gamma}_M(\wedge^d T^*M \otimes V^*F)  \ar[d] & 
	\\ 
\mathbold{\Gamma}_M(F)
	\, , } 
\]

\vspace{-2mm} 
\noindent where the projection $\pi_\CF$ is given by postcomposition of (plots of) sections of $\wedge^d T^*M \otimes V^*F$
with the (manifold) vector bundle projection $\pi_F: \wedge^d T^*M \otimes V^*F\rightarrow F$. 
\end{definition}

 In this sense, the map $0_\CF: \CF \rightarrow T^*\CF$ of Eq. \eqref{VariationalCotangentZeroSection} is really the canonical zero-section 
 of the smooth set vector bundle $\pi_\CF: T^*\CF \rightarrow \CF$ over the off-shell space of fields. Similarly the Euler--Lagrange 
 operator $\CE \CL : \CF \rightarrow T^* \CF$ is another such section, i.e.,
 $\pi_\CF\circ  \CE \CL (\phi) = \pi_F \circ E L \circ j^\infty \phi= \pi_0 \circ j^\infty \phi = \phi$ as
 $EL:J^\infty_M F\rightarrow \wedge^d T^*M \otimes V^*F $ is a bundle map over $F$ (Lem. \ref{SourceFormsDefineDifferentialOperators}).
 
\begin{remark}[\bf Non-variational field theories] Recall, by Lem. \ref{SourceFormsDefineDifferentialOperators}, that
there exist differential operators
$P:\Gamma_M(F) \rightarrow \Gamma_M(\wedge^d T^*M \otimes V^*F)$, inducing smooth sections 
$\CP:\mathbold{\Gamma}_M(F) \rightarrow \mathbold{\Gamma}_M(\wedge^d T^*M \otimes V^*F)$ of the variational cotangent bundle, which arise from source forms
$\Omega^{d,1}_s(J^\infty_M F)$ which cannot be variational according to Lem. \ref{EulerLagrangeComplexCohomology}. 
In other words, differential operators of this form cannot express the criticality of some action functional, and so represent local 
field theories which are \textit{non-variational}. 
Nevertheless, we may still consider the smooth set of solutions of the partial differential equation $\CP=0$, 
exactly as the pullback of Cor. \ref{CriticalSmoothSetPullback}.
\end{remark}

In the case of a \textit{noncompact spacetime} $M$ the action is not necessarily defined on the whole of field space $\CF$, since fields of
noncompact support may yield noncompactly supported densities under $\CL:\CF\rightarrow \Omega^{d}_{\mathrm{Vert}}(M)$ -- which are \textit{not} 
integrable. One could restrict to compactly supported fields $\CF_{c}\hookrightarrow \CF$, in which case the action is well defined 
and the criticality condition, along with the functoriality result, follow verbatim. 
However, fields of non-compact support are considered physically viable field configurations -- and so this option is too restrictive. 

\medskip 
The proper formulation of the criticality condition for a generic local Lagrangian $\CL$ on a noncompact spacetime $M$ is \textit{defined} by 
the joint criticality condition of all its charges (Def. \ref{ChargesOnFieldSpace}). Indeed, notice that in the previously examined case of 
a compact spacetime without boundary, the well-defined action $S= \int_M \circ \CL$ may be thought of as the charge $\CL_M$ of the local $d$-form 
current $\CL$, over the compact submanifold $M$. For a general spacetime, the integral does not exist over the whole of $M$, 
but it does exist over compact $d$-dimensional submanifolds - i.e., the charges
\vspace{-2mm}
$$
S_{\Sigma^d}\equiv \CL_{\Sigma^d}:= \int_{\Sigma^d} \circ \; \CL \;\;:\;\; \CF \longrightarrow   y(\FR)
$$

\vspace{-2mm}
\noindent exist for all compact submanifolds $\Sigma^d\rightarrow M$. The criticality condition (Def. \ref{CriticalRkPoints}) may be consistently 
defined for each of the charges, with an appropriate modification arising that takes care of the potential boundary of the submanifold. 
That is, we only consider $1$-parameter variations that fix the value of the jet of the (plots of) field at the boundary -- viewed as 
a `boundary condition' for the charge.  

\begin{definition}[\bf  Critical plots of charge]
\label{CriticalRkPointsofCharge}
Let $S_{\Sigma^d}:\CF \rightarrow y(\FR)$ be the charge/action of a local Lagrangian $\CL$ over a $d$-dimensional compact 
submanifold $\Sigma^d\hookrightarrow M$. The \textit{critical $\FR^k$-plots} of $S_{\Sigma^d}$ is the subset of $\FR^k$-plots
\vspace{-2mm}
\begin{align}
		\mathrm{Crit}(S_{\Sigma^d})(\FR^k):=  \Big\{\phi^k\in \CF(\FR^k) \; \big{|}\; \partial_t S_{\Sigma^d} (\phi^k_t)|_{t=0}=0, \hspace{0.3cm}
  \forall \, \phi^k_t\in  \CF_{\phi^k,\, \partial \Sigma^d}(\FR^k\times \FR^1_t) \Big\} ,		
\end{align}

\vspace{-2mm}
\noindent
where 
\vspace{-2mm}
$$
\CF_{\phi^k,\, \partial \Sigma^d}(\FR^k\times \FR^1_t) = \Big\{\phi^k_t \in \CF(\FR^k\times \FR^1_t) \, \big{|} \, \phi^k_{t=0}= 
\phi^k \in \CF(\FR^k) \, \, \, \mathrm{and} \, \, \, j^\infty\phi^k_{t=t_0} |_{\partial \Sigma^d} 
= j^\infty \phi^k |_{\partial \Sigma^d}, \hspace{0.3cm} \forall \, t_0\in \FR^1_t \Big\}
$$

\vspace{-2mm}
\noindent
is the set of $1$-parameter families of $\FR^k$-plots of fields that start at $\phi^k$ and have constant jet along the boundary 
of $\Sigma^d$.
\end{definition}
In less technical terms, the criticality condition for the charges employs 1-parameter variations which keep all the `derivatives' 
of the field fixed along the boundary of the given submanifold.\footnote{In other words, the value of the field is fixed in an `infinitesimal 
neighborhood' around $\partial \Sigma^d$. This will be made into a rigorous definition in \cite{GSS-2}.} Such variations include, 
in particular, variations that are constant in some open neighborhood of the boundary (i.e., in germs around the boundary), 
as most often used in the rigorous (point-set) description of the action principle (see \cite{Chris}).
 Our definition will necessarily yield the same critical set, but is more suited to make contact with the notion of boundary 
 conditions in the case of a spacetime with boundary (see Remark \ref{Rem-boundary} below).

\begin{definition}[\bf  Critical $\FR^k$-plots of Lagrangian]
\label{CriticalRkPointsofLagrangian}
Let $\CL:\CF\rightarrow \Omega^d_\mathrm{Vert}(M)$ be a local Lagrangian over a (potentially) non-compact d-dimensional spacetime $M$
without boundary. The \textit{critical $\FR^k$-plots} of $\CL$ is the subset of $\FR^k$-plots
\vspace{-2mm}
\begin{align}
		\mathrm{Crit}(\CL)(\FR^k):=  \Big\{\phi^k\in \CF(\FR^k) \; \big{|}\; 
  \phi^k\in \mathrm{Crit}(S_{\Sigma^d})(\FR^k), \hspace{0.3cm} \forall \; 
  \mathrm{compact \hspace{0.1cm} submanifolds \hspace{0.1cm}} \Sigma^d\hookrightarrow M \Big\} .		
\end{align}
\end{definition}
With these definitions at hand, and by the (petit) sheaf theoretic nature of the fields space (Rem. \ref{FieldSpaceAsASheafOfSheaves}), 
the result of Prop. \ref{Crit(S)functoriality} readily extends to the case of non-compact spacetimes $M$.

\begin{proposition}[\bf Crit($\CL$) functoriality]
\label{Crit(L)functoriality}	Let $M$ be a noncompact manifold without boundary, $F\rightarrow M$ a fiber bundle, $\CF$ its smooth set
of sections and $\CL: \CF=\CF \rightarrow \Omega^{d}_\mathrm{Vert}(M)$ a (smooth) local Lagrangian density given by 
  \vspace{-1mm} 
$$
\CL = L\circ j^\infty 
$$ 

  \vspace{-1mm} 
\noindent  for some Lagrangian density bundle map $L:J^\infty_M F \rightarrow \wedge^d T^* M$. Then the
assignment $\FR^k \mapsto \mathrm{Crit}(\CL)(\FR^k)$ defines a smooth set, which is equivalently given by the pullback diagram of smooth sets
  \vspace{-2mm} 	
 \[
	\xymatrix@=1.6em  {\mathrm{Crit}(\CL) \ar[d] \ar[rr] &&   \CF \ar[d]^{\mathcal{EL}} 
		\\ 
		\CF\ar[rr]^-{0_\CF}  && T^*_\mathrm{var} \CF
	\, , }
	\]

   \vspace{-2mm} 
\noindent 	where $\mathcal{EL}$ is the (smooth) Euler--Lagrange differential operator \eqref{EulerLagrangerSmoothSetMap} of $\CL$\,.
	\begin{proof}
	The proof follows exactly as in Prop. \ref{Crit(S)functoriality} up to Eq. \eqref{VariationOfLagrangianDensityOnRkPoints}
   \vspace{-1mm} 
 \begin{align*}
\partial_t \big(L\circ j^\infty \phi^k_t \big) \big|_{t=0}
&=\big\langle \CE \CL(\phi^k_0)\, , \, \partial_t \phi^k_t|_{t=0} \big\rangle 
\; -\;  
 \dd_M \theta_\CL\big(\partial_t \phi^k_t |_{t=0}\big)\quad \in \quad  \Omega^d_\mathrm{Vert}(M)\big(\FR^k\big)\, .
\end{align*}

  \vspace{-1mm} 
\noindent  The criticality condition is obtained by demanding the criticality of each of the charges of $\CL$. 
Integrating over any compact $\Sigma^d\hookrightarrow M$, and for any variation $\phi_t$ as in Def. \ref{CriticalRkPointsofCharge}, we have
  \vspace{-1mm} 
 \begin{align*}
\partial_{t} S_{\Sigma^d}(\phi^k_t)|_{t=0}&= 
\int_{\Sigma^d} \big\langle \mathcal{E}\mathcal{L}(\phi^k_0)\, , \, \partial_t \phi^k_t|_{t=0} \big\rangle 
\; - \; \int_{\Sigma^d} \dd_M \theta_\CL\big(\partial_t \phi^k_t |_{t=0}\big)\\
 &= \int_{\Sigma^d} \big\langle \mathcal{E}\mathcal{L}(\phi^k_0)\, , \, \partial_t \phi^k_t|_{t=0} \big\rangle
 \; - \; \int_{\partial \Sigma^d} (j^\infty \phi^k_0)^* \iota_{\partial_t j^\infty \phi^k_t |_{t=0}} 
 \theta_L \\
 &= \int_{\Sigma^d} \big\langle \mathcal{E}\mathcal{L}(\phi^k_0)\, , \, \partial_t \phi^k_t|_{t=0} \big\rangle 
\end{align*}

\vspace{-1mm}
\noindent 
where the boundary term in this case vanishes since $j^\infty \phi^k_t |_{\partial \Sigma^d}$ is constant in $t$ by assumption,
and hence the corresponding prolongated tangent vector $\partial_t j^\infty \phi^k_t |_{t=0}:M\rightarrow VJ^\infty F$ vanishes 
at the boundary $\partial \Sigma^d$. By Lem. \ref{LinePlotsRepresentTangentVectors} and its application to higher plots 
(see Eq. \eqref{PlotsOfTangentBundleViaCurvesofPlots} and its footnote), 
this holds for all tangent vectors $\partial_t\phi^k_t |_{t=0}= X_{\phi_0^k,\Sigma^d}\in T_{\phi^k_0}\CF$ 
whose support lies in the interior of $\Sigma^d$.
\footnote{In more detail, this entails that the tangent vector $X_{\phi_0,\Sigma^d}:M \rightarrow VF$ at $\phi_0\in \CF(*)$ is such that 
$ X_{\phi_0,\Sigma^d}(p)= 0 \in V_{\phi_0(p)}F $ if $p\notin \mathrm{int}(\Sigma^d)$, and similarly for tangent vectors over $\FR^k$-plots.}
Since this holds for \textit{all} such tangent vectors, it follows by the Fundamental 
Lemma of variational calculus that 
\vspace{-2mm}
\begin{align*}
	\phi_0^k \in \mathrm{Crit}(S_{\Sigma^d})(\FR^k)  \quad \iff \quad  \mathcal{E}\mathcal{L}(\phi^k_0)\big|_{\mathrm{int}(\Sigma^d)}
 =0_{\phi^k_0} \big|_{\mathrm{int}(\Sigma^d)} \;\; \in \; C^{\infty}(\FR^k)\hat{\otimes} \Gamma_{\mathrm{int} (\Sigma^d)}(\wedge^d T^*M \otimes V^*F)\, ,
\end{align*}

\vspace{-1mm}
\noindent 
i.e., $\phi_0^k$ is a critical plot of the charge over $\Sigma^d$ if and only if it satisfies the Euler--Lagrange equations on 
the interior of $\Sigma^d$.

Now a plot $\phi^k_0$ is a critical plot for $\CL$ (Def. \ref{CriticalRkPointsofLagrangian}) if it is a critical plot of the
corresponding charge \textit{for all} such $\Sigma^d\hookrightarrow M$, and hence if and only if 
$$
\CE \CL(\phi_0^k) \big|_{\mathrm{int}(\Sigma^d)}=0\big|_{\mathrm{int}(\Sigma^d)} 
$$
\textit{for all} compact submanifolds $\Sigma^d\hookrightarrow M$. Taking the collection of (closed) balls $\bar{B}_p^d\hookrightarrow M$ 
around any point $p\in M$, we may cover the spacetime $M$ with compact submanifolds. Hence $\phi^k_0$ is a critical plot of $\CL$ 
if any only if it satisfies the Euler--Lagrange equations on an open neighborhood $B_p^d\hookrightarrow M$ of every point. 
Since the Euler--Lagrange differential operator $\CE \CL$ is furthermore a map of sheaves in the petit sense of 
Rem. \ref{FieldSpaceAsASheafOfSheaves}, $\phi^k_0$ satisfies the Euler--Lagrange equations \textit{locally} if and only if 
it satisfies them \textit{globally}. In other words
\vspace{-1mm}
\begin{align*}
	\phi_0^k \in \mathrm{Crit}(\CL)(\FR^k) \quad \iff  \quad 
 \mathcal{E}\mathcal{L}(\phi^k_0)=0_{\phi^k_0} \;\; \in \; C^{\infty}(\FR^k)\hat{\otimes} \Gamma_M(\wedge^d T^*M \otimes V^*F)\, ,
\end{align*}

\vspace{-1mm}
\noindent 
and so the rest of the claim follows as in Cor. \ref{CriticalSmoothSetPullback}.
 \end{proof}
\end{proposition}
 
\begin{remark}[{\bf Criticality on spacetimes with boundary}]\label{Rem-boundary}
Field theories on spacetimes with boundary require further care, both in the above mathematical treatment and the corresponding physical interpretation. 
We briefly summarize the relevant cases at the level of fields (with general plots following similarly), without expanding on the fine details here.

\smallskip 
\noindent {\bf (i)} Consider first the case of a local field theory $(\CF,\CL)$ on a compact spacetime $M$ with (connected) boundary
$\partial M \neq \emptyset $. In this case, one may consistently compose with integration over $M$ over the full field space, and 
the local action functional is well-defined 
 $S=\int_M \circ \CL : \CF \rightarrow y(\FR)$. Naively, one has two similar but different definitions for the criticality of the action functional.
 As we will see, each of these corresponds to a different choice of \textit{boundary conditions}.

 {\bf (a)} We may use Def. \ref{CriticalRkPoints} which has no restrictions on the variations along the boundary $\partial M$, in which case the criticality is equivalent to the vanishing of \textit{both} the Euler--Lagrange equations $\CE \CL(\phi)|_{\mathrm{int}(M)}= 0_{\phi} |_{\mathrm{int}(M)}$ on the interior of $M$ \textit{and} of the integral
\vspace{-1mm}
$$\int_{\partial M} \theta_\CL\big(\partial_t \phi_t |_{t=0}\big) := \int_{\partial M} (j^\infty \phi)^* \iota_{\partial_t j^\infty \phi_t |_{t=0}} 
 \;  \theta_L \, 
$$ 

\vspace{0mm} 
\noindent along the boundary, for all 1-parameter variations. The vanishing of the latter is equivalent to a set of conditions relating the field 
and its jets along the boundary -- which depend on the explicit form of Lagrangian and in particular the `boundary term' $ \theta_L$. 
Since this term is defined up to the addition of $\dd_H$-exact $(d-1,1)$-forms on $J^\infty_M$, this means that \textit{in the presence of a boundary} 
its \text{choice} in the decomposition $\dd_V L = EL + \dd_H \theta$ should actually be considered as data entering the definition of the field theory $(\CF,\CL)$.
Such boundary conditions are called ``\textit{natural}'' or ``\textit{Neumann}'' boundary conditions, and may also be expressed as the 
vanishing of a smooth map $\mathrm{Nm}(\phi|_{\partial M})=0$. The critical smooth set is then again a pullback, given by the intersection
of those fields that satisfy the Euler--Lagrange equations on the interior $\CF_{\CE \CL \circ |_{\mathrm{int}(M)}=0}\hookrightarrow \CF$ 
and those that satisfy the Neumann boundary conditions $\CF_{\mathrm{Nm}\circ |_{\partial M}=0 }\hookrightarrow \CF$,
 
 \vspace{-2mm}
 \[
\xymatrix@R=.9em@C=1.6em  {\mathrm{Crit}^{\mathrm{Nm}}(S) \ar[d] \ar[rr] &&   \CF_{\mathrm{Nm}\circ |_{\partial M}=0 } \ar[d]
	\\ 
	\CF_{\CE \CL \circ |_{\mathrm{int}(M)}=0}\ar[rr]  && \CF
	\, . } 
\]
 
 \vspace{-2mm}
 \noindent We note that each of the smooth subspaces $\CF_{\CE \CL \circ |_{\mathrm{int}(M)}=0}$ and $\CF_{\mathrm{Nm}\circ |_{\partial M}=0 }$ are
 themselves a pullback, respectively, given by the vanishing of the corresponding smooth maps.

 {\bf(b)} We may instead use Def. \ref{CriticalRkPointsofCharge} applied on $M$ itself, viewed tautologically as a compact submanifold of itself,
 which restricts the (jets of) variations to vanish along the boundary. As we will see, this is in fact a special case of Def. \ref{CriticalRkPoints},
 applied however to a smooth subspace of the smooth field space $\CF$. The interpretation in this case is that one has made a $\textit{choice}$ of 
 boundary values for the dynamical fields $\phi$ and their derivatives. \footnote{More precisely if the resulting Euler--Lagrange operator is of order $k$, 
 the boundary jet values are fixed only up to order $k-1$. This is to avoid overdetermining the resulting PDE, and (hopefully) resulting 
 in a well-posed boundary value problem. The corresponding differences in the following discussion are simply in notation. }
Indeed, let $ \psi : \partial M \rightarrow J^\infty_M F$ be a smooth section over the boundary (i.e., covering $\partial M\hookrightarrow M$), 
thought of as the fixed values of the (jets) of the fields we consider (see also Def. \ref{SectionsOnInfinitesimalNeighborhoodOfSubmanifold}). 
Such boundary conditions on the fields are called ``\textit{imposed}'' or ``\textit{Dirichlet}'' boundary conditions. In mathematical terms, 
this means we are not actually considering the smooth action functional $S:\CF\rightarrow y(\FR)$ on the full space of fields, 
but instead its restriction on the smooth subspace
$$
\CF^{\mathrm{Dir}}_{\psi} \longhookrightarrow \CF 
$$
with $\FR^k$-plots of constant infinite jet value 
$j^\infty \psi$ on the boundary, 
$$
\CF^{\mathrm{Dir}}_{ \psi}(\FR^k)=\big\{\phi^k \in \CF(\FR^k) \; \big|\; j^\infty \phi^k|_{\partial M}(x,-)=  \psi (-) \hspace{0.2cm} \forall x\in \FR^k \big\}\,.
$$
It is now immediate that applying the general Def. \ref{CriticalRkPoints} on the smooth action functional $S:\CF^{\mathrm{Dir}}_{ \psi} \rightarrow y(\FR)$ 
immediately recovers exactly that of Def. \ref{CriticalRkPointsofCharge} applied on $S:\CF\rightarrow y(\FR)$ as a charge on $M$. In this case, 
the boundary term appearing in the variation of the functional vanishes by \textit{definition}, and so criticality is equivalent to the vanishing
of the Euler--Lagrange on the interior $\CE \CL (\phi)|_{\mathrm{int}(M)}=0_\phi |_{\mathrm{int}(\Sigma^d)}$. The critical smooth set is yet again given by a pullback,
\vspace{-2mm}
 \[
\xymatrix@R=1.6em@C=3em  {\mathrm{Crit}^{\mathrm{Dir}}_{ \psi}(S) \ar[d] \ar[rr] &&   \CF^{\mathrm{Dir}}_{ \psi} \ar[d]^{\CE \CL \circ |_{\mathrm{int}(M)}}
	\\ 
	\CF^{\mathrm{Dir}}_{ \psi}\ar[rr]^-{0_\CF \circ |_{\mathrm{int}(M)}}  && T^*_{\mathrm{var}} \CF |_{\mathrm{int}(M)}
	\, . } 
\]
 
 \vspace{-2mm}
 \noindent
One can consider all possible Dirichlet boundary values simultaneously, and hence demand instead the criticality over the smooth subspace 
$$
\CF^{\mathrm{Dir}}\longhookrightarrow \CF
$$
with $\FR^k$-plots of constant value on the boundary $\CF^{\mathrm{Dir}}(\FR^k) = \underset{ \psi :\partial M \rightarrow J^\infty_M F}{\coprod} 
\CF_{ \psi}^\mathrm{Dir}(\FR^k)$. This is a smooth subspace of the space of field $\CF$ which has the \textit{same points}
$\CF^{\mathrm{Dir}}(*)=\CF(*)$, but \textit{different higher} plots (and in particular less of them). The critical smooth set is given analogously by 
\vspace{-2mm}
 \[
\xymatrix@R=1.6em@C=3em   {\mathrm{Crit}^{\mathrm{Dir}}(S) \ar[d] \ar[rr] &&   \CF^{\mathrm{Dir}}\ar[d]^{\CE \CL \circ |_{\mathrm{int}(M)}}
	\\ 
	\CF^{\mathrm{Dir}}\ar[rr]^-{0_\CF\circ |_{\mathrm{int}(M)}}  && T^*_{\mathrm{var}}\CF |_{\mathrm{int}(M)}
	\, , } 
\]
 
 \vspace{-2mm}
 \noindent by which it also follows that $\mathrm{Crit}^{\mathrm{Dir}}(S)(\FR^k) = \underset{\psi}{\coprod} \mathrm{Crit}_{ \psi}^\mathrm{Dir}(S)(\FR^k)$. 
 We note that each of the smooth subspaces  $\CF_{ \psi}^\mathrm{Dir}(\FR^k)$ and $\CF^{\mathrm{Dir}}$ may also be seen as pullbacks, respectively.

\vspace{1mm} 
\noindent {\bf (ii)} If $M$ is noncompact with (connected) boundary $\partial M\neq \emptyset$, the criticality condition 
is defined locally by considering compact submanifolds as in Def. \ref{CriticalRkPointsofLagrangian}. However, there is
the following important distinction due to the appearance of the \textit{ambient} boundary. For submanifolds contained
in the \textit{interior} of $M$, one uses Def. \ref{CriticalRkPointsofCharge}, while for submanifolds intersecting the
\textit{ambient} boundary $\partial M$, there is a choice of boundary conditions -- as in the compact case. The resulting 
smooth critical sets follow the above description, analogously. One way to think of the distinction between the ambient 
and interior boundaries is as follows: For any point $p$ of a boundary $\partial \Sigma^d \subset \mathrm{int}(M)$ lying 
in the interior of the spacetime, there necessarily exists a closed ball $\bar{B}_p^d$ around it with $B_p^d \subset \mathrm{int}(M)$. 
For a given field configuration $\phi$, the criticality of the action on $\bar{B}_p^d$ implies that $\phi$ satisfies the
Euler--Lagrange equations at $p$. This condition is then to be viewed as a Dirichlet boundary condition for the criticality
of the charge over $\Sigma^d$, hence the restriction of the (jets of) the variations Def. \ref{CriticalRkPointsofCharge}. 
Crucially, this line of thought breaks down for points in the ambient boundary $p\in \partial M$, allowing for 
a \textit{choice} of boundary condition as outlined above. 

\vspace{1mm} 
\noindent {\bf (iii)} If the boundary has several connected components $\partial M= \coprod_{i\in I} N_i $, one may choose different 
boundary conditions independently for each component, and the description of the smooth critical sets is modified accordingly.

\vspace{1mm} 
\noindent {\bf (iv)} For any of the cases above, and any choice of boundary condition, it is not guaranteed that the resulting boundary
value problem will have a solution, and even if it does, it may not be unique. Such questions are within the realm of analysis of PDEs,
which falls outside the scope of this manuscript.
\end{remark} 

\begin{remark}[\bf Criticality via moduli space of 1-forms]\label{CriticalityViaModuliSpaceOf1-forms}
We may reformulate the discussion of criticality condition (Def. \ref{CriticalRkPoints}) in a slightly more functorial way, by employing
smooth 1-forms on $\CF$ in the sense of Def. \ref{nformsonSmoothSet}. Unfortunately, the situation is more convoluted than the naive 
interpretation of the notation suggests. Recall (Def. \ref{DifferentialofSmoothMap})
the de Rham differential 1-form $\dd S\in \Omega^1(\CF)$  of a smooth map $S:\CF\rightarrow y(\FR)$ 
may be defined as the composition
\vspace{-2mm} 
\begin{align*}
	\dd S \;:\; \CF \xlongrightarrow{\;\; S\;\;} y(\FR)\cong \mathbold{\Omega}_{\mathrm{dR}}^0 
 \xlongrightarrow{\;\; \dd \;\;} \mathbold{\Omega}_{\mathrm{dR}}^1 \, .
\end{align*}

\vspace{-2mm} 
\noindent
In particular, when evaluating on $\FR^1$-plots, this takes the form 
\vspace{-2mm} 
\begin{align*}
	\CF(\FR^1)&\longrightarrow \Omega^1(\FR^1) \nn \\
	\phi_t &\longmapsto {\partial_t} S(\phi_t) \dd t ,
\end{align*}

\vspace{-2mm} 
\noindent
naturally encoding the partial derivative of  \eqref{PartialDerivativeOfActionOnFR^1plots}. This is naturally extended 
to a smooth map out of the path space $\mathbold{P}(\CF)=[\FR^1,\CF]$ of Ex. \ref{PathSpaceOfFields} 
\vspace{-4mm} 
\begin{align*}
	\dd_t S
      \;:\; 
    \mathbold{P}(\CF) 
    &\xlongrightarrow{\quad S \quad} 
    \mathbold{\Omega}_{\mathrm{dR, \mathrm{Vert}}}^0
    (\FR^1_t) \xlongrightarrow{\quad \dd_t \quad} \mathbold{\Omega}_{\mathrm{dR},\mathrm{Vert}}^1(\FR^1_t) 
 \\
    \phi^k_t & \quad \longmapsto  \quad S(\phi^k_t) \qquad \quad 
   \longmapsto \quad \partial_t S(\phi^k_t) \cdot \dd t \, ,
\end{align*}

\vspace{-2mm} 
\noindent
where we have used 
$ 
  [\FR^1_t,\FR]
    \cong 
  y(\FR)\hat{\otimes} 
  C^\infty(\FR^1_t) 
  \cong 
  \mathbold{\Omega}_{\mathrm{dR}, \mathrm{Vert}}^{0}(\FR^1_t)
$. 
The smooth (vertical) de Rham differential 
$\dd_t : \mathbold{\Omega}_{\mathrm{Vert}, \mathrm{dR}}^0(\FR^1_t)
\rightarrow \mathbold{\Omega }_{\mathrm{Vert},\mathrm{dR}}^1(\FR^1_t)\cong y(\FR)\hat{\otimes}\Omega^1(\FR^1_t)
$ 
here coincides with that of  \eqref{SmoothVerticaldeRham}. This composition naturally encodes all the partial derivatives along 1-parameter
families of $\FR^k$-plots appearing in Def. \ref{CriticalRkPoints}. However, this is somewhat misleading since neither of these maps is 
exactly what is denoted as the variation `$\delta S$' in the physics literature. Indeed :

\vspace{1mm} 
\noindent {\bf (i)} The `variation $\delta S$' of the smooth function $S:\CF\rightarrow y(\FR)$, viewed as acting on paths of (plots) of fields,
may be encoded  by further composing with the smooth evaluation map
\vspace{-2mm}
\begin{align*}
\mathrm{ev}_{t=0} \;:\; \Omega^1_{\mathrm{Vert}}(\FR^1_t) &\longrightarrow y(\FR) \otimes(T^*_0\FR^1) \cong y(\FR)\\
f_t^k\cdot \dd t &\longmapsto f_{t=0}^k\cdot \dd t |_{t=0}\, ,
\end{align*}
\vspace{-2mm}
and so via the smooth map
\begin{align*}
\mathrm{ev}_{t=0} \circ \dd_t S \;:\; \mathbold{P}(\CF) &\longrightarrow y(\FR)\otimes T^*_0\FR^1\cong y(\FR) \\
\phi^k_t &\longmapsto \partial_t S(\phi^k_t) |_{t=0}\cdot \dd t|_{t=0}\, .
\end{align*}

\vspace{-1mm} 
\noindent {\bf (ii)} Comparison with Def. \ref{CriticalRkPoints} shows that the critical $\FR^k$-plots are \textit{not} given by those that vanish 
under $\dd S$,  $\dd_t S$ or $\mathrm{ev}_{t=0}\circ \dd _t$. The latter of the three encodes the variation $\delta S$, but the criticality 
vanishing condition is on $\FR^k$-plots of the actual field space $\CF$ and \textit{not} of $\mathbold{P}(\CF)$.  Direct inspection gives  
\begin{align*}
		\mathrm{Crit}(S)(\FR^k)\cong  \Big\{\phi^k\in \CF(\FR^k) \; \big{|}\; \delta_{\phi^k_t} S = \dd_t S (\phi^k_t)|_{t=0}=0 , \hspace{0.3cm}
  \forall \, \phi^k_t\in  \CF_{\phi^k}(\FR^k\times \FR^1_t) \Big\} \, .		
\end{align*}

\vspace{1mm} 
\noindent {\bf (iii)} For a generic smooth map $S:\CF\rightarrow y(\FR)$, by Ex. \ref{TangentVectorsPathsOfFieldsAndDerivations}, the composed map 
\begin{align*}
\mathrm{ev}_{t=0}\circ \dd_t S \; :\;  \mathbold{P}(\CF) \longrightarrow  y(\FR)
\end{align*}
may not necessarily depend \textit{only} on the tangent vector at $\phi^k$ corresponding to each $\phi^k_t$, i.e., does not necessarily factor through
$\mathbold{P}(\CF) \rightarrow T\CF$ of Ex. \ref{PathSpaceOfFields}. For a local function S, e.g.
$S=\CP_{\Sigma^p}= \int_{\Sigma^p} P\circ j^\infty \in C^\infty_\mathrm{loc}(\CF)$, for which the above derivative does depend only on the 
corresponding tangent vector (Ex. \ref{TangentVectorsAndDerivationsofLocalFunctions}), then the variation does (uniquely) smoothly 
factor through the tangent bundle 
$$
\mathrm{ev}_{t=0}\circ \dd_t S \;:\; \mathbold{P}(\CF) \longrightarrow T\CF \longrightarrow  y(\FR) \, ,
$$
in which case it can be thought of as a $1$-form in the traditional sense of a real-valued map out of $T \CF$. Indeed, this can be made much more explicit
via the use of the local bicomplex of \cref{TheBicomplexOfLocalFormsSection} (see Lem. \ref{LocalDiffFormsAsDeRhamForms} and
Eq. \eqref{VerticalTransgressedDifferentialOfAction}). In the synthetic setting of \cite{GSS-2}, 
this factorization will be guaranteed, \footnote{This will be essentially \textit{by definition} of (infinitesimally) thickened smooth sets, 
whereby smooth maps are such that they `preserve the infinitesimal structure'.} a fact related to the classifying nature of 
$\mathbold{\Omega}_{\mathrm{dR}}^1$ (see Rem. \ref{ClassifyingFormsDoNotClassifyInSmoothSet}). 

\vspace{1mm} 
\noindent {\bf (iv)}  It is still not guaranteed that the criticality condition is functorial for an arbitrary local function. Intuitively, this is
due to the fact that in general, a smooth fiber-wise linear map $T\CF\rightarrow y(\FR)$ is not necessarily given by a section of some `cotangent bundle'. 
If the local function in question is a smooth action functional $S=\int_M \circ \CL$ on a compact manifold $M$, for some local Lagrangian $\CL$, 
then Prop. \ref{Crit(S)functoriality} shows that the induced variation map on the tangent bundle 
is given by
\begin{align*} 
T \CF &\longrightarrow y(\FR) \\
\partial_t \phi^k_t |_{t=0} & \; \longmapsto \; \int_M \langle \CE \CL(\phi^k)\, , \, \partial_t \phi^k_t |_{t=0} \rangle 
\end{align*}
which makes the criticality condition functorial. Cor. \ref{CriticalSmoothSetPullback} amplifies this by noting that the map out of the tangent 
bundle is equivalently determined via the section of the \textit{variational} cotangent bundle (Def. \ref{LocalCotangentBundleToFieldSpace}) over $\CF$

\vspace{-2mm} 
	\[ 
\xymatrix@R=1.6em@C=2.5em{ &&  T^*_{\mathrm{Var}} \CF  \ar[d]
	\\ 
	\CF\ar[rru]^{\CE \CL} \ar[rr]^>>>>>>>{\id_\CF} && \CF \, . 
}   
\]
We stress that this description fails for generic local functions and the induced map $T\CF\rightarrow y(\FR)$ 
obtained via point {\bf (iii)}. 

\vspace{1mm} 
\noindent {\bf (v)} The above discussion applies in all cases of spacetimes, compact or noncompact, with or without boundary, 
by demanding the local criticality of charges and encoding the potential boundary condition choices in  $\mathbold{P}(\CF)$.
\end{remark}

The critical set description may be employed to show that \textit{all} notions of (finite) symmetries of a Lagrangian field
theory $(\CL,\CF)$ (Def. \ref{FiniteSymmetryofLagrangianFieldTheory}) do in fact preserve the on-shell space of fields,
extending the result of Prop. \ref{LocalSymmetryPreservesOnshellSpace}.

\begin{proposition}[\bf Symmetries preserve on-shell space]
\label{SymmetryPreservesOnshellSpace} Any (spacetime covariant) symmetry $\CD:\CF\rightarrow \CF$ of a Lagrangian
field theory such that $\CL \circ \CD = f^*\circ (\CL + \dd_M \CK)$, where $\CK = K\circ j^\infty$ for some 
$K\in \Omega^{d-1,0}(J^\infty_M F)$ and $f:M\xrightarrow{\sim} M$, preserves the smooth critical subspace 
$\mathrm{Crit}(\CL) \hookrightarrow \CF$ of on-shell fields. That is, the diffeomorphism 
$\CD:\CF\rightarrow \CF$ restricts to a diffeomorphism
$$
\CD|_{\mathrm{Crit}(\CL)} \;:\; \mathrm{Crit}(\CL)\xlongrightarrow{\sim} \mathrm{Crit}(\CL)\, .
$$ 
\end{proposition}
\begin{proof} For the sake of being concise, we prove this in the case of a compact spacetime with boundary, with the noncompact 
case being completely analogous using the criticality of local charges instead. The following holds for either Neumann 
or Dirichlet boundary conditions, and so we do not indicate the choice in the notation. 
Let $\CL^{\CD}:= \CL\circ \CD  : \CF \rightarrow \Omega^{d}_{\mathrm{Vert}}(M)$, then for any plot $\phi_t\in \CF( \FR^1_t)$ 
 \begin{align*}
    \partial_t \CL^\CD(\phi_t) |_{t=0} &= \partial_t \big( f^*\CL(\phi_t) +f^*\dd_M \CK(\phi_t) \big) |_{t=0} \\
    &= f^*\big( \partial_t \CL(\phi_t)  |_{t=0} +\dd_M \CK(\phi_t)  |_{t=0} \big) \\
    &= f^* \Big( \big\langle \CE (\CL +\dd_M \CK) (\phi_0), \, \partial_t \phi_t |_{t=0} \big\rangle - 
    \dd_M \theta_\CL\big(\partial_t \phi_t |_{t=0}\big) \Big) \\
    &= f^* \big\langle \CE\CL (\phi_0), \, \partial_t \phi_t |_{t=0} \big\rangle - 
    f^* \dd_M \theta_\CL\big(\partial_t \phi_t |_{t=0}\big)
 \end{align*}
 where in the second line we used that $f^*$ and $\partial_t$ commute as operators on $\Omega^d(M)\hat{\otimes}C^\infty(\FR^1)$, 
 in the third we employ Eq. \eqref{VariationOfLagrangianDensityOnPoints} for the local Lagrangian $\CL+\dd_M \CK$ and in the 
 fourth we use the fact that exact Lagrangians have trivial Euler--Lagrange operators via Eq. \eqref{ExactLagrangianTrivialELequation}. 
 
 Since $f:M\rightarrow M$ is a diffeomorphism, it preserves the integral of top-forms (up to a sign) and so 
 \begin{align*}
 \partial_t S^{\CD}(\phi_t) |_{t=0} &=\int_M \partial_t \CL^\CD(\phi_t) |_{t=0}  
 = \int_M f^* \big\langle \CE\CL (\phi_0), \, \partial_t \phi_t |_{t=0} \big\rangle 
 - \int_M  f^* \dd_M \theta_\CL\big(\partial_t \phi_t |_{t=0}\big)  \\
 &=\pm \int_M \big\langle \CE\CL (\phi_0), \, \partial_t \phi_t |_{t=0} \big\rangle
 - (\pm) \int_M \dd_M \theta_\CL\big(\partial_t \phi_t |_{t=0}\big) \\
 &= \pm \partial_t S(\phi_t)|_{t=0}\, .
 \end{align*} 
 It follows $\phi_0\in \CF$ is a critical point for $S^\CD$ if and only if it is so for $S$. The argument follows 
 identically for $\FR^k$-plots, and hence 
 \vspace{-2mm}
\begin{align}\label{CritL=CritLP}
 \mathrm{Crit}(\CL^\CD)\cong \mathrm{Crit}(\CL) \cong \CF_{\CE\CL=0}\, .
 \end{align}

 \vspace{-1mm}
Next, let $\phi_0\in \CF(*)$ and define $\psi_0:=\CD(\phi_0)\in \CF(*)$. Since $\CP:\CF\rightarrow \CF$ is a diffeomorphism, any 
line plot $\psi_t\in \CF(\FR^1_t)$ starting at $\psi_0$ is given (uniquely) by $\CD(\phi_t)$ for some $\phi_t\in \CF(\FR^1_t)$. 
By Eq. \eqref{VariationOfLagrangianDensityOnPoints}, it is again the case that
\begin{align*}
\partial_t \CL(\psi_t)|_{t=0}= \big\langle \CE\CL (\psi_0), \, \partial_t \psi_t |_{t=0} \big\rangle 
\; - \; \dd_M \theta_\CL\big(\partial_t \phi_t |_{t=0}\big) \, ,
\end{align*}
i.e., $\psi_0=\CD(\phi_0)$ is a critical point for $\CL$ if and only if it satisfies the Euler--Lagrange equations of $\CL$ and the chosen boundary conditions.
On the other hand, by definition
\vspace{-2mm} 
\begin{align*}
\partial_t \CL(\psi_t)|_{t=0} :&= \partial_t \big(\CL\circ \CD(\phi_t)\big) |_{t=0} \\&= \partial_t \CL^\CD(\phi_t)|_{t=0}
\end{align*}

\vspace{-2mm} 
\noindent
whose integral vanishes for all $1$-parameter variations if and only if $\phi_0 \in \mathrm{Crit}(\CL)$, by Eq. \eqref{CritL=CritLP}. 
That is,  $\phi_0$ is an on-shell field for $\CL$ if and only if $\CD(\phi_0)$ is an on-shell field for $\CL$,
\vspace{-2mm} 
$$
\phi_0\in \mathrm{Crit}(\CL) \hspace{1.2cm} \iff \hspace{1.2cm}  \CD(\phi_0) \in \mathrm{Crit}(\CL) \,. 
$$

\vspace{-2mm} 
\noindent The discussion follows identically for $\FR^k$-plots, and hence the result follows.
\end{proof}

\newpage 

\addtocontents{toc}{\protect\vspace{-10pt}}
\section{Infinitesimal symmetries} 

\subsection{Local infinitesimal symmetries and Noether's First Theorem}
\label{EvolutionaryVectorFieldsAndNoetherTheoremsSection}

Having introduced local symmetries (Def. \ref{FiniteSymmetryofLagrangianFieldTheory}) of a Lagrangian field theory, we now 
proceed to define their infinitesimal version. 
These are \textit{local} vector fields on field space that preserve the Lagrangian up to an exact local Lagrangian,
and whose interplay with a certain subspace of 
vertical vector fields on the infinite jet bundle immediately yields Noether's First and Second theorems. 
As Lem. \ref{InfinitesimalSpacetimeCovariantSymmetries} will show, local vector fields do in fact capture the infinitesimal
version of spacetime covariant symmetries as well, thus justifying the focus solely \textit{local} vector fields.

\begin{definition}[\bf  Local vector fields]
\label{LocalVectorFields}
A smooth vector field $\CZ\in \CX(\CF)=\Gamma_\CF(T\CF)$ on field space is \textit{local} if it is given by
\vspace{-2mm} 
$$
\CZ=Z\circ j^\infty 
$$ 

\vspace{-2mm} 
\noindent for some bundle map 
\vspace{-1mm} 
\[ 
\xymatrix@C=2.8em@R=.2em  {J^\infty_M F \ar[rd]_{\pi^\infty_0} \ar[rr]^Z &   &  VF \ar[ld]
	\\ 
& F & 
}   
\]

\vspace{-1mm} 
\noindent  over the total space $F$ of the field bundle $F\rightarrow M$, via Lem. \ref{DiffOpsAsSmoothMaps}.
The subset of local vector fields is denoted by $\CX_{\mathrm{loc}}(\CF)\subset \CX(\CF)$.
\end{definition}
As usual, the name is justified since for each $\phi\in \Gamma_M(F)$ 
the value of the tangent vector $\CZ_\phi$ at each $x\in M$ depends only `locally' on $\phi$, via its jet 
$j^\infty_x \phi$ at $x\in M$. By Lem. \ref{FunctionsOnInftyJetBundle}, such bundle maps are locally represented by (finite) sums
\vspace{-2mm} 
\begin{align}
Z=Z^{a}\frac{\partial}{\partial u^a}\, ,     
\end{align}

\vspace{-2mm} 
\noindent where $\{Z^{a}\}\subset C^\infty(J^\infty_M F)$ are (locally defined) smooth functions on the infinite jet bundle, and 
$\{\frac{\partial}{\partial u^a}\}$ is the local coordinate basis for vertical tangent vectors on $F$. Thus in the physics literature, by abuse of notation
as in Ex. \ref{DiffeomorphismForVector-valuedFieldTheoryViaTarget} and \ref{DiffeomorphismForVector-valuedFieldTheoryViaBase},  
such a local vector field $\mathcal{Z} \in \CX(\CF)$ is usually denoted by 
\vspace{-2mm} 
\begin{align}\label{VectorFieldOnFieldSpaceAbuseOfNotation}
\CZ(\phi)&=\CZ^a(\phi)\cdot \frac{\delta}{\delta \phi^a}= Z^a\big(\phi, \{\partial_I \phi\}_{|I|\leq k}  \big) \cdot \frac{\delta}{\delta \phi^a}\, ,  
\hspace{1cm} \mathrm{or} \hspace{1cm}  \delta_Z \phi^a =\CZ^a(\phi)=  Z^a\big(\phi, \{\partial_I \phi\}_{|I|\leq k}\big) \, ,
\end{align} 
with the latter thought of as a smooth (and local) `infinitesimal transformation of the field'. The latter notation will be fully 
justified in \cref{TheBicomplexOfLocalFormsSection} and Lem. \ref{LocalCartanCalculus}. Recall, we have already introduced examples 
of such local vector fields in Ex. \ref{DiffeomorphismForVector-valuedFieldTheoryViaTarget} and Ex. \ref{DiffeomorphismForVector-valuedFieldTheoryViaBase}, 
where the corresponding bundle maps factor globally through $J^0_M F$ and $J^1_M F$, respectively.
Given their importance, the bundle maps inducing local vector fields are given a special name. 
\begin{definition}[\bf  Evolutionary vector field]\label{EvolutionaryVectorFields}
The set of \textit{evolutionary vector fields} $\CX_{\mathrm{ev}}(J^\infty_M F)$ is the set of smooth bundle maps 
\vspace{-2mm} 
\[ 
\xymatrix@C=2.8em@R=.2em  {J^\infty_M F \ar[rd]_{\pi^\infty_0} \ar[rr]^Z &   &  VF \ar[ld]
	\\ 
& F & 
}   
\]

\vspace{0mm} 
\noindent  over the total space $F$ of the field bundle $F\rightarrow M$.
\end{definition}

\begin{example}[\bf Differentiating finite local diffeomorphisms]\label{DifferentiatingFiniteLocalDiffeomorphisms}
Consider a smooth 1-parameter family of \textit{local} diffeomorphisms $\CP_t\in [\CF,\CF](\FR^1)$ starting at the identity $\CP_{0}=\id_\CF$, via Ex. \ref{VectorFieldsAsInfinitesimalDiffeomorphisms}, and so such that $\CP_t:= P_t \circ j^\infty $ for some smooth 1-parameter family of bundle maps
\vspace{-2mm} 
\[ 
\xymatrix@C=2.8em@R=.2em  {\FR^1\times J^\infty_M F \ar[rd]\ar[rr]^{P_t} &   &  F\;, \ar[ld]
	\\ 
& M & 
}   
\]

\noindent with $P_{t=0}=\pi^\infty_0:J^\infty_M F\rightarrow F$. Running through the differentiation of Ex. \ref{VectorFieldsAsInfinitesimalDiffeomorphisms} 
(also carried out explicitly in the particular cases of Ex. \ref{DiffeomorphismForVector-valuedFieldTheoryViaTarget} and Ex. \ref{DiffeomorphismForVector-valuedFieldTheoryViaBase}), 
the interested reader may verify that its 
infinitesimal version, i.e., the induced vector field $\partial_t \CP_t |_{t=0}\in \CX(\CF)$ is {\it local} in the sense of Def. \ref{LocalVectorFields}. 
In particular,
\vspace{-1mm} 
$$
\partial_t \CP_t |_{t=0} = Z\circ j^\infty \, ,
$$
where $Z=\partial_t P_t |_{t=0}$ is the induced evolutionary vector field 
\vspace{-2mm} 
\[ 
\xymatrix@C=2.8em@R=.2em  {J^\infty_M F \ar[rd]_{\pi^\infty_0} \ar[rr]^{Z} &   &  VF \ar[ld] \, .
	\\ 
& F & 
}   
\]

\vspace{-2mm} 
\noindent
\noindent In other words, local vector fields are indeed the infinitesimal version of 1-parameter local diffeomorphisms.
\end{example}

Strictly speaking, evolutionary vector fields are \textit{not} vector fields on $J^\infty_M F$, since they do not take values in $T(J^\infty_M F)$. 
Nevertheless, they may be uniquely `prolongated' to a Lie sub-algebra of vertical vector fields on $J^\infty_M F$ whose explicit form and properties 
rigorously justify several formulas appearing in field theory. For evolutionary vector fields arising by differentiating $1$-parameter families of 
bundle maps $P_t:J^\infty_M F\rightarrow F$ over M as in Ex. \ref{DifferentiatingFiniteLocalDiffeomorphisms}, the corresponding vector fields on 
$J^\infty_M F$ arise by differentiating the corresponding prolongated families $\pr P_t :J^\infty_M F \rightarrow J^\infty_M F$ 
from Def. \ref{ProlongationOfJetBundleMap}. More generally, we have the following result \cite{Anderson89}.

\vspace{2mm} 
\begin{proposition}[\bf Prolongated evolutionary vector fields]
\label{ProlongationEvolutionaryVectorFields}
$\,$

\noindent {\bf (i)}
For any $Z\in \CX_{\mathrm{ev}}(J^\infty_M F)$, there exists a unique `prolongated' vertical vector field 
$\pr Z:J^\infty_M F\rightarrow VJ^\infty_M F$, i.e.,
such that
\begin{itemize}
\item[{\bf (a)}] it projects to $Z$ 
\vspace{-1mm} 
$$
\dd\pi^\infty_0 \circ {\rm pr}Z = Z \, ,
$$
\item[{\bf (b)}] it commutes with the horizontal differential 
$$
[\iota_{{\rm pr}Z}, \dd_H] =0 \, .
$$
\end{itemize}
\noindent {\bf (ii)} Conversely, any vertical vector field $\hat{Z}\in \CX_V(J^\infty_M F)$ such that 
$[\iota_{\hat{Z}}, \dd_H] =0 $ defines an evolutionary vector field via 
$$
\dd \pi_{0}^\infty \circ Z \;:\; J^\infty_M F \longrightarrow VJ^\infty_M F \longrightarrow VF \, .
$$

\vspace{-2mm} 
\noindent {\bf (iii)} The two processes are inverse to each other.
\end{proposition}

In local coordinates, the coefficients of any vertical vector field
$\hat{Z}=0+\sum_{|I|=0}^\infty \hat{Y}^a_I \frac{\partial}{\partial u^a_I}$ such that 
$[\iota_{\hat{Z}}, \dd_H]=0$ are necessarily of the form $\hat{Y}^a_I= D_I(\hat{Y}^a)$, where 
$\hat{Y}^a=\hat{Y}^a_0$ are the coefficients of 
$\{\frac{\partial}{\partial u^a}\}$. This can be checked by evaluating the condition 
$[\iota_{\hat{Z}}, \dd_H]=0$ inductively on the local 
`generators' $\big\{x^\mu,\, \{u^a_I\}_{0\leq|I|}, \,\dd x^\mu, \,\{\dd_H u^a_I\}_{0\leq|I|}, \,\{\dd_V u^a_I\}_{0\leq|I|} \big\}$ 
of the variational bi-complex. That is,
\vspace{-2mm} 
$$
\hat{Z}= \sum_{|I|=0}^\infty D_I(\hat{Y}^a) \frac{\partial}{\partial u^a_I}\, ,
$$

\vspace{-2mm} 
\noindent and so if $Z=Z^a \frac{\partial}{\partial u^a}\in \CX_{\mathrm{ev}}(J^\infty_M F)$ is an evolutionary vector field, 
then its prolongation is given locally by
\vspace{-2mm} 
\begin{align}\label{ProlongationEvolutionayVectorFieldlocally}
\pr Z =\sum_{|I|=0}^\infty D_I(Z^a) \frac{\partial}{\partial u^a_I} \, .
\end{align}

\begin{remark}[\bf On the nomenclature of functional forms]\label{SmoothMapsInducedbyFunctionalForms}
Let $\omega^{p,q}\in \Omega^{p,q}(J^\infty_M F)$ be any $(p,q)$-form. Then for each q-tuple of evolutionary vector 
fields $Z_1,\cdots,Z_q\in \CX_{\mathrm{ev}}(J^\infty_M F)$, there is an induced horizontal $p$-form 
$\iota_{\pr Z_1}\cdots \iota_{\pr Z_q} \om^{p,q}= \om^{p,q}(\pr Z_1,\cdots, \pr Z_q) \in \Omega^{p,0}(J^\infty_M F)$ 
and hence an induced smooth current 
\vspace{-2mm} 
\begin{align*}
\CF&\longrightarrow \Omega^{p}_{\mathrm{Vect}}(M)\\
\phi^k &\longmapsto (j^\infty \phi^k)^* \om^{p,q}(\pr Z_1,\cdots, \pr Z_q)\, . 
\end{align*}

\vspace{-1mm} 
\noindent For the case of $p=d$, the decomposition 
$\Omega^{d,q}(J^\infty_M F) \cong \Omega^{d,q}_f(J^\infty_M F)\oplus \dd_H \Omega^{d-1,q}(J^\infty_M F)$ of Prop. \ref{InteriorEulerProperties}
says that the non-trivial part of the above current is the functional form component of $\om^{d,q}$. Similarly, the value 
of the corresponding charges defined over compact  submanifolds without boundary is completely determined by the functional form component. 
In other words, the resulting charges (functionals) on field space are completely determined by the subspace of functional
forms, justifying the name given. Thus, each functional form $\omega^{d,q}\in\Omega^{d,q}_f(J^\infty_M F)$ may be thought of as defining 
a family of Lagrangian densities, (smoothly) parametrized by $q$-tuples of evolutionary vector fields.
\end{remark}

\begin{corollary}[\bf  Cartan calculus for evolutionary vector fields]\label{EvolutionaryCartanCalculus}
$\,$

\noindent  The evolutionary vector fields $\CX_{\mathrm{ev}}(J^\infty_M F)$ may be naturally identified 
with a Lie sub-algebra of $\CX_V(J^\infty_M F)$.

\noindent {\bf (i)} The Lie derivative along an
(prolongated) evolutionary vector field is given by
\vspace{-2mm} 
$$
\mathbb{L}_{\pr Z} = [\iota_{\pr Z} , \dd_V] 
$$

\vspace{-2mm} 
\noindent and satisfies
$$
\mathbb{L}_{\pr Z} \dd_H = \dd_H \mathbb{L}_{\pr Z}\, , \hspace{2cm}\mathbb{L}_{\pr Z} \dd_V = \dd_V \mathbb{L}_{\pr Z} \, .
$$
\noindent {\bf (ii)}  The Lie bracket of any two (vertical) 
evolutionary vector fields is vertical 
$[{\pr Z_1}, \pr Z_2]\in \CX_V(J^\infty_M F)$. 

\noindent {\bf (iii)}  
For any two evolutionary vector fields 
\vspace{-3mm} 
$$
\iota_{[{\pr Z_1}, \pr Z_2]}\dd_H= \dd_H \iota_{[{\pr Z_1}, \pr X_Z]}\, .
$$
That is, $[{\pr Z_1}, \pr Z_2]$ is (the prolongation) of some evolutionary vector field $Z_3 \in \CX_\mathrm{ev}(J^\infty F)$. 

\noindent {\bf (iv)} In local coordinates, 
$$[{\pr Z_1}, \pr Z_2]=\pr Z_3$$
where $Z_3=\Big(D_I Z_1^b \cdot \frac{\partial Z_2^a}{\partial u^b_I}-D_I Z_2^b \cdot \frac{\partial Z_1^a}{\partial u^b_I}\Big)
\cdot  \frac{\partial}{\partial u^a}.$ 
\end{corollary}

\begin{proof}
Part {\bf (i)} follows from the commutation $[\iota_{\pr Z}, \dd_H]=0$, part {\bf (ii)} 
from vertical vector fields being involutive, part {\bf (iii)} by the general Cartan Calculus formula 
$\iota_{[{\pr Z_1}, \pr Z_2]}=[L_{\pr Z_1}, \iota_{\pr Z_2}]$, and part {\bf (iv)} by direct checking. 
\end{proof}  
The corollary implies the existence of an induced Lie algebra structure on the set of local vector fields. Via Ex. \ref{DifferentiatingFiniteLocalDiffeomorphisms}, 
it may be checked that this is the infinitesimal version of the group structure of the subgroup of local diffeomorphisms  
$\mathrm{Diff}_\mathrm{loc}(\CF)\hookrightarrow \mathrm{Diff}(\CF) $. 

\begin{definition}[\bf  Lie algebra of local vector fields]\label{LieAlgebraOfLocalVectorFields}
$\,$

\noindent {\bf (i)}  The Lie bracket of two local vector fields, $\CZ_1=Z_1\circ j^\infty, \; \CZ_2 = Z_2 \circ j^\infty \in \CX_{\mathrm{loc}}(\CF)$, is given by 
\vspace{-1mm} 
$$
[\CZ_1, \CZ_2]:= Z_3 \circ j^\infty\, ,
$$

\vspace{-1mm} 
\noindent where $Z_3\in \CX_{\mathrm{ev}}(J^\infty_M F)$ is the evolutionary vector field corresponding to $[\pr Z_1, \pr Z_2] \in \CX_V(J^\infty_M F)$, 
as in  Cor. \ref{EvolutionaryCartanCalculus} \!{(iii)}. 

\noindent {\bf (ii)} In local coordinates, by Eq. \eqref{HorizontalVectorFieldBasisAction} 
and Cor. \ref{EvolutionaryCartanCalculus} \!{(iv)}, this is locally given by 
\vspace{-2mm} 
$$
\CZ_3(\phi)= [\CZ_1,\CZ_2](\phi)
= \bigg(\frac{\partial}{\partial x^ I}\big( Z_1^b\circ j^\infty \phi\big) \cdot \Big(\frac{\partial Z_2^a}{\partial u^b_I}\circ j^\infty \phi\Big)
\; - \; \frac{\partial}{\partial x^I}\big( Z_2^b \circ j^\infty \phi\big) \cdot \Big(\frac{\partial Z_1^a}{\partial u^b_I} \circ j^\infty \phi\Big)\!\bigg) 
\cdot  \frac{\partial}{\partial u^a}\, ,
$$

\vspace{-1mm} 
\noindent or abusively as in Eq. \eqref{VectorFieldOnFieldSpaceAbuseOfNotation} and Rem. \ref{TreatingPartialDerivativesAsIndependent} by
\vspace{-1mm} 
$$
\CZ_3(\phi) = \bigg(\frac{\partial \CZ^b_1(\phi)}{\partial x^I} \cdot
\frac{\delta Z^a_2(\phi, \{\partial_J \phi\}_{|J|\leq k_2})}{\delta \phi^b_I}
\;\; - \;\; 
\frac{\partial \CZ^a_2(\phi)}{\partial x^I} \cdot \frac{\delta Z^b_1(\phi, \{\partial_J \phi\}_{|J|\leq k_1})}{\delta \phi^b_I} \bigg)\cdot
\frac{\delta}{\delta \phi^a}\, . 
$$
\end{definition}

Note that the action of a (prolongated) evolutionary vector field on horizontal forms $P\in \Omega^{p,0}(J^\infty_M F)$ simplifies further
\begin{align*}
\mathbb{L}_{\pr Z} P&= [\iota_{\pr Z}, \dd_V]=\iota_{\pr Z}  \dd_V P \, ,
\end{align*}
since $\iota_{ \pr Z} P=0 $ as $\pr Z$ is vertical. Locally, 
if $P= P_{\mu_1\cdots \mu_p}(x,\{u_I^a\}_{|I|\leq k})\cdot  \dd x^{\mu_1}\wedge \cdots \wedge \dd x^{\mu_p}$ the action takes 
the familiar form
\begin{align}\label{ActionOfEvolutionaryVectorFieldOnHorizontalForms} 
\mathbb{L}_{\pr Z} P&= \iota_{\pr Z} \bigg( \sum_{|I|=0}^\infty 
\frac{\partial P_{\mu_1\cdots \mu_p}}{\partial u^a_I } \cdot \dd_v u^a_I \wedge \dd x^{\mu_1}\wedge \cdots \wedge \dd x^{\mu_p} \!\bigg)\nn \\
&= \sum_{|I|=0}^\infty D_I(Z^a) \cdot \frac{\partial P_{\mu_1\cdots \mu_p}}{\partial u^a_I }
\cdot \dd x^{\mu_1}\wedge \cdots \wedge \dd x^{\mu_p}
\\
&= \pr Z( P_{\mu_1\cdots \mu_p})\cdot \dd x^{\mu_1}\wedge \cdots \wedge \dd x^{\mu_p}\, . \nn
\end{align}
As we will demonstrate, this encodes the textbook presentation of the action of local vector fields on local currents and functions on field space. 
To see this, note that by \eqref{HorizontalVectorFieldBasisAction} and for any field 
$\phi\in \Gamma_M(F)$, the local formula for 
the prolonganted vector field \eqref{ProlongationEvolutionayVectorFieldlocally} gives 
$$
\pr Z \circ j^\infty \phi =\sum_{|I|=0}^\infty \frac{\partial}{\partial x^I} 
\big(Z^a\circ j^\infty \phi\big) \frac{\partial}{\partial u^a_I} 
$$ 
as a section of $VJ^\infty F$ over $M$. That is, as a section of $\mathbold{\Gamma}_M(VJ^\infty F) 
\rightarrow \mathbold{\Gamma}_M(F)$ 
we may abusively denote the (prolongated) vector field $\pr \CZ := \pr Z\circ j^\infty$ on $\CF$ by
\vspace{-3mm}  
\begin{align}\label{ProlongatedVectorFieldonFieldSpaceAbusively}
\pr \CZ(\phi) = \sum_{|I|=0}^{\infty} \frac{\partial \CZ^a(\phi)}{\partial x^I}\frac{\delta}{\delta (\partial_I  \phi^a)}\, ,
\end{align}

\vspace{-2mm} 
\noindent similar to \eqref{VectorFieldOnFieldSpaceAbuseOfNotation}. 

\medskip 
Analogously, by the local formula 
\eqref{ActionOfEvolutionaryVectorFieldOnHorizontalForms} of $\mathbb{L}_{\pr Z} P$, the value of the induced 
differential operator $\mathbb{L}_{\pr Z} P\circ j^\infty_M: \CF \rightarrow \Omega^{p}_{\mathrm{Vert}}(M)$ on a field $\phi$ is given by
\vspace{-2mm} 
$$
\mathbb{L}_{\pr Z} P \circ j^\infty(\phi)=\sum_{|I|=0}^\infty \frac{\partial (Z^a\circ j^\infty \phi)}{\partial x^I}
\cdot \frac{\partial P_{\mu_1\cdots \mu_p}}{\partial u^a_I }\circ j^\infty \phi \cdot \dd x^{\mu_1}\wedge \cdots \wedge \dd x^{\mu_p}\, ,
$$

\vspace{-2mm} 
\noindent which is presented, abusing notation as in Rem. \ref{TreatingPartialDerivativesAsIndependent}, as
\vspace{-2mm} 
\begin{align}\label{ProlongatedActionOnHorizontalFormAbusively}
\mathbb{L}_{\pr Z} P \circ j^\infty(\phi)= 
&= \pr \CZ(\phi) \Big( P_{\mu_1\cdots \mu_p}\big(\phi,\{\partial_J \phi\}_{|J|\leq k}\big) \! \Big) 
\cdot \dd x^{\mu_1}\wedge \cdots \wedge \dd x^{\mu_p} \nn
\\
&=\sum_{|I|=0}^\infty \frac{\partial \CZ^a(\phi)}{\partial x^I} \cdot
\frac{\delta P_{\mu_1\cdots \mu_p}\big(\phi,\{\partial_J \phi \}_{|J|\leq k}\big) }{\delta (\partial_I \phi^a) }
\cdot \dd x^{\mu_1}\wedge \cdots \wedge \dd x^{\mu_p}\,.
\end{align}

\begin{definition}[\bf  Action of local vector fields on currents]\label{ActionOfLocalVectorFieldsOnCurrents}
The Lie algebra of local vector fields $\CX_{\mathrm{loc}}(\CF)$ (Def. \ref{LocalVectorFields}) acts on  
local currents (Def. \ref{CurrentOnFieldSpace}) via the 
prolongation of the corresponding evolutionary vector field. Explicitly:

\noindent {\bf (i)} If $Z\in \CX_{\mathrm{ev}}(J^\infty_M F)$ is an evolutionary vector 
field with induced local vector field $\CZ=Z\circ j^\infty \in \CX_{\mathrm{loc}}(\CF)$ and $P\in\Omega^{p,0}(J^\infty_M F)$ is a horizontal 
$(p,0)$-form with induced current $\CP= P \circ j^\infty: \CF \rightarrow \Omega^{p}_{\mathrm{Vert}}(M)$, then
\begin{align}
\CZ (\CP):= \mathbb{L}_{\pr Z}(P) \circ j^\infty \;:\; \CF \longrightarrow \Omega^{p}_{\mathrm{Vert}}(M)\, .
\end{align}
\noindent {\bf (ii)}  Locally, the value of the resulting current on a field $\phi \in \CF$ is given by
\vspace{-2mm} 
\begin{align*}
\CZ (\CP)(\phi)=\sum_{|I|=0}^\infty \frac{\partial (Z^a\circ j^\infty \phi)}{\partial x^I} 
\cdot \frac{\partial P_{\mu_1\cdots \mu_p}}{\partial u^a_I }\circ j^\infty \phi \cdot \dd x^{\mu_1}\wedge \cdots \wedge \dd x^{\mu_p} 
\end{align*}

\vspace{-2mm} 
\noindent and similarly on $\FR^k$-plots of fields.

\noindent {\bf (iii)}  Via the abuse of notation \eqref{ProlongatedVectorFieldonFieldSpaceAbusively} and \eqref{ProlongatedActionOnHorizontalFormAbusively}, 
the value on a field $\phi$ may be equivalently calculated as 
\begin{align*}
 \CZ(\CP)(\phi)&=  \mathbb{L}_{\pr Z} P\circ j^\infty \phi = \pr \CZ(\phi) \Big( P_{\mu_1\cdots \mu_p}\big(\phi,\{\partial_J \phi \}_{|J|\leq k}\big)\! \Big) 
 \cdot \dd x^{\mu_1}\wedge \cdots \wedge \dd x^{\mu_p}  \\
 &= \sum_{|I|=0}^\infty \frac{\partial \CZ^a(\phi)}{\partial x^I} \cdot
 \frac{\delta P_{\mu_1\cdots \mu_p}\big(\phi,\{\partial_J \phi \}_{|J|\leq k}\big) }{\delta (\partial_I \phi^a) } 
 \cdot \dd x^{\mu_1}\wedge \cdots \wedge \dd x^{\mu_p}\, , 
 \end{align*}

 \vspace{-2mm} 
 \noindent by treating $\{\partial_I\phi^a\}_{|I|\leq k}$ as independent and computing the corresponding partial derivatives, 
 via Rem. \ref{TreatingPartialDerivativesAsIndependent}.

 \noindent {\bf (iv)} Using further the `infinitesimal transformation of the field' notation, it may also be written as
 \begin{align*}
  \delta_Z \CP (\phi) =  \sum_{|I|=0}^\infty \frac{ \partial (\delta_Z \phi^a) }{\partial x^I}  \cdot
 \frac{\delta P_{\mu_1\cdots \mu_p}\big(\phi,\{\partial_J \phi \}_{|J|\leq k}\big) }{\delta (\partial_I \phi^a) }
 \cdot \dd x^{\mu_1}\wedge \cdots \wedge \dd x^{\mu_p}
 \end{align*}

 \vspace{-2mm}
\noindent recovering the formulas implicitly used in the physics literature. The latter will be fully justified in \cref{TheBicomplexOfLocalFormsSection} 
and Lem. \ref{LocalCartanCalculus}.
\end{definition}

By construction and Cor. \ref{EvolutionaryCartanCalculus}, 
this is a Lie algebra action via derivations with respect to the (graded) algebra structure of currents. In \cref{TheBicomplexOfLocalFormsSection}, 
this action will be identified with the ``Lie derivative along $\CZ$'', in an appropriate sense. Integrating the resulting currents along submanifolds, 
we get an induced action on charges and hence on  the algebra of local functions on field space.

\begin{definition}[\bf  Action of local vector fields on local functions]\label{ActionOfLocalVectorFieldsOnLocalFunctions}
The Lie algebra of local vector fields $\CX_{\mathrm{loc}}(\CF)$ acts on charges, the generators of smooth local functions $C^\infty_\mathrm{loc}(\CF)$ 
of Def. \ref{LocalFunctionsOnFieldSpace} via the action on the corresponding current, and is
extended as a derivation to all of $C^\infty_\mathrm{loc}(\CF)$. 

\noindent {\bf (i)} If $X\in \CX_{\mathrm{ev}}(J^\infty_M F)$ is an evolutionary vector field with induced local 
vector field $\CZ=X\circ j^\infty \in \CX_{\mathrm{loc}}(\CF)$ and $P\in\Omega^{p,0}(J^\infty_M F)$ is a horizontal
$(p,0)$-form with induced charge $\CP_{\Sigma^p}=\int_{\Sigma^p} P \circ j^\infty \in C^\infty_\mathrm{loc}(\CF)$, 
then
\begin{align}
\CZ (\CP_{\Sigma^p}):=\int_{\Sigma^p}\CZ (\CP) = 
\int_{\Sigma^p} \mathbb{L}_{\pr Z}(P) \circ j^\infty  \in C^\infty_{\mathrm{loc}}(\CF)\, .
\end{align}
\noindent {\bf (ii)}  If 
$P'\in \Omega^{p',0}(J^\infty_M )$ is an other horizontal $(p',0)$ form with induced charge $\CP'_{\Sigma^{p'}}$, then 
$$
\CZ (\CP_{\Sigma^p}\cdot \CP
'_{\Sigma^{p'}}):= \CZ (\CP_{\Sigma^p}) \cdot \CP'_{\Sigma^{p'}} + \CP_{\Sigma^p}\cdot \CZ (\CP'_{\Sigma^{p'}})\,.
$$ 
\end{definition}

\begin{remark}[\bf Tangent vectors and derivations of local functionals]\label{TangentVectorsAndDerivationsofLocalFunctions}Let
$Z_\phi = \partial_t \phi_t |_{t=0}\in T_\phi(\CF)=\Gamma_M(VF)$ be a tangent vector at $\phi\in \CF$, represented by 
an $\FR^1$-plot $\phi_t\in \mathbold{\Gamma}_M(F)(\FR^1)$. 

\noindent {\bf (i)} Recall, as in Rem. \ref{TangentVectorsPathsOfFieldsAndDerivations}, there is an induced $\FR$-valued derivation of local functionals 
\vspace{-2mm} 
\begin{align*}
    C^\infty_\mathrm{loc}(\CF) &\longrightarrow \FR \\
    \CP_{\Sigma^p} &\longmapsto \partial_t(\CP_{\Sigma^p}\circ \phi_t) |_{t=0} \, . \nn 
\end{align*}

\vspace{-2mm} 
\noindent {\bf (ii)} Contrary to the case of general smooth functions on $\CF$ as in 
Rem. \ref{TangentVectorsPathsOfFieldsAndDerivations}, the derivation on local functionals is clearly independent of the representative $\FR^1$-plot. 
To see this, notice that for each $x\in M$
\begin{align*}
 \partial_t (P\circ j^\infty \phi_t(x) )|_{t=0}
  &= \iota_{\partial_t j^\infty \phi_t(x) |_{t=0}} \dd_V P|_{j^\infty \phi_0(x)} 
 \end{align*}
as in Eq. \eqref{InfinityJetVerticalVectorFromProlongationofPlotActionVia1form}, where the latter `prolongated' tangent vector 
$\partial_t j^\infty\phi_t|_{t=0}\in \Gamma_M(VJ^\infty F)$ depends only on $\partial_t \phi_t |_{t=0}\in \Gamma_M(VF)$, as can be seen explicitly in 
Ex. \ref{InfinityJetVerticalVectorFromProlongationofPlot}. Notice that varying pointwise in $x\in M$ this contraction 
only \textit{partially} defines a`$(p,0)$-form on $J^\infty_M F$', that is only along the image of $j^\infty \phi$. 
However, it does define a top-form on the base $M$, denoted by
\begin{align*}
\partial_t (P\circ j^\infty \phi_t )|_{t=0}
&=(j^\infty \phi)^* \big( \iota_{\partial_t j^\infty \phi_t |_{t=0}} \dd_V P) 
\end{align*}
 where the pullback form acts as in Lem. \ref{HorizontalDifferentialBaseDeRhamCompatibility}. Thus, on the induced local function 
 $\CP_{\Sigma^p}$ we have
 \vspace{-2mm} 
\begin{align*}
\partial_t(\CP_{\Sigma^p} \circ \phi_t)\big|_{t=0} &= 
\int_{\Sigma^p} \partial_t( P \circ j^\infty \phi_t) \big|_{t=0} \\
&= \int_{\Sigma^p} (j^\infty \phi)^*\iota_{\partial_t j^\infty \phi_t |_{t=0}} \dd_V P \, ,
\end{align*}

\vspace{-1mm} 
\noindent which manifestly depends only on $\phi\in \CF$ and $Z_\phi =\partial_t \phi_t |_{t=0} \in T_\phi \CF$. 

\noindent {\bf (iii)} This observation may be used as an alternative definition of the action of local vector fields. 
Indeed, for a local vector field $\CZ$ we may instead define $\CZ (\CP_{\Sigma^p})$ by its value on field configurations as
\vspace{-2mm} 
\begin{align}\label{AlternativeDefinitionofLocalVectorFieldAction}
\CZ (\CP_{\Sigma^p})(\phi):= \partial_{t} \phi_t |_{t=0} (\CP_{\Sigma^p})=
\int_{\Sigma^p}\partial_t(P \circ j^\infty \phi_t)|_{t=0} \, ,
\end{align}

\vspace{-1mm} 
\noindent where $\CZ_\phi =Z\circ j^\infty_M (\phi) \in T_\phi(\CF)$, 
is represented as 
\vspace{-2mm} 
$$
\CZ_\phi = \partial_t \phi_t |_{t=0}\, ,
$$

\vspace{-2mm} 
\noindent for some smooth $\FR^1$-plot of fields $\phi_t$ (Lem. \ref{LinePlotsRepresentTangentVectors}). 
This recovers the expression of Def. \ref{ActionOfLocalVectorFieldsOnLocalFunctions}
since 
\vspace{-2mm}
\begin{align*}
\partial_t(P \circ j^\infty \phi_t)|_{t=0}&= (j^\infty \phi)^*\iota_{\partial_t j^\infty \phi_t |_{t=0}} \dd_V P \\
&= (\iota_{\pr Z} \dd_V ) \circ j^\infty \phi = L_{\pr Z} P \circ j^\infty \phi\;.
\end{align*} 

\vspace{-2mm}
\noindent {\bf (iv)} The same statements apply for $\FR^k$-plots.
Note that this is in contrast with a general vector field, which might not necessarily define an action on smooth (or local) 
functionals, as described in \eqref{PathSectionDerivation}. From this description, the derivation extension of 
Def. \ref{ActionOfLocalVectorFieldsOnLocalFunctions} is deduced as a property instead. 

\noindent {\bf (v)} For a local vector field 
$\CZ = \partial_t \CP_t |_{t=0}$ arising from a 1-parameter family of local diffeomorphisms as in
Ex. \ref{DifferentiatingFiniteLocalDiffeomorphisms}, the Lie algebra action corresponds to the differentiation of the pullback map
\vspace{-2mm} 
$$
\CP_t^* \;:\; C^\infty_\mathrm{loc}(\CF)\longrightarrow C^\infty_\mathrm{loc}(\FR^1\times \CF)\, ,
$$

\vspace{-2mm} 
\noindent and hence is the infinitesimal version of the $\mathrm{Diff}_\mathrm{loc}(\CF)$ action on $C^\infty(\CF)$.
\end{remark}

At this point, we are well-equipped to define the notion of an infinitesimal local symmetry of a local Lagrangian field theory.
With this notion at hand, the statement and proof of Noether's First and Second theorems follow easily, as a result on the 
bicomplex of the infinite 
jet bundle $J^\infty_M F$, which naturally pulls back to the field space as a statement about local currents and their charges.

\begin{definition}[\bf  Infinitesimal local symmetry of Lagrangian field theory]\label{InfinitesimalLocalSymmetryOfLagrangian}
An evolutionary vector field $Z \in \CX_{\mathrm{ev}}(J^\infty_M F)$ is an \textit{infinitesimal symmetry} of
a Lagrangian density $L\in \Omega^{d,0}(J^\infty_M F)$ if 
\vspace{-2mm}
\begin{align}
\mathbb{L}_{\pr Z} L= \dd_H K_Z
\end{align}

\vspace{-2mm}
\noindent
for some $K_Z \in \Omega^{d-1,0}(J^\infty_M F)$. Equivalently, 
by Lem. \ref{HorizontalDifferentialBaseDeRhamCompatibility} and Def. \ref{ActionOfLocalVectorFieldsOnCurrents}, 
the corresponding local vector field $\CZ=Z\circ j^\infty_M \in \CX_\mathrm{loc}(\CF)$ is an \textit{infinitesimal local symmetry} 
of the corresponding Lagrangian field theory $(\CF, \CL)$ if
\vspace{-2mm}
$$
\CZ(\CL)= \dd_M (\CK_Z ) \, ,
$$

\vspace{-2mm}
\noindent i.e., if it preserves the Lagrangian up to a trivial local Lagrangian.
\end{definition}
It is easy to see that the Lie bracket $[\CZ_1,\CZ_2]$ (Def. \ref{LieAlgebraOfLocalVectorFields}) of any two infinitesimal local
symmetries $\CZ_1,\CZ_2\in \CX_{\mathrm{loc}}(\CF)$ of a local field theory $(\CF,\CL)$ is also a symmetry. This follows since
\vspace{-2mm}
\begin{align}\label{SubalgebraofLocalInfinitesimalSymmetriesCalculation}
\mathbb{L}_{[\pr Z_1, \pr Z_2]} L &= \mathbb{L}_{\pr Z_1}\, (\mathbb{L}_{\pr Z_2} L) - \mathbb{L}_{\pr Z_2} ( \mathbb{L}_{\pr Z_1} L) =\mathbb{L}_{\pr Z_1} (\dd_H K_{Z_2}) - \mathbb{L}_{\pr Z_2} ( \dd_H K_{Z_1})\\ 
&= \dd_H( \mathbb{L}_{\pr Z_1} K_{Z_2} -\mathbb{L}_{\pr Z_2} K_{Z_1}) \, , \nn 
\end{align}

\vspace{-1mm}
\noindent where we used symmetry assumption of $\CZ_1,\CZ_2$, and then the Cartan calculus for evolutionary vector fields of Cor. \ref{EvolutionaryCartanCalculus}. 
That is, the subspace $\CX_{\mathrm{loc}}^{\CL}(\CF)$ of local vector fields consisting of local symmetries of a field theory $(\CF,\CL)$ is a Lie subalgebra
\vspace{-1mm} 
\begin{align}\label{SubalgebraofLocalInfinitesimalSymmetries}
\big(\CX_{\mathrm{loc}}^{\CL}(\CF),\, [-,-]\big) \longhookrightarrow \big(\CX_{\mathrm{loc}}(\CF),\, [-,-]\big) \, ,
\end{align}
which is the infinitesimal version of the (smooth) subgroup inclusion $\mathrm{Diff}_{\mathrm{loc}}^{\CL}(\CF)\hookrightarrow 
\mathrm{Diff}_{\mathrm{loc}}$ from part {\bf (c)} of Rem. \ref{OnSpacetimeCovariantSymmetries}.  

\bigskip 
There is a class of infinitesimal local symmetries that exist for \textit{every} local Lagrangian field theory. For reasons that will become apparent shortly, 
these are called \textit{trivial} infinitesimal symmetries. 

\newpage 
\begin{example}[\bf Trivial infinitesimal symmetries]\label{TrivialInfinitesimalSymmetries}
Let $L\in \Omega^{d,0}(J^\infty_M F)$ be an arbitrary Lagrangian density. For any bundle map
$T:J^\infty_M F \rightarrow \wedge^2 VF \otimes \wedge^d TM$ over $F$, there is an induced evolutionary vector field
$T\cdot \mathrm{EL}\in \CX_{\mathrm{ev}}(J^\infty_M F)$ given by the composition of bundle maps over $F$
\vspace{-2mm}
\begin{align*}
T \cdot \mathrm{EL}:J^\infty_M F \xrightarrow{\quad T \otimes \mathrm{EL} \quad} \wedge^2 VF\otimes 
\wedge^d TM \otimes V^*F \otimes \wedge^d T^*M \xlongrightarrow{\quad \sim \quad} (\wedge^2 VF \otimes V^*F) 
\otimes (\wedge^d TM \otimes \wedge^d T^*M) \longrightarrow VF  
\end{align*}

\vspace{-1mm}
\noindent where the second map is the isomorphism swapping the order of the fibers, while the last map combines the fiberwise contractions 
$\iota_{(-)}(-): \wedge^2 VF \otimes V^*F \rightarrow VF $ and $\iota_{(-)}(-): \wedge^d TM \otimes \wedge^d T^*M \xrightarrow{\sim} M\times \FR$.
In local coordinates, we have 
$T=T^{[ab]}\cdot \frac{\partial}{\partial u^a}\wedge\frac{\partial}{\partial u^b} \cdot \frac{\partial}{\partial x^1}\wedge \cdots \wedge\frac{\partial}{\partial x^d}$ 
for some collection of (local) functions $T^{[ab]}\in C^\infty (J^\infty_M F)$ antisymmetric in the indices, and similarly 
$\mathrm{EL}= \mathrm{EL}_a \cdot \dd_F u^a \wedge \dd x^1 \cdots \dd x^d$, thus the induced `trivial' evolutionary vector field is locally
\vspace{-2mm}
\begin{align*}
T\cdot \mathrm{EL}= T^{[ab]}\cdot \mathrm{EL}_b \cdot \frac{\partial}{\partial u^a}\, , 
\end{align*}

\vspace{-2mm}
\noindent which is the form that often appears in the physics literature.\footnote{As customary, the global structure of $M$, $F\rightarrow M$ and 
induced bundles is often ignored -- hence defining a trivial symmetry simply by a collection of functions $K^{[ab]}$ antisymmetric in the indices.
This is not sufficient to define a \textit{global} evolutionary vector field on non-trivial field bundles.}
By construction, this is `trivially' a symmetry of the density $L$
\vspace{-1mm}
\begin{align*}
\mathbb{L}_{\pr(T\cdot \mathrm{EL})} L &= \iota_{\pr(T\cdot \mathrm{EL})} \dd_V L 
\\
&= \iota_{\pr(T\cdot \mathrm{EL})} \mathrm{E L} + \iota_{\pr(T\cdot \mathrm{EL})} \dd_H  \theta_L 
\\
&= 0 + \dd_H \big( - \iota_{\pr(T\cdot \mathrm{EL})}  \theta_L \big)
\end{align*}

\vspace{-1mm}
\noindent where the first term vanishes by its local coordinate expression (see also proof of Prop. \ref{Noether1st}),
$\iota_{\pr(T\cdot \mathrm{EL})} \mathrm{E L} = T^{[ab]} \cdot EL_a\cdot EL_b $ and using symmetry and antisymmetry of the indices, 
while the second becomes horizontally exact by the commutation relation for evolutionary vector fields. Lastly, notice that any 
such evolutionary vector field vanishes on the shell $S_L\hookrightarrow J^\infty_M F$ of the Lagrangian (Def. \ref{ShellOfLagrangian}). 
Equivalently, the corresponding local vector field $\CK\cdot \CE \CL:= (K\cdot EL) \circ j^\infty = (K\circ j^\infty) \cdot (EL\circ j^\infty)$ 
is an infinitesimal symmetry of the field theory $(\CF, \CL)$
\vspace{-2mm}
\begin{align*}
\CT \cdot \CE \CL= \CT^{[ab]}\cdot \CE \CL_b \cdot \frac{\delta}{\delta \phi^a}\hspace{0.5cm} \in \hspace{0.5cm} \CX_{\mathrm{loc}}^{\CL}(\CF)\,  , 
\end{align*}

\vspace{-1mm}
\noindent which vanishes on the smooth subspace of on-shell fields $\CF_{\CE \CL}\hookrightarrow \CF.$ It is easy to see that the Lie bracket (Def. \ref{LieAlgebraOfLocalVectorFields}) of any two trivial infinitesimal symmetries is again trivial. In other words, the set of all trivial 
infinitesimal local symmetries is a further Lie subalgebra of \eqref{SubalgebraofLocalInfinitesimalSymmetries}, denoted by
\vspace{-1mm}
\begin{align*}
(\CX_\mathrm{loc}^{\CL,\mathrm{triv}}(\CF),\,  [-,-] )
\;\; \longhookrightarrow \;\; 
(\CX_{\mathrm{loc}}^{\CL}(\CF),\, [-,-])\, .  
\end{align*}
\vspace{-1mm}
\noindent The integrated (finite) version of this fact further justifies the name `trivial'.
\end{example}

\begin{example}[\bf Differentiating finite local symmetries]\label{DifferentiatingFiniteLocalSymmetries}
Consider the case where $\CD_{t}=D_t\circ j^\infty$ is a 1-parameter (finite) local symmetry of a Lagrangian field theory $(\CF,\CL)$ as in Def. \ref{FiniteSymmetryofLagrangianFieldTheory}, starting at the identity, and so
$$
\CL\circ \CD_t = \CL + \dd_M \CK_t 
$$
for some $\CK_t= K_t \circ j^\infty $, where $K_t:\FR^1\times J^\infty_M F \rightarrow \wedge^{d-1}T^*M$ is a 1-parameter family of smooth 
bundle maps over $M$, with $K_0=0_M:J^\infty_M\rightarrow \wedge^{d-1} T^*M$ the zero bundle map. 
Then the interested reader may verify, proceeding as in Ex. \ref{DifferentiatingFiniteLocalDiffeomorphisms}, that the differentiated version 
of the relation recovers exactly that 
of Def. \ref{InfinitesimalLocalSymmetryOfLagrangian}
$$
\CZ(\CL)= \dd_M (\CK_\CZ ) \, ,
$$
with\footnote{Strictly speaking, $\partial_t K_t |_{t=0}$ is a map 
$J^\infty_M F \rightarrow V(\wedge^{d-1} T^*M) \cong \wedge^{d-1} T^*M \times_M \wedge^{d-1}T^*M$ which covers $K_0=0_M$ via 
$\pr_1:  \wedge^{d-1} T^*M \times_M \wedge^{d-1}T^*M \rightarrow \wedge^{d-1} T^*M $. Thus it is completely determined by the projection to 
the section factor, which is the one we tacitly mean above.} 
$K_Z= \partial_t K_t |_{t=0}:J^\infty_M F \rightarrow \wedge^{d-1}(T^*M)$ and
$\CZ=Z\circ j^\infty$, for $Z=\partial_t D_t |_{t=0}$.  In other words, local 
infinitesimal symmetries of a Lagrangian field theory are indeed the infinitesimal version of finite local symmetries. In particular, 
if a trivial infinitesimal local symmetry $\CT\cdot \CE \CL$ of Ex. \ref{TrivialInfinitesimalSymmetries} is integrable, then the 
corresponding symmetry $\CP_t$ is necessarily the identity\footnote{Strictly, only for $t\in (-\epsi,\epsi) \subset \FR^1$ for some 
$\epsi \in \FR$.} on the subspace on-shell fields $\CF_{\CE \CL}$, since $\CT \cdot \CE \CL=\partial_t \CD_t |_{t=0} $ vanishes on
$\CF_{\CE \CL}\hookrightarrow \CF$.   
\end{example}

The main and most famous application of infinitesimal local symmetries is via Noether's First Theorem, by which each symmetry produces a conserved current.

\newpage 
\begin{proposition}[\bf Noether's First Theorem]\label{Noether1st} \cite{Noether}
Let $Z\in \CX_{\mathrm{ev}}(J^\infty_M F)$ be an infinitesimal symmetry with $\mathbb{L}_{\pr Z} L= \dd_H K_Z$.

\noindent {\bf (i)}  Then the $(d-1,0)$-form 
\vspace{-3mm} 
\begin{align}\label{ConservedCurrentOnJetBundle}
P_Z:=K_Z +\iota_{\pr Z} \theta_L
\end{align}

\vspace{-3mm}
\noindent
satisfies 
\vspace{-3mm}
\begin{align*}
\dd_H P_Z= \iota_{\pr Z} EL = \langle EL, Z\rangle \, , 
\end{align*}

\vspace{-2mm}
\noindent where $\dd_V L = EL + \dd_H \theta_L$ as in \eqref{LagrangianVerticalDifferentialDecomposition}, and the 
latter pairing is the duality bundle map 
$\langle-,-\rangle:(\wedge^d T^*M \otimes V^*F)\otimes VF \rightarrow \wedge^d T^*M$ over $M$. 
In particular, $P_Z$ is horizontally closed on the shell of $\mathrm{S}_L\hookrightarrow y(J^\infty_M F)$. 

\vspace{1mm} 
\noindent {\bf (ii)}  It follows the induced current 
\vspace{-4mm} 
\begin{align*}
\CP_Z:= \CK_\CZ + \iota_\CZ \theta_\CL \equiv \CK_\CZ + (\iota_{\pr Z}\theta_L) \circ j^\infty  \end{align*}

\vspace{-2mm}
\noindent satisfies 
\vspace{-3mm} 
\begin{align*}
\dd_M\CP_Z = \langle\CE \CL, \CZ \rangle 
\end{align*}

\vspace{-2mm}
\noindent and so it is conserved on the space of on-shell fields $\CF_{\mathrm{EL}=0}\hookrightarrow \CF$.
\begin{proof}
We have 
\vspace{-3mm} 
\begin{align*}
\dd_H (P_Z)&= \dd_H K_Z + \dd_H \iota_{\pr Z} \theta_L = \mathbb{L}_{\pr Z} L -\iota_{\pr Z} \dd_H \theta_L \\
&=\iota_{\pr Z} \dd_V L - \iota_{\pr Z} \dd_H \theta_L =\iota_{\pr Z} EL + \iota_{\pr Z} \dd_H \theta_L - \iota_{\pr Z} \dd_H \theta_L \\
&= \iota_{\pr Z} EL\, ,
\end{align*}

\vspace{-2mm}
\noindent
where the first line follows by the symmetry assumption and the second by the commutation relation $[\dd_H, \iota_{\pr Z}\theta_L ]=0$ 
for evolutionary vector fields. The identification $\iota_{\pr Z} EL= \langle EL, Z\rangle$ follows since $EL$ is a source form, 
and in particular corresponds to a bundle map $J^\infty_M F \rightarrow V^*F \otimes \wedge^d T^*M $ over F,
according to Lem. \ref{SourceFormsDefineDifferentialOperators}. Explicitly,  $EL=EL^a \cdot \dd_V u^a \wedge \dd x^1\wedge \cdots \wedge \dd x^d$ 
locally, and hence only the first component 
of $\pr Z= Z^a \frac{\partial}{\partial u^a} + \sum_{|I|\geq 1}^{\infty} D_I(Z^a) \frac{\partial}{\partial u^a_I}$ contributes 
to the contraction, i.e., the component of its evolutionary vector field $Z$. The induced current $\CP_{Z}=P_Z \circ j^\infty$ is 
on-shell conserved by Stokes' Theorem \ref{StokesTheoremLocalFunctions}. In particular the formula for $\dd_M \CP_Z $ follows 
immediately since $\langle EL, Z\rangle \circ j^\infty_M = \langle  EL\circ j^\infty, Z\circ j^\infty\rangle = \langle \CE \CL, \CZ\rangle.$
\end{proof}
\end{proposition}
As with Eq. \eqref{ThetaNotationalShorthand}, we have used the notation for the pullback of the contraction $\iota_{\pr Z} \theta_L$ to field space $\CF$
\vspace{-2mm}
\begin{align}\label{ThetaContractionShorthandNotation}
 \iota_\CZ \theta_\CL := (\iota_{\pr Z}\theta_L) \circ j^\infty \, .
\end{align}

\vspace{-2mm}
\noindent 
The underlying mathematical meaning of this notation will become apparent in \cref{TheBicomplexOfLocalFormsSection}, as an actual contraction 
of the local vector field $\CZ:= Z\circ j^\infty$ and the `local form $\theta_\CL$' on $\CF\times M$ (see Eq. \eqref{ContractionOfLocal1form}). 
For the moment, we will only need it and treat it as a useful notation.

\smallskip 
The calculation may be read in reverse, implying the converse statement: Suppose there exists a current $\CP_Z$ which is
on-shell preserved, in such a way that $\dd_M \CP_Z = \langle \CE \CL, \CZ \rangle$ for some local vector field $\CZ= Z\circ j^\infty_M$.
Then the vector field is an infinitesimal symmetry of the local Lagrangian with 
$\CZ (\CL) = \dd_M \big(\CP_Z - (\iota_{\pr Z} \theta_L \circ j^\infty_M)\big)$. An immediate `trivial' application of Noether's First Theorem 
is on the trivial symmetries from Ex. \ref{TrivialInfinitesimalSymmetries}, whereby $\CK_{\CT\cdot \CE \CL } = - (\iota_{\pr T\cdot E L} \theta_L)$ 
and hence $\CP_{\CT \cdot \CE \CL} = 0$ identically. Strictly speaking, an infinitesimal symmetry $\CZ\in \CX_\loc^\CL(\CF)$ 
actually defines a family of conserved currents: One may add an arbitrary horizontally closed 
$(d-1,0)$-form $K_Z$ and similarly $\iota_{\pr Z} \dd_M T'$ for an arbitrary horizontally (closed and hence) exact $(d-1,1)$-form.

\begin{example}[\bf O($n$)-model conserved currents]
Consider the case of the O(n)-model Lagrangian $\CL$ from Ex. \ref{VectorValuedFieldTheoryLagrangian}, and the vector fields 
$\CZ^A(\phi) = A^{a}_{\, b} \cdot \phi^b \cdot \frac{\delta}{\delta \phi^a}$ from Ex. \ref{DiffeomorphismForVector-valuedFieldTheoryViaTarget} 
and $\CZ^\nu(\phi) = \mathbb{L}_\nu (\phi)^a  \cdot \frac{\delta}{\delta \phi^a}= \nu^\mu \cdot \partial_\mu \phi^a \cdot \frac{\delta}{\delta \phi^a}$ 
from Ex. \ref{DiffeomorphismForVector-valuedFieldTheoryViaBase}. 
These are local with corresponding the evolutionary vector fields $Z^A= A^{a}_{\,b} \cdot u^b \cdot \frac{\partial}{\partial u^a}$ 
and $Z^{\nu}(\phi) = \nu^\mu \cdot u^a_\mu \cdot \frac{\partial}{\partial u^a}$, respectively.

\vspace{1mm} 
\noindent {\bf {(i)}} Choosing $A=(A^a_{\,b})\in \mathfrak{o}(n)$, it is immediate that $\CZ^A(\CL)(\phi):= \mathbb{L}_{\pr Z^A}(L) \big(j^\infty \phi) = 0$ 
due to the $O(n)$-invariance of the Euclidean product $\langle -,-\rangle $ on the target $W$. Recalling the variational decomposition 
$\delta L = EL + \dd_H \theta_L$ with $\theta_L= - \langle \dd_V u , \star \dd_H u \rangle_g $ from Eq. \eqref{O(n)ModelJetBundleVariationalDecomposition},
the corresponding on-shell conserved current is
\vspace{-2mm}
\begin{align*}
\CP_{\CZ^A} (\phi)  &= 0+ \iota_{\CZ^A} \theta_\CL  (\phi) := (\iota_\pr \theta_L) (j^\infty \phi) 
\\ 
&= - \big(Z^{A,a}\wedge \star \dd_H u_a\big) \circ j^\infty \phi  = - \CZ^A(\phi)^a \wedge \star \dd_M\phi 
\\ &= -A^a_{\, b} \cdot \phi^b \wedge \star \dd_M \phi_a  \;.
\end{align*}

\vspace{-2mm}
\noindent
{\bf {(ii)}} Choosing $\nu=\nu^\mu \cdot \frac{\partial}{\partial x^\mu}\in \CX(M)$ to be a Killing vector field of the metric $g$, $\mathbb{L}_\nu (g)=0$, 
it follows that $\CZ^\nu(\CL)(\phi):= \mathbb{L}_{\pr Z^\nu}(L) \big(j^\infty \phi) =  \mathbb{L}_{\nu} \big(\CL(\phi)\big)= \dd_M \iota_\nu\big(   \CL(\phi)\big)$. 
The corresponding on-shell conserved current is
\vspace{-2mm}
\begin{align*}
\CP_{\CZ^\nu} (\phi)  &=  \iota_\nu \big(  \CL(\phi)\big) + \iota_{\CZ^\nu} \theta_\CL  (\phi) \\ 
&=  \iota_\nu \big(  \CL(\phi)\big) -\mathbb{L}_\nu(\phi)^a \wedge \star \dd_M \phi_a \, ,
\end{align*}

\vspace{-2mm}
\noindent which may be further expanded in local coordinates as in Ex. \ref{VectorValuedFieldTheoryLagrangian}. 
In the case of Minkowski spacetime with its Killing Lie algebra being the Poincar\'e Lie algebra, these currents comprise the energy-momentum tensor.
\end{example}

\begin{remark}[\bf Integrating infinitesimal local  symmetries]\label{IntegratingInfinitesimalLocalSymmetries}
Although every 1-parameter local diffeomorphism induces a local vector field on field space (Ex. \ref{DifferentiatingFiniteLocalDiffeomorphisms}), 
it is not necessarily the case that every local vector field integrates to a 1-parameter local diffeomorphism. Indeed, even in the case where 
$\CZ = Z \circ j^\infty$ for some $Z$ that globally factors through a finite order jet $J^k_M F$, it generally only integrates to a local 
flow  on $J^k_M F$, which is not enough to define a diffeomorphism on the full field space $\CF=\mathbold{\Gamma}_M(F)$. Moreover, it might 
be that a local vector field integrates to a spacetime covariant diffeomorphism on $\CF$ instead (see Lem. \ref{InfinitesimalSpacetimeCovariantSymmetries}). 
Nevertheless, we stress that Noether's First theorem (and Second, see Prop. \ref{Noether2nd}) applies for \textit{any} infinitesimal
local symmetries, and the existence of corresponding conserved charges is independent of the integrability properties. 
\end{remark}

Let us close off this subsection by proving how the infinitesimal versions of finite spacetime covariant symmetries 
(Def. \ref{FiniteSymmetryofLagrangianFieldTheory}) are in fact \textit{local} infinitesimal symmetries, justifying our focus -- and generally 
of the physics literature -- on the latter. 

\begin{proposition}[{\bf Infinitesimal spacetime covariant symmetries are local}]\label{InfinitesimalSpacetimeCovariantSymmetries}
Let $\CD_t$ be a smooth 1-parameter family of spacetime covariant symmetries of a classical field theory $(\CF,\CL)$ starting at the identity. 
That is,
\vspace{-2mm} 
$$
\CL \circ \CD_{t} = f_t^* \circ \CL + f^*_t \circ \dd_M \CK_t 
$$

\vspace{-2mm} 
\noindent where $\CD_{t}=D_t \circ j^\infty(-) \circ (\id,f^{-1}_t) \in [\CF,\CF](\FR^1)$ is induced by a smooth 1-parameter family of bundle maps 
\vspace{-2mm} 
\[ 
\xymatrix@C=4em@R=.4em  {\FR^1\times J^\infty_M F \ar[dd] \ar[rr]^{D_t} &   &  F \ar[dd]
	\\  \\ 
\FR^1\times M \ar[rr]^{f_t} &  &  M \, .
}    
\]

\vspace{-2mm}
\noindent covering a 1-parameter family of diffeomorphisms $f_t:\FR^1\times M\rightarrow M$, for some $\CK_t=K_t \circ j^\infty$ where
$K_t: \FR^1\times J^\infty_M F \rightarrow \wedge^{d-1} T^*M$ is a 1-parameter family of smooth bundle maps over $M$. Then: 

\noindent {\bf (i)} The vector field $\partial_t \CD_t |_{t=0}: \CF \rightarrow T\CF$ is local 
\vspace{-2mm}
\begin{align*}
\partial_t \CD_t |_{t=0}= \CZ_{\mathrm{ev}}= Z_{\mathrm{ev}} \circ j^\infty \hspace{0.5cm} \in \hspace{0.5cm} \CX_{\mathrm{loc}}(\CF) \, ,
\end{align*}

\vspace{-1mm}
\noindent for an evolutionary vector field $Z_{\mathrm{ev}}: J^\infty_M F\rightarrow VF$ over F.

\noindent {\bf (ii)} The induced infinitesimal action is an infinitesimal local symmetry of $(\CF,\CL)$
\vspace{-2mm} 
\begin{align*}
 \partial_t (\CL \circ \CD_{t}) |_{t=0} \equiv \CZ_\mathrm{ev}(\CL) =  \dd_M (\CK_{Z_\ev})
\end{align*}

\vspace{-2mm} 
\noindent 
for some $K_{Z_\mathrm{ev}}: J^\infty_M F\rightarrow \wedge^{d-1} T^*M$ over M.
\end{proposition}

\begin{proof} We show the calculation at the level of $*$-plots of fields, with that of higher $\FR^k$-plots being analogous. 
Running through the differentiation of Ex.  \ref{VectorFieldsAsInfinitesimalDiffeomorphisms}
(see also Ex. \ref{DifferentiatingFiniteLocalDiffeomorphisms}), it follows that
\vspace{-2mm}
\begin{align*}
\partial_t \CD_t(\phi)&= \partial_t D_t|_{t=0} \circ j^\infty(\phi) \, \circ \, \id_M + \partial_t \big(\pi^\infty_0 \circ j^\infty\phi \, \circ \, f_t^{-1}\big) |_{t=0} \\
&= \partial_t\big( \pi^\infty_0 \circ \pr D_t)|_{t=0} \, \circ \, j^\infty \phi + \dd (\pi^\infty_0  \, \circ \, j^\infty \phi)\circ  \big(\partial_t f^{-1}_t |_{t=0}\big) \\
&= \dd \pi^\infty_0 \, \circ \, (\partial_t \pr D_t |_{t=0}) \, \circ \, j^\infty \phi + \dd \pi^\infty_0 \, \circ \, \dd (j^\infty \phi) \big(\partial_t f^{-1}_t |_{t=0}\big) \\
&=: \dd \pi^\infty_0 \, \circ \, \pr Z \circ j^\infty \phi - \dd \pi^\infty_0 \, \circ \, H(X) \circ j^\infty \phi \, ,
\end{align*}

\vspace{-2mm}
\noindent where the first line follows by the product rule, the second by the definition (Def. \ref{ProlongationOfJetBundleMap}) of the prolongation 
$\pr D_t : \FR^1\times J^\infty_M F\rightarrow J^\infty_M F$ and the chain rule, while the third again by the chain rule. The last line follows by 
defining the vector field $\pr Z:= \partial_t \pr D_t |_{t=0} : J^\infty_M F \rightarrow T J^\infty_M F $ and recalling\footnote{We use
$H (X)_{j^\infty_p \phi} = H_{j^\infty_p\phi}\big( X (j^\infty_p \phi)\big)$ as a shorthand for 
$H\big((\pi^\infty_M)^* X\big)(j^\infty_p \phi) = H_{j^\infty_p\phi}\big((\pi^\infty_M)^* X (j^\infty_p \phi)\big)$, since the splitting is really 
a map $H:J^\infty_M F\times_M TM \rightarrow TJ^\infty_M F$.} that $H_{j^\infty_p\phi}\big( X (j^\infty_p \phi)\big):= \dd(j^\infty \phi)_p X_p$ 
defines the horizontal lift of tangent vectors on $M$ to tangent vectors on $J^\infty_M F$ \big(see Eq. \eqref{InftyJetBundleHorizontalTangentCoordinate}, 
Prop. \ref{SmoothSplittingProp}\big), where we define $X:= \partial_t f_t |_{t=0}:M\rightarrow TM$, hence that $\partial_t f^{-1}_t |_{t=0}=-X$. 

We highlight that the vector field $\pr Z:J^\infty_M F\rightarrow TJ^\infty_M F$ is \textit{not}  vertical and in particular not the prolongation of some evolutionary vector
field (unless $f_t=\id_M$). However, the horizontal splitting of $TJ^\infty_M F$ decomposes the vector field as $\pr Z= (\pr Z)_V + H(X) $, with the horizontal component
being necessarily the lift of $X$ since $Z$ covers $X$ -- by construction. Working in local coordinates (see \eqref{HorizontalVerticalVectorFieldCoordinates}) and 
using the local characterization of prolongated evolutionary vector fields \eqref{ProlongationEvolutionayVectorFieldlocally}, one sees that the vertical component is
in fact necessarily the prolongation $(\pr Z)_V = \pr Z_\mathrm{ev}$ of an evolutionary vector field  $Z_{\mathrm{ev}}:J^\infty_M F\rightarrow VF$ over $F$
(see also \cite[Prop. 1.20]{Anderson89}  for detailed coordinate formulas). Thus
\vspace{-2mm}
\begin{align*}
\partial_t \CD_{t}|_{t=0}&= \dd \pi^\infty_0 \circ \pr Z_{\mathrm{ev}} \circ j^\infty + \dd \pi^\infty_0 \circ H(X)\circ j^\infty  - \dd \pi^\infty_0 \circ H(X)\circ j^\infty   \\
 &=  Z_{\mathrm{ev}} \circ j^\infty =: \CZ_{\mathrm{ev}}
\end{align*} 

\vspace{-2mm}
\noindent is in fact a \textit{local} vector field $\CF$. Completely analogously, it follows that the infinitesimal action \footnote{A priori, the left-hand side is defined only through Rem. \ref{AlternativeDefinitionofLocalVectorFieldAction}.} is given by
\vspace{-2mm}
\begin{align*}
    \partial_t \big(\CL \circ \CD_t\big) |_{t=0}&= \pr Z_\mathrm{ev}(L) \circ j^\infty + \mathbb{L}_X ( \CL) - \mathbb{L}_X \circ  (\CL )  \\
    &= \CZ_\mathrm{ev} (\CL) \, .
\end{align*}

\vspace{-2mm}
\noindent with the latter action being that of Def. \ref{ActionOfLocalVectorFieldsOnCurrents}.

Assuming further that $\CD_t$ is a 1-parameter covariant symmetry of the Lagrangian, we furthermore get 
\vspace{-2mm}
\begin{align*}
    \partial_t \big(\CL \circ \CD_t\big) |_{t=0}&= \partial_t \big( f^*_{t} \circ \CL  +  f^*_t \circ \dd_M \CK_t ) |_{t=0}\\
    &= \mathbb{L}_X (\CL) + \mathbb{L}_X (\dd_M \CK_{0}) + \dd_M(\partial_t \CK_t |_{t=0}) \\
    &= \dd_M \big(\iota_X \CL + \iota_{X} \dd_M \CK_0 + \dot{\CK}_0  \big) \, ,
\end{align*}

\vspace{-2mm}
\noindent where in the first line we used the product rule, then the Cartan calculus on $M$ and by defining 
$\dot{\CK}_0= \dot{K}_0\circ j^\infty $ where $\dot{K}_0:= \partial_t K_t |_{t=0}: J^\infty_M F \rightarrow \wedge^d T^*M $ 
(see Ex. \ref{DifferentiatingFiniteLocalSymmetries}). Lastly, it follows immediately in coordinates that
$\iota_X \CL(\phi):= \iota_X (L\circ j^\infty\phi) = 
(\iota_{H(X)} L ) \circ j^\infty \phi $ since the Lagrangian density $L$ is a horizontal form, and similarly for the second term. 
Combining this with the compatibility of Lem. \ref{HorizontalDifferentialBaseDeRhamCompatibility}, we arrive at
\vspace{-2mm}
$$
\partial_t \big(\CL \circ \CD_t\big) |_{t=0} = \dd_M (K_{Z_{\mathrm{ev}}} \circ j^\infty) 
$$

\vspace{-1mm} 
\noindent
where $K_{Z_{\mathrm{ev}}}= \iota_{H( X)} \big( L + \dd_H K_0 \big) + \dot{K}_0 \, \in \, \Omega^{d-1,0}(J^\infty_M F)$, which completes the proof.
\end{proof}

The above, somewhat abstract, statement applies in virtually all local field theories with spacetime covariant symmetries (General relativity,
Yang--Mills, Chern--Simons). It is considerably easier to interpret the calculation in explicit examples. 
\begin{example}[\bf Differentiating spacetime symmetry of O($n$)-model]\label{InfinitesimalSpacetimeSymmetryOfO(N)-model}
Recall the O($n$)-model field theory of Ex. \ref{VectorValuedFieldTheoryLagrangian}, with field space $\CF=[M,W]$ and local Lagrangian 
\vspace{-2mm}
$$
\CL(\phi)=\tfrac{1}{2}\big(\langle \dd_M\phi\, ,\,  \dd_M\phi \rangle_g +  c_2 \cdot \langle \phi , \phi \rangle + \tfrac{1}{2} c_4 \cdot 
(\langle \phi,\phi \rangle)^2 \big) \cdot \dd \mathrm{vol}_g \, ,
$$

\vspace{-2mm}
\noindent where $g$ is a background metric on $M$ and $\langle-,-\rangle$ an inner product on $W$.
By Ex. \ref{LocalSymmetriesOfO(n)Model}, any isometry $f:M\rightarrow M$ induces a spacetime covariant symmetry via 
the pullback action $\CD=f^*: [M,W]\rightarrow [M,W]$
\vspace{-2mm}
\begin{align}\label{O(N)spacetimecovariantsymmetryEqRef}
\CL\circ f^*(\phi) = f^* \circ \CL(\phi) \, .
\end{align}

\vspace{-2mm}
\noindent 
Consider any 1-parameter family of isometries $f_t: \FR^1\times M \rightarrow M$ starting at the identity, with induced Killing vector field $v=\partial_t f_t |_{t=0} \in \CX(M)$ on spacetime. Differentiating the corresponding 1-parameter family action on $\CL$, i.e., the left-hand side of Eq. \eqref{O(N)spacetimecovariantsymmetryEqRef}, it immediately follows that 
\vspace{-2mm}
\begin{align*}
 \partial_{t} \CL (f^*_t \phi) |_{t=0} &= \big(\langle \dd_M( \mathbb{L}_v \phi)\, ,\,  \dd_M(\phi) \rangle_g +  c_2 \cdot \langle  \mathbb{L}_v \phi , \phi \rangle + c_4 \cdot 
(\langle \mathbb{L}_v\phi,  \phi \rangle)^2 \big) \cdot \dd \mathrm{vol}_g \, .   
\end{align*}

\vspace{-2mm}
\noindent Next recall the vector field on field space $\CZ^v:\CF \rightarrow T\CF$ from Ex. \ref{DiffeomorphismForVector-valuedFieldTheoryViaBase} corresponding to $f_t^*:\FR^1\times \CF\rightarrow \CF$. It is immediately seen to be \textit{local} as 
\vspace{-2mm}
$$
\CZ^\nu(\phi) = \mathbb{L}_v(\phi^a) \cdot \frac{\partial}{\partial u^a}= Z^\nu \circ j^\infty(\phi)
$$

\vspace{-1mm}
\noindent for $Z^\nu: J^\infty_M F \rightarrow VF$ defined locally by $\nu^\mu \cdot u^a_\mu \cdot \frac{\partial}{\partial u^a}$. It follows by the explicit formula above that the infinitesimal action the diffeomorphism $f^*:\CF\rightarrow \CF$ on the Lagrangian is equivalently given by the action of the corresponding local vector field (Def. \ref{ActionOfLocalVectorFieldsOnCurrents})
\vspace{-2mm}
$$ 
\partial_{t} \CL (f^*_t \phi) |_{t=0} = \CZ^\nu(\CL) (\phi) =  \big(\langle \dd_M( \mathbb{L}_v \phi)\, ,\,  \dd_M(\phi) \rangle_g +  c_2 \cdot \langle  \mathbb{L}_v \phi , \phi \rangle + c_4 \cdot 
(\langle \mathbb{L}_v\phi,  \phi \rangle)^2 \big) \cdot \dd \mathrm{vol}_g \, .
$$

\vspace{-1mm}
\noindent
At this point, either by using the fact that $\nu\in \CX(M)$ is a Killing vector field or by differentiating the right-hand side of Eq. \eqref{O(N)spacetimecovariantsymmetryEqRef}, it furthermore follows that the infinitesimal action satisfies
\vspace{-2mm}

$$
\partial_{t} \CL (f^*_t \phi) |_{t=0} = \CZ^\nu(\CL) (\phi) = \dd_M \big (\iota_\nu \CL(\phi) \big) \, , 
$$
\vspace{-2mm}
\noindent
i.e., an infinitesimal local symmetry as per Lem. \ref{InfinitesimalSpacetimeCovariantSymmetries}.

\end{example}

\subsection{Infinitesimal gauge symmetries and Noether's Second Theorem}
\label{Sec-gauge} 

There is a class of infinitesimal local symmetries that induce redundancies in the physical interpretation of classical 
field theories, in the sense that they obstruct (Prop. \ref{LocalGaugeSymmetryObstructsCauchySurface}) the existence of a `\textit{Cauchy surface}' (Def. \ref{CauchySurface}), i.e., a well-defined set of initial conditions which induce a \textit{unique} evolution of fields via the Euler--Lagrange
equations.
These are symmetries that can be freely parametrized by `gauge' parameters which are functions, or more generally sections
of vector bundles, on the spacetime $M$. Such symmetries induce currents that are conserved off-shell, 
\footnote{In general, this holds up to a trivial current that vanishes on-shell (Cor. \ref{TrivialCurrentsFromGaugeSymmetries}).}
which may 
equivalently be interpreted as inducing interrelations between the components of the Euler--Lagrange differential operator.

\begin{lemma}[{\bf Motivating case of a gauge symmetry}]\label{LocalGaugeSymmetrySimpleExample}
Let $\CZ = Z\circ j^\infty \in \CX(\CF)$ be a local vector field. Suppose the local vector field 
$\mathcal{Z}_f = Z_f\circ j^\infty := (f\cdot Z) \circ j^\infty  \in \CX(\CF)$ is an infinitesimal symmetry of the smooth 
Lagrangian $\CL:\CF \rightarrow \Omega^{d}_\mathrm{Vert}(M)$, for any $f\in C^\infty(M)$. That is, 
$$
\CZ_f (\CL) = \dd_M \CK_{Z_f} \, ,
$$
for some $K_{Z_f}\in \Omega^{d-1,0}(J^\infty_M F)$, for every $f\in C^\infty(M)$. 

\noindent {\bf (i)} 
Then the $(d-1)$-form current $\CP_{Z}:\CF \rightarrow \Omega^{d-1}_{\mathrm{Vert}}(M)$ corresponding to $\CZ \in \CX(\CF)$ is conserved \textit{off-shell} 
\vspace{-1mm} 
$$
\dd_M \CP_Z =  0 \, : \, \CF\longrightarrow \Omega^{d}_{\mathrm{Vert}}(M)\, ,
$$

\vspace{-2mm} 
\noindent and not only on the smooth space of on-shell fields $\CF_{\CE \CL}\hookrightarrow \CF$. 

\noindent {\bf (ii)}  Equivalently, there exist the \textit{off-shell} relation between the Euler-Lagrange 
differential operators of $\CL$
\vspace{-1mm} 
$$
\langle \CE \CL, \mathcal{Z} \rangle = \CE \CL_a \cdot \mathcal{Z}^a = 0 \, .
$$
\begin{proof}
Since $(f\cdot Z) \circ j^\infty $ is a symmetry of $\CL$, by Noether's First Theorem 
\vspace{-2mm} 
\begin{align*}
\dd_{M} \CP_{f\cdot Z} &= 
\langle \CE \CL , (f\cdot Z) \circ j^\infty \rangle = \langle \CE \CL, Z\circ j^\infty\rangle \cdot f \\
&= \dd_M \CP_Z \cdot f
\end{align*}

\vspace{-2mm} 
\noindent for any $f\in C^\infty(M)$. In other words, for any field configuration $\phi \in \CF$ (and any higher plot),
$\dd_M \CP_Z(\phi) \cdot f \in \Omega^{d}(M)$ is an exact top-form for all $f\in C^\infty(M)$. 
Integrating over closed balls $\bar{B}^d_{p}\subset M$ around any point $p\in M$ in the base manifold, 
$$
\int_{\bar{B}^d_p} \dd_M \CP_Z(\phi) \cdot f =0 
$$
for all $f\in C^\infty(M)$, which implies that $\dd_M \CP_Z(\phi)|_{B^d_p} =0$ in the interior. Thus, $\CP_Z(\phi)$ is
locally - and hence globally closed for any field configuration $\phi\in \CF$. \footnote{Strictly speaking, this only shows that $\CP_Z(\phi)$
is closed in the \textit{interior} $\mathrm{int}(M)\subset M$.} It follows that, as a smooth current on field space,
$$
0=\dd_M \CP_Z: \CF \longrightarrow \Omega^{d}_{\mathrm{Vert}}(M)\, , 
$$
and so $\CP_Z$ is conserved off-shell. By the formula $\dd_M \CP_Z = \langle \CE \CL, \CZ \rangle $, being conserved off-shell is 
equivalent to the identities
$$
\langle \CE \CL, \CZ \rangle= \CE \CL_a \cdot \
\CZ^a =0 
$$

\vspace{-2mm} 
\noindent between the components of the smooth Euler--Lagrange operator.
\end{proof}
\end{lemma}
The illustration above captures the main features of gauge symmetries. The `gauge' parameters 
$f\in C^\infty(M)\cong \Gamma_M(M\times \FR)$ parametrize an (infinite-dimensional) family of symmetries
$\{\CZ_f:= (f\cdot Z)\circ j^\infty \}\subset \CX_\mathrm{loc}(\CF)$, whose dependence on $f$ is $\textit{local}$. 
In the above example, $\CZ_f$ depends on f via its value at points, i.e., the zero-jet of $f$. These induce non-trivial interrelations between the components of the Euler--Lagrange operator, and the resulting current is conserved off-shell.

\medskip 
More generally, an infinitesimal local gauge symmetry is an infinite-dimensional subspace 
$\{\CR_e:= R_{e}\circ j^\infty \}\subset \CX_\mathrm{loc}^{\CL}(\CF)$ parametrized by `\textit{gauge parameters}', sections 
$e\in \Gamma_M(E)$ of a gauge parameter vector bundle $E\rightarrow M$, such that the dependence of each $\CR_e$ on $e$ is 
linear and local, i.e.,
via the infinite jet prolongation of $e$. The resulting currents satisfy similar off-shell properties, which are equivalent 
to a set of identities relating the components of the Euler--Lagrange operator. We now turn to make this intuition precise.

\begin{definition}[\bf  Infinitesimal gauge symmetry]\label{ParametrizedGaugeSymmetries}
A (parametrized) collection of  \textit{infinitesimal local gauge symmetries} of a Lagrangian 
density $L$ is an infinite-dimensional subspace of evolutionary vector fields on $J^\infty_M F$ that may be identified with 
the image of a $\FR$-linear map
\vspace{-2mm} 
\begin{align*}
R_{(-)} : \Gamma_M(E) &\longrightarrow \CX_\mathrm{ev}(J^\infty_M F)\, \\ 
e &\longmapsto R\circ (j^\infty e, \id_{J^\infty_M F})
\end{align*}

\vspace{-2mm} 
\noindent where $E\rightarrow M$ is a ``\textit{gauge parameter}'' vector bundle over M, and
\vspace{-2mm} 
\[ 
\xymatrix@C=1.8em@R=.4em  { J^\infty_M E\times_M J^\infty_M F 
 \ar[rd] \ar[rr]^{\hspace{0.8cm} R} &   &  VF \ar[ld]
	\\ 
& F & 
}   
\]

\vspace{-2mm} 
\noindent is a smooth bundle map linear in the $J^\infty_M E$ fibers, \footnote{The fiber product $J^\infty_M G\times_M J^\infty_M F$
is taken in $\mathrm{LocProMan}\hookrightarrow \FrechetManifolds$. Equivalently, it is the infinite jet bundle $J^\infty_M (G\times_M F)$ 
on the fiber product of the bundles $G$ and $F$, and so manifestly a locally pro-manifold.}
such that each evolutionary vector field 
\vspace{-2mm} 
$$
R_{e}:J^\infty_M F \cong  M\times_M J^\infty_M F \xlongrightarrow{\;\;\big(j^\infty e,\;  \id_{J^\infty_M F}\big)\;\;} 
J^\infty_M E \times_M J^\infty_M F \xlongrightarrow{\;R\;} VF \,  
$$

\vspace{-1mm} 
\noindent is a symmetry of the Lagrangian density
\vspace{-1mm} 
$$
\mathbb{L}_{\pr R_e} L = \dd_H K_{R_{e}}\, ,
$$
for some $K_{R_{e}}\in \Omega^{d-1,0}(J^\infty_M F)$ with local dependence on $e$, i.e., given by a bundle map 
$K_R : J^\infty_M E \times_M J^\infty_M F \rightarrow \wedge^{d-1}T^*M$ over $M$. 
\end{definition}

Composing with the jet 
prolongation along $F$, this corresponds to a subspace 
of local vector fields of $\CF$, given by the image of the (smooth) $\FR$-linear map 
$$
\CR_{(-)}:= R_{(-)}\circ j^\infty \;:\; \Gamma_M(E) \longrightarrow \CX_{\mathrm{loc}}(\CF) \, ,
$$
\noindent such that for each $e\in \Gamma_M(E)$ it defines a local symmetry of the local Lagrangian $\CL$
\vspace{-1mm} 
$$
\CR_e(\CL)= \dd_M \CK_{R_e}\, .
$$ 

\vspace{-2mm}
\noindent Let us expand the definition in local coordinates, where it appears in a more familiar form. Let $\{x^\mu,c^\beta\}$ be a compatible 
coordinate chart for $E\rightarrow M$ with induced coordinates $\big\{x^\mu,\{c^\beta_K\}_{0\leq |K|}\big\}$ for $J^\infty_M E\rightarrow M$. 
Then by the linearity assumption, the gauge parameter bundle map $R:J^\infty_M E \times_M J^\infty_M F \rightarrow VF$ is locally of the form
\vspace{-2mm} 
\begin{align}\label{GaugeParameterBundleMapLocally}
R&= \sum_{|K|=0}^\infty c^\beta_K \cdot R^{aK}_\beta \cdot \frac{\partial}{\partial u^a} \\[-1pt] \nn
&=\Big(c^\beta \cdot R_\beta^a + c^\beta_\mu \cdot \R_\beta^{a\mu} + c^\beta_{\mu_1\mu_2}\cdot R_\beta^{a \mu_1 \mu_2}
+\cdots \Big) \cdot \frac{\partial}{\partial u^a}
\end{align}
with the sum necessarily (locally) terminating, where each of the coefficients $\{R^a_\beta, R^{a\mu}_\beta, R^{a\mu_1\mu_2}_\beta, \cdots \}$ 
is a smooth function on $J^\infty_M F$, and hence also of (locally) finite order in the coordinates $\big\{x^\mu, \{u^a_I\}_{0\leq |I|}\big\}$. 
Thus, for a gauge parameter $e\in \Gamma_M(E)$ the corresponding evolutionary vector field takes the form
\vspace{-2mm}
\begin{align*}
R_e&=\sum_{|K|=0}^\infty \frac{\partial e^\beta}{\partial x^K} \cdot R^{aK}_\beta \cdot \frac{\partial}{\partial u^a}  \nn\\[-1pt]
&=\Big(e^\beta\cdot  R_\beta^a + \frac{\partial e^\beta}{\partial x^\mu} \cdot \R_\beta^{a\mu} + 
\frac{\partial e^\beta}{\partial x^{\mu_1} \partial x^{\mu_2}}\cdot R_\beta^{a \mu_1 \mu_2}  
+\cdots \Big) \cdot \frac{\partial}{\partial u^a} \, .
\end{align*}
By abuse of notation, the corresponding local vector field on $\CF$ is often denoted by
\vspace{-1mm} 
\begin{align*}
\CR_e(\phi)&= \sum_{|K|=0}^\infty \frac{\partial e^\beta}{\partial x^K} \cdot \CR^{aK}_\beta(\phi) \cdot \frac{\delta}{\delta \phi^a}= \sum_{|K|=0}^\infty \frac{\partial e^\beta}{\partial x^K} \cdot R^{aK}_\beta(\phi,\{\partial_I\phi\}_{|I|\leq k}) \cdot \frac{\delta}{\delta \phi^a} 
\nn\\[-1pt]
&= \Big(e^\beta \cdot \CR_\beta^a(\phi) + \frac{\partial e^\beta}{\partial x^\mu} \cdot \CR_\beta^{a\mu}(\phi)
+ \frac{\partial e^\beta}{\partial x^{\mu_1} \partial x^{\mu_2}}\cdot \CR_\beta^{a \mu_1 \mu_2}(\phi)  +\cdots \Big) \cdot
\frac{\delta}{\delta \phi^a} \, ,
\end{align*}
or even as an `infinitesimal transformation of the field'
\vspace{-1mm} 
$$
\delta_{R_e} \phi^a= e^\beta \cdot \CR_\beta^a(\phi) + \frac{\partial e^\beta}{\partial x^\mu} \cdot \CR_\beta^{a\mu}(\phi)
+ \frac{\partial e^\beta}{\partial x^{\mu_1} \partial x^{\mu_2}}\cdot \CR_\beta^{a \mu_1 \mu_2}(\phi)  +\cdots\, ,
$$
exactly matching the formulae appearing in the physics literature (see e.g. \cite{HenneauxTeitelboim92}).

\begin{example}[\bf Electromagnetic gauge symmetry]\label{ElectromagnetismGaugeSymmetry}
The archetypical example of a gauge symmetry is that of pure electromagnetism, expressed in terms of gauge potentials.
Consider spacetime $(M,g)$ as a (pseudo)-Riemannian $d$-dimensional manifold, and field bundle as the cotangent bundle 
$F= T^*M$ with induced coordinates $\{x^\mu, u_\mu\}$. The smooth set of electromagnetic fields is 
$\CF=\mathbold{\Gamma}_M(T^*M)\cong \Omega^1_{\mathrm{Vert}}(M)$, and the local Lagrangian of pure electromagnetism is 
\vspace{-2mm} 
\begin{align*}
 \CL(A) &:= \tfrac{1}{2} \cdot F_A \wedge  \star F_A  = \tfrac{1}{2}\langle F_A , F_A \rangle_g \cdot \dd \mathrm{vol}_g  \equiv \tfrac{1}{2}\langle \dd_M A ,\dd_M A \rangle_g \cdot \dd \mathrm{vol}_g  
 \\ &= \tfrac{1}{2} g^{\mu \nu} g^{\sigma \rho}\cdot  \partial_\sigma  A_\mu \cdot \partial_\rho  A_\nu  \cdot \sqrt{|g|}\cdot \dd x^1 \cdots \dd x^d \, ,  
\end{align*}
\noindent 

\vspace{-2mm} 
\noindent where $F_A:= \dd_M A $ is the field strength\footnote{For a spacetime of the form $\FR^1_t \times N$, 
the traditional electric and magnetic field strengths are recovered via the decomposition 
$F= E\wedge \dd t + B = E_{it} \dd x^i \wedge \dd t + B_{i j} \dd x^i \wedge \dd x^j$.} (curvature) of the 
gauge potential (connection) $A\in \Omega^1(M)$. Since $F= T^*M$ is a vector bundle it follows that $VF\cong T^*M \times_M T^*M$, and so  $T\CF \cong \CF\times \CF$ (see also Rem. \ref{ScalarFieldTheoryTangetVectors}). Hence,
for each function $e\in C^\infty(M)\cong \Gamma_M(M\times \FR)$ on spacetime there exists the \textit{constant} vector field 
\vspace{-2mm}
\begin{align*}
\CR_{e}: \CF &\longrightarrow T\CF \cong \CF\times \CF \\
A& \longmapsto (A, \dd_M e) \, ,
\end{align*}

\vspace{-2mm}
\noindent which in local coordinates may be represented by
\vspace{-2mm}
$$
\CR_e(A):=  \partial_\mu e \cdot \frac{\delta}{\delta A_\mu} =  \partial_\mu e \cdot \delta^{\mu}_\nu \cdot  \frac{\delta}{\delta A_\nu}\quad \in \quad \CX_\loc(\CF)\, .
$$

\vspace{-1mm}
\noindent 
 Note that the dependence on the chosen function $e\in \Gamma_M(M\times \FR)$ is through $J^1_M (M\times \FR)$. It is infinitesimal symmetry of the Lagrangian 
 as can be checked in coordinates via Def. \ref{ActionOfLocalVectorFieldsOnCurrents}, or directly by the nilpotency of the de Rham differential
\vspace{-2mm}
\begin{align}\label{ElectroMagnetismGaugeInvarianceEq}
\CR_e(\CL) (A) &= \tfrac{1}{2} \big(  \dd_M( \dd_M e) \wedge  \star \dd_M A  + \dd_M A \wedge \star \dd_M(\dd_M e) \big) \nn
\\ &= \dd_M^2 e \wedge \star \dd_M A =0 \, .  
\end{align}

\vspace{-2mm}
\noindent Thus, it defines a parametrized gauge symmetry, with gauge parameter bundle $E=M\times \FR$ and coordinates $\{x^\mu, c\}$, and the corresponding 
bundle map $R:J^\infty_M (M\times \FR) \times_M J^\infty_M (T^*M)\longrightarrow T^*M\times_M T^*M \cong V(T^*M)$ over $T^*M$ given locally by
\vspace{-2mm}
$$
c_\mu \cdot \delta^\mu_\nu \cdot \frac{\partial}{\partial u_\mu } \, .
$$

\vspace{-2mm}
\noindent
For completeness, let us note the electromagnetism Euler--Lagrange operator, i.e., the (free) Maxwell equations take the form
\vspace{-2mm}
\begin{equation}
\label{ElectroMagnetismEulerLagrangeOp}
\CE \CL(A) = \dd_M \star F_A = \dd_M \star \dd_M A \quad \in \quad \Omega^{d-1}(M) ,
\end{equation}

\vspace{-0mm}
\noindent where $\Omega^{d-1}(M) \cong \Gamma_M(\wedge^d T^*M \times_{M}  TM) \cong \Gamma_M(\wedge^d T^*M \otimes V^*F)$, as expected. 
The current corresponding to each infinitesimal symmetry $\CR_e\in \CX_\loc(\CF)$ is given by
\vspace{-2mm}
\begin{align}\label{ElectroMagnetismGaugeCurrent}
\CP_{R_e}(A) = \dd_M e \wedge \star \dd_M A  \, ,
\end{align}

\vspace{-2mm} \noindent which is obviously conserved on $\CF_{\CE \CL}$. 
\end{example}

\begin{remark}[\bf The Lie algebra of all infinitesimal gauge symmetries]\label{LieAlgebraOfAllGaugeSymmetries} For the sake of completeness,
we comment on some abstract aspects of gauge symmetries, defined via arbitrary parametrizations as above. Some of these points are explicitly 
mentioned and further expanded in \cite{HenneauxTeitelboim92}, while others are implicitly necessary and slightly supplement the
(local coordinate) description of gauge symmetries therein. 

\noindent {\bf (i)} The \textit{actual} gauge symmetries $\CX_{\mathrm{loc}}^{\CL,\CR}(\CF)\subset \CX_{\mathrm{loc}}^{\CL}(\CF)$ 
are identified by the \textit{image} of \textit{some} local $\FR$-linear map $\CR_{(-)}:\Gamma_M(E)\rightarrow \CX_\mathrm{loc}^{\CL}(\CF)$ as in Def. \ref{ParametrizedGaugeSymmetries}. In particular, the gauge parameter $E\rightarrow M$ bundle and corresponding parametrization bundle map $R$ are \textit{not unique}.\footnote{ For instance, let $E'\rightarrow M$ be any other vector bundle and consider the map $\tilde{\CR}_{(-)}: \Gamma_M(E\times_M E')\cong \Gamma_M(E)\oplus \Gamma_M(E')\rightarrow \CX_\mathrm{loc}^{\CL}(\CF)$ acting via $\CR_{(-)}$ on the first component and trivially on the second.} Furthermore, it might be that some gauge symmetries are in the image of some parametrization $\CR_{(-)}: \Gamma_M(E) \rightarrow \CX_\mathrm{loc}^{\CL}(\CF)$, but \textit{not} in the image of another $\hat{\CR}_{(-)}: \Gamma_M(\hat{E}) \rightarrow \CX_\mathrm{loc}^{\CL}(\CF)$. \footnote{For instance, take $\hat{E}\rightarrow M$ to be a proper subbundle of $E\rightarrow M$ for an injective $\CR_{(-)}:\Gamma_M(E)\rightarrow \CX_\mathrm{loc}^{\CL}(\CF)$. Then the restricted parametrization $\hat{\CR}_{(-)}:= \CR_{(-)}|_{\hat{E}} : \Gamma_M(\hat{E})\rightarrow \CX_\mathrm{loc}^{\CL}(\CF)$ identifies strictly less gauge symmetries via its image.}

\noindent {\bf (ii)} By Eq. \eqref{SubalgebraofLocalInfinitesimalSymmetries}, it follows that for any two gauge parameters $e_1,e_2 \in \Gamma_M(E)$ of a (parametrized) gauge symmetry $R$, the commutator $[\CR_{e_1},\CR_{e_2}]$ is also a local symmetry of $(\CF,\CL)$. Crucially, however, it is not necessarily possible to express  $[\CR_{e_1},\CR_{e_2}]$ as $R_{e_3}$ for some other gauge parameter 
$e_3\in \Gamma_M(E)$.\footnote{This is possible for all gauge parameters if there exists a local `bracket map' $[-,-]^\CR_{(-)}: \Gamma_M(E)\times \Gamma_M(E)\times \CF \rightarrow \Gamma_M(J^\infty E)$ such that $\CR_{[e_1,e_2]^\CR_{(\phi)}}(\phi)= [\CR_{e_1}(\phi), \CR_{e_1}(\phi)]_{\CX_{\mathrm{loc}}(\CF)}$. The bracket generally might depend on the dynamical fields (in a local manner), and so the composition $\Gamma_M(E)\times \Gamma_M (E) \times \CF \rightarrow \Gamma_M(J^\infty E)\times \CF \rightarrow T\CF $ defines a bundle map over $\CF$, and so (generally) a \textit{Lie algebroid} structure, rather than a Lie algebra. Nevertheless, for many physical examples (e.g. General Relativity, Yang-Mills, Chern-Simons theories) this happens to be an actual Lie algebra map.} Nevertheless, commutators of (parametrized) gauge symmetries are also considered to be \textit{gauge} symmetries. 

\noindent {\bf (iii)} We have defined an infinite-dimensional family of local symmetries, the \textit{trivial} symmetries $\CX_{\mathrm{loc}}^{\CL,\mathrm{triv}}(\CF)\hookrightarrow \CX_{\mathrm{loc}}^{\CL}(\CF)$ of Ex. \ref{TrivialInfinitesimalSymmetries}. Strictly speaking, these are not parametrized by sections of some vector bundle over $M$, but instead by vector bundle maps out $K: J^\infty_M F\rightarrow \wedge^2 VF \otimes \wedge^d TM$ over $F$. Nevertheless, trivial symmetries are also considered to be \text{gauge} symmetries. 

\noindent {\bf (iv)} It follows that formally, the Lie algebra of all (implicit) infinitesimal gauge symmetries is the minimal subalgebra (in fact ideal) of local symmetries, containing all parametrized local symmetries $\big\{ \CX_{\mathrm{loc}}^{\CL,\CR}(\CF) \hookrightarrow \CX_{\mathrm{loc}}^{\CL}(\CF)\big\}_\CR$ and the trivial symmetries $\CX_{\mathrm{loc}}^\CL(\CF)$, 
\vspace{-1mm}
$$
\big( \CX_{\mathrm{loc}}^{\CL,\mathrm{gauge}}\, , [-,-]\big) \longhookrightarrow \big(\CX^{\CL}_{\mathrm{loc}}(\CF), \, [-,-] \big)\, .
$$

\noindent {\bf (v)} A parametrized gauge symmetry $\CR_{(-)}:\Gamma_M(E)\rightarrow \CX^{\CL}_{\mathrm{loc}}(\CF)$ is called a \textit{generating set} if \text{any} infinitesimal gauge symmetry $\CZ\in  \CX_{\mathrm{loc}}^{\CL,\mathrm{gauge}}$ may be written as 
\vspace{-1mm}
$$
\CZ= R \circ (C \times_F \id_{J^\infty_M F}) \circ j^\infty + \CK \cdot \CE \CL 
$$

\vspace{1mm}
\noindent for some bundle maps $C: J^\infty_M F\rightarrow J^\infty_M E$ and $K:J^\infty_M F \rightarrow \wedge^2 VF \otimes \wedge^d TM$. 
In simple words the vector field $\CZ(\phi)$ may be expressed via $\CR$ using a `field dependent gauge parameter' $\CC(\phi)=C(j^\infty \phi)$, 
up to a trivial gauge transformation $\CK\cdot \CE \CL(\phi)$, and so in local coordinates 
\vspace{-1mm}
$$
\delta_Z \phi^a = \Big( \CC^\beta(\phi) \cdot \CR_\beta^a(\phi) + \frac{\partial \CC^\beta(\phi)}{\partial x^\mu}
\cdot \CR_\beta^{a\mu}(\phi) + \cdots \Big) + \CK^{[ab]}(\phi) \cdot \CE \CL_b (\phi)\, .
$$

\vspace{-1mm}
\noindent There is a lot more to be said along these lines, which is however outside the scope of this manuscript. 
For more details and explicit examples of the above concepts, albeit in the topologically trivial cases, see \cite{HenneauxTeitelboim92}.
\end{remark}

\medskip 

The abstract characterization of the full Lie algebra of (implicit) infinitesimal gauge symmetries \big(Rem. \ref{LieAlgebraOfAllGaugeSymmetries} {(iv)}\big)
is great as an abstract backbone, but is not directly amenable to explicit applications. In practice, one always treats gauge symmetries via explicit
parametrizations. Noether's Second Theorem may be viewed as a way of detecting the existence of (parametrized) gauge symmetries of a local field theory $(\CF,\CL)$, 
and may be deduced immediately as an application of Noether's First (Prop. \ref{Noether1st}), essentially via an integration by parts argument.

\begin{proposition}[\bf Noether's Second Theorem]\label{Noether2nd}
Let
\vspace{-3mm}
$$
\CR_{(-)}:\Gamma_M(E) \longrightarrow \CX_{\mathrm{loc}}(\CF)
$$

\vspace{-1mm}
\noindent be an $\FR$-linear map defined by $\CR_{e}:= R\circ \big(j^\infty e, j^\infty(-) \big)$ for some bundle map
\vspace{-2mm} 
\[ 
\xymatrix@C=3em@R=-.1em  { J^\infty_M E\times_M J^\infty_M F 
 \ar[rd] \ar[rr]^{\hspace{0.7cm} R} &   &  VF \ar[ld]
	\\ 
& F & 
}   
\]

\vspace{-2mm} 
\noindent
linear in the $J^\infty_M E$ fibers. 

\noindent {\bf (i)} Then $\CR_{(-)}$ parametrizes a collection of  infinitesimal gauge symmetries of a local field theory 
$(\CF,\CL)$ if and only if the Euler--Lagrange differential operator satisfies the ``Noether identity"
\vspace{-2mm} 
$$
0=\CN^{R}\circ(\CE \CL\times_\CF \id_\CF)\quad : \quad \Gamma_M(F)\longrightarrow \Gamma_M(V^*F \otimes \wedge^d T^*M)\times \Gamma_M (F) 
\longrightarrow \Gamma_M(E^*\otimes \wedge^d T^*M)\, ,
$$

\vspace{-1mm} 
\noindent where $\CN^R:\Gamma_M(V^*F \otimes \wedge^d T^*M)\times \Gamma_M (F)\rightarrow \Gamma_M(E^*\otimes \wedge^d T^*M)$ 
is the formal adjoint differential operator 
(see, e.g., \cite[\S V.1.3]{Courant})
of $\CR=\Gamma_M(E)\times\Gamma_M(F) \rightarrow \Gamma_M(VF)$. 

\vspace{0.5mm} 
\noindent {\bf (ii)} Explicitly, $\CN^R$ is  defined via (cohomological) integration by parts\footnote{In terms of jet bundles, this calculation corresponds to lifting the Lagrangian density via the projection  $J^\infty_M E\times_M J^\infty_M F\rightarrow J^\infty_M F$, and applying  the interior Euler operator (Def. \ref{InteriorEulerOperator}) 
on $J^\infty_M(E\times_M F)$. This is, in particular, one way to justify the global validity of the local formulas which follow.} with respect to   
$e\in \Gamma_M(E)$ so that for any $\tilde{\phi}\in \Gamma_M(V^*F \otimes \wedge^d T^*M)$, $\phi\in \Gamma_M(F)$ and $e\in \Gamma_M(E)$
\vspace{-2mm} 
$$
\langle \tilde{\phi}, \, \CR_e(\phi) \rangle_{VF} = \langle \CN^R(\tilde{\phi},\phi) ,\,  e \rangle_{E} + \dd_M \CJ (\phi,\tilde{\phi},e)\, , 
$$
for some differential operator $\CJ:\Gamma_M(F)\times\Gamma_M(V^*F \otimes \wedge^d T^*M) \times \Gamma_M(E) \rightarrow \Omega^{d-1}(M)$ which is, 
in particular, $\FR$-linear in the second and third entries. 

\noindent {\bf (iii)} In local coordinates,
$\langle\CN^R(\tilde{\phi},\phi),\, e\rangle_E= \sum_{|J|
=0}^\infty (-1)^{|J|} \cdot \frac{\partial}{\partial x^J}\big(\tilde{\phi}_a \cdot \CR^{a J}_b(\phi)\big)\cdot  e^b \cdot \dd x^1 \cdots \dd x^d$
and so, for all $b$, the condition reads
\vspace{0mm} 
$$
0= \sum_{|J|=0}^\infty (-1)^{|J|} \cdot \frac{\partial}{\partial x^J}\big(\CE \CL_a(\phi) \cdot \CR^{aJ}_b(\phi) \big) \, .
$$
\begin{proof}
If $\CR_e$ defines a symmetry, then by Noether's First Theorem, the corresponding charge satisfies 
\vspace{-2mm} 
\begin{align*}
\dd_M \CP_{R_e} (\phi)= \langle \CE \CL(\phi),\, \CR_e(\phi)  \rangle_{VF}= \sum_{|J|=0}^\infty \CE \CL_a(\phi) \cdot 
\frac{\partial e^b}{\partial x^J} \cdot \CR^{aJ}_b(\phi)\, .
\end{align*}

\vspace{-3mm} 
\noindent `Integrating by parts'  (cohomologically) on the right-hand side with respect to the gauge parameter $e^b$,
\vspace{-1mm} 
\begin{align*}
\dd_M \CP_{R_e} (\phi) &=\langle\CN^R\big(\CE \CL(\phi),\phi \big),\, e\rangle_E+ \dd_M \CJ \big(\phi,\CE \CL(\phi),e\big) \\
&=
e^b\cdot  \sum_{|J|=0}^\infty (-1)^{|J|}\cdot  \frac{\partial}{\partial x^J}\big(\CE \CL_a(\phi) \cdot \CR^{aJ}_b(\phi) \big)  + \dd_M \CJ\big(\phi, \CE \CL(\phi), e\big)\, ,
\end{align*}

\vspace{-2mm} 
\noindent one sees that $\langle \CN^R\big(\CE\CL(\phi),\phi\big),\, e\rangle_E=\dd_M\big(\CP_{R_e}(\phi) + \CJ(\phi,\CE \CL(\phi), e)\big)$ 
is an exact form on $M$ for each field $\phi\in \CF(*)$ and parameter $e\in \Gamma_M(E)$. Proceeding as in the proof of 
Lem. \ref{LocalGaugeSymmetrySimpleExample}, integrating locally over balls $\bar{B}_{p}^d \subset M$ around each point $p\in M$ implies that
\vspace{-1mm}
$$ 
\CN^R\big(\CE\CL(\phi),\phi\big) = 0
$$

\vspace{0mm}
\noindent on (the interior of) $M$, for \textit{all} field configurations $\phi \in \CF(*)$. Conversely, reading the equations backward, 
the vanishing identity implies that $\CR_e$ is a symmetry for every $e\in \Gamma_M(E)$.
\end{proof}
\end{proposition}

\begin{corollary}[\bf  Trivial local currents from gauge symmetries]\label{TrivialCurrentsFromGaugeSymmetries} 
$\,$

\noindent {\bf (i)} The conserved current $\CP_{R_e}$ 
of a gauge symmetry $\CR_e \in \CX_{\mathrm{loc}}^{\CL,\mathrm{gauge}}(\CF)$ satisfies
\vspace{-2mm}
$$
\dd_M \CP_{R_e}=\dd_ M\CJ_e \, : \, \CF \longrightarrow \Omega^{d}_{\mathrm{Vert}}(M)
$$ 

\vspace{-1mm}
\noindent 
 for some local $\CJ_e: \CF \rightarrow \Omega^{d-1}_{\mathrm{Vert}}(M)$ which vanishes on the subspace $\CF_{\CE \CL} \hookrightarrow \CF$ of on-shell fields. 
 
 \noindent {\bf (ii)} Furthermore, if the $(d-1)$-cohomology of $F$ vanishes, then $\CP_{R_e}$ is of the form
\vspace{-2mm}
$$
\CP_{R_e} = \CJ_e + \dd_M \CT_e 
$$

\vspace{-2mm}
\noindent
 for some local $\CT_e: \CF\rightarrow \Omega^{d-2}_{\mathrm{Vert}}(M)$. In particular, it defines a trivial on-shell local current $
\CP_{R_e}|_{\CF_{\CE \CL}} = \dd_M \CT_e $.

\end{corollary}

\begin{proof}
By the proof of Prop. \ref{Noether2nd}, we have
the equality of $d$-form currents on the field space $\dd_M \CP_{R_e} = \dd_M \CJ_e$, where
\vspace{-2mm}
$$
\CJ_e := \CJ \big(-, \CE \CL (-),e\big) \, : \, \CF \longrightarrow \Omega^{d-1}_\mathrm{Vert}(M)
$$

\vspace{-2mm}
\noindent is a $(d-1)$-current whose dependence is linear in the second entry, and hence vanishes on-shell. 

Now, both $\CP_{R_{e}}$  and  $\CJ_e$ arise from bi-differential operators $\CP_{R_{(-)}}, \, \CJ_{(-)}\, : \,$ $\CF \times \mathbold{\Gamma}_M (E) \rightarrow \Omega^{d-1}_\mathrm{Vert}(M)$ induced by horizontal $(d-1)$-forms on $J^\infty(E\times_M F)$,  and so by Prop. \ref{EulerLagrangeComplexCohomology} if the $(d-1)$-cohomology of $E\times_M F$ vanishes, then $\dd_M ( \CP_{R_{e}}- \CJ_e )=0$ implies 
\vspace{0mm}
$$
\CP_{R_e} = \CJ_e + \dd_M \CT_e  
$$

\vspace{0mm}
\noindent
for some $\CT_e : \CF \rightarrow \Omega^{d-2}_\mathrm{Vert}(M)$. Note that, since $E\times_M F $ is a (contractible) vector bundle over $F$, 
 their cohomologies coincide. Since $\CK_{e}$ vanishes on the smooth subspace of on-shell fields, it follows that
\vspace{-2mm}
$$
\CP_{R_e} |_{\CF_{\CE \CL}} = 0 + \dd_M \CT_e|_{\CF_{\CE \CL}} \, : \, \CF_{\CE \CL} \longrightarrow \Omega^{d-1}_{\mathrm{Vert}}(M)\, .
$$

\vspace{-6mm}
\end{proof}

The vanishing of the $(d-1)$-cohomology in the above corollary is usually satisfied explicitly for physical theories (for instance, for field bundles 
being vector bundles over Minkowski spacetime). If instead only the $(d-1)$-cohomology of the spacetime vanishes, then the current is pointwise 
(in the field space $\CF_{\CE \CL}$) $\dd_M$-exact,  but perhaps not with a local dependence in the field (Rem. \ref{SpacetimeCohomologyAndTrivialityOfCurrents}). In \cite{Zuckerman} the statement is claimed without a proof and no mention of the cohomology assumption. In \cite{Blohmann23b} 
the statement is claimed as in \cite{Zuckerman}, but we are not able to follow the proof therein.

\vspace{2mm}
\begin{example}[\bf Noether identity for electromagnetism]\label{ElectromagnetismNoether}
Consider the case of electromagnetism from Ex. \ref{ElectromagnetismGaugeSymmetry}. Cohomologically integrating by parts with respect to 
the gauge parameter the gauge invariance equation Eq. \eqref{ElectroMagnetismGaugeInvarianceEq}, we get
\vspace{-2mm}
\begin{align*}
0&= \dd_M^2 e \wedge \star \dd_M A \\ &= \dd_M\big( \dd_M e \wedge \star \dd_M A \big) + \dd_M e \wedge \dd_M \star \dd_M A \\
&= \dd_M \CP_{R_e}(A) + \dd_M e \wedge  \CE \CL(A) \\ & = \dd_M \CP_{R_e}(A) +
\dd_M \big(e \wedge  \CE \CL(A) \big) - e \wedge \dd_M \big(\CE \CL (A) \big) \, ,
\end{align*} 

\vspace{-2mm}
\noindent where we identified the conserved current \eqref{ElectroMagnetismGaugeCurrent} and the Euler--Lagrange operator \eqref{ElectroMagnetismEulerLagrangeOp}.
It follows that the off-shell Noether identity corresponding to the electromagnetic gauge symmetry is given by 
\vspace{-2mm}
$$
\CN^R(\CE \CL(A)):= \dd_M \circ \CE \CL(A) = \dd_M (\dd_M \star \dd_M A) =0 \, ,
$$

\vspace{-2mm}
\noindent and the corresponding current satisfies the off-shell identity
\vspace{-2mm}
$$
\dd_M \CP_{R_e} = \dd_M \CJ\big(A,\CE \CL(A),e\big) := \dd_M \big(- e\wedge \CE \CL(A) \big) \, .
$$

\vspace{-2mm}
\end{example}

\begin{remark}[\bf Physical observables from gauge symmetries]
The potentiality of a trivial resulting current does not
imply that gauge symmetries do not produce meaningful observables via their charges over compact submanifolds with 
boundary or non-compact submanifolds with suitable asymptotics, for choices of gauge parameters with support on the boundary
or suitable asymptotic support, respectively (e.g. \cite{BB01}\cite{ABS}). 
Although we do not pursue reviewing this matter here, this story 
should also naturally take place in this generalized smooth setting.
\end{remark}

\begin{remark}[\bf Nomenclature of local and global symmetries] The adverbs  ``local" and ``global" have several, but related meanings 
in the physics literature. In this manuscript, we use local as in the mathematics literature, to mean factoring through an infinite jet bundle, 
or to indicate that a certain quantity is defined only on a subspace of an ambient space. In the physics literature, gauge symmetries 
are often referred to simply as `local symmetries' in the sense that they are freely parametrized by (gauge) parameters that are not `constant'
through space-time. In other words, their value may be chosen to vary `locally from point to point in space-time. In contrast, the term
`global symmetry' (also called rigid) is reserved for symmetries that are not part of a gauge symmetry, meaning they can be only parametrized by a finite 
(or countable) dimensional parameter vector space. More formally, the set of \textit{global} or \textit{rigid} infinitesimal symmetries of
a Lagrangian is defined as the quotient of all local infinitesimal symmetries by the (ideal) of gauge symmetries \big(Rem. \ref{LieAlgebraOfAllGaugeSymmetries} {(iv)}\big)
\vspace{-2mm}
$$
\CX_{\mathrm{loc}}^{\CL,\mathrm{rig}}(\CF):= \CX_\mathrm{loc}^{\CL} (\CF)\, /\,  \CX^{\CL,\mathrm{gauge}}_{\mathrm{loc}}(\CF)\, .
$$
\end{remark}

\newpage 
\subsection{Cauchy surfaces and obstructions by gauge symmetries}
\label{Sec-Cauchy} 

The all-important feature of gauge symmetries, both physically and mathematically, is that they \textit{`obstruct the unique evolution'} 
of fields obeying the Euler--Lagrange equations. Let us make this mathematically precise in our current context of smooth sets.

\begin{definition}[\bf  Sections on infinitesimal neighborhood of submanifold]\label{SectionsOnInfinitesimalNeighborhoodOfSubmanifold}
Let $\iota_{\Sigma^p}: \Sigma^p\hookrightarrow M$ be a submanifold of spacetime. The \textit{sections on the infinitesimal neighborhood of
$\Sigma^p$} of the bundle $F\rightarrow M$ are defined as the set of sections of the infinite jet bundle $J^\infty_M F$ over $\Sigma^{p}$, 
\vspace{-2mm}
\begin{align}
\Gamma_{\Sigma^{p}}\big(J^\infty_M F|_{\Sigma^{p}}\big)\cong 
\big\{ \psi: \Sigma^{p} \rightarrow J^\infty_M F \;  \big{|} \; \pi_M \circ \psi = \iota_{\Sigma^{p}}\big\} \, .   
\end{align}
The corresponding smooth set $\mathbold{\Gamma}_{\Sigma^{p}}\big(J^\infty_M F|_{\Sigma^{p}}\big)$ follows as in Def. \ref{SectionsSmoothSet}.
\end{definition}
The nomenclature will become apparent in \cite{GSS-2}, whereby the `infinitesimal neighborhood of $\Sigma^p$ in $M$' will be a \textit{bona-fide} 
(thickened) smooth subspace of $M$, over which sections of $F$ (and not $J^\infty_M F$) correspond to the above set. For now, it is enough to note 
the interpretation of the definition: Any section $\psi:\Sigma^p\rightarrow J^\infty_M F |_{\Sigma^p}$ (smoothly) determines a particular instance
of the potential values for sections $\phi$ of $F\rightarrow M$ and their derivatives to any order on the submanifold $\Sigma^p$, $\psi(x) = [j^\infty_x \phi]$. 
In other words, one may ask whether there exists a field $\phi \in \CF(*)= \Gamma_M(F)$ whose jet restricts to the chosen section $\psi$ on 
the infinitesimal neighborhood of $\Sigma^p$. Along these lines, choosing the submanifold to be of codimension-1 and demanding further compatibility 
with the Euler--Lagrange equations, we may define the smooth set of initial conditions on a submanifold $\Sigma^{d-1}$ for a given Lagrangian
field theory $(\CF,\CL)$ on the spacetime $M$.
\begin{definition}[\bf  Initial data]\label{InitialData}
Let $(\CF,\CL)$ be a local Lagrangian field theory on $M$ and $\Sigma^{d-1}\hookrightarrow M$ a submanifold of codimension $1$. The (smooth) 
set of \textit{initial data of $(\CF,\CL)$ on $\Sigma^{d-1}$} is defined by the subset of (plots of) sections over the infinitesimal
neighborhood of $\Sigma^{d-1}$ that factor through the prolongated shell $S_L^\infty \hookrightarrow J^\infty_M F$ (Def. \ref{ShellOfLagrangian}) of the Lagrangian L,
\begin{align}
\mathrm{InData}_{\CL}( \Sigma^{d-1}):= 
\mathbold{\Gamma}_{\Sigma^{d-1},\pr EL=0}\big(J_M^\infty F|_{\Sigma^{d-1}}\big)
\longhookrightarrow \mathbold{\Gamma} _{\Sigma^{d-1}}\big(J_M^\infty F|_{\Sigma^{d-1}}\big)\, .
\end{align}
In particular, a point $\psi$ of $\mathrm{InData}_{\CL}(\Sigma^{d-1})$ is a section of the form
\vspace{-2mm} 
\[ 
\xymatrix@C=1.8em@R=.4em  {& S^\infty_L  \ar[rd] \ar[rr] &   & J^\infty_M F\ar[ld]
	\\ 
\Sigma^{d-1} \ar[ru]^{\psi} \ar[rr] & & M & 
} \, ,  
\]

\vspace{-2mm} 
\noindent and similarly for higher plots.

\end{definition}
There is a natural restriction map from the space of fields to sections over the infinitesimal neighborhood of $\Sigma^{d-1}$, 
given by the compostion 
$(-)|_{\Sigma^{d-1}}\circ j^\infty : \CF \rightarrow \mathbold{\Gamma}_{M}(J^\infty_M F) \rightarrow \mathbold{\Gamma} _{\Sigma^{d-1}}\big(J_M^\infty F|_{\Sigma^{d-1}}\big)$.
By definition of the space of initial data on $\Sigma^{d-1}$, this descends to a map
\begin{align}\label{InitialDataRestriction}
(-)|_{\Sigma^{d-1}} \circ j^\infty \;:\; \CF_{\CE \CL} \longrightarrow\mathrm{InData}_\CL(\Sigma^{d-1})
\end{align}
since any on-shell field $\phi:M\rightarrow F$ field factors through the prolongated shell via the jet prolongation, and hence 
its restriction to any submanifold does as well.  
\begin{definition}[\bf  Cauchy surface]\label{CauchySurface}
A (connected) codimension-1 submanifold $\Sigma^{d-1}\hookrightarrow M$ is a \textit{Cauchy surface} of the Lagrangian field theory $(\CF,\CL)$ if 
the smooth map
\vspace{-2mm}
\begin{align}
\mathrm{Cau}_{\Sigma^{d-1}}:= (-)|_{\Sigma^{d-1}} \circ j^\infty \;:\; \CF_{\CE \CL} \xlongrightarrow{\quad \sim \quad} \mathrm{InData}_\CL(\Sigma^{d-1})   
\end{align}

\vspace{-2mm}
\noindent is a diffeomorphism.
\end{definition}
The interpretation is exactly as in the physics literature. At the level of $*$-plots, this entails that there is a bijection between fields $\phi:M\rightarrow F$ 
that satisfy the Euler--Lagrange equations $\CE \CL(\phi)=0$ on $M$, and the set of initial conditions $\psi:\Sigma^{d-1}\rightarrow J^\infty_M F$ consistent with 
the Euler--Lagrange equations,  $\mathrm{Im}(\psi)\subset S_L \hookrightarrow J^\infty_M F$. In simple words, the derivatives of the on-shell fields along 
a Cauchy surface serve as a parametrization of the full space of on-shell fields. 
\footnote{In general, the definition might be relaxed to yield a 
diffeomorphism for on-shell fields restricted to an open neighborhood $U$ of $\Sigma^{d-1}$. Everything we say here applies to such a modification, 
and hence we will not be explicit about it.}
In our current setting of smooth sets, we see that this parametrization, if it exists, is in fact
automatically smooth.
\begin{remark}[\bf On the order of initial data]\label{OnTheOrderOfInitialData}
Strictly speaking, in practice, for an Euler--Lagrange differential operator of order $k$, $EL^k: J^k_M F \rightarrow \wedge^d T^*M \otimes V^*F$, 
it is sufficient to consider (compatible) initial data in terms of jets along 
a hypersurface, up to order $k$. That is, the above definitions are trivially modified by replacing $S^\infty_L \hookrightarrow J^\infty_M F$ 
with the $k$-order shell $S_{L,k} \hookrightarrow J^k_M F$ (see Rem. \ref{ManifoldStructureOnTheShell}). In special cases, these smooth sets can be 
identified -- but not in general. Nevertheless, everything we say below holds verbatim for both notions of initial data.\footnote{For globally finite order 
differential operators, it is not hard to see that there is a bijection - and hence a diffeomorphism - between the \textit{images} of the restrictions 
$\mathrm{Cau}_{\Sigma^{d-1}}:= (-)|_{\Sigma^{d-1}} \circ j^\infty : \CF_{\CE \CL} \longrightarrow \mathbold{\Gamma}_{\Sigma^{d-1}}\big(J_M^\infty F|_{\Sigma^{d-1}}\big)$ 
and $\mathrm{Cau}_{\Sigma^{d-1}}^k:= (-)|_{\Sigma^{d-1}} \circ j^k : \CF_{\CE \CL} \longrightarrow \mathbold{\Gamma}_{\Sigma^{d-1}}\big(J_M^k F|_{\Sigma^{d-1}}\big)$. 
Hence, when $\Sigma^{d-1}$ is a Cauchy surface, the two notions of initial data coincide.}
Our choice is purely formal, so as to avoid talking about a fixed finite order of the Lagrangian and the 
induced Euler--Lagrange operator altogether.  
\end{remark}

\begin{example}[\bf Particle mechanics Cauchy surface]\label{ParticleMechanicsCauchySurface}
As a check, we consider the particularly simple example of free (non-relativistic) particle mechanics. Although everything we say generalizes 
straightforwardly for a free particle moving in an arbitrary Riemannian manifold $(M,g)$, for the sake of exposition we write formulas explicitly
only for the case of flat, Euclidean space $(\FR^d,\delta)$. This is, in fact, a special 
case of the O($n$)-model from Ex. \ref{VectorValuedFieldTheoryLagrangian} via the following choices. 
Choosing the spacetime to be only `time' $M=\FR^1_t$ with its canonical metric, then the Euclidean target $W\cong (\FR^d ,\delta)$ 
is interpreted as the `space', and the field space is the smooth path space $\CF=\mathbold{P}(\FR^d)=[\FR^1_t,\FR^d]$. 
In the case where the coupling constants vanish, the Lagrangian reduces to the free particle Lagrangian
\vspace{-1mm}
\begin{align*}
\CL(\gamma) &= \langle \dd_t \gamma, \star \dd_t \gamma \rangle = \langle \partial_t \gamma , \partial_t \gamma \rangle \cdot \dd t 
= \partial_t \gamma^a \cdot \partial_t  \gamma_a \cdot \dd t \, .
\end{align*}

\vspace{-1mm}
\noindent The corresponding Euler--Lagrange operator is simply Newton's Law
\vspace{-2mm}
\begin{align*}
\CE \CL(\gamma) &= \dd_t \star \dd_t \gamma  = \partial_t^2 \gamma \cdot \dd t \, , 
\end{align*}

\vspace{-1mm}
\noindent 
and the on-shell space of solutions consists of (plots of) straight lines in $\FR^d$
\vspace{-1mm}
\begin{align*}
\CF_{\CE \CL} = \mathrm{Lines}(\FR^d) \longhookrightarrow [\FR_t^1,\FR^d]\, ,  
\end{align*}

\vspace{-1mm}
\noindent 
whose *-plots are simply straight lines $\gamma= v \cdot t + c  : \FR_t \rightarrow \FR^d$. Similarly, $\FR^k$-plots are parametrized straight lines
$\gamma^k= v^k \cdot t + c^k : \FR^k \times \FR_t \rightarrow \FR^d$, for arbitrary functions $v^k \in C^\infty(\FR^k)$ and $c^k \in C^\infty(\FR^k)$. 

The space of positions and velocities at any time instant $t_0\in \FR^1_t$ should be a copy of the tangent bundle $T\FR^d$, 
and as is well known (and obvious) the position and velocity at a time instant fully determine the set of straight lines. 
Thus, it is naturally expected that the full space of on-shell fields $\mathrm{Lines}(\FR^d)$ is \textit{diffeomorphic} 
to the tangent $T \FR^d$, in an appropriate sense. To see how this description arises from our definition, notice that 
the vanishing of the Euler--Lagrange source form $EL = u_{tt}^a \cdot \dd_V  u_a  \cdot \dd t \in \Omega^{d=1,1}\big(J^\infty_{\FR^1_t}(\FR^1_t\times \FR^d)\big)$, and of its prolongation $\pr EL$,
imposes vanishing conditions on all jets of order $k\geq 2$, and no conditions for $k\leq 1$. Hence  its prolongated shell $S^\infty_L$ is 
canonically identified with the first jet bundle under the canonical embedding
\vspace{-2mm}
\begin{align*}
J^1_{\FR_t}(\FR^1_t\times \FR^d) &\longhookrightarrow J^\infty_{\FR^1_t}(\FR^1_t\times \FR^d)\\
\big(\gamma(t_0), \dot{\gamma}(t_0) \big) &\longmapsto \big(\gamma(t_0), \dot{\gamma}(t_0), 0,\dots 0\big)\, ,
\end{align*}

\vspace{-2mm}
\noindent 
which makes sense globally in this case of trivial bundles. Now, connected submanifolds of codimension-1 of the base $M=\FR^1_t$ 
are necessarily single points, i.e., instants of time $t_0 \in \FR^1_t$. Hence the smooth set of initial data 
$\mathrm{InData}_{\CL}(\{0\})$ (Def. \ref{InitialData}) along  $\{t_0=0\}\hookrightarrow  \FR^1_t$,  
has $*$-plots $\psi$ that factor through the shell
\vspace{-2mm} 
\[ 
\xymatrix@C=1.8em@R=.4em  {& J^1_{\FR_t}(\FR^1_t\times \FR^d)  \ar[rd] \ar[rr] &   & J^\infty_{\FR_t}(\FR^1_t\times \FR^d)\ar[ld]
	\\ 
\{0\} \ar[ru] \ar[rr] & &\FR^1_t & 
} \, ,  
\]

\vspace{-2mm} 
\noindent and similarly for higher plots. Note that such sections are simply points (and plots) of the first jet manifold  
$J^1_{\FR_t}(\FR^1_t\times \FR^d)$, and so the smooth set of initial data is simply the Yoneda embedding of the first jet
space $\mathrm{InData}_{\CL}(\{0\})\cong y(J^1_{\FR_t}(\FR^1_t\times \FR^d))$. But the manifold of first jets of
curves is by definition the tangent bundle  $ J^1_{\FR_t}(\FR^1_t\times \FR^d) \cong_\SmoothManifolds T \FR^d $, 
and so it follows that
\vspace{-5mm}
\begin{align*}
\mathrm{InData}_{\CL}(\{0\}) \cong_\SmoothSets y(T\FR^d) \, .
\end{align*}

\vspace{-2mm}
\noindent
In other words, the smooth set of initial conditions is simply (the Yoneda embedding of) the tangent bundle $T\FR^d$ of the target $\FR^d$,
which shows immediately that the restriction 
\vspace{-2mm}
\begin{align}\label{ParticleMechanicsIsomorphism}
\mathrm{Cau}_{\{0\}}:= (-)|_{\{t=0\}} \circ j^1 : \CF_{\CE \CL} &\; \longrightarrow\; \mathrm{InData}_\CL(\{0\}) \cong y(T\FR^d)  \\
\gamma^k= v^k \cdot t + c^k &\; \longmapsto \; (c^k , v^k) \, , \nn 
\end{align}

\vspace{-2mm}
\noindent is a diffeomorphism. Here the choice of time instant $\{t_0\}\hookrightarrow \FR^1_t$ is immaterial, and hence 
any time instant is a Cauchy surface -- but through a different Cauchy isomorphism. We will revisit this setting in Ex. \ref{Ex-free-cov} below. 
\end{example}

Generally, neither the surjectivity nor the injectivity of the underlying point-set map \eqref{InitialDataRestriction} to the initial data 
is guaranteed, both of which depend highly on the detailed form of the given 
Euler--Lagrange differential operator $\CE \CL$. In non-technical terms, surjectivity says that for any choice of initial data 
$\psi:\Sigma^{d-1}\rightarrow J^\infty_M F$, there exists \textit{some} field configuration $\phi:M\rightarrow F$ whose jet prolongation restricts 
to $\psi$ on $\Sigma^{d-1}$, and so $\phi$ is thought as the `\textit{evolution}' of the initial condition $\psi$ on $M$ under the Euler--Lagrange 
equations $\CE \CL =0$. Injectivity says that in the cases of initial conditions $\psi$ for which some evolution $\phi$ \textit{exists}, 
then it is also \textit{unique}. 

\smallskip 
The general study of such \textit{well-posedness} questions is quite delicate and falls within the realm of analysis of PDEs, 
which is outside of the scope of this manuscript. However, an immediate general result that evades analytical arguments is that
\textit{finite gauge} symmetries \textit{obstruct} the existence of Cauchy surfaces -- by ruling out the injectivity of the underlying point-set map.

\begin{proposition}[\bf Gauge symmetry obstructs Cauchy surface]\label{LocalGaugeSymmetryObstructsCauchySurface}
Let $\CR_{(-)}:\Gamma_M(E)\rightarrow \CX_{\mathrm{loc}}(\CF)$ be an infinitesimal (local) gauge symmetry of a Lagrangian field theory $(\CF,\CL)$ 
and $\Sigma^{d-1}\hookrightarrow M$ a codimension-1 submanifold. Let $e\in \Gamma_M(E)$ be a (nonzero) gauge parameter with support contained strictly 
within $\mathrm{int}(M-
\Sigma^{d-1})$. Then, assuming 
 the infinitesimal symmetry $\CR_{e}$ is {\bf (a)} integrable and {\bf (b)} not zero on $\CF_{\CE \CL}$ (hence not trivial), the restriction to initial data
 \vspace{-2mm} 
$$
(-)|_{\Sigma^{d-1}} \circ j^\infty  \;:\; \CF_{\CE \CL} \longrightarrow\mathrm{InData}_\CL(\Sigma^{d-1})
$$

\vspace{-2mm} 
\noindent
 is \textit{not} a diffeomorphism. In other words, the submanifold $\Sigma^{d-1}$ is not a Cauchy surface.
\end{proposition}
\begin{proof}
 Denote by $\CD_{(e)}:=\CD_{(e),t=1}:\CF\rightarrow \CF$ the diffeomorphism which is the integrated version of the local $\CR_e$ 
 vector field, i.e., $\partial_t \CD_{(e),t} |_{t=0} = \CR_{e}$. This exists by assumption and hence is in particular
 a \textit{local} symmetry (Rem. \ref{DifferentiatingFiniteLocalSymmetries}). 
 By Prop. \ref{LocalSymmetryPreservesOnshellSpace}, it preserves the on-shell space of fields
 $\CD_{(e)}:\CF_{\CE \CL} \xrightarrow{\sim} \CF_{\CE \CL}$, 
 and so if $\phi\in \CF_{\CE \CL}(*)$ is an on-shell field, then so is $\CD_{(e)}(\phi) $.
 Furthermore, the infinitesimal symmetry $\CR_{e}$ 
 is by assumption not zero on $\CF_{\CE \CL}$ (hence nontrivial, Ex. \ref{TrivialInfinitesimalSymmetries}, Ex. \ref{DifferentiatingFiniteLocalSymmetries}).
 Consequently,  $\CD_{(e)}$ is \textit{not} the identity map $\CD_{(e)}\neq \id_{\CF_{\CE \CL}}$ on the space of on-shell fields,  so that 
 (without loss of generality) 
 \vspace{-2mm} 
 $$
 \CD_{(e)}(\phi)\neq \phi\, . 
 $$

 \vspace{-2mm} 
\noindent
 Next, since the support of $e$ is strictly within $\mathrm{int}(M-\Sigma^{d-1})$, it follows that for any $p\in \Sigma^{d-1}$ 
 there exists some open neighborhood
 $U_{p}\subset M$ on which $e$ vanishes, $e|_{U_p}=0$. It follows that, by the $\FR$-linear dependence on $e$ and the petit sheaf 
 nature of the field space $\CF$ 
 (Rem. \ref{FieldSpaceAsASheafOfSheaves}), the \textit{local} vector field $\CR_{e}$ vanishes on $U_p$, i.e., 
 $\CR_{e}|_{U_p} =0 \in \CX_{\mathrm{loc}}(\CF|_{U_p})$.
 Thus, the corresponding integrated diffeomorphism must be the \textit{identity} when restricted to $U_p\subset M$,
\vspace{-2mm} 
$$
\CD_{(e)}|_{U_p}=\id_{\CF|_{U_{p}}} \, : \, \CF_{\CE \CL}|_{U_p} \longrightarrow \CF_{\CE \CL} |_{U_p} \, ,
$$

\vspace{-2mm} 
\noindent for some open neighborhood around $p\in \Sigma^{d-1}$, for \textit{every} $p\in \Sigma^{d-1}$. 

Finally, to prove the claim of the proposition, it is enough to prove the map is not bijective on $*$-plots. Let $\phi\in \CF_{\CE \CL}(*)$ 
be any on-shell field configuration. Since local diffeomorphism on $\CF$ are  in particular maps of (petit) sheaves on $M$, we have
\vspace{-1mm} 
$$
\CD_{(e)}(\phi) |_{U_p}=\CD_{(e)}|_{U_p}(\phi|_{U_p})= \id_{\CF|_{U_{p}}} (\phi|_{U_p})= \phi |_{U_p} \in \CF_{\CE\CL} |_{U_{p}}\, 
$$
by which it follows that the infinite jets coincide
\vspace{-2mm} 
$$
j^\infty_p\big(\CD_{(e)}(\phi)\big) = j^\infty_p \phi  \, \in  J^\infty_M F\, ,
$$

 \vspace{-2mm} 
\noindent
 for every $p\in \Sigma^{d-1}$. That is, both on-shell fields $\CD_{(e)}(\phi)$ and $\phi$ define the \textit{same} initial data on $\Sigma^{d-1}$
\vspace{-2mm} 
$$
j^\infty\big(\CD_{(e)} \phi) |_{\Sigma^{d-1}} = j^\infty \phi |_{\Sigma^{d-1}} \; : 
\; \Sigma^{d-1} \longrightarrow S_L\xhookrightarrow{\;\;\;} J^\infty_M F \, .
$$

\vspace{-2mm} 
\noindent 
Therefore, there are two \textit{different} on-shell configurations $\CD_{(e)}(\phi), \, \phi 
$ with the same initial data $j^\infty \phi |_{\Sigma^{d-1}}$, hence the map 
\vspace{-2mm} 
$$
(-)|_{\Sigma^{d-1}} \circ j^\infty \;:\; \CF_{\CE \CL}(*) \longrightarrow\mathrm{InData}_\CL(\Sigma^{d-1})(*)
$$

\vspace{-2mm} 
\noindent is \textit{not} injective.  
\end{proof}

\begin{corollary}[\bf  Nonexistence of Cauchy surfaces]\label{NonexistenceOfCauchySurfaces}
Let $\CR_{(-)}:\Gamma_M(E)\rightarrow \CX_{\mathrm{loc}}(\CF)$ be an infinitesimal (local) gauge symmetry of a Lagrangian field theory $(\CF,\CL)$
such that $\CR_{e}$ is {\bf (a)} integrable and {\bf (b)} not zero on $\CF_{\CE \CL }$, for all gauge parameters $e\in \Gamma_M(F) $ of compact support. 
Then there does not exist any Cauchy surface for $(\CF,\CL)$.    
\end{corollary}
\begin{proof}
Let $\Sigma^{d-1}$ be a codimension-$1$ submanifold and choose (nonzero) a gauge parameter $e\in \Gamma_M(E)$ with compact support in 
$\mathrm{int}(M-\Sigma^{d-1})$. Such sections exist since $E\rightarrow M$ is a locally trivial bundle, by employing bump functions. 
By the assumptions on $\CR_{e}$, the result of Prop. \ref{LocalGaugeSymmetryObstructsCauchySurface} applies. The argument applies 
for any candidate submanifold $\Sigma^{d-1}$, and so the result follows.   
\end{proof}

\begin{example}[{\bf Instances of nonexistence}]
The assumptions of Cor \ref{NonexistenceOfCauchySurfaces} are satisfied in all examples of fundamental field theories with gauge symmetries. 
We indicate the corresponding ingredients without spelling out the details in two such cases: 

\vspace{.5mm} 
\noindent {\bf(i)} For the theory of General Relativity, the infinitesimal gauge symmetries are parametrized by vector fields $\CX(M) =\Gamma_M(TM)$ on spacetime. 
Those of compact support $\CX_{c}(M)\subset \CX(M)$ \big(and their induced vector fields on $\CF=\mathrm{Met}(M)$\big) arise by differentiating diffeomorphisms
$\mathrm{Diff}_c(M)\subset \mathrm{Diff}(M)$ \big(and the induced action on $\mathrm{Met}(M)$\big)\footnote{Strictly speaking, the (finite) gauge symmetries
in the case of General Relativity are spacetime covariant symmetries (Def. \ref{FiniteSymmetryofLagrangianFieldTheory}). Nevertheless, the on-shell space of 
fields is still preserved (Prop. \ref{SymmetryPreservesOnshellSpace}), and hence Prop. \ref{LocalGaugeSymmetryObstructsCauchySurface} and
Cor. \ref{NonexistenceOfCauchySurfaces} apply with minor modifications. } which are the identity outside a compact submanifold. 

\vspace{.5mm}
\noindent {\bf (ii)} Similarly for Yang-Mills theory in the trivial topological sector, the infinitesimal gauge symmetries are parametrized by 
$C^\infty(M,\frg)= \Gamma_M(M\times \frg)$. 
The gauge parameters of compact support $C^\infty_{c}(M,\frg)\subset C^\infty(M,\frg)$  arise by differentiating maps  $C^\infty_c(M,G)\subset C^\infty(M,G)$
which map to the identity $e\in G$ outside a compact submanifold (along with the induced vector fields on $\CF= \Omega^1(M,\frg)$ by differentiating 
the induced finite diffeomorphisms 
on $\Omega^{1}(M,\frg)$, respectively). A special case of this is (abelian) $G=U(1)$ Yang-Mills theory in the guise of electromagnetism from 
Ex. \ref{ElectromagnetismGaugeSymmetry}. By the linearity of the field space $\CF=\Omega^1_\mathrm{Vert}(M)$ and the constancy of each 
vector field $\CR_e: \CF \rightarrow T\CF \cong \CF\times \CF$, given by $A^k\mapsto \dd_M e$, the diffeomorphism corresponding to the 
integrated infinitesimal transformation is given on $\FR^k$-plots of fields by 
\vspace{-1mm}
\begin{align*} 
 \CD_{(e)} \;:\; \Omega_\mathrm{Vert}^1(M) &\longrightarrow \Omega_\mathrm{Vert}^1(M) 
 \\[-2pt] 
A^k &\longmapsto A^k + \dd_M e \, .
\end{align*}

\end{example} 

\begin{remark}[{\bf Gauge equivalence classes}]
To preserve the assumption that classical physics is \textit{deterministic}, i.e., with a unique `time evolution', one is forced to consider any two field 
configuration $\phi_1,\phi_2 \in \CF(*)$ of a local field theory $(\CF,\CL)$ which are related by a (finite) gauge transformation as being equivalent. 
It follows that a physically viable observable should assign the same value for any gauge equivalent field configuration. One way to implement this 
rigorously is to consider the quotient set of  fields by gauge transformations, or in our setting, the quotient \textit{smooth}
set -- a colimit construction which \textit{exists} in smooth sets 
\vspace{-3mm} 
\begin{align}\label{QuotientByGaugeSymmetries}
\CF\, /\, \mathrm{Diff}^{\CL,\mathrm{gauge}}(\CF)\hspace{0.5cm} \in \hspace{0.5cm} \SmoothSets
\end{align}

\vspace{-1mm} 
\noindent on which all (gauge invariant) local smooth functionals out of $\CF$ descend to. We note that for this quotient construction to make sense,
it is crucial that the Euler--Lagrange operator is covariant under local (and spacetime covariant) symmetries, which is indeed the case 
\eqref{CovarianceofEulerLagrangeOperator}. 
For the case of electromagnetism, this is the smooth set corresponding to the quotient $\Omega^1(M) / \dd_M \Omega^0(M)$.

\end{remark}

\begin{remark}[\bf On gauge equivalence and redundancy]\label{OnGaugeEquivalencyandRedundancy}
We briefly remark on two ways that the viewpoint of gauge symmetries as mere redundancies can be insufficient. The concepts involved in the discussion below 
will be rigorously fleshed out in \cite{GSS-3}.

\vspace{.5mm} 
\noindent {\bf (i)} This treatment is sufficient for most purposes of \textit{classical} field theory with fields being sections of fiber bundles 
$\CF=\mathbold{\Gamma}_M(F)$. However, for the purposes of `gauge fixing'\footnote{Roughly, this means finding an equivalent local field theory 
$(\CF',\CL')$ with \textit{no gauge symmetries}. Equivalency, among other conditions, means that the two underlying on-shell spaces should be 
diffeomorphic $\CF_{\CE \CL}/ \mathrm{Diff}^{\CL,\mathrm{gauge}}(\CF)\cong \CF'_{\CE \CL' } $.} and quantization it is extremely useful to 
consider this as a quotient in a homotopical/higher geometrical sense, thus resulting into a \textit{smooth groupoid} instead (see \cref{outlook}). 
The \textit{infinitesimal}/\textit{perturbative} (and algebraic) incarnation of this construction appears in the physics literature under
the name {\it BRST complex} (see \cite{HenneauxTeitelboim92}). 

\vspace{.5mm} 
\noindent {\bf (ii)} There is another, related, sense in which the phrase `gauge equivalence is redundancy' should be carefully interpreted. 
Throughout this text, we have dealt with field spaces being sections of fiber bundles $\CF=\mathbold{\Gamma}_M(F)$ \textit{over the spacetime }$M$,
which are by definition (petit) sheaves on $M$ (Rem. \ref{FieldSpaceAsASheafOfSheaves}). This description includes, in particular, globally defined 
$G$-gauge fields $\Omega^1(M,\frg)\cong \Gamma_M\big(T^*M\times_M (M\times \frg)\big)$. In other words, principal $G$-bundle connections over $M$ 
in the \textit{trivial topological sector}. Taking into account dynamics via a Lagrangian which is invariant under the action of $C^\infty(M,G)$,
one may (or may not) consider the corresponding homotopy quotient groupoid as the actual physical field space. However, gauge fields in the 
non-trivial topological sector, i.e., $G$-connections on $P\rightarrow M$ for an arbitrary $G$-bundle (e.g., Dirac monopole in electromagnetism, 
instantons in Yang-Mills theory), do \textit{not} follow this description. In particular, the description of a \textit{single} gauge field $A$ 
as an object \textit{living on spacetime}\footnote{As opposed to its equivalent description of a \textit{global} $1$-form 
$\tilde{A}\in \Omega^{1}(P,\frg)$ on the corresponding principal $G$-bundle $P\rightarrow M$.} $M$ (with respect to a good open cover
$\coprod_{i\in I} U_i \rightarrow M$) is given by a pair of families of 

{\bf (a)} locally defined fields  $\big\{A_i \in
      \Omega^1(U_i, \mathfrak{g})\big\}_{i\in I}$ \textit{and} 

 {\bf (b)} locally defined ``transition maps''\footnote{Corresponding to the isomorphisms between two different
 trivializations of $P\rightarrow M$ over $U_i$ and $U_j$.} $\big\{
      g_{ij} 
      \in
      C^\infty(U_{ij}, G) \big\}_{i,j\in I}$ on the intersections $U_{ij}= U_i\cap U_j$ such that 

\vspace{-2mm}
\begin{equation}
  \label{GaugeTransformations}
      A_{j}
      =
      g_{ji}\cdot  A_{i} \cdot g_{ji}^{-1}
      +
      g_{ji}\cdot \dd 
      g_{ji}^{-1}  \quad  \mathrm{and} \quad 
      g_{ij}=g_{ik}\cdot g_{kj} 
\end{equation}
      
      \vspace{0mm}
      \noindent on overlaps $U_{ij}$ and triple overlaps $U_{ijk}=U_i\cap U_j \cap U_k $, respectively. In other words, the collection of all $G$-gauge fields 
      on $M$ does not form a (petit) mapping/section sheaf of Rem. \ref{FieldSpaceAsASheafOfSheaves}. Instead, they may be identified with an appropriate
      mapping space construction in a higher categorical setting, via maps into a classifying moduli stack
      $M\rightarrow \mathbf{B}G_{\mathrm{conn}}$ \eqref{BGConn}, hence forming a (petit) \textit{mapping stack} on $M$! 
     
     One notices that the ``transformations'' relating the $\big\{A_i\in \Omega^1(U_i,\frg) \big\}$ via the locally defined transition maps 
     $\big\{g_{ij} \in C^\infty(U_{ij}, G)\big\}_{i,j\in I}$ have the same form as the action of (locally defined) \textit{gauge transformations}
     in the case of the trivial topological sector. However, these transition maps are \textit{data} entering the definition of the field configuration. 
     They are \textit{not} to be viewed as redundancies and one should \textit{not} quotient by them. Indeed, if one defines instead the field space to
     be \textit{equivalence classes} $[A_i]$ of locally defined forms by such transformations, then the information about the non-trivial topological 
     sectors is lost. In particular, if these (locally defined) transition maps were merely redundancies, then all gauge fields would be necessarily 
     equivalent to a globally defined $A\in \Omega^1(M,\frg)$.

     The actual gauge equivalences that may be partially viewed as redundancies in the vein of {\bf (i)} are instead parametrized by locally defined 
     families\footnote{In contrast to transition the transition maps, these are defined locally on the open subsets $U_i$, rather than the 
     intersections $U_{ij}$.} $\big\{\lambda_i \in C^\infty(U_i,G) \big\}_{i\in I}$ acting on any gauge field\footnote{Corresponding to pulling 
     back the connection on the total space $P\rightarrow M$ along a principal $G$-bundle 
     automorphism $\bar{P}\rightarrow P$ over $M$.} $A= \big( \{A_i\}_{i\in I}, \, \{g_{ij}\}_{i,j\in I} \big)$ by
     \vspace{-1mm}
\begin{align*}
\{A_i\}_{i\in I} &\longmapsto \{\lambda_{i}^{-1} \cdot A_i \cdot \lambda_i + \lambda_i^{-1} \cdot \dd \lambda_i \}_{i\in I} \\
\{g_{ij}\}_{i,j\in I} &\longmapsto \{ \lambda^{-1}_i \cdot g_{ij} \cdot \lambda_j \}_{i,j\in I} \, . 
\end{align*}

\vspace{-1mm}
\noindent
In fact, these natural automorphisms of the full nonperturbative field space consisting of $G$-gauge fields of arbitrary topological sectors 
are automatically encoded in the (smooth) mapping stack/groupoid construction $[M, \mathbf{B}G_{\mathrm{conn}}]$ as homotopies/higher morphisms 
between the two gauge field configurations \cite{FSS12}\cite{FSS13}\cite{FSS14}, as we will make further explicit in
future installments of this series.

\end{remark}

\newpage 

\addtocontents{toc}{\protect\vspace{-10pt}}
\section{Presymplectic structure of local field theories}
\label{TheBicomplexOfLocalFormsSection}

The setting of smooth sets is the natural arena to make fully rigorous the observation of \cite{Zuckerman} (considerably expanded in \cite{DF99}), 
that the usual manipulations of classical \textit{local} field theory (for example as described in the previous sections) may be  
concisely expressed via the bicomplex of \textit{local} differential forms on the product `smooth space' $\CF\times M$ of off-shell fields 
and the spacetime, which arises as the `pullback' of the variational bicomplex on $J^\infty_M F$. 

\medskip 
The point is that in our current context 
the spaces in question \textit{exist} as smooth sets. There is a natural notion of a \textit{smooth} tangent bundle on $\CF\times M$, 
and hence differential forms on it, whereby one can genuinely pullback forms on $J^\infty_M F$ via the \textit{smooth prolongated} evaluation 
map $\mathrm{ev} \circ (j^\infty\times \id_M) : \CF\times M \rightarrow J^\infty_M F$, hence identifying a subset of differential forms with 
an induced bi-complex structure. This recovers the bicomplex of local forms as implicitly used in \cite{Zuckerman} and \cite{DF99}. 
For any local field theory $(\CF,\CL)$, the main additional insight of this formalism is the study of the induced presymplectic 
$(d-1,2)$-form $\om_\CL$ on $\CF\times M$, its restriction to $\CF_{\CE \CL}\times M$ and its transgression to an on-shell
presymplectic $2$-form $\int_{\Sigma^{d-1}}\om_\CL|_{\CE \CL}$ on $\CF_{\CE \CL}$ upon integration along an appropriate submanifold, yielding 
the \textit{covariant phase space} of a local field theory as a presymplectic smooth set.

\medskip 
As previously, we rigorously define the relevant structures within the category of smooth sets, while also describing the exact 
correspondence with the formulas appearing in the physics literature.\footnote{A similar formalization is described in \cite{Blohmann23b}, 
which however takes place in `pro-diffeological' spaces. And hence in turn pro-smooth sets, rather than actual smooth sets. 
This is not ideal and has its limitations -- see Rem. \ref{ProVsLocProJetBundle}. Further discrepancies in definitions and results 
in comparison to our approach will be noted throughout this section. However, a discussion of the presymplectic and induced 
Poisson structures does not appear therein.} Our presentation includes a detailed discussion of the 
off-shell presymplectic current of field theories, and the induced
bracket structures of Noether and Hamiltonian local currents, which in good situations restricts to the on-shell product $\CF_{\CE \CL}\times M$. 
Finally, by transgressing the presymplectic current to an actual presymplectic $2$-form on $\CF$, we describe the analogous induced Poisson
bracket on Hamiltonian functionals. In good situations, this restricts to the \textit{covariant Poisson bracket} of Hamiltonian
functionals on the space on-shell fields $\CF_{\CE \CL}$, i.e., the Poisson bracket corresponding to the presymplectic structure 
of the covariant phase space.

\medskip 
In the mathematical physics literature, early accounts for the covariant phase space approach include 
\cite{So77}\cite{Kij73}\cite{KS76} \cite{Gaw72}\cite{GK74}\cite{KT79}. 
It was rediscovered in \cite{Wi86}\cite{Zuckerman}, see also \cite{CrnW}.
Detailed historical accounts, with relevant references, of the development of the covariant
phase space formalism can be found in \cite{Gi23} and \cite{Khavkine14}.
In the latter, the approach is via the infinite jet bundle, and hence closer to ours, where furthermore 
the relation with Peierls bracket (defined in
terms of Green functions, when these exist) is explained in detail.

\medskip 
\noindent {\bf Comment on notation:} In this final section, we omit reference 
to the fully faithful Yoneda embedding symbol, as often practiced in the sheaf categorical literature. This is meant to alleviate the otherwise extremely heavy notation of this final chapter and to act as a stepping stone for the omission of the corresponding Yoneda embeddings that will appear in the following parts of the series. We hope that, by now, the reader should be able to tell from the ambient 
context which category the objects we refer to live in ($\SmoothManifolds$, $\mathrm{LocProMan}$ or $\SmoothSets$).

\subsection{The local bicomplex and its Cartan calculus}
\label{sec-cartan}

We start off by defining the tangent bundle on the product smooth set $\CF \times M \in \SmoothSets$ to be
\vspace{-2mm}
$$
T(\CF\times M):= T\CF \times TM \hspace{0.5cm} \in \hspace{0.5cm} \SmoothSets \, ,
$$

\vspace{-2mm}
\noindent
and so at the level of points $T_{(\phi,p)} (\CF\times M) = T_\phi \CF \times T_p M = \Gamma_M(\phi^*VF) \times T_p M \subset T(\CF\times M)(*)$. 
The intuition is that both tangent vectors on M and $\CF$ are given by `infinitesimal curves' in $M$ and $\CF$ respectively, and thus 
so should tangent vectors on $\CF\times M$. This intuition will, once again, become a constructive definition in the synthetic
sense -- recovering the above (see \cite{GSS-2}). Viewing $\CF \times M$ as a \textit{trivial} bundle over spacetime
\vspace{-2mm}
\[
\xymatrix@=0.7em  { \CF\times M \ar[dd]^{\pr_2} & 
	\\  \\
M
	\, , } 
\]

\vspace{-2mm}
\noindent we may think of the (by definition) splitting of the tangent bundle $T(\CF\times M)$ as a horizontal splitting along the projection to $M$. 
In other words, 
\vspace{-3mm}
\begin{align}\label{SplittingOfProductTangentBundle}
V(\CF\times M) \times_{\CF\times M} H(\CF\times M) :=  (T\CF\times M)\times_{\CF 
 \times M} (\CF\times TM) \cong T\CF \times TM\, , 
\end{align}

\vspace{-1mm}
\noindent
and so $T\CF\times M \hookrightarrow T\CF\times TM$ is the \textit{vertical} subbundle while $\CF\times TM \hookrightarrow T\CF\times TM$ is the \textit{horizontal} subbundle. 

By definition, differential $m$-forms on both the spacetime manifold $M$ and the infinite jet bundle $J^\infty_M F$ (Def. \ref{mformsInftyJetBundle}) 
are given by fiber-wise linear bundle maps out of their tangent bundle (in $\SmoothSets$). We define differential forms on $\CF\times M$ 
in the same vein and in line with differential forms on $\CF$ from Def. \ref{DifferentialFormsOnFieldSpace}.
\begin{definition}[\bf  Forms on $\CF\times M$]\label{DifferentialFormsOnProduct} 
The set of differential m-forms on $\CF\times M$ is defined as
\begin{align}\Omega^m(\CF\times M) := \mathrm{Hom}^{\mathrm{fib.lin. an.}}_{\SmoothSets}
\big(T^{\times m}(\CF\times M)\,,\, \FR \big) \, , 
\end{align} 

\vspace{-1mm} 
\noindent i.e., smooth real-valued, fiber-wise linear antisymmetric maps with respect to the fiber-wise linear structure on the $m$-fold fiber product
\vspace{0mm} 
$$
T^{\times m}(\CF\times M):= T(\CF\times M)\times_{\CF\times M}\cdots \times_{\CF\times M} T(\CF\times M)
$$ 

\vspace{0mm} 
\noindent of the tangent bundle over the $\CF\times M$.  
\end{definition}
As usual, the collection of differential forms of all degrees forms a graded $\FR$-vector space
\vspace{0mm}
$$
\Omega^\bullet(\CF\times M):= \bigoplus_{m\in \NN} \Omega^{m}(\CF\times M )\, .
$$

\vspace{-1mm}
\noindent Furthermore, it is straightforward to define what a \textit{horizontal} and \textit{vertical} form on $\CF \times M$ should be. 
For instance, as with the infinite jet bundle (Def. \ref{HorizontalVertical1formDefinition}), a horizontal $1$-form $\om$ on $\CF\times M$
is one that vanishes under the restriction
\vspace{-2mm} 
\begin{align}\label{Horizontal1formOnProduct}
\om|_{V(\CF\times M)} \;:\; T\CF \times M \longhookrightarrow T\CF\times TM \longrightarrow \FR \;,
\end{align}

\vspace{0mm}
\noindent and analogously for vertical $1$-forms, and even further for horizontal /  vertical $m$-forms on $\CF\times M$. At this point,
the development of a bi-complex structure and, furthermore, a corresponding Cartan calculus on $\CF\times M$ reaches several stumbling blocks:

\begin{remark}[\bf Caveats with the `bicomplex of $\CF\times M$']\label{CaveatsWithTheBicomplexOfProduct} Both \cite{Zuckerman} and \cite{DF99}
seem to assume (albeit unnecessarily) the existence of a bi-complex structure on the set of \textit{all} differential forms\footnote{Although 
they do not explicitly define what this set consists of. We presume they mean a version of our description.} $\Omega^{\bullet}(\CF\times M)$, 
along with a Cartan Calculus with respect to vector fields on $\CF\times M$. However, the (potential) existence of this structure suffers from the following ambiguities:

\noindent {\bf (a)} Even though there is a notion of horizontal/vertical $1$-forms on $\CF\times M$, it is not obvious that \textit{every} $1$-form $\om$ on $\CF\times M$ splits as 
\vspace{-2mm}
$$
\om \overset{?}{=} \om_V +\om_H \, .
$$

\vspace{0mm}
\noindent Recall, the reason this holds for finite-dimensional manifolds and the infinite jet bundle is that $1$-forms, in these cases, are in 1-1 correspondence with
maps of modules of vector fields (Rem. \ref{1formInfinityJetAsFrechetMap}). This in turn uses the paracompact structure of the underlying (Fr\'{e}chet) manifolds to
extend tangent vectors on $J^\infty_M F$ to vector fields. This is not at all obvious for the case of the field space $\CF$, and in turn for $\CF\times M$. In other words, 
this seems to be a non-trivial obstacle towards defining the bigrading on $\Omega^\bullet(\CF\times M)$. A potential proof that might overcome this is sketched in Rem. \ref{CaveatOnTransgressedVerticalDifferential}. Nevertheless, there is a further crucial issue.

\noindent {\bf (b)} There is no obvious definition for a `de Rham differential' 
\vspace{-2mm}
$$ 
\dd_\CF : \Omega^{m}(\CF)\longrightarrow \Omega^{m+1}(\CF)\, .
$$

\vspace{-1mm}
\noindent
Nevertheless, we note that for any vector field $\CZ:\CF \rightarrow T \CF$ there is a natural contraction / interior product operation \footnote{This extends to higher
degrees. For $m=2$, the precomposition $T\CF\times_\CF \CF \xrightarrow{(\id,\CZ)} T\CF\times_\CF T\CF \rightarrow \FR $ defines a $1$-form by noting the canonical 
isomophism $T\CF\times_\CF  \CF \cong T\CF$ as a fiber product over $\CF$. } 
\vspace{-2mm}
\begin{align*}
\iota_{\CZ}:\Omega^1(\CF) &\longrightarrow C^\infty(\CF) \\ 
\om &\longmapsto \om \circ \CZ \, , 
\end{align*}

\vspace{-2mm}
\noindent which is however insufficient for a Cartan calculus in the absence of a differential. Naturally, this is directly related to the observations 
of Rem. \ref{TangentVectorsPathsOfFieldsAndDerivations} and Eq. \eqref{PathSectionDerivation}, since the would-be differential would necessarily define 
a derivation / Lie derivative $\mathbb{L}_{\CZ}= \iota_{\CZ} \dd_\CF$ on $C^\infty(\CF)$. The same statements follow through for the case of $\CF\times M$,
even without the further demand of an underlying bi-grading structure. 

\noindent {\bf (c)} We note that the problems of {\bf (a), (b)} are evaded if one instead takes as a definition of forms on $\CF$ that of the classifying 
space $\mathbold{\Omega}^m_{\mathrm{dR}}$ (Def. \ref{nformsonSmoothSet}). As explained therein, a de Rham differential exists for this version of forms on $\CF$, 
and similarly for those on $\CF\times M$. In fact, one can see that a natural bicomplex structure exists\footnote{To see this, recall that 
$\mathrm{Hom}_{\SmoothSets}(\CF\times M, \, \mathbold{\Omega}^{1}_{\mathrm{dR}})\cong  \mathrm{Hom}_{\SmoothSets}\big(\CF, \, [M, \mathbold{\Omega}^{1}_{\mathrm{dR}}] \big)$
by the internal hom property. Evaluating on test probes, the mapping space on the right is canonically isomorphic to 
$\big(y(\FR)\hat{\otimes} \Omega^1(M) \big) \bigoplus \big(\mathbold{\Omega}^{1}_{\mathrm{dR}} \hat{\otimes} C^\infty(M)\big)$, hence splitting 
a de Rham $1$-form $\om \in \Omega^{1}_\mathrm{dR}(\CF\times M)$ as $\om= \om^{1,0}+ \om^{0,1}$. This generalizes to de Rham m-forms, and hence splitting the
de Rham differential (Def. \ref{DifferentialofSmoothMap}) as $\dd_\mathrm{dR} = \dd_{\mathrm{dR}}^\CF + \dd_{\mathrm{dR}}^M$, inducing a bi-complex structure.}
on $\Omega^{\bullet}_{\mathrm{dR}}(\CF\times M):= \mathrm{Hom}_{\SmoothSets}(\CF\times M, \, \mathbold{\Omega}^{\bullet}_{\mathrm{dR}})$. However, in the 
setting of smooth sets the relation of the (classifying) de Rham forms with those of Def. \ref{DifferentialFormsOnProduct} is unclear.\footnote{As we have hinted
before in Rem. \ref{ClassifyingFormsDoNotClassifyInSmoothSet}, the situation is different in the setting of infinitesimally 
thickened smooth sets where the two notions will, indeed, be intricately related. }
Furthermore, there is no obvious interior product operation, which once 
again obstructs the construction of a Cartan calculus.   
\end{remark}
The resolution, from our perspective, is that \textit{local} classical field theory -- and hence all of the constructions and arguments 
of \cite{Zuckerman}\cite{DF99} -- only requires the existence of a bi-complex structure (and a corresponding Cartan Calculus) on the subset 
of \textit{local} forms on $\CF\times M$ (along with the set of local vector fields), hence completely bypassing the above conundrum.
This \textit{local bicomplex} does \textit{exist} within our context of smooth sets, and does indeed rigorously recover the statements and 
formulas of the above references. 
To that end, recall the smooth evaluation map 
$
\mathrm{ev}:\mathbold{\Gamma}_M(F) \times M \rightarrow F
$
from Eq. \eqref{SmoothEvaluationMap}. There is an analogous smooth evaluation map, which we denote by the same symbol, 
\vspace{-2mm}
\begin{align*}
\mathrm{ev} \,:\, \mathbold{\Gamma}_M(J^\infty F) \times M &\longrightarrow J^\infty_M F \\
(\tilde{\phi},p) &\longmapsto \tilde{\phi}(p) \, , 
\end{align*}

\vspace{-2mm}
\noindent
and similarly for higher plots. Precomposing along the smooth jet prolongation $j^\infty:\CF\rightarrow \mathbold{\Gamma}_M(J^\infty F)$, 
we may define the smooth prolongated evaluation map with values in $J^\infty_M F$.

\begin{definition}[\bf  Prolongated evaluation]\label{ProlongatedEvaluationMap}
The \textit{prolongated evaluation map} $\ev^\infty: \CF\times M \rightarrow J^\infty_M F$ is defined as the composition of maps of smooth sets
\vspace{-3mm}
\begin{align*}
\ev^\infty \,:\, \CF \times M \xlongrightarrow{\quad (j^\infty,\, \id_M)\quad} \mathbold{\Gamma}_M(J^\infty F)\times M \xlongrightarrow{\;\;\ev \;\;} J^\infty_M F\, . 
\end{align*}

\vspace{-2mm}
\noindent Explicitly, at the level of $*$-plots this is given by 
\vspace{-2mm}
$$
\ev^\infty (\phi,p)= j^\infty \phi(p) \hspace{0.5cm} \in \hspace{0.5cm} J^\infty_M F\, ,
$$

\vspace{-2mm}
\noindent and similarly for higher plots.
\end{definition}

At this point, we note that both the source $\CF\times M$ and target $J^\infty_M F$ of the prolongated evaluation consist of smooth spaces with a well-defined 
tangent bundle, in terms of infinitesimal curves. \footnote{In \cite{GSS-2} all of the mentioned tangent bundles will be special instances of the synthetic tangent 
bundle construction, making the intuition of infinitesimal curves precise. The pushforward map motivated below will then follow naturally as the synthetic 
pushforward of $\ev^\infty$.} Moreover, by Lem. \ref{LinePlotsRepresentTangentVectors}, we can represent any tangent vector 
$(\CZ_\phi, X_p)\in T_\phi \CF \times T_p M \subset T\CF \times TM$ by a line plot $(\phi_t, p_t) \in \big(\CF\times M\big) (\FR^1_t)$ starting at $(\phi,p)$. 
Since tangent vectors on $J^\infty_M F$ are also represented by line plots (Lem. \ref{InfinityJetTangentVectorsAsCurves}), this suggests the existence of 
a pushforward map along $\ev^\infty$. Indeed, we may calculate the derivative of $\ev^\infty\circ (\phi_t,p_t): \FR^1 \rightarrow J^\infty_M F$ at $t=0$, 
\vspace{-1mm}
\begin{align*}
\partial_t \big( \ev^\infty(\phi_t ,p_t) \big)  |_{t=0} &= \partial_t \big(j^\infty \phi_t \circ (\id_{\FR^1_t}, p_t) \big)  |_{t=0}  \\
&= \partial_t (j^\infty \phi_t )|_{t=0} ( p_{t=0} ) + \partial_t \big( j^\infty \phi_{t=0} \circ (\id_{\FR^1_t}, p_t) \big)|_{t=0} \\
&= \partial_t (j^\infty \phi_t )|_{t=0} (p) + (\dd j^\infty \phi)_p  (\partial_t p_{t} |_{t=0})  \\
&= j^\infty \CZ_\phi(p) + (\dd j^\infty \phi)_p X_p \, , 
\end{align*}

\vspace{-1mm} 
\noindent
where the derivative is computed as in Lem. \ref{InfinityJetTangentVectorsAsCurves} and Ex. \ref{InfinityJetVerticalVectorFromProlongationofPlot}, 
hence via the corresponding compatible family $\{j^k \phi_t \circ (\id_{\FR^1_t}, p_t): \FR^1 \rightarrow J^k_M F\}_{k\in \NN}$, whereby the 
chain rule applies -- as done implicitly in the second equality. The latter equalities show that the derivative naturally combines the two
previously introduced pushforward maps: The `vertical pushforward' of Eq. \eqref{pushforwardsoftangentfieldvectors} via the prolongation of the tangent 
vector $\CZ_\phi$ at $\phi\in \CF$ evaluated at $p\in M$
\vspace{-1mm}
$$
j^\infty \CZ_\phi (p) \hspace{0.5cm}\in \hspace{0.5cm}  V_{j^\infty_p \phi} J^\infty_M F \, ,
$$

\vspace{-1mm}
\noindent
and the `horizontal pushforward' (horizontal lift) of Eq. \eqref{PushforwardAlongProlongationOfField} via the differential of the prolongation
of $\phi\in \CF$ at $p\in M$ 
\vspace{-1mm}
$$
(\dd j^\infty \phi)_p X_p \hspace{0.5cm} \in \hspace{0.5cm} H_{j^\infty_p \phi} J^\infty_M F\, .
$$

\vspace{-1mm}
\noindent
The resulting tangent vector on $J^\infty_M F$ depends only on the tangent vectors $(\CZ_\phi, X_p)$, and so we may disregard the representative 
line-plots and define the pushforward of $\ev^\infty$ directly as a map of tangent vectors.
\begin{definition}[\bf  Pushforward of prolongated evaluation]\label{PushforwardOfProlongatedEvaluation}
The  \textit{pushforward} of the prolongated evaluation $\ev^\infty: \CF\times M \rightarrow J^\infty_M F$ is defined by
\vspace{-2mm}
\begin{align*}
\dd \ev^\infty \,:\, T\CF\times TM &\; \longrightarrow  \; T J^\infty_M F\\
(\CZ_\phi, X_p) &\; \longmapsto \; j^\infty \CZ_\phi(p) + (\dd j^\infty \phi)_p X_p \,,
\end{align*}

\vspace{-2mm}
\noindent and similarly for higher plots.
\end{definition}
In a compatible coordinate chart for $J^\infty_M F$ around $p\in M$ where $\CZ_{\phi} = \CZ_{\phi}^a \cdot \frac{\partial}{\partial u^a} \in T_\phi \CF $ 
and $X_{p}=X^\mu \cdot \frac{\partial}{\partial x^\mu} |_{p}\in T_p M$, the pushforward is given by (see maps \eqref{pushforwardsoftangentfieldvectors}
and \eqref{PushforwardAlongProlongationOfField})
\vspace{-1mm}
\begin{align}\label{PushforwardOfProlongatedEvaluationLocalCoords}
\bigg(\CZ_{\phi}^a \cdot \frac{\partial}{\partial u^a}, \, X^\mu \cdot \frac{\partial}{\partial x^\mu} \Big|_{p} 
\bigg) \;\; \longmapsto \;\;    \sum_{|I|=0}^\infty \frac{\partial \CZ_\phi^a}{\partial x^I} (p)
\cdot \frac{\partial}{\partial u^a_I} \Big|_{j^\infty_p \phi} + X^\mu \bigg(\frac{\partial}{\partial x^\mu} \Big\vert_{j^\infty_p \phi} 
+ \sum_{|I|=0}^{\infty} \frac{\partial \phi^a} {\partial x^{I+\mu}} (p)
\cdot    \frac{\partial}{\partial u^a_I} \Big\vert_{j^\infty_p \phi}\bigg)\, . 
\end{align}

\vspace{-2mm}
\noindent
We highlight that, by construction, the pushforward $\dd \ev^\infty$ respects the natural splitting of the tangent bundles $T\CF \times TM$ 
from Eq. \eqref{SplittingOfProductTangentBundle}, and that of $TJ^\infty_M F$  from Cor. \ref{CanonicalSmoothSplitting}, as bundles over $M$
\vspace{-3mm} 
\[ 
\xymatrix@C=3.5em@R=.2em  {V(\CF\times M)\times_{\CF\times M} H(\CF\times M) \ar[rd] \ar[rr]^{\quad \, \, \dd \ev^\infty}
&   & VJ^\infty_M F\times_{J^\infty_M F} HJ^\infty_M F \ar[ld]
	\\ 
& M & 
}.  
\]

\vspace{-2mm}
\noindent 
The idea now is to define the subset of \textit{local forms on $\CF\times M$} by pulling back forms in the variational bi-complex on 
$J^\infty_M F$ via $e_\infty:\CF\times M \rightarrow J^\infty_M F$, by which we mean precomposing differential forms with the 
pushforward map on tangent vectors. In other words, the \textit{pullback} \textit{local} form  $(\ev^\infty)^*\om \in \Omega^{\bullet}(\CF\times M)$ 
of a differential form $\om \in \Omega^{\bullet,\bullet}(J^\infty_M F)$ is defined by the composition
\vspace{-3mm}
\begin{align}\label{PullbackAlongProlongatedEvaluation}
(\ev^\infty)^*\om  \;:\; T\CF\times TM \xlongrightarrow{\;\; \dd \ev^\infty\;\; } J^\infty_M F \xlongrightarrow{\;\; \om \;\;} \FR \, . 
\end{align}

\vspace{-2mm}
\noindent Since forms on $J^\infty_M F$ do have a well-defined bigrading, and the pushforward map respects the respective splittings, 
it follows that the image of the map of $\FR$-vector spaces
\vspace{-2mm}
\begin{align*}
(\ev^\infty)^* \;:\; \Omega^{\bullet,\bullet}(J^\infty_M F) \longrightarrow \Omega^{\bullet}(\CF\times M) \, . 
\end{align*}

\vspace{-2mm}
\noindent induces a bi-grading -- whose horizontal and vertical components coincide with the corresponding notions defined \eqref{Horizontal1formOnProduct}
directly on $T\CF \times TM$. Thus, it follows that the vertical and horizontal differentials on $J^\infty_M F$ induce two (anti-commuting) 
differentials on the image of $(\ev^\infty)^*$. In total, the pullback identifies a \textit{bi-complex structure} on the subspace of  
\textit{local forms} on the smooth set $\CF\times M$. 

\begin{definition}[\bf  Bicomplex of local forms]\label{BicomplexOfLocalForms}
The \textit{bicomplex of local forms} on $\CF\times M$ is defined as the image of the variational bicomplex on $J^\infty_M F$ under
the pullback map $(\ev^\infty)^*$. That is, 
\vspace{-1mm}
\begin{align*}
\Omega^{\bullet,\bullet}_{\mathrm{loc}}(\CF\times M) := \mathrm{Im}\big( (\ev^\infty)^*
\; : \; \Omega^{\bullet,\bullet}(J^\infty_M F) \longrightarrow \Omega^\bullet(\CF\times M) \big) \quad \subset \;\; \Omega^\bullet(\CF\times M)\, ,   
\end{align*}

\vspace{-1mm}
\noindent equipped with the induced horizontal and vertical differentials 
\vspace{-2mm}
$$
\dd_M \big((\ev^\infty)^* \om  \big):= (\ev^\infty)^* (\dd_H \om)\, ,  
\hspace{2cm} \delta \big((\ev^\infty)^* \om  \big)\equiv \dd_\CF \big((\ev^\infty)^* \om  \big):= (\ev^\infty)^* (\dd_V \om) \, ,
$$

\vspace{-2mm}
\noindent respectively.
\end{definition}

By construction, the local bi-complex is a module over local functions, or equivalently local $0$-forms on $\CF\times M$, i.e.,
$C^\infty_\mathrm{loc}(\CF\times M) \cong \Omega^{0,0}_\loc (\CF\times M)$. It is instructive to spell out the form of the local 
bicomplex in local coordinates, which will make contact with the formulas of \cite{DF99}. To that end, for any 1-form 
$\om=\om_H+\om_V= (\om_H)_\mu \cdot \dd x^\mu  + \sum_{|I|=0}^{\infty} (\om_V)_a^{I} \cdot \dd_V u^a_I$ on $J^\infty_M F$ 
and any tangent vector $(\CZ_\phi, X_p)\in T_\phi \CF \times T_p M$, we may compute 
\vspace{-2mm}
\begin{align}\label{ActionOfLocal1FormInCoordinates}
\big((\ev^\infty)^* \om  \big)|_{(\phi,p)} (\CZ_\phi, X_p) &= \om|_{j^\infty_p \phi} \Big( (\dd \ev^\infty)_{(\phi,p)} \big(\CZ_\phi,p\big) \Big)\nn  \\
&= (\om_H +\om_V)|_{j^\infty_p \phi} \big( j^\infty \CZ_\phi(p) + (dj^\infty \phi)_p X_p \big)  \\
&= \om_V |_{j^\infty_p \phi} \big(j^\infty \CZ_\phi(p)\big) + \om_H |_{j^\infty_p \phi}\big( (\dd j^\infty \phi)_p X_p\big) \nn \\
&=  \sum_{|I|=0}^\infty \big(\om_V\big)_a^I(j^\infty_p \phi)  \cdot \frac{\partial \CZ_\phi^a}{\partial x^I} (p)
+ \big(\om_H\big)_\mu(j^\infty_p \phi) \cdot X^\mu \, , \nn 
\end{align}

\vspace{-1mm}
\noindent
where in the third equality we used the vertical/horizontal properties of the two tangent vectors, respectively, and finally used their coordinate form \eqref{PushforwardOfProlongatedEvaluationLocalCoords}. Note that, even though $e^\infty: \CF\times M\rightarrow J^\infty_M F$ might not be 
surjective \footnote{In more detail, as per Rem. \ref{FieldSpaceAsASheafOfSheaves}, this is really a map of petit sheaves (and hence the pullback too) -- i.e., 
over each $U\subset M$. Regardless of the (non)-existence of global sections, the restricted map $ev^\infty|_U: \Gamma_U(F)\times U\rightarrow J^\infty_U F$ 
might still not be surjective for all $U\subset M$, and the comment applies for each open neighborhood. It is, however, an epimorphism as a map of (petit) 
sheaves, i.e., surjective on stalks.} or submersive (surjective on tangent spaces), the differentials of a local form are well-defined, i.e., independent of the chosen representative. 
This follows since pullback $(\ev^\infty)^*\om $ local form depends manifestly only the values of $\om\in J^\infty_M F$ along the image
of $\ev^\infty$ (see \eqref{ActionOfLocal1FormInCoordinates}).

\paragraph{Standard notation for the local bicomplex.}  Recalling the abuse of notation from \eqref{VectorFieldOnFieldSpaceAbuseOfNotation}, 
we may denote the vertical tangent vector $\CZ_\phi= \CZ_\phi^a \cdot \frac{\partial}{\partial u^a} \, \in\,  T_\phi \CF = \Gamma_M(\phi^*VF)$ as
\vspace{-1mm}
$$
\CZ_{\phi} = \CZ^a \cdot \frac{\delta}{\delta \phi^a} \Big\vert_{(\phi,p)} \hspace{0.5cm} \in \hspace{0.5cm} T_\phi \CF   \longhookrightarrow T_\phi \CF \times T_p M\, .
$$

\vspace{-1mm}
\noindent 
Following this trend, the pullback of a vertical $1$-form $\om_V = \sum_{|I|=0}^{\infty} (\om_V)_a^{I} \cdot \dd_V u^a_I =\sum_{|I|=0}^{\infty} (\om_V)_a^{I} \big(x^\mu,\{u^b_J\}_{|J|\leq k}\big) \cdot \dd_V u^a_I  $ may be denoted as 
\vspace{-3mm}
\begin{align}\label{LocalVertical1formAbusively}
\big((\ev^\infty)^*\om_V\big) :&= \sum_{|I|=0}^{\infty} (\om_V)_a^{I}\circ \ev^\infty \cdot (\ev^\infty)^* \dd_V u^a_I  
\\[-1pt]
&= \sum_{|I|=0}^{\infty} (\om_V)_a^{I} \big(x^\mu,\{\partial_J \phi^b \}_{|J|\leq k}\big) \cdot \delta(\partial_I \phi^a)   \hspace{0.5cm} \in \hspace{0.5cm} \Omega^{0,1}_\mathrm{loc}(\CF\times M) \nn 
\end{align}

\vspace{-2mm}
\noindent
with the implicit understanding that the notation here\footnote{The abuse of notation would be slightly less confusing if we used the notation $\hat{\phi}^a_I$ from Ex. \ref{FieldPointAmplitudeExample}. We comply with the standard physics literature notation as in \cite{DF99}, to make direct contact with the formulas therein.} 
means $\partial_I \phi^a := (\ev^\infty)^* u^a_I= u^a_I \circ \ev^\infty \in C^\infty_{\mathrm{loc}}(\CF\times M) $, and so
\vspace{-1mm}
$$
\delta(\partial_I \phi^a) \equiv \delta ( u^a_I \circ \ev^\infty) := (\ev^\infty)^* \dd_V u^a_I \, . 
$$

\vspace{-1mm}
\noindent Under this abuse of notation, the rigorous calculation of \eqref{ActionOfLocal1FormInCoordinates} justifies the following useful (formal) manipulation from \cite{DF99}
\vspace{-3mm}
\begin{align}\label{Local1FormActionOnTangentAbusively}
\delta(\partial_I \phi^a)\big\vert_{\phi,p} \bigg( \CZ^b_\phi \cdot \frac{\delta}{\delta \phi^b} \Big\vert_{(\phi,p)}\bigg) = 
\partial_I \bigg( \delta \phi^a \Big(\CZ^b_\phi \cdot \frac{\delta}{\delta \phi^b}\Big) \! \bigg) (p) =
\partial_I \CZ^a_\phi(p) = \frac{\partial \CZ^a_\phi}{\partial x^I} (p) \, ,
\end{align}

\vspace{-1mm}
\noindent
where, as explained, the partial derivative symbol on the left-hand side is a formal symbol, while on the right-hand 
side it is an actual derivative. In other words, with the caveats and abuse of notation mentioned, the `vertical differential commutes with partial derivatives'
\vspace{-2mm}
\begin{align*}
`` \, \delta(\partial_I \phi^a) = \partial_I (\delta \phi^a)\, "\, ,
\end{align*}

\vspace{-2mm}
\noindent
which is precisely how these symbols are usefully manipulated in explicit calculations appearing in the literature.

Along the same lines, given a smooth function $f= f\big(x,\{u^b_J\}_{|J|\leq k}\big)\in C^\infty(J^\infty_M F)$ with induced local function $(\ev^\infty)^*f \in \Omega^{0,0}_{\mathrm{loc}}(\CF\times M)$, again by \eqref{ActionOfLocal1FormInCoordinates}, the corresponding local vertical 1-form may be computed and denoted as
\vspace{-2mm}
\begin{align}\label{VerticalLocalDifferentialInCoordinates}
\delta\big((\ev^\infty)^*f) :&= (\ev^\infty)^* \dd_V f= \sum_{|I|=0}^{\infty} \frac{\partial f}{\partial u^a_I} \circ \ev^\infty \cdot (\ev^\infty)^* \dd_V u^a_I \nn \\
&= \sum_{|I|=0}^{\infty} \frac{\delta f\big(x, \{\partial_J \phi^b\}_{|J|\leq k}\big)}{\delta(\partial_I \phi^a)} \cdot  \delta(\partial_I \phi^a)
\end{align}

\vspace{-2mm}
\noindent
On the other hand, the corresponding local horizontal 1-form is given by
\vspace{-1mm}
\begin{align*}
\dd_M \big((\ev^\infty)^*f) :&= (\ev^\infty)^* \dd_H f = (\ev^\infty)^* (D_\mu f \cdot \dd_H x^\mu) = D_\mu f \circ \ev^\infty \cdot (\ev^\infty)^*\dd_H x^\mu \, , 
\end{align*}
\vspace{-2mm}
\noindent
which at any point $(\phi,p)\in \CF\times M$, again by \eqref{ActionOfLocal1FormInCoordinates}, is further given by 
\vspace{-1mm}
\begin{align*}
\dd_M \big((\ev^\infty)^*f) |_{(\phi,p)} &= D_\mu f (j^\infty_p \phi) \cdot \dd x^\mu \big|_p = 
\Big(\frac{\partial f}{\partial x^\mu}\circ j^\infty \phi \Big)(p) \cdot \dd x^\mu \Big|_p  \\ 
&= \dd_M (f \circ j^\infty\phi) |_{p} 
\end{align*}

\vspace{-2mm}
\noindent
as a map $T_\phi \CF\times T_p M \xrightarrow{\pr_2} T_p M \rightarrow \FR$, where in the second equality we used the chain rule 
\eqref{HorizontalVectorFieldBasisAction}. 
This justifies the notation $\dd_M$ for the horizontal differential on $\Omega^{\bullet,\bullet}_\mathrm{loc}(\CF\times M)$ -- it is 
indeed simply the de Rham differential along the spacetime $M$, computed at each fixed field configuration $\phi \in \CF$ 
\vspace{-2mm}
\begin{align}\label{HorizontalLocalDifferentialInCoords}
\dd_M \big((\ev^\infty)^*f\big)= \frac{\partial f\big(x, \{\partial_J \phi^b \}_{|J|\leq k}\big)}{\partial x^\mu} \cdot \dd x^\mu 
\end{align}

\vspace{-1mm}
\noindent
In these terms, the compatibility of the differentials from Prop. \ref{HorizontalDifferentialBaseDeRhamCompatibility} is then simply the 
further pullback, \footnote{The pullback of the inclusion $\iota_\phi: M\hookrightarrow \CF\times M$ is again defined via its pushforward 
$\dd \iota_\phi: T_p M \hookrightarrow T_\phi \CF \times T_p M $. It follows the pullback is given by 
$\iota_\phi^*= \ev_\phi \circ \pr_H :  \Omega^{\bullet,\bullet}_\mathrm{loc}(\CF\times M)\rightarrow 
\Omega^{\bullet,0}_\mathrm{loc}(\CF\times M)\rightarrow \Omega^\bullet(M)$.} 
for any $\phi \in \CF$, along
$\iota_{\phi}: M\cong \{\phi\} \times M  \longhookrightarrow \CF\times M \, ,$
\vspace{-2mm}
$$
\iota_\phi^* (\ev^\infty)^* \dd_H f = (\ev^\infty \circ \iota_\phi)^* \dd_H f = (j^\infty\phi)^* \dd_H f = \dd_M (f\circ j^\infty\phi) \, .
$$

\vspace{-1mm}

This discussion naturally generalizes to arbitrary local $(p,q)$-forms on $\CF\times M$ and their differentials. 
Indeed, for any $(p,q)$-form $\om\in \Omega^{p,q}(J^\infty_M F)$ with local representation as in \eqref{pqformInfinityJetLocalCoordinates}, 
the pullback local form may be denoted as
\vspace{-1mm}
\begin{align}\label{LocalpqFormAbusively}
(\ev^\infty)^* \om =  \sum_{I_1,\cdots, I_p=0} \om_{\mu_1\cdots \mu_p a_1 \cdots a_q}^{I_1\dots I_q}
\big(x, \{\partial_J \phi^b \}_{|J|\leq k}\big)\cdot  \dd_M x^{\mu_1}\wedge \cdots
\wedge \dd_M x^{\mu_p}\wedge \delta(\partial_{I_1} \phi^{a_1}) \wedge\cdots \wedge \delta(\partial_{I_q} \phi^{a_q}) \, , 
\end{align}

\vspace{-1mm}
\noindent
with the action of the vertical $\delta$ and horizontal $\dd_M$ differentials given as a derivation via \eqref{VerticalLocalDifferentialInCoordinates} and \eqref{HorizontalLocalDifferentialInCoords}. In particular, this local coordinate description shows\footnote{Abstractly, this follows by the internal hom property.
For instance, $\mathrm{Hom}_{\SmoothSets}(\CF\times TM, \, \FR) \cong\mathrm{Hom}_{\SmoothSets}\big(\CF, \, [TM,\FR]\big)$, and so the subset of local, 
fiberwise linear maps can be seen to be 
$\Omega^{1,0}_{\mathrm{loc}}(\CF\times M\, , \FR)\equiv \mathrm{Hom}_{\SmoothSets}^\mathrm{fib.lin.loc.}(\CF\times TM, \, \FR)
\cong\mathrm{Hom}_{\SmoothSets}^\mathrm{loc}\big(\CF, \, \Omega^{1}_{\mathrm{Vert}}(M) \big)\hookrightarrow \mathrm{Hom}_{\SmoothSets}\big(\CF, \, [TM,\FR]\big) $, 
where the latter is given by identifying the subobject $\Omega^1_\mathrm{Vert}(M)\hookrightarrow [TM,\FR]$.} that the local $p$-form currents 
of Def. \ref{CurrentOnFieldSpace}, and equivalently of \eqref{CurrentsViaPullback}, further coincide with the subset of horizontal local forms on $\CF\times M$
\vspace{-1mm}
$$
\big\{ \CP:= P\circ j^\infty \, : \, \CF \rightarrow \Omega^{p}_{\mathrm{Vert}}(M) \big\}  \hspace{0.5cm}  
\iff \hspace{0.5cm} \big\{ \CP:=(\ev^\infty)^*P \in \Omega^{p,0}(\CF \times M) \big\} \, ,
$$

\vspace{-1mm}
\noindent along with the actions of the corresponding differentials along $M$.

\paragraph{Local Cartan calculus.}
Recall for any (smooth) vector field $\CZ: \CF\times M \rightarrow T\CF \times TM$, there is a natural interior product operation
\big(Rem \ref{CaveatsWithTheBicomplexOfProduct} {(b)}\big)
\vspace{-3mm}
\begin{align*}
\iota_{\CZ} \,:\, \Omega^1_\mathrm{loc}(\CF\times M) &\longrightarrow C^\infty(\CF\times M) \\ 
\om &\longmapsto \om\circ \CZ\, ,  
\end{align*}

\vspace{-2mm}
\noindent
since a differential form is, by definition, a smooth map $T\CF\times TM \rightarrow \FR$. As the notation suggests, there is no reason the resulting smooth
function $\iota_\CZ \om = \om\circ \CZ \in C^\infty(\CF\times M)=\Omega^{0}(\CF\times M)$ should be \textit{local}. Nevertheless, this is indeed the case for
\textit{local} vector fields $\CZ= Z\circ j^\infty \in \CX_{\mathrm{loc}}(\CF)\hookrightarrow \CX(\CF\times M)$ viewed as vector fields
on \footnote{Formally, $\CZ= Z\circ j^\infty : \CF \rightarrow T\CF $ defines a smooth vector field on $\CF\times M$ by $(\phi,p)\mapsto (\CZ(\phi), 0_p)\in T_\phi \CF\times T_p M$.
We will use the same notation for the resulting vector field on $\CF\times M$.} $\CF\times M$   constant along $M$, with vanishing horizontal components. 
Indeed, one has
\vspace{-1mm} 
$$
\iota_\CZ\big((\ev^\infty)^*\om\big) = (\ev^\infty)^* (\iota_{\pr Z} \om) \hspace{0.5cm} \in  \hspace{0.5cm} \Omega^{0}_\loc(\CF\times M)
$$

\vspace{-1mm} 
\noindent which can be seen via
\vspace{-2mm} 
\begin{align}\label{ContractionOfLocal1form}
\iota_{\CZ}\big((\ev^\infty)^*\om\big) |_{(\phi,p)} &= \big((\ev^\infty)^*\om \circ \CZ\big) (\phi,p) = (\ev^\infty)^*\om|_{(\phi,p)} (\CZ_\phi)
\nn \\
& =\om_V |_{j^\infty_p \phi} \big(j^\infty (Z\circ j^\infty \phi)(p)\big) =  \om_V |_{j^\infty_p \phi} \big( \pr Z \circ j^\infty \phi (p)\big) \\
&= \om_V (\pr Z) \circ \ev^\infty (\phi,p)
= (\ev^\infty)^* (\iota_{\pr Z} \om) (\phi,p)\, , \nn
\end{align}
where the second line follows by \eqref{ActionOfLocal1FormInCoordinates}, with the identification $j^\infty(Z\circ j^\infty \phi) = \pr Z\circ j^\infty \phi$ given, for instance, by the local descriptions \eqref{pushforwardsoftangentfieldvectors} and \eqref{ProlongatedVectorFieldonFieldSpaceAbusively}. Using the standard notation $\CZ= Z^c\big(x^\mu, \{\partial_J \phi^b\}_{|J|\leq k'}\big) \cdot \frac{\delta}{\delta \phi^c} $ of \eqref{VectorFieldOnFieldSpaceAbuseOfNotation} 
 and $ (\ev^\infty)^*\om=\sum_{|I|=0}^{\infty} (\om_V)_a^{I} \big(x^\mu,\{\partial_J \phi^b \}_{|J|\leq k}\big) \cdot  \delta(\partial_I \phi^a\big)$ of \eqref{LocalVertical1formAbusively}, along with the abuse of \eqref{Local1FormActionOnTangentAbusively}, this may be also calculated as
 \vspace{-1mm}
\begin{align}\label{ContractionOfLocal1formAbusively}
\iota_\CZ\big((\ev^\infty)^*\om\big) 
& = \sum_{|I|=0}^{\infty} (\om_V)_a^{I} \Big(x^\mu,\{\partial_J \phi^b \}_{|J|\leq k}\Big) 
\cdot  \delta(\partial_I \phi^a) \bigg(Z^c\Big(x^\mu, \{\partial_J \phi^b\}_{|J|\leq k'}\Big) 
\cdot\frac{\delta}{\delta \phi^c}\bigg) \nn  \\
&= \sum_{|I|=0}^{\infty} (\om_V)_a^{I} \Big(x^\mu,\{\partial_J \phi^b \}_{|J|\leq k}\Big) 
\cdot  \partial_I  \bigg(Z^c\Big(x^\mu, \{\partial_J \phi^b\}_{|J|\leq k'}\Big) 
\cdot\delta \phi^a \Big(\frac{\delta}{\delta \phi^c}\Big)\!\bigg) \\
&=\sum_{|I|=0}^{\infty} (\om_V)_a^{I} \Big(x^\mu,\{\partial_J \phi^b \}_{|J|\leq k}\Big) 
\cdot  \frac{\partial Z^a\big(x^\mu, \{\partial_J \phi^b\}_{|J|\leq k'} \big)}{\partial x^I} \, , \nn
\end{align}

\vspace{-1mm}
\noindent which is the form in which it appears in \cite{DF99}.

\medskip 
The interior product with respect to local vector fields extends to arbitrary local forms as a (graded) 
derivation \footnote{Equivalently, for any $\om\in \Omega^{2}_{\mathrm{loc}}(\CF\times M)$ the contracted local $1$-form is 
the map $\iota_\CZ \om := T\CF \cong T\CF \times_\CF \CF \xrightarrow{(\id,\CZ)} T\CF\times_\CF T\CF \xrightarrow{\om} \FR$, and similarly for higher forms.}   
\vspace{-2mm}
$$
\iota_{\CZ} \,:\, \Omega^{p,q}_\loc(\CF\times M) \longrightarrow \Omega^{p,q-1}_\mathrm{loc}(\CF\times M) \, ,
$$ 

\vspace{-1mm}
\noindent and so given locally as above by acting on the pullback local forms on their presentation of \eqref{LocalpqFormAbusively}. 
In particular, for a horizontal local form 
$\CP=(\ev^\infty)^*P \in \Omega_\loc^{p,0}(\CF\times M)$, or equivalently $\CP= P\circ j^\infty :\CF \rightarrow \Omega^{p}_\mathrm{Vert}(M)$,
the contraction of its vertical differential $\delta \CP \in \Omega^{p,1}_\loc (\CF\times M)$ along local vector fields immediately
recovers exactly the action of Def. \ref{ActionOfLocalVectorFieldsOnCurrents}
\vspace{-2mm}
$$\iota_\CZ\delta \CP = \CZ(\CP) \hspace{0.5cm} \in \hspace{0.5cm} \Omega^{d,0}_\loc(\CF\times M) \, ,
$$

\vspace{-2mm}
\noindent which may be seen at the abstract level by \eqref{ContractionOfLocal1form}, or via the coordinate formulas \eqref{ContractionOfLocal1formAbusively}. This suggests that the action of local vector fields on $p$-form currents is a kind of Lie derivative on $\CF\times M$. 

\begin{definition}[\bf  Lie derivative along local vector field]
\label{LieDerivativeOfLocalVectorField}
Let $\CZ\in \CX_{\loc}(\CF)$ be any local vector field. The Lie derivative along $\CZ$ is defined by $\mathbb{L}_\CZ:= [\iota_\CZ, \delta]$, 
that is
\vspace{-2mm} 
\begin{align*}
\mathbb{L}_\CZ  \,:\, \Omega^{p,q}_\loc(\CF\times M) & \; \longrightarrow \; \Omega^{p,q}_\loc(\CF\times M) \\
(\ev^\infty)^*\om & \; \longmapsto \; \iota_\CZ \big(\delta (\ev^\infty)^*\om \big) + \delta \big(\iota_\CZ (\ev^\infty)^*\om \big) \, .
\end{align*}
\end{definition}
Expanding the right-hand side using Def. \ref{BicomplexOfLocalForms} and \eqref{ContractionOfLocal1form}, this is equivalently
\vspace{-1mm}
\begin{align*}
\mathbb{L}_\CZ (\ev^\infty)^*\om = (\ev^\infty)^* ( \iota_{\pr Z} \dd_V \om + \dd_V\iota_{\pr Z} \om) = (\ev^\infty)^*(\mathbb{L}_{\pr Z} \om) \, . 
\end{align*}

\vspace{-1mm}
\noindent We could have used the total differential $\delta+\dd_M$ to define the Lie derivative, but this is redundant since the interior product
of vertical local vector fields on $\CF\times M$ commutes with $\dd_M$,
\vspace{-2mm}
$$
\iota_\CZ \dd_M = - \dd_M \iota_\CZ \, , 
$$

\vspace{-2mm}
\noindent
as can be seen abstractly, or by the local coordinate reprentations. More generally, we have the following result.

\begin{proposition}[\bf Local Cartan calculus]\label{LocalCartanCalculus}
The Lie algebra of local vector fields  $\big(\CX_\loc(\CF),\, [-,-])$ (Def. \ref{LieAlgebraOfLocalVectorFields}) acts on the local bicomplex 
$\big(\Omega^{\bullet,\bullet}(\CF\times M),\, \delta,\dd_M\, \big)$ via the Lie derivative (Def. \ref{LieDerivativeOfLocalVectorField}) 
and satisfies 
\vspace{-2mm}
$$
\mathbb{L}_\CZ \, \dd_M = \dd_M \mathbb{L}_\CZ \, , \hspace{2cm} \mathbb{L}_\CZ \, \delta = \delta\,  \mathbb{L}_\CZ \, . 
$$
\end{proposition}
\begin{proof}
This follows essentially by pulling back the relations of Cor. \ref{EvolutionaryVectorFields}. It can also be deduced directly using the 
local descriptions of \eqref{LocalVertical1formAbusively}, \eqref{HorizontalLocalDifferentialInCoords}, \eqref{LocalpqFormAbusively}, and \eqref{ContractionOfLocal1formAbusively}, 
all of which have been rigorously justified.
\end{proof}

This local Cartan Calculus is essentially the core machinery in the descriptions of \cite{Zuckerman} and \cite{DF99}, here defined in full 
detail with $\SmoothSets$. It allows us to further explain the notation in \cref{EvolutionaryVectorFieldsAndNoetherTheoremsSection}, 
and also to lift the proofs of statements therein to proofs computed directly on $\CF\times M$, as usually (implicitly) practiced in 
the physics literature. For instance, the `infinitesimal transformation of a field' notation of \eqref{VectorFieldOnFieldSpaceAbuseOfNotation} 
is rigorously identified as the contraction 
\vspace{-2mm}
$$
\iota_\CZ \delta \phi^a \equiv \delta_{\CZ} \phi^a \, ,
$$

\vspace{-2mm}
\noindent
where now $\delta \phi^a$ is really the vertical differential of the (locally defined) smooth function $\phi^a=u^a\circ ev^\infty$ 
on $\CF\times M$. Analogously, the corresponding abuse of notation from Def. \ref{ActionOfLocalVectorFieldsOnCurrents} for the action 
of local vector fields on currents is identified as the contraction
\vspace{-2mm}
$$
\iota_\CZ \delta \CP \equiv \delta_\CZ \CP\, .  
$$

\vspace{-1mm}
\noindent Similarly, the calculation that infinitesimal local symmetries of a local Lagrangian $\CL\in \Omega^{d,0}_\loc(\CF\times M)$ 
close as a subalgebra \eqref{SubalgebraofLocalInfinitesimalSymmetriesCalculation} 
is equivalently expressed as
\vspace{-2mm}
\begin{align*}
\mathbb{L}_{[\CZ_1, \CZ_2]} \CL &= \mathbb{L}_{\CZ_1}\, (\mathbb{L}_{\CZ_2} \CL) - \mathbb{L}_{\CZ_2} ( \mathbb{L}_{\CZ_1} \CL) 
=\mathbb{L}_{\CZ_1} (\dd_M \CK_{\CZ_2}) - \mathbb{L}_{ \CZ_2} ( \dd_M \CK_{\CZ_1})\\ 
&= \dd_M( \mathbb{L}_{\CZ_1} \CK_{\CZ_2} -\mathbb{L}_{\CZ_2} \CK_{\CZ_1}) \, , 
\end{align*}

\vspace{-1mm}

\noindent and, similarly, the proof of Noether's First Theorem (Prop. \ref{Noether1st}) as
\vspace{-2mm} 
\begin{align*}
\dd_M (\CP_\CZ)&= \dd_M \CK_\CZ + \dd_M \iota_{\CZ} \theta_\CL = \mathbb{L}_{\CZ} \CL -\iota_{\CZ} \dd_M \theta_\CL \\
&=\iota_{\CZ} \delta \CL - \iota_{\CZ} \dd_M \theta_\CL =\iota_{\CZ} \CE \CL + \iota_{\CZ} \dd_M \theta_\CL - \iota_{\CZ} \dd_M \theta_\CL\\
&= \iota_{\CZ} \CE\CL\, ,
\end{align*}

\vspace{-2mm}
\noindent where $\CE \CL$ here stands for the induced vertical 1-form on $\CF\times M$ given by
$(\CZ_\phi, X_p) \mapsto \langle \CE \CL(\phi), \CZ_\phi\rangle (p) $. Following the above, it immediately follows that the 
conserved current $\CP_{[\CZ_1,\CZ_2]}$ corresponding to the Lie bracket $[\CZ_1,\CZ_2]\in \CX_\loc^\CL(\CF)$ of two local symmetries is given by
\vspace{-4mm}
\begin{align}\label{CurrentOfLiebracketOfSymmetries}
 \CP_{[\CZ_1,\CZ_2]}= \mathbb{L}_{\CZ_1} \CK_{\CZ_2} - \mathbb{L}_{\CZ_2} \CK_{\CZ_1} + \iota_{[\CZ_1,\CZ_2]} \theta_\CL\, .
\end{align}

\vspace{-2mm}
\noindent Once again, we note that this actually defines a family of conserved currents, with the freedom to add $\dd_M$-closed terms 
to $\CK_{\CZ_1}, \CK_{\CZ_2}$ and $\dd_M$-exact terms to $\theta_\CL$. We will come back with a better characterization of the 
right-hand side in terms of the conserved currents $\CP_{\CZ_1}$ and $\CP_{\CZ_2}$ themselves via Def. \ref{BracketOfNoetherPairs}.

\medskip 
Expanding the above formulas in local coordinates along the lines of \eqref{LocalVertical1formAbusively}, \eqref{HorizontalLocalDifferentialInCoords}, 
\eqref{LocalpqFormAbusively}, and \eqref{ContractionOfLocal1formAbusively}, recovers verbatim those appearing in the physics literature, 
thus rigorously fully justifying their validity as statements about \textit{smooth} and \textit{local} geometry on the full \textit{smooth} 
space of fields $\CF$ and its smooth subspace of on-shell fields $\CF_{\CE \CL}$.

\begin{remark}[\bf Extended Local Cartan calculus]
$\,$

\noindent {\bf (i)} At a purely mathematical level, an analogous discussion applies for a larger class of `local' vector fields on $\CF\times M$: 

\begin{itemize} 
\item[{\bf(a)}] Decomposable vector fields of the form $\CZ+X \in \CX_\loc(\CF)\times \CX(M)\hookrightarrow \CX(\CF\times M)$ are considered in \cite{DF99}.
The corresponding Cartan calculus corresponds to two commuting Cartan calculi given by the local vertical of Lem. \ref{LocalCartanCalculus}, 
along with that of the base manifold $\big([-,-]_M, \dd_M , \iota \big)$. In particular, the Lie derivative in this case is the decomposed one 
$\mathbb{L}_{\CZ + X}= [\iota_\CZ ,\delta] + [\iota_X , \dd_M] \equiv \mathbb{L}_\CZ + \mathbb{L}_X$.

\item[{\bf (b)}] It can be further (maximally) extended to include vector fields induced by general smooth bundle maps $J^\infty_M F\rightarrow T F$, 
whose prolongation splits (\cite[Prop. 1.20]{Anderson89}) as the prolongation of an evolutionary vector field $Z:J^\infty M\rightarrow VF$, 
and the horizontal lift of some bundle map  $X:J^\infty_M F\rightarrow TM$ over $M$. Thus the induced local vector fields are of the form 
$\CZ + \CX $ on $\CF\times M$, where $\CZ$ is a local vector field on $\CF$, and $\CX(\phi)$ is a vector field on $M$ for each $\phi \in \CF$, 
with a local but \textit{not constant} dependence along $\CF$. The corresponding larger, more complicated, Cartan Calculus on $\CF\times M$ 
is carefully spelled out in \cite{Del}, where the vector fields in question are referred to as ``\textit{insular}''. In this case, the Lie
derivative is given by the general form with respect to the total differential 
$\mathbb{L}_{\CZ+ \CX}= [\iota_{\CZ+ \CX}, \delta+ \dd_M] = \mathbb{L}_\CZ + [\CX, \delta+ \dd_M]$, covering the previous cases of vertical local 
and decomposable vector fields as special instances.
\end{itemize} 
\vspace{.5mm} 
\noindent {\bf (ii)}  From a physical and field-theoretic perspective, these extensions seem redundant. 
Indeed, the vector fields needed are to be physically interpreted,
in particular, as infinitesimal symmetries of the \textit{actual field space} $\CF$. As we have shown throughout \cref{EvolutionaryVectorFieldsAndNoetherTheoremsSection}
(and in particular Lem. \ref{InfinitesimalSpacetimeCovariantSymmetries}), this notion is fully captured solely by local vector fields on $\CF$. 
This perspective is implicitly shared by \cite{Zuckerman}. In \cite{DF99}, the addition of a vector field along $M$ (constant along $\CF$) is 
used to bring the action of infinitesimal local symmetries of a Lagrangian into a nicer form, therein termed a ``\textit{manifest symmetry}''.
Nevertheless, this is \textit{not always} possible (see therein) and, further, the action of any decomposable symmetry $\CZ + X$ on a
Lagrangian can always be re-expressed\footnote{This follows easily: In the extended local Cartan calculus, 
$\mathbb{L}_{\CZ+X} \CL = \mathbb{L}_\CZ \CL + \mathbb{L}_X \CL = \mathbb{L}_\CZ \CL + \dd_M (\iota_X \CL)$ since $\CL$ is in particular a top-form along $M$. 
Thus if $\CZ+X$ is an infinitesimal symmetry, $\mathbb{L}_{\CZ+X} \CL  = \dd_M \CK_{\CZ +X}$ for some local $(d-1,0)$-form $\CK_{\CZ +X}$,
and so $\CZ$ is a symmetry with $\mathbb{L}_\CZ \CL = \dd_M ( \CK_{\CZ+X} - \iota_X \CL)$.} as a symmetry of $\CZ$ solely. Hence, the 
extension of {\bf(a)} above seems to be merely a calculational tool.
The more general vector fields of {\bf (b)}, have (currently) no physical interpretation or application, as far as we are aware. 
\end{remark}

\paragraph{\bf Relation to de Rham forms on $\CF\times M$.} As hinted in Rem. \ref{CaveatsWithTheBicomplexOfProduct} {\bf{(c)}} and the corresponding 
footnote, the de Rham forms on $\CF\times M$ defined via the classifying space $\mathbold{\Omega}^\bullet_{\mathrm{d R}}$ carry a natural bicomplex structure. 
Employing Lem. \ref{JetBundleDiffFormsAsDeRhamForms}, it is possible to identify the bicomplex of local forms on $\CF\times M$ arising from globally 
finite order forms on the jet bundle with a subcomplex of the de Rham forms.

\begin{lemma}[\bf Local forms on $\CF\times M$ as de Rham forms]\label{LocalDiffFormsAsDeRhamForms}
The subalgebra $\Omega^\bullet_{\loc,\mathrm{glb}}(\CF\times M)\hookrightarrow \Omega^\bullet_\loc (\CF\times M)$ arising by pulling back the globally
finite order differential forms $\Omega^\bullet_\mathrm{glb} (J^\infty_M F)$ on the infinite jet bundle is canonically identified with a subalgebra
of de Rham forms on the $\CF\times M$. That is, there is a canonical injective DCGA map 
\begin{align*}
 \Omega^\bullet_{\loc,\mathrm{glb}} (\CF\times M)\longhookrightarrow \Omega^\bullet_\mathrm{d R} (\CF\times M)\, ,
\end{align*}
\noindent
 that furthermore respects the corresponding bi-complex structures. 
\end{lemma}
\begin{proof}
 In cases where the pushforward of the prolongated evaluation map $\dd \ev^\infty : T(\CF\times M)\rightarrow T(J^\infty_M F)$ is surjective, and hence 
 its precomposition is injective on the sets of differential forms, this follows immediately by Lem. \ref{JetBundleDiffFormsAsDeRhamForms}. 
 More explicitly, for a local $m$-form 
 \vspace{-0mm}
 $$
 (\mathrm{ev}^\infty)^* \om = \om  \circ \dd \ev^\infty \;\; : \;\; T^{\times m}(\CF\times M) \longrightarrow T^{\times m}(J^\infty_M F) \longrightarrow \FR \, ,
 $$

\vspace{-1mm}
\noindent where $\om = \pi_k^*\om^k : T^{\times m}(J^\infty_M F) \rightarrow T^{\times m}(J^k_M F)\rightarrow \FR $, the corresponding de Rham $m$-form is given by 
\vspace{-2mm}
 $$
 \tilde{\om}\circ \ev^\infty = \tilde{\om}^k\circ \pi_k \circ \ev^\infty \;\; : \;\; \CF\times M \longrightarrow J^\infty_M F\longrightarrow J^k_M F \longrightarrow \mathbold{\Omega}_\mathrm{d R} \, \, .
 $$

\vspace{-2mm}
\noindent
 In the cases where the pushfoward is not surjective, the statement follows by noticing that any local form on $\CF\times M$ depends on the underlying 
 jet bundle differential form only via its values along the image of $\ev^\infty$ inside $J^\infty_M F$. That is, if $\om_1^k, \om_2^k \in \Omega^m(J^k_M F)$
 define the same local $m$-form $ (\mathrm{ev}^\infty)^* \pi_k^* \om_1^k = (\mathrm{ev}^\infty)^* \pi_k^* \om_2^k \in \Omega^m(\CF \times M)$ then it follows 
 that the corresponding de Rham forms $\tilde{\om}_1^k, \tilde{\om}_2^k : J^k_M F\rightarrow \mathbold{\Omega}^m_{\mathrm{d R}}$ define the same de Rham $m$-forms 
 \vspace{-2mm}
 $$
 \tilde{\om}_1^k\circ \pi_k \circ \ev^\infty= \tilde{\om}_2^k\circ \pi_k \circ \ev^\infty \;\; :  \;\; \CF\times M \longrightarrow \Omega^m_{\mathrm {d R}} \, ,
 $$

 \vspace{-2mm}
 \noindent
 since the latter depend only on the image of (plots of) $\ev^\infty$ in $J^\infty_M F$, and vice versa.
 
 To see that above map respects the corresponding bigradings, recall that the bigrading of de Rham forms is induced precisely by the product structure of
 $\CF \times M$ (see Rem. \ref{CaveatsWithTheBicomplexOfProduct} {\bf{(c)}} and the corresponding footnote). On the other hand, the bi-grading of local 
 differential forms on $\CF\times M$, is induced by that of $J^\infty_M F$ -- which in turn corresponds to the product structure of $T(\CF \times M)$(see Eqs. \eqref{SplittingOfProductTangentBundle} and \eqref{PullbackAlongProlongatedEvaluation}). It follows that the inclusion 
 $\Omega^\bullet_{\loc,\mathrm{glb}} (\CF\times M)\longhookrightarrow \Omega^\bullet_\mathrm{d R} (\CF\times M)$ necessarily respects the two bigradings.
 Arguing along the same lines, it follows that it also respects the corresponding differentials and wedge products.
\end{proof}

If, as we expect, it is actually the case that $\Omega^\bullet(J^\infty_M F)\cong \Omega^\bullet_{\mathrm{d R}}(J^\infty_M F)$ as maps of smooth sets, 
then the above embedding canonically extends to local forms on $\CF\times M$ induced by any \textit{locally} finite order forms on $J^\infty_M F$ 
(although this is not necessary for virtually all existing examples of fields theories).
The identification of the above Lemma is implicit in  \cite{FSS13}\cite{FSS14}\cite{FRS14}, and is now officially formally justified. This abstract 
reinterpretation of local forms on $\CF\times M$ as (particular) maps into the classifying space $\mathbold{\Omega}^{\bullet}_{\mathrm{d R}}$ allows 
for many useful categorical arguments and constructions (see therein). Nevertheless, as we have made clear, the traditional description as maps out 
of the tangent bundle is the most appropriate picture to make direct contact with the computational formulas and symbols appearing in the physics 
literature. This fact will be further amplified in the following sections.

\newpage

\subsection{Presymplectic current and induced brackets}
\label{SecPresymplecticCurrentPropertiesAndBrackets}

In this section, among several related results, we will rigorously formalize the observation of Zuckerman \cite{Zuckerman} 
that any local Lagrangian field theory $(\CF,\CL)$ induces an ``on-shell conserved (pre)symplectic $(d-1)$-current'' 
on its on-shell space of fields. To that end, recall the (cohomological) 
`integration by parts formula' on the jet bundle \eqref{LagrangianVerticalDifferentialDecomposition}, which may 
now be pulled back via $\ev^\infty : \CF\times M \rightarrow J^\infty_M F$ to give
\vspace{-3mm} 
\begin{align}\label{CohomologicalIntegrationByPartsOnFieldSpace}
\delta \CL = \CE \CL + \dd_M \theta_\CL \hspace{0.5cm} \in \hspace{0.5cm} \Omega^{d,1}_{\loc}(\CF\times M)  \, ,
\end{align}

\vspace{-2mm} 
\noindent where the term $\CE \CL$ here stands for the corresponding vertical 1-form on $\CF\times M$ given by
\vspace{-1mm}
\begin{align}\label{EulerLagrangeOperatorAs(d,1)form}
(\CZ_\phi, X_p) \longmapsto \langle \CE \CL(\phi), \CZ_\phi\rangle (p) + \iota_{X_p} \CE \CL(\phi) \, .
\end{align}

\vspace{-1mm}
\noindent
Note that the latter term of the decomposition\footnote{ The above decomposition gives an equivalent characterization of
on-shell fields, as those over which the vertical differential of the Langrangian 
\vspace{-2mm}
$$\delta \CL|_{\phi}:   \big(T_{\phi} \CF \times TM\big)^{\times (d+1)} \longhookrightarrow  \big(T \CF \times TM\big)^{\times (d+1)} \longrightarrow \FR  
$$

\vspace{-1mm}
\noindent vanishes, up to a $\dd_M$-exact \textit{local} form $\dd_M \theta_\CL$ on $\CF\times M$. However, this is not as natural as the criticality condition and pullback smooth set characterization.} 
\eqref{LagrangianVerticalDifferentialDecomposition}, i.e. the `boundary term',  is defined up to a \textit{choice} of $\theta_\CL:= (\ev^\infty)^*\theta_L$, 
with the freedom to add a $\dd_M$-closed local $(d-1,1)$-form. In fact, due to Prop. \ref{AcyclicityTheorem}, this freedom is 
necessarily up to $\dd_M$-exact local $(d-1,1)$-forms. 

\begin{definition}[\bf  Presymplectic current]\label{PresymplecticCurrent}
The \textit{presymplectic current} of a Lagrangian field theory $(\CF, \CL)$ is defined by
\vspace{-1mm} 
$$
\om_\CL := \delta \theta_\CL \quad \in \quad  \Omega^{d,2}_\loc(\CF\times M) ,
$$

\vspace{-1mm} 
\noindent
up to a \textit{choice} of a $\dd_M$-exact $(d,1)$-form. 
\end{definition}
This object is also termed the ``universal current'' in \cite{Zuckerman}. The presymplectic adjective is justified since it 
is in particular a $2$-form in the direction of $\CF$, and further $\delta$-closed by construction 
\vspace{-1mm}
$$
\delta \om_\CL = \delta^2 \theta_\CL = 0 \;.
$$

\vspace{-1mm}
\noindent 
It is often the case that this is a \textit{degenerate} $2$-form in the $\CF$ direction, in an appropriate sense, and hence generally only 
\textit{presymplectic} and not actually symplectic. In particular, this degeneracy occurs whenever infinitesimal gauge symmetries exist (Prop. \ref{Gaugesymmetryvspresymplecticcurrent}). 
Due to this definition, the chosen local $(d-1,1)$-form 
$\theta_\CL$ is also known as the \textit{presymplectic potential current}.

\begin{example}[\bf Presymplectic current for the O($n$)-model and electromagnetism]\label{PresymplecticCurrentExamples}
$\,$

\noindent {\bf {(i)}} By pulling back the variational decomposition on $J^\infty_M F$ from Ex. \ref{O(n)ModelOnshellFields}, 
or computing directly on $\CF\times M$, it follows that for the O($n$)-model Lagrangian
\vspace{-2mm}
\begin{align*}
\delta \CL = \CE\CL + \dd_M \theta_\CL = - \langle \dd_M \star \dd_M \phi + \star \phi, \delta \phi \rangle  
- \dd_M ( \langle \delta \phi \,, \star \dd_M \phi  \rangle \big) \, ,  
\end{align*}

\vspace{-2mm}
\noindent
and so the presymplectic potential and current are given by
\vspace{-2mm}
\begin{align*}
\theta_\CL= - \langle \delta \phi \,, \star \dd_M \phi  \rangle = \delta \phi^a \wedge \star \dd_M \phi_a \qquad \mathrm{and} \qquad \om_\CL = \delta \theta_\CL = + \delta \phi^a \wedge \star \delta \dd_M \phi_a \, . 
\end{align*}

\vspace{-2mm}
\noindent
{\bf(ii)} Similarly, the variational decomposition of the Lagrangian of pure electromagnetism from Ex. \ref{ElectromagnetismGaugeSymmetry} reads 
\vspace{-2mm}
\begin{align*}
\delta \CL = \CE\CL + \dd_M \theta_\CL = - \delta A \wedge \dd_M  \star \dd_M A - \dd_M (\delta A \wedge \star \dd_M A)  \, , 
\end{align*}

\vspace{-2mm}
\noindent
and so the presymplectic potential and current are given by
\vspace{-2mm}
\begin{align*}
\theta_\CL= - \delta A\wedge  \star \dd_M A \qquad \mathrm{and} \qquad \om_\CL = \delta \theta_\CL = \delta A \wedge \star \delta \dd_M A \, . 
\end{align*}

\vspace{-2mm}
\noindent

\end{example}

The idea is now to `pullback/restrict' the presymplectic current on $\CF\times M$ to a $(d-1,2)$-form on the smooth subspace 
$\CF_{\CE \CL}\times M$ 
and deduce its resulting properties by restricting Eq. \eqref{CohomologicalIntegrationByPartsOnFieldSpace}. 
For this to make sense, an appropriate definition of the tangent bundle (and hence forms)
on $\CF_{\CE \CL}$ is required. 

\medskip 
The correct intuition behind the construction is, once again, tangent vectors as first-order infinitesimal curves within 
$\CF_{\CE \CL}$. As pointed out throughout the text, this too will become a special instance of the synthetic tangent bundle
in \cite{GSS-2}. For now, we motivate this by considering actual $1$-parameter families of on-shell fields and calculating 
the induced condition on the corresponding tangent vectors. 

\medskip 
Consider first a 1-parameter family of (off-shell) fields $\phi_t\in \CF(\FR^1)$ 
starting at $\phi \in \CF(*)$ and the corresponding image $\FR^1$-plot of the variational cotangent bundle
$T^*_\mathrm{var}(\CF)= \mathbold{\Gamma}_M(V^*F \otimes \wedge^d T^*M)$ under the Euler--Lagrange operator
\vspace{-1mm} 
$$
\CE \CL (\phi_t) \quad \in \quad T^*_\var \CF(\FR^1)  \, .
$$

\vspace{-1mm} 
\noindent 
Denoting the corresponding tangent vector by $\CZ_\phi:= \partial_t \phi_t |_{t=0} \in T_\phi \CF$ and computing the 
derivative \footnote{As usual, this means pointwise in $M$ (see Ex. \ref{TangentVectorsOnFieldSpace}).} 
at $t=0$ in local coordinates via the chain rule gives
\vspace{-2mm} 
\begin{align*}
 \partial_t \CE \CL_a (\phi_t) \big\vert_{t=0} 
 &= \partial_t\big( EL_a  \circ j^\infty \phi_t \big) 
 \\
& = \sum_{|I|=0} \bigg(\frac{\partial EL_a}{\partial u^b_I}\circ j^\infty\phi \bigg) 
 \cdot \partial_t\big(u^b_I\circ j^\infty \phi_t\big) \big\vert_{t=0} \\ 
 & = \sum_{|I|=0} 
\bigg(\frac{\partial EL_a}{\partial u^b_I}\circ j^\infty\phi \bigg) \cdot 
 \frac{
\partial}{\partial t} \frac{\partial \phi^b_t}{\partial x^I}\Big\vert_{t=0} \\
&=  \sum_{|I|=0} 
 \frac{\delta \CE \CL_a (\phi)}{\delta (\partial_I \phi^b)} \cdot 
 \frac{\partial\CZ_\phi^b}{\partial x^I}\; .
\end{align*}
The right-hand side manifestly depends only on the underlying tangent vector at $\phi\in \CF(*)$, and furthermore defines a tangent vector 
\footnote{The variational cotangent bundle is a smooth set of sections, hence Def. \ref{KinematicalTangentBundleToFieldSpace} applies. 
We will not need the explicit form of this bundle.} in $T_\var^*(\CF)$. Hence, as with the pushforward of the prolongated evaluation 
(Def. \ref{PushforwardOfProlongatedEvaluation}), it defines a pushforward map of tangent bundles along $\CE \CL$, which does in fact naturally correspond 
to the synthetic pushforward \cite{GSS-2}.

\begin{definition}[\bf  Pushforward of Euler--Lagrange operator]
The \textit{pushforward of the Euler--Lagrange operator} $\CE \CL: \CF \rightarrow T^*_\var \CF $ is defined by
\vspace{-2mm}
\begin{align*}
\CE \CL_*: T\CF & \; \longrightarrow \; T (T^*_\var \CF)\\
\CZ_\phi & \; \longmapsto \;  \sum_{|I|=0} 
 \frac{\delta \CE \CL_a (\phi)}{\delta (\partial_I \phi^b)} \cdot 
 \frac{\partial\CZ_\phi^b}{\partial x^I} \cdot \dd u^a \, ,
\end{align*}

\vspace{-2mm}
\noindent and similarly for higher plots.

\end{definition}
Here we use $\CE \CL_*$ for the pushforward rather than `$\dd \CE \CL$' so as to avoid confusion with the actual differentials of the bicomplex, 
when $\CE \CL$ is viewed as a $(d,1)$-form on $\CF \times M$. Consider now the case where $\phi_t\in \CF_{\CE \CL}(\FR^1)$ is a 1-parameter 
family of \textit{on-shell fields} starting at $\phi \in \CF_{\CE \CL}(*)$, and so 
\vspace{0mm}
$$
\CE \CL(\phi_t) = 0 \hspace{0.5cm} \in \hspace{0.5cm} T^*_\var\CF(\FR^1)\, .
$$

\vspace{0mm}
\noindent 
Differentiating as above, it follows that corresponding the tangent vector $\CZ_\phi$ in $\CF$ is \textit{tangent to $\CF_{\CE \CL}$} 
in the sense that it vanishes along the pushforward
\vspace{-2mm}
\begin{align}\label{JacobiEquation}
\CE \CL_*(\CZ_\phi)= \sum_{|I|=0} 
 \frac{\delta \CE \CL_a (\phi)}{\delta (\partial_I \phi^b)} \cdot 
 \frac{\partial\CZ_\phi^b}{\partial x^I} \cdot \dd u^a= 0  \hspace{0.5cm} \in \hspace{0.5cm} T_{\CE\CL(\phi)} (T^*_\var \CF) \;.
\end{align}

\vspace{-2mm}
\noindent 
This equation is also known in the literature as the ``linearized Euler--Lagrange equations'' or the ``\textit{Jacobi equation}'' \cite{Zuckerman}, 
and the corresponding tangent vectors $\CZ_\phi \in \Gamma_M(\phi^*VF)$ as ``\textit{Jacobi fields}''. The non-trivial conditions on
the tangent vector are those of $|I|\geq 1$, since the $|I|=0$ case vanishes by the Euler--Lagrange equation on $\phi\in \CF_{\CE \CL}(*)$. 
The set of all tangent vectors to $\CF_{\CE \CL}$ may be thought of as the pullback/intersection of the pushforward
$\CE \CL_*: T\CF(*) \rightarrow T(T^*_\var \CF)(*)$ and the (fiberwise) zero map $0_{\CF * }: T\CF(*) \rightarrow T(T^*_\var \CF)(*)$.
All sets appearing in this intersection are of course the underlying points of actual smooth sets, 
and hence the corresponding \textit{tangent bundle smooth set}
may be defined as the pullback/intersection with the zero map in $\SmoothSets$.
\begin{definition}[\bf  Tangent bundle of on-shell field space]\label{OnShellTangentBundle}
The \textit{smooth tangent bundle} $T \CF_{\CE\CL}$ to the smooth subspace of on-shell fields is defined as the pullback 
  \vspace{-1mm} 	
 \[
	\xymatrix@=1.6em  {T \CF_{\CE \CL} \ar[d] \ar[rr] &&   T\CF \ar[d]^{\mathcal{EL}_*} 
		\\ 
		T\CF \ar[rr]^-{{0_\CF}_*}  && T(T^*_\var\CF)
	\, , }
	\]

   \vspace{-2mm} 
\noindent 	
in $\SmoothSets$. In other words, this is the smooth spaces with $\FR^k$-plots

\vspace{-2mm}
$$
T\CF_{\CE\CL}(\FR^k)= \big\{ \CZ_{\phi^k} \in T \CF \;  | \; \CE \CL_*(\CZ_{\phi^k})= 0_{\phi^k} \big\} \, .
$$
\end{definition}

\newpage 

\begin{remark}[\bf On-shell tangent vectors vs line plots]
$\,$

\noindent {\bf (i)} This is in line with the definition of on-shell tangent vectors of \cite{Zuckerman}, 
and hence also implicit in \cite{DF99}. 

\noindent {\bf (ii)} However, even though the definition is motivated by considering on-shell line plots and their corresponding on-shell 
tangent vectors, 
it is \textit{not true} that \textit{every} on-shell tangent vector $\CZ_\phi \in T \CF_{\CE \CL}(*)$ is represented by a line plot 
in $\CF_{\CE \CL}$ for arbitrary Lagrangians. In other words, the analog of Lem. \ref{LinePlotsRepresentTangentVectors} fails for the on-shell space of fields.

\noindent {\bf (iii)} In fact, this does not stem from the infinite dimensionality of the setting, but can be also seen in the
$0$-dimensional field theory of Ex. \ref{FiniteDimensionalCriticalSmoothSet}. Indeed, consider the case where $N=\FR$ and 
$S:N\rightarrow \FR$ is given by $S(x)=\frac{x^3}{3}$, so that $\dd S : \FR \rightarrow T^*\FR= \FR^2$ is given by 
$\dd S(x)= \big(x, x^2)$. Then the smooth critical locus of $S$ is a single point $\mathrm{Crit}(S)=\{0\}\hookrightarrow \FR $,
but the points of tangent bundle of $\mathrm{Crit}(S)$ in the above sense are given by $T_0 \FR \hookrightarrow T\FR$.
Intuitively, this encodes the tangency of the parabola graph and the horizontal $x$-axis in $\FR^2$. 

\noindent {\bf (iv)} Generally in the field-theoretic context, the question of whether all on-shell tangent vectors 
are represented by on-shell paths of fields depends highly on the explicit form of the Lagrangian density $L$, and is intricately related to the fact that the prolongated shell $S^\infty_L$ might not be a manifold (Rem. \ref{ManifoldStructureOnTheShell}). An explicit field theoretic example of this fact may be found in \cite{Blohmann23b}.
\footnote{Therein a different definition of tangent bundle is employed, 
inspired by constructions internal to diffeological spaces, which only includes on-shell tangent vectors which
are represented by line plots.}

\noindent {\bf (v)} This conundrum of intuition is bypassed in the infinitesimally thickened setting where tangent vectors are given, 
by definition, as \textit{actual} infinitesimal line plots \cite{GSS-2} whereby the above pullback construction 
will arise as a proposition.
\end{remark}

The definition of differential forms on $T \CF_{\CE \CL}$ and $T \CF_{\CE \CL}\times M$ now follows as in Def. \ref{DifferentialFormsOnProduct}, i.e.,
as fiberwise linear antisymmetric maps into $\FR \in \SmoothSets$. The canonical subspace embedding
\vspace{-2mm}
$$
\iota_{\CE \CL} \,:\, \CF_{\CE \CL} \times M \longhookrightarrow \CF \times M
$$

\vspace{-2mm}
\noindent 
induces a pushforward embedding map $T\CF_{\CE \CL} \times TM\hookrightarrow T\CF \times TM$ of tangent bundles, and hence a pullback/restriction of differential
forms. Here, we are only interested in the image of the local bicomplex under the restriction
\begin{align*}
(\iota_{\CE \CL})^* \,:\, \Omega^{\bullet,\bullet}_\loc(\CF \times M)&\longrightarrow \Omega^\bullet(\CF_{\CE \CL}\times M) \\ 
(ev^\infty)^*\om &\longmapsto (ev^\infty)^*\om |_{\CE \CL}\, ,
\end{align*} which defines the \textit{on-shell local forms} on $\CF_{\CE \CL}\times M$. In other words, on-shell local forms are in 1-1 correspondence 
with equivalency classes of local forms that agree when restricted to the on-shell subspace. Practically, this is quite straightforward: 
A local 1-form $(\ev^\infty)^*\om: T\CF \times TM \rightarrow \FR $ 
induces a 1-form on $\CF_{\CE \CL} \times M$ simply by `restricting its domain'
\begin{align*}
    (\ev^\infty)^*\om |_{\CE \CL} \;:\; T \CF_{\CE \CL} \times TM \longhookrightarrow T\CF \times TM \longrightarrow \FR  \, .
\end{align*}
In the local coordinate description of \eqref{LocalpqFormAbusively}, one use the same notation
\vspace{-2mm}
\begin{align*}
(\ev^\infty)^* \om |_{\CE \CL} =  \sum_{I_1,\cdots, I_p=0} \om_{\mu_1\cdots \mu_p a_1 \cdots a_q}^{I_1\dots I_q}\big(x, \{\partial_J \phi^b \}_{|J|\leq k}\big)\cdot  \dd_M x^{\mu_1}\wedge \cdots
\wedge \dd_M x^{\mu_p}\wedge \delta(\partial_{I_1} \phi^{a_1}) \wedge\cdots \wedge \delta(\partial_{I_q} \phi^{a_q}) \, , 
\end{align*}

\vspace{-2mm}
\noindent with the implicit understanding that the right-hand side is to be evaluated \textit{only} on on-shell fields and their tangent vectors. 
Along the same lines, the \textit{on-shell local vector fields} $\CX_\loc(\CF_{\CE \CL})\hookrightarrow \CX (\CF_{\CE \CL}\times M)$ are 
(equivalency classes of) restrictions of local vector 
fields $\CF \rightarrow T\CF $ which factor through $T \CF_{\CE \CL}$ (and hence satisfy \eqref{JacobiEquation} at each $\phi\in \CF_{\CE \CL}(*)$).

\begin{remark}[\bf On-shell Cartan calculus caveats]\label{OnShellCartanCalculusCaveats}
It is not obvious that the local Cartan calculus of $\Omega^{\bullet,\bullet}_\loc(\CF\times M)$ descends to a Cartan calculus of local forms and 
vector fields on $\CF_{\CE \CL}\times M$, although often implicitly assumed in the literature. \footnote{This is often by stating the (often unjustified)
assumption that $\CF_{\CE \CL}$ itself is \textit{some kind} of an infinite-dimensional manifold.}

\noindent {\bf (i)} In more detail, the vertical differential 
$\delta$ does not, in general, descend. 
For instance, consider the case where two local $(p,0)$-forms $\CP, \tilde{\CP}$ define the same on-shell local $(p,0)$-form 
$\CP|_{\CE \CL}=\tilde{\CP}|_{\CE \CL} \in \Omega^{p,0}(\CF_{\CE \CL} \times M)$, i.e., such that 
$\CP,\tilde{\CP}: \CF \rightarrow \Omega^{p}_\mathrm{Vert}(M)$ coincide along $\CF_{\CE \CL}$ but not necessarily on the complement. 
In that case, the restrictions of their vertical differentials $\delta \CP, \delta \tilde{\CP}$ agree if and only if the difference 
of currents is (locally) of the form
\vspace{-2mm}
\begin{align}\label{CurrentsDifferByEL}
(\CP-\tilde{\CP})(\phi)= \CJ^I(\phi) \cdot \partial_I \CE \CL(\phi) \, ,
\end{align}

\vspace{-2mm}
\noindent 
(see proof of Prop. \ref{Gaugesymmetryvspresymplecticcurrent}).  The same statements hold for the Lie bracket of local vector fields 
descending to the set of on-shell local vector fields. 

\noindent {\bf (ii)}   Technically, at the level of the jet bundle, this is the case whenever the ideal of
$C^\infty(J^\infty_M F)$ generated by (the components of) the prolongated Euler--Lagrange bundle map 
$\pr EL:J^\infty_M F \rightarrow J^\infty_M(\wedge^d T^*M\otimes V^*F)$ is \textit{point-determined} (see e.g. \cite{MoerdijkReyes}). 
Of course, the analogous statements can even fail in the finite-dimensional field-theoretic 
setting of Ex. \ref{FiniteDimensionalCriticalSmoothSet}, but it is true if, for instance, the corresponding action is regular enough 
such that the on-shell critical set is an embedded submanifold (see e.g.  \cite[Thm. 1.1]{HenneauxTeitelboim92}).

\noindent {\bf (iii)}  Along the same lines, the field-theoretic differential $\delta$ descends to $\CF_{\CE \CL}\times M$, for instance,
in the cases where the prolongated shell $S_L^\infty \hookrightarrow J^\infty_M F$ is an embedded Fr\'{e}chet submanifold (Rem. \ref{ManifoldStructureOnTheShell}).
This is the case when the diagram of finite order prolongated shells $S^q_{L,k}\rightarrow S^{q-1}_{L,k}$ from Rem. \ref{ManifoldStructureOnTheShell} 
consists of smooth (fin. dim.) manifolds and fibrations  \cite{Tsu82}\cite{GP17}. In other words, when $S^\infty_L$ is
a locally-pro submanifold of the infinite jet bundle (Def. \ref{JetBundleLocPro}). PDEs with this property are called `formally integrable', 
which is often the case for those arising in field theory.

\noindent {\bf (iv)} 
The general question of whether smooth functions, and hence also forms, on $J^\infty_M F$ which vanish on the prolongated shell $S^\infty$ of 
a differential operator $P$ are proportional to its prolongation $\pr P$ has, of course, been explored 
in the PDEs literature. Indeed, such conditions may be viewed as the defining ingredients of `\textit{diffieties}' inside a jet bundle, 
\footnote{Traditionally, this school employs the pro-manifold picture -- but many of the constructions and claims should 
be applicable to the locally pro-manifold picture too (as done e.g. in \cite{GP17}).}
with the restriction of the variational bicomplex on $S^\infty$ \cite{Tsu82} being closely related to the `\textit{Secondary Calculus}' 
of the diffiety \cite{Vin81}\cite{Vin84} \cite{Vin13}\cite{Vitag}.

\noindent {\bf (v)}  We do not need these analytical details for our developments, but we note that such explicit sufficient conditions in 
terms of the Euler--Lagrange source form may be found, for instance, in \cite{KV11}. Similar conditions in the field-theoretic setting may 
be found in \cite[\S 12]{HenneauxTeitelboim92}. \footnote{The jet bundle technology is not explicitly used therein, but the conditions
are in fact implicitly taking place in $J^\infty_M F$.}  
Many of the results that follow do not require the use
an on-shell Cartan calculus, as the manipulations take place on $\CF \times M$, with the on-shell restrictions being applied 
a posteriori. We will refer to this remark, in the instances where the extra assumptions are needed. Nevertheless, 
as far as we are aware, sufficient regularity conditions are satisfied in all fundamental Lagrangians of interest \cite[\S 12]{HenneauxTeitelboim92}.
\end{remark}

\begin{lemma}[{\bf Local vector field is on-shell iff preserves EL-operator}]
\label{LocalVectorFieldIsOnShellIffPreservesELoperator}
A local vector field $\CZ\in \CX_\loc(\CF)$ defines an on-shell vector field if and only if it preserves the Euler--Lagrange source 
form $\CE \CL\in \Omega^{d,1}(\CF\times M)$ along $\CF_{\CE\CL}$. That is, for any on-shell field $\phi \in \CF_{\CE \CL}$ 
\vspace{-1mm}
$$ 
\CZ(\phi) \in T_\phi \CF_{\CE \CL}  \hspace{0.5cm} \iff \hspace{0.5cm} \mathbb{L}_\CZ (\CE \CL) \big{\vert}_{\{\phi\}\times M} = 0 \, .
$$
\end{lemma}
\begin{proof}
Using the local Cartan calculus and standard abuse of notation $\CE \CL = \CE\CL_a \cdot \delta \phi^a $, we compute
\vspace{-2mm}
\begin{align*}
\mathbb{L}_\CZ (\CE \CL)&= (\iota_\CZ \delta +  \delta \iota_\CZ) \CE \CL \\ 
&= \iota_\CZ \big( \delta \CE \CL_a \wedge \delta \phi^a \big) + \delta ( \CE \CL_a \cdot \CZ^a) \\
&= \iota_\CZ \delta \CE \CL_a \cdot \delta \phi^a - \delta \CE \CL_a \cdot \CZ^a + \delta \CE \CL_a \cdot \CZ^a - \CE \CL_a \cdot \delta \CZ^a  
\\
&= \iota_\CZ \delta \CE \CL_a \cdot \delta \phi^a  - \CE \CL_a \cdot \delta \CZ^a \\
&= \iota_\CZ \bigg( \sum_{|I|=0}^{\infty} \frac{\delta \CE \CL_a}{\delta(\partial_I \phi^b)}  \cdot  \delta(\partial_I \phi^b) \bigg) 
\wedge \delta \phi^a  + \CE \CL_a \cdot \delta \CZ^a \\[-2pt]
&= \sum_{|I|=0}^{\infty} \bigg(\frac{\delta \CE \CL_a}{\delta(\partial_I \phi^b)}  \cdot  \frac{\partial \CZ^b}{\partial x^I} \bigg) 
\wedge \delta \phi^a + \CE \CL_a \cdot \delta \CZ^a  \, . 
\end{align*}

\vspace{-2mm}
\noindent At any on-shell field $\phi \in \CF_{\CE \CL}$, the second term vanishes by definition. The result then follows, since the
vanishing of the remaining term is exactly the tangency condition, i.e., the Jacobi equation of \eqref{JacobiEquation}.
\end{proof}
This means $\mathbb{L}_\CZ(\CE \CL)|_{\{\phi\}\times M}$ actually vanishes at any $\phi \in \CF_{\CE \CL}$ as a map of tangent vectors on $\CF\times M$, and not only those tangent to $\CF_{\CE \CL}\times M$. 

This result allows us to prove the infinitesimal version of Prop. \ref{LocalSymmetryPreservesOnshellSpace} and Prop.
\ref{SymmetryPreservesOnshellSpace}, bearing in mind also Prop. \ref{InfinitesimalSpacetimeCovariantSymmetries}. 
That is, since any finite local or spacetime covariant symmetry $\CP\in \mathrm{Diff}_\loc^\CL(\CF)$ of a Lagrangian field theory 
$(\CF,\CL)$ preserves the on-shell field space $\CF_{\CE \CL}$, then any infinitesimal symmetry $\CZ\in \CX_\loc^\CL(\CF)$ should
be \textit{tangent} to $\CF_{\CE \CL}$. In other words, an infinitesimal symmetry of $\CL$ should define an on-shell vector field, 
which is indeed the case.

\begin{proposition}[\bf Infinitesimal symmetry is tangent to $\CF_{\CE \CL}$]\label{InfinitesimalSymmetryTangentToOnshellFieldSpace}
Let $\CZ\in \CX_\loc(\CF)$ be an infinitesimal symmetry of a Lagrangian $\CL\in \Omega^{d,0}_\loc(\CF\times M)$, and so
$\mathbb{L}_\CZ \CL = \dd_M \CK_\CZ$. 
Then $\CZ$ defines an on-shell vector field $\CZ|_{\CE \CL} \in \CX_\loc(\CF_{\CE \CL})$, i.e., 
\vspace{-2mm}
$$
\CZ(\phi) \quad \in \quad T_\phi \CF_{\CE \CL}
$$

\vspace{-2mm}
\noindent
for each $\phi \in \CF_{\CE \CL}$.
\end{proposition}

\begin{proof}
Let $\CP_\CZ = \CK_\CZ + \iota_\CZ \theta_\CL $ be the on-shell conserved current of the infinitesimal local 
symmetry $\CZ$ (Prop. \ref{Noether1st}), i.e., $\iota_\CZ \CE \CL = \dd_M \CP_\CZ $. 
Computing the Lie derivative of $\CE \CL$ along $\CZ$,
\vspace{-2mm}
\begin{align*}
\mathbb{L}_\CZ \CE \CL &= (\iota_\CZ \delta + \delta \iota_\CZ) \CE \CL \\
&= \iota_ \CZ \delta ( \delta \CL - \dd_M \theta_\CL) + \delta \dd_M \CP_\CZ \\
&= 0 - \dd_M (\iota_\CZ \delta \theta_\CL) - \dd_M ( \delta \CK_\CZ + \delta \iota_\CZ \theta_\CL ) \\
&=-\dd_M ( \mathbb{L}_\CZ \theta_\CL + \delta \CK_\CZ )
\end{align*}

\vspace{-2mm}
\noindent where we used the bi-complex commutation relations and the local Cartan calculus of Lem. \ref{LocalCartanCalculus} repeatedly. 
Next, it follows that for \textit{any} local vector field $\CZ$ (not necessarily a symmetry), by  \cite[Cor. 3.22]{Anderson89}
or even directly in local coordinates, that 
\vspace{-2mm}
$$
\mathbb{L}_\CZ ( \CE \CL) = \CE (\mathbb{L}_\CZ \CL) \, .
$$

\vspace{-2mm}
\noindent
This is the infinitesimal version of the covariance property of the Euler--Lagrange operator from Prop. \ref{LocalSymmetryPreservesOnshellSpace}. 
In other words, the Lie derivative $\mathbb{L}_\CZ (\CE \CL)$ may be equivalently seen as the Euler--Lagrange operator / source form of the
local Lagrangian $\mathbb{L}_\CZ \CL = (\mathbb{L}_{\pr Z} L) \circ j^\infty $.

Combining the two equations for the case where $\CZ$ is a local symmetry, we have 
\vspace{-2mm}
$$
\CE (\mathbb{L}_\CZ \CL) = - \dd_M (  \mathbb{L}_\CZ \theta_\CL + \delta \CK_\CZ ) \, ,
$$

\vspace{-2mm}
\noindent 
which at the level of the infinite jet bundle says that the source form $E(\mathbb{L}_{\pr Z} L ) \in \Omega^{d-1,1}_S (J^\infty_M F)$ is $\dd_H$-exact.
By Prop. \ref{InteriorEulerProperties}, this is necessarily the zero source form, and hence
\vspace{-2mm}
\begin{align*}
\mathbb{L}_\CZ (\CE \CL) = \CE (\mathbb{L}_\CZ \CL ) = 0 \hspace{0.5cm} \in \hspace{0.5cm} \Omega^{d-1,0}_\loc(\CF\times M) \, . 
\end{align*}

\vspace{-2mm}
\noindent 
This equation holds as forms on $\CF \times M$ and hence, in particular, restricts to equations of forms on $\CF_{\CE \CL}\times M$. 
This completes the proof by the result of Lem. \ref{LocalVectorFieldIsOnShellIffPreservesELoperator}.
\end{proof}

This statement appears in Theorem 13 (a) of \cite{Zuckerman} without proof, and is implicitly used throughout \cite{DF99}. 
A more involved proof appears also in \cite{Blohmann23b}. \footnote{It seems to us the general covariance fact 
$\CE (\mathbb{L}_\CZ \CL)  = \mathbb{L}_\CZ (\CE \CL)$ is not noted therein.} 
 Note that the equation 
\vspace{-1mm}
\begin{align}\label{InfinitesimalSymmPreservesELequation}
\mathbb{L}_\CZ (\CE \CL) = 0 \, ,
\end{align}

\vspace{-1mm}
\noindent appearing in the proof may be interpreted as the statement that (infinitesimal) symmetries of an action/Lagrangian 
are also symmetries of the induced Euler--Lagrange equations.
A calculation along the lines of Lem. \ref{LocalVectorFieldIsOnShellIffPreservesELoperator} shows that the vertical differential 
of the Euler--Lagrange $(d-1,1)$-form vanishes on-shell.

\begin{lemma}[{\bf $\CE \CL$ is $\delta$-closed on-shell}]\label{EulerLagrangeFormIsDeltaClosedOnshell}
The Euler--Lagrange local $(d,1)$-form $\CE \CL\in \Omega^{d,1}_\loc(\CF\times M)$ satisfies
\vspace{-2mm} 
$$
\delta \CE \CL |_{\CE \CL} = 0 \quad \in \quad \Omega^{d,2}_\loc(\CF_{\CE \CL}\times M)\, . 
$$
\end{lemma}
\begin{proof}
For any two on-shell tangent vectors $\CZ^{1,2}_\phi= (\CZ^{1,2}_\phi)^a \cdot \frac{\delta}{\delta \phi^a}\,  
\in\, T_\phi \CF_{\CE \CL} \hookrightarrow T_\phi \CF$, we have 
\vspace{-2mm}
\begin{align*}
\delta \CE \CL |_{\CE \CL} (\CZ^1_\phi, \CZ^2_\phi) &=  \sum_{|I|=0}^{\infty} \frac{\delta \CE \CL_a(\phi)}{\delta(\partial_I \phi^b)} 
\cdot  \delta(\partial_I \phi^b) \wedge \delta \phi^a\Big\vert_{(\phi,p )}  \big(\CZ^1_\phi, \CZ^2_\phi) \\[-1pt]
& =  \sum_{|I|=0}^{\infty} \bigg(\frac{\delta \CE \CL_a(\phi)}{\delta(\partial_I \phi^b)}  \cdot  \frac{\partial (\CZ^1_\phi)^b}{\partial x^I} 
\cdot  (\CZ^1_\phi)^a\bigg) (p) \, - \, (1\leftrightarrow 2) \\[-2pt]
&= 0
\end{align*}

\vspace{-3mm}
\noindent where we used the standard abuse of notation and the fact that each of the terms vanishes since they are proportional 
to the Jacobi equation \eqref{JacobiEquation} for each \textit{on-shell} tangent vector, respectively. The case of general 
tangent vectors $\CZ_\phi + X_p \in T_\phi \CF_{\CE \CL} \times T_p M $ follows similarly.
\end{proof}
This implies, in particular, that 
\vspace{0mm}
$$
\delta \CE \CL |_{\CE \CL} ( \CZ^1,\CZ^2) = 0 \quad \in \quad \Omega^{2,d}_\loc(\CF_{\CE \CL}\times M)  
$$

\vspace{-1mm}
\noindent
for any two on-shell vector fields $\CZ_1, \CZ_2 \in \CX_\loc(\CF_{\CE \CL})$, as can be also checked directly using the induced 
local Cartan calculus. Furthermore, Lem. \ref{EulerLagrangeFormIsDeltaClosedOnshell} immediately implies that the presymplectic
current is \textit{conserved} on-shell.

\begin{corollary}[\bf  Presymplectic current is on-shell conserved]\label{PresymplecticCurrentIsOnshellConserved}
The presymplectic current $\om_\CL = \delta \theta_\CL \in \Omega^{d-1,2}_\loc(\CF\times M)$ satisfies
\vspace{-2mm}
$$
\dd_M \om_\CL |_{\CE \CL} = 0 \quad \in \quad \Omega^{d,2}_\loc(\CF_{\CE \CL}\times M)\, .
$$
In other words, it is an \textit{on-shell conserved current}.
\end{corollary}
\begin{proof}
Applying the vertical differential on Eq. \eqref{CohomologicalIntegrationByPartsOnFieldSpace}, we have
\vspace{-1mm}
$$
0=\delta^2 \CL = \delta \CE \CL - \dd_M \delta \theta_\CL \, .
$$

\vspace{-1mm}
\noindent
That is, 
$
\dd_M \om_\CL = \delta \CE \CL \in  \Omega^{d,2}_\loc(\CF\times M) 
$, 
so that the result follows by restricting to $\CF_{\CE \CL}\times M$ and using Lem. \ref{EulerLagrangeFormIsDeltaClosedOnshell}.
\end{proof}

By definition, infinitesimal local symmetries preserve the Lagrangian $\CL$, up to a trivial Lagrangian, and we have shown they preserve 
the Euler--Lagrange operator $\CE \CL$ (Prop. \ref{InfinitesimalSymmetryTangentToOnshellFieldSpace}). It should not come as a
surprise that they further preserve the induced presymplectic current, up to a `trivial' $\dd_M$-exact $(d-1,2)$ current.

\begin{lemma}[{\bf Infinitesimal symmetries preserve presymplectic current}]\label{InfinitesimalSymmetryPreservesPresymplecticCurrent}
Let $\CZ\in \CX_\loc(\CF)$ be an infinitesimal symmetry of $\CL\in \Omega^{d,0}_\loc(\CF\times M)$. Then $\CZ$ preserves the 
presymplectic current $\om_\CL=\delta \theta_\CL$
\vspace{-2mm}
$$
\mathbb{L}_\CZ \om_\CL = \dd_M \CB_\CZ \hspace{0.5cm} \in \hspace{0.5cm} \Omega^{d-1,2}_\loc(\CF\times M) \, , 
$$

\vspace{-2mm}
\noindent
up to an exact local $(d-1,2)$-form. 
\end{lemma}
\begin{proof}
The local Cartan calculus gives
\vspace{-2mm}
\begin{align*}
\mathbb{L}_\CZ \om_\CL = \mathbb{L}_\CZ \delta \theta_\CL = \delta \mathbb{L}_\CZ \theta_\CL\, ,
\end{align*}

\vspace{-1mm}
\noindent while by the proof of Prop. \ref{InfinitesimalSymmetryTangentToOnshellFieldSpace}, $
\dd_M (\mathbb{L}_\CZ \theta_\CZ) = + \dd_M (\delta \CK_\CZ) $. 
Thus, applying $\dd_M$ to the first equation we get 
\vspace{-1mm}
\begin{align*}
 \dd_M \mathbb{L}_\CZ \om_\CL &= \dd_M \delta (\mathbb{L}_\CZ \theta_\CL )= - \delta \dd_M (\mathbb{L}_\CZ \theta_\CL) \\
 &=-\delta \dd_M \delta \CK_\CZ = + \delta^2 \dd_M \CK_\CZ \\
 &=0 \, .
\end{align*}

\vspace{-2mm}
\noindent 
At the level of the $J^\infty_M F$, $\mathbb{L}_{\pr Z} \om_L $ is a $\dd_H$-closed $(d-1,2)$ form, and so by Prop. \ref{AcyclicityTheorem} it 
is in fact exact $\mathbb{L}_{\pr Z} \om_L = \dd_H B_Z$. Pulling back to $\CF \times M$ via the prolongated evaluation map completes the proof.
\end{proof}
 
Borrowing intuition from finite-dimensional symplectic geometry, this result says that an infinitesimal local symmetry is a
`\textit{symplectic vector field}' with respect to $\om_\CL=\delta \theta_\CL$, up to a horizontally exact local form. Of course, 
as one might expect, this result is not restricted to infinitesimal symmetries. Indeed, any finite symmetry 
(Def. \ref{FiniteSymmetryofLagrangianFieldTheory}) of a Lagrangian field theory preserves the induced presymplectic current. 
We include a brief proof of the local case, which parallels the above, for completeness.

\begin{lemma}
[\bf Finite local symmetries preserve presymplectic current]\label{FiniteSymmetryPreservesPresymplecticCurrent}
Let $\CD=D\circ j^\infty :\CF\rightarrow \CF$ be a local symmetry of $(\CF,\CL)$ as in Def. \ref{FiniteSymmetryofLagrangianFieldTheory}. 
Then $\CD$ preserves the presymplectic current $\om_\CL=\delta \theta_\CL$
\vspace{-2mm}
$$
\CD^* \om_\CL  = \dd_M \CB_\CD \quad \in \quad \Omega^{d-1,2}_\loc(\CF\times M) \, , 
$$

\vspace{-2mm}
\noindent
up to an exact local $(d-1,2)$-form. The pullback here means $\CD^*\om := (\ev^\infty)^*\big(\pr D^*\om)$ where
$\pr D:J^\infty_M F\rightarrow J^\infty_M F$ is the prolongated bundle map (Def. \ref{ProlongationOfJetBundleMap}).
\footnote{This can be equivalently defined via the pushforward $\dd \CD: T\CF \rightarrow T\CF$ along $\CD$.}
\end{lemma}
\begin{proof}
 We prove this for the case of a local finite symmetry, leaving the spacetime covariant one for the interested reader. 
 The proof is analogous to the infinitesimal case. Let $\CD= D \circ j^\infty:\CF \rightarrow \CF$ be any local diffeomorphism.  
 Then
\vspace{-2mm}
$$
\CE (\CD^*\CL) = \CD^*(\CE \CL) = \CD^* ( \delta \CL - \dd_M \theta_\CL) 
= \delta  \CD^* \CL -  \dd_M \CD^* \theta_\CL
$$

\vspace{-2mm}
\noindent where the first step is the covariance of the Euler--Lagrange operator from Prop. \ref{SymmetryPreservesOnshellSpace}
and \eqref{CovarianceofEulerLagrangeOperator}. The commutation of both differentials with local self maps $\CF\rightarrow \CF$ 
can be seen in the local coordinate abuse of notation, and appears at the level of the jet bundle in 
 \cite[Thm. 3.15]{Anderson89}. This is simply the integrated analog of the evolutionary/local Cartan calculus relations.

Consider now the case where $\CD$ is actually local symmetry of $\CL$, i.e., 
$\CD^*\CL = \CL + \dd_M \CK_\CD$. Then, we also have
\vspace{-2mm}
$$
\CE (\CD^*\CL)= \CE (\CL + \dd_M \CK_\CD) = \CE \CL \, ,
$$

\vspace{-2mm}
\noindent whereby combining with the equation above, we arrive at
\vspace{-2mm}
$$
\CE \CL= \delta  \CD^* \CL -  \dd_M \CP^* \theta_\CL \, .
$$

\vspace{-2mm}
\noindent Applying $\delta$ on both sides,
\vspace{0mm}
$$
\delta \CE \CL= 0 +  \dd_M \CD^* \delta \theta_\CL = \dd_M \CD^* \om_\CL \, ,
$$

\vspace{-0mm}
\noindent but as in the proof of Cor. \ref{PresymplecticCurrentIsOnshellConserved}, it is also the case that
$\delta \CE \CL = \dd_M \CD^* \om_\CL$. Subtracting the last equations, we get
\vspace{-2mm}
$$
\dd_M (\CD^* \om_\CL - \om_\CL) = 0\, , 
$$

\vspace{-2mm}
\noindent
which at the level of the jet bundle says that $\pr D^* \om_L - \om_L$ is a horizontally closed $(d-1,2)$-form. Once again, 
the result follows as an application of Prop. \ref{AcyclicityTheorem}.
\end{proof}

Following the finite-dimensional heuristics, the result says that a finite symmetry of a local Lagrangian field theory is a
`\textit{symplectomorphism}' with respect to $\om_\CL$, up to a horizontally exact local form. Infinitesimal symmetries 
and their currents are further intricately related to the properties of the presymplectic current.

\begin{lemma}[{\bf Symmetry conserved current vs symplectic current}] \label{ContractedSymplecticFormAndConservedCharge}
Let $\CZ \in \CX_\loc(\CF) $ be an infinitesimal symmetry of $\CL\in \Omega^{d,0}_\loc(\CF\times M)$ with conserved current
$\CP_\CZ \in \Omega^{d-1,0}(\CF\times M)$. Then the local $(d-1,1)$-form $\iota_\CZ \om_\CL + \delta \CP_\CZ$ is horizontally exact, 
that is 
\vspace{-1mm}
$$
\iota_\CZ \om_\CL + \delta \CP_\CZ  = \dd_M  \CT_\CZ  \quad \in \quad \Omega^{d-1,1}_\loc(\CF\times M) \, ,
$$

\vspace{0mm}
\noindent for some $\CT_\CZ \in \Omega^{d-2,1}_\loc(\CF\times M)$. That is, $\iota_\CZ \om_\CL + \delta \CP_\CZ$ is a`trivial'
$(d-1,1)$-form on $\CF\times M$.
\end{lemma}
\begin{proof}
Recall by Noether's First Theorem (Prop. \ref{Noether1st}), the conserved current satisfies $\iota_\CZ \CE \CL = \dd_M \CP_\CZ$. 
Computing the Lie derivative of $\CE \CL$ along $\CZ$ as in the proof of Prop. \ref{InfinitesimalSymmetryTangentToOnshellFieldSpace}, 
we have 
\vspace{-1mm}
\begin{align*}
0=\mathbb{L}_\CZ \CE \CL &= \dd_M (- \iota_\CZ \delta \theta_\CL - \delta \CP_\CZ ) \, , \end{align*}

\vspace{-2mm}
\noindent and hence
\vspace{-2mm}
\begin{align*}
\dd_M (\iota_\CZ \om_\CL + \delta \CP_\CZ ) =0 \, .
\end{align*}

\vspace{-1mm}
\noindent At the level of the $J^\infty_M F$, this says that $\iota_{\pr Z} \om_L +\dd_V P_{\pr Z} $ is a $\dd_H$-closed $(d-1,2)$ form, 
and so by Prop. \ref{AcyclicityTheorem} it is in fact exact $\iota_{\pr Z} \om_L +\dd_V P_{\pr Z} = \dd_H T_\CZ$. Pulling back to 
$\CF \times M$ via the prolongated evaluation map completes the proof. 
\end{proof}

Let us expand on the interpretation of the above result in terms of the presymplectic potential current $\theta_\CL$. Substituting 
$\om_\CL= \delta \theta_\CL$ and $\CP_\CZ= \iota_\CZ \theta_\CL + \CK_\CZ$, then the above reads
\vspace{-2mm}
\begin{align}\label{HamiltonianVectorFieldCondition}
\mathbb{L}_\CZ \theta_\CL + \delta \CK_\CZ = \dd_M \CT_\CZ\, ,
\end{align}

\vspace{-2mm}
\noindent and so $\CZ$ is heuristically a `\textit{Hamiltonian vector field}' for $\omega_\CL$, up to a local horizontally 
exact form (a trivial current). In the same vein and equivalently, the current $\CP_\CZ$ is a `\textit{Hamiltonian form}' for $\om_\CL$. 
From this point of view, Lem. \ref{InfinitesimalSymmetryPreservesPresymplecticCurrent} follows as a consequence by applying $\delta$, 
with $\CB_\CZ = -\delta \CT_\CZ$, i.e., any Hamiltonian vector field is a symplectic vector field. We will come back to this interpretation 
shortly (see discussion around \eqref{HamiltonianCurrentCondition}).

Before that, recall by Cor. \ref{TrivialCurrentsFromGaugeSymmetries} that the current of a gauge symmetry satisfies $\dd_M \CP_e = \dd_M \CJ_e $ 
for some local current $\CJ_e=\CJ(-,\CE \CL(-), e)$ that is linear in the second entry of $\CE \CL$, and in particular vanishes 
on-shell $\CJ_e|_{\CE \CL}=0$,. In that case, repeating the proof above verbatim gives
\vspace{-2mm}
\begin{align}\label{GaugeSymmetryConservedCurrentVsSymplecticCurrent}
\iota_{\CR_e} \om_\CL + \delta  \CJ_e  = \dd_M  \tilde{\CT}_e  \quad \in \quad \Omega^{d-1,1}_\loc(\CF\times M) \, . 
\end{align}

\vspace{-2mm}
\noindent This relation implies the following degenerate behavior for the presymplectic current. 

\begin{proposition}[\bf Gauge symmetry vs presymplectic current]
\label{Gaugesymmetryvspresymplecticcurrent}
Let $\CR_e \in \CX^\CL_\loc(\CF)$ by a gauge symmetry of $\CL$, for any gauge parameter $e\in \Gamma_M E$.
Then the induced presymplectic current satisfies the on-shell relation
\vspace{-2mm} 
\begin{align*}
\iota_{\CR_e} \om_\CL \big{\vert}_{\CE \CL} = \dd_M  \tilde{\CT}_e \big{\vert}_{\CE \CL}
\quad \in \quad \Omega^{d-1,1}_\loc(\CF_{\CE \CL}\times M) \, , 
\end{align*}

\vspace{-2mm} 
\noindent for some local $(d-2,1)$-form $\tilde{\CT}_e$. In particular, for every on-shell tangent vector 
$\CZ_\phi \in T_\phi \CF_{\CE \CL}$,
the evaluation of the presymplectic form is an exact $(d-1)$-form on spacetime
\vspace{-2mm}
\begin{align*}
\om_\CL\big(\CR_e(\phi),\, \CZ_\phi\big)  = \dd_M \big( \tilde{\CT}_e(\CZ_\phi) \big) 
\quad \in \quad \Omega^{d-1}(M) \, .  
\end{align*}
\end{proposition}

\begin{proof}
By Lem. \ref{ContractedSymplecticFormAndConservedCharge} and since we are dealing with a symmetry, we have the  off-shell relation \eqref{GaugeSymmetryConservedCurrentVsSymplecticCurrent}, 
where
$
\CJ_e: \CF \to \Omega^{d-1}_\mathrm{Vert}(M)
$
is a local $(d-1)$-current $\CJ_e= \CJ(-, \CE \CL(-), e)$ linear in the second entry. Hence, it suffices to show that 
the vertical differential $\delta \CJ_e$ vanishes on-shell. We show this locally using the standard abuse of notation. 
By construction (see also the proof of Prop. \ref{Noether2nd}), $\CJ_e$ is of the form
\vspace{-3mm}
\begin{align*}
\CJ_e = \sum_{|I|=0}\CJ^{ I}_e(\phi)\cdot \partial_I \CE \CL(\phi) \, ,
\end{align*}

\vspace{-3mm}
\noindent for some collection of differential operator $\{\CJ^{ I}_e(\phi)\}_{|I|=0}$ of $\phi$ (also local in $e\in \Gamma_M(E)$) 
but whose precise form is immaterial. Applying the vertical differential along $\CF$,
\vspace{-2mm}
\begin{align*}
\delta \CJ_e = \sum_{|I|=0} \delta \CJ^{ I}_e(\phi)\cdot \partial_I \CE \CL(\phi) +
\CJ^{I}_e (\phi) \cdot \delta \big(\partial_I \CE \CL(\phi)\big) \, ,
\end{align*}

\vspace{-2mm}
\noindent where former terms in the sum immediately vanish on-shell, in the same manner as $\CJ_e$ itself.
For the latter remaining terms, working as in Eq. \eqref{Local1FormActionOnTangentAbusively} it is easy
to see that `the partial derivatives commute with the vertical differential', i.e.,
\vspace{-2mm}
$$
\delta\big( \partial_I \CE \CL(\phi) \big) (\CZ^1_\phi, \CZ^2_\phi) = \partial_I \big (\delta \CE \CL (\CZ^1_\phi, \CZ^2_\phi)\big)
\quad \in \quad C^\infty(M) \, ,
$$

\vspace{-1mm}
\noindent for any pair of tangent vectors on the field space. Using Lem. \ref{EulerLagrangeFormIsDeltaClosedOnshell}, 
it follows that such terms vanish when applied to $\textit{on-shell}$ tangent vectors, and so
\vspace{0mm}
$$
\delta(\partial_I \CE \CL) |_{\CE \CL} = \partial_I (\delta   \CE \CL |_{ \CE \CL} ) = 0\, . 
$$

\vspace{-2mm}
\noindent Overall, we deduce that
\vspace{0mm}
$$
\delta \CJ_e |_{\CE \CL} = 0 \hspace{0.5cm} \in \hspace{0.5cm} \Omega^{d-1,1}_\loc (\CF_{\CE \CL}\times M )\, , 
$$

\vspace{-1mm}
\noindent which completes the proof.
\end{proof}

\begin{remark}[\bf A subtlety in the literature]
The original source \cite{Zuckerman} claims the statement in Prop. \ref{Gaugesymmetryvspresymplecticcurrent} 
(after composing with integration) in Theorem 13(b), and gives (a sketch of) a proof that relies on the 
`$\dd_M$-exactness' of the corresponding conserved current of the gauge symmetry (Cor. \ref{TrivialCurrentsFromGaugeSymmetries}).
From our point of view, however, the exactness of the charge depends on the $(d-1)$-cohomology of the field bundle. 
This seems to be an obstacle to the validity of the proof therein. Our proof for the degeneracy of the presymplectic 
current evades this issue. 
\end{remark}

\vspace{-4mm} 
\paragraph{Brackets of Noether and Hamiltonian currents.}  As we have eluded in \eqref{CurrentOfLiebracketOfSymmetries}, there is a natural
`bracket operation' for conserved local $(d-1,0)$-form currents, which maps the currents $\CP_{\CZ_1}, \CP_{\CZ_2}$ of any two local 
infinitesimal symmetries of a local Lagrangian $\CL$ to the current of the corresponding commutator symmetry $\CP_{[\CZ_1,\CZ_2]}$. 
In fact, as we will see, this is closely related to a more general `Poisson-like' bracket on a larger set of `Hamiltonian' local $(d-1,0)$-form currents, 
i.e., those that satisfy the `Hamiltonian' condition of Lem. \ref{ContractedSymplecticFormAndConservedCharge}, which is induced by the 
presymplectic current $\om_\CL$. To that end, we first \textit{fix a choice}\footnote{In principle, one can enhance the discussion below
to keep track of this choice. We conform with the literature and do not explicitly pursue this here.} 
of presymplectic current potential $\theta_\CL$ in the variational decomposition $\delta \CL = \CE \CL + \dd_M \theta_\CL$.

Having fixed the presymplectic potential, then an infinitesimal symmetry $\CZ\in \CX_\loc(\CF) $ such that $\mathbb{L}_\CZ \CL = \dd_M \CK_\CZ$
defines a family of conserved currents $\CP_\CZ = \CK_\CZ + \iota_\CZ \theta_\CL$ parametrized by the freedom of adding a horizontally
closed form to $\CK_\CZ$. We keep track of these families by considering instead pairs of conserved charges and symmetries.

\begin{definition}[\bf  Noether pairs]\label{NoetherPairs}
The set of \textit{Noether pairs} of a local Lagrangian field theory $(\CF,\CL)$, with respect to a presymplectic current potential $\theta_\CL$, 
is defined by 
\vspace{-2mm}
\begin{align*}
\mathrm{NoethPrs}(\CL,\theta_\CL):= \Big\{(\CP_\CZ, \CZ) \in \Omega^{d-1,0}_\loc(\CF\times M)\times \CX_\loc(\CF) 
\, \, \big\vert \, \, \dd_M \CP_\CZ = \iota_\CZ \CE \CL \Big\}\, .
\end{align*}

\vspace{-2mm}
\noindent Equivalently (Prop. \ref{Noether1st}), Noether pairs are in 1-1 correspondence with pairs $(\CK_\CZ,\CZ)$ such 
that $\mathbb{L}_\CZ \CL = \iota_\CZ \delta \CL = \dd_M \CK_\CZ$, via $\CP_\CZ = \CK_\CZ +\iota_\CZ \theta_\CL$.
\end{definition}
Equation \eqref{CurrentOfLiebracketOfSymmetries} defines a well-defined map of Noether pairs, by assigning the corresponding 
current $\CP_{[\CZ_1,\CZ_2]}$ of the commutator symmetry for of two Noether pairs
\vspace{-2mm}
\begin{align*}
\big((\CP_{\CZ_1}, \CZ_1), \, (\CP_{\CZ_2}, \CZ_2)\big) \longmapsto \big(\mathbb{L}_{\CZ_1} \CK_{\CZ_2}
- \mathbb{L}_{\CZ_2} \CK_{\CZ_1} + \iota_{[\CZ_1,\CZ_2]} \theta_\CL, \,  [\CZ_1, \CZ_2] \big)\, .
\end{align*}

\vspace{-2mm}
\noindent Using $\CP_{\CZ_{1,2}}= \CK_{\CZ_{1,2}} - \iota_{\CZ_{1,2}}\theta_\CL$ and the local Cartan calculus repeatedly,
the first term on the right-hand side may be recast as
\vspace{-2mm}
\begin{align*}
\mathbb{L}_{\CZ_1} \CK_{\CZ_2} - \mathbb{L}_{\CZ_2} \CK_{\CZ_1} + \iota_{[\CZ_1,\CZ_2]} \theta_\CL 
&= \mathbb{L}_{\CZ_1} \CP_{\CZ_2} - \iota_{\CZ_2} ( \mathbb{L}_{\CZ_1} \theta_\CL +\delta \CK_{\CZ_1} ) \\
&= \tfrac{1}{2} \big( \mathbb{L}_{\CZ_1} \CP_{\CZ_2} - \mathbb{L}_{\CZ_2} \CP_{\CZ_1}\big) - \tfrac{1}{2}\iota_{\CZ_2} \big( \mathbb{L}_{\CZ_1} \theta_\CL 
 +  \delta \CK_{\CZ_1} \big) + \tfrac{1}{2} \iota_{\CZ_1} \big(\mathbb{L}_{\CZ_2} \theta_\CL +\delta \CK_{\CZ_2}  \big)\, . 
\end{align*}

\vspace{-2mm}
\noindent Recall that by \eqref{HamiltonianVectorFieldCondition},  $\mathbb{L}_{\CZ_{1,2}} \theta_\CL + \delta \CK_{\CZ_{1,2}} = \dd_M \CT_{\CZ_{1,2}}$ 
for some local current $\CT_{\CZ_{1,2}}$, and so the terms apart from the Lie derivatives of original currents $\CP_{\CZ_1}, \CP_{\CZ_2}$ 
are in fact $\dd_M$-exact. Thus, the above may be interpreted as saying that the current $\CP_{[\CZ_1,\CZ_2]}$ of the commutator 
of two symmetries is given via the Lie derivatives along the currents, up to a $\dd_M$-exact term -- i.e., a trivial local $(d-1,0)$-form current.
In other words, to compute the corresponding charge $\CP_{[\CZ_1,\CZ_2]}^{\Sigma^{d-1}}$ over a compact, boundaryless submanifold
$\Sigma^{d-1}\hookrightarrow M$, it is sufficient to consider only $\mathbb{L}_{\CZ_1} \CP_{\CZ_2}$ (or its antisymmetrization). 

\begin{definition}[\bf  Bracket of Noether pairs]\label{BracketOfNoetherPairs}
 The bracket of Noether pairs 
 \vspace{-1mm}
 $$
 \{-,-\}_N: \mathrm{NoethPrs}(\CL,\theta_\CL) \times \mathrm{NoethPrs}(\CL,\theta_\CL)\longrightarrow \mathrm{NoethPrs}(\CL,\theta_\CL)
 $$
 
 \vspace{-2mm}\noindent is defined by assigning the pair of commutator of symmetries and the corresponding conserved current 
 \vspace{-1mm}
\begin{align*}
\big\{(\CP_{\CZ_1}, \CZ_1), \, (\CP_{\CZ_2}, \CZ_2)\big\}_N :&=  \big( \mathbb{L}_{\CZ_1} \CP_{\CZ_2} 
+ \dd_M \iota_{\CZ_2} \CT_{\CZ_1} , \, [\CZ_1,\CZ_2] \big) \\
&= \big( \tfrac{1}{2} ( \mathbb{L}_{\CZ_1} \CP_{\CZ_2} - \mathbb{L}_{\CZ_2} \CP_{\CZ_1}) 
+ \tfrac{1}{2} \dd_M ( \iota_{\CZ_{2}} T_{\CZ_{1}}-  \iota_{\CZ_{1}} T_{\CZ_{2}}), \, [\CZ_1,\CZ_2] \big)\, ,
\end{align*}

\vspace{-1mm}
\noindent 
where, by \eqref{HamiltonianVectorFieldCondition}, the latter terms $\dd_M$-exact terms are given by
$\dd_M \iota_{\CZ_{2,1}} T_{\CZ_{1,2}}=-\iota_{\CZ_{2,1}}(\mathbb{L}_{\CZ_{1,2}} \theta_\CL + \delta \CK_{\CZ_{1,2}})$.
\end{definition}
The bracket of Noether pairs is manifestly $\FR$-linear and antisymmetric. However, a tedious but straightforward calculation
shows that the appearance of the $\dd_M$-exact currents spoils the Jacobi identity, in turn by $\dd_M$-exact currents.
That is, the Jacobiator evaluates to
\vspace{-2mm}
\begin{align}\label{NoetherBracketJacobiator}
\mathrm{Jac}_{\{-,-\}_N}\big((\CP_{\CZ_1}, \CZ_1), \, (\CP_{\CZ_2}, \CZ_2), \, (\CP_{\CZ_3}, \CZ_3) \big)
:&= \Big\{ \big\{(\CP_{\CZ_1}, \CZ_1), \, (\CP_{\CZ_2}, \CZ_2)\}_N\, , \,(\CP_{\CZ_3}, \CZ_3) \Big\}_N  + \mathrm{cyc}(1,2,3) \nn \\
&= \big( \dd_M (\cdots), \, 0 \big) \, , 
\end{align}

\vspace{-2mm}
\noindent for some horizontally exact (trivial) current, whose precise form is not important (see \eqref{HamiltonianBracketJacobiator}).
It follows that the Noether bracket does not define a Lie algebra on Noether pairs, but instead only on the \textit{quotient} of 
Noether pairs by `trivial', i.e., horizontally exact $(d-1,0)$-form currents.

\begin{lemma}[{\bf Dickey Lie algebra of Noether currents}]\label{DickeyLieAlgebraOfNoetherCurrents}
$\,$

\noindent {\bf (i)} The bracket of Noether pairs descends to a Lie bracket on the quotient
\vspace{-2mm}
$$
\mathrm{NoethPrs}(\CL,\theta_\CL) / \dd_M \Omega^{d-2,0}_\loc(\CF\times M)\, ,  
$$

\vspace{-2mm}
\noindent known as the Dickey Lie bracket of (on-shell) conserved currents (see \cite{Dickey91}).

\noindent {\bf (ii)} Furthermore, the Dickey 
Lie algebra is part of a short exact sequence of Lie algebras 
\vspace{-1mm}
\begin{align}\label{DickeyLieAlgebraCentralExtension}
H_{\loc,\dd_M}^{d-1,0}(\CF\times M) \longhookrightarrow 
\mathrm{NoethPrs}(\CL,\theta_\CL) / \dd_M \Omega^{d-2,0}_\loc(\CF\times M) \longrightarrow \CX^\CL_\loc(\CF) \, ,
\end{align}

\vspace{-2mm}
\noindent where the former horizontal $(d-1,0)$-cohomology, i.e., the off-shell conserved currents modulo trivial currents, being thought 
of as a trivial Lie algebra. In other words,  the Dickey Lie algebra of (on-shell) conserved currents is a central extension of the
Lie algebra of infinitesimal local symmetries of the Lagrangian field theory.
\begin{proof}
The bracket of Noether pairs \eqref{BracketOfNoetherPairs} descends to a Lie bracket since the Jacobiator \eqref{NoetherBracketJacobiator} 
vanishes in cohomology. The short exact sequence follows immediately: the first map is the canonical inclusion 
$[\tilde{\CP}]\mapsto ([\tilde{\CP}],\, 0)$ and the second is the canonical projection $([\tilde{\CP_\CZ}],\, \CZ)\mapsto \CZ $ 
onto symmetries of $\CL$. The kernel of the former coincides with the image of the latter, by definition. Note that
this subset represents exactly the remaining freedom of adding a $\dd_M$-closed $(d-1,0)$-form to each current class $[\CP_{\CZ}]$ 
corresponding to a fixed symmetry $\CZ$. 
\end{proof}
\end{lemma}
In principle, a lot of information is discarded if works on the quotient by `trivial' currents. From a physical point of view, 
information of charges over submanifolds with boundary is ignored. Certain homotopical Lie algebraic structures naturally encode the breaking 
of the Jacobi identity, in a controlled homotopical manner, which reduce to the above quotient description in an appropriate truncation
(Rem. \ref{HigherPoissonAlgebrasOfLocalObservables}), but otherwise retain the full information of the bracket of currents. 

\medskip 
There is a more general bracket defined on a larger class of $(d-1,0)$-local currents, which are not necessarily on-shell conserved. 
Physically, upon integration over a codimension 1 submanifold, these correspond to observables that are not necessarily conserved 
over `time evolution'. These are analogs of Hamiltonian functions from finite-dimensional (pre)symplectic geometry. 
Explicitly, these are local $(d-1,0)$-form currents $\CH_\CZ\in \Omega^{d-1,0}_\loc(\CF\times M)$ satisfying the condition  
\vspace{-2mm}
\begin{align}\label{HamiltonianCurrentCondition}
\iota_\CZ \om_\CL + \delta \CH_\CZ  = \dd_M  \CT_\CZ  
\end{align}

\vspace{-2mm}
\noindent
for \textit{some} local vector field $\CZ\in \CX_\loc(\CF)$ and local form $\CT_\CZ \in \Omega^{d-2,1}_\loc(\CF\times M)$. Local currents 
satisfying this condition are known as \textit{Hamiltonian currents}\footnote{The right-hand side is often demanded to be identically zero. 
However, to include all conserved currents and more general Hamiltonian observables, it is natural to relax the condition up to 
a horizontally exact local form.} of $\om_\CL$, while the corresponding local vector fields $\CZ$ as \textit{Hamiltonian vector fields} 
of $\om_\CL$.

\medskip 
Recall that this is exactly the relation appearing in Lem. \ref{ContractedSymplecticFormAndConservedCharge}, which is now 
interpreted as the statement that all Noether conserved currents of a field theory $(\CF,\CL)$ are in fact Hamiltonian for its
presymplectic current $\om_\CL=\delta \theta_\CL$. However, generally speaking, a Hamiltonian current need not be conserved on-shell, 
i.e., generically $\dd_M \CH_\CZ |_{\CE \CL} \neq 0$, and similarly a Hamiltonian vector field need not be a symmetry of $\CL$. 
Nevertheless, it is still true that Hamiltonian vector fields are a symmetry of $\om_\CL$, since by applying $\delta$ 
to \eqref{HamiltonianCurrentCondition} they preserve the presympectic current, up to a horizontally exact form
\vspace{-2mm}
$$
\mathbb{L}_\CZ \om_\CL = \dd_M \CB_\CZ \, , 
$$

\vspace{-2mm}
\noindent
where $\CB_\CZ = -\delta \CT_\CZ$. In other words, Hamiltonian vector fields of $\om_\CL$ are \textit{symplectic} vector fields.

\medskip 
As with symmetries of a Lagrangian and Noether currents, to a fixed Hamiltonian vector field $\CZ$ satisfying \eqref{HamiltonianCurrentCondition}, 
there exists a corresponding \textit{family} of Hamiltonian currents $\CH_\CZ$ parametrized by the freedom of adding any $(d-1,0)$-current $\tilde{\CH}$ 
whose vertical differential vanishes - up to a horizontally exact form, $\delta \tilde{\CH} = \dd_M \tilde{\CT}$ for some 
$\tilde{\CT} \in \Omega^{d-2,1}_\loc(\CF\times M)$. To keep track of these families, we consider instead pairs of Hamiltonian currents and vector fields. 
\begin{definition}[\bf  Hamiltonian Pairs]\label{Hamiltonian pairs}
The set of \textit{Hamiltonian pairs} of a local Lagrangian field theory $(\CF,\CL)$, with respect to a presymplectic potential $\theta_\CL$ and hence 
current $\om_\CL=\delta \theta_\CL$, is defined by 
\vspace{-1mm}
\begin{align*}
\mathrm{HamPrs}(\CF\times M,\om_\CL):= \big\{(\CH_\CZ, \CZ) \in \Omega^{d-1,0}_\loc(\CF\times M)\times \CX_\loc(\CF) 
\, \, \big\vert \, \, \exists \, \CT_\CZ\, \, \mathrm{s.t.}\, \,  \iota_\CZ \om_\CL + \delta \CH_\CZ  = \dd_M  \CT_\CZ  \big\}\, .
\end{align*}

\vspace{-1mm}
\noindent 
The subspace
\vspace{0mm}
$$\mathrm{StrHamPrs}(\CF\times M,\om_\CL)\longhookrightarrow \mathrm{HamPrs}(\CF\times M,\om_\CL) 
$$

\vspace{0mm}
\noindent of Hamiltonian pairs such that $\iota_\CZ \om_\CL + \delta \CH_\CZ=0$ vanishes exactly (and not up to a $\dd_M$-exact current) 
are called \textit{strict} Hamiltonian pairs.
\end{definition}
Mimicking the construction of a Poisson bracket in finite-dimensional (pre)symplectic geometry, there is an analogous `Poisson-like'
bracket of Hamiltonian currents. 
\begin{definition}[\bf  Bracket of Hamiltonian pairs]\label{BracketOfHamiltonianPairs}
The bracket of Hamiltonian pairs 
 \vspace{-2mm}
 $$
 \{-,-\}_H \,:\, \mathrm{HamPrs}(\CF\times M,\om_\CL) \times \mathrm{HamPrs}(\CF\times M,\om_\CL)\longrightarrow \mathrm{HamPrs}(\CF\times M,\om_\CL)
 $$
 
 \vspace{-2mm}\noindent is defined by assigning the contraction of the presymplectic current
 \vspace{-2mm}
\begin{align*}
\big\{(\CH_{\CZ_1}, \CZ_1), \, (\CH_{\CZ_2}, \CZ_2)\big\}_H :&=\big(-\iota_{\CZ_1}\iota_{\CZ_2} \om_\CL, \, [\CZ_1,\CZ_2] \big)\, .
\end{align*}
\end{definition}

The fact that this is well-defined, i.e., the contraction $-\iota_{\CZ_1}\iota_{\CZ_2} \om_\CL$ is indeed a Hamiltonian current for 
$[\CZ_1,\CZ_2]\in \CX_\loc(\CF)$, follows by a straightforward application of the local Cartan calculus. Indeed, a short calculation shows
\vspace{-2mm}
\begin{align*}
- \delta \iota_{\CZ_1}\iota_{\CZ_2} \om_\CL = -\iota_{[\CZ_1,\CZ_2]} \om +
\dd_M \big(\iota_{\CZ_2} \delta \CT_{\CZ_1} - \iota_{\CZ_1} \delta  \CT_{\CZ_2} \big) \, ,
\end{align*}

\vspace{-2mm}
\noindent 
where $\CT_{\CZ_1},\,  \CT_{\CZ_2} \in \Omega^{d-2,1}_\loc(\CF\times M)$ are any representative local $(d-2,1)$-forms for the Hamiltonian current
condition \eqref{HamiltonianCurrentCondition}, for each of the pairs respectively. This is simply the Hamiltonian current condition for the pair
$(-\iota_{\CZ_1}\iota_{\CZ_2} \om_\CL, \, [\CZ_1,\CZ_2])$, hence exhibiting the bracket $\{-,-\}_H$ as a well-defined map of Hamiltonian pairs. 

The relation between the two brackets of Def. \ref{BracketOfNoetherPairs} and Def. \ref{BracketOfHamiltonianPairs} follows easily, 
since by applying a further contraction on the defining Hamiltonian condition \eqref{HamiltonianCurrentCondition} we have
\vspace{-2mm}
\begin{align}\label{BracketOfHamiltonianPairsViaLieDerivative}
-\iota_{\CZ_1}\iota_{\CZ_2}\om_\CL &= \mathbb{L}_{\CZ_1} \CH_{\CZ_2} + \dd_M\iota_{\CZ_1} \CT_{\CZ_2}  \\
&=\tfrac{1}{2}\big(\mathbb{L}_{\CZ_1} \CH_{\CZ_2} -\mathbb{L}_{\CZ_2} \CH_{\CZ_1}\big) 
+ \tfrac{1}{2} \dd_M\big(\iota_{\CZ_1} \CT_{\CZ_2} - \iota_{\CZ_2} \CT_{\CZ_1} \big) \nn
\end{align}

\vspace{-2mm}
\noindent with the latter equality holding by the manifest antisymmetry of the double contraction.
\begin{lemma}[{\bf Relation between brackets of currents}]\label{RelationBetweenBracketsOfCurrents}
On the subspace of Noether pairs $\mathrm{NoethPrs}(\CL,\theta_\CL) \hookrightarrow \mathrm{HamPrs}(\CF\times M,\om_\CL)$, the bracket 
of Hamiltonian pairs coincides with that of Noether pairs, up to a horizontally exact current. Explicitly, 
\vspace{-2mm}
\begin{align*}
\big\{(\CP_{\CZ_1}, \CZ_1), \, (\CP_{\CZ_2}, \CZ_2)\big\}_H &= \big(\mathbb{L}_{\CZ_1} \CP_2 + \dd_M \iota_{\CZ_1} \CT_{\CZ_2}, \, [\CZ_1,\CZ_2]  \big) \\
&= \big\{(\CP_{\CZ_1}, \CZ_1), \, (\CP_{\CZ_2}, \CZ_2)\big\}_N + \big(  \dd_M(\iota_{\CZ_1} \CT_{\CZ_2} - \iota_{\CZ_2} \CT_{\CZ_1}) ,\,0\big) \, ,
\end{align*}
for any $(\CP_{\CZ_1}, \CZ_1), \, (\CP_{\CZ_2}, \CZ_2) \, \in \, \mathrm{NoetherPair}(\CL,\theta_\CL)$.

\vspace{-2mm}
\end{lemma}
\begin{proof}
By Lem. \ref{ContractedSymplecticFormAndConservedCharge}, any Noether current is in particular Hamiltonian, hence Noether pairs
are indeed a vector subspace of Hamiltonian pairs. The result follows immediately by considering \eqref{BracketOfHamiltonianPairsViaLieDerivative}
for the case of the Hamiltonian currents being Noether currents and comparing with the formulas from Def. \ref{BracketOfNoetherPairs}.   
\end{proof}

As with the case of the Noether bracket, the bracket of Hamiltonian pairs is naturally antisymmetric and $\FR$-linear, but the horizontally-exact 
terms spoil Jacobi identity. An easy way to see this is as follows: Since $\om_\CL$ is a local $(d-1,2)$-form, it follows that the contraction 
with any three local vector fields vanishes $\iota_{\CZ_1} \iota_{\CZ_2} \iota_{\CZ_3} \om_\CL =0$. Applying $\delta$ on the equation and using 
the local Cartan calculus repeatedly, one arrives at
\vspace{-2mm}
$$
0 =  \big ( - (\iota_{\CZ_1} \iota_{[\CZ_2,\CZ_3]}\om_\CL + \mathrm{cyc}(1,2,3) \big) 
+ \dd_M \big( - \iota_{\CZ_1} \iota_{\CZ_2} \delta T_{\CZ_3} + \mathrm{cyc}(1,2,3) \big) \, ,
$$

\vspace{-2mm}
\noindent and so the Jacobiator reads 
\vspace{-2mm}
\begin{align}\label{HamiltonianBracketJacobiator}
\mathrm{Jac}_{\{-,-\}_H}\big((\CH_{\CZ_1}, \CZ_1), \, (\CH_{\CZ_2}, \CZ_2), \, (\CH_{\CZ_3}, \CZ_3) \big)
:&= \Big\{ \big\{(\CH_{\CZ_1}, \CZ_1), \, (\CH_{\CZ_2}, \CZ_2)\}_H\, , \,(\CH_{\CZ_3}, \CZ_3) \Big\}_H  + \mathrm{cyc}(1,2,3)  \nn \\
&= \Big( \dd_M\big( \iota_{\CZ_1} \iota_{\CZ_2} \delta T_{\CZ_3} + \mathrm{cyc}(1,2,3) \big ) , \, 0 \Big) \, .
\end{align}

\vspace{-2mm}
\noindent The explicit form of the Jacobiator \eqref{NoetherBracketJacobiator} for the Noether bracket can be deduced from
\eqref{HamiltonianBracketJacobiator}, by substituting the relation from Lem. \ref{RelationBetweenBracketsOfCurrents}. 

There are two ways to identify a Lie algebra structure from the above bracket. Obviously, 
as in finite-dimensional geometry, $\{-,-\}_H$ defines a strict Lie algebra structure on the subspace of \textit{strict} Hamiltonian pairs
\vspace{-2mm}
\begin{align}\label{LieAlgebraOfStrictHamiltonianPairs}
\big( \mathrm{StrHamPrs}(\CF\times M,\om_\CL) \, , \{-,-\}_H \big)
\end{align}

\vspace{-2mm}
\noindent
where the $\dd_M$-exact terms never appear. Otherwise, the bracket induces a Lie algebra structure on the quotient of Hamiltonian pairs by horizontally 
exact currents, inside which the Dickey Lie algebra of Lem. \ref{DickeyLieAlgebraOfNoetherCurrents} naturally sits as a Lie subalgebra.
\begin{lemma}[{\bf Local Poisson Lie algebra of Hamiltonian currents}]\label{LocalPoissonLieAlgebraOfHamiltonianCurrents}
$\,$

\noindent {\bf (i)} The bracket of Hamiltonian pairs descends to a Lie bracket on the quotient
\vspace{-2mm}
\begin{align*}
\mathrm{HamPrs}(\CF\times M,\om_\CL) / \dd_M \Omega^{d-2,0}_\loc(\CF\times M)\, ,  
\end{align*}

\vspace{-2mm}
\noindent known as the local Poisson Lie bracket of (Hamiltonian) currents.

\noindent {\bf (ii)} Furthermore, the local Poisson Lie algebra 
is part of a short exact sequence of Lie algebras 
\vspace{-2mm}
\begin{align}\label{LocalPoissonLieAlgebraCentralExtension}
\Omega^{d-1,0}_{\delta=\dd_M}(\CF\times M) / \dd_M \Omega^{d-2,0}_\loc(\CF\times M) 
\; \longhookrightarrow \; 
\mathrm{HamPrs}(\CF\times M,\om_\CL) / \dd_M \Omega^{d-2,0}_\loc(\CF\times M) 
\; \longrightarrow \;
\CX^{\mathrm{Ham},\om_\CL}_\loc(\CF) \, ,
\end{align}

\vspace{-2mm}
\noindent with the subspace of local $(d-1,0)$-forms on the left being those $\tilde{\CH}$ whose vertical differential is $\dd_M$-exact 
$\delta \tilde{\CH}=\dd_M \tilde{\CT}$, modulo horizontally exact currents, thought of as a trivial Lie algebra. That is,  
the local Poisson Lie algebra of Hamiltonian currents (modulo trivial currents) is a central extension of the Lie algebra of  
Hamiltonian vector fields of $\CX^{\mathrm{Ham},\om_\CL}_\loc(\CF)$.

\begin{proof} The proof follows identically as in Lem. \ref{DickeyLieAlgebraOfNoetherCurrents}, by replacing the Noether with 
the corresponding Hamiltonian counterparts. 
The bracket of Hamiltonian pairs \eqref{BracketOfHamiltonianPairs} descends to a Lie bracket since the Jacobiator \eqref{HamiltonianBracketJacobiator} 
vanishes in cohomology. The short exact sequence follows immediately: the first map is the canonical inclusion $[\tilde{\CH}]\mapsto ([\tilde{\CH}],\, 0)$ 
and the second is the canonical projection $([\tilde{\CP_\CZ}],\, \CZ)\mapsto \CZ $ onto Hamiltonian vector fields of $\om_\CL$. 
The kernel of the former coincides with the image of the latter, by definition. This subset represents exactly the remaining freedom 
of adding a $\delta$-closed $(d-1,0)$-form - up to an $\dd_M$-exact current - to each class Hamiltonian current class $[\CH_{\CZ}]$
corresponding to a fixed Hamiltonian vector field $\CZ$. 
\end{proof}
\end{lemma}
Strictly speaking, Hamiltonian pairs are only a vector space and they do not have a commutative algebra structure. That is, the bracket of Hamiltonian pairs 
is only a Lie bracket and not Poisson. The commutative algebra structure appears after transgressing currents as per Def. \ref{LocalFunctionsOnFieldSpace}, 
whereby one may analogously define Hamiltonian functionals with respect to the transgressed symplectic $2$-form. It is in that setting that the actual 
Poisson structure appears. We will explain this treatment in the following section and relate it to the brackets of currents defined above 
in Lem. \ref{PoissonVsCurrentBrackets}.

\begin{remark}[\bf Brackets of on-shell Hamiltonian currents]\label{OnshellBracketsCaveats}
For cases of classical field theories in which the local Cartan calculus descends to $\CF_{\CE \CL}\times M$, and so the requirements
of Rem. \ref{OnShellCartanCalculusCaveats} are fulfilled, the above discussion applies analogously to \textit{on-shell} Hamiltonian pairs. 
More explicitly, these are on-shell pairs
\vspace{-2mm} 
$$
(\CH_\CZ|_{\CE\CL},\CZ|_{\CE \CL})\;\; \in \Omega^{d-1,0}_\loc(\CF_{\CE \CL}\times M) \times \CX_{\loc}(\CF_{\CE \CL})
$$

\vspace{-2mm} 
\noindent that satisfy the corresponding on-shell Hamiltonian condition 
$\iota_\CZ \om_\CL|_{\CE \CL} + \delta \CH_\CZ|_{\CE \CL} = \dd_M \CT_\CZ |_{\CE \CL} $. 
It follows (modulo Rem. \ref{OnShellCartanCalculusCaveats}), that such pairs are represented by pairs
$(\CH_\CZ, \CZ)\in \Omega^{d-1,0}_\loc(\CF\times M)$ such that $\CZ$ is tangent to $\CF_{\CE \CL}$
and the Hamiltonian condition is satisfied off-shell up to a $(d-1,1)$-form proportional to a differential operator
applied to the Euler--Lagrange form 
$\iota_\CZ \om_\CL + \delta \CH_\CZ = \dd_M \CT_\CZ + \CJ^I \partial_I \CE \CL$. The above results follow for 
the on-shell Hamiltonian pairs by replacing the off-shell objects and conditions with the corresponding on-shell versions.
\end{remark}

\begin{remark}[\bf Higher Poisson algebras of local observables]\label{HigherPoissonAlgebrasOfLocalObservables}
$\,$

\noindent {\bf (i)} Just as with the bracket of Noether pairs, taking the strict quotient by horizontally
exact Hamiltonian currents discards a lot of information. 
In particular, information on charges of Hamiltonian currents over submanifolds with boundary is lost in this quotient. Readers familiar with higher/homotopical algebraic structures will naturally point out 
that the breaking of the Jacobi identity -- up to a homologically trivial term -- suggests an underlying structure
of an $L_\infty$-algebra, i.e., a ``higher Poisson $L_\infty$-algebra of local observables''. 
This is indeed the case and ideas along this line of thought have been pursued in the literature, often under 
the name of ``higher presymplectic'' geometry. 

\vspace{1mm} 
\noindent {\bf (ii)} We will not be delving into a discussion of higher structures in this manuscript. 
The interested reader may find relevant ideas and constructions 
in the following manuscripts: The earliest investigation specifically on the bigraded situation of 
the jet bundle and field theory as described 
in the current manuscript is found in   \cite{BarnichStasheff}. General investigations of such structures in terms 
of higher presymplectic geometry, i.e., in terms of finite-dimensional manifolds equipped with closed $m$-forms 
(without a bigrading), were initiated in \cite{Rogers}. 
Further investigations with a viewpoint towards quantization, topological terms and BPS charges were pursued in \cite{FRS14}\cite{BPS}. 
A closely related treatment in terms of $L_\infty$-algebras of local observables, but not equivalent, to that of \cite{BarnichStasheff} may 
be found in \cite{Del}. Therein, the discussion differs from ours since the total grading on the bicomplex of local forms is employed, and 
Hamiltonian pairs are defined relative to the $(\delta+\dd_M)$-closed ``Poincar\'e-Cartan'' $\tilde{\om}_\CL:= \om_\CL - \CE \CL$.  
\end{remark}

\subsection{The covariant phase space, off-shell and on-shell Poisson brackets} 
\label{Sec-cov}

The results of \cref{SecPresymplecticCurrentPropertiesAndBrackets}, which take place on $\CF\times M$ and its smooth subspace $\CF_{\CE \CL}\times M$,
may now be \textit{transgressed} to statements about the actual field space $\CF$ and the subspace of on-shell fields $\CF_{\CE \CL}$, 
essentially by integrating the horizontal parts of the bi-graded forms involved over submanifolds of spacetime $M$. 
We will use the notation $\varpi^{p,q}:=(\ev^\infty)^*\om^{p,q}$.

\begin{definition}[\bf  Transgression]\label{Transgression}
Let $\Sigma^p\hookrightarrow M$ be a compact oriented submanifold. The transgression of local $(p,q)$-forms on $\CF\times M$ to $q$-forms on $\CF$
\vspace{-2mm}
\begin{align*}
\tau_{\Sigma^{p}} \,:\, \Omega^{p,q}_\loc(\CF\times M) \longrightarrow \Omega^{q}(\CF)
\end{align*}

\vspace{-2mm}
\noindent
is defined by
\vspace{-2mm}
\begin{align*}
\tau_{\Sigma^p}\big(\varpi^{p,q}\big) \;:\; (T\CF)^{\times q} &\;\; \longrightarrow \quad  \FR \\
(\CZ_\phi^1,\cdots, \CZ_\phi^q) &\;\; \longmapsto \;\; \int_{\Sigma^p} \big(\varpi^{p,q}\big)_\phi(\CZ^1_\phi,\cdots, \CZ^q_\phi, \cdots ) \, ,
\end{align*}

\vspace{-2mm}
\noindent and similarly for higher plots. We denote by 
\vspace{-1mm}
$$ 
\Omega^\bullet_{\loc,\Sigma^{p}}(\CF) \longhookrightarrow \Omega^\bullet(\CF)
$$ 

\vspace{0mm}
\noindent
the subalgebra of forms generated by transgressed local forms along $\Sigma^p$, and by $\Omega^\bullet_{\loc}(\CF)$ generated by transgressed
local forms along arbitrary compact oriented submanifolds.
\end{definition}

For $p=0$, transgression along $0$-dimensional submanifolds (points) is given simply by evaluation. Note that for $q=0$ the transgression of 
a local $p$-form current $\CP$ is simply its charge over $\Sigma^{p}$
from Def. \ref{ChargesOnFieldSpace}.

\newpage 

\begin{remark}[\bf Transgression over non-compact submanifolds]\label{TransgressionOverNoncompact}
For local $(p,0)$-forms, the transgression may be defined over noncompact submanifolds by restricting to the subspace of compactly supported field 
$\CF_{\mathrm{cpt}}\hookrightarrow \CF$. However, for local $(p,q)$-forms with $q\geq 1$, the transgression may be defined over noncompact submanifolds 
over the full field space, by restricting the domain of the transgressed map to be products of $T_\mathrm{cpt}\CF\hookrightarrow T\CF$, whose points
are tangent vectors to field space with compact support along the fibers. That is, a tangent vector $\CZ_\phi\in T_{\mathrm{cpt},\phi}\CF$ is a
section $\CZ_\phi\in \Gamma_M(\phi^*VF)$ which maps to the zero element in each fiber $ V_{\phi(x)}F$ outside a compact
region\footnote{This can be further relaxed to tangent vectors with compact support only \textit{after 
 restriction} along $\Sigma^p$.} of $M$. Thus, the transgression of a local $(p,q)$-form over a noncompact submanifold $\Sigma^{p}\hookrightarrow M$ is,
 strictly speaking not a $q$-form on $\CF$, but instead a smooth map $(T_{\mathrm{cpt}}\CF)^{\times q}\rightarrow \FR$. Everything we say below applies 
 verbatim for such situations with minimal modifications, which we shall not make explicit. 
\end{remark}

Of special interest are the transgressions of the presymplectic potential current $\theta_\CL:= (\ev^\infty)^* \theta_L $ and, more importantly, 
of the presymplectic current $\om_\CL = \delta \theta_\CL $ from Def. \ref{PresymplecticCurrent}.
\begin{definition}[\bf  Presymplectic potential and 2-form on field space]\label{PresymplecticPotentialAnd2formOnFieldSpace}
The presymplectic potential $\Theta_{\CL, \Sigma^{d-1}}$ and presymplectic 2-form $\Omega_{\CL, \Sigma^{d-1}}$ of a Lagrangian field theory 
$(\CF,\CL)$, with respect to a compact oriented submanifold $\Sigma^{d-1}$, are defined as the transgressions of the corresponding currents
\vspace{-2mm}
\begin{align*}
\Theta_{\CL, \Sigma^{d-1}}:= \tau_{\Sigma^{d-1}} (\theta_\CL) \;\; : \;\; T\CF \longrightarrow \FR 
\end{align*}

\vspace{-2mm}
\noindent
and
\vspace{-3mm}
\begin{align*}
\hspace{7mm} \Omega_{\CL, \Sigma^{d-1}}:= \tau_{\Sigma^{d-1}} (\om_\CL) \;\; : \;\; T\CF\times T\CF \longrightarrow \FR  \, ,   
\end{align*}

\vspace{-2mm}
\noindent respectively.

\end{definition}
Now if $\Sigma^{d-1}$ is without boundary, then both the presymplectic potential and $2$-form with respect to $\Sigma^{d-1}$
are uniquely defined -- in contrast with the corresponding currents which are only defined up to an addition of a $\dd_M$-exact current. 
In turn, these may be restricted/pulled back to forms on the space of on-shell fields $\CF_{\CE \CL}$. 
Equivalently, the currents may be first restricted to $\CF_{\CE \CL}\times M$ and then transgressed via Def. \ref{Transgression}. 
In particular, the presymplectic 2-form 
\vspace{-2mm}
\begin{align*}
\Omega_{\CL, \Sigma^{d-1}}\big\vert_{\CE \CL}:= \tau_{\Sigma^{d-1}} (\om_\CL) \big\vert_{\CE \CL} \equiv
\tau_{\Sigma^{d-1}} (\om_\CL \vert_{\CE \CL} ) \quad : \quad T \CF_{\CE \CL} \times T \CF_{\CE \CL} 
\longhookrightarrow T\CF\times T\CF \longrightarrow \FR     
\end{align*}

\vspace{-1mm}
\noindent depends only on the cobordism class of $\Sigma^{d-1}$. 

\begin{lemma}[{\bf On-shell presymplectic 2-form and cobordism classes}]\label{OnShell2formAndCobordisms}
Let $B^d$ be a cobordism between two compact oriented submanifolds $\partial B^d = \Sigma_1^{d-1}\coprod \Sigma_2^{d-1}$.
Then the corresponding transgressed on-shell presymplectic 2-forms coincide,
\vspace{-4mm}
\begin{align*}
\Omega_{\CL, \Sigma_1^{d-1}}\big\vert_{\CE \CL} = \Omega_{\CL, \Sigma_2^{d-1}}\big\vert_{\CE \CL} \quad \in \quad \Omega^{2}(\CF_{\CE \CL})\, . 
\end{align*}

\vspace{-4mm}
\begin{proof}
By Lem. \ref{PresymplecticCurrentIsOnshellConserved}, the presymplectic current is conserved on-shell $\dd_M \om_\CL |_{\CE \CL}=0$ 
and so for any cobordism $B^d$ between two compact oriented submanifolds $\partial B^d = \Sigma_1^{d-1}\coprod \Sigma_2^{d-1}$, we have
\vspace{-2mm}
\begin{align*}
0=\; \int_{B^d} \dd_M \om_\CL |_{\CE \CL} = \; \int_{\Sigma_1^{d-1}} \om_\CL|_{\CE \CL} -\int_{\Sigma_2^{d-1}} \om_\CL|_{\CE \CL} \, ,
\end{align*}

\vspace{-2mm}
\noindent and the result follows.
\end{proof}

\end{lemma}
In particular, this implies that if $\Sigma^{d-1}$ is cobordant to the empty manifold, i.e., is homologically trivial in $M$, 
then the corresponding on-shell $2$-form is trivial.
We have called $\Omega_{\CL, \Sigma^{d-1}}\in \Omega^{2}(\CF)$ a `presymplectic' 2-form on $\CF$, and similarly its restriction on $\CF_{\CE \CL}$, 
by virtue of $\tau_{\Sigma^{d-1}}(\delta \om_\CL) = 0$. In order to witness it as a truly presymplectic $2$-form on $\CF$, and that 
$\Theta_{\CL,\Sigma^{d-1}}$ is its potential, a notion of a differential on the algebra generated by transgressed forms 
$\Omega^{\bullet}_{\loc}(\CF)\hookrightarrow \Omega^{\bullet}(\CF)$ is required. Intuitively, and as implicitly practiced in the physics 
literature, this is given by computing the vertical differential under the integral. The well-definiteness of this operation, however, 
requires justification.

\begin{lemma}[{\bf Transgressed vertical differential}]\label{TransgressedVerticalDifferential}
The vertical differential $\delta:\Omega^{p,q}_\loc(\CF\times M)\rightarrow \Omega^{p,q+1}_\loc(\CF\times M)$ transgresses along any compact 
oriented submanifold $\Sigma^{p}$ to a well-defined map
\vspace{-2mm}
\begin{align*}
  \delta \,:\, \Omega^\bullet_{\loc,\Sigma^{p}}(\CF) &\longrightarrow \Omega^{\bullet+1}_{\loc,\Sigma^{p}}(\CF)
\end{align*}

\vspace{-3mm}
\noindent
defined on transgressed local forms by
\vspace{-3mm}
$$
\delta \,\big( \tau_{\Sigma^p}(\varpi^{p,q})\big)
:=\tau_{\Sigma^p}\big(\delta( \varpi^{p,q})\big)
\equiv \int_{\Sigma^{p}} \delta(\varpi^{p,q}) \, , 
$$

\vspace{-2mm}
\noindent
and extends to the full algebra generated by transgressed local forms as a derivation $\Omega^\bullet_\loc(\CF)\rightarrow \Omega^{\bullet+1}_\loc(\CF)$. 
\end{lemma}
\begin{proof}
The only ambiguous part of the definition is that there might be two different local $(p,q)$-forms $\varpi^{p,q}\neq  
\tilde{\varpi}^{p,q}:=(\ev^\infty)^*\tilde{\om}^{p,q}$ such that their transgressions agree
\vspace{-2mm}
$$
\int_{\Sigma^p} \varpi^{p,q} = \int_{\Sigma^p} \tilde{\varpi}^{p,q} \, .
$$

\vspace{-2mm}
\noindent 
Thus for the transgressed vertical differential to be a well-defined map, it better be the case that the transgressions 
of the vertical differentials $\delta \varpi^{p,q},\, \delta \tilde{\varpi}^{p,q}$ also coincide 
\vspace{-2mm}
$$
\int_{\Sigma^p} \delta \varpi^{p,q} = \int_{\Sigma^p}\delta \tilde{\varpi}^{p,q} \, .
$$

\vspace{-2mm}
The proof of this follows by an inductive argument. Consider first case of two local $(p,0)$-form currents 
$\CP,\tilde{\CP} \in \Omega^{p,0}(\CF\times M)$ with the same 
charge $\int_{\Sigma^p} \CP = \int_{\Sigma^p} \tilde{\CP}$ along $\Sigma^{p}$. In this case, the result follows essentially
as with the variation of the charge of a Lagrangian (see proof of Prop. \ref{Crit(S)functoriality}). 
Indeed, since the charge maps agree as maps of smooth sets, it follows that for any $1$-parameter family of fields 
\vspace{-1mm}
$$
\int_{\Sigma^p} \CP(\phi_t) = \int_{\Sigma^p} \tilde{\CP}(\phi_t) \quad \in \quad C^\infty(\FR_t) \, .
$$ 

\vspace{0mm}
\noindent The derivative evaluated at $t=0$ gives 
$\int_{\Sigma^p}\delta \CP_{\phi_0}(\partial_t \phi_t |_{{t}=0})= \int_{\Sigma^p}\delta \tilde{\CP}_{\phi_0}(\partial_t \phi_t |_{{t}=0})$ 
along the lines of Ex. \ref{InfinityJetVerticalVectorFromProlongationofPlot}. Since any tangent vector on $\CF$ is represented by a line plot 
(Lem. \ref{LinePlotsRepresentTangentVectors}), the result follows.

Next, consider the case of two local $(p,1)$-forms $\varpi^{p,1}, \,   \tilde{\varpi}^{p,1}$ such that the transgressions agree $\int_{\Sigma^p} \varpi^{p,1} = \int_{\Sigma^p} \tilde{\varpi}^{p,1}$. It follows that precomposing with local vector field $\CZ_1: \CF \rightarrow T\CF$ implies
\vspace{-2mm}
$$
\int_{\Sigma^p} \iota_{\CZ_1} \varpi^{p,1} = \int_{\Sigma^p} 
 \iota_{\CZ_1} \tilde{\varpi}^{p,1} \, ,
$$

\vspace{-2mm}
\noindent i.e., the equality of the transgressions of the local $(p,0)$-forms  $\iota_{\CZ_1} \varpi^{p,1}$ and $ \iota_{\CZ_1} \tilde{\varpi}^{p,1}$. 
But we have already shown that for $(p,0)$-forms this further implies 
\vspace{-1mm}
$$
\int_{\Sigma^p} \delta \iota_{\CZ_1} \varpi^{p,1} = \int_{\Sigma^p} 
 \delta \iota_{\CZ_1} \tilde{\varpi}^{p,1} \, ,
$$

\vspace{0mm}
\noindent whereby precomposing with a further local vector field $\CZ_2: \CF\rightarrow T\CF$ gives
\vspace{-2mm}
$$
\int_{\Sigma^p} \iota_{\CZ_2} \delta \iota_{\CZ_1} \varpi^{p,1} = \int_{\Sigma^p} \iota_{\CZ_2} 
 \delta \iota_{\CZ_1} \tilde{\varpi}^{p,1} \, ,
$$

\vspace{-2mm}
\noindent or in terms of Lie derivatives 
\vspace{0mm}
$$
\int_{\Sigma^p} \mathbb{L}_{\CZ_2} \iota_{\CZ_1} \varpi^{p,1} = \int_{\Sigma^p} \mathbb{L}_{\CZ_2} 
  \iota_{\CZ_1} \tilde{\varpi}^{p,1} \, .
$$

\vspace{0mm}
\noindent Collecting the above, we may use the local Cartan calculus (Lem. \ref{LocalCartanCalculus}) under the integral to compute
\vspace{-2mm}
\begin{align*}
\int_{\Sigma^p} \delta \varpi^{p,1} (\CZ_1,\CZ_2) &=  \int_{\Sigma^p}   \mathbb{L}_{\CZ_1} \iota_{\CZ_2} \varpi^{p,1} -  \mathbb{L}_{\CZ_2} \iota_{\CZ_1} \varpi^{p,1} - \iota_{[\CZ_1,\CZ_2]} \varpi^{p,1} \\
&= \int_{\Sigma^p}   \mathbb{L}_{\CZ_1} \iota_{\CZ_2} \tilde{\varpi}^{p,1} -  \mathbb{L}_{\CZ_2} \iota_{\CZ_1} \tilde{\varpi}^{p,1} - \iota_{[\CZ_1,\CZ_2]}\tilde{\varpi}^{p,1} \\
&= \int_{\Sigma^p} \delta \tilde{\varpi}^{p,1} (\CZ_1,\CZ_2) \, . 
\end{align*}

\vspace{-2mm}
\noindent The case of transgressed local $(p,q)$-forms follows analogously, by induction.
\end{proof}

\begin{remark}[\bf Caveat on transgressed vertical differential]\label{CaveatOnTransgressedVerticalDifferential}
$\,$

\noindent {\bf (i)} Strictly speaking, the proof of Lem. \ref{TransgressedVerticalDifferential} only shows that the differential is well-defined on the maps 
of local vector fields $\big(\CX_\loc(\CF)\big)^{\times q} \rightarrow C^\infty_\loc(\CF)$ induced by transgressed local forms. 
This is sufficient for all the uses of the transgressed vertical differential that will follow. To prove that it is actually well-defined 
as a map of transgressed differential forms viewed as bundle maps $(T\CF)^{\times q}\rightarrow \FR$, it is necessary to show that any 
tangent vector $\CZ_\phi \in T_\phi(\CF)$ extends to a \textit{local} vector field on $\CF$. 
We expect this is indeed true, but the proof requires Fr\'{e}chet analytical details, which we do not need to delve into here. 

\vspace{1mm} 
\noindent {\bf (ii)}
A sketch of the extension would be as follows: A tangent vector $\CZ_\phi \in T_\phi \CF$ is a smooth bundle map $\CZ_\phi: M\rightarrow VF$ 
covering $\phi$, and its jet prolongation (Ex. \ref{InfinityJetVerticalVectorFromProlongationofPlot}) defines 
$j^\infty \CZ_\phi : M \rightarrow  V J^\infty_M F$ covering $j^\infty \phi: M \rightarrow J^\infty_M F $. 
The image $j^\infty \phi (M) \hookrightarrow J^\infty_M F$ should be a closed embedded Fr\'{e}chet submanifold of the Hausdorff 
and paracompact infinite jet bundle (Def. \ref{JetBundleLocPro}), by essentially the same proof of the finite-dimensional manifold statement. 
Thus the prolongation of $\CZ_\phi$ is equivalently a section of the vector bundle $VJ^\infty_M F |_{j^\infty\phi(M)}$ over $j^\infty \phi(M)$. 
Composing with the pushforward of the projection $\dd \pi^\infty_0 : VJ^\infty_M F\rightarrow VF$ gives a section
$J^\infty_M F |_{j^\infty \phi(M)} \rightarrow (\pi^\infty_0)^*VF|_{j^\infty \phi(M)}$. Assuming the result of extension of 
vector bundle sections (see proof of Lem. \ref{LinePlotsRepresentTangentVectors}) over closed submanifolds extends to the 
Hausdorff, paracompact case of the infinite jet bundle and a finite rank vector bundle over it, then there exists an extended 
section $J^\infty_M F \rightarrow (\pi^\infty_0)^*VF$ over $J^\infty_M F$. Equivalently, this is a bundle map 
$Z : J^\infty_M F \rightarrow VF$ over $F$, i.e., an evolutionary vector field such that
$\CZ (\phi):= Z \circ j^\infty(\phi) = \CZ_\phi \in T_\phi \CF$. Thus the induced local vector field 
$\CZ:= Z\circ j^\infty \in \CX_\loc(\CF)$ extends $\CZ_\phi \in T_\phi \CF$.
\end{remark}

Note that, by the proof of Prop. \ref{Crit(L)functoriality}, the criticality/on-shell condition (Def. \ref{CriticalRkPointsofLagrangian}) 
of a field $\phi$ (and similarly of plots) corresponding to a \textit{local} Lagrangian $\CL$, may be equivalently expressed as the joint
vanishing of the vertical transgressed differential 
\footnote{This is essentially the transgressed characterization from the footnote of the decomposition 
\eqref{LagrangianVerticalDifferentialDecomposition}.} 
\vspace{-2mm}
\begin{align}\label{VerticalTransgressedDifferentialOfAction}
\delta S_{\Sigma^d}|_\phi \;:\; T_\phi \CF \longrightarrow \FR \, ,
\end{align}

\vspace{-2mm}
\noindent
for all compact oriented submanifolds $\Sigma^d\hookrightarrow M$. This provides a further (equivalent) interpretation for the `variation' symbol of local
Lagrangians/action functionals as used in the physics literature (cf. Rem. \ref{CriticalityViaModuliSpaceOf1-forms}{{(iii)}}).
Furthermore, using the transgressed vertical differential, it is \textit{actually} the case that $\Theta_{\CL,\Sigma^{d-1}}\in \Omega^1_\loc(\CF) $ 
is the presymplectic potential of $\Omega_{\CL,\Sigma^{d-1}}\in \Omega^2_\loc(\CF)$, and that the latter is indeed presymplectic, i.e., 
\vspace{-2mm}
$$
\Omega_{\CL,\Sigma^{d-1}}= \delta \Theta_{\CL,\Sigma^{d-1}} \qquad \mathrm{and} \qquad \delta \Omega_{\CL,\Sigma^{d-1}} = 0 \, . 
$$

\vspace{-2mm}
In the case of a product spacetime with a \textit{chosen} distinguished `space and time' splitting $M^d\cong N^{d-1}\times \FR $, 
as usually assumed in physical 
theories, there is a uniquely determined $2$-form on the on-shell space of fields $\CF_{\CE \CL}$, which makes it canonically into a presymplectic smooth set.

\begin{definition}[\bf  Covariant phase space]\label{CovariantPhaseSpace} Let $(\CF,\CL)$ be a local classical field theory on a product
spacetime $M=N\times \FR$, where $N$ is a compact oriented manifold. The \textit{covariant phase space} of $(\CF,\CL)$ is defined 
as the presymplectic smooth set
\vspace{-2mm}
$$(\CF_{\CE \CL}\, , \Omega_{\CL}) \, ,$$

\vspace{-2mm}
\noindent 
where $\Omega_{\CL}:= \Omega_{\CL, N\times\{0\}}\big\vert_{\CE \CL} \equiv \Omega_{\CL, N\times\{t_0\}}\big\vert_{\CE \CL} $ is the canonical transgression 
of the on-shell presymplectic current by Lem. \ref{OnShell2formAndCobordisms}.
\end{definition}
Strictly speaking, if $N$ has a non-trivial boundary, then the definition of $\Omega_\CL$ actually depends on the choice of presymplectic current $\theta_\CL$, 
parametrized by the addition of $\dd_M$-exact $(d-1,1)$-currents. Assuming the requirements of Rem. \ref{EvolutionaryCartanCalculus} are met, so that the local Cartan calculus descends to $\CF_{\CE \CL}\times M$, then the corresponding vertical differential also transgresses to $\Omega^\bullet_\loc(\CF_{\CE \CL})$, by the same proof as in Lem. \ref{TransgressedVerticalDifferential}, and hence $\Omega_\CL$ is truly a presymplectic form $\delta \Omega_\CL$ on $\CF_{\CE \CL}$.

\medskip 
Of course, as we have hinted in the discussion of Lem. \ref{Gaugesymmetryvspresymplecticcurrent}, if infinitesimal gauge symmetries exist then the covariant 
phase space cannot be symplectic,  in that the 2-form $\Omega_\CL$ is necessarily degenerate. 

\begin{lemma}[{\bf Gauge symmetries imply degeneracy}]\label{GaugeSymmetriesImplyDegeneracy}
Let $(\CF,\CL)$ be a local classical field theory on $N\times \FR$ with some non-trivial (parametrized) infinitesimal gauge symmetries 
$\CR_{(-)}: \Gamma_M(E)\rightarrow \CX_\loc^\CL(\CF)$. Then the canonical on-shell presymplectic current $\Omega_\CL$ is \textit{degenerate}.
\end{lemma}
\begin{proof}
Recall by Lem. \ref{Gaugesymmetryvspresymplecticcurrent}, for any gauge parameter $e\in \Gamma_M(E)$ 
\vspace{-2mm}
\begin{align*}
\iota_{\CR_e} \om_\CL \big{\vert}_{\CE \CL} = \dd_M  \tilde{\CT}_e \big{\vert}_{\CE \CL}
\quad \in \quad \Omega^{d-1,1}_\loc(\CF_{\CE \CL}\times M) \, . 
\end{align*}

\vspace{-1mm}
\noindent Recall, also by Prop. \ref{InfinitesimalSymmetryTangentToOnshellFieldSpace}, that the restriction $\CR_{e}|_{\CE \CL}$ does 
indeed define a vector field on $\CF_{\CE \CL}$, i.e., $\CR_{e}(\phi)\in T_\phi (\CF_{\CE \CL})$ for any $\phi \in \CF_{\CE \CL}(*)$. 
Choosing $e$ to have non-trivial support along $N\times \{0\}$, then transgressing along $N\times \{0\}$ gives
\vspace{-2mm}
$$
\Omega_\CL(\CR_{e}(\phi),-)= 0 \quad : \quad T_\phi \CF_{\CE \CL} \hookrightarrow \FR \, ,
$$

\vspace{-2mm}
\noindent for all on-shell fields $\phi$ and some $\CR_e(\phi)\neq 0 \in T_\phi \CF_{\CE \CL}$.
That is, $\Omega_\CL$ is a degenerate $2$-form on $\CF_{\CE \CL}$.
\end{proof}
Thus if a theory does have non-trivial infinitesimal gauge symmetries, one can only hope the pairing is non-degenerate on the smooth set 
quotient\footnote{This assumes all infinitesimal gauge symmetries are integrable. The tangent bundle on the quotient smooth set 
corresponds to the quotient of $T \CF_{\CE \CL}$ by the induced pushforward action of the gauge symmetries.} $\CF 
/ \mathrm{Diff}^{\CL,\mathrm{gauge}} (\CF)$ from \eqref{QuotientByGaugeSymmetries}. The pair of the quotient of on-shell fields by gauge 
symmetries with the induced presymplectic two-form
\vspace{-3mm}
\begin{align}\label{ReducedPhaseSpace}
\big( \CF_{\CE \CL} 
/ \mathrm{Diff}^{\CL,\mathrm{gauge}} (\CF)\, ,\, \bar{\Omega}_\CL \big) 
\end{align}

\vspace{-2mm}
\noindent
is known as the \textit{reduced covariant phase space}. We highlight that the presymplectic $2$-form does descend to the $T\CF_{\CE \CL}/\mathrm{Diff}^{\CL,\mathrm{gauge}} (\CF)$ 
quotient since finite local symmetries preserve it, by Lem. \ref{FiniteSymmetryPreservesPresymplecticCurrent}.

In the case where $N\equiv N\times \{0\}\hookrightarrow N\times \FR$ is a  Cauchy surface, and so in particular no gauge symmetries
exist\footnote{More generally, Cauchy surfaces may be formulated for the quotient of field configurations by gauge symmetries instead. 
Hence, with initial data uniquely determining an on-shell field, up a to gauge transformation. 
We will not expand on the details here.}, then one can transport the presymplectic structure along the defining isomorphism of 
Def. \ref{CauchySurface} to the smooth set of initial data on the Cauchy surface.

\begin{definition} [\bf Non-covariant phase space associated to Cauchy surface]\label{PhaseSpaceOfCauchySurface}
  The  \textit{non-covariant phase space} of $(\CF,\CL)$
with respect to the Cauchy surface $N=N\times \{0\} \hookrightarrow N\times \FR$ is defined as the presymplectic smooth set 
\vspace{-1mm}
$$\big(\mathrm{InData}_\CL(N)\, , \Omega_{\CL}^{\mathrm{Cau}}\big) \, ,$$

\vspace{-2mm}
\noindent 
where $\Omega_{\CL}^{\mathrm{Cau}}$ is the pullback of $\Omega_\CL$ along the isomorphism 
$\mathrm{Cau}_{N}=(-)|_{N} \circ j^\infty : \CF_{\CE \CL} \xlongrightarrow{\sim } \mathrm{InData}_\CL(N)$. 
\end{definition} 
Of course, for this to make sense, the correct notion of a tangent bundle on $\mathrm{InData}_\CL(N) \hookrightarrow \mathbold{\Gamma}_N (J^\infty_M F|_N)$
is assumed. We will not go
into the technical details here, \footnote{Nevertheless, these follow immediately abstractly in the infinitesimally thickened setting of \cite{GSS-2}.}
but the interested reader may verify that this is the subspace of 
\vspace{-2mm}
$$
\mathbold{\Gamma}_N\big( VJ^\infty_M F \big)\, ,
$$

\vspace{-2mm}
\noindent consisting of (plots of) sections which factor through the (synthetic) tangent bundle of the prolongated shell\footnote{With the current technology, 
this may be defined analogously to Def. \ref{OnShellTangentBundle}, i.e., as the (smooth) zero-locus of the pushforward 
$\pr EL_*: T J^\infty_M F \rightarrow T J^\infty_M F(V^*F \otimes \wedge^d T^*M)$.} $S^\infty_L\rightarrow J^\infty_M F$.
A particularly simple but illustrative example is that of the free particle. 

\begin{example}[\bf Covariant and non-covariant free particle phase space]
\label{Ex-free-cov}
Recall the free particle from Ex. \ref{ParticleMechanicsCauchySurface}, with Lagrangian $\CL=\partial_t \gamma^a \cdot \partial_t  \gamma_a \cdot \dd t$ and Euler--Lagrange form
$\CE \CL(\gamma)= \partial_t^2 \gamma \cdot \dd t$. Reducing the presymplectic current of the $O(n)$-model from Ex. \ref{PresymplecticCurrentExamples} 
to the particle case, or by calculating directly, 
\vspace{-2mm} 
$$
\om_\CL = \delta \gamma^a \wedge \delta (\partial_t \gamma_a) \quad \in \;\; \Omega_{\loc}^{0,2}\big(\mathbold{P}(\FR^d) \times \FR^1_t\big)\,.
$$

\vspace{-2mm} 
\noindent Transgression along the Cauchy surface $\{0\} \hookrightarrow \FR^1_t$ reduces to  evaluation at $t=0$, so that 
\vspace{-2mm} 
$$ 
\Omega_{\CL, \{0\}}= \delta \gamma^a \wedge \delta (\partial_t \gamma_a) \big\vert_{t=0} \quad \in \;\; \Omega^2_\loc\big(\mathbold{P}(\FR^d)\big)
$$

\vspace{-2mm} 
\noindent which further restricts to a presymplectic 2-form on the on-shell field space of $\mathrm{Lines}(\FR^d)\hookrightarrow \mathbold{P}(\FR^d)$. 
To be more explicit, recall that by linearity the tangent bundle of paths is $T \mathbold{P}(\FR^d) \cong \mathbold{P}(\FR^d)\times \mathbold{P}(\FR^d)$ (Rem. \ref{ScalarFieldTheoryTangetVectors}). Hence the tangent vectors over a path $\gamma \in \mathbold{P}(\FR^d)$ are
\vspace{-2mm}
$$ b^a \cdot \frac{\delta}{\delta \gamma^a}\big|_{\gamma}\equiv (\gamma, b)\quad \in \quad  T_{\gamma}\mathbold{P}(\FR^d) \cong \{\gamma\} \times \mathbold{P}(\FR^d)\, , 
$$ 

\vspace{-2mm}
\noindent
that is,  a copy of all paths. These are represented by $\FR^1_s$-plots of the form $\gamma + s\cdot b$. 
It follows that the presymplectic $2$-form $\Omega_{\CL, \{0\}}$ is given by the smooth map
\vspace{-2mm}
\begin{align*}
\Omega_{\CL, \{0\}} \;:\; T \mathbold{P}(\FR^d) \times_{\mathbold{P}(\FR^d)} T \mathbold{P}(\FR^d)  &\longrightarrow \FR 
\\
(\gamma, b_1,b_2)  &\longmapsto b^a_1(0) \cdot \partial_t b_{2a} (0) - b_{2}^a(0) \cdot \partial_t b_{1a}(0) \, ,
\end{align*}

\vspace{-2mm}
\noindent
which is in particular constant along the base $\mathbold{P}(\FR^d)$. From this form, it can be directly checked that the map is fiber-wise nondegenerate, 
i.e., a genuine \textit{symplectic} 2-form. However, since $b_1, b_2$ are arbitrary paths in $\FR^d$, the value does depend on the chosen transgressing 
submanifold, i.e., the chosen time instant $0\in \FR^1_t$. 

Let us now restrict this $2$-form to the on-shell field space $\mathrm{Lines}(\FR^d)\hookrightarrow \mathbold{P}(\FR^d)$. By representing tangent vectors 
at $\gamma$ via $\gamma +s\cdot b$, or otherwise, it follows that the pushforward of the Euler--Lagrange operator and the corresponding Jacobi equation 
\eqref{JacobiEquation} is 
$\CE \CL_* (\gamma,b) = (\partial_t^2\gamma, \partial_t^2 b) = 0$. Thus tangent vectors to lines are also lines, and so the tangent bundle
(Def. \ref{OnShellTangentBundle}) to the on-shell fields $\mathrm{Lines}(\FR^d)\hookrightarrow\mathbold{P}(\FR^d)$ is given by the smooth subset
\vspace{-1mm}
\begin{align*}
T \mathrm{Lines} (\FR^d) \cong \mathrm{Lines} (\FR^d) \times \mathrm{Lines} (\FR^d) \longhookrightarrow   T \mathbold{P}(\FR^d) \cong  \mathbold{P}(\FR^d)\times \mathbold{P}(\FR^d) \, .
\end{align*}

\vspace{-1mm}
\noindent
The on-shell restriction of the $2$-form $\Omega_{\CL,\{0\}}$ is given by exactly the same formula, with domain lines rather than all paths, 
\vspace{-2mm}
\begin{align*}
\Omega_\CL:= \Omega_{\CL, \{0\}}|_{\mathrm{Lines}(\FR^d)} \;\;:\;\; T \mathrm{Lines} \times_{\mathrm{Lines}(\FR^d)} T \mathrm{Lines}(\FR^d)  &\longrightarrow \FR 
\\[-2pt]
(\gamma, b_1,b_2)  &\longmapsto b^a_1(0) \cdot \partial_t b_{2 \, a} (0) - b_{2}^a(0) \cdot \partial_t b_{1 \, a}(0) \, ,
\end{align*}

\vspace{-2mm}
\noindent In contrast to the off-shell symplectic $2$-form, this is in fact independent of the transgressing time instant, as guaranteed by
Lem. \ref{OnShell2formAndCobordisms}, since both tangent vectors appearing are \textit{lines}, i.e. necessarily of the form 
$b_{1,2}= w_{1,2} \cdot t + d_{1,2} \,  \in  \, C^\infty(\FR^1_t, \FR^d)$. This completes the explicit description of the covariant
phase space (Def. \ref{CovariantPhaseSpace}) of the free particle
\vspace{-1mm} 
$$
\big( \mathrm{Lines}(\FR^d), \, \Omega_\CL \big)  \, .
$$

\vspace{-2mm} 
Moving onto the non-covariant phase space associated to the Cauchy surface $\{0\}\hookrightarrow \FR^1_t$, recall the initial data ``Cauchy diffeomorphism'' from Eq. \eqref{ParticleMechanicsIsomorphism}
\vspace{-4mm}
\begin{align*}
\mathrm{Cau}_{\{0\}} \;:\; \mathrm{Lines}(\FR^d) &\; \xlongrightarrow{\sim}\;  T\FR^d \cong \FR^d\times \FR^d 
\\[-2pt]
\gamma = v \cdot t + c &\; \longmapsto \; (c , v) \, , \nn 
\end{align*}

\vspace{-2mm}
\noindent
in which case the notion of a  tangent bundle on the initial data smooth set, and also the pushforward along the isomorphism, is transparent. 
Explicitly, let $\{x^a,p^a\}$ be the canonical coordinates on $T\FR^d$, then the pushforward acts on tangent vectors by
\vspace{-2mm}
\begin{align*}
(\mathrm{Cau}_{\{0\}})_{* \gamma} \;:\; T_{\gamma} \mathrm{Lines}(\FR^d) \cong \mathrm{Lines}(\FR^d)  &\; \xlongrightarrow{\sim}\;  T_{\mathrm{Cau}_{\{0\}}(\gamma)}(T\FR^d) 
\\[-2pt]
b = w \cdot t + d &\; \longmapsto \; b^a(0) \cdot \frac{\partial}{\partial x^a} + \partial_t b^a(0) \cdot \frac{\partial}{\partial p^a} \, , \nn 
\end{align*}

\vspace{-2mm}
\noindent
which for $b=w\cdot t +d $ is simply $d^a \cdot \frac{\partial}{\partial x^a} + w^a \cdot \frac{\partial}{\partial p^a}$\, . 
Thus, pulling back the covariant symplectic form $\Omega_\CL$ by the inverse of $(\mathrm{Cau}_{\{0\}})_*$, we get the induced 
symplectic form $\Omega_\CL^{\mathrm{Cau}}$ on $T\FR^d$ given by 
\vspace{-2mm}
\begin{align*}
\Omega_\CL^{\mathrm{Cau}} \;\; :\;\; T(T\FR^d)\times_{T\FR^d} T(T\FR^d) &\; \xlongrightarrow{}\;  \FR \\
\bigg(d_1^a \cdot \frac{\partial}{\partial x^a} + w_1^a \cdot \frac{\partial}{\partial p^a}, \, d_2^a \cdot \frac{\partial}{\partial x^a}
+ w_2^a \cdot \frac{\partial}{\partial p^a}\bigg) &\; \longmapsto \;\; d_1^a \cdot w_{2 \, a} - d_2^a \cdot w_{1\, a} \, .
\end{align*}

\vspace{-2mm}
\noindent Thus, in terms of coordinates, the non-covariant phase space of the free particle, associated with the time instant 
$\{0\}\hookrightarrow \FR^1_t$, is the symplectic manifold 
\vspace{0mm}
$$
\big(T\FR^d, \, dx^a \wedge \dd p_a \big)\, ,
$$

\vspace{-2mm}
\noindent which, under the isomorphism $T\FR^d \xrightarrow{\sim} T^*\FR^d$ induced by the Euclidean metric, is precisely 
the canonical cotangent bundle symplectic manifold. 
\end{example}

\paragraph{Covariant Poisson brackets} 
We conclude this first part of the series by detailing the transgressed version of the Hamiltonian current condition, and hence defining an (off-shell) 
Poisson bracket on the algebra of Hamiltonian functionals on $\CF$ corresponding to the (chosen) transgressed symplectic $2$-form $\Omega_{\CL,\Sigma^{d-1}}$. 
For spacetimes of the form $N\times \FR$, and modulo Rem. \ref{OnShellCartanCalculusCaveats}, this same story applies to the case of on-shell Hamiltonian 
functionals and hence defining a canonical Poisson bracket on the covariant phase space $(\CF_{\CE \CL}, \Omega_{\CL})$, corresponding to the 
canonical on-shell transgressed presymplectic $2$-form.

\begin{definition}[\bf  Hamiltonian functionals]\label{HamiltonianFunctionals}
A transgressed local functional
\vspace{-2mm}
$$ \CH_{\CZ,\Sigma^{d-1}}= \tau_{\Sigma^{d-1}} (\CH_{\CZ}) \equiv \int_{\Sigma^{d-1}} \CH_{\CZ} \quad \in  \quad C^\infty_\loc(\CF)
$$ 

\vspace{-2mm}
\noindent
is \textit{Hamiltonian} for the presymplectic smooth set $(\CF, \Omega_{\CL,\Sigma^{d-1}})$, corresponding to a compact oriented submanifold 
$\Sigma^{d-1}\hookrightarrow M$, if there exists a local vector field $\CZ\in \CX_\loc(\CF)$ such that
\vspace{-2mm}
$$\iota_{\CZ} \Omega_{\CL,\Sigma^{d-1}} + \delta \CH_{\CZ, \Sigma^{d-1}} = 0 \, . 
$$

\vspace{-2mm}
\noindent The minimal subalgebra generated by transgressed Hamiltonian functionals over $\Sigma^{d-1}$ is the \textit{algebra of
Hamiltonian functionals} over $\Sigma^{d-1}$ 
\vspace{-1mm}
$$
\mathrm{HamAlg}(\CF, \Omega_{\CL,\Sigma^{d-1}}) \longhookrightarrow C^\infty_\loc(\CF)\, . 
$$

\vspace{-2mm}
\end{definition}
Strictly speaking, we should be speaking of transgressed \textit{Hamiltonian (functional) pairs} $(\CH_\CZ, \CZ)\in C^\infty_\loc(\CF) \times \CX_\loc (\CF)$, 
for the same reasons as in the case of Hamiltonian vector fields and corresponding currents. We will suppress this point in favor of easing the heavy notation.

\begin{remark}[\bf Hamiltonian currents vs. Hamiltonian functionals]\label{HamCurrentsVsHamFunctionals}
$\,$

\noindent {\bf (i)} Assuming $\partial \Sigma^{d-1}=\emptyset$, then all Hamiltonian currents $\CH_{\CZ}\in \Omega^{d-1,0}_\loc(\CF\times M)$ 
transgress along $\Sigma^{d-1}$ to yield examples of Hamiltonian functionals, since the Hamiltonian current condition \eqref{HamiltonianCurrentCondition} 
transgresses to the Hamiltonian functional condition. 

\vspace{.5mm} 
\noindent {\bf (ii)}  If $\Sigma^{d-1}$ has a boundary, then only the subset of strict Hamiltonian currents transgresses to Hamiltonian functionals
along $\Sigma^{d-1}$. However, these do not exhaust all transgressed Hamiltonian functionals. In particular, for the case of no boundary, the 
Hamiltonian functional condition of Def. \ref{HamiltonianFunctionals} demands the corresponding current condition\footnote{This statement is
the converse to Stoke's Theorem (see \cite[\S 7.3]{Lafontaine}).  
The exact form on the right-hand side is possibly non-local in the field space.}, i.e.,  
$\iota_\CZ \omega_\CL |_{\Sigma^{d-1}} + \delta \CH_{\CZ} |_{\Sigma^{d-1}}= \dd_{\Sigma^{d-1}}(...)$ being exact, only when restricted to the 
submanifold $\Sigma^{d-1}$. 

\vspace{.5mm} 
\noindent {\bf (iii)}  Crucially, however, this does not imply the Hamiltonian current condition over $M$, i.e.,
$\iota_\CZ \omega_\CL + \delta \CH_{\CZ}$ being $\dd_M$-exact on the full spacetime $M$. 
 An explicit such case is considered in Ex. \ref{Ex-last}.
\end{remark}

The algebra of Hamiltonian functionals over $\Sigma^{d-1}$ is comprised of products of transgressed local functionals
$\CH_{\CZ,\Sigma^{d-1}}$ over $\Sigma^{d-1}$, each of which satisfies the Hamiltonian condition for some local vector field $\CZ$. 
Of course, any Hamiltonian functional satisfies the Hamiltonian condition, however potentially via a smooth but \textit{non-local} 
(functional) vector field.

\begin{lemma}[{\bf Hamiltonian functionals are Hamiltonian}]\label{HamiltonianFunctionalsAreHamiltonian}
Any Hamiltonian functional $\bar{\CH}_{\Sigma^{d-1}}\in \mathrm{HamAlg}(\CF,  \Omega_{\CL,\Sigma^{d-1}}) $ given by
\vspace{-2mm}
$$
\bar{\CH}_{\Sigma^{d-1}}= \CH_{\CZ_1,\Sigma^{d-1}}\cdot  \ldots \cdot  \CH_{\CZ_n,\Sigma^{d-1}}
= \tau_{\Sigma^{d-1}} (\CH_{\CZ_1})\cdot \ldots \cdot  \tau_{\Sigma^{d-1}} (\CH_{\CZ_n}) \quad : \quad \CF \longrightarrow \FR 
$$

\vspace{0mm}
\noindent
satisfies the Hamiltonian condition
\vspace{-1mm}
$$\iota_\CZ \Omega_{\CL,\Sigma^{d-1}} + \delta \hat{\CH}_{\Sigma^{d-1}} = 0 \, , 
$$
\vspace{-2mm}
\noindent for the smooth vector field $\CZ:= \sum_{i=0}^{n}  \big( \CH_{\CZ_1,\Sigma^{d-1}}\cdots 
 \hat{\CH}_{\CZ_i,\Sigma^{d-1}}\cdots  \CH_{\CZ_n,\Sigma^{d-1}}\big) \cdot \CZ_{i} \, \in \, \CX(\CF) $ where $(\hat{-})$ means omit.
\end{lemma}

\newpage 
\begin{proof}
Recall that the algebra of Hamiltonian functionals is generated by transgressions of particular local $(d-1,0)$-form currents, and hence the
differential $\delta: \mathrm{HamAlg}(\CF, \Omega_{\CL,\Sigma^{d-1}})  \rightarrow \Omega^{1}(\CF)$ makes sense as the unique derivation extension
of that on charges
(Lem. \ref{TransgressedVerticalDifferential}).  We prove the statement for $n=2$, with the general case following by induction. 
Consider two transgressed Hamiltonian functionals $ \CH_{\CZ_1,\Sigma^{d-1}} , \,  \CH_{\CZ_2,\Sigma^{d-1}}$, and their product Hamiltonian functional 
 \vspace{-2mm}
$$
\CH_{\CZ_1,\Sigma^{d-1}} \cdot  \CH_{\CZ_2,\Sigma^{d-1}} \quad \in \quad C^\infty_\loc (\CF) \, . 
$$

\vspace{-2mm} 
\noindent By the derivation property of the vertical differential, it immediately follows that 
\vspace{-2mm}
\begin{align*}
\delta ( \CH_{\CZ_1,\Sigma^{d-1}} \cdot  \CH_{\CZ_2,\Sigma^{d-1}} ) 
&= \delta \CH_{\CZ_1,\Sigma^{d-1}} \cdot  \CH_{\CZ_2,\Sigma^{d-1}} +\CH_{\CZ_1,\Sigma^{d-1}} \cdot  \delta  \CH_{\CZ_2,\Sigma^{d-1}}  \\ 
&= - \iota_{\CZ_1} \Omega_{\CL , \Sigma^{d-1}} \cdot \CH_{\CZ_2,\Sigma^{d-1}}- \CH_{\CZ_1,\Sigma^{d-1}} \cdot  \iota_{\CZ_2} \Omega_{\CL, \Sigma^{d-1} }\\ 
&= - \iota_{\big(\CH_{\CZ_2,\Sigma^{d-1}}\cdot \CZ_1 + \CH_{\CZ_1,\Sigma^{d-1}} \cdot \CZ_2 \big) } \Omega_{\CL, \Sigma^{d-1}} \, .
\end{align*}

\vspace{-7mm} 
\end{proof}

Having identified the algebra of (smooth) Hamiltonian functionals of the presymplectic smooth set $(\CF, \Omega_{\CL,\Sigma^{d-1}})$,
the definition of the corresponding Poisson bracket follows exactly as in finite-dimensional (pre)symplectic geometry.
\begin{definition}[\bf  Off-shell Poisson bracket]\label{OffShellPoissonBracket}
Let $\Sigma^{d-1}$ be a compact oriented submanifold. The \textit{Poisson bracket} of Hamiltonian functionals
\vspace{-2mm}
\begin{align*}
\{-,-\}_{\Sigma^{d-1}} \;:\; \mathrm{HamAlg}(\CF, \Omega_{\CL,\Sigma^{d-1}}) \times \mathrm{HamAlg}(\CF, \Omega_{\CL,\Sigma^{d-1}}) 
\; \longrightarrow \;
\mathrm{HamAlg}(\CF, \Omega_{\CL,\Sigma^{d-1}})
\end{align*}

\vspace{-2mm}
\noindent
of the presymplectic smooth set $(\CF,\Omega_{\CL, \Sigma^{d-1}})$ is defined by 
\vspace{-2mm}
\begin{align*}
\big\{\CH_{\CZ_1,\Sigma^{d-1}}\, ,\, \CH_{\CZ_2, \Sigma^{d-1}}\big\}_{\Sigma^{d-1}} 
:= -\iota_{\CZ_1} \iota_{\CZ_2} \Omega_{\CL,\Sigma^{d-1}} 
\end{align*}

\vspace{-2mm}
\noindent for any two transgressed Hamiltonian functionals, and extended as a bi-derivation to the full algebra of Hamiltonian functionals.  
\end{definition}
It is not hard to see on arbitrary Hamiltonian functionals, given by products of transgressed Hamiltonian functionals,
the Poisson bracket coincides with the contraction of the presymplectic 2-form with the corresponding non-local
vector fields from Lem. \ref{HamiltonianFunctionalsAreHamiltonian}. 

\smallskip 
It remains to show that the bracket is well-defined, i.e., indeed independent of the representative currents and vector fields, 
and further mapping into Hamiltonian functionals, rather than arbitrary smooth functionals. Moreover, it remains to justify 
the name \textit{Poisson}, i.e., that it does actually satisfy the Jacobi identity.

\begin{lemma}[{\bf Poisson bracket is Lie}]\label{PoissonBracketIsLie}
The bracket of Def. \ref{OffShellPoissonBracket} is well-defined and satisfies the Jacobi identity, turning the algebra of Hamiltonian 
functionals into a Poisson algebra
\vspace{-2mm}
$$
\Big(\mathrm{HamAlg}\big(\CF, \Omega_{\CL,\Sigma^{d-1}}\big), \, \{-,-\}_{\Sigma^{d-1}}\Big ) \, .
$$

\vspace{-2mm}
\end{lemma}
\begin{proof}
We show the statement for transgressed Hamiltonian functionals, with the general case following by the biderivation property.
To see that the bracket is independent of the chosen representative currents, recall by Lem. \ref{TransgressedVerticalDifferential}
that for any two local currents  
$\CH_{\CZ_2}, \, \hat{\CH}_{\hat{\CZ}_2}$ such that their transgression/integral along $\Sigma^{d-1}$ agrees 
$\CH_{\CZ_2, \Sigma^{d-1}} =  \hat{\CH}_{\hat{\CZ}_2, \Sigma^{d-1}}$,
it follows that their differentials also agree $\delta \CH_{\CZ_2, \Sigma^{d-1}} =  \delta \hat{\CH}_{\hat{\CZ}_2, \Sigma^{d-1}}$. 
Assuming now the transgressed functionals are both Hamiltonian with respect to the vector fields $\CZ_2, \hat{\CZ}_2$,
respectively, it follows that
\vspace{-2mm}
\begin{align*}
\big\{\CH_{\CZ_1,\Sigma^{d-1}} \, ,\, \hat{\CH}_{\hat{\CZ}_2, \Sigma^{d-1}}\big\}_{\Sigma^{d-1}} 
:&= -\iota_{\CZ_1} \iota_{\hat{\CZ}_2} \Omega_{\CL,\Sigma^{d-1}} = \iota_{\CZ_1} \delta \hat{\CH}_{\hat{\CZ_2},\Sigma^{d-1}} \\
&= \iota_{\CZ_1} \delta \CH_{\CZ_2,\Sigma^{d-1}} =-\iota_{\CZ_1} \iota_{\CZ_2} \Omega_{\CL,\Sigma^{d-1}} \\
&= \big\{\CH_{\CZ_1,\Sigma^{d-1}} \,,\, \CH_{\CZ_2, \Sigma^{d-1}}\big\}_{\Sigma^{d-1}}\, .
\end{align*}

\vspace{-2mm}
\noindent
The fact that the image of the bracket sits inside Hamiltonian functionals follows by applying the local Cartan calculus 
under the integral, giving
\vspace{-1mm}
$$
\delta\, \big\{\CH_{\CZ_1,\Sigma^{d-1}}\, ,\, \CH_{\CZ_2, \Sigma^{d-1}}\big\}_{\Sigma^{d-1}}
= \delta \iota_{\CZ_1} \iota_{\CZ_2} \Omega_{\CL,\Sigma^{d-1}}= \iota_{[\CZ_1,\CZ_2]} \Omega_{\CL,\Sigma^{d-1}} \, .
$$

\vspace{0mm}
\noindent Finally, the Jacobi identity follows from the same calculation as that of Hamiltonian currents \eqref{HamiltonianBracketJacobiator} 
(where the Jacobi identity however fails). That is, for any three transgressed Hamiltonian functionals 
$\CH_{\CZ_1, \Sigma^{d-1}},\,  \CH_{\CZ_2, \Sigma^{d-1}},\,  \CH_{\CZ_3, \Sigma^{d-1}}$, we have the trivial vanishing 
$\iota_{\CZ_1}\iota_{\CZ_2} \iota_{\CZ_3} \Omega_{\CL,\Sigma^{d-1}}$ by degree reasons. Applying the (transgressed) 
vertical differential (Lem. \ref{TransgressedVerticalDifferential}) on both sides and computing using the local Cartan
calculus under the integral yields 
\vspace{-2mm}
$$
0 =  \iota_{\CZ_1} \iota_{[\CZ_2,\CZ_3]}\Omega_{\CL,\Sigma^{d-1}} + \mathrm{cyc}(1,2,3) \, ,
$$

\vspace{-2mm}
\noindent which yields exactly the Jacobi identity for the Poisson bracket via Def. \ref{OffShellPoissonBracket}.
\end{proof}
Even though (transgressed) Hamiltonian functionals are not exhausted by transgressions of Hamiltonian currents
(Rem. \ref{HamCurrentsVsHamFunctionals}), the local `Poisson' Lie algebras of Hamiltonian pairs of 
Lem. \ref{PoissonBracketIsLie} and  \eqref{LieAlgebraOfStrictHamiltonianPairs} do map into 
the Poisson Lie algebra of Hamiltonian functionals upon transgression, providing plenty of particular instances 
of Hamiltonian functionals and their Poisson brackets.

\begin{corollary}[\bf  Brackets of Hamiltonian currents and functionals] \label{PoissonVsCurrentBrackets}
Let $\Sigma^{d-1}\hookrightarrow M$ be a compact oriented submanifold.  

\noindent {\bf (i)} If $\Sigma^{d-1}$ is without boundary, then  transgression $\tau_{\Sigma^{d-1}}: \mathrm{HamPrs}(\CF\times M, \om_\CL)\rightarrow 
\mathrm{HamAlg}(\CF, \Omega_{\CL,\Sigma^{d-1}})$ intertwines the bracket of Hamiltonian pairs (Def. \ref{BracketOfHamiltonianPairs})
and the corresponding Poisson bracket 
\vspace{-1mm}
\begin{align*}
\tau_{\Sigma^{d-1}}\big(  \{\CH_{\CZ_1}, \,  \CH_{\CZ_2}\}_H \big) = \big\{ \CH_{\CZ_1,\Sigma^{d-1}}, \, \CH_{\CZ_2, \Sigma^{d-1}} \big\}   
\quad : \quad \CF \longrightarrow \FR \;.
\end{align*}

\vspace{-1mm}
\noindent
In particular, the transgression over $\Sigma^{d-1}$ maps the central
extension local Poisson Lie algebra of Lem. \ref{LocalPoissonLieAlgebraOfHamiltonianCurrents} into the Poisson Lie algebra of Hamiltonian functionals 
over $\Sigma^{d-1}$ of Lem \ref{PoissonBracketIsLie}.

\noindent {\bf (ii)} If $\Sigma^{d-1}$ is with boundary, then the same holds for the subset transgressed strict Hamiltonian currents. In particular,
the transgression over $\Sigma^{d-1}$ maps the Lie algebra of strict Hamiltonian Pairs of  \eqref{LieAlgebraOfStrictHamiltonianPairs} 
into the Poisson Lie algebra of Hamiltonian functionals over $\Sigma^{d-1}$. 
\end{corollary}
\begin{proof}
This follows immediately since the Poisson bracket of Hamiltonian functionals is defined via $\Omega_{\CL,\Sigma^{d-1}}$ which is the transgression 
of $\om_\CL$, while the bracket of Hamiltonian currents is defined via the latter presymplectic current. Explicitly, for a submanifold without boundary, 
the Hamiltonian current condition (Eq. \eqref{HamiltonianCurrentCondition}) transgresses to the Hamiltonian functional condition
(Def. \ref{HamiltonianFunctionals}) since the $\dd_M-$exact terms vanish upon integration. Furthermore, 
\vspace{-2mm}
\begin{align*}
 \tau_{\Sigma^{d-1}}\big(  \{\CH_{\CZ_1}, \,  \CH_{\CZ_2}\}_H \big) &= -\int_{\Sigma^{d-1}} \iota_{\CZ_1} \iota_{\CZ_2} \om_\CL  
 =- \iota_{\CZ_1} \iota_{\CZ_2} \int_{\Sigma^{d-1}} \om_\CL \\
 &= - \iota_{\CZ_1} \iota_{\CZ_2} \Omega_{\CL,\Sigma^{d-1}} =  \big\{ \CH_{\CZ_1,\Sigma^{d-1}}, \, \CH_{\CZ_2, \Sigma^{d-1}} \big\} 
\end{align*}

\vspace{-2mm}
\noindent The case of transgression over manifolds with boundary follows identically, since no $\dd_M$-exact terms appear in the strict Hamiltonian current condition.
\end{proof} 

\begin{remark}[\bf Poisson algebra of the covariant phase space]\label{PoissonAlgebraOfTheCovariantPhaseSpace}
Consider the case of a product space-time $M=N\times \FR$, whereby there exists a canonical on-shell presymplectic 2-form making up 
the covariant phase space  (Def. \ref{CovariantPhaseSpace})
\vspace{-2mm}
$$(\CF_{\CE \CL}, \Omega_\CL)\, .$$

\vspace{-2mm}
\noindent Assuming the requirements of Rem. \ref{OnShellCartanCalculusCaveats} are met, so that the local Cartan calculus descends to $\CF_{\CE \CL}\times M$ 
and hence further transgresses to $\CF_{\CE \CL}$ as in Lem. \ref{TransgressedVerticalDifferential}, then the above discussion follows verbatim for the 
case of on-shell Hamiltonian functionals of $(\CF_{\CE \CL}, \Omega_\CL)$, and the canonical Poisson algebra of functionals they generate, determined
by the canonical presymplectic $2$-form. In more detail, these are generated by on-shell transgressed functionals 
$\CH_{\CZ,N} |_{\CE \CL} \in C^\infty_\loc(\CF_{\CE \CL})$ that satisfy the on-shell Hamiltonian condition
$\iota_{\CZ} \Omega_{\CL}|_{\CE \CL} + \delta \CH_{\CZ, N}|_{\CE \CL} = 0$ for some on-shell vector field 
$\CZ|_{\CE \CL} \in \CX_{\loc}(\CF_{\CE \CL})$. Modulo Rem. \ref{OnShellCartanCalculusCaveats}, these are represented by off-shell functionals
$\CH_{\CZ,N}\in C^\infty(\CF)$ that satisfy the off-shell Hamiltonian functional condition, up to the transgression a differential operator 
applied to the Euler--Lagrange form 
\vspace{-2mm}
$$
\iota_{\CZ} \Omega_{\CL} + \delta \CH_{\CZ, N} = \int_{N} \CJ^I \partial_I \CE \CL \, .
$$

\vspace{-2mm}
\end{remark}

\begin{example}[\bf Free particle covariant Poisson algebra]\label{Ex-last} Along the lines of Ex. \ref{Ex-free-cov},
recall the free particle description of the off-shell presymplectic current 
\vspace{-2mm} 
$$
\om_\CL = \delta \gamma^a \wedge \delta (\partial_t \gamma_a) \quad \in \;\; \Omega_{\loc}^{0,2}\big(\mathbold{P}(\FR^d) \times \FR^1_t\big)\,.
$$

\vspace{-2mm} 
\noindent and induced on-shell symplectic $2$-form 
\vspace{-2mm} 
$$ 
\Omega_{\CL}= \delta \gamma^a \wedge \delta (\partial_t \gamma_a) \big\vert_{t=t_0} \quad \in \;\; \Omega^2_\loc\big(\mathrm{Lines}(\FR^d)\big) \, ,
$$

\vspace{-2mm} 
\noindent
which is independent of the chosen $\{t_0\}\hookrightarrow \FR^1_t$. Using the identification of the on-shell field space with the symplectic manifold 
$(T\FR^d, \, \dd x^a\wedge \dd p_a)$, via the diffeomorphism described in Ex. \ref{Ex-free-cov}, and the straightforward description of Hamiltonian
functions on the latter, it is often claimed that the on-shell functionals 
\vspace{-2mm}
$$ \Gamma^a(t_0):= \gamma^a |_{t=t_0}\quad , \quad  \dot{\Gamma}_a(t_0):= \partial_t \gamma_a |_{t=t_0} 
\quad  \in \quad  C^\infty_\loc \big(\mathrm{Lines}(\FR^d)\big)
$$

\vspace{-2mm}
\noindent 
are both Hamiltonian (conjugate variables) for any $t_0\in \FR^1_t$, with Hamiltonian vector fields  
\vspace{-2mm} 
$$ 
\CZ_{\Gamma^a(t_0)}\overset{?}{=} \frac{\delta}{\delta(\partial_t \gamma^a)} \quad , \quad  \CZ_{\dot{\Gamma}_a(t_0)}
:= - \frac{\delta}{\delta \gamma^a} \quad  \in \quad  \CX_\loc \big(\mathrm{Lines}(\FR^d)\big) \, ,
$$ 
respectively. The question mark on the former is intentional since this symbol does \textit{not} define a local vector field on the field space.
Of course, this is the naive formula substitution of the canonical Hamiltonian vector fields corresponding to the non-covariant Hamiltonian
conjugate variables $x^a, p^a \in C^\infty(T\FR^d)$ on the non-covariant phase space at $\{0\}\hookrightarrow \FR^1_t$. 
We now make precise what this is  actually supposed to mean, thus justifying its formal manipulation.

Firstly, we note the contraction of the (off-shell) symplectic current $\om_\CL$ with the latter well-defined local vector field
$-\frac{\delta}{\delta \gamma^a}$ gives 
\vspace{-4mm}
\begin{align*}
\iota_{\CZ_{\dot{\Gamma}_a(t_0)}} \om_\CL &= \iota_{-\frac{\delta}{\delta \gamma^a}} \big( \delta \gamma^b \wedge \delta (\partial_t \gamma_b)\big)\\
&= - \delta^{a}_{\, b}\cdot  \delta(\partial_t \gamma_b) + \delta \gamma^a \cdot \partial_t (\delta^{a}_{\,b}) = - \delta(\partial_t \gamma_a) \, ,    
\end{align*}

\vspace{-2mm}
\noindent and so $\partial_t \gamma_a\in \Omega^{0,0}\big(\mathbold{P}(\FR^d)\times \FR^1_t\big)$ is an off-shell Hamiltonian current
(Def. \ref{HamiltonianCurrentCondition}). This immediately transgresses to the off-shell Hamiltonian functional 
$\partial_t \gamma_a |_{t=t_0}\in C^\infty_\loc\big(\mathbold{P}(\FR^d)\big)$. Since $-\frac{\delta}{\delta \gamma^a}$ is tangent to the 
on-shell fields $\mathrm{Lines}(\FR^d)$, this indeed restricts to give the \textit{Hamiltonian} functional
\vspace{-2mm}
$$
\dot{\Gamma}_a(t_0):= \partial_t \gamma_a |_{t=t_0} \quad  \in \quad  C^\infty_\loc \big(\mathrm{Lines}(\FR^d)\big)
$$

\vspace{-2mm}
\noindent which corresponds to the observable that measures the velocity (momentum) of the particle moving on a line at $t=t_0$. 

Next, consider the well-defined off-shell local vector field given by 
\vspace{-2mm}
$$
t\cdot \frac{\delta}{\delta \gamma_a} \;\; \in \CX_\loc\big(\mathbold{P}(\FR^d)\big) \, 
$$

\vspace{-2mm}
\noindent 
which is the vector field $\mathbold{P}(\FR^d)\longrightarrow T \mathbold{P}(\FR^d)$ given by $\gamma \rightarrow (\gamma, t \cdot e_a)$.
In other words, it is the (constant) vector field on the space of paths that assigns (to any path $\gamma$) the tangent vector path 
$b= 1\cdot t \cdot e_a + 0$ of constant speed $v=1$ in the direction $e_a\in \FR^d$, passing through the origin. Since tangent vectors 
to the on-shell fields $\mathrm{Lines}(\FR^d)$ are also lines, it follows that this vector field does restrict to an on-shell vector field. 
Its contraction with the off-shell presymplectic current gives 
\vspace{-2mm}
\begin{align*}
\iota_{t\cdot \frac{\delta}{\delta \gamma_a}} \om_\CL 
&= t \cdot \delta^{a}_{\, b}\cdot  \delta(\partial_t \gamma_b)
+ \delta \gamma_a \cdot \partial_t (t \cdot \delta^{a}_{\,b})\\
&= t\cdot \delta(\partial_t \gamma_a) - \delta \gamma^a  \, ,    
\end{align*}

\vspace{-2mm}
\noindent which, however, still does \textit{not} exhibit $\delta \gamma^a$ as a Hamiltonian current. Nevertheless, the 
transgression along $t=0$ does exhibit 
\vspace{-2mm} 
$$
\gamma^a |_{t=0} \quad \in \quad C^\infty_\loc\big(\mathbold{P}(\FR^d)\big)
$$ 

\vspace{-2mm} 
\noindent
as an \textit{off-shell} Hamiltonian functional for $\Omega_{\CL,0}\in \Omega^2_\loc\big(\mathbold{P}(\FR^d)\big)$, which restricts to 
an on-shell Hamiltonian  functional for the canonical symplectic $2$-form, i.e.,
\vspace{-2mm}
\begin{align*}
\iota_{t\cdot \frac{\delta}{\delta \gamma_a}} \Omega_{\CL} = \big( t\cdot \delta(\partial_t \gamma_a) - \delta \gamma^a\big)|_{t=0}
= - \delta \gamma^a |_{t=0}\, .
\end{align*}

\vspace{-2mm} 
\noindent Consequently,  
$$\Gamma^a(0)= \gamma^a |_{t=0} \quad \in \quad C^\infty_\loc\big(\mathrm{Lines}(\FR^d)\big)
$$ 
is indeed an on-shell Hamiltonian functional for the on-shell vector field 
\vspace{-2mm} 
$$ 
\CZ_{\Gamma^a(0)}:= t\cdot \frac{\delta}{\delta \gamma^a} \, .
$$

\vspace{-2mm} 
\noindent
Analogously
$$\Gamma^a(t_0)= \gamma^a |_{t=t_0}
$$ 
is Hamiltonian for any $t_0$, with corresponding vector field
\vspace{-3mm} 
$$ 
\CZ_{\Gamma^a(t_0)}:= (t-t_0)\cdot \frac{\delta}{\delta \gamma^a} \, .
$$

\vspace{-2mm}
\noindent Consequently, the induced Poisson bracket of the two Hamiltonian functionals is given by
\vspace{-2mm}
\begin{align*}
\big\{\Gamma^a(t_0)\, ,\, \dot{\Gamma}_b(t_0) \big\}  
: &= -\iota_{\CZ_{\Gamma^a(t_0)}} \iota_{\CZ_{\dot{\Gamma}_a(t_0)}} \Omega_{\CL} \\
&=   +\iota_{\CZ_{\Gamma^a(t_0)}}  \delta(\partial_t \gamma_a) \big|_{t=t_0} 
= \delta^a_{\, b}
\end{align*}

\vspace{-2mm}
\noindent as expected by the non-covariant Poisson algebra structure. Since any local functional with higher jet dependence than $1$ vanishes 
on the space of on-shell fields, it follows that the
Poisson algebra of Hamiltonian functionals is generated by these two (families) of transgressed functionals. 
We leave it to the reader to verify that, under the Cauchy isomorphism of Ex. \ref{Ex-free-cov}, the local vector field 
$t\frac{\delta}{\delta \gamma^a}$ maps to the canonical Hamiltonian vector field $\frac{\partial}{\partial p_a}$ on $T\FR^d$, 
which exhibits $x^a$ as a Hamiltonian function, conjugate to $p_a$. This shows that, even though the notation used is mathematically ill-defined, 
the resulting statement of the position and velocity observables being Hamiltonian functionals and `conjugate' is nevertheless correct. 

\smallskip 
A similar story can be repeated in the case of the ${\rm O}(n)$-model's (e.g. scalar field) and electromagnetism's conjugate variables
viewed as covariant phase space Hamiltonian functionals, and their corresponding Hamiltonian vector fields.

\end{example}

\newpage 

\addtocontents{toc}{\protect\vspace{-10pt}}
\section{Outlook}
\label{outlook} 

In conclusion, we have shown that smooth sets constitute an excellent context for formulating classical field theory 
in its vanilla form for plain bosonic fields. In follow-up articles in 
this series \cite{GSS-2}\cite{GSS-3}, we will generalize smooth sets in various natural ways enhancing this correspondence to 

{\bf (a)} infinitesimal structure / perturbative field theory,

{\bf (b)} fermionic field theory, and 

{\bf (c)} higher gauge theory by invoking, in turn, synthetic differential supergeometry and higher topos theory.

\medskip 
\noindent We briefly outline some of the basic ideas (cf. 
\cite{FSS14}\cite{Alfonsi}\cite{Schreiber24}\cite{GSS-3}):

\medskip 
\noindent
\begin{itemize}[leftmargin=.4cm]
\item 
{\bf Synthetic geometry and infinitesimal structure.}
A remarkable (if maybe underappreciated) fact of differential geometry says that passage from smooth manifolds to their ordinary 
real algebras of smooth functions is a full embedding into the opposite category $\mathrm{CAlg}_{{}_{\FR}}$ of commutative real algebras 
(``Milnor's exercise'', cf. \cite[\S 35.8-10]{KMS}):
\vspace{-3mm} 
\begin{equation}
  \label{EmbeddingOfSmoothManifoldsIntoCAlgOp}
  \begin{tikzcd}[
    row sep=-3pt,
    column sep=20pt
  ]
    \SmoothManifolds
    \ar[
      rr,
      hook,
      "{
        C^\infty(-)
      }"
    ]
    &&
    \mathrm{CAlg}_{{}_{\FR}}^{\mathrm{op}}
    \\
    M 
    \ar[
      rr,
      phantom,
      "{
        \longmapsto
      }"
    ]
    &&
    C^\infty(M)
  \end{tikzcd}
\end{equation}

\vspace{-2mm} 
\noindent This means that smooth manifolds behave more like affine schemes in algebraic geometry than the usual definitions may suggest, 
thus providing access to analogous constructions. In particular, we may regard any full subcategory of $\mathrm{CAlg}_{{}_{\FR}}$
larger than the image of \eqref{EmbeddingOfSmoothManifoldsIntoCAlgOp} as {\it being} a category of ``generalized manifolds'' which 
do not exist as point-set models but (only) through their would-be algebras of functions.

\medskip 
Concretely, since the algebra of smooth functions on the $r$-th order {\it infinitesimal neighborhood} of the origin in $\FR^m$ 
should be the real polynomials in $m$ variables modulo the relation that any $r+1$st power of them vanishes, we may take this to 
be the dual definition of the {\it infinitesimal halo} $\DD^m_r \hookrightarrow \FR^m$ around the origin and take the category 
of {\it infinitesimally thickened} Cartesian spaces to be the full subcategory on the objects of the following form:
\vspace{-2mm} 
\[
  \begin{tikzcd}[
    row sep=-4pt
  ]
    \ThickenedCartesianSpaces
    \ar[
      rr,
      hook,
      "{
        C^\infty(-)
      }"
    ]
    &&
    \mathrm{CAlg}_{{}_{\FR}}^{\mathrm{op}}
    \\
    \FR^k
    \times
    \DD^m_r
    \ar[
      rr,
      phantom,
      "{ \longmapsto }"
    ]
    &&
    C^\infty(\FR^k)
    \otimes_{{}_{\FR}}
    \FR[\epsilon_1, ... \epsilon_m]
    /
    (\epsilon^{r+1})
    \,.
  \end{tikzcd}
\]

\vspace{-3mm} 
\noindent This category becomes a site by declaring the open covers to be of the form $\big\{ U_i \times \DD^m_r \xhookrightarrow{ \iota_i \times \mathrm{id} } \FR^k \times \DD^m_r \big\}_{i \in I}$ for the $\iota_i$ constituing a (differentiably) good open cover of ordinary Cartesian spaces as before. The resulting sheaf
topos of {\it (infinitesimally) thickened smooth sets}
\[
  \FormalSmoothSets
  \;:=\;
  \mathrm{Sh}(\ThickenedCartesianSpaces)
\]
is known as the {\it Cahiers topos}, a ``well-adapted model'' for ``synthetic differential geometry''. Here smooth manifolds (and generally smooth sets) 
co-exist with actual infinitesimal objects such as the first-order infinitesimal interval $\DD^1_1$. Note that the latter contains just a single point
\vspace{-2mm} 
\[
  \ast \simeq \FR^0
  \xhookrightarrow[  \exists ! ]{\quad \iota \quad}
  \DD^1_1
\]

\vspace{-1mm} 
\noindent and yet is larger than that point (receives more plots, namely from other infinitesimal objects).

For instance, the tangent bundle $TM \xrightarrow{p} M$ of a smooth manifold $M$ (itself a smooth manifold) appears now literally as the space 
of images of the infinitesimal interval in $M$, in that 
\vspace{-3mm} 
$$
  \begin{tikzcd}[row sep=12pt, column sep=5pt]
    {[\DD^1_1, M]}
    \ar[r, phantom, "{ \simeq }"]
    \ar[
      d,
      "{
        [\iota, \DD^1_1] \;\;
      }"{swap}
    ]
    &
    T M
    \ar[d, "{ p }"]
    \\
    {[\ast, M]}
    \ar[r, phantom, "{ \simeq }"]
    &
    M
    \mathrlap{\,.}
  \end{tikzcd}
$$

\vspace{-2mm} 
\noindent 
Proceeding along these lines, one finds a convenient and powerful realization of plenty of field-theoretic concepts in $\FormalSmoothSets$ (cf. \cite{KS17}).
In particular, \textit{every} tangent bundle appearing in this first part of the series is in fact recovered as an example of the synthetic 
tangent bundle construction, with the target being instead the field space $\CF$, a subspace thereof, or the infinite jet bundle $J^\infty_M F$ 
(viewed as thickened smooth sets). Similarly, the set of infinite jets of the field bundle $F\rightarrow M$ arises as sections over infinitesimal
neighborhoods $\DD_p \hookrightarrow M$ of the base and, in fact, 
the infinite jet bundle as a (thickened) smooth set arises via a natural categorical construction. Moreover, it should be the case that the perturbation theory of a local field theory $(\CF,\CL)$ around a point $\phi_0 \in \CF$ in field space may be rigorously defined and described as the \textit{actual} restriction of the Lagrangian to the infinitesimal neighborhood   $\DD_{\phi_0} \hookrightarrow \CF$ of the chosen field.
\\


\newpage 
\vspace{-2mm} 
\item
{\bf Supergeometry and fermionic fields.}
In this vein, also fermionic supergeometry finds its natural home simply by further enlarging the ambient category of commutative algebras \eqref{EmbeddingOfSmoothManifoldsIntoCAlgOp} to that of $\mathbb{Z}/2$-graded commutative algebras, to be denoted $\mathrm{sCAlg}_{{}_{\FR}}$.
Here we want the ``smooth functions'' on a purely odd super-space $\FR^{0\vert q}$ to be the Grassmann-algebra $\wedge^\bullet\big( (\FR^q)^\ast \big)$ 
on $q$ generators, and hence we may define {\it super-Cartesian spaces} to be those which constitute the following full subcategory:
\vspace{-2mm} 
\[
  \begin{tikzcd}[row sep=-4pt]
    \SuperCartesianSpaces
    \ar[
      rr,
      "{ C^\infty(-) }"
    ]
    &&
    \mathrm{sCAlg}_{{}_{\FR}}^{\mathrm{op}}
    \\
    \FR^{k \vert q}
    \times
    \DD^m_r
    \ar[
      rr,
      phantom,
      "{ \longmapsto }"
    ]
    &&
    C^\infty(\FR^k)
    \otimes
    \wedge^\bullet\big(
      (\FR^q)^\ast
    \big)
    \otimes
    \FR[\epsilon_1,...\epsilon_m]
    /
    (
     \epsilon^{r+1}
    )
    \,.
  \end{tikzcd}
\]

\vspace{-1mm} 
\noindent With the open covers of such super-Cartesian spaces defined in the same 
manner as before, we obtain the topos of  {\it super smooth sets}
\vspace{-2mm} 
\[
  \SuperSmoothSets
  \;:=\;
  \mathrm{Sh}\big(
    \SuperCartesianSpaces
  \big)
  \,.
\]
This topos is the home of classical fermionic fields (cf. \cite{KonechnySchwarz98}\cite{Sachse08}).  In particular, it is the internal hom,
mapping space construction that encodes fermionic field spaces as used in the physics literature. Such spaces often have \textit{no underlying points} at all, and the non-trivial information is instead encoded via higher $\FR^{0|q}$-plots. 
\\

\item
{\bf Higher geometry and gauge fields.}
On the other hand, the characteristic property of (gauge) bosons is that their field spaces have a relaxed notion of {\it equality}, 
in that two (plots/families of) gauge fields
\vspace{-2mm} 
\[
  \Phi,
  \Phi'
  \,:\,
  \FR^{k}
  \xrightarrow{\phantom{---}}
  \mathrm{Fields}
\]

\vspace{-2mm} 
\noindent 
may be nominally distinct and yet identified via gauge transformations $\begin{tikzcd}\!\!\! \Phi \ar[r, "{g}", "{\sim}"{swap}] & \Phi'
\!\!\!\end{tikzcd}$,    
that are invertible and may be associatively composed:

\vspace{-4mm} 
\begin{equation*}
  \label{CompositionOfGaugeTransformations}
  \begin{tikzcd}[row sep=-50pt, column sep=large]
    \FR^{k}
    \ar[
      rr,
      bend left=60,
      "{ \Phi }"{description, name=s}
    ]
    \ar[
      rr,
      bend right=60,
      "{ \Phi'' }"{description, name=t}
    ]
    \ar[
      from=s,
      to=t,
      bend left=20,
      Rightarrow,thick, color=greenii,
      "{ h \circ g }"{sloped, swap, pos=.76},
      "{ \sim }"{sloped, pos=.2},
    ]
    \ar[
      rr,
      crossing over,
      "{ 
        \Phi'
      }"{description, name=i, pos=.27}
    ]
    \ar[
      from=s,
      to=i,
      shorten=2pt,
      Rightarrow, thick, color=greenii,
      "{ g }"{swap, pos=.6},
      "{ \sim }"{sloped, pos=.4, yshift=-.5pt, swap}
    ]
    \ar[
      from=i,
      to=t,
      shorten=2pt,
      Rightarrow, thick, color=greenii,
      "{ h }"{swap, pos=.2},
      "{ \sim }"{sloped, swap, pos=.35, yshift=-.5pt, swap},
    ]
    &&
    \mathrm{Fields}
    \mathrlap{\,.}
  \end{tikzcd}
\end{equation*}
\vspace{-2mm} 

\noindent 
This means that the plots of gauge field spaces no longer form plain sets, but {\it groupoids} (exposition of groupoids includes \cite{Weinstein96}).

\smallskip

Furthermore, for {\it higher} gauge fields, two such gauge transformations, in turn, may be nominally distinct and yet identified by ``gauge-of-gauge transformations''
\vspace{-2mm} 
\begin{equation*}
  \begin{tikzcd}[scale=1.5]
    \FR^{k}
    \ar[
      rr,
      bend left=40,
      "{ \Phi }"{description, name=s}
    ]
    \ar[
      rr,
      bend right=40,
      "{ \Phi' }"{description, name=t}
    ]
    \ar[
      from=s,
      to=t,
      shorten=2pt,
      bend left=50,
      Rightarrow, thick, color=greenii,
      "{ \sim }"{sloped, yshift=-.5pt},
      "{\ }"{name=tt, swap}
    ]
    \ar[
      from=s,
      to=t,
      shorten=2pt,
      bend right=50,
      Rightarrow, thick, color=greenii,
      "{ \sim }"{sloped, yshift=.5pt},
      "{\ }"{name=ss}
    ]
    \ar[
      from=ss,
      to=tt,
      phantom, thick, color=orangeii,
      "{
        \Rightarrow
      }"{scale=1.5},
      "{ \sim }"{yshift=4.5pt, scale=.7, pos=.2}
    ]
    \ar[
      rr,
      -,
      color=orangeii,
      shorten <=20pt,
      shorten >=16pt,
      line width=.5pt,
      shift left=1.2pt
    ]
    &&
    \mathrm{Fields}
  \end{tikzcd}
\end{equation*}

\vspace{-2mm} 
\noindent
satisfying a 2-dimensional analog of composition and associativity, as schematically indicated here:
\vspace{-3mm} 
\begin{equation*}
  \label{ThreeSimplex}
  \begin{tikzcd}[
    row sep=25pt,
    column sep=23
  ]
    \mathllap{
      \scalebox{.7}{
        \color{gray}
        $(\Delta^0)$
      }
    }
    \Phi
    &&
    \phantom{\Phi'}
    \\[-25pt]
    \mathllap{
      \scalebox{.7}{
        \color{gray}
        $(\Delta^1)$
      }
    }
    \Phi
    \ar[
      rr,
      "{ g }"
    ]
    &&
    \Phi'
    \\[-25pt]
    &
    |[alias=two]|
    \Phi'
    \ar[
      dr,
      "{ h }"
    ]
    \\
    \Phi
    \ar[
      ur,
      "{ 
        \mathllap{
          \scalebox{.9}{
            \color{gray}
            $(\Delta^2)$
            \;\;\;\,
          }
        }
        g 
      }"
    ]
    \ar[
      rr,
      "{ h \circ g }"{swap},
      "{\ }"{name=onethree}
    ]
    &&
    \Phi''
    \ar[
      from=two,
      to=onethree,
      Rightarrow, thick, color=greenii,
      shorten <=5pt,
      "{
        \mu(g,h)
      }"{description, pos=.6}
    ]
  \end{tikzcd}
  \begin{tikzcd}[column sep=24pt]
    &[-10pt]&
    &[-50pt] 
    |[alias=two]|
    \Phi'
    \ar[
      ddr,
      "{\ }"{name=twofour, swap}
    ]
    &[-5pt]
    \\
    \\
    |[alias=one]|
    \Phi
    \ar[
      rrrr,
      "{\ }"{name=onefour}
    ]
    \ar[
      drr,
      "{\  }"{name=onethree}
    ]
    \ar[
      uurrr,
      "{
        \mathllap{
          \scalebox{.9}{
            \color{gray}
            $(\Delta^3)$
          }
        }
      }"{yshift=15pt}
    ]
    &&&&
    |[alias=four]|
    \Phi'''
    \mathrlap{\,,}
    \\[-3pt]
    &&
    |[alias=three]|
    \Phi''
    \ar[
      urr,
      "{\ }"{name=one}
    ]
    \ar[
      from=two, 
      to=onefour,
      Rightarrow, thick, color=greenii,
      crossing over,
      shorten=5pt
    ]
    \ar[
      from=two, 
      to=onethree,
      Rightarrow, thick, color=greenii,
      shorten <=10pt,
      shorten >=1pt,
      crossing over
    ]
    \ar[
      from=three, 
      to=onefour,
      Rightarrow, thick, color=greenii,
      crossing over,
      shorten=5pt
    ]
    \ar[
      from=two,
      to=three,
      crossing over,
      "{\ }"{name=twothree}
    ]
    \ar[
      from=three, 
      to=twofour,
      shorten <= 10pt,
      shorten >= 2pt,
      Rightarrow, thick, color=greenii,
      crossing over,
    ]
  \end{tikzcd}
  \hspace{-3pt}
\end{equation*}

\vspace{-3mm} 
\noindent and so on to ever higher order gauge transformations, now making the plots form {\it higher groupoids}, which may be thought of as Kan simplicial sets, for exposition see \cite[\S 7]{Friedman12}).

Hence the plots of a higher gauge field spaces form generally not a plain set but a Kan simplicial set, organizing generally not into a plain presheaf but into a (Kan-){\it simplicial presheaf} (cf. \cite{Jardine87}).
$$
  \mathrm{PSh}\big(
    \SuperCartesianSpaces
    ,\,
    \mathrm{sSet}_{\mathrm{Kan}}
  \big)
  \;=\;
  \Bigg\{\!\!\!
  \begin{tikzcd}[
    row sep=-4pt
  ]
    \SuperCartesianSpaces^{\mathrm{op}}
    \ar[
      rr,
      "{
        \CF
      }"
    ]
    &&
    \mathrm{sSet}_{\mathrm{Kan}}
    \\
    \FR^{k}
    \ar[
      rr,
      phantom,
      "{ \longmapsto }"
    ]
    &&
    \mathrm{Plots}(-,\mathcal{G})
  \end{tikzcd}
\!\!\!  \Bigg\}.
$$
For example, the moduli classifying stack of $G$-Yang-Mills fields, cf. \eqref{CharPolynomials}, 
has groupoids of plots given as follows (\cite[Ex. 2.2.4]{FSS13}\cite[\S 2.4]{FSS14}\cite[Ex. 2.11]{BSS18}):
\begin{equation}
  \label{BGConn}
  \def\arraystretch{1.6}
  \begin{array}{l}
      \mbox{$\mathbold{Plots}$}\big(
        \FR^{k}
        ,
        \mbox{$\mathbold{B}$}G_{\mathrm{conn}}
      \big)_2
    \;:=\;
    \left\{\!\!\!\!\!
    \adjustbox{raise=4pt}{
    \def\arraystretch{1.4}
    \begin{tikzcd}[
      column sep=15pt
    ]
      &
      |[alias=two]|
      A_1
      \ar[
        dr,
        "{ g_2 }"
      ]
      \\
      A_0
      \ar[
        rr,
        "{
          g_2
          \cdot
          g_1
        }"{swap},
        "{\ }"{name=onethree}
      ]
      \ar[
        ur,
        "{
          g_1
        }"
      ]
      &&
      A_2
      \ar[
        from=two,
        to=onethree,
        shorten <=4pt,
        Rightarrow, thick, color=greenii,
        "{ \exists! }"
      ]
    \end{tikzcd}
    }
    \,\middle\vert\,
    \def\arraystretch{1.2}
    \begin{array}{l}
      A_i \in
      \Omega^1(\FR^{k})
      \otimes \mathfrak{g},
      \\
      g_i 
      \in
      C^\infty(\FR^{k}, G)
      \\
      A_i
      =
      g_i A_{i-1} g_i^{-1}
      +
      g_i 
      \mathrm{d}
      g_i^{-1}
    \end{array}
  \!\!\!  \right\}.
  \end{array}
\end{equation}
Here a plot of $\mathbf{B}G_{\mathrm{conn}}$ is a $G$-gauge field configuration on the probe chart $\FR^{n}$, given by a $\mathfrak{g} = \mathrm{Lie}(G)$-valued 
1-form $A$ (a gauge potential) and an equivalence of plots is a smooth $G$-valued function $g$ on the probe chart relating two such 1-forms by the usual 
gauge transformations \eqref{GaugeTransformations}. As hinted in Rem. \ref{OnGaugeEquivalencyandRedundancy}, the \textit{set} of $G$-gauge fields on a 
spacetime manifold $M$ is then given by maps $M\rightarrow \mathbf{B}G_{\mathrm{conn}}$, with the internal hom object $\big[M,\mathbf{B}G_{\mathrm{conn}}\big]$
encoding the gauge equivalency transformations between any two gauge fields via its $\Delta^1$-plots.

\smallskip 
As before, this means that the naive formulas known from physics apply {\it on probes} and the (now: higher) topos theory takes care of producing 
global structures from these probes. We will further explain how this works in \cite{GSS-3}.

\medskip 
Generally, we emphasize that even without considering gauge fields and/or gauge symmetries, higher structures are implicit throughout the treatment 
of a local field theory. In particular, the finite (and infinitesimal) symmetries of a Lagrangian field theory (Def. \ref{FiniteSymmetryofLagrangianFieldTheory}) 
can be interpreted as saying that the
Lagrangian is preserved \textit{up to homotopy}, in a precise sense. Similarly, a careful discussion of the corresponding conserved and Hamiltonian
currents shows that higher structures in the guise of infinitesimal versions of $\infty$-groupoids, i.e., $L_\infty$-algebras, are secretly 
dwelling within local classical field theories (Rem. \ref{HigherPoissonAlgebrasOfLocalObservables}). 
\end{itemize}

\vspace{1cm} 
\noindent {\bf Acknowledgements.} 
 We thank Urs Schreiber for detailed and helpful discussions 
on various aspects of this paper. We also thank Lukas M\"uller and Dmitri Pavlov for useful comments on an earlier draft. 

\vspace{1cm} 
\noindent {\bf Data availability.} 
No additional research data beyond the data included and cited in this work
are needed to validate the research findings presented.

\newpage 


\end{document}